\documentclass[aip,jmp,reprint,author-numerical,onecolumn,nobibnotes]{revtex4-1}

\usepackage[utf8]{inputenc}
\usepackage{mathrsfs}
\usepackage{dsfont}
\usepackage{amsmath}
\usepackage{amssymb}
\usepackage{amsthm}
\usepackage{amsfonts}
\usepackage{amstext}
\usepackage{amsopn}
\usepackage{amsxtra}
\usepackage{hyperref}
\usepackage{graphicx}
\newcommand{\dps}{\displaystyle}
\newcommand{\ii}{\infty}
\newcommand\R{{\ensuremath {\mathbb R} }}
\newcommand\bS{{\ensuremath {\mathbb S} }}
\newcommand\bE{{\ensuremath {\mathbb E} }}
\newcommand\C{{\ensuremath {\mathbb C} }}
\newcommand\N{{\ensuremath {\mathbb N} }}
\newcommand\Z{{\ensuremath {\mathbb Z} }}

\newcommand\1{{\ensuremath {\mathds 1} }}
\newcommand\bP{{\ensuremath {\mathds P} }}
\newcommand\bQ{{\ensuremath {\mathds Q} }}
\renewcommand\phi{\varphi}

\newcommand{\gS}{\mathfrak{S}}

\newcommand{\wto}{\rightharpoonup}

\newcommand{\cR}{\mathcal{R}}

\newcommand{\cX}{\mathcal{X}}

\newcommand{\cE}{\mathcal{E}}

\newcommand{\cF}{\mathcal{F}}

\newcommand{\cD}{\mathcal{D}}

\newcommand{\cL}{\mathscr{L}}
\newcommand{\eps}{\epsilon}

\newcommand{\dr}{{\rm d}r}
\newcommand{\dx}{{\rm d}x}
\newcommand{\dy}{{\rm d}y}
\newcommand{\dz}{{\rm d}z}
\newcommand{\dk}{{\rm d}k}
\newcommand{\du}{{\rm d}u}
\newcommand{\dv}{{\rm d}v}
\newcommand{\dt}{{\rm d}t}
\newcommand{\rd}{{\rm d}}
\renewcommand{\epsilon}{\varepsilon}
\newcommand\pscal[1]{{\ensuremath{\left\langle #1 \right\rangle}}}

\renewcommand{\geq}{\geqslant}
\renewcommand{\leq}{\leqslant}

\renewcommand{\tilde}{\widetilde}
\newcommand{\nn}{\nonumber}

\theoremstyle{plain}
\newtheorem{theorem}{Theorem}
\newtheorem{lemma}[theorem]{Lemma}
\newtheorem{corollary}[theorem]{Corollary}

\newtheorem{conjecture}[theorem]{Conjecture}

\newtheorem{definition}[theorem]{Definition}

\newtheorem{remark}[theorem]{Remark}

\begin{document}

\title{Coulomb and Riesz gases: The known and the unknown} 

\author{Mathieu Lewin}
\email[]{mathieu.lewin@math.cnrs.fr}
\affiliation{CNRS \& CEREMADE, Universit\'e Paris-Dauphine, PSL University, Place de Lattre de Tassigny, 75 016 PARIS, FRANCE}

\date{\today}

\begin{abstract}
We review what is known, unknown and expected about the mathematical properties of Coulomb and Riesz gases. Those describe infinite configurations of points in $\R^d$ interacting with the Riesz potential $\pm |x|^{-s}$ (resp.~$-\log|x|$ for $s=0$). Our presentation follows the standard point of view of statistical mechanics, but we also mention how these systems arise in other important situations (e.g. in random matrix theory). The main question addressed in the article is how to properly define the associated infinite point process and characterize it using some (renormalized) equilibrium equation. This is largely open in the long range case $s<d$. For the convenience of the reader we give the detail of what is known in the short range case $s>d$. In the last part we discuss phase transitions and mention what is expected on physical grounds.

\bigskip

\footnotesize \noindent \copyright~2022 by the author. This paper may be reproduced, in its entirety, for non-commercial purposes
\end{abstract}

\pacs{05.20.-y,02.30.Em}

\maketitle 


\begin{flushright}
\sl in memory of Freeman J. Dyson (1923--2020)
\end{flushright}

\tableofcontents

\section{Introduction: the Riesz potential}
A \emph{Riesz gas}\cite{Riesz-38} is a family of probabilities over random infinite configurations of points in~$\R^d$ (point processes), which has a specific behavior with respect to scaling. It depends on the two parameters
\begin{equation}
 s\in(-2,+\ii],\qquad \Upsilon=\beta\rho^{\frac{s}d}\in[0,\ii],
 \label{eq:parameters}
\end{equation}
where $\rho>0$ is the average number of points per unit volume (also called the intensity), $\beta=1/T\in(0,\ii]$ is the inverse temperature which controls the amount of randomness in the system and $s$ is a degree of homogeneity. More precisely, a Riesz gas is a \emph{Gibbs point process}~\cite{Georgii-11} for which the interaction energy of every point $x_{j_0}$ with all the other points $x_j$ in an infinite configuration is (formally) given by
\begin{equation}
\sum_{j\neq j_0} V_s(x_{j_0}-x_j)
\label{eq:interaction_energy_intro}
\end{equation}
with the homogeneous potential of degree $-s$
\begin{equation}
 V_s(x)=\begin{cases}
 |x|^{-s}&\text{for $s\in(0,+\ii]$,}\\
 -\log|x|&\text{for $s=0$,}\\
 -|x|^{-s}&\text{for $s\in(-2,0)$.}\\
 \end{cases}
\label{eq:def_V_s_intro}
\end{equation}
For $s=0$ the name \emph{log gas} is also sometimes used. When $s\to0^+$ we have $V_s(x)=1-s\log|x|+o(s)_{s\to0^+}$ for every $x\neq0$ and thus  retain the first non-trivial term in the expansion. The \emph{Coulomb gas} corresponds to $s=d-2$, in which case $V_s$ is proportional to the fundamental solution of the Laplacian:
\begin{equation}
 -\Delta V_{d-2}=\begin{cases}
|\bS^{d-1}|\,(d-2)\,\delta_0&\text{for } d\geq3,\\
2\pi \,\delta_0,&\text{for } d=2,\\
2 \,\delta_0,&\text{for } d=1.
\end{cases}
 \label{eq:fund_solution_Coulomb}
\end{equation}
For $s=+\ii$ we recover the \emph{hard sphere gas} with impenetrable spheres of radius $1/2$.

In the definition~\eqref{eq:def_V_s_intro}, the signs are chosen to ensure that $V_s$ is \emph{repulsive}, that is, \emph{decreasing with $|x|$}. This way, the points will not be too close from each other. The amount of repulsion depends on the parameter $s$ and on the distance between the points. The repulsion is stronger at small distances for large $s$, and at large distances for small $s$. A natural threshold is given by $s=d$, the space dimension. For $s>d$ the potential $V_s$ is integrable at infinity but not at the origin. In this case the series~\eqref{eq:interaction_energy_intro} converges for any reasonable infinite configuration of points $\{x_j\}_{j\in\N}\subset\R^d$, for instance when the smallest distance between the $x_j$'s is positive. This is called the \emph{short range case}. On the contrary, when $s\leq d$ the series~\eqref{eq:interaction_energy_intro} will usually never converge and it has to be properly \emph{renormalized}. This is the \emph{long range case} to which belongs the important Coulomb potential at $s=d-2$.

In order to be able to renormalize the series~\eqref{eq:interaction_energy_intro} in the long range case $s\leq d$, it will be very important that $V_s$ be \emph{positive-type}, that is, with a positive Fourier transform. Recall that the Fourier transform of $|x|^{-s}$ is
\begin{equation}
\widehat{\frac1{|x|^s}}(k)=\frac{2^{\frac{d}2-s}\Gamma\left(\frac{d-s}{2}\right)}{\Gamma\left(\frac{s}{2}\right)}\,{\rm f.p.}\!\left(\frac1{|k|^{d-s}}\right),\qquad 0\neq s<d,\quad s\notin -2\N,
 \label{eq:Fourier_V_s}
\end{equation}
where, for $s<0$, ${\rm f.p.}$ denotes the Hadamard finite part (see Ref.~\onlinecite[VII.7]{Schwartz}). In the denominator $\Gamma(s/2)$ changes sign at every non-positive even integer. This explains our constraint that $s>-2$ in~\eqref{eq:def_V_s_intro}. When $-4<s<-2$, no multiple of $|x|^{-s}$ is both repulsive and positive-type.

Whereas $s$ controls the repulsion between the points, the other parameter $\Upsilon=\beta \rho^{s/d}$ is used to monitor the amount of randomness in the system:
\begin{itemize}
 \item When $\Upsilon$ is large, our point process will usually be \emph{strongly correlated}, with the positions of the individual points highly dependent of the others ones. Our Riesz point process could be non  unique (which is related to the breaking of symmetries, as we will see). For $\Upsilon=+\ii$ uniqueness will in fact never hold.
 \item When $\Upsilon$ is small, the point process will be unique, hence invariant under all isometries of $\R^d$. Correlations will decay fast, which means that the points in two regions located far away from each other will be almost independent. This situation is usually called a `gas' in statistical physics.
\end{itemize}
Even if for large values of $\Upsilon$ the `Riesz gas' is in fact not a `gas', this name is nevertheless commonly used in the literature to illustrate that there are infinitely many points in the system.

Long range systems such as~\eqref{eq:def_V_s_intro} for $s\leq d$ play a very important role in physics~\cite{DauRufAriWil-02,CamDauRuf-09,CamDauFanRuf-14}. Typical examples include galaxies or self-gravitating stars (interacting with the Newton force), charged systems such as plasmas (interacting with the 3D Coulomb force), two-dimensional and geophysical flows (interacting with the 2D Coulomb force), dipolar systems such as dielectrics and diamagnets.

Long range Riesz gases also appear in many unexpected mathematical situations, including Ginzburg-Landau vortices~\cite{SanSer-12}, random matrices~\cite{Mehta-10,Forrester-10}, eigenvalues of random Schr\"odinger operators~\cite{AizWar-15}, quantum chaos~\cite{BohGiaSch-84}, Fekete points on manifolds~\cite{HarSaf-04,BorHarSaf-19}, complex geometry~\cite{BerBouWitt-11}, Laughlin functions~\cite{LieRouYng-19,Klevtsov-16}, zeros of random polynomials~\cite{ZeiZel-10}, zeros of the Riemann function~\cite{Montgomery-73,RudSar-96,BouKea-13}, modular forms and sphere packing problems~\cite{CohKumMilRadVia-19_ppt,Viazovska-21,PetSer-20}. Riesz gases appear everywhere and seem to be sort of universal. They have been the object of many old and recent works. Unfortunately, each area comes with its own definition of what a Riesz gas is and its own set of tools to study it. Those are not always easily transferred to other fields of research.

In this paper we will review what is known, unknown and expected for Riesz gases. Our definition of Riesz gases will follow the standard point of view of statistical mechanics, but we will also mention how these systems arise in some other important situations. We will provide some proofs, when they are simple enough or hard to find in the literature. What follows does not attempt to be an exhaustive treatment of Riesz gases and only reflects the point of view of the author. There are other possible approaches to the problem. In particular, we will \emph{not} speak at all about the quantum problem and multi-component systems (which have several types of points).

In spite of the large amount of works on the subject, many important questions are still completely open. Several of them will be mentioned in this work. The most important is probably the mere \emph{existence} of the Riesz point process and its characterization in terms of some equilibrium equation. This is well known for $s>d$ but open in many cases for $s<d$. One related difficulty is to give a meaning to the renormalized potential, a question which will occupy a large part of the article.

The review is dedicated to the memory of Freeman J. Dyson, who was extremely influential in the subject. With Wigner he has essentially created the mathematical field of random matrix theory~\cite{Dyson-62a,*Dyson-62b,*Dyson-62c,DysMeh-63a,DysMeh-63b}, where the log gas ($s=0$) appears naturally, as we will see in Section~\ref{sec:random_matrices}. In this context it is often called the \emph{Dyson gas}\cite{WieZab-00}. He gave the first proof of the instability of quantum bosonic matter with Coulomb forces~\cite{Dyson-67} and, with Lenard, of the stability of fermionic matter~\cite{DysLen-67,DysLen-68}. He was the one mentioning to Montgomery that the conjectured distribution of zeros of the Riemann Zeta function is the same as a log gas in dimension $d=1$~\cite{Montgomery-73}. In 1969, he proved\cite{Dyson-69} that short range lattice Riesz gases can undergo phase transitions in one dimension, if $s\leq2$. In 1971, he considered a classical Coulomb gas with two species of charges and a uniform background (now called a Wigner-Dyson lattice) and suggested this might be relevant for external crustal layers of neutron stars\cite{Dyson-71}.

In Sections~\ref{sec:thermo_limit_short_range} and~\ref{sec:thermo_limit_long_range} we will discuss how Riesz gases are defined in the framework of statistical mechanics. We take a bounded domain $\Omega\subset\R^d$ and a finite number of points $N$ in $\Omega$, with $\rho=N/|\Omega|$. We then look at what is happening to the Gibbs point process, in the limit $N \to\ii$ and $\Omega\nearrow\R^d$. This is usually called a \emph{thermodynamic limit}. In Section~\ref{sec:periodic} we discuss another way of approaching the problem using analytic continuation and periodized systems. Section~\ref{sec:confined} is a quick outline of confined systems, which are very often encountered in practical situations such as random matrices. In Section~\ref{sec:transitions} we discuss some known and conjectured properties of Riesz gases. We particularly focus on \emph{phase transitions}, that is, the uniqueness or non-uniqueness of the Riesz gas depending on the value of the parameters $s$ and $\Upsilon$. We mention there many results from the physical literature. Finally, Section~\ref{sec:appendix} contains some additional proofs.

\section{Thermodynamic limit in the short range case $s>d$}\label{sec:thermo_limit_short_range}

In this section and the following one, we discuss the construction of Riesz gases using the \emph{thermodynamic limit}, as is usually done in statistical mechanics~\cite{Ruelle,Lanford-73}. This method focuses right away on the point process in the whole space $\R^d$ and naturally preserves the symmetries of the problem, in particular the scaling invariance. As we will see later, in applications Riesz gases often appear in other limits. In many problems they describe the behavior of the system at a certain \emph{microscopic scale}, so that seeing it requires solving first the macroscopic problem and then zooming at the right length. These complications do not appear in the thermodynamic limit.

In order to better explain the difficulties of the long range case, we found it natural to first recall in this section what is known in the short range case $s>d$, which is very well understood since the 60--70s. Most of the results we will quote here are due to Ruelle\cite{Ruelle}, Dobru\v{s}in-Minlos\cite{DobMin-67} and Georgii\cite{Georgii-76} but some are a bit more difficult to locate in the literature. Many theorems hold the same with an arbitrary interaction potential decaying fast enough at infinity. However, the proofs for our potential $V_s$ tend to be much simpler, due to its positivity, and we thought we would as well provide some details. As is usual in statistical mechanics, we concentrate first on the convergence of the thermodynamic functions before looking at the point process itself.

Several physical systems may be appropriately described by such purely repulsive power-law potentials $|x|^{-s}$ with $s>d$, including for instance colloidal particles in charge-stabilized suspensions~\cite{PauAckWol-96,SenRic-99,DeiBalWilWei-01} or certain metals under extreme thermodynamic conditions~\cite{HooYouGro-72}. For $s=2$ in dimension $d=1$, one obtains an integrable system, the \emph{classical Calogero-Sutherland-Moser model}~\cite{Calogero-71,Sutherland-71,Sutherland-71b,Moser-75,Sutherland-04}.

\subsection{Canonical ensemble}
We consider $N$ distinct points $x_1,...,x_N$ in a bounded domain $\Omega\subset\R^d$ and then look at the limit $N\to\ii$ and $\Omega\nearrow \R^d$ with $N/|\Omega|\to\rho>0$, the desired intensity of our process. We will often take for simplicity  $\Omega=\ell\omega$ where $\omega$ is a given set of measure $|\omega|=1$, with $\ell\sim (N/\rho)^{1/d}$. Our point process will be defined in terms of the \emph{Riesz energy} of the $N$ points
\begin{equation}
\cE_s(x_1,...,x_N):=\sum_{1\leq j<k\leq N}V_s(x_j-x_k)=\sum_{1\leq j<k\leq N}\frac1{|x_j-x_k|^s},\qquad s>d,
 \label{eq:def_cE_s}
\end{equation}
and will depend on a temperature $T=1/\beta$.

We start with $T=0$ where we just minimize the energy of $N$ points in the domain $\Omega$. We denote by
\begin{equation}
\boxed{ E_s(N,\Omega):=\min_{x_1,...x_N\in\overline\Omega}\cE_s(x_1,...,x_N)}
 \label{eq:def_E_s}
\end{equation}
this minimum and call $\cX_s(N,\Omega)$ the set of the \emph{optimal configurations} $(x_1,...,x_N)\in\overline\Omega^N$ realizing the minimum. Note that, by scaling,
\begin{equation}
E_s(N,\Omega)=
\lambda^s\,E_s\big(N,\lambda \Omega\big),\qquad \cX_s(N,\lambda\Omega)=\lambda\,\cX_s(N,\Omega), \qquad\forall\lambda>0.
\label{eq:scaling_N_T0}
\end{equation}
In particular, we can always fix the density $\rho=N/|\Omega|$ to be any number that we like. The parameter $\rho$ plays no role in our problem. In some cases it will however be instructive to remember how things explicitly depend on $\rho$. We could equivalently fix completely the domain $\Omega$ and not increase its size in the limit $N\to\ii$, but then we would have to zoom in order to see how the points are arranged at the microscopic scale. This point of view is discussed later in Section~\ref{sec:confined}.

Our goal is to look at the limit $N\to\ii$ and at the positions of the optimal points. In the end we hope to obtain an infinite configuration of points with an average of $\rho$ points per unit volume, and thus occupying the whole space. For this it is important that the points do not concentrate too much nor leave big holes. For such good configurations the interaction~\eqref{eq:interaction_energy_intro} of any point $x_{j_0}$ with the other points $x_j$ will be of order one. Due to the double sum, we thus expect the total energy~\eqref{eq:def_E_s} to be of order $N$. This is confirmed by the following elementary result, which can be found in many similar forms in the literature.

\begin{lemma}[Bounds on $E_s(N,\Omega)$]\label{lem:simple_estim_E_s}
Let $s>d$. There exists two constants $c_1,c_2>0$ depending only on $d$ and $s$, such that
\begin{equation}
c_1 \left(\frac{N}{|\Omega|}\right)^{1+\frac{s}d}\leq \frac{E_s(N,\Omega)}{|\Omega|}\leq c_2 \left(\frac{N}{|\Omega|}\right)^{1+\frac{s}d}
\label{eq:estim_E_N}
\end{equation}
for every bounded open set $\Omega$ with $|\partial\Omega|=0$, and every $N\geq N_0$ large enough. Here $N_0$ only depends on the `shape' of $\Omega$, that is, on $\omega=|\Omega|^{-1/d}\Omega$ and not on the volume $|\Omega|$.
\end{lemma}

In our case we will look at the limit where $N/|\Omega|\sim\rho$ is a fixed density. The result thus says that $E_s(N,\Omega)$ is of order $\rho^{1+s/d}|\Omega|\sim \rho^{s/d}N$.

\begin{proof}
Let $\omega$ be any bounded open set with $|\omega|=1$ and $|\partial\omega|=0$. By the scaling property~\eqref{eq:scaling_N_T0}, it suffices to prove the inequality~\eqref{eq:estim_E_N} for $\Omega=N^{1/d}\omega$, that is, at density $\rho=N/|\Omega|=1$. For the upper bound  we place our points on a lattice with density $>1$, for instance $\Z^d/2$ which has density $2^d$. The number of points of $\Z^d/2$ intersecting $\Omega$ is at least equal to the number of cells located inside $\Omega$ and can thus be estimated by
$$\#\Omega\cap \frac{\Z^d}{2}\geq 2^d|\Omega|-2^d|\partial\Omega+ C|=2^dN\left(1-\big|\partial\omega+N^{-\frac1d}C\big|\right),$$
where $2^d|\partial\Omega+ C|$ is an estimate on the number of cells intersecting the boundary and $C=[-\frac14,\frac14)^d$ is the cube of side length $1/2$. Since $\partial\omega$ is compact, we have $|\partial\omega+N^{-1/d}C|\to|\partial\omega|=0$ by dominated convergence and thus there are $2^dN+o(N)$ points in $\Omega$. This is more than necessary. We place our $N$ points on any such sites of $\Omega\cap \Z^d/2$. Completing the series, we obtain the upper bound
$$E_s(N,\Omega)=\frac12\sum_{j=1}^N\sum_{k\neq j}\frac{1}{|x_j-x_k|^s}\leq \frac{N}2\sum_{z\in\frac{\Z^d}2\setminus\{0\}}\frac{1}{|z|^s}=2^{s-1}N\sum_{z\in\Z^d\setminus\{0\}}\frac{1}{|z|^s},$$
as soon as there are at least $N$ points in $\Omega$, that is, $|\partial\omega+N^{-\frac1d}C|\leq 1-2^{-d}$. This is the claimed upper bound in~\eqref{eq:estim_E_N}.

For the lower bound, we consider a lattice with density $<1$, for instance $2\Z^d$. This defines the tiling $\R^d=\cup_{z\in 2\Z^d}(z+C')$
where, this time, $C'=[-1,1)^d$. The number of cells intersecting $\Omega$ satisfies, similarly as before,
\begin{equation*}
K:=\#\{z\in 2\Z^d\ :\ (z+C')\cap \Omega\neq\emptyset\}
\leq2^{-d}N\left(1+\big|\partial\omega+N^{-\frac1d}C'\big|\right).
\end{equation*}
We call $z_1,...,z_K\in2\Z^d$ the centers of the $K$ cubes intersecting $\overline\Omega$. Let $x_1,...,x_N$ be a minimizer for $E_s(N,\Omega)$ and denote by $n_k:=\#\{x_1,...,x_N\}\cap (z_k+C')$ the number of points located in the cube $z_k+C'$. Ignoring the interactions between the points in different cubes and using that $|x_i-x_j|\leq 2\sqrt{d}$ for the $n_k(n_k-1)/2$ pairs in each cube, we obtain
\begin{equation*}
E_s(N,\Omega)\geq \frac{1}{2^{1+s}d^{\frac{s}2}}\sum_{k=1}^K n_k(n_k-1)=\frac{1}{2^{1+s}d^{\frac{s}2}}\left(\sum_{k=1}^K n_k^2-N\right).
\end{equation*}
To bound $\sum_{k=1}^K n_k^2$ from below, we write
$$\sum_{k=1}^K n_k^2=\sum_{k=1}^K \left(n_k-\frac{N}{K}\right)^2+\frac{N^2}{K}\geq \frac{N^2}{K}$$
since $\sum_{k=1}^Kn_k=N$. This provides the lower bound
$$E(N,\Omega)\geq \frac{N}{2^{1+s}d^{\frac{s}2}}\left(\frac{N}{K}-1\right)\geq \frac{N}{2^{2+s}d^{\frac{s}2}}(2^d-1)$$
when $N$ is large enough so that $|\partial\omega+N^{-\frac1d}C'|\leq (2^d-1)/(2^d+1)$.
\end{proof}

Next we turn to the positive temperature case $T=1/\beta>0$. We have to consider random sets of $N$ points and a probability measure on such sets, which is the same as taking a symmetric positive measure $\bP$ on $\Omega^N$ with the normalization condition
\begin{equation}
 \frac{1}{N!}\int_{\Omega^N}\rd \bP(x_1,...,x_N)=1
 \label{eq:normalization_bP}
\end{equation}
and the constraint that the `diagonal' (where some of the $x_j$ coincide) has zero $\bP$-measure. The $1/N!$ is because we want to count each configuration only once. Our probability measure is obtained by minimizing the \emph{free energy} (`energy minus $T\times$entropy')
\begin{multline*}
\cF_s(\beta,N,\Omega,\bP)=\frac{1}{N!}\int_{\Omega^N}\cE_s(x_1,...,x_N)\bP(x_1,...,x_N)\,\dx_1\cdots \dx_N\\+\frac{\beta^{-1}}{N!}\int_{\Omega^N}\bP(x_1,...,x_N)\log\bP(x_1,...,x_N)\,\dx_1\cdots \dx_N,
\end{multline*}
under the constraint~\eqref{eq:normalization_bP}. The value of the minimum is
\begin{equation}
\boxed{F_s(\beta,N,\Omega)=\min_{\bP}\cF_s(\beta,N,\Omega,\bP)=-\beta^{-1} \log Z_s(\beta,N,\Omega),}
\label{eq:Gibbs_var}
\end{equation}
with unique minimizer
\begin{equation}
\boxed{ \bP_{s,\beta,N,\Omega}(x_1,...,x_N)=Z_s(\beta,N,\Omega)^{-1}e^{-\beta\cE_s(x_1,...,x_N)}}
 \label{eq:def_bP_canonical}
\end{equation}
called the \emph{canonical Gibbs measure}. The normalization factor
\begin{equation}
Z_s(\beta,N,\Omega)=\frac1{N!}\int_{\Omega^N}e^{-\beta\cE_s(x_1,...,x_N)}\dx_1\cdots \dx_N=e^{-\beta F_s(\beta,N,\Omega)}
\label{eq:partition_fn_canonical}
\end{equation}
is called the \emph{partition function}. Note that $\bP_{s,\beta,N,\Omega}$ in~\eqref{eq:def_bP_canonical} is in fact a smooth function on $\Omega^N$ which vanishes exponentially fast on the diagonal. The scaling relation now reads
\begin{equation}
F_s(\beta,N,\Omega)=
\lambda^{s}\,F_s\big(\beta\lambda^s,N,\lambda\,\Omega\big)+\beta^{-1}N\log(\lambda^d).
\label{eq:scaling_N_T}
\end{equation}
Choosing $\lambda=\rho^{1/d}$ we are reduced to considering the case where $\rho=1$ and $\beta$ is replaced by the parameter $\Upsilon=\beta\rho^{s/d}$ announced in the introduction. In the limit $\beta=1/T\to+\ii$, $F_s(\beta,N,\Omega)$ converges to our previous minimum energy $E_s(N,\Omega)$. The following contains simple bounds similar to that of Lemma~\ref{lem:simple_estim_E_s}.

\begin{lemma}[Bounds on $F_s(\beta,N,\Omega)$]\label{lem:simple_estim_F_s}
Under the same conditions as in Lemma~\ref{lem:simple_estim_E_s}, we have
\begin{multline}
c_1\left(\frac{N}{|\Omega|}\right)^{1+\frac{s}d}+\frac{N\beta^{-1}}{|\Omega|}\left(\log\frac{N}{|\Omega|}-1\right)\leq \frac{F_s(\beta,N,\Omega)}{|\Omega|}\\
\leq c_3\left(\frac{N}{|\Omega|}\right)^{1+\frac{s}d}+\frac{N\beta^{-1}}{|\Omega|}\left(\log\frac{N}{|\Omega|}+c_4\right).
\label{eq:estim_F_N}
\end{multline}
\end{lemma}

\begin{proof}
Using $\cE_s\geq E_s(N,\Omega)$ and Stirling's estimate $N!\geq(N/e)^N$, we get
\begin{equation}
F_s(\beta,N,\Omega)\geq E_s(N,\Omega)-\beta^{-1}\log\frac{|\Omega|^N}{N!}\geq E_s(N,\Omega)+N\beta^{-1}\left(\log\frac{N}{|\Omega|}-1\right).
 \label{eq:estim_lower_simple_temp}
\end{equation}
Inserting the lower bound~\eqref{eq:estim_E_N} from Lemma~\ref{lem:simple_estim_E_s}, we obtain the lower bound in~\eqref{eq:estim_F_N}. For the upper bound we argue as in the proof of Lemma~\ref{lem:simple_estim_E_s}, smearing out the points a little to have a finite entropy. We give ourselves a smooth non-negative function $\chi$ supported in the ball $B_{1/4}$ (centered at the origin and of radius $1/4$) with $\int\chi=1$ and $\int\chi|\log\chi|<\ii$. Assuming $N/|\Omega|=\rho=1$ for simplicity, we then place $N$ independent points, identically distributed with $\chi$, around the same $N$ points $z_1,...,z_N\in \Z^d/2\cap \Omega$ as in the $T=0$ case. In other words, we take the symmetric probability
$$\bP=\sum_{\sigma\in\gS_N}\prod_{j=1}^N\chi(x_j-z_{\sigma(j)}).$$
From the minimization principle~\eqref{eq:Gibbs_var}, we find
$$F_s(\beta,N,\Omega)\leq \frac1{N!}\int \cE_s\bP+\frac1{N!\beta}\int\bP\log\bP\leq Nc_3+N\beta^{-1}c_4$$
with the constants
$$c_3=\sum_{z\in\Z^d/2\setminus\{0\}}\iint_{\R^d\times\R^d}\frac{\chi(x)\chi(y)}{|x-y+z|^s}\dx\,\dy,\qquad c_4=\int_{\R^d}\chi\log\chi.$$
The upper bound in~\eqref{eq:estim_F_N} follows from~\eqref{eq:scaling_N_T} with $\lambda^d=N/|\Omega|$.
\end{proof}

We have seen that $E_s(N,\Omega)$ and $F_s(\beta,N,\Omega)$ behave linearly in $N$ and $|\Omega|$ at fixed density $\rho=N/|\Omega|$. Next, we address the existence of the thermodynamic limit. The following is a standard result dating from the 60s.

\begin{theorem}[Canonical thermodynamic functions~\cite{Ruelle-63a,Fisher-64,FisRue-66,Ruelle}]\label{thm:limit_C}
Assume that $s>d$. Let $\omega$ be any bounded open set with $|\omega|=1$ and $|\partial\omega|=0$. Then, for any $\rho>0$ and $\beta>0$, the following limits
\begin{equation}
\lim_{\substack{N\to\ii\\ \frac{N}{\ell^d}\to\rho}}\frac{E_s(N,\ell\omega)}{\ell^d}=e(s)\rho^{1+\frac{s}d},\qquad
\lim_{\substack{N\to\ii\\ \frac{N}{\ell^d}\to\rho}}\frac{F_s(\beta,N,\ell\omega)}{\ell^d}=f(s,\beta,\rho)
\label{eq:thermo_limit_free_energy}
\end{equation}
exist and are independent of $\omega$. The function $f$ satisfies the relation
\begin{equation}
 f(s,\beta,\rho)=\rho^{1+\frac{s}d}f\big(s,\beta\rho^{\frac{s}d},1\big)+\beta^{-1}\rho\log\rho.
 \label{eq:relation_f}
\end{equation}
We have
$$c_1\leq e(s)\leq c_2,\qquad c_1-\Upsilon^{-1}\leq f(s,\Upsilon,1)\leq c_3+\Upsilon^{-1}c_4$$
for all $s>d$ and $\Upsilon>0$, with the constants given by Lemmas~\ref{lem:simple_estim_E_s} and~\ref{lem:simple_estim_F_s}.
The function $(\rho,\beta)\mapsto \beta f(s,\beta,\rho)$ is convex in $\rho$ and concave in $\beta$. In particular, $\Upsilon\mapsto f(s,\Upsilon,1)$ is continuous on $(0,\ii)$. At infinity, we have
$\lim_{\Upsilon\to+\ii}f(s,\Upsilon,1)=e(s)$.
\end{theorem}

The theorem for any $\rho>0$ follows immediately from the case $\rho=1$ by scaling. The relation~\eqref{eq:relation_f} implies that, at small density,
$f(s,\beta,\rho)\sim_{\rho\to0^+} \beta^{-1}\rho\log\rho$
which is the leading term of the entropy per unit volume of the Poisson point process. In the following we will often write for simplicity $f(s,\Upsilon):=f(s,\Upsilon,1)$ and hope that this does not cause any confusion. The functions $s\mapsto e(s)$ and $(s,\Upsilon)\mapsto f(s,\Upsilon)$ are respectively the \emph{energy per unit volume} and the \emph{free energy per unit volume} of our Riesz gas, at unit density. They can be shown to be continuous in $s>d$. Except in some exceptional cases, their precise value is unknown. We will explain later in Section~\ref{sec:crystal_conjecture} that $e(s)=\zeta(s)$ in dimension $d=1$~\cite{Ventevogel-78}, where $\zeta$ is the Riemann Zeta function, and will make some explicit conjectures on the value of $e(s)$ in dimensions $d=2,3$. The value of $e(s)$ is also known in dimensions $d\in\{8,24\}$~\cite{CohKumMilRadVia-19_ppt}. Finally, the case $s=2$ in dimension $d=1$ is the classical Calogero-Sutherland-Moser model~\cite{Calogero-71,Sutherland-71,Sutherland-71b,Moser-75,Sutherland-04} which is integrable\cite{GalMar-73,Choquard-00}, see Remark~\ref{rmk:CSM} below.

The usual way of proving the existence of the limits~\eqref{eq:thermo_limit_free_energy} is to start with the case of hypercubes and to show that the (free) energy is subadditive, up to small error corrections~\cite{Ruelle-63a,Fisher-64,FisRue-66,Ruelle}. For this the idea is to split the cube into $2^d$ equal smaller cubes and to evaluate the energy of the state obtained by taking independent optimizers in each of the smaller cubes, inserting security corridors. The error is just the interaction energy between the cubes, which is small thanks to the corridors and the integrability of $V_s$ at infinity. The proof for any domain $\omega$ is then done by tiling it with smaller cubes. In fact, the limits~\eqref{eq:thermo_limit_free_energy} hold for general sequences $\Omega_N\nearrow\R^d$ satisfying some regularity conditions of the boundary. It is not necessary that $\Omega_N$ is the rescaling of a fixed $\omega$.

The limit~\eqref{eq:thermo_limit_free_energy} for $e(s)$ was also shown in the recent Ref.~\onlinecite{HarSaf-04,HarSaf-05}, with the same method of proof that we have just described, and for $f(s,\beta,\rho)$ in a more general situation in Ref.~\onlinecite{HarLebSafSer-18}. In the limit $s\to d^+$, we have
\begin{equation}
\lim_{s\to d^+}(s-d)e(s)=\lim_{s\to d^+}(s-d)f(s,\Upsilon,1)=\frac{\pi^{d/2}}{\Gamma\!\left(\frac{d}2\right)}=\frac{|\bS^{d-1}|}{2}.
\label{eq:pole_e_s}
\end{equation}
This is proved in Ref.~\onlinecite{HarMicSaf-19} for $e(s)$ and for $f(s,\Upsilon,1)$ this follows from~\eqref{eq:estim_lower_simple_temp} and the same arguments as in Lemma~\ref{lem:simple_estim_F_s}. Thus the energy and free energy per unit volume diverge when $s\to d^+$, which is the threshold between the short and long range cases. The leading term does not depend on the inverse temperature $\beta$.

\begin{remark}[Large deviations]
At $\beta<\ii$, it is also possible to prove a \emph{Large Deviation Principle}, which is more precise than the thermodynamic limit in Theorem~\ref{thm:limit_C}. For our short range Riesz gas this was done in Ref.~\onlinecite{Georgii-94,Georgii-95} and later in a more general situation in Ref.~\onlinecite{HarLebSafSer-18}. We refer to Ref.~\onlinecite{Lewis-88a,Lewis-88e,LewZagPul-88,Ellis-85,Touchette-09} for a discussion of the importance of large deviations in the context of statistical mechanics.
\end{remark}

\subsection{Grand-canonical ensemble}
Instead of fixing the density $\rho$, it is often very convenient to work in the \emph{grand canonical setting}, where we allow random fluctuations of the number of points, and then control the average density by means of a dual variable $\mu$, called the \emph{chemical potential}. The grand-canonical problem has very nice algebraic properties which simplify many proofs and, in most cases, one can know everything on the canonical problem using the grand-canonical ensemble.

The grand canonical partition function is a kind of Laplace transform of the canonical one:
\begin{equation}
Z_s^{\rm GC}(\beta,\mu,\Omega):= \sum_{n=0}^\ii \frac{e^{\beta\mu n}}{n!}\int_{\Omega^n}e^{-\beta\cE_s(x_1,...,x_n)}\dx_1\cdots \dx_n=1+e^{\beta\mu}+\sum_{n=2}^\ii e^{\beta\mu n} Z_s(\beta,n,\Omega),
\end{equation}
where we used the convention that $E_s(n,\Omega)=0$ for $n\in\{0,1\}$. The grand-canonical free energy is
\begin{equation}
\boxed{F_s^{\rm GC}(\beta,\mu,\Omega):=-\beta^{-1}\log Z_s^{\rm GC}(\beta,\mu,\Omega).}
\end{equation}
The corresponding grand-canonical Gibbs measure is a collection of measures $\bP=(\bP_0,\bP_1,...)$ where each $\bP_n$ is the density for $n$ points, given by
$$\bP_n(x_1,...,x_n)=Z_s^{\rm GC}(\beta,\mu,\Omega)^{-1}e^{-\beta\cE_s(x_1,...,x_n)+\beta\mu n}.$$
This is the unique minimizer of the grand-canonical free energy
$$\cF_s^{\rm GC}(\beta,\mu,\Omega,\bP):=\sum_{n\geq2}\frac1{n!}\int_{\Omega^n}\cE_s\bP_n+\beta^{-1}\sum_{n\geq0}\frac{1}{n!}\int_{\Omega^n}\bP_n\log\bP_n$$
under the normalization constraint $\sum_{n\geq0}\bP_n(\Omega^n)/n!=1$.
The scaling relation now takes the form
\begin{equation}
F_s^{\rm GC}(\beta,\mu,\Omega)=\lambda^sF_s^{\rm GC}\Big(\beta\lambda^s\,,\,\frac{\mu-\beta^{-1}\log(\lambda^d)}{\lambda^s}\,,\,\lambda\Omega\Big).
\label{eq:scaling_N_T_GC}
\end{equation}
We may thus always work at $\mu=0$ after choosing $\lambda=e^{\beta\mu /d}$. We see that the parameter $\Upsilon=\beta\rho^{s/d}$ in the canonical problem is replaced by $\widetilde\Upsilon=\beta e^{\beta\mu s/d}$ in the grand-canonical setting. In other words, the activity $z:=e^{\beta\mu }$ plays the role of a density.
Similarly, at $T=0$ we may define
\begin{equation}
\boxed{E_s^{\rm GC}(\mu,\Omega)=\min_{n\geq0}\left\{E_s(n,\Omega)-\mu n\right\}=\!\!\min_{\substack{n\geq0\\ x_1,...,x_n\in\Omega}}\left\{\sum_{1\leq j<k\leq n}\frac1{|x_j-x_k|^s}-\mu\,n\right\},}
\label{eq:def_E_s_GC}
\end{equation}
which is the limit of $F_s(\beta,\mu,\Omega)$ when $\beta\to\ii$. The minimum in~\eqref{eq:def_E_s_GC} is always attained at a finite $n$ since, by Lemma~\ref{lem:simple_estim_E_s}, $E_s(n,\Omega)$ grows like $n^{1+s/d}$ in the limit $n\to\ii$ when $\Omega$ is fixed. We have
$$E_s^{\rm GC}(\mu,\Omega)=\lambda^{s}E_s^{\rm GC}\left(\lambda^{-s}\mu,\lambda\Omega\right).$$

The thermodynamic limit $|\Omega|\to\ii$ is similar to the canonical case, the two situations being related by a Legendre transform.

\begin{theorem}[Grand-canonical thermodynamic functions~\cite{Ruelle-63a,Fisher-64,FisRue-66,Ruelle}]\label{thm:limit_GC}
Assume that $s>d$. Let $\omega$ be any bounded open set with $|\omega|=1$ and $|\partial\omega|=0$. Then for every $\beta>0$ and $\mu\in\R$,
\begin{align}
\lim_{\ell\to\ii}\frac{E_s^{\rm GC}(\mu,\ell\omega)}{\ell^d}&=\min_{\rho\geq0}\{\rho^{1+\frac{s}d}e(s)-\mu\rho\}=-\frac{s e(s)^{-\frac{d}{s}}}{d\left(1+\frac{s}{d}\right)^{1+\frac{d}{s}}}\mu_+^{1+\frac{d}s},
\label{eq:thermo_limit_energy_GC}\\
\lim_{\ell\to\ii}\frac{F_s^{\rm GC}(\beta,\mu,\ell\omega)}{\ell^d}
&=\min_{\rho>0}\big\{f(s,\beta,\rho)-\mu\rho\big\}=:g(s,\beta,\mu).
\label{eq:thermo_limit_free_energy_GC}
\end{align}
The function $g$ satisfies $g(s,\beta,\mu)=z^{1+\frac{s}d}g(s,\beta z^{\frac{s}d},0)$, with $z=e^{\beta\mu}$.
The function $(\beta,\mu)\mapsto \beta g(s,\beta,\mu)$ is strictly concave in $\mu$ and concave in $\beta$.
\end{theorem}

In Theorem~\ref{thm:limit_C} we have seen that the limiting canonical free energy is a convex function of $\rho$. Therefore, it is as well the Legendre transform of the grand canonical one:
\begin{equation}
f(s,\beta,\rho)=\max_{\mu\in\R}\left\{g(s,\beta,\mu)+\mu\rho\right\}.
\label{eq:Legendre_GC_C}
\end{equation}
To any $\rho$ we can associate the $\mu$'s solving this maximum and to any $\mu$ we can associate the $\rho$'s satisfying the minimum in~\eqref{eq:thermo_limit_free_energy_GC}. The convexity implies that the functions are differentiable, except possibly on a countable set. Whenever the derivative exists, we have $\mu(\rho)=\partial_\rho f(s,\beta,\rho)$ and $\rho(\mu)=-\partial_\mu g(s,\beta,\mu)$. Recall that a jump in the derivative of a convex function corresponds to a constant slope over an interval for its Legendre transform.

The strict concavity in $\mu$ stated in Theorem~\ref{thm:limit_GC} is very important. For positive potentials as in our situation, it was proved by Ginibre in Ref.~\onlinecite{Ginibre-67}. The argument has then been rewritten in a more general context by Ruelle in Ref.~\onlinecite[Sec.~4]{Ruelle-70}. The strict concavity implies that the derivative in $\rho$ of the canonical free energy cannot have any jump, that is, $\Upsilon\mapsto f(s,\Upsilon)$ is in fact $C^1$. For any $\rho$ the maximum in~\eqref{eq:Legendre_GC_C} is attained at a \emph{unique $\mu=\mu(\rho)$}, given by
$$\boxed{\mu(\rho)=\frac\partial{\partial\rho}f(s,\beta,\rho).}$$
On the other hand, the grand-canonical free energy $\mu\mapsto g(s,\beta,\mu)$ can in principle have jumps in its derivative with respect to $\mu$, corresponding to (first order) phase transitions. At such a point several phases of different densities co-exist. Those also correspond to intervals where the canonical free energy $f(s,\beta,\rho)$ is linear in $\rho$.

At zero temperature we have due to~\eqref{eq:thermo_limit_energy_GC}
\begin{equation}
\mu(\rho)=\left(1+\frac{s}{d}\right)e(s)\,\rho^{\frac{s}d},\qquad \rho(\mu)=\frac{\mu_+^{\frac{d}s}}{\left(1+\frac{s}{d}\right)^{\frac{d}s}e(s)^{\frac{d}s}}\qquad\text{at $T=0$.}
\label{eq:mu_rho_T0}
\end{equation}
In this case, the grand-canonical free energy has no jump in its derivative and there are no phase transition when the density is varied. This is of course due to the scaling invariance of the system. Note that all the negative $\mu$'s give the same density $\rho=0$, which is obvious from the definition~\eqref{eq:def_E_s_GC} since $E_s(\Omega,n)>0$ for all $n\geq2$.

\begin{remark}[Extensivity of variance]\label{rmk:Ginibre}
In a bounded domain $\Omega$ the first two derivatives of the free energy with respect to $\mu$ equal
$$\frac{\partial}{\partial\mu}F_s^{\rm GC}(\beta,\mu,\Omega)=-\bE_{s,\beta,\mu,\Omega}[n],\qquad \frac{\partial^2}{\partial\mu^2}F_s^{\rm GC}(\beta,\mu,\Omega)=-\beta\left(\bE_{s,\beta,\mu,\Omega}[n^2]-\bE_{s,\beta,\mu,\Omega}[n]^2\right),$$
where $\bE_{s,\beta,\mu,\Omega}[\cdot]$ denotes the expectation in the grand-canonical Gibbs state and $n$ is the number of points. In other words, the second derivative is proportional to the variance of the number of points. In Ref.~\onlinecite{Ginibre-67}, Ginibre proves that this variance satisfies
\begin{equation}
 \bE_{s,\beta,\mu,\Omega}[n^2]-\bE_{s,\beta,\mu,\Omega}[n]^2\geq \frac{\bE_{s,\beta,\mu,\Omega}[n]}{1+e^{\beta\mu}\beta^{\frac{d}s}\int_{\R^d}(1-e^{-|x|^{-s}})\,\dx}.
 \label{eq:variance_Ginibre}
\end{equation}
Since $\bE_{s,\beta,\mu,\Omega}[n]\sim \rho(\mu)|\Omega|$, this provides a lower bound proportional to the volume. In other words, the strict concavity of $\mu\mapsto g(s,\beta,\mu)$ follows from the variance being an extensive quantity. This is related to the \emph{non} hyperuniformity of the Gibbs point process\cite{GhoLeb-17}, which we will discuss later in Section~\ref{sec:prop}.
\end{remark}

\begin{remark}[Grand-canonical free energy for $s=2$ in $d=1$]\label{rmk:CSM}
Using a work of Ruijsenaars\cite{Ruijsenaars-95}, Choquard proved in Ref.~\onlinecite{Choquard-00} that for the classical Calogero-Sutherland-Moser model $s=2$ in dimension $d=1$, we have the explicit formula
\begin{equation}
 g(2,\beta,\mu)=-\frac{\sqrt2}{\pi\beta^{\frac32}}\int_0^\ii \phi^{-1}\left(\sqrt{2\pi\beta} e^{\beta\mu}\,e^{-k^2}\right)\,\dk
 \label{eq:g_CSM}
\end{equation}
where $\phi(x)=xe^x$ and $\phi^{-1}$ denotes its inverse on $\R_+$. In particular, $g$ is a real-analytic function of $(\mu,\beta)$ on $\R\times(0,\ii)$. The \emph{quantum} Calogero-Sutherland-Moser model can be mapped to a non-interacting gas obeying the Haldane-Wu fractional statistics~\cite{Haldane-91,Wu-94,BerWu-94,Ha-94,Isakov-94,MurSha-94}. The function $\phi$ in~\eqref{eq:g_CSM} is what remains from this mapping in a semi-classical limit~\cite{BhaMurSen-10}. Vaninsky\cite{Vaninsky-97,Vaninsky-00} coined a formula for the canonical function $f(2,\beta,\rho)$ but there seems to exist no proof that this is the Legendre transform of~\eqref{eq:g_CSM}. In Ref.~\onlinecite{BhaMurSen-10} the first terms in the expansion of $f(2,\beta,\rho)$ at small density are derived from~\eqref{eq:g_CSM}.
\end{remark}

\subsection{Local bounds and definition of the point process}\label{sec:local_bd_short}
Now that we have recalled the definition and properties of the macroscopic thermodynamic functions, we look at the Gibbs measure itself. At $T=0$ we have to simply study the positions of the points. In order to make sure that the limit is non trivial, we need to prove \emph{local bounds}.

\bigskip

\paragraph{Zero temperature.}
We start with the zero-temperature case and prove that for minimizing positions, the points $x_j$ never get too close to each other and cannot leave too big holes. This is what is needed to pass to the limit locally and get an infinite configuration of points. It is instructive to first deal with the easier grand-canonical case.

\begin{lemma}[Grand-canonical local bounds, $T=0$]\label{lem:local_bound_GC_T0}
Let $\mu>0$ be any fixed number. Let $\Omega$ be any bounded open set with $|\partial\Omega|=0$. Let $N$ be so that
$$E_s(N,\Omega)-\mu N=\min_{n\geq0}\big\{E_s(n,\Omega)-\mu n\big\}$$
and $x_1,...,x_N$ be any minimizer for $E_s(N,\Omega)$. We assume that $N\geq2$, which is the case if for instance ${\rm diam}(\Omega)> \mu^{-\frac1s}$.

\smallskip

\noindent $(i)$ For any $1\leq j_0\leq N$, we have
\begin{equation}
 \sum_{j\neq j_0}\frac1{|x_j-x_{j_0}|^s}\leq \mu.
 \label{eq:estim_pot_GC}
\end{equation}
In particular, the smallest distance between the points satisfies $\min_{j\neq k}|x_j-x_k|\geq \mu^{-\frac1s}$.

\smallskip

\noindent $(ii)$ We have
\begin{equation}
 \sum_{j=1}^N\frac1{|x-x_{j}|^s}\geq \mu,\qquad\forall x\in\overline\Omega.
 \label{eq:estim_pot_GC_lower}
\end{equation}
This implies the existence of a universal constant $r$ (depending only on $d$ and $s$) such that any ball of radius $r\mu^{-\frac1s}$ and center $x\in\Omega$ contains at least one of the $x_j$'s.
\end{lemma}

One should interpret~\eqref{eq:estim_pot_GC} as an upper bound on the decrease of energy when we remove the point $x_{j_0}$ from the system. The other estimate~\eqref{eq:estim_pot_GC_lower} provides a lower bound on the increase of energy when we add one point to the system, at the position $x_{N+1}=x$. Understanding the variations of the energy when one point is removed or added is the key to obtain local bounds.

\begin{proof}
We have $E_s(0,\Omega)=E_s(1,\Omega)=0$ and $E_s(2,\Omega)={\rm diam}(\Omega)^{-s}$. The optimal $N$ thus satisfies $N\geq2$ if ${\rm diam}(\Omega)>\mu^{-1/s}$. The minimality of $N$ means that
\begin{equation}
 E_s(N,\Omega)\leq E_s(n,\Omega)+\mu (N-n),\qquad\forall n\geq0.
 \label{eq:minimality_N_GC}
\end{equation}
We then write
\begin{align*}
E_s(N,\Omega)&=\sum_{\substack{1\leq j<k\leq N\\ j,k\neq j_0}}\frac1{|x_j-x_k|^s}+\sum_{j\neq j_0}\frac1{|x_j-x_{j_0}|^s}\\
&\geq E_s(N-1,\Omega)+\sum_{j\neq j_0}\frac1{|x_j-x_{j_0}|^s}\geq E_s(N,\Omega)+\sum_{j\neq j_0}\frac1{|x_j-x_{j_0}|^s}-\mu
\label{eq:decomp_one_particle}
\end{align*}
where we have used~\eqref{eq:minimality_N_GC} for $n=N-1$. We obtain~\eqref{eq:estim_pot_GC}. For~\eqref{eq:estim_pot_GC_lower} we add an extra point in the position $x_{N+1}=x$ to the minimizer for $E_s(N,\Omega)$ and use it as a trial state for $E_s(N+1,\Omega)$. We obtain
$$E_s(N)+\sum_{j=1}^N\frac1{|x-x_j|^s}\geq E_s(N+1)\geq E_s(N)+\mu,$$
using again~\eqref{eq:minimality_N_GC}, this time with $n=N+1$.

To prove the statement concerning the balls of radius $r\mu^{-1/s}$, we use the following elementary lemma.

\begin{lemma}\label{lem:simple_estim_sum}
There exists a universal constant $C=C(s,d)$ such that for every $r\geq1$,
$$\sum_{j}\frac1{|y_j|^s}\leq \frac{C}{r^{s-d}}$$
for any points $y_j\in\R^d\setminus B_r(0)$ with the property that $\min_{j\neq k}|y_j-y_k|\geq1$.
\end{lemma}

We choose $r$ so that, in the lemma, $C/r^{s-d}<1$. By scaling we deduce that if there is no $x_j$ in a ball of radius $r\mu^{-1/s}$ centered at some $x$, then
$\sum_{j=1}^N|x-x_j|^{-s}< \mu$.
This is because $\min|x_j-x_k|\geq\mu^{-1/s}$ by $(i)$. When $x\in\Omega$, this contradicts~\eqref{eq:estim_pot_GC_lower} and thus shows the last part of the statement.
\end{proof}

\begin{proof}[Proof of Lemma~\ref{lem:simple_estim_sum}]
For any $y\in B_{1/2}(y_j)$, we have by the triangle inequality $|y|\leq |y_j|+1/2\leq (3/2)|y_j|$, since $|y_j|\geq r\geq1$. Thus $|y_j|^{-s}\leq (3/2)^s |B_{1/2}|^{-1}\int_{B_{1/2}(y_j)}|y|^{-s}\,\dy$. Since the balls $B_{1/2}(y_j)$ are disjoint due to the distance between the $y_j$ and all contained in $\R^d\setminus B_{r/2}(0)$, we obtain
$$\sum_j \frac1{|y_j|^s}\leq (3/2)^s|B_{\frac12}|^{-1}\int_{|y|\geq \frac{r}2}\frac{\dy}{|y|^s}=\frac{3^sd}{(s-d)r^{s-d}}.$$
\end{proof}

The previous proof in the grand-canonical case uses that $\mu$ provides a bound on the variation of the energy when we add or remove one point from the system. In the canonical case we have to first estimate this variation and then the argument is the same as before. For simplicity, we state our result at density $\rho=1$. As usual, the general case follows by scaling.

\begin{lemma}[Local bounds in the canonical case, $T=0$]\label{lem:local_bound_C_T0}
Assume that $s>d$. Let $\omega$ be any bounded open set with $|\omega|=1$ and $|\partial\omega|=0$. Let $\Omega=N^{1/d}\omega$. Then we have
\begin{equation}
E_s(N-1,\Omega)\geq E_s(N,\Omega)-\mu_1,\qquad E_s(N+1,\Omega)\geq E_s(N,\Omega)+\mu_2
\label{eq:estim_one_particle}
\end{equation}
for two universal constants $\mu_1,\mu_2>0$ and $N$ large enough. The conclusions of Lemma~\ref{lem:local_bound_GC_T0} hold for a minimizer $x_1,...,x_N$ of $E_s(N,\Omega)$, with $\mu$ replaced by $\mu_1$ for $(i)$ and by $\mu_2$ for $(ii)$.
\end{lemma}

The lemma is from Dobru\v{s}in-Minlos in 1967 (Ref.~\onlinecite[Sec.~4]{DobMin-67}), see also Georgii in Ref.~\onlinecite[Lem.~(6.2)]{Georgii-76}. Similar arguments appeared later in Ref.~\onlinecite{KuiSaf-98,BorHarSaf-08,HarSafWhi-12,HarSafVla-17}.

\begin{proof}
Assume $N\geq2$ and let $x_1,...,x_{N-1}\in\overline\Omega$. We look at the set
$A:=\Omega\setminus \bigcup_{j=1}^{N-1}B_\delta(x_j)$
obtained by removing small balls around the points. Its volume satisfies
$$|A|\geq |\Omega|-(N-1)|B_1|\delta^d\geq (N-1)\left(1-|B_1|\delta^d\right)\geq \frac{N-1}2$$
if we choose $\delta= (2|B_1|)^{-1/d}$.
Then we compute the average
$$|A|^{-1}\int_A\sum_{j=1}^{N-1}\frac{\dx}{|x-x_j|^s}\leq 2\int_{|x|\geq \delta}\frac{\dx}{|x|^s}=\frac{d(2|B_1|)^{\frac{s}d}}{s-d}=:\mu_1.$$
This proves that there exists an $x\in A\subset\Omega$ such that $\sum_{j=1}^{N-1}|x-x_j|^{-s}\leq \mu_1$. Thus we have
$E_s(N,\omega)\leq \cE_s(x_1,...,x_{N-1},x)\leq \cE_s(x_1,...,x_{N-1})+\mu_1$.
Optimizing over $(x_1,...,x_{N-1})$ we obtain $E_s(N,\Omega)\leq E_s(N-1,\Omega)+\mu_1$. In fact, this argument works for $n$ points in an arbitrary domain $\Omega$ under the sole condition that $2\leq n\leq |\Omega|+1$.

For the other bound we consider a minimizer $x_1,...,x_{N+1}$ for $E(N+1,\Omega)$ and call $\eta=\min_{j\neq k}|x_j-x_k|/2$ half of the smallest distance between the points. For $N$ large enough, we have $\eta\leq (2/|B_1|)^{1/d}$. This is because the balls $B_{\eta}(x_i)$ are disjoint and therefore
$$(N+1)|B_1|\eta^d=\left|\bigcup_{j=1}^{N+1}B_\eta(x_j)\right|\leq N\left(1+\big|\partial \omega+N^{-\frac1d}B_\eta\big|\right)=N+o(N).$$
Up to permutations we can assume that $\eta=|x_{N+1}-x_N|/2$ and then
$$E_s(N+1,\Omega)\geq E_s(N,\Omega)+\sum_{j=1}^N\frac1{|x_{N+1}-x_j|^s}\geq E_s(N,\Omega)+\frac{|B_1|^{\frac{s}d}}{2^{s\frac{d+1}{d}}}.$$
\end{proof}

In Lemma~\ref{lem:local_bound_C_T0} we have only considered the two pairs $(N-1,N)$ and $(N,N+1)$ in a domain of volume $|\Omega|=N$. There are similar bounds for any pair $(N,N+1)$ when $N$ is of the same order as $|\Omega|$.

The local bounds in Lemmas~\ref{lem:local_bound_GC_T0} and~\ref{lem:local_bound_C_T0} allow us to pass to the thermodynamic limit and get, after extraction of a subsequence, an infinite configuration of points in $\R^d$. We would also like to pass to the limit in the minimization problem solved by the points. Of course, since we end up with infinitely many points, their total energy is infinite. However we can still express their optimality by moving, adding and deleting finitely many $x_j$'s, and writing that the energy must go up. Due to the positive distance between the points and the short range nature of the potential, the energy shift is finite. This leads to the following definition.

\begin{definition}[Equilibrium configurations and Riesz point process at $T=0$]\label{def:equilibrium}
Let $s>d$ and $\mu>0$. An \emph{equilibrium configuration at chemical potential $\mu$} is an infinite collection of points $X=\{x_j\}_{j\in\N}\subset\R^d$ such that there exists $\eps>0$ with

\smallskip

\noindent$(i)$ $\min_{j\neq k}|x_j-x_k|\geq\eps \mu^{-\frac1s}$;

\smallskip

\noindent$(ii)$ any ball of radius $\eps^{-1}\mu^{-\frac1s}$ contains at least one of the $x_j$'s;

\smallskip

\noindent$(iii)$ for any bounded domain $D\subset\R^d$, we have
 \begin{multline}
  \sum_{1\leq j<k\leq n}\frac1{|y_j-y_k|^s}+\sum_{j=1}^n\sum_{\ell=N+1}^\ii\frac1{|y_j-x_\ell|^s}\\
  \geq \sum_{1\leq j<k\leq N}\frac1{|x_j-x_k|^s}+\sum_{j=1}^N\sum_{\ell=N+1}^\ii\frac1{|x_j-x_\ell|^s}+\mu(N-n)
  \label{eq:DLR_T0}
 \end{multline}
for all $n\geq0$ and all $y_1,...,y_n\in \overline D$, after relabeling the $x_j$ so that $x_1,...,x_N\in \overline D$ and $x_j\in\R^d\setminus \overline D$ for $j\geq N+1$. In other words, $x_1,...,x_N$ solve the minimization problem
\begin{equation}
\min_{\substack{n\geq0\\ y_1,...,y_n\in \overline D}}\left\{\sum_{1\leq j<k\leq n}\frac1{|y_j-y_k|^s}+\sum_{j=1}^n\sum_{\ell=N+1}^\ii\frac1{|y_j-x_\ell|^s}-\mu n\right\}.
\label{eq:boundary_T0}
\end{equation}
We call $\cX_{s,\mu}$ the set containing all such equilibrium configurations. It is invariant under translations, rotations and it is closed for the local convergence of points. We have the scaling relation $\cX_{s,\mu_2}=(\mu_2/\mu_1)^{\frac1s}\cX_{s,\mu_1}$ for any $\mu_1,\mu_2>0$.
A \emph{Riesz point process} at temperature $T=0$ and chemical potential $\mu$ is by definition a point process $\mathscr{P}$ which concentrates on $\cX_{s,\mu}$, that is, for which the above properties~$(i)$--$(iii)$ hold $\mathscr{P}$--almost surely. The convex set of such processes is denoted by $\cR_{s,\ii,\mu}$.
\end{definition}

The points $x_1,...,x_N$ could also be allowed to move outside of $D$. This does not change anything since $D$ can be arbitrarily large. The property~\eqref{eq:DLR_T0} is the zero-temperature version of the famous Dobru\v{s}in-Lanford-Ruelle (DLR)~\cite{Dobrushin-68a,Dobrushin-68b,Dobrushin-69,LanRue-69} condition which will be discussed below. We have however not found it stated anywhere in the literature.

The local minimization problem~\eqref{eq:boundary_T0} is similar to the definition of $E^{\rm GC}_s(\mu,D)$ in~\eqref{eq:def_E_s_GC}, except for the second sum in the minimum which involves the potential generated by the point $x_j$ outside of $D$. Due to the positive distance between the $x_j$'s, this potential is very small well inside $D$, by Lemma~\ref{lem:simple_estim_sum}. It is essentially only seen by the $y_j$ close to the boundary of $D$. This potential is often called a ``boundary condition'' for the points inside. One can prove that the minimum in~\eqref{eq:boundary_T0} equals $E_s^{\rm GC}(\mu,\ell D)+O(\ell^{d-\frac{s-d}{s-d+1}})$, uniformly with respect to the points outside, under the assumption that $|D|=1$ and $\partial D$ is smooth enough.\footnote{For a smooth domain (for instance satisfying~\eqref{eq:hyp_domain} below), there are of the order of $\ell^{d-1}R$ points located at a distance $R$ from the boundary. Those see a bounded potential. For the $N+o(N)$ particles inside, at a distance $\geq R$, the potential induced by the particles outside is of order $O(R^{d-s})$ by Lemma~\ref{lem:simple_estim_sum}. Thus the energy shift is of order $\ell^{d-1}R+\ell^dR^{d-s}$. After optimizing over $R$, this provides the claimed $O(\ell^{d-\frac{s-d}{s-d+1}})$.} In particular, the points $x_1,...,x_N$ almost minimize $E_s^{\rm GC}(\mu,\ell D)$. From the Legendre relations in Theorem~\ref{thm:limit_GC}, this can be used to show that the points $x_j$ have a well defined intensity and energy per unit volume
\begin{equation}
\lim_{\ell\to\ii}\frac{\#X\cap \ell D}{\ell^d}=\rho(\mu),\qquad \lim_{\ell\to\ii}\frac{\cE_s(\#X\cap \ell D)}{\ell^d}=e(s)\rho(\mu)^{1+\frac{s}d}.
\label{eq:intensity_T0}
\end{equation}
We recall that $\rho(\mu)$ is defined in~\eqref{eq:mu_rho_T0}.

The local bounds of Lemma~\ref{lem:local_bound_GC_T0} can be used to prove that any sequence of optimizers for our grand canonical problem $E_s^{\rm GC}(\mu,\ell\omega)$ converges locally to an equilibrium configuration $X=\{x_j\}_{j\in\N}$ when $\ell\to\ii$, after extraction of a subsequence.

\begin{theorem}[Convergence to equilibrium configurations]\label{thm:CV_equilibrium}
Let $s>d$ and $\mu>0$. Let $\omega$ be an open set containing the origin with $|\partial\omega|=0$. For any $\ell\geq1$, denote by $X_\ell=\{x_{1,\ell},...,x_{N_\ell,\ell}\}$ an optimizer for $E_s^{\rm GC}(\mu,\ell\omega)$. After extraction of a subsequence $\ell_n\to\ii$, $X_{\ell_n}$ converges locally to an infinite equilibrium configuration $X\in \cX_{s,\mu}$. In particular, we have $\cX_{s,\mu}\neq\emptyset$ for all $\mu>0$.
\end{theorem}

The set $\cX_{s,\mu}$ is invariant under translations and rotations, hence is always infinite. However, it seems reasonable to believe that there will only exist finitely many equilibrium configurations, up to translations and rotations:
\begin{equation}
\cX_{s,\mu}\overset{?}=\bigcup_{j=1}^J\mu^{-\frac1s}\big\{R\mathscr{L}_j+\tau\ :\ R\in SO(d),\ \tau\in\R^d\big\}.
 \label{eq:crystal_conjecture}
\end{equation}
The \emph{crystallization conjecture} (Conjecture~\ref{conj:crystal_state} below) states that the $\cL_j$ are all Bravais lattices~\cite{BlaLew-15}. In dimensions $d\in\{1,8,24\}$ we expect that $J=1$ with a universal lattice $\cL_1$ independent of $s$~\cite{Ventevogel-78,CohKumMilRadVia-19_ppt} (of course $\cL_1=\Z$ in dimension $d=1$). In other dimensions, the optimal lattice could depend on $s$, and then we expect that $J>1$ for some particular values of $s$ where two or more lattices give the same answer. This is what is believed to happen in $d=3$ with the Body-Centered Cubic (BCC) and Face-Centered Cubic (FCC) lattices. This is all discussed later in Section~\ref{sec:crystal_conjecture}

At $T=0$, a Riesz point process is just a probability measure over the equilibrium configurations in $\cX_{s,\mu}$. Such point processes have intensity $\rho=\rho(\mu)$ by~\eqref{eq:intensity_T0} and form a convex set. After averaging, we see that there exists rotation-invariant and/or translation-invariant Riesz point processes.

We have discussed here the easier grand-canonical case. A canonical equilibrium configuration $X=\{x_j\}_{j\in\N}$ is by definition one satisfying all the same properties $(i)$--$(iii)$ of Definition~\ref{def:equilibrium}, with
\begin{itemize}
 \item $\mu^{-\frac1s}$ replaced by $\rho^{-\frac1d}$ in $(i)$ and $(ii)$,
 \item only $n=N$ allowed in~\eqref{eq:DLR_T0}, hence $\mu$ can be discarded there.
\end{itemize}
Choosing $\mu$ so that $\rho(\mu)=\rho$ in~\eqref{eq:mu_rho_T0} we see that the corresponding grand-canonical configurations $\cX_{s,\mu}$ are all canonical. Since $\rho\mapsto\mu(\rho)$ is one-to-one, we expect that these are the only ones but have not found this stated anywhere in the literature. The convergence in the canonical case is studied in Ref.~\onlinecite{HarLebSafSer-18} after performing some averages over translations.

\begin{remark}[Minimizing the (free) energy per unit volume]\label{rmk:energy_per_unit_vol}
Let $X=\{x_j\}_{j\in\N}$ be any infinite configuration of points in $\R^d$. Then we can define the (upper and lower) energy per unit volume and density by
$$\overline{e}_s(X):=\limsup_{R\to\ii}\frac{\cE_s(X\cap B_R)}{|B_R|},\qquad\underline{e}_s(X):=\liminf_{R\to\ii}\frac{\cE_s(X\cap B_R)}{|B_R|},$$
$$\overline{\rho}(X):=\limsup_{R\to\ii}\frac{\#(X\cap B_R)}{|B_R|},\qquad\underline{\rho}(X):=\liminf_{R\to\ii}\frac{\#(X\cap B_R)}{|B_R|}.$$
These macroscopic quantities do not allow for a fine understanding of the optimal configurations. For instance, they do not change if $X$ is modified on a compact set. In addition, the limits might be different if we replace the ball $B_R$ by another set. Nevertheless, these concepts have proved very useful in some situations~\cite{CohKumMilRadVia-19_ppt,HarLebSafSer-18}. For instance, it follows from the definition that the minimal energy satisfies
$e(s)\rho^{1+\frac{s}d}\leq \inf \{\underline{e}_s(X)\ :\ \underline{\rho}(X)=\rho\}$.
Let $X\in\cX_{s,\mu}$ be an equilibrium configuration as in Definition~\ref{def:equilibrium}, which is known to exist by Theorem~\ref{thm:CV_equilibrium}. Then we have $\overline{e}_s(X)=\underline{e}_s(X)=e(s)\rho(\mu)^{1+s/d}$ and $\overline\rho(X)=\underline\rho(X)=\rho(\mu)$ by~\eqref{eq:intensity_T0}. Choosing $\mu$ so that $\rho(\mu)=\rho$, this proves that
\begin{equation}
e(s)\rho^{1+\frac{s}d}=\min_{\substack{X\subset\R^d\\ \underline\rho(X)=\rho}} \underline{e}_s(X).
\label{eq:var_principle_energy_per_unit_vol}
\end{equation}
There is a similar expression in the grand canonical case with $\underline{e}_s(X)-\mu\overline\rho(\mu)$ and without constraint. Thus $e(s)$ has a variational interpretation in terms of infinite configurations of points. Instead of considering individual point configurations $X$ which may not have a clear density or energy per unit volume, it is sometimes better to work with \emph{stationary (that is, translation-invariant) point processes}\cite{BorSer-13,Leble-15,Leble-16} (see Lemma~\ref{lem:periodic_energy_pt_process} below). Variational characterizations of the type of~\eqref{eq:var_principle_energy_per_unit_vol}, involving general point processes, have played an important role in works of Georgii~\cite{Georgii-94,Georgii-95,Leble-16} at positive temperature.
\end{remark}

\bigskip

\paragraph{Positive temperature.}
Next we turn to the positive temperature case. All the configurations of points are now possible and we have to control the probabilities that the points are badly placed. Instead of considering the smallest distance and the largest hole, we discuss weaker bounds which are enough to pass to the limit. We again start with the easier grand-canonical case, for which we introduce the \emph{$k$-point correlation function} of the Gibbs measure in a domain $\Omega$
\begin{equation}
\rho^{(k)}_{s,\beta,\mu,\Omega}(x_1,...,x_k):=Z^{\rm GC}_s(\beta,\mu,\Omega)^{-1}\sum_{n\geq0}\frac{e^{\beta\mu (n+k)}}{n!}\int_{\Omega^n}e^{-\beta \cE_s(x,y)}\dy.
\label{eq:correlation_fn}
\end{equation}
We have used here the short-hand notation $x=(x_1,...,x_k)$, $y=(y_1,...,y_n)$, $\dy=\dy_1\cdots\dy_n$ and $\cE_s(x,y)=\cE_s(x_1,...,x_k,y_1,...,y_n)$. The expectation of a random variable $f$ defined on finite or infinite configurations of points is given by
$$\bE_{s,\beta,\mu,\Omega}[f]:=Z^{\rm GC}_s(\beta,\mu,\Omega)^{-1}\sum_{n\geq0}\frac{e^{\beta\mu n}}{n!}\int_{\Omega^n}f(\{x_1,...,x_n\})e^{-\beta \cE_s(x)}\dx_1\cdots\dx_n.$$
Our goal will be to control the energy $\cE_s(D\cap X)$ and the number of points $n_D(X)=\# X\cap D$ in any given domain $D\subset\R^d$, independently of the large domain $\Omega$. A calculation shows that
\begin{equation}
\bE_{s,\beta,\mu,\Omega}\left[\frac{n_D!}{(n_D-k)!}\right]=\int_{D^\ell}\rho_{s,\beta,\mu,\Omega}^{(k)},\qquad \bE_{s,\beta,\mu,\Omega}\big[\cE_s(D\cap \cdot)\big]=\frac12\iint_{D^2}\frac{\rho_{s,\beta,\mu,\Omega}^{(2)}(x,y)}{|x-y|^s}\dx\,\dy.
 \label{eq:nb_correlation}
\end{equation}
The question is therefore to control the correlation functions, which is very easy for a positive interaction~\cite{Ruelle}.

\begin{lemma}[Local bounds in the grand-canonical case, $T>0$~\cite{Ruelle}]\label{lem:local_bound_GC}
Let $s>d$ and $\Omega\subset \R^d$ be any bounded domain. Then we have the universal bounds
\begin{equation}
\rho_{s,\beta,\mu,\Omega}^{(k)}(x_1,...,x_k)\leq e^{\beta\mu k}e^{-\sum_{1\leq \ell<m\leq k}\frac{\beta}{|x_\ell-x_m|^s}}\leq e^{\beta\mu k},
\label{eq:estim_correlations}
\end{equation}
\begin{equation}
e^{-\frac{\beta }{4{\rm diam}(D)^s}}\;\bE_{s,\beta,\mu,\Omega}\left[e^{\frac{\beta (n_D)^2}{4{\rm diam}(D)^s}}\right]\leq \bE_{s,\beta,\mu,\Omega}\left[e^{\beta\cE_s(D\cap \cdot)}\right]\leq \exp\left(e^{\beta\mu}|D|\right),
\label{eq:estim_local_energy}
\end{equation}
for any bounded domain $D\subset\R^d$.
\end{lemma}

The bound~\eqref{eq:estim_correlations} implies that our Gibbs measure has a uniformly bounded average local energy and number of points. The second bound~\eqref{eq:estim_local_energy} gives that the probability of having more than $\lambda$ points or an energy larger than $\lambda$ in a domain $D$ decays exponentially in $\lambda$, at a rate depending on the size of $D$.

\begin{proof}
Since the potential $V_s(x)=|x|^{-s}$ is positive, we have $\cE_s(x,y)\geq  \cE_s(x)+\cE_s(y)$. Inserting in~\eqref{eq:correlation_fn} immediately provides~\eqref{eq:estim_correlations}. For~\eqref{eq:estim_local_energy} we need to use that for any random variable $f$ with support in a domain $D$
$$\bE_{s,\beta,\mu,\Omega}[f]=\frac1{Z^{\rm GC}_s(\beta,\mu,\Omega)}\sum_{n,k\geq0}\frac{e^{\beta\mu (n+k)}}{n!\,k!}\int_{D^k\times (\Omega\setminus D)^{n-k}}\!\!f(\{x\})e^{-\beta \cE_s(x,y)}\dx\,\dy.$$
This is shown by looking at all the possible numbers of points in and outside $D$.
Taking $f(X)=\exp(\beta \cE_s(X\cap D))$ we obtain the last bound in~\eqref{eq:estim_local_energy}. For the first bound we use that for any $x=(x_1,...,x_n)\in D^n$
$$\cE_s(x)\geq \frac{n(n-1)}{2\,{\rm diam}(D)^s}\geq \frac{n^2-1}{4\,{\rm diam}(D)^s},$$
which concludes the proof.
\end{proof}

There exist similar bounds in the canonical case. The $k$-point correlation of the canonical Gibbs measure is defined by
\begin{equation}
\rho^{(k)}_{s,\beta,N,\Omega}(x):=Z_s(\beta,N,\Omega)^{-1}\frac{1}{(N-k)!}\int_{\Omega^{N-k}}e^{-\beta \cE_s(x,y)}\dy
\label{eq:correlation_fn_C}
\end{equation}
so that we have the pointwise inequality
\begin{equation}
\rho^{(k)}_{s,\beta,N,\Omega}(x)\leq \frac{Z_s(\beta,N-k,\Omega)}{Z_s(\beta,N,\Omega)} e^{-\beta\cE_s(x)}.
\label{eq:estim_correlations_C}
\end{equation}
The same proof as in Lemma~\ref{lem:local_bound_C_T0} was used by Dobru\v{s}in-Minlos (Ref.~\onlinecite[Sec.~4]{DobMin-67}, see also Georgii in Ref.~\onlinecite[Lem.~(6.2)]{Georgii-76}) to show that
$$\frac{Z_s(\beta,N-k,\Omega)}{Z_s(\beta,N,\Omega)}\leq 2^ke^{\beta\mu_1 k},$$
for fixed $k$ and $N$ large-enough (depending on $k$), in any domain satisfying $|\Omega|=N$, where $\mu_1$ is the constant from Lemma~\ref{lem:local_bound_C_T0}. This provides the desired bound on the correlation functions in the canonical case.

With these bounds at hand we can pass to the thermodynamic limit and obtain a point process satisfying similar local bounds. This is what corresponds to $(i)$ and $(ii)$ in Definition~\ref{def:equilibrium}. The positive temperature equivalent of $(iii)$ is
called the Dobru\v{s}in-Lanford-Ruelle (DLR) condition~\cite{Dobrushin-68a,Dobrushin-68b,Dobrushin-69,LanRue-69}, and states that the conditional probability of the points in a domain $D$, given the positions $x=\{x_\ell\}$ of the points outside of $D$ is the grand-canonical probability on $D$
\begin{equation}
\bP_{\beta,\mu,D,x}:= \left(\frac{e^{\beta \mu n}}{Z^{\rm GC}(D,x)}e^{-\beta  \left(\sum_{1\leq j<k\leq n}\frac1{|y_j-y_k|^s}+\sum_{j=1}^n\sum_{\ell}\frac1{|y_j-x_\ell|^s}\right)}\right)_{n\geq0}.
 \label{eq:specification}
\end{equation}
This property holds in a finite domain $\Omega$ and pertains in the limit, almost surely with respect to the outside\cite{Preston-74}. The second sum in~\eqref{eq:specification} is almost-surely finite, due to the average bounds from Lemma~\ref{lem:local_bound_GC}.

\begin{definition}[Riesz point process at $T>0$]\label{def:Gibbs_DLR}
A \emph{Riesz point process} at inverse temperature $\beta$ and chemical potential $\mu\in\R$ is a point process satisfying the same average local bounds as in~\eqref{eq:estim_correlations} and~\eqref{eq:estim_local_energy} for any bounded domain $D$ and the DLR condition that the conditional probability of the points in a domain $D$ given the positions $x_\ell$ in $\R^d\setminus D$ is almost surely given by~\eqref{eq:specification}. The convex set of such processes is denoted by $\cR_{s,\beta,\mu}$ and it is non empty.
\end{definition}

A Riesz point process does not necessarily have a well defined intensity since there can be phase transitions and thus several $\rho$'s corresponding to one $\mu$, which is different from $T=0$. One can also define a concept of canonical Gibbs state. It is proved by Georgii in Ref.~\onlinecite{Georgii-76} that those are all convex combinations of grand-canonical ones.

At $T=0$ we knew from the invariance under translations and rotations that the set $\cX_{s,\mu}$ of equilibrium points cannot be reduced to one point, hence so does the convex set $\cR_{s,\ii,\mu}$. On the other hand, at positive temperature it is perfectly possible that $\cR_{s,\beta,\mu}$ be a single point process, invariant under isometries. The question of whether this happens or not is fundamental in the understanding of phase transitions and will be discussed later in Section~\ref{sec:transitions}.

Instead of using the DLR conditional probability~\eqref{eq:specification}, one can define the Riesz point process through the Kirkwood-Salsburg (KS) equations~\cite{Ruelle-63,Ruelle}. This is an infinite hierarchy of integral equations involving the correlation functions, taking the form
\begin{multline}
\rho^{(k)}(x_1,...,x_k)=e^{-\beta \sum_{j=2}^k\frac1{|x_j-x_1|^s}}e^{\beta\mu}\times\bigg\{\rho^{(k-1)}(x_2,...,x_k)+\\
+\sum_{n=1}^\ii \frac1{n!}\int_ {(\R^d)^n}\prod_{j=1}^n\left(e^{-\frac\beta{|y_j-x_1|s}}-1\right) \rho^{(k-1+n)}(x_2,...,x_k,y_1,...,y_n)\,\dy_1\cdots \dy_n \bigg\}.
\label{eq:KS}
\end{multline}
Yet another point of view is given by the  Bogoliubov-Born-Green-Kirkwood-Yvon (BBGKY) equations, which involve gradients and take the form
\begin{multline}
\nabla_{x_1}\rho^{(k)}(x_1,...,x_k)=-\beta\sum_{j=2}^k\nabla V_s(x_1-x_j)\rho^{(k)}(x_1,...,x_k)\\
 -\beta\int_{\R^d}\nabla V_s(x_1-y)\rho^{(k+1)}(x_1,...,x_k,y)\,\dy.
\label{eq:BBGKY_short}
\end{multline}
It is proved in Ref.~\onlinecite{Dobrushin-68c,LanRue-69,Ruelle-70,GenSim-12} (see also the discussion in Ref.~\onlinecite{GruLugMar-78}) that these are completely equivalent points of view. A solution of the KS or BBGKY equations with correlation functions satisfying the bounds~\eqref{eq:estim_correlations} defines a point process which satisfies the DLR condition~\eqref{eq:specification}, and conversely. There are in fact other equivalent characterizations such as the Kubo-Martin-Schwinger (KMS) condition but they will not be discussed in this article.

\section{Thermodynamic limit in the long range case $s<d$}\label{sec:thermo_limit_long_range}

After this long description of Riesz point processes in the short range case $s>d$, we turn to the more difficult long range case $s<d$, for which many results are still open. We will not discuss the threshold $s=d$ in this article.

\subsection{Canonical ensemble}
We can start the same as in the short range case and investigate the minimal energy
\begin{equation}
 \widetilde E_s(N,\Omega)=\min_{x_1,...,x_N\in\overline\Omega}\sum_{1\leq j<k\leq N}V_s(x_j-x_k)
 \label{eq:def_E_s_stupid}
\end{equation}
where we recall that $V_s$ is given by~\eqref{eq:def_V_s_intro}. It is however not possible to construct good configurations of points this way. The points repel a lot at large distances due to the non-integrability of the potential, and do not repel that much anymore when they are close. The consequence is that they will all escape to a neighborhood of the boundary. This is sometimes called the \emph{evaporation catastrophe}\cite{Gallavotti-99}. In fact, for $-2<s\leq d-2$ the minimum is always attained for $x_1,...,x_N$ all exactly on $\partial\Omega$, and our set $\Omega$ ends up being completely empty. To prove this claim, we recall that
$$\begin{cases}
 \dps -\Delta\frac1{|x|^s}=\frac{s(d-2-s)}{|x|^{s+2}} &\text{in $\R^d\setminus\{0\}$, for $s\neq0$,}\\[0.3cm]
 \dps -\Delta(-\log|x|)=\frac{d-2}{|x|^{2}} &\text{in $\R^d\setminus\{0\}$, for $s=0$. }\\
  \end{cases}
$$
Due to the sign in~\eqref{eq:def_V_s_intro}, we find that $\cE_s$ is superharmonic with respect to each $x_{j_0}$ on $\Omega\setminus\{x_j\}_{j\neq j_0}$, hence attains its minimum at the boundary. This cannot be at one $x_j$ where the potential diverges to $+\ii$ or vanishes, and thus we must have $x_{j_0}\in\partial\Omega$. When $d-2<s<d$, there will be points everywhere but many more close to the boundary than in the interior of~$\Omega$.

Gathering at (or close to) the boundary is the best that our points can do to compensate the slow decay of the potential, but this is not sufficient to make the energy behave well. For $\omega$ any smooth bounded domain so that $|\omega|=1$, it can be proved that (for $s>0$)
\begin{align}
\widetilde{E}_s(N,N^{\frac1d}\omega)&=N^{-\frac{s}d}\min_{y_1,...,y_N\in \overline\omega}\sum_{1\leq j<k\leq N}\frac1{|y_j-y_k|^s}\nn\\
&\underset{N\to\ii}{\sim}\frac{N^{2-\frac{s}d}}2\min_{\nu}\iint_{\overline\omega\times\overline\omega}\frac{\rd\nu(x)\,\rd\nu(y)}{|x-y|^s}=\frac{N^{2-\frac{s}d}}{2\,{\rm Cap}_s(\omega)}\label{eq:mean-field}
\end{align}
where the right side is the Riesz capacity of the domain $\omega$~\cite{Landkof-72} and the second minimum is over all probability measures $\nu$ supported in $\overline\omega$. Unlike the short range case treated in Lemma~\ref{lem:simple_estim_E_s}, the energy grows much faster than $N$ in a large domain. The limit~\eqref{eq:mean-field} dates back at least to Choquet~\cite{Choquet-58} and has appeared in many forms in the literature~\cite{Landkof-72}. There is a similar limit for $s\leq0$.

Adding a fixed temperature will not help. The points will now be forced to visit the whole of $\Omega$ but this is so costly that they will very rarely do so. In fact, the exact same limit as~\eqref{eq:mean-field} holds for the free energy
$$\widetilde F_s(\beta,N,N^{\frac1d}\omega)=-\beta^{-1}\log\left(\int_{(N^{\frac1d}\omega)^N}e^{-\beta\sum_{1\leq j<k\leq N}V_s(x_j-x_k)}\dx_1\cdots \dx_N\right),$$
independently of the value of the temperature $T=1/\beta$. One should take $T\sim N^{\frac{d-s}d}$ in order to see an effect of the temperature to leading order, but this will only affect the value of the constant in~\eqref{eq:mean-field} without changing the behavior in $N^{2-\frac{s}d}$~\cite{MesSpo-82,Kiessling-89,Kiessling-93,CagLioMarPul-92,*CagLioMarPul-95,KieSpo-99}.

We have to find a way of compensating the strong repulsion between the points and prevent their escape to the boundary. On the other hand, we wish to keep nice scaling properties such as~\eqref{eq:scaling_N_T0} and~\eqref{eq:scaling_N_T}. The solution to this problem is well known in physics. The idea is to add a uniform compensating background of density $\rho_b$ which acts as a renormalization of the energy. We thus define
\begin{multline}
\cE_s(x_1,...,x_N,\Omega,\rho_b):=\sum_{1\leq j<k\leq N}V_s(x_j-x_k)-\rho_b\sum_{j=1}^N\int_\Omega V_s(x_j-y)\dy\\+\frac{\rho_b^2}2\iint_{\Omega\times\Omega}V_s(x-y)\,\dx\,\dy.
\label{eq:def_cE_s_Jellium}
\end{multline}
The second and third terms on the right side of~\eqref{eq:def_cE_s_Jellium} are respectively interpreted as the interaction between the points and the uniform background, and the self-energy of the background. The last term is a constant added for convenience, but the second term depends on the location of the points and it can drastically modify their optimal position. It seems very natural to enforce the constraint that $\rho=N/|\Omega|=\rho_b$ and we will soon do so, but for the moment we keep $\rho_b>0$ arbitrary. When $\rho_b=0$ we recover the problematic energy in~\eqref{eq:mean-field}. The corresponding minimum energy is
\begin{equation}
\boxed{E_s(N,\Omega,\rho_b):=\min_{x_1,...,x_N\in\overline\Omega}\cE_s(x_1,...,x_N,\Omega,\rho_b).}
\label{eq:def_E_s_Jellium}
\end{equation}
We use the same notation $\cE_s$ and $E_s$ for $s<d$ as in the short range case $s>d$. We hope this does not cause any confusion. The link with $s>d$ will be clarified later in Section~\ref{sec:periodic}. By convention we will always include the background when $s<d$.
The uniformity of the background provides the same scaling relation as in the short range case, with the exception of $s=0$:
\begin{equation}
E_s(N,\Omega,\rho_b)=\lambda^s\,E_s\big(N,\lambda \Omega,\rho_b\lambda^{-d}\big)+\log\lambda\frac{(N-\rho_b|\Omega|)^2-N}{2}\delta_0(s).
\label{eq:scaling_N_T0_Jellium}
\end{equation}
Our notation means that the second term is only present for $s=0$. The potential seen by any point $x_{j_0}$ in the system is now given by
$$\sum_{\substack{j=1\\j\neq j_0}}^NV_s(x_j-x_{j_0})-\rho_b\int_\Omega V_s(y-x_{j_0})\dy.$$
The hope is that the second term compensates the first if $\rho=\rho_b$ and the points are sufficiently well arranged in the domain $\Omega$.  For instance, if the $x_j$ are uniformly and independently distributed in $\Omega$, the expectation is
$$\left(\frac{N-1}{|\Omega|}-\rho_b\right)\int_\Omega V_s(y-x_{j_0})\,\dy=\left(\rho-\rho_b-\frac1{|\Omega|}\right)\int_\Omega V_s(y-x_{j_0})\,\dy.$$
The integral $\int_\Omega V_s(y-x_{j_0})\,\dy$ behaves like $|\Omega|^{1-\frac{s}{d}}$, a divergence which can be compensated by the factor $(\rho-\rho_b-|\Omega|^{-1})$ only if $\rho_b=\rho$.

The subtraction of a uniform background as in~\eqref{eq:def_cE_s_Jellium} has a long history in physics. In the Coulomb case $s=d-2$, this is often called the \emph{Jellium model}, a name which seems to have been first suggested by Herring at a conference in 1952~\cite{Herring-52,Hugues-06}. The points are interpreted as negative charges moving in a positively charged jelly. Another name which is also found in the literature is the \emph{one-component plasma}.\footnote{In the physics literature, the name `Jellium' is often employed for electrons (which are quantum with spin), whereas the `one-component plasma' is mainly used for classical particles as considered in the present article.} The model seems to have been proposed around 1900 by J.~J.~Thomson~\cite{Thomson-04} -- based on previous ideas of W.~Thomson (Lord Kelvin) -- in order to describe the electrons in an atom, before the discovery of the nuclei. In this context, it is often called the \emph{plum pudding model}. The first mathematical results go back to F\"oppl who studied it in his dissertation with Hilbert in 1912~\cite{Foppl-12}. In a celebrated work, Wigner~\cite{Wigner-34} introduced the quantum version of this model in 1934. Since the renormalization procedure~\eqref{eq:def_cE_s_Jellium} is independent of~$s$, in this paper we will use the name `Jellium' for all $s<d$.

The homogeneous background is a caricature of what charged particles usually experience in real physical systems, but the Jellium model has nevertheless been found to provide both qualitative and quantitative results in a large number of practical situations. It is the reference model in Density Functional Theory~\cite{LunMar-83,ParYan-94,PerKur-03}, where it appears in the Local Density Approximation~\cite{HohKoh-64,KohSha-65,LewLieSei-19,LewLieSei-19_ppt} and is used for deriving the most efficient empirical functionals~\cite{Perdew-91,PerWan-92,Becke-93,PerBurErn-96,SunPerRuz-15,Perdew_etal-16}. Valence electrons in alcaline metals have been found to be described by Jellium to a high precision, for instance in sodium~\cite{Huotari-etal-10} and lithium~\cite{Hiraoka-etal-20}. The Jellium model is also believed to be a good approximation to the deep interior of white dwarfs, where the density of particles is very dense, as was studied first by Salpeter in 1961~\cite{Salpeter-61} (see also Ref.~\onlinecite{BauHan-80} and Ref.~\onlinecite[Chap.~11]{VanHorn-15}). In this case the atoms are fully ionized and the nuclei evolve in a uniform background of electrons. This is all for $d=3$ and $s=d-2=1$ (Coulomb case) but there are many other applications for different values of $d$ and $s$. For instance, $s=0$ plays a role for star polymer solutions, at least at short distances\cite{WitPin-86,Likos-01}. Other values of $s$ and $d\in\{1,2,3\}$ can be artificially produced in the laboratory by tuning the interaction between cold atoms using lasers~\cite{Zhang_etal-17}.

When $\rho_b=N/|\Omega|$, it is useful to think that each point $x_j$ owns a small piece $\Omega_j$ (to be determined) of the background of volume $|\Omega_j|=1/\rho_b$. The potential generated by a point $x_{j_0}$ with its background $\Omega_{j_0}$  behaves like
\begin{equation}
\frac1{|x-x_j|^s}-\rho_b\int_{\Omega_j}\frac{\dy}{|x-y|^s}=s\frac{x\cdot P_j}{|x|^{s+2}}+\frac{s}2\frac{\,x^TQ_jx}{|x|^{s+4}}+O\left(\frac1{|x|^{s+3}}\right)_{|x|\to\ii}
 \label{eq:dipole}
\end{equation}
where
$$P_j=x_j-\rho_b\int_{\Omega_j}y\,\dy,\qquad Q_j=(s+2)x_jx_j^T-|x_j|^2-\rho_b\int_{\Omega_j}\big((s+2)yy^T-|y|^2\big)\dy$$
are the dipole and quadrupole for the Riesz interaction. A similar expansion holds for $s\leq0$. The interaction between two such compounds goes at infinity like $1/|x|^{s+2}$ and even like $1/|x|^{s+4}$ if $P_j=0$. We thus see how the background can serve to improve the decay at infinity of the potential. This generates the hope that the system will be stable for $s\geq d-2$ or maybe even $s\geq d-4$. It turns out that the situation is much better: the background stabilizes the system for all $s>0$ as well as all $-2< s\leq0$ if we impose neutrality, in all space dimensions $d$.

\begin{lemma}[Stability of Jellium~\cite{LewLieSei-19b}]\label{lem:lower_bound_Jellium}
Let $d\geq1$ and $-2<s<d$. For a universal constant $c_1>0$ (depending only on $d$ and $s$), we have
\begin{equation}
E_s(N,\Omega,\rho_b) \geq -
 \begin{cases}
\dps c_1\rho_b^{\frac{s}d}N&\text{for $s>0$, $\rho_b\geq0$,}\\[0.2cm]
\dps \left(c_1+\frac{\log\rho_b}{2d}\right)N&\text{for $s=0$, $\rho_b=\frac{N}{|\Omega|}$,}\\
0&\text{for $-2< s<0$, $\rho_b=\frac{N}{|\Omega|}$.}
 \end{cases}
 \label{eq:lower_bound_Jellium}
\end{equation}
\end{lemma}

This is the equivalent of Lemma~\ref{lem:simple_estim_E_s} in the long range case, except that for $s>0$ $\rho^{s/d}$ is replaced by $\rho_b^{s/d}$ in~\eqref{eq:lower_bound_Jellium}. Note that $E_s(N,\Omega,\rho_b)$ can now be negative for $0\leq s<d$ but becomes again positive for $s<0$.

The lemma is taken from Ref.~\onlinecite[Appendix]{LewLieSei-19b}, but a similar argument was given earlier in Ref.~\onlinecite[App.~B.2]{CotPet-19b}, only for $0<s<d$. The Coulomb case $s=d-2$ was treated much earlier by \textcite{LieNar-75} and \textcite{SarMer-76}, based on ideas of Onsager~\cite{Onsager-39}. The argument deeply relies on the fact, mentioned in the introduction, that $V_s$ has a positive Fourier transform. In the case $s\leq0$ its Fourier transform is a singular distribution but it is positive when restricted to neutral systems, hence the constraint that $\rho_b=N/|\Omega|$ in~\eqref{eq:lower_bound_Jellium}. We quickly outline the proof since it relies on an inequality which we will need later.

\begin{proof}
When $-2<s<0$, since $V_s(0)=0$ we can introduce $\nu:=\sum_{j=1}^N\delta_{x_j}-\rho_b\1_\Omega$ and write
\begin{equation}
\cE_s(x_1,...,x_N,\Omega,\rho_b)=\frac12\iint_{\Omega\times\Omega}V_s(x-y)\,\rd\nu(x)\,\rd\nu(y)=\frac{2^{d-s-1}\pi^{\frac{d}2}\Gamma\left(\frac{d-s}{2}\right)}{-\Gamma\left(\frac{s}{2}\right)}\int_{\R^d}\frac{|\widehat{\nu}(k)|^2}{|k|^{d-s}}\,\rd  k. \label{eq:energy_Fourier_Jellium}
\end{equation}
The last integral converges when $\nu(\R^d)=N-\rho_b|\Omega|=0$ and $s>-2$, since then
$$\widehat{\nu}(k)=-i (2\pi)^{-\frac{d}2}k\cdot \left(\sum_{j=1}^Nx_j-\rho_b\int_\Omega y\,\rd y\right) +O(|k|^2)_{k\to0}.$$
The right side of~\eqref{eq:energy_Fourier_Jellium} is non-negative since $\Gamma(s/2)<0$ for $-2<s<0$.
The case $s\geq0$ is more complicated, due to the singularity at the origin. The idea is to replace $V_s$ in the Jellium energy $\cE_s$ by a truncated potential $V_{s,\eps}$, which is bounded at the origin and still has a positive Fourier transform. Adding the missing diagonal term $j=k$ in the interaction between the points, one obtains that the Jellium energy with $V_{s,\eps}$ equals
\begin{equation}
 \frac{1}2\iint_{\Omega\times\Omega}V_{s,\eps}(x-y)\,\rd\nu(x)\,\rd\nu(y)-\frac{N}{2}V_{s,\eps}(0)\geq -\frac{N}{2}V_{s,\eps}(0),
 \label{eq:lower_bound_Jellium_proof}
\end{equation}
with the same $\nu$ as above. To define the regularized potential $V_{s,\eps}$ we use that for any radial function $\chi\geq0$ with $\int_{\R^d}\chi=1$, we have by scaling
\begin{equation}
\frac{1}{|x|^s}=c(s,\chi)\int_0^\ii \chi\ast\chi (x/r)\, r^{-1-s}{\rm d}r
\label{eq:FeffDeLaLlave}
\end{equation}
where $c(s,\chi)^{-1}=|\bS^{d-1}|^{-1}\iint_{\R^d\times\R^d}\chi(x)\chi(y)|x-y|^{s-d} \dx\,\dy$. This suggests to introduce the truncated potential
$$V_{s,\eps}(x):=c(s,\chi) \int_{\eps}^\ii \chi\ast\chi (x/r) \, r^{-1-s}{\rm d}r$$
which satisfies
$V_{s,\eps}(0)=\frac{c(s,\chi)\eps^{-s}}{s}\int_{\R^d}\chi^2$ and $0\leq V_{s,\eps}\leq V_s$ a.e. Replacing $V_s$ by $V_{s,\eps}$ decreases the first and third term in the Jellium energy~\eqref{eq:def_E_s_Jellium}, whereas the error in the interaction energy can be estimated by
\begin{equation}
 \sum_{j=1}^N\int_\Omega (V_s-V_{s,\eps})(x_j-y)\,\dy\leq N\int_{\R^d}\big(V_s(y)-V_{s,\eps}(y)\big)\,\dy.
 \label{eq:estim_interaction_bkgd_truncation}
\end{equation}
Since $\int_{\R^d}(V_s-V_{s,\eps})<\ii$, this is of order $N$, as claimed. The case $s=0$ is obtained by looking at the limit $s\to0^+$~\cite{LewLieSei-19}.
\end{proof}

\begin{remark}[The log gas]
In the limit $s\to0^+$, we have $V_s(x)\to1$ and thus
$$\lim_{s\to0^+}E_s(N,\Omega,\rho_b)=\frac{(N-\rho_b|\Omega|)^2}{2}-\frac{N}{2}.$$
This is of order $N$ only in the neutral case $N=\rho_b|\Omega|$. Understanding minimizers requires expanding to the next order in $s$, leading to the definition $V_0(x):=-\log|x|$ and the expansion
$$E_s(N,\Omega,\rho_b)=-\frac{N}{2}+s\,E_0(N,\Omega,\rho_b)+o(s)_{s\to0^+},\qquad \text{for $N=\rho_b|\Omega|$.}$$
We hope it does not create any confusion that $E_0$ is the derivative of $E_s$ at $s=0^+$, and not the value of the function.
\end{remark}

\begin{remark}[The threshold $s=-2$]
The case $s=-2$ is similar to $s=0$. A calculation shows that
\begin{equation*}
\lim_{s\to-2}\cE_{s}(x_1,...,x_N,\Omega,\rho_b)=\left(\rho_b|\Omega|-N\right)\left(\sum_{j=1}^N|x_j|^2-\rho_b\int_\Omega|y|^2\,\dy\right)
+\left|\sum_{j=1}^Nx_j-\rho_b\int_{\Omega}y\,\dy\right|^2.
\label{eq:energy_s_-2}
\end{equation*}
This is non-negative when $\rho_b=N/|\Omega|$. The last term vanishes when  $\sum_{j=1}^Nx_j=\rho_b\int_\Omega y\,\dy$, and this proves that $E_{s}(N,\Omega,\rho_b)\to0$ when $s\to(-2)^+$ and $N=\rho_b|\Omega|$.
The minimizers for $E_s$ converge to the configurations satisfying $\sum_{j=1}^Nx_j=\rho_b\int_\Omega y\,\dy$ which further minimize the next order term. It thus seems  appropriate to define $V_{-2}(x):=-|x|^2\log|x|$ and work with the two constraints
$$N=\rho_b|\Omega|,\qquad \sum_{j=1}^Nx_j=\rho_b\int_\Omega y\,\dy.$$
With these charger and dipole conditions, we will have $\cE_s\geq0$ for all $-4<s<-2$ if we again flip the sign of the potential and take $V_s(x)=|x|^{|s|}$. Adding more and more constraints, we can in fact go on like this and consider arbitrarily negative values of $s$. At each negative even integer we should consider $V_{-2k}(x)=\pm |x|^{2k}\log|x|$. In order not to complicate the discussion, we will restrict here our attention to $s>-2$, which is already a much larger region than what was considered in most of the literature so far.
\end{remark}

\begin{remark}[Explicit values]
After choosing an explicit $\chi$, it is possible to provide a concrete value for $c_1$. In the Coulomb case $s=d-2$, one can in fact obtain surprisingly good bounds following a different method due to Lieb and Narnhofer~\cite{LieNar-75}. The idea is to replace $V_s$ by $V_s\ast\nu$ for a radial function $\nu\geq0$ with $\int_{\R^d}\nu=1$. Newton's theorem implies $V_s\ast\nu\leq V_s$, and the interaction with the background is estimated as in~\eqref{eq:estim_interaction_bkgd_truncation}. After optimizing over $\nu$ (the optimum is the uniform measure of a certain ball), this provides the constant~\cite{LieNar-75,SarMer-76}
$$c_1=\begin{cases}
\dps -\frac1{12}&\text{for $s=-1$ in $d=1$,}\\[0.3cm]
\dps \frac38+\frac14\log\pi&\text{for $s=0$ in $d=2$,}\\[0.3cm]
\dps \frac{d^2}{2(d+2)}\left(\frac{2\pi^{\frac{d}2}}{d\Gamma(d/2)}\right)^{1-\frac2d}&\text{for $s=d-2$ in $d\geq3$.}
      \end{cases}
$$
In $d=3$ one obtains $c_1=(3/5)(9\pi/2)^{1/3}\simeq 1.4508$ which is surprisingly close to the conjectured best value $-\zeta_{\rm BCC}(1)\simeq 1.4442$ (see Section~\ref{sec:crystal_conjecture}). In dimension $d=2$, we have $c_1\simeq 0.6612$, which is also remarkably close to the expected best constant $-\zeta'_{\rm trg}(0)\simeq  0.6606$ (see Ref.~\onlinecite[Prop.~B.1]{Lauritsen-21}). In dimension $d=1$, the constant is optimal since we will see in Theorem~\ref{thm:crystallization} that $\lim_{N\to\ii}E_{-1}(N,[0,N],1)/N=-\zeta(-1)=1/12$ where $\zeta$ is the Riemann Zeta function.

The constant $c_1$ obtained by this method has a simple physical interpretation. For $N=|\Omega|$ and $\rho_b=1$, the opposite of $c_1$ is exactly the Jellium energy of a system composed of points sufficiently far away and a background $\Omega$ equal to the union of the balls of volume $1$ centered at the points. The lower bound $-c_1$ implies that this is the minimizer if we would optimize both over the positions of the points and the shape of the background $\Omega$. In chemistry, this is called the ``point charge plus continuum approximation'' of Jellium~\cite{SeiPer-94}.
\end{remark}

When $s\leq 0$, non-neutral systems have a very large \emph{negative} energy, hence the necessity of the neutrality constraint. For instance, taking $N$ points uniformly distributed over $\Omega$, we obtain the simple upper bound
$$E_s(N,\Omega,\rho_b)\leq \left(\frac12\left(\frac{N}{|\Omega|}-\rho_b\right)^2-\frac{N}{2|\Omega|^2}\right)\iint_{\Omega\times\Omega}V_s(x-y)\dx\,\dy.$$
If $N/|\Omega|\to\rho\neq \rho_b$, the right side behaves like ${\rm sgn}(s)|\Omega|^{2-s/d}$ for $s\neq0$ and $-|\Omega|^2\log|\Omega|$ for $s=0$. The constraint that $\rho_b=N/|\Omega|$ is therefore necessary to have a lower bound of order $N$ when $s\leq0$. This will force us to work canonically. There will be no well defined grand-canonical model for $s\leq0$.

When $s>0$ non-neutral systems will also behave badly, but they have a large \emph{positive} energy. It is therefore expected that neutrality will automatically arise and should not be a necessary assumption. The grand-canonical problem will thus be perfectly well defined when $s>0$, but it will essentially be trivial. Changing the value of the chemical potential $\mu$ will not affect the final density which will always be $\rho=\rho_b$. This follows from the following corollary of the proof of Lemma~\ref{lem:lower_bound_Jellium}.

\begin{lemma}[Simple estimate on the total charge]\label{lem:neutral_Jellium}
Assume that $0<s<d$. We have for all $N\geq1$
\begin{equation}
E_s(N,\Omega,\rho_b) \geq \frac{c_1'(N-\rho_b|\Omega|)^2}{|\Omega|^{\frac{s}{d}}}-c_1\rho_b^{\frac{s}{d}}N.
\label{eq:estim_charge}
\end{equation}
for a positive constant $c_1'$ depending only on the shape of $\Omega$, that is, on the set $\omega=|\Omega|^{-\frac1d}\Omega$.
\end{lemma}

There are similar estimates on the higher moments. For instance the same proof gives for the dipole
$$\cE_s(x_1,...,x_N,\Omega,\rho_b)\geq \frac{c_1''}{|\Omega|^{\frac{s+2}{d}}}\left|\sum_{j=1}^Nx_j-\rho_b\int_\Omega y\,\dy\right|^2-c_1\rho_b^{\frac{s}{d}}N.$$
These estimates become better when $s$ decreases, which is a simple manifestation of the long range of $V_s$.

When $E_s(N,\Omega,\rho_b)$ is of order $|\Omega|$ we conclude that
$|{N}/{|\Omega|}-\rho_b|\lesssim |\Omega|^{-\frac{d-s}{2d}}\to0$
which is the neutrality mentioned before. In particular, in the Coulomb case $s=d-2$ we find that $N-\rho_b|\Omega|$ must be at most a surface term $|\Omega|^{\frac{d-1}{d}}$. This is optimal, since it is well known that a charge on the surface can generate a constant potential inside, which thus shifts the energy by a constant times $N$. In fact, for $\Omega=(N/\rho_b)^{1/d}\omega$, the difference $(E_s(N+Q,\Omega,\rho_b)-E_s(N,\Omega,\rho_b))/N$ converges to a positive limit depending only on the $s$-capacity of $\omega$, whenever $Q\sim q N^{\frac{s+d}{2d}}$ with $q\neq0$. This is proved in the Coulomb case $s=d-2$ in Ref.~\onlinecite[Sec.~3.5]{LieNar-75} for balls (see also Ref.~\onlinecite[Sec.~VI]{LieLeb-72} and Ref.~\onlinecite{GraSch-95b}). We provide the proof for $d-2<s<d$ in Section~\ref{sec:net_charge} below.

\begin{proof}
By scaling we can assume that $\Omega=N^{\frac1d}\omega$ with $|\omega|=1$, so that $\rho=1$. From the estimate in~\eqref{eq:lower_bound_Jellium_proof} we know that
\begin{align}
\cE_s(x_1,...,x_N,\Omega,\rho_b)&\geq  \frac{1}2\iint_{\Omega\times\Omega}V_{s,\eps}(x-y)\,\rd\nu(x)\,\rd\nu(y)-c_1N\rho_b^{\frac{s}{d}} \nn\\
&=\frac{c(s,\chi)}2 \int_\eps^\ii \left(\int_{\R^d}(\chi_r\ast\nu)^2\right)r^{d-1-s}{\rm d}r-c_1N\rho_b^{\frac{s}{d}},\label{eq:estim_chi}
\end{align}
with $\nu=\sum_{j=1}^N\delta_{x_j}-\rho_b\1_\Omega$ and $\chi_r(x)=r^{-d}\chi(x/r)$. We assume for simplicity that $\chi$ has its support in the unit ball $B_1$.
By Cauchy-Schwarz we have
$$N^2\big(1-\rho_b\big)^2=\left(\int_{\R^d}\chi_r\ast \nu\right)^2\leq |\Omega_r|\int_{\R^d}(\chi_r\ast \nu)^2$$
where $\Omega_r=\Omega+B_{r}=N^{1/d}(\omega+B_{N^{-1/d}r})$ contains the support of $\chi_r\ast \nu$. Using that $|\omega+B_{R}|\leq C(1+R^d)$ for any $R>0$, where $C$ depends on $\omega$, we obtain $|\Omega_r|\leq C(N+r^d)$ and thus
\begin{align*}
\int_\eps^\ii\!\! \int_{\R^d}(\chi_r\ast\nu)^2\,r^{d-1-s}{\rm d}r&\geq \int_{N^{\frac1d}}^\ii \int_{\R^d}(\chi_r\ast\nu)^2\,r^{d-1-s}{\rm d}r\\
&\geq \frac{N^2(1-\rho_b)^2}{C}\int_{N^{\frac1d}}^\ii \frac{r^{d-1-s}}{N+r^d}{\rm d}r=N^{2-\frac{s}{d}}\frac{(1-\rho_b)^2}{C}\int_{1}^\ii \frac{r^{d-1-s}}{1+r^d}{\rm d}r.
\end{align*}
\end{proof}

It remains to discuss upper bounds on $E_s(N,\Omega,\rho_b)$ in the neutral case and show that it is of order $N$, as expected. The situation is more complicated than in the short range case since we really have to place the points everywhere at the correct density, in order to appropriately screen the background. In fact, the proof gets more involved when $s$ decreases and more stringent conditions are needed on the shape of the domain $\Omega$ for $s\leq0$. The main tool for proving upper bounds is the following simple inequality (called the ``method of cells'' in Ref.~\onlinecite{SarMer-76}).

\begin{lemma}[General upper bound]\label{lem:general_upper_bound_Jellium}
Assume that $-2< s<d$. Let $\Omega\subset\R^d$, $N\in\N$ and $\rho_b>0$ be such that $N=\rho_b|\Omega|$. Suppose that $\Omega$ is the disjoint union of some measurable sets $\Omega_j$ with $N_j:=\rho_b|\Omega_j|\in\N$. Then we have
\begin{equation}
 E_s(N,\Omega,\rho_b)\leq -\sum_{j}\frac{\rho_b^2}{2N_j}\iint_{\Omega_j\times\Omega_j}V_s(x-y)\,\dx\,\dy.
 \label{eq:general_upper_bound_Jellium}
\end{equation}
\end{lemma}

\begin{proof}
Place $N_j$ points uniformly distributed in each $\Omega_j$, that is, use the trial measure $\bP=\text{Sym}\big(\otimes_{j=1}^N(\rho_b\1_{\Omega_j}/N_j)^{\otimes N_j})$ where
\begin{equation}
\text{Sym}(\bQ)(x_1,...,x_N):=\sum_{\sigma\in\gS_N}\bQ(x_{\sigma(1)},...,x_{\sigma(N)})
\label{eq:symmetrization}
\end{equation}
denotes the symmetrization of a probability measure $\bQ$.
\end{proof}

For $N_j=1$, the inequality tells us that any domain $\Omega$ which can be partitioned into $N$ sets $\Omega_j$, all of the same volume $1/\rho_b$ and with a uniformly bounded diameter must have an energy of order $N$. This works for all $s\in(-2,d)$ and includes all sufficiently smooth connected domains, as shown recently by~\textcite{GigLeo-17}. We next explain how to use the lemma with fewer assumptions on $\Omega$, at the expense of restricting the considered values of $s$ to $[-1,d)$.

\begin{lemma}[Upper bound for $s\geq-1$]\label{lem:upper_bound_Jellium}
Let $N\in\N,\rho_b>0$ and $\Omega$ be a bounded domain so that $N=\rho_b|\Omega|$ and $|\partial\Omega|=0$. For $N\geq N_0$ large enough, we have
$$E_s(N,\Omega,\rho_b)\leq c_2 N\begin{cases}
-\rho_b^{\frac{s}{d}}&\text{for $s>0$,}\\
1-\frac{\log\rho_b}{2d}&\text{for $s=0$,}\\
\rho_b^{\frac{s}{d}}&\text{for $-1\leq s<0$.}
\end{cases}
$$
Here $c_2$ only depends on $s,d$, as well as on $\omega=|\Omega|^{-1/d}\Omega$ when $s=-1$, whereas $N_0$ only depends on $s,d,\omega$. If $s\in[-1,0]$ we need to further assume that $\omega$ has a regular boundary in the sense that
\begin{equation}
 |\partial \omega+B_r|\leq Cr,\qquad\forall r\leq r_0.
 \label{eq:hyp_domain}
\end{equation}
\end{lemma}

\begin{proof}
By scaling we can assume that $\rho_b=1$ and $\Omega=N^{1/d}\omega$.  The idea is the same as in the proof of Lemma~\ref{lem:simple_estim_E_s}, except that we cannot use lattices with densities different from 1, which would badly screen the uniform background. We thus consider the tiling of $\R^d$ made of cubes of side length one and take for $\Omega_j$ all the cubes $C_j$ contained in $\Omega$. Let $K$ be the number of those cubes, which satisfies
$N\big(1-|\partial\omega+N^{-\frac1d}C|\big)\leq K\leq N\big(1+|\partial\omega+N^{-\frac1d}C|\big)$
by the arguments in Lemma~\ref{lem:simple_estim_E_s}. We call $\Omega_{K+1}=\Omega\setminus\bigcup_{C_j\subset\Omega} C_j$ the missing part, which has the volume
$|\Omega_{K+1}|=N-K\leq N|\partial\omega+N^{-\frac1d}C|=o(N)$.
Applying Lemma~\ref{lem:general_upper_bound_Jellium} gives
\begin{equation*}
E_s(N,\Omega,\rho_b)\leq -\frac{K}{2}\iint_{C\times C}V_s(x-y)\,\dx\,\dy-\frac{1}{2|\Omega_{K+1}|}\iint_{\Omega_{K+1}\times\Omega_{K+1}}V_s(x-y)\,\dx\,\dy.
\end{equation*}
When $s>0$ we can just discard the last term which is negative and obtain the desired upper bound $E_s(N,\Omega,\rho_b)\leq -c_2N$ with for instance $c_2=(1/4)\iint_{C\times C}V_s(x-y)\dx\,\dy$. When $s\leq 0$ the last term is (or can be at $s=0$) positive and must be estimated. Using that $\Omega_{K+1}$ has diameter of order $N^{1/d}$, the last term is of order $N^{\frac{|s|}{d}}|\Omega_{K+1}|$ for $s<0$ and $|\Omega_{K+1}|\log N$ for $s=0$. It is hard to go further without more assumptions on $\omega$. Under the assumption~\eqref{eq:hyp_domain}, we obtain $|\Omega_{K+1}|\lesssim N^{\frac{d-1}{d}}$ and thus our error term is a $O(N^{\frac{|s|+d-1}{d}})$ (resp. $O(N^{\frac{d-1}{d}}\log N)$ for $s=0$). For $s\in(-1,0]$, this is a $o(N)$ and we can take the same $c_2:=\iint_{C\times C}|V_s(x-y)|\dx\,\dy$. For $s=-1$ we take $c_2=\iint_{C\times C}|x-y|\dx\,\dy+{\rm diam}(\omega)$.
\end{proof}

We have seen that for sufficiently well behaved domains $\Omega$ (depending on~$s$), the energy is of order $N$ in the neutral case.  We next turn to the positive temperature case and define
$$\boxed{F_s(\beta,N,\Omega,\rho_b):=-\beta^{-1}\log Z_s(\beta,N,\Omega,\rho_b),}$$
with
\begin{equation}
Z_s(\beta,N,\Omega,\rho_b):=\frac1{N!}\int_{\Omega^N}e^{-\beta\cE_s(x_1,...,x_N,\Omega,\rho_b)}\dx_1\cdots dx_N.
\label{eq:Z_Jellium}
\end{equation}
The corresponding Gibbs measure is
$$\bP_{s,\beta,N,\Omega,\rho_b}:=Z(\beta,N,\Omega,\rho_b)^{-1}e^{-\beta\cE_s(x_1,...,x_N,\Omega,\rho_b)}.$$
The scaling relation reads
\begin{equation}
 F_s(\beta,N,\Omega,\rho_b)=\lambda^{s}\,F_s\Big(\beta\lambda^s,N,\lambda\,\Omega,\lambda^{-d}\rho_b\Big)+\beta^{-1}N\log(\lambda^d)-\frac{N}{2}\delta_0(s)\log\lambda.
 \label{eq:scaling_N_T_Jellium}
\end{equation}
We can obtain a lower bound on $F_s(\beta,N,\Omega,\rho_b)$  using the same estimate as in~\eqref{eq:estim_lower_simple_temp} and Lemma~\ref{lem:lower_bound_Jellium}. As for upper bounds, the estimate~\eqref{eq:general_upper_bound_Jellium} becomes
\begin{align*}
F_s(\beta,N,\Omega,\rho_b)&\leq -\sum_{j}\frac{\rho_b^2}{2N_j}\iint_{\Omega_j\times\Omega_j}V_s(x-y)\,\dx\,\dy+\beta^{-1}\sum_j\log\left(\frac{N_j!}{|\Omega_j|^{N_j}}\right)\\
&\leq -\sum_{j}\frac{\rho_b^2}{2N_j}\iint_{\Omega_j\times\Omega_j}V_s(x-y)\,\dx\,\dy+N\beta^{-1}\log\rho_b
 \label{eq:general_upper_bound_Jellium_T}
\end{align*}
and the argument is the same as when $T=0$.

The previous estimates suggest that the thermodynamic limit could exist for all $s>-2$, for sufficiently smooth sequences of domains. Unfortunately, this has not been proved in full generality, to our knowledge. The best result known so far seems to be the following.

\begin{theorem}[Canonical thermodynamic functions, long range case]\label{thm:limit_C_Jellium}
Let $d\geq1$, $\max(0,d-2)\leq s<d$. We also allow $s=-1$ if $d=1$. Let $\omega$ be any smooth bounded open set with $|\omega|=1$. Then, for any $\rho=\rho_b>0$ and $\beta>0$, the following limits
\begin{equation}
\lim_{\substack{N\to\ii\\ \ell^d=N/\rho}}\frac{E_s(N,\ell\omega,\rho)}{\ell^d}=e(s)\rho^{1+\frac{s}d}-\delta_0(s)\frac{\rho\log\rho}{2d},\qquad
\lim_{\substack{N\to\ii\\ \ell^d=N/\rho}}\frac{F_s(\beta,N,\ell\omega,\rho)}{\ell^d}=f(s,\beta,\rho)
\label{eq:thermo_limit_free_energy_Jellium}
\end{equation}
exist and are independent of $\omega$. The function $f$ satisfies the relation
\begin{equation}
 f(s,\beta,\rho)=\rho^{1+\frac{s}d}f\big(s,\beta\rho^{\frac{s}d},1\big)+\left(\beta^{-1}-\frac{\delta_0(s)}{2d}\right)\rho\log\rho.
 \label{eq:relation_f_Jellium}
\end{equation}
\end{theorem}

We use again the same notation $e(s)$ and $f(s,\beta,\rho)$ as for the short range case $s>d$. The link between the two cases will be discussed in Section~\ref{sec:periodic}.

The theorem was proved by Lieb-Narnhofer\cite{LieNar-75} for $s=1$ in dimension $d=3$, based on earlier work by Lieb-Lebowitz\cite{LieLeb-72} (see also Ref.~\onlinecite{LieSei-09}), and was extended to $s=d-2$ in all dimensions by Sari-Merlini\cite{SarMer-76}. Those proofs make use of Newton's theorem, which is however specific to the Coulomb case. For $s=-1$ in $d=1$ the result is due to~\textcite{Kunz-74}. The proof for the other values of $s$ is due to Serfaty \emph{et al} in a long series of works~\cite{SanSer-12,SanSer-14a,SanSer-15,BorSer-13,PetSer-17,RouSer-16,RotSer-15,Leble-15,LebSer-17,Serfaty-19,ArmSer-21}. There the model is written in an external potential and in terms on the electric field instead of the charge densities. The connection to our definition of Jellium is detailed at zero temperature in Ref.~\onlinecite[Lemma~2.6]{CotPet-19b} and in Section~\ref{sec:confined} below (see, in particular, Remark~\ref{rmk:confined_applies}).
We provide a different proof of Theorem~\ref{thm:limit_C_Jellium} for $s>\max(0,d-2)$ in Section~\ref{sec:net_charge} below. In the Coulomb case $s=d-2\geq0$, Armstrong and Serfaty have provided in Ref.~\onlinecite{ArmSer-21} a quantitative bound of the order $O(\ell^{-1})$ for the two limits in~\eqref{eq:thermo_limit_free_energy_Jellium}.

We insist that the limits~\eqref{eq:thermo_limit_free_energy_Jellium} concern the neutral case where the volume $\ell^d$ is exactly equal to $N/\rho=N/\rho_b$. The result does \emph{not} hold under the sole condition that $N\ell^{-d}\to\rho=\rho_b$ as was the case in the short range case in Theorem~\ref{thm:limit_C}. In fact, one can show that
\begin{align}
 \lim_{\substack{N,\ell\to\ii\\ \ell^{-\frac{d+s}{2}}(N-\rho\ell^d)\to q}}\frac{E_s(N,\ell\omega,\rho)}{\ell^d}&=\left(e(s)+\frac{q^2}{2\,{\rm Cap}_s(\omega)}\right)\rho^{1+\frac{s}d},\nn\\
 \lim_{\substack{N,\ell\to\ii\\ \ell^{-\frac{d+s}{2}}(N-\rho\ell^d)\to q}}\frac{F_s(\beta,N,\ell\omega,\rho)}{\ell^d}&=f(s,\beta,\rho)+\frac{q^2}{2\,{\rm Cap}_s(\omega)}\rho^{1+\frac{s}d}
 \label{eq:net_charge}
\end{align}
for any fixed $q\in[0,+\ii]$, where ${\rm Cap}_s(\omega)$ is the Riesz $s$-capacity defined in~\eqref{eq:mean-field}. This was proved for balls in the Coulomb case $s=d-2$ in Ref.~\onlinecite{Kunz-74,LieNar-75}. We extend this result to all $\max(0,d-2)<s<d$ in Section~\ref{sec:net_charge} (see Corollary~\ref{cor:net_charge}). In particular, we see that we need $\ell^d=N/\rho+o(N^{\frac{d+s}{2d}})$ to obtain the limit to $e(s)\rho^{1+s/d}$ or $f(s,\beta,\rho)$ in~\eqref{eq:thermo_limit_free_energy_Jellium}.

Note that for the log gas $s=0$ the scaling relation~\eqref{eq:relation_f_Jellium} gives
$$ f(0,\beta,\rho)=f\big(0,\beta,1\big)\,\rho+\frac{2d-\beta}{2\beta d}\rho\log\rho.$$
This is a convex function of $\rho$ for $\beta<2d$ but a concave function for $\beta>2d$. It is exactly linear at $\beta=2d$. In fact, $\frac{2d-\beta}{2\beta d}\rho$ is equal to the pressure, which is thus negative for $\beta>2d$~\cite{SalPra-63}.
Using a circular background the free energy could be explicitly computed at the special point $\beta=2$ in dimension $d=2$ in Ref.~\onlinecite{DeuDewFur-79,AlaJan-81,Forrester-98}, leading to
\begin{equation}
 f(0,2,\rho)=-\frac{\log(2\pi^2)}{4}\rho+\frac{\rho\log\rho}{4}\qquad \text{for $s=0$ and $\beta=2$ in dimension $d=2$}.
 \label{eq:formula_f_2D_beta2}
\end{equation}
In dimension $d=1$ it is in fact possible to compute $f(0,\beta,\rho)$ for all values of $\beta>0$ and $\rho>0$, by first periodizing the system. The formula is provided later in~\eqref{eq:f_log_gas} in Remark~\ref{rmk:log_gas}. This is again an integrable system, related to the quantum Calogero-Sutherland-Moser model (see Section~\ref{sec:random_matrices} below). In 1D, the point $\beta=2$ at which the free energy stops to be convex is a BKT phase transition, see Section~\ref{sec:sine-beta}.

\subsection{Grand-canonical ensemble}
We now turn to the grand-canonical case. Since the energy behaves badly for $s\leq0$ in the non-neutral case, we have to assume $s>0$. We define
\begin{equation}
\boxed{E_s^{\rm GC}(\mu,\Omega,\rho_b):=\min_{n\geq0}\left\{E_s(n,\Omega,\rho_b)-\mu n\right\},}
\label{eq:def_E_s_GC_Jellium}
\end{equation}
as well as
\begin{equation}
\boxed{F_s^{\rm GC}(\beta,\mu,\Omega,\rho_b):=-\beta^{-1}\log Z_s^{\rm GC}(\beta,\mu,\Omega,\rho_b),}
\end{equation}
with
\begin{equation}
Z_s^{\rm GC}(\beta,\mu,\Omega,\rho_b):=1+e^{\beta\mu}+\sum_{n=2}^\ii e^{\beta\mu n} Z_s(\beta,n,\Omega,\rho_b).
\label{def:Z_GC}
\end{equation}
As we have explained, the chemical potential $\mu$ will have no effect on the bulk density of the system and might only affect the density close to the boundary. Let us introduce $N:=\rho_b|\Omega|$ and assume for simplicity that this is an integer. Then we have $E_s^{\rm GC}(\mu,\Omega,\rho_b)\leq E_s(N,\Omega,\rho_b)-\mu N\leq (c_2-\mu)N\rho_b^{s/d}$ by Lemma~\ref{lem:upper_bound_Jellium}. On the other hand, denoting by $N_\mu$ an integer so that $E_s(N_\mu,\Omega,\rho_b)-\mu N_\mu=E^{\rm GC}_s(\mu,\Omega,\rho_b)$, we obtain by Lemma~\ref{lem:neutral_Jellium}
$$(N_\mu-N)^2\leq (c_1')^{-1}|\Omega|^{\frac{s}d}\rho_b^{\frac{s}d} (c_2+c_1-\mu)N=o(N^2).$$
This proves the claimed charge neutrality in the limit. There is a similar argument at positive temperature.

In spite of its triviality, the grand-canonical problem is still a useful model for $s>0$. For instance, the proof of the existence of the thermodynamic limit is simplified by the fact that the number of points can be taken arbitrary and the result is known for all $s>0$.

\begin{theorem}[Grand-canonical thermodynamic functions, long range case]\label{thm:limit_GC_Jellium}
Assume that $0<s<d$. Let $\omega$ be any bounded open set with $|\omega|=1$ and $\partial\omega$ satisfying~\eqref{eq:hyp_domain}. Then for every $\beta>0$, $\mu\in\R$ and $\rho_b>0$, the following limit exist
\begin{equation}
\lim_{\ell\to\ii}\frac{E_s^{\rm GC}(\mu,\ell\omega,\rho_b)}{\ell^d}=e^{\rm GC}(s)\rho_b^{1+\frac{s}d}-\mu\,\rho_b,\quad
\lim_{\ell\to\ii}\frac{F_s^{\rm GC}(\beta,\mu,\ell\omega,\rho_b)}{\ell^d}
=f^{\rm GC}(s,\beta,\rho_b)-\mu\,\rho_b.
\label{eq:thermo_limit_GC_Jellium}
\end{equation}
The function $f^{\rm GC}$ satisfies the same scaling relation as in~\eqref{eq:relation_f_Jellium}. If $s\geq d-2$, then $e^{\rm GC}(s)=e(s)$ and $f^{\rm GC}(s,\beta,\rho_b)=f(s,\beta,\rho_b)$, the canonical functions from Theorem~\ref{thm:limit_C_Jellium}.
\end{theorem}

The existence of the limit as well as the equality with the canonical problem is proved in the Coulomb case in Ref.~\onlinecite{LieNar-75,SarMer-76}. A different proof of the existence of the limit can be provided in dimension $d=3$ by following the method introduced in Ref.~\onlinecite{HaiLewSol_1-09,*HaiLewSol_2-09} based on the Graf-Schenker inequality~\cite{GraSch-95}. For other values of $s\in(0,d)$, one can use a similar inequality due to Fefferman~\cite{Fefferman-85,Hughes-85,Gregg-89}, as recently shown in Ref.~\onlinecite{CotPet-19,CotPet-19b} at $T=0$. On the contrary to the usual method in the short range case (outlined after Theorem~\ref{thm:limit_C}), the Graf-Schenker and Fefferman approaches are based on establishing a \emph{lower bound} on the free energy in a large domain in terms of the one in smaller domains (Figure~\ref{fig:GS}). Since the local number of points typically fluctuates, this approach is well suited to the grand-canonical ensemble. 

For $d-2<s<d$, we have not found the equality with the canonical problem stated in the literature and we provide a proof in Section~\ref{sec:net_charge} below. In fact, we can even prove the existence of the thermodynamic limit in the canonical case (Theorem~\ref{thm:limit_C_Jellium}) using Theorem~\ref{thm:limit_GC_Jellium} (see Corollary~\ref{cor:equivalence_ensembles} in Section~\ref{sec:net_charge} below).

Long range potentials can sometimes lead to non-equivalent ensembles in statistical mechanics~\cite{DauRufAriWil-02,CamDauRuf-09,CamDauFanRuf-14}. This is already the case for $s\leq0$, where the grand canonical problem is unbounded whereas the canonical problem is finite. The non-equivalence for $s\leq0$ is also manifest in the non-convexity of the (free) energy as a function of $\rho$ (when it exists). We do not quite know what to expect for all $s<d-2$ in dimension $d\geq3$ but at least conjecture the equivalence of the canonical and grand-canonical ensembles for $s\geq d-4$. In the physical dimensions $d\leq 3$ this would already cover all the possible exponents $s>0$.

Due to the necessary neutrality of the system, the free energy ends up being exactly linear in the chemical potential $\mu$. Recall that it is strictly concave in the short range case (Theorem~\ref{thm:limit_GC}). From Remark~\ref{rmk:Ginibre} this suggests that the variance of the number of points should be a $o(|\Omega|)$, which is related to the hyperuniformity of the point process discussed later in Section~\ref{sec:prop}.

\begin{figure}[t]
\includegraphics[width=7cm]{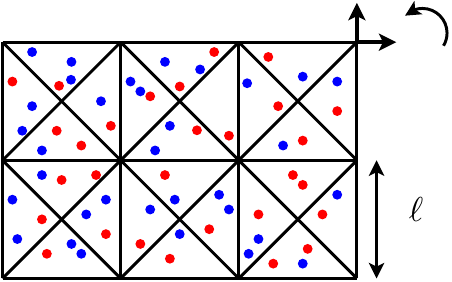}
\caption{Consider positive and negative charges $e_j\in\{\pm1\}$ located at some $x_j\in\R^3$ and represented here by the colored dots. The total Coulomb interaction is $E_{\rm tot}=\sum_{1\leq j<k\leq N}e_je_k/|x_j-x_k|$. Take then a tiling $(T_\alpha)$ of $\R^3$, made exclusively of congruent tetrahedra of side length $\ell$. The Graf-Schenker inequality\cite{GraSch-95} states that there exists an appropriate translation and rotation of the tiling (depending on the $x_j$) for which $E_{\rm tot}\geq \sum_\alpha E_\alpha-CN/\ell$ where $E_\alpha$ is the Coulomb interaction of the points inside the tetrahedron $T_\alpha$ and $N$ is the total number of points. In other words, we can neglect the interaction between the tetrahedra in a lower bound, up to a small error. The interpretation is that there is always some kind of screening as far as lower bounds are concerned. In our case one of the two charge densities is uniform and not composed of points, but the bound holds the same. The Graf-Schenker inequality was used in Ref.~\onlinecite{GraSch-95,HaiLewSol_1-09,*HaiLewSol_2-09,LewLieSei-18,LewLieSei-19} to prove that the thermodynamic limit exists for 3D Coulomb systems. A similar but more complicated bound was introduced before by Conlon, Lieb and Yau~\cite{ConLieYau-89}, involving cubes and a replacement of Coulomb by Yukawa. For other values of $s>0$, Fefferman had developed earlier in Ref.~\onlinecite{Fefferman-85} a similar inequality but each tiling domain had to be written as the union of balls of many different sizes. The advantage of the latter approach is that it can be generalized to all $s>0$ in all dimensions $d\geq1$\cite{Hughes-85,Gregg-89,CotPet-19}. It is an interesting question to find a generalization of the simpler Graf-Schenker inequality to other dimensions and potentials. \label{fig:GS}}
\end{figure}

\subsection{Local bounds and definition of the point process}
We turn to local bounds and the definition of the Riesz point process. This is a very active subject at the moment and few results have been established so far. Even if local bounds are important, they will in general not be enough to guarantee that the potential
$$\Phi_\Omega(x)=\sum_{x_j\in\Omega}V_s(x-x_j)-\rho_b\int_\Omega V_s(x-y)\,\dy$$
admits a well defined limit $\Phi(x)$ almost-surely when $\Omega\nearrow\R^d$. This is needed to pass to the limit in
the equilibrium equations such as DLR. The convergence of $\Phi_\Omega$ should hold because the points are sufficiently well positioned so as to screen the background, but this is difficult to prove in general.

\bigskip

\paragraph{Local bounds.}
At zero temperature, local bounds are known. This is an unpublished result of Lieb in the Coulomb case, which has been used and generalized to all $d-2 \leq s<d$ in Ref.~\onlinecite{PetSer-17,RouSer-16,RotSer-15,LebSer-17,LieRouYng-19,Rougerie-22_ppt}. For points on a manifold, related results can be found in Ref.~\onlinecite{Dahlberg-78,Dragnev-02,DraSaf-07,BraDraSaf-14,HarRezSafVol-19}. The low temperature regime is also studied in the 2D Coulomb case in Ref.~\onlinecite{AmeOrt-12,Ameur-18,Ameur-21,AmeRom-22_ppt}.

\begin{lemma}[Separation]\label{lem:separation}
Let $d\geq1$ and $\max(0,d-2)\leq s<d$. Let $\rho_b>0$, $N\geq2$ and $\Omega\subset\R^d$ be any bounded open set. Let $x_1,...,x_N$ be a minimizer for $E_s(N,\Omega,\rho_b)$. There exists a universal constant $\delta>0$ (depending only on $s,d$) such that for every $x_{j_0}\in\Omega$ with ${\rm d}(x_{j_0},\partial\Omega)\geq \delta\rho_b^{-1/d}$, we have
$\min_{j\neq j_0}|x_{j_0}-x_j|\geq \delta\rho_b^{-\frac{1}d}$.
\end{lemma}

\begin{proof}
We start with Lieb's proof in the easier Coulomb case $s=d-2$. One can take $\delta=|B_1|^{-1/d}$, which means that each point in the interior of $\Omega$ has around it a ball of volume $\rho_b^{-1}$ containing no other point. Let us take $j_0$ as in the statement and assume by contradiction that there exists $j_1$ such that $x_{j_1}$ is in the ball $B_R(x_{j_0})\subset\Omega$ centered at $x_{j_0}$ and of radius $R:=\delta\rho_b^{-1/d}$. After relabeling the points we may assume that $j_0=2$ and $j_1=1$. Next we consider the energy as a function of $x_1$ and remark that
\begin{align}
\cE_{s,\rho_b,\Omega}(x_1,...,x_N)=&\cE_{s,\rho_b,\Omega}(x_2,...,x_N)+\left(V_s(x_1-x_2)-\rho_b\int_{B_R(x_2)}V_s(x_1-y)\,dy\right)\nn\\
&\qquad+\left(\sum_{j=3}^NV_s(x_1-x_j) -\rho_b\int_{\Omega\setminus B_R(x_2)}V_s(x_1-y)\,dy\right)\nn\\
=:&\cE_{s,\rho_b,\Omega}(x_2,...,x_N)+W_1(x_1)+W_2(x_1).
\label{eq:decomp_min_distance}
\end{align}
In the Coulomb case $s=d-2$ we have
\begin{equation*}
-\Delta_{x_1} W_2=c_d\sum_{x_j\in B_R(x_2)}\delta_{x_\ell}\geq0,\qquad \text{in $B_R(x_2)$}
\end{equation*}
with $c_d>0$ the constant in~\eqref{eq:fund_solution_Coulomb}.
Therefore $W_2$ is subharmonic and attains its minimum at the boundary of $B_R(x_2)$. By Newton's theorem, $W_1(x_1)=(V_s-|B_R|^{-1} V_s\ast\1_{B_R})(x_1-x_2)$ is strictly positive inside $B_R(x_2)$ and vanishes at the boundary. Hence $W_1+W_2$ must attain its minimum at the boundary of $B_R(x_2)$, which contradicts the minimality of the energy with respect to $x_1$.

For other values of $s$, this argument has been generalized by Petrache and Serfaty in Ref.~\onlinecite[Thm.~5]{PetSer-17}. The main point is that the potential $\mu\ast|x|^{-s}$ is the restriction to $\R^d\times\{0\}$ of the solution $\widetilde W$ to the degenerate elliptic equation $-{\rm div}|x_{d+1}|^{s+1-d}\nabla \widetilde W=\tilde{c}_{d,s}\mu(x)\delta_0(x_{d+1})$ in $\R^{d+1}$~\cite{CafSil-07}. Applying the (degenerate) maximum principle~\cite{FabKenSer-82} in $(\R^d\setminus\Omega\cup B_R(x_2))\times\R$, we see that $\widetilde W$ attains its minimum on $\partial\Omega\cup \partial B_R(x_2))\times\{0\}$. Hence the function $W_2$ in~\eqref{eq:decomp_min_distance} must again attain its minimum on $\partial\Omega\cup \partial B_R(x_2)$. Choosing $R$ small enough such that
$W_1(x)>\max_{\partial B_R(x_2)}W_1$ for all $x\in B_R(x_2)$ allows to conclude.
\end{proof}

\begin{remark}
Lieb's separation result was generalized by~\textcite{LieRouYng-19} as follows. To any configuration $x_1,...,x_K\in\R^d$ of $K$ points one can associate a unique set $B(x_1,...,x_K)\subset\R^d$ of volume $K/\rho_b$ containing the points so that the Coulomb potential generated by $\sum_{j=1}^K\delta_{x_j}-\rho_b\1_{B(x_1,...,x_K)}$ vanishes completely outside of  $B(x_1,...,x_K)$. In general this set will not be fully included into our background $\Omega$ but it is if the points are sufficiently inside $\Omega$. In this case, one can prove that for a minimizer the other $N-K$ points must all lie outside of $B(x_1,...,x_K)$. Furthermore, we have the inclusion $B(x_1,...,x_{K-1})\subset B(x_1,...,x_K)$. For one point $B(x)$ is just the ball of radius $(|B_1|\rho_b)^{-1/d}$ and we recover the previous result. By induction we can thus let $\Omega_j:=B(x_1,...,x_j)\setminus B(x_1,...,x_{j-1})$, after ordering the points properly. This allows to partition the interior of $\Omega$ into subsets $\Omega_j\ni x_j$ with the property that the charge $\delta_{x_j}-\rho_b\1_{\Omega_j}$ does not interact with any of the other compounds. This is one possible splitting of the background among the points, which we have already mentioned several times. The procedure stops for the charges close to the boundary, which should not matter in the thermodynamic limit. The ``screening regions'' of Ref.~\onlinecite{LieRouYng-19} are also known as ``subharmonic quadrature domains'' in the potential theory literature\cite{Sakai-82,GusSha-05,GusPut-07,Gustafsson-04}, where they are obtained by some kind of partial balayage. See Ref.~\onlinecite{Rougerie-22_ppt} for more details.
\end{remark}

Lemma~\ref{lem:separation} applies to all $N$, but for $N\gg \rho_b|\Omega|$ most of the points will accumulate at the boundary $\partial\Omega$ where the estimate does not hold. When the energy is of order $N$ we can show that there are only $o(N)$ points close to the boundary, so that Lemma~\ref{lem:separation} covers most of the points.

\begin{lemma}[Number of points close to the boundary] \label{lem:nb_points_boundary}
Let $d\geq1$ and $0<s<d$. Let $\rho_b>0$, $N\geq2$ and $\Omega\subset\R^d$ be any bounded open set.
For $\delta$ small enough (depending only on $\omega=|\Omega|^{-1/d}\Omega$), the number of points at a distance $R\leq \delta |\Omega|^{1/d}$ to the boundary satisfies
\begin{multline}
\#\{x_j\in\Omega\ : \rd(x_j,\partial\Omega)<R\}\leq N\frac{|\{x\in\Omega\ :\ \rd(x,\partial\Omega)\leq3R\}|}{|\Omega|}\\
+\frac{C|\Omega|^{\frac12}}{R^{\frac{d-s}{2}}}\left(\cE_s(x_1,...,x_N,\Omega,\rho_b)+c_1N\rho_b^{\frac{s}d}\right)^{\frac12}.
\end{multline}
\end{lemma}

When $N$ and the energy are both of order $|\Omega|$, we find that there are $o(N)$ points located at any distance $1\ll R\ll |\Omega|^{1/d}$ to the boundary, hence in particular also at a finite distance. When $\omega$ has a regular boundary in the sense of~\eqref{eq:hyp_domain} we can make this more quantitative and obtain after optimizing over $R$ that there are at most $O(N^{1-\frac{d-s}{d(2+d-s)}})$ points located at a distance $R\sim N^{\frac{2}{d(2+d-s)}}$ to the boundary $\partial\Omega$.

\begin{proof}
We assume again $\rho=1$. We take $\chi=|B_1|^{-1}\1_{B_1}$ in~\eqref{eq:estim_chi}. Letting $q_{r,\tau}(X):=\chi_r\ast\nu(\tau)=|B_r|^{-1}(\#X\cap B_r(\tau)-\rho_b |B_r\cap \Omega|)$ be the charge per unit volume in the ball $B_r(\tau)$ centered at $\tau$, the bound~\eqref{eq:estim_chi} gives for $R\geq\eps$,
\begin{equation}
\frac1R\int_R^{2R}\int_{\R^d}q_{r,\tau}(X)^2\rd\tau\,\dr\leq \frac{C}{R^{d-s}}\left(\cE_s(X,N^{1/d}\omega,\rho_b)+c_1N\rho_b^{\frac{s}d}\right).
\label{eq:estim_local_charge}
\end{equation}
When $RN^{-1/d}\leq\delta$ is small enough, the set $\Omega^-_{3R}=\{\tau\in\Omega\ :\ \rd(\tau,\partial \Omega)\geq 3R\}$ is non empty and we can restrict the $\tau$ integral to this set. Using the Cauchy-Schwarz inequality and $|\Omega_{3R}^-|\leq|\Omega|$, we find
\begin{align*}
\frac1R\int_R^{2R}\int_{\Omega_{3R}^-}q_{r,\tau}(X)^2\rd\tau\,\dr&\geq |\Omega|^{-1}\left(\frac1R\int_R^{2R}\int_{\Omega_{3R}^-}q_{r,\tau}(X)\,\rd\tau\right)^2\\
&= |\Omega|^{-1}\left(\frac1{R}\int_R^{2R}\int_{\Omega_{3R}^-}\frac{\#X\cap B_r(\tau)}{|B_r|}\,\rd\tau-|\Omega_{3R}^-|\right)^2.
\end{align*}
Since $\#X\cap\Omega_{R}^-\geq |B_r|^{-1}\int_{\Omega^-_{3R}}\#X\cap B_r(\tau)\rd\tau$ for $r\leq 2R$, this gives
\begin{equation*}
\#X\cap\Omega_{R}^-\geq|\Omega_{3R}^-| - \frac{C|\Omega|^{\frac12}}{R^{\frac{d-s}{2}}}\left(\cE_s(X,N^{1/d}\omega,\rho_b)+c_1N\rho_b^{\frac{s}d}\right)^{\frac12}
\end{equation*}
and concludes the proof.
\end{proof}

The local charge $q_{r,\tau}(X)$ used in the previous proof is also often called the \emph{discrepancy} (per unit volume). At $T=0$, Lemma~\ref{lem:separation} says that it is uniformly bounded for minimizers, away from the boundary of $\Omega$. The estimate~\eqref{eq:estim_local_charge} says that it is small in average for $r\gg1$, which is going to be useful later.

The next natural step is to prove that there are no big hole in the system, similarly as in Lemmas~\ref{lem:local_bound_GC_T0} and~\ref{lem:local_bound_C_T0}. To our knowledge, this is not understood for all values of $s$. In Ref.~\onlinecite{RotSer-15,PetRot-18,ArmSer-21} it is proved that $q_{r,\tau}(X)=O(1/r)$ for a minimizer $X$ of the canonical Coulomb problem, uniformly in $\tau$ far enough from the boundary. This implies that any ball of radius $r$ must contain of the order of $r^d$ points. Instead of looking at the holes, it is equivalent to ask what is the smallest radius $r$ so that $\cup_jB_r(x_j)$ covers the whole of $\Omega$ (perhaps with a neighborhood of the boundary removed). This is called the \emph{covering radius}~\cite{BorHarSaf-19}. Weaker average bounds on $q_{r,\tau}$ are proved later in Section~\ref{sec:proof_thm_Jellium} for all $s>0$.

At positive temperature, average bounds on $q_{r,\tau}$ are more difficult to obtain. It is shown by Lebl\'e and Serfaty in Ref.~\onlinecite[Lem.~3.2]{LebSer-17} that $\bE[q_{R,0}^2]\leq {C}/{R^{d-s}}$, for any infinite \emph{translation-invariant (stationary)} point process with finite Jellium energy per unit volume, in the case $\max(0,d-2)\leq s<d$. One can in fact get the same bound for all $0<s<d$ (averaged over $[R,2R]$) by integrating~\eqref{eq:estim_local_charge} against the point process. It is more complicated to deal with non translation-invariant systems (without performing an average over translations, that is, look at the ``empirical field''). An estimate on the local charge was provided in Ref.~\onlinecite{Leble-17,BauBouNikYau-17,BauBouNikYau-19} in the 2D Coulomb case, but only for sets of diameter $R\sim N^\eps$ for any $\eps>0$. The desired local bounds were finally proved very recently by~\textcite{ArmSer-21} for $s=d-2$ in all dimensions $d\geq2$, in the canonical case. Their result is formulated with an external confining potential but also applies to a uniform background, by Remark~\ref{rmk:confined_applies} in Section~\ref{sec:confined}. After passing to the limit, this provides a limiting point process in the Coulomb case, for all values of $\beta$.

Boursier has recently obtained\cite{Boursier-21_ppt} rigidity results about the fluctuations of the individual points in the case $0<s<1$ in dimension $d=1$, which imply very precise (average) local bounds. In the case of the 1D log gas $s=0$, much more is known due to the link with random matrices explained in Section~\ref{sec:confined}, see for instance Ref.~\onlinecite{BouErdYau-12,BouErdYau-14}.

For stationary point processes, it is possible to get around explicit local bounds and obtain some local tightness using the finiteness of the entropy per unit volume. Such an argument goes back to Georgii and Zessin~\cite{GeoZes-93,Georgii-11} and was crucially used for Riesz gases in Ref.~\onlinecite{DerVas-20,DerHarLebMai-21,DerVas-21_ppt}. More precisely, the entropy controls the expectation of $n_D\log n_D$ in any domain $D$ (see for instance Ref.~\onlinecite[Lem.~6.2]{MarLewNen-22_ppt}).

\bigskip

\paragraph{Infinite Riesz point processes.}
With local bounds at hand, the next step is to pass to the limit and get either infinite optimal configurations at $T=0$, or a point process at $T>0$, satisfying a DLR-type condition. The main difficulty here is to give a meaning to the potential, which is the formal limit
\begin{equation}
\Phi(x):=\lim_{\Omega\nearrow\R^d}\left(\sum_{x_j\in\Omega} V_s(x-x_j)-\rho_b\int_{\Omega}V_s(x-y)\,\dy\right).
\label{eq:def_Phi}
\end{equation}
This potential should appears in the DLR equations and it is interpreted as a renormalization of the infinite potential $\sum_{j} V_s(x-x_j)$. Local bounds are in general not enough to properly define the potential~\eqref{eq:def_Phi}.  One has to use more carefully the fact that we work with a minimizer for $T=0$ and a Gibbs measure for $T>0$.

A special situation is $d-1<s<d$, which has recently been considered by Dereudre and Vasseur~\cite{DerVas-21_ppt}. In this case, a local bound on the average number of points implies that $\nabla\Phi(x)$ is finite almost surely and it remains to show that $\Phi(x)$ is almost surely bounded for one $x$. This was used in Ref.~\onlinecite{DerVas-21_ppt} to prove the convergence of the Gibbs state at $T>0$ to a solution of the (properly renormalized) canonical and grand-canonical DLR equations. To be more precise, the authors started with the periodic model discussed later in Section~\ref{sec:periodic_BC} to ensure translation-invariance (hence a uniform density $\rho^{(1)}$), but their result applies the same to our situation, after performing an average over translations.

In a previous work~\cite{DerHarLebMai-21}, Dereudre, Hardy, Lebl\'e and Ma\"ida had managed to treat the case $s=0$ in dimension $d=1$ (1D log gas), using some \emph{a priori} local bounds from Ref.~\onlinecite{LebSer-17}. This case is better understood due to the link with random matrix theory. The corresponding point process had in fact already been constructed in Ref.~\onlinecite{ValVir-09,KilSto-09,Nakano-14,ValVir-17} but the (renormalized) DLR equations were first justified in Ref.~\onlinecite{DerHarLebMai-21}. Apart from the special 1D Coulomb case $s=-1$ which was already completely understood at the end of the 70s\cite{Kunz-74,AizMar-80} (see Section~\ref{sec:Kunz} below), the work of Dereudre, Hardy, Lebl\'e and Ma\"ida gave the first rigorous justification of DLR for long range systems.

The DLR characterization might not be the easiest path for Coulomb and Riesz gases. It was suggested by Gruber, Lugrin and Martin~\cite{GruLugMar-78,GruLugMar-80} back in the 80s that the BBGKY equations might be more adapted since they involve the field $\nabla V_s$ which decays better and is integrable at infinity for $s>d-1$. We discuss these equations in Section~\ref{sec:prop} below and it would be interesting to make the connection with Ref.~\onlinecite{DerVas-21_ppt}.

Since we always think of $\Phi(x)$ as a renormalization of the divergent potential $\sum_{j=1}^\ii V_s(x-x_j)=+\ii$, one fundamental question is to identify the infinite configurations $X=\{x_j\}_{j\in\N}$ for which the infinite series can be renormalized in a natural and unambiguous way. Our Riesz point process should concentrate on such configurations. Our train of thought in this article is that the uniform background (also sometimes called the \emph{integral compensator}, see Ref.~\onlinecite[Rmk.~1.15]{DerVas-21_ppt}) is the right approach, at least for not too low values of $s$. In Section~\ref{sec:periodic} we will compare it with another method based on analytic continuation in $s$.

\subsection{Equilibrium configurations for $d-2\leq s<d$}
We state here a result in the case $T=0$ for $d-2\leq s<d$, which has not been treated in Ref.~\onlinecite{DerHarLebMai-21,DerVas-21_ppt} and is the equivalent of Theorem~\ref{thm:CV_equilibrium} in the long range case. We are able to renormalize the potential for our equilibrium configuration, that is, show the existence of the function $\Phi$ in~\eqref{eq:def_Phi}, without performing any average over translations. The following result seems to be new and its detailed proof is provided later in Section~\ref{sec:proof_thm_Jellium}.

\begin{theorem}[Equilibrium configurations for $d-2\leq s<d$]\label{thm:infinite_conf_Jellium}
We assume $0<s<d$ in dimensions $d\in\{1,2\}$ and $d-2\leq s<d$ in dimensions $d\geq3$. Let $\omega$ be any domain of volume $|\omega|=1$ so that $|\partial\omega|=0$. Let $\rho_b>0$ and $\mu\in\R$. Consider any minimizer $X_\ell=\{x_{1,\ell},...,x_{N_\ell,\ell}\}\subset \Omega:=\ell\omega$ for the grand-canonical problem $E^{\rm GC}_s(\mu,\ell\omega,\rho_b)$. Up to extraction of a subsequence, translation of $\omega$, and relabelling the $x_{j,\ell}$, we have the following properties:

\smallskip

\noindent$\bullet$ $x_{j,\ell}\to x_j$ as $\ell\to\ii$ for any fixed $j\geq1$. The infinite configuration of points $X=\{x_j,\ j\geq1\}$ satisfies $|x_j-x_k|\geq\delta\mu^{-1/s}$ for $j\neq k$, with $\delta>0$ the same constant as in Lemma~\ref{lem:separation}.

\smallskip

\noindent$\bullet$ The potential
$$\Phi_\ell(x):=\sum_{j=1}^{N_\ell}\frac1{|x-x_{j,\ell}|^s}-\rho_b\int_{\Omega}\frac{\dy}{|x-y|^s}$$
is bounded below on $\Omega$ and locally bounded from above on $\Omega\setminus X_\ell$, independently of $\ell$. The sequence $\Phi_\ell$ converges as $\ell\to\ii$ to a function $\Phi\in C^0(\R^d\setminus X)\cap L^1_{\rm loc}(\R^d)$, in $L^1_{\rm loc}(\R^d)$ and locally uniformly in the sense that for any $R>0$
\begin{equation}
\Phi_\ell(x)-\sum_{x_{j,\ell}\in \overline{B_R}}\frac1{|x-x_{j,\ell}|^s}\underset{\ell\to\ii}\longrightarrow \Phi(x)-\sum_{x_{j}\in \overline{B_R}}\frac1{|x-x_{j}|^s}\quad\text{uniformly on $\overline{B_R}$.}
 \label{eq:limit_Phi_local}
\end{equation}
Denoting by
\begin{equation}
\Phi^{(j_0)}=\lim_{\ell\to\ii}\sum_{\substack{j=1\\ j\neq j_0}}^{N_\ell}\frac1{|x_{j_0}-x_{j,\ell}|^s}-\rho_b\int_{\Omega}\frac{\dy}{|x_{j_0}-y|^s}=\lim_{x\to x_{j_0}}\left(\Phi(x)-\frac1{|x-x_{j_0}|^s}\right)
\label{eq:limit_interaction_j0}
\end{equation}
the limit of the interaction of any $x_{j_0,\ell}$ with the rest of the system, we can express for $s>d-2$ and all $x\in\R^d\setminus\{x_{j_0}\}$
 \begin{equation}
  \Phi(x)=\Phi^{(j_0)}+\frac1{|x-x_{j_0}|^s}+\sum_{j\neq j_0}\left(\frac1{|x-x_j|^s}-\frac1{|x_{j_0}-x_j|^s}+s\frac{(x-x_{j_0})\cdot (x_{j_0}-x_j)}{|x_{j_0}-x_j|^{s+2}}\right).
\label{eq:limit_Phi}
\end{equation}
If $s=d-2$, $\Phi$ solves the equation
\begin{equation}
-\Delta \Phi=(d-2)|\bS^{d-1}|\left(\sum_j \delta_{x_j}-\rho_b\right)
\label{eq:Poisson_Phi_Jellium}
\end{equation}
in the sense of distributions on $\R^d$. In all cases, $\Phi$ is uniquely determined from the infinite configuration $X=\{x_j\}$, up to a constant.

\smallskip

\noindent$\bullet$ The limiting infinite configuration $X$ satisfies the equilibrium equations
\begin{equation}
\sum_{j=1}^n\Phi_{D^c}(y_j)+\cE_s(Y,D,\rho_b)\geq \sum_{j=1}^N\Phi_{D^c}(x_j)+\cE_s(X\cap  \overline{D},D,\rho_b)+\mu(n-N),
\label{eq:DLR_T0_Jellium}
\end{equation}
for any domain $D\subset \R^d$, any $Y=\{y_1,...,y_n\}\subset \overline D$ and after relabelling the $x_j$ so that $\{x_1,...,x_N\}=X\cap \overline D$, where
$$\Phi_{D^c}(x):=\Phi(x)-\sum_{x_j\in\overline{D}}\frac{1}{|x-x_j|^s}+\rho_b\int_D\frac{\dy}{|x-y|^s}$$
denotes the potential induced by the system outside of $D$.
\end{theorem}

The series on the right side of~\eqref{eq:limit_Phi} is convergent for $d-2<s<d$ since the summand behaves as $|x_j|^{-(s+2)}$ for large $j$ and the points $x_j$ are well separated. When $d-1<s<d$, we will prove that the third term vanishes after summing, and thus obtain the simpler formula
\begin{equation}
 \Phi(x)=\Phi^{(j_0)}+\frac1{|x-x_{j_0}|^s}+\sum_{j\neq j_0}\left(\frac1{|x-x_j|^s}-\frac1{|x_{j_0}-x_j|^s}\right)\qquad\text{for $d-1<s<d$}.
 \label{def:Phi_simpler}
\end{equation}
The last sum was called the `move function' in Ref.~\onlinecite{DerHarLebMai-21,DerVas-21_ppt}. The formula~\eqref{def:Phi_simpler} formally amounts to shifting the divergent series by an infinite ($x$-independent) constant, which seems a natural procedure since $\nabla \Phi$ is finite. For $d-2<s<d$, the less intuitive $x$-dependent third term in~\eqref{eq:limit_Phi} naturally appears in our proof because we look at an equilibrium configuration and use that the energy is stationary at a minimizer. We see no easy way of bringing more derivatives to renormalize the series in a similar manner for $s< d-2$. At $s=d-2$ we are not able to provide an explicit form for $\Phi$ but know that it is the \emph{unique} solution (up to constants) to Poisson's equation~\eqref{eq:Poisson_Phi_Jellium}, which is bounded-below and grows at most quadratically at infinity (see Lemma~\ref{lem:Liouville}).

It may worry the reader that given the infinite configuration $X=\{x_j\}$, the corresponding potential $\Phi$ seems to be only explicitly known up to a constant. This is probably unavoidable. If we fix a specific representation for $\Phi$ (e.g. the right side of~\eqref{eq:limit_Phi} without the first constant), then we can include the missing unknown constant into the chemical potential $\mu$ and thereby obtain equilibrium DLR equations with a \emph{renormalized chemical potential $\mu_{\rm ren}$}. It seems reasonable to expect that for such a fixed representation of $\Phi$, there exists a solution for only one value of $\mu_{\rm ren}$. When performing the thermodynamic limit with a different $\mu$, some points would then have to escape to a neighborhood of the boundary of $\ell\omega$ so as to generate a constant potential inside and thus change $\mu$ to $\mu_{\rm ren}$. Such a phenomenon is predicted to happen in the Coulomb case at positive temperature in Ref.~\onlinecite{Imbrie-82,FedKen-85} as well as in Ref.~\onlinecite[Sec.~II.F]{BryMar-99}. It has never been proved, to our knowledge.

In physics, renormalization often coincides with an analytic continuation. We are unable to show this for minimizers but believe it is the case for $s>d-2$, since those are conjectured to be periodic. In the periodic case, we prove the equivalence in the next section.

\section{Analyticity and periodicity}\label{sec:periodic}

Periodic systems play a central role in the analysis of Riesz gases. They are believed to be optimal in some situations (e.g. at zero temperature). They also arise when the points are placed on the torus in the thermodynamic limit, instead of a container $\Omega$ with hard walls. The energy per unit volume and the potential of a periodic configuration of points can be expressed in terms of the \emph{Epstein Zeta function}~\cite{Epstein-06}, which is a $d$--dimensional generalization of the Riemann Zeta function. This naturally leads to interesting problems in complex analysis and analytic number theory.

In the Coulomb case, the definition of the potential $\Phi(x)$ of an infinite periodic configuration poses some problems which have raised a lot of confusion in the literature. This will also be discussed in this section.

\subsection{Background as an analytic continuation}\label{sec:periodic_continuation}

In this first section we exhibit a connection between the long and short range cases, by proving that the potential of an infinite system with background is for $s<d$ the \emph{analytic continuation in $s$} of the same system with a short range potential $s>d$. Results of this type go back to Ref.~\onlinecite{BorBorTay-85,BorBorShaZuc-88,BorBorSha-89,BorBorStr-14} and they shed a new light on the role of the uniform background, at least for not too low values of $s$. We start with general point configurations before we turn to the special case of periodic lattices.

\subsubsection{General case}

Let $X=\{x_j\}_{j\in\N}$ be any infinite configuration of points in $\R^d$. We assume that the points are well separated:
$\inf_{j\neq k}|x_j-x_k|>0$.
This is sufficient to define the potential
$$\Phi(s,x)=\sum_{j=1}^\ii\frac{1}{|x-x_j|^s},\qquad \text{for $s>d$ and any $x\in \R^d\setminus X$.}$$
In fact this defines an analytic function in $s$ on the half plane $\{\Re(s)>d\}$ for every fixed $x\in\R^d\setminus X$. The question we are asking here is whether this potential admits a meromorphic continuation in $s$ and if this extension coincides with the limit obtained by inserting a background
$$\Phi(s,x)=\lim_{\Omega\nearrow \R^d}\left(\sum_{x_j\in\Omega}\frac{1}{|x-x_j|^s}-\rho_b\int_\Omega\frac{\dy}{|x-y|^s}\right),\qquad \text{for $s<d$.}$$
For configurations leaving no big hole, we expect that $\Phi(s,x)$ will diverge when $s\to d^+$, so that there will always be a pole at $s=d$. In fact, we expect the residue to be related to the density of points, which is the constant $\rho_b$ which we have to choose for the background.

The points must be sufficiently well placed so as to screen a uniform background and the allowed values of $s<d$ will depend on the quality of the screening. The following is a simple result inspired of Ref.~\onlinecite{BorBorShaZuc-88,BorBorSha-89,BorBorStr-14,BlaBriLio-02,GeSan-21}, where we only require that the number of points in large balls is sufficiently close to the corresponding volume.

\begin{lemma}[Background as an analytic continuation in $s$ with discrepancy bounds]\label{lem:counting}
Consider an infinite configuration of points $X=\{x_j\}_{j\in\N}\subset\R^d$ such that $\inf_{j\neq k}|x_j-x_k|>0$. Let $x\in\R^d\setminus X$ and assume that
\begin{equation}
\left|\#X\cap B_R(x)-\rho_b\frac{|\bS^{d-1}|R^d}d\right|\leq CR^{d-\alpha},\qquad \forall R\geq C
\label{eq:counting_points}
\end{equation}
for some $\rho_b,C>0$ and $0<\alpha\leq d$. Then the potential $\Phi(s,x):=\sum_{j\geq1}|x-x_j|^{-s}$, initially defined on the half plane $\{\Re(s)>d\}$, admits a meromorphic extension to $\{\Re(s)>d-\alpha\}$ with a unique simple pole at $s=d$, of residue $\rho_b|\bS^{d-1}|$. This analytic extension is given by the limit
\begin{equation}
\Phi(s,x)=\lim_{R\to\ii}\left(\sum_{x_j\in B_{R}(x)}\frac{1}{|x-x_j|^s}-\rho_b\int_{B_{R}(x)}\frac{\dy}{|x-y|^s}\right)
 \label{eq:Phi_s_x}
\end{equation}
for all $d-\alpha<\Re(s)<d$.
\end{lemma}

Our background is here a ball centered at the point $x$ for simplicity, but more general situations can be considered. All periodic lattices satisfy~\eqref{eq:counting_points} with a constant $C$ independent of $x$, and $\alpha=2$ in dimensions $d\geq5$~\cite{Gotze-04}, $\alpha<2$ in dimension $d=4$, $\alpha<2-2/(d+1)$ in dimensions $d\in\{2,3\}$~\cite{Landau-15,Landau-24} and $\alpha=1$ in dimension $d=1$. Thus for periodic configurations we can reach $s>d-2$ for $d\geq4$ and $s>d-2+2/(d+1)$ for $d\in\{1,2,3\}$.
There exist better estimates for the cubic lattice~\cite{IviKraKuhNow-06}. The analyticity can be wrong at $s=d-2$ for periodic configurations, see Lemma~\ref{lem:periodic_V} and Remark~\ref{rmk:shift} below.

The range of the analytic extension provided in Lemma~\ref{lem:counting} is optimal under condition~\eqref{eq:counting_points}. Indeed, consider the infinite configuration $Y=\{|k|^{\frac\alpha{d-\alpha}}k,\ k\in\Z^d\}$ which has $O(R^{d-\alpha})$ points in a ball $B_R$ and yields a pole at $s=d-\alpha$. Adding it to a nice (e.g.~periodic) configuration $X$ with only one pole at $s=d$, we obtain a configuration with two poles at $d$ and $d-\alpha$.

\begin{proof}
Let us define
\begin{equation}
 f_R(s)=\sum_{x_j\in B_R(x)}\frac{1}{|x-x_j|^s}-\rho_b\int_{B_R(x)\setminus B_1(x)}\frac{\dy}{|x-y|^s}
 \label{eq:def_f_R}
\end{equation}
and prove that it converges to an analytic function on the half plane $\{\Re(s)>d-\alpha\}$ in the limit $R\to\ii$. We introduce the locally finite measure $\nu:=\sum_{x_j\in\R^d\setminus B_1(x)}\delta_{x_j-x}-\rho_b\1_{\R^d\setminus B_1(0)}$ (we have re-centered the system at 0 for simplicity). Integrating by parts in radial coordinates gives
\begin{equation*}
f_R(s)=\sum_{x_j\in B_1(x)}\frac{1}{|x-x_j|^s}+\int_{|y|\leq R}\frac{\rd\nu(y)}{|y|^s}=\sum_{x_j\in B_1(x)}\frac{1}{|x-x_j|^s}+\frac{\nu(B_R)}{R^s}+s\int_1^R\frac{\nu(B_r)}{r^{s+1}}\dr.
\end{equation*}
The first sum on the right side contains finitely many terms due to the positive distance between the points and it is thus analytic. From the assumption~\eqref{eq:counting_points}, we have $|\nu(B_r)|\leq C(1+r^{d-\alpha})$ and therefore the second term $\nu(B_R)/R^s$ goes to zero for $\Re(s)>d-\alpha$, when $R\to\ii$. The last integral is convergent in the limit $R\to\ii$ and the limiting function
$$f(s)=\sum_{x_j\in B_1(x)}\frac{1}{|x-x_j|^s}+s\int_1^\ii\frac{\nu(B_r)}{r^{s+1}}\dr$$
is analytic on $\{\Re(s)>d-\alpha\}$.

When $\Re(s)>d$, the two terms in~\eqref{eq:def_f_R} converge separately and we find
\begin{equation}
f(s)=\Phi(s,x)-\rho_b\int_{\R^d\setminus B_1(x)}\frac{\dy}{|x-y|^s}=\Phi(s,x)-\frac{\rho_b|\bS^{d-1}|}{s-d}
\label{eq:comput_Phi_f}
\end{equation}
with $\Phi(s,x)$ the potential in the statement. This proves that $s\mapsto \Phi(s,x)$ admits a meromorphic extension to $\{\Re(s)>d-\alpha\}$ with a simple pole at $s=d$, of residue $\rho_b|\bS^{d-1}|$. On the other hand, for $d-\alpha<\Re(s)<d$, we can write
\begin{multline*}
\lim_{R\to\ii}\left(\sum_{x_j\in B_{R}(x)}\frac{1}{|x-x_j|^s}-\rho_b\int_{B_{R}(x)}\frac{\dy}{|x-y|^s}\right)\\
=\lim_{R\to\ii}\left(f_R(s)-\rho_b\int_{B_1(x)}\frac{\dy}{|x-y|^s}\right)=f(s)+\frac{\rho_b|\bS^{d-1}|}{s-d}.
\end{multline*}
This coincides with the meromorphic extension of $\Phi(s,x)$, by~\eqref{eq:comput_Phi_f}.
\end{proof}

The previous result is based on the rather global information~\eqref{eq:counting_points} on the number of points in large balls. The range of validity in $s$ can be improved under some more local assumptions on the positions of the points. The following is a slight generalization of a result in Ref.~\onlinecite{BorBorShaZuc-88,BorBorSha-89,BorBorStr-14}, which only dealt with periodic systems (see also Ref.~\onlinecite{Lauritsen-21}). The spirit is again that each point owns a small piece of the background, of constant volume.

\begin{lemma}[Background as an analytic continuation in $s$ with local bounds]\label{lem:analytic_potential}
Consider an infinite configuration of points $X=\{x_j\}_{j\in\N}\subset\R^d$. Assume that $\R^d=\cup_j\overline{\Omega_j}$ for some disjoints measurable sets $\Omega_j$ satisfying
$|\Omega_j|=\rho_b^{-1}$ and $B_r(x_j)\subset \Omega_j\subset B_{1/r}(x_j)$ for some $0<r<1$ and all $j\geq1$. Let $x\in\R^d\setminus X$. Then the potential $\Phi(s,x):=\sum_{j\geq1}|x-x_j|^{-s}$ admits a meromorphic extension to $\{\Re(s)>d-1\}$ with a unique simple pole at $s=d$. If furthermore
\begin{equation}
\rho_b\int_{\Omega_j}x\,\dx=x_j,\qquad\forall j\geq1,
\label{eq:no_dipole}
\end{equation}
then the same holds on $\{\Re(s)>d-2\}$. For every $d-1<\Re(s)<d$, we have
\begin{equation}
 \Phi(s,x)=\lim_{\ell\to\ii}\left(\sum_{x_j\in \ell\omega}\frac{1}{|x-x_j|^s}-\rho_b\int_{\ell\omega}\frac{\dy}{|x-y|^s}\right)
 \label{eq:Phi_s_x2}
\end{equation}
for all $\omega\subset\R^d$ containing the origin, with $|\omega|=1$ and a boundary as in~\eqref{eq:hyp_domain}. Under the condition~\eqref{eq:no_dipole}, we have for $d-2<\Re(s)<d$
\begin{align}
\Phi(s,x)&=\lim_{\ell\to\ii}\left(\sum_{j}\frac{\chi(x_j/\ell)}{|x-x_j|^s}-\rho_b\int_{\R^d}\frac{\chi(y/\ell)\,\dy}{|x-y|^s}\right)\nn\\
&=\lim_{R\to\ii}\left(\sum_{x_j\in B_R}\frac{1}{|x-x_j|^s}-\rho_b\int_{\!\!\underset{x_j\in B_R}\cup\Omega_j}\frac{\dy}{|x-y|^s}\right)
 \label{eq:regularized_limit}
\end{align}
where $\chi\in C^\ii_c(\R^d)$ is any function such that $\chi(0)=1$.
\end{lemma}

For balls the first part is a consequence of Lemma~\ref{lem:counting} since the assumptions on the $\Omega_j$'s imply that~\eqref{eq:counting_points} holds with $\alpha=1$. The second part is an improvement for $d-2<\Re(s)\leq d-1$, under the no dipole assumption~\eqref{eq:no_dipole}. The limit~\eqref{eq:regularized_limit} might not exist for general rescaled sets $\Omega=\ell\omega$ as in~\eqref{eq:Phi_s_x}. It is well known that in this context smooth cut-offs (such as on the first line of~\eqref{eq:regularized_limit}) are better behaved (Ref.~\onlinecite[Sec.~3.7]{Tao-13}).
In dimension $d=1$ a similar result holds for $-1<s<0$ but the potential is the opposite of the analytic continuation, due to our choice of sign for $V_s$. At $s=0$ it is the derivative of the analytic continuation.

\begin{proof}
We argue similarly as in Lemma~\ref{lem:counting}. We consider the series
$$f(s)=\sum_{j\geq1}\left(\frac1{|x-x_j|^s}-\rho_b\int_{\Omega_j\setminus B_1(x)}\frac{\dy}{|x-y|^s}\right)$$
which defines an analytic function over the whole half plane $\{\Re(s)>d-1\}$. Indeed, our assumptions on $X$ and the $\Omega_j$'s imply that only finitely many $\Omega_j$ can intersect $B_1(x)$, whereas for the other ones we have by~\eqref{eq:dipole}
\begin{equation}
 \left|\frac1{|x-x_j|^s}-\rho_b\int_{\Omega_j}\frac{\dy}{|x-y|^s}\right|\leq \frac{C}{|x_j|^{\Re(s)+1}}
 \label{eq:dipole_expansion_Omega_j}
\end{equation}
for $j$ large enough. The truncated series
\begin{align*}
f_{\ell}(s)&:=\sum_{x_j\in\ell\omega}\left(\frac1{|x-x_j|^s}-\rho_b\int_{\Omega_j\setminus B_1(x)}\frac{\dy}{|x-y|^s}\right)\\
&=\sum_{x_j\in\ell\omega}\frac1{|x-x_j|^s}-\rho_b\int_{\!\!\underset{x_j\in\ell\omega}\cup \Omega_j\setminus B_1(x)}\frac{\dy}{|x-y|^s},
\end{align*}
converges locally uniformly to $f(s)$ on $\{\Re(s)>d-1\}$. For $\Re(s)>d$ we get $f(s)=\Phi(s,x)-\rho_b|\bS^{d-1}|/(s-d)$ hence $\Phi(s,x)$ admits the mentioned meromorphic extension to $\{\Re(s)>d-1\}$. On the other hand, for $d-1<\Re(s)<d$ and $\ell$ large enough, we have
\begin{equation*}
f_{\ell}(s)=\Phi_\ell(s,x)+\rho_b\int_{B_1}\frac{\dy}{|y|^s}
+\rho_b\left(\int_{\ell\omega}\frac{\dy}{|x-y|^s}-\int_{\!\!\underset{x_j\in\ell\omega}\cup \Omega_j}\frac{\dy}{|x-y|^s}\right),
\end{equation*}
with $\Phi_\ell(s,x)$ the function inside the limit on the right side of~\eqref{eq:Phi_s_x}. The last term involves only the $\Omega_j$ intersecting the boundary of $\ell\omega$ and those are at most at a distance $2r$ from this boundary. It can thus be estimated by
$$\left|\int_{\!\!\underset{x_j\in\ell\omega}\cup \Omega_j}\frac{\dy}{|x-y|^s}-\int_{\ell\omega}\frac{\dy}{|x-y|^s}\right|\leq \frac{C\ell^{d-1}r}{\ell^{\Re(s)}}\to0$$
due to the regularity assumption on $\omega$. We deduce that $\Phi_\ell(s,x)$ admits the limit
$$\lim_{\ell\to\ii}\Phi_\ell(s,x)=f(s)-\rho_b\int_{B_\eps}\frac{\dy}{|y|^s}=f(s)+\frac{\rho_b|\bS^{d-1}|}{s-d}$$
which concludes the proof for $d-1<\Re(s)<d$.

If the dipole vanishes as in~\eqref{eq:no_dipole}, the exponent on the right side of~\eqref{eq:dipole_expansion_Omega_j} is replaced by $\Re(s)+2$ and $f(s)$ is analytic on $\{\Re(s)>d-2\}$. The proof for the second limit~\eqref{eq:regularized_limit} works exactly the same as above. For the first limit we only have to prove that for $d-2<\Re(s)<d$
\begin{equation}
\lim_{\ell\to\ii}\sum_j\int_{\Omega_j\setminus B_1(x)}\frac{\chi_\ell(y)-\chi_\ell(x_j)}{|x-y|^s}\dy=0.
\label{eq:chi_ell}
\end{equation}
This follows from the fact that, for $j$ large enough and $y\in\Omega_j$,
$$\frac{\chi_\ell(y)-\chi_\ell(x_j)}{|x-y|^s}=\frac{\nabla \chi(\ell^{-1}x_j)}{\ell|x_j|^s}\cdot (y-x_j)+O\left(\frac{1}{\ell^2|x_j|^{\Re(s)}}+\frac{1}{\ell|x_j|^{\Re(s)+1}}\right).$$
After integration over $\Omega_j$, the first term vanishes. Summing over $j$ the second and third terms give a $O(\ell^{d-\Re(s)-2})$ for $d-2<\Re(s)<d-1$, a $O(\log\ell/\ell)$ for $\Re(s)=d-1$ and a $O(\ell^{-1})$ for $d-1<\Re(s)<d$.
\end{proof}

\subsubsection{Periodic systems}\label{sec:periodic_potential}
We now discuss the periodic case. Let
$$\cL=v_1\Z+\cdots +v_d\Z$$
be any lattice, with $(v_1,...,v_d)$ a (not necessarily orthonormal) basis of $\R^d$. We call $Q=\{x\in\R^d\ :\ |x|<|x-z|,\ \forall z\in\cL\}$ the Wigner-Seitz cell of $\cL$, which contains all the points closer to the origin than to any other point of the lattice. This is a polyhedron. Note that since $\cL=-\cL$, then we have as well $Q=-Q$. This implies that $Q$ has no dipole: $\int_Q y\,\dy=0$. In the short range case, we define the $\cL$--periodic interaction potential
\begin{equation}
\boxed{V_s^{\rm \cL}(x)=\sum_{z\in\cL}\frac{1}{|x-z|^s},\qquad\text{for $x\notin\cL$ and $s>d$}.}
\label{eq:V_s_periodic_short_range}
\end{equation}
This corresponds to the potential $\Phi(s,x)$ studied in the previous subsection. We also introduce the interaction between any point of the lattice with all the other ones
\begin{equation}
M_{\cL}(s):=\lim_{x\to0}\left(V_s^\cL(x)-\frac1{|x|^s}\right)=\sum_{z\in\cL\setminus\{0\}}\frac1{|z|^s},\qquad\text{for $s>d$,}
\label{eq:Madelung}
\end{equation}
which is called the \emph{Madelung constant of the lattice $\cL$}.\cite{Madelung-18,LieSim-77b,CatBriLio-98} We have $M_\cL(s)=2\zeta_\cL(s)$ where
\begin{equation}
\boxed{\zeta_\cL(s):=\frac12\sum_{z\in\cL\setminus\{0\}}\frac1{|z|^s}}
\label{eq:zeta_L}
\end{equation}
is the \emph{Epstein Zeta function}~\cite{Epstein-06} of the lattice $\cL$. When $d=1$ and $\cL=\Z$ we recover the usual Riemann Zeta function.

Since $\int_Qy\,\dy=0$, Lemma~\ref{lem:analytic_potential} applies with $X=\cL$ and $\Omega_z=z+Q$. It gives that $V^{\cL}_s(x)$ and $M_{\cL}(s)$ possess a meromorphic extension to at least $\{\Re(s)>d-2\}$ and that the potential for $d-2<\Re(s)<d$ can be computed as a limit with a uniform background. It turns out that the extension exists on the whole plane $\C$, with a unique pole at $s=d$.
However, the simple background fails to reproduce the analytic extension for $s\leq d-2$, even when the limit exists. This has caused a lot of confusion. Our goal in this section is to explain where this issue is coming from and how the background can be modified in order to get the analytic extension over the whole of $\C\setminus\{d\}$. Essentially, we will have to introduce some oscillations at the boundary of $\Omega$.

The analytic extension will be given in terms of the $\cL$--periodic potential $V_s^\cL$ which has the same Fourier coefficients as $V_s$ and satisfies $\int_Q V_s^\cL=0$. Its Fourier transform reads
\begin{equation}
\widehat{V_s^\cL}=\begin{cases}
\dps {\rm sgn}(s)\,\frac{2^{\frac{3d}{2}-s}\pi^{d}\,\Gamma\left(\frac{d-s}{2}\right)}{\Gamma\left(\frac{s}{2}\right)|Q|}\sum_{k\in\cL^*\setminus\{0\}}\frac{\delta_k}{|k|^{d-s}}&\text{for $s\in(-2,0)\cup(0,d)$,}\\
\dps \frac{2^{\frac{3d}2+1-s}\pi^{d}\Gamma\left(\frac{d}{2}\right)}{|Q|}\sum_{k\in\cL^*\setminus\{0\}}\frac{\delta_k}{|k|^{d}}&\text{for $s=0$,}
\end{cases}
 \label{eq:def_V_s_L}
\end{equation}
where $\cL^*$ is the lattice dual to $\cL$, that is, the one generated by the dual basis $v_1^*,...,v_d^*$ times $2\pi$.

The following is in the same spirit as Lemmas~\ref{lem:counting} and~\ref{lem:analytic_potential}, except that it is valid for a much larger range of $s$, and that we have to be careful with the signs.

\begin{lemma}\label{lem:periodic_V}
The potential $V_s^\cL(x)$ defined for $\{\Re(s)>d\}$ admits a meromorphic extension $\widetilde{V}_s^\cL(x)$ to $\C$ with a unique and simple pole at $s=d$, and
\begin{equation}
\widetilde{V}_s^\cL(x)= \begin{cases}
 V_s^\cL(x)&\text{for $s\in(0,\ii)\setminus\{d\}$,}\\
 -V_s^\cL(x)&\text{for $s<0$,}
 \end{cases}
 \label{eq:V_s_extension_analytic}
\end{equation}
where $V_s^\cL$ is the periodic function defined in~\eqref{eq:def_V_s_L}. At $s=0$ we have $V_0^\cL(x)=\frac{\partial}{\partial s}\widetilde V_s^\cL(x)\big|_{s=0}$.
Similarly, we have
$$M_{\cL}(s):=\lim_{x\to0}\left(V_s^\cL(x)-\frac1{|x|^s}\right)=2\begin{cases}
\zeta_\cL(s)&\text{for $s\in(0,\ii)\setminus\{d\}$,}\\
-\zeta_\cL(s)&\text{for $s<0$.}
\end{cases}
$$
and
$$M_{\cL}(0):=\lim_{x\to0}\left(V_0^\cL(x)+\log|x|\right)=2\zeta_\cL'(0)\qquad \text{for $s=0$.}$$
Assume furthermore that $d-2<s<d$ or that $s=d-2$ and $d\int_Qy_iy_j\,\dy=\int_Q|y|^ 2\,\dy \delta_{ij}$ (no quadrupole). Then we have
\begin{multline}
\lim_{\ell\to\ii}\left(\sum_{z\in \cL\cap \in \ell\omega}V_s(x-z)-\rho_b\int_{\cup_{z\in\ell\omega}(Q+z)}\frac{\dy}{|x-y|^s}\right)\\
=V_s^\cL(x)+ \delta_{d-2}(s)\frac{|\bS^{d-1}|}{2d}\int_Q|y|^2\,\dy+\frac{\log|Q|}{d|Q|}\delta_0(s)
 \label{eq:shift}
\end{multline}
for all $x\in\R^d\setminus \cL$ and $\omega$ containing the origin so that $|\omega|=1$ and $|\partial\omega|=0$.
\end{lemma}

We discuss the proof informally. The analytic extension is well known and only the formula~\eqref{eq:V_s_extension_analytic} requires an argument. Let us start with the case $\max(0,d-2)<\Re(s)<d$. From the proof of Lemma~\ref{lem:analytic_potential}, we see that it suffices to show the formula
\begin{equation}
V_s^\cL(x)=\sum_{z\in\cL}\left(V_s-|Q|^{-1}\1_Q\ast V_s\right)(x-z)\qquad\text{for $d-2<\Re(s)<d$}.
\label{eq:V_s_periodic_Jellium}
\end{equation}
This is well known and the proof goes as follows. The function $f=V_s-|Q|^{-1}\1_Q\ast V_s$ is integrable. Its Fourier transform is proportional to $|k|^{s-d}(1-(2\pi)^{d/2}|Q|^{-1}\widehat{\1_Q}(k))$. This is a $O(|k|^{2-d+\Re(s)})$ at the origin, hence $\int_{\R^d} f=0$.
Since $\widehat{\1_Q}(k)=0$ for all $0\neq k\in\cL^*$, we have $\widehat{f}(k)=\widehat{V^\cL_s}(k)$ for all $k\in\cL^*$ and thus obtain~\eqref{eq:V_s_periodic_Jellium} from Poisson's summation formula.

At $s=d-2$ and if the quadrupole vanishes, the series on the right of~\eqref{eq:V_s_periodic_Jellium} converges but it does not coincide with $V_s^\cL$. As stated in~\eqref{eq:shift}, there appears a constant shift proportional to $\int_Q|y|^2\dy$.\footnote{The shift $\int_Q|y|^2\,\dy$ is also called the \emph{lattice quantizer}~\cite{ConSlo-99} and it is minimal for the BCC lattice~\cite{BarSlo-83} in 3D.} This was proved in Ref.~\onlinecite{BorBorShaZuc-88,BorBorSha-89,BorBorStr-14,LewLie-15} and is further discussed in Remark~\ref{rmk:shift} below. More screening is required for $s\leq d-2$ and the right formula to be used is rather cumbersome at first sight. It reads
\begin{equation}
V_s^\cL(x)=\sum_{z\in\cL}\left(V_s+\sum_{m=1}^{M}(-1)^m{M\choose m}\,V_s\ast \left(\frac{\1_Q}{|Q|}\right)^{\ast m}\right)\!\!(x-z)\qquad \text{for $s>d-2M$,}
\label{eq:V_s_periodic_general}
\end{equation}
where $\phi^{\ast m}:=\phi\ast\cdots\ast \phi$ denotes the iterated convolution. The new function $f$ in the sum now has a Fourier transform proportional to $|k|^{s-d}(1-(2\pi)^{d/2}|Q|^{-1}\widehat{\1_Q}(k))^{M}$ and the arguments are the same as before. The proof of Lemma~\ref{lem:analytic_potential} can be adapted to handle the new terms and this is how one can get the stated analytic continuation to the whole of $\C\setminus\{d\}$, increasing $M$ step by step.

For instance, for $M=2$ and $d-4<s<d$ we get the doubly-screened potential\cite{LewLie-15,BlaLew-15,Lauritsen-21}
\begin{equation}
 V_s^\cL(x)=\sum_{z\in\cL}\left(V_s-2|Q|^{-1}V_s\ast\1_Q+ |Q|^{-2}V_s\ast\1_Q\ast\1_{Q}\right)(x-z).
 \label{eq:V_s_periodic_double_screening}
\end{equation}
When $d-2<s<d$ this coincides with~\eqref{eq:V_s_periodic_Jellium} due to the fact that
\begin{multline*}
\sum_{z\in\cL}\left(-|Q|^{-1}V_s\ast\1_Q+ |Q|^{-2}V_s\ast\1_Q\ast\1_{Q}\right)(x-z)\\=\frac1{|Q|}\int_{\R^d}\left(V_s-|Q|^{-1}V_s\ast\1_Q\right)(x-y)\,\dy=0,
\end{multline*}
as we have seen above. At $s=d-2$, the last integral does not vanish and provides the shift in~\eqref{eq:shift}. The doubly-screened potential~\eqref{eq:V_s_periodic_double_screening} corresponds to the modified background $\rho_b(2\1_\Omega-|Q|^{-1}\1_\Omega\ast\1_Q)$ instead of the simple sharp background $\rho_b\1_\Omega$. This new background has a slightly different behavior in the neighborhood of $\partial\Omega$ but still converges to the constant $\rho_b$ over $\R^d$. This is how the background has to be modified for periodic systems, in order to recover the expected analytic continuation for low values of $s$. Instead of modifying the background it might be possible to instead deform the lattice close to the boundary.

Since periodic systems have some ergodicity, we can also investigate the \emph{Jellium energy per unit volume}, in addition to the potential. For later purposes, we slightly complicate the situation and assume that we have $N$ points in the unit cell $Q$, which are repeated infinitely many times in space. This corresponds to taking $N$ translated copies of the lattice $\cL$. 

\begin{lemma}[Periodic energy]\label{lem:periodic_energy}
Let $\cL$ be an arbitrary lattice of Wigner-Seitz cell $Q$. Consider $N$ distinct points $x_1,...,x_N\in Q$ and the infinite periodic system $X=\{x_i+z,\ i=1,...,N,\ z\in\cL\}$. Denote by $X_R:=\{x_i+z,\ i=1,...,N,\ z\in\cL\cap B_R\}$ and $\Omega_R:=\cup_{z\in \cL\cap B_R}(z+Q)$ the points and cells intersecting the ball $B_R$. Then, for $d-2<s<d$, the Jellium energy per unit volume converges to
\begin{equation}
\lim_{R\to\ii}\frac{\cE_s\big(X_R,\Omega_R,N|Q|^{-1}\big)}{|B_R|}
=\frac1{|Q|}\bigg(\sum_{1\leq j<k\leq N}V_s^\cL(x_j-x_k)+\frac{NM_\cL(s)}{2}+\frac{\log|Q|}{2d}\delta_0(s)\bigg).
\label{eq:limit_periodic}
\end{equation}
If $|Q|^{-1}\int_Qy\,\dy=N^{-1}\sum_{j=1}^Nx_j$, then the same result holds for $\max(-2,d-4)<s<d$.
\end{lemma}

The convergence of the energy holds in a larger range of $s$ than for the potential. This is because the energy can be expressed in terms of the doubly-screened potential~\eqref{eq:V_s_periodic_double_screening}~\cite{LewLie-15,BlaLew-15,Lauritsen-21}. Physically, each cell $z+Q$ contains $N$ points and a uniform background of density $\rho_b=N/|Q|$, hence is neutral. The energy involves the interactions between these cells, which is summable for all $s>d-2$. If each cell has no dipole, this is summable for all $s>d-4$. Recall that $NM_\cL(s)/2$ in~\eqref{eq:limit_periodic} is the interaction of each of the $N$ points with its periodic copies, which is not contained in the first sum. For $N=1$, only the Madelung term remains in Lemma~\ref{lem:periodic_energy} and we conclude that the Jellium energy per unit volume of the infinite lattice $X=\cL$ with uniform background $\rho_b=|Q|^{-1}$ is just $\rho_b\,M_\cL(s)/2$, for the mentioned values of $s$~\cite{Lauritsen-21}.

From Lemmas~\ref{lem:periodic_V} and~\ref{lem:periodic_energy}, we see that the Jellium energy per unit volume of an infinite periodic system with background (the right side of~\eqref{eq:limit_periodic}) equals the analytic continuation of the corresponding energy in the short range case $s>d$ (modulo a sign for $s<0$ and a derivative for $s=0$). Thus the background does the analytic continuation at the level of the energy per unit volume, as it did for the potential.

There are similar results for periodic point processes, which naturally occur at positive temperature. One simple lemma in this direction is as follows.

\begin{lemma}[Energy of periodic point processes]\label{lem:periodic_energy_pt_process}
Let $\mathscr{P}$ be an $\cL$--periodic point process on $\R^d$ with finite local moments, which is repulsive and clustering in the sense that
\begin{equation}
 \left|\rho_\mathscr{P}^{(2)}(x,y)\right|\leq Ce^{-\frac{1}{C|x-y|^a}},\qquad \left|\rho_\mathscr{P}^{(2)}(x,y)-\rho_\mathscr{P}^{(1)}(x)\rho_\mathscr{P}^{(1)}(y)\right|\leq \frac{C}{1+|x-y|^b},
 \label{eq:clustering}
\end{equation}
for some $a,b,C>0$. Denote by $\mathscr{P}_R$ the localization of $\mathscr{P}$ to the set $\Omega_R=\cup_{z\in B_R}\overline{(z+Q)}$. Then for all $s>d$, the average energy per unit volume converges to
\begin{equation}
\lim_{R\to\ii}\frac{\bE_{\mathscr{P}_R}\left[\cE_s(\cdot) \right]}{|\Omega_R|}=\frac1{2|Q|}\iint_{Q\times\R^d}\frac{\rho_\mathscr{P}^{(2)}(x,y)}{|x-y|^s}\,\dx\,\dy.
\label{eq:limit_periodic_clustering_short_range}
\end{equation}
Let $\rho_b=|Q|^{-1}\int_Q\rho_\mathscr{P}^{(1)}$. For $\max(d-2,d-b)<s<d$, the Jellium energy per unit volume converges to
\begin{multline}
\lim_{R\to\ii}\frac{\bE_{\mathscr{P}_R}\left[\cE_s(\cdot,\Omega_R,\rho_b) \right]}{|\Omega_R|}
=\frac1{2|Q|}\iint_{Q\times\R^d}V_s(x-y)\left(\rho^{(2)}_\mathscr{P}(x,y)-\rho_\mathscr{P}^{(1)}(x)\rho_\mathscr{P}^{(1)}(y)\right)\dx\,\dy\\
+\frac1{2|Q|}\iint_{Q\times Q}V^\cL_s(x-y)\rho_\mathscr{P}^{(1)}(x)\rho_\mathscr{P}^{(1)}(y)\, \dx\,\dy
\label{eq:limit_periodic_clustering_Jellium}
\end{multline}
and this is the analytic extension in $s$ of~\eqref{eq:limit_periodic_clustering_short_range}. If $\int_Q y\,\rho^{(1)}_\mathscr{P}(y)\,\dy=0$, the result is valid with $d-4$ instead of $d-2$.
\end{lemma}

This is the periodic extension of similar results in Ref.~\onlinecite{BorSer-13,Leble-16} for translation-invariant point processes  (which, of course, are also periodic). The second term in~\eqref{eq:limit_periodic_clustering_Jellium} is then absent since $\rho^{(1)}_\mathscr{P}$ is constant. The  pointwise upper bound on $\rho^{(2)}_\mathscr{P}$ in~\eqref{eq:clustering} is not at all optimal and has been chosen for simplicity. What we really need is that the point process is sufficiently repulsive so as to make the right side of~\eqref{eq:limit_periodic_clustering_short_range} finite on the diagonal $x=y$. Recall, however, that our short range Riesz Gibbs point processes at inverse temperature $\beta$ all satisfy the exponential repulsion by Lemma~\ref{lem:local_bound_GC}. The clustering property~\eqref{eq:clustering} is also known to be valid, but only for sufficiently small values of $\beta$, as we will explain later in Theorem~\ref{thm:unique}.

\begin{proof}
We remove the index $\mathscr{P}$ on the correlation functions for shortness. The periodicity of the point process implies that $\rho^{(2)}(x+z,y+z)=\rho^{(2)}(x,y)$ and $\rho^{(1)}(x+z)=\rho^{(1)}(x)$ for all $z\in \cL$. In particular, $\rho^{(1)}$ is $\cL$--periodic. Integrating over $y$ in the unit cell $Q$, we obtain from~\eqref{eq:clustering} that $\rho^{(1)}$ is in fact bounded. The localization $\mathscr{P}_R$ has by definition the correlation functions $\rho^{(k)}_{\mathscr{P}_R}=\1_{\Omega_R}^{\otimes k}\rho^{(k)}$ and therefore
$$\bE_{\mathscr{P}_R}\left[\cE_s(\cdot) \right]=\frac12\iint_{\Omega_R\times\Omega_R}V_s(x-y)\rho^{(2)}(x,y)\,\dx\,\dy,$$
by~\eqref{eq:nb_correlation}. Using the periodicity of $\rho^{(2)}$, this converges to the limit mentioned in~\eqref{eq:limit_periodic_clustering_short_range} (see Ref.~\onlinecite[Prop.~2.1]{CatBriLio-01} for a similar argument).

In the Jellium case we have similarly
\begin{align*}
\bE_{\mathscr{P}_R}\left[\cE_s(\cdot,\Omega_R,\rho_b) \right]
&=\frac12\iint_{\Omega_R\times\Omega_R}V_s(x-y)\Big(\rho^{(2)}(x,y)-2\rho^{(1)}(x)\rho_b+\rho_b^2\Big)\dx\,\dy\\
&=\frac12\iint_{\Omega_R\times\Omega_R}V_s(x-y)\Big(\rho^{(2)}(x,y)-\rho^{(1)}(x)\rho^{(1)}(y)\Big)\dx\,\dy\\
&\qquad +\frac12\iint_{\Omega_R\times\Omega_R}V_s(x-y)\big(\rho^{(1)}(x)-\rho_b\big)\big(\rho^{(1)}(y)-\rho_b\big)\dx\,\dy.
\end{align*}
Under our assumption~\eqref{eq:clustering}, when we divide by $|B_R|$ the first term converges to
$$\frac1{2|Q|}\iint_{Q\times\R^d}V_s(x-y)\left(\rho^{(2)}(x,y)-\rho^{(1)}(x)\rho^{(1)}(y)\right)\dx\,\dy,$$
similarly as in the short range case. Since $\rho^{(1)}$ is periodic and $\int\rho^{(1)}=\rho_b|Q|$, the second term can be written as
\begin{align}
&\frac12\iint_{\Omega_R\times\Omega_R}V_s(x-y)\big(\rho^{(1)}(x)-\rho_b\big)\big(\rho^{(1)}(y)-\rho_b\big)\, \dx\,\dy\nn\\
&=\frac12\iint_{Q\times Q}\rho^{(1)}(x)\rho^{(1)}(y)\sum_{z,z'\in B_R}\bigg(V_s(x-y+z-z')-|Q|^{-1}\int_QV_s(u-y+z-z')\,\du\nn\\
&\quad -\frac1{|Q|}\int_QV_s(x-v+z-z')\,\dv+\frac1{|Q|^2}\iint_{Q\times Q}V_s(u-v+z-z')\du\,\dv \bigg) \dx\,\dy.\label{eq:rho_1_P_split}
\end{align}
The function in the parenthesis is summable in $z-z'$ when $d-2<s<d$. Passing to the limit, we obtain
$$\frac1{2|Q|}\iint_{Q\times Q}f(x,y)\rho^{(1)}(x)\rho^{(1)}(y) \dx\,\dy$$
with the $\cL$--periodic kernel (recall~\eqref{eq:V_s_periodic_Jellium})
\begin{align*}
f(x,y)&=\sum_{z\in\cL}\bigg(V_s(x-y+z)-\frac1{|Q|}\int_QV_s(u-y+z)\,\du\\
&\qquad-\frac1{|Q|}\int_QV_s(x-v+z)\,\dv +\frac1{|Q|^2}\iint_{Q\times Q}V_s(u-v+z)\du\,\dv\bigg)\\
&=V_s^\cL(x-y)+\frac1{|Q|}\int_{\R^d}\!\!\bigg(V_s(u+x-y)-V_s(u-y)-V_s(u+x)+\frac1{|Q|}V_s\ast\1_Q(u)\bigg)\du.
\end{align*}
From the behavior in Fourier we see that the last integral vanishes for $s>d-2$.
The proof of the analytic extension is similar to that of Lemma~\ref{lem:analytic_potential}. In the case $s>d-4$ we use that $\rho^{(1)}_\bP$ has no dipole to introduce more cancellation in~\eqref{eq:rho_1_P_split}.
\end{proof}

\begin{remark}[Potential of infinite periodic Coulomb systems]\label{rmk:shift}
Consider a lattice without quadrupole as in Lemma~\ref{lem:periodic_V}. Let
$$\Phi^{(0)}:=\lim_{\ell\to\ii}\left(\sum_{0\neq z\in \cL\cap \in \ell\omega}V_s(z)-\rho_b\int_{\cup_{z\in\ell\omega}(Q+z)}\frac{\dy}{|y|^s}\right)$$
be the interation energy of any point (e.g. the origin) with the rest of the system, where the background is a union of cells.
For $d-2<s<d$ this is equal to the Madelung constant $M_\cL(s)=2\zeta_\cL(s)$ by Lemma~\ref{lem:periodic_V}, which coincides with twice the energy per unit volume by Lemma~\ref{lem:periodic_energy}. However, the two do not coincide in the Coulomb case by~\eqref{eq:shift}. This is very confusing and has generated some controversy in the literature. We learnt this from Ref.~\onlinecite{BorBorShaZuc-88,BorBorSha-89,BorBorStr-14} and this has led to important technical difficulties for the uniform electron gas~\cite{LewLie-15,CotPet-19b,LewLieSei-19b}.
It turns out that the controversy started much earlier in 1979, after a paper of Hall~\cite{Hall-79} based on an unpublished remark by Plaskett in 1959. The conendrum raised by Hall was discussed in several papers in the 80s.\cite{DeWette-80,IhmCoh-80,HalRic-80,Hall-81,AlaJan-81,NijRui-88} The conclusion is that although the energy per unit volume is unambiguous, the point interaction $\Phi^{(0)}$ and the potential $\Phi(x)$ are only defined up to a constant for a periodic Coulomb system. One can get different shifts for the same configuration depending on how the background grows\cite{NijRui-88}. The electric field is itself defined unambiguously. This complicates a lot the study of infinite Coulomb systems. Similar problems occurr for random systems~\cite{CanLahLew-13,Lahbabi-PhD} at positive temperature~\cite{ChoFavGru-80}.
\end{remark}

\subsection{The crystallization conjecture}\label{sec:crystal_conjecture}
In physics, it is often believed that interacting systems are periodic (crystallized) at small (hence in particular at zero) temperature. This is called the \emph{crystallization conjecture}~\cite{Uhlenbeck-68,Radin-87,BlaLew-15} and its physical aspects will be discussed more thoroughly in Section~\ref{sec:transitions}. Although this is certainly not valid for all possible interactions, it is believed that the conjecture holds for Riesz gases.

At zero temperature, the precise statement is that optimal configurations are all Bravais lattices, for all admissible $s$ (that is, are periodic with exactly one point per unit cell). Due to Lemmas~\ref{lem:periodic_V} and~\ref{lem:periodic_energy}, the formulation of this conjecture in terms of the Epstein Zeta function reads as follows.

\begin{conjecture}[Crystallization -- energy]\label{conj:crystal}
For any $s\in (-2,\ii)\setminus\{d\}$, the minimal energy per unit volume of the Riesz gas is
\begin{equation}
e(s)=\min_{\substack{\cL\\ |Q|=1}}\begin{cases}
\zeta_\cL(s)&s>0,\\
\zeta_\cL'(0)&s=0,\\
-\zeta_\cL(s)&s<0.
\end{cases}
\label{eq:conj_crystal}
\end{equation}
\end{conjecture}

We have in fact only defined $e(s)$ for $s\geq\max(0,d-2)$ in Theorem~\ref{thm:limit_C_Jellium} and $s=-1$ in $d=1$, but we mentioned that $e(s)$ probably exists for all $s>-2$. The existence of the limit is part of the conjecture for $s<d-2$.

In dimension $d=1$ no minimum over $\cL$ is necessary and the conjecture involves the Riemann Zeta function. In higher dimensions, it is still needed to determine the optimal lattice $\cL$, which could depend on $s$. From the conjecture it would follow that in the regions where the optimal lattice $\cL$ is independent of $s$, $e(s)$ is analytic (up to the sign change at $s=0$ and the pole at $s=d$).

In dimension $d=1$ the conjecture has been proved in many cases. The easiest is  the Coulomb case $s=-1$ for which the Jellium energy of $N$ points in the interval $I_N=[-N/2,N/2]$ can be expressed as
\begin{align}
\cE_{-1}(X,I_N,1)&=-\sum_{1\leq j<k\leq N}|x_j-x_k|+\sum_{j=1}^N\int_{-\frac{N}{2}}^{\frac{N}{2}}|x_j-y|\,\dy-\frac12\iint_{[-\frac{N}{2},\frac{N}{2}]^2}|x-y|\dx\,\dy\nn\\
&=\sum_{j=0}^{N-1}\left(x_j-j+\frac{N+1}{2}\right)^2+\frac{N}{12}\label{eq:1D_Jellium_energy}
\end{align}
if we order the points so that $x_1\leq\cdots\leq x_N$, see Ref.~\onlinecite[Eq.~(5)]{Kunz-74}. From this we see immediately that the unique minimizer is when the points are in the center of their unit cell, $x_j=j-(N+1)/2$. Passing to the limit $N\to\ii$ provides $e(-1)=1/12=-\zeta(-1)$, as claimed. In the short range case $s\in(1,\ii)$, the result $e(s)=\zeta(s)$ is due to Ventevogel~\cite{Ventevogel-78} (see also Ref.~\onlinecite{MarMayRakSaf-04,CohKum-07,HarLebSafSer-18,BorHarSaf-19}). For $s=0$, the equality is due to Sandier and Serfaty~\cite{SanSer-14a} (see also Ref.~\onlinecite{Leble-15}). The case $s\in(0,1)$ is handled by~\textcite{Leble-16}. For points on the circle $\bS^1$, crystallization holds at any finite $N$ for all $s>-2$ (see Theorem~\ref{thm:limit_C_per} below as well as Ref.~\onlinecite{Toth-59,Ventevogel-78} and Ref.~\onlinecite[Thm.~2.3.3]{BraHarSaf-09}), but the link to the Jellium model does not seem to have been established rigorously for $s<0$, $s\neq-1$. All these results rely deeply on the convexity of $V_s$ as a function of $|x|$ and on the fact that one can express everything in terms of the distance between consecutive points on the line or on a curve.

The \emph{Cohn-Kumar conjecture}~\cite{CohKum-07} states that some special lattices $\cL$ are \emph{universally optimal} in particular dimensions. Universal optimality means that the lattice is the global minimizer for every completely monotone short-range interaction potential of the square distance (equivalently, every Gaussian $e^{-\alpha|x|^2}$). This covers our Riesz potential $V_s$ for all $s>d$. \textcite{PetSer-20} have proved that the Cohn-Kumar conjecture would in fact imply Conjecture~\ref{conj:crystal} not only for $s>d$ but also for all $\max(0,d-2)\leq s<d$. In this case, the optimal lattice $\cL$ in~\eqref{eq:conj_crystal} is therefore independent of $s$ on $[\max(0,d-2),\ii)\setminus\{d\}$, so that $e(s)$ is analytic and equals the corresponding Epstein Zeta function $\zeta_\cL(s)$.

With the above definition, the universal optimality of $\cL=\Z$ in dimension $d=1$ was in fact proved by Ventevogel and Nijboer in Ref.~\onlinecite{VenNij-79-2}. More was later shown in Ref.~\onlinecite{CohKum-07} (see also Ref.~\onlinecite{,CohKumMilRadVia-19_ppt}). The Cohn-Kumar conjecture in dimension $d=2$ states that the triangular lattice is universally optimal and it is still an open problem. It would imply that $e(s)$ is also analytic in dimension $d=2$ (again up to the pole at $s=2$ and the derivative at $s=0$).

The Cohn-Kumar conjecture was recently proved in the groundbreaking work Ref.~\onlinecite{CohKumMilRadVia-19_ppt} by Cohn \emph{et al.} in dimensions $d\in\{8,24\}$ with, respectively, the $E_8$ and Leech lattices. The current situation is, therefore, the following.

\begin{theorem}[Crystallization]\label{thm:crystallization}
Conjecture~\ref{conj:crystal} is true for all $s\in\{-1\}\cup[0,1)\cup(1,\ii)$ in dimension $d=1$ and all $s\in [d-2,\ii)\setminus\{d\}$ in dimensions $d\in\{8,24\}$.
\end{theorem}

There is no universally optimal lattice in the physical dimension $d=3$. Numerical computations~\cite{Fuchs-35,PlaNijHov-51,ColMar-60,Nijboer-75,BonMar-77,BlaLew-15,BetSamTra-21_ppt} suggest that the optimal lattice is Face-Centered Cubic (FCC) for $s\geq3/2$ and Body-Centered Cubic (BCC) for $0<s\leq3/2$. Both should be optimal at $s=3/2$. This conjecture can be traced back to Nijboer in 1975 for the specific values $s\in\{1,3/2,2\}$ (see Ref.~\onlinecite[p.~83]{Nijboer-75}), and is now usualy called the \emph{Sarnak-Str\"ombergsson conjecture}\cite{SarStr-06}. If the conjecture is correct, $e(s)$ is therefore an analytic function on $(0,\ii)\setminus\{3/2,3\}$, with a jump in its derivative at $s=3/2$ (first order phase transition) and a pole at $s=3$. The analyticity over $(0,\ii)\setminus\{3\}$ is wrongly conjectured in Ref.~\onlinecite{BraHarSaf-12,BorHarSaf-19}.

Conjecture~\ref{conj:crystal} has a long history in physics and mathematical physics. It is a general principle that zero-temperature systems are ``very often'' crystallized. This fundamental problem was formalized for instance in Ref.~\onlinecite{Uhlenbeck-68,Radin-87,BlaLew-15}. The case of the Riesz potential was considered in many physical works, for both $s>d$ and $s<d$, as we will outline in Section~\ref{sec:transitions}. In the 3D Coulomb case (or more generally $s=1$ and $d\in\{1,2,3\}$), the conjecture is due to Wigner~\cite{Wigner-34,Wigner-38} in 1934. When discussing electrons, he wrote in Ref.~\onlinecite[p.~1010]{Wigner-34} that \emph{``if the electrons had no kinetic energy, they would settle in configurations which correspond to the absolute minima of the potential energy. These are closed-packed lattice configurations, with energies very near to that of the body-centered lattice.''} A year later, Fuchs~\cite{Fuchs-35} made the careful computation of the Jellium energy of the BCC lattice $\zeta_{\rm BCC}(1)$ and confirmed it is better than other lattices (see also Ref.~\onlinecite{SarMerCal-76}). That this is related to the Epstein Zeta function was known since works of Madelung~\cite{Madelung-18}, Ewald~\cite{Ewald-21} and Emersleben~\cite{Emersleben-23}. The analytic extension is implicit in the formulas involving Theta functions used in the literature\cite{PlaNijHov-51,ColMar-60,Sholl-67,BonMar-77,BorGlaMPh-13,GiuVig-05} but to our knowledge it was formalized first in Ref.~\onlinecite{BorBorTay-85,BorBorShaZuc-88,BorBorSha-89,BorBorStr-14}.

A \emph{Wigner crystal} of (quantum) electrons is very hard to observe in the laboratory, due to the necessity of working at very low temperature and density (to suppress quantum fluctuations). This is a very active field of research in experimental physics. Crystallization has only been observed very recently, for one-dimensional systems in Ref.~\onlinecite{Shapir-etal-19}, in 2D in Ref.~\onlinecite{Li_etal-21} and in 3D in Ref.~\onlinecite{Smolenski_etal-21,Zhou_etal-21}, thus largely confirming Wigner's prediction.

The 2D Coulomb case ($s=0$ in $d=2$) also has a long history in physics, since it describes the behavior of vortices in superfluids~\cite{Aftalion-07} and superconductors~\cite{BetBreHel-94,SanSer-12,Serfaty-14b}. In this setting, the expected minimizer (the triangular lattice) is often called the \emph{Abrikosov lattice}~\cite{Abrikosov-57}. This was observed very quickly in experiments for superconductors~\cite{CriJacMadFar-64,Hess_etal-89} and much later for rotating Bose-Einstein condensates~\cite{Dalibard-00,Ketterle-01}.

There are many computations in the physics literature supporting the conjecture. We defer the discussion of these results to Section~\ref{sec:transitions}, where we add the temperature. Conjecture~\ref{conj:crystal} also appears in many forms in the mathematical literature. Part of it can be found  in Ref.~\onlinecite{BetBreHel-94,AftBla-06,SarStr-06,CohKum-07,SanSer-12,BraHarSaf-12,SanSer-15,BetSan-18,BorHarSaf-19} but there are many other examples. We refer to Ref.~\onlinecite{BlaLew-15} for more details.

\begin{remark}[Optimal lattices]
Finding the optimal lattice $\cL$ for the minimum on the right side of~\eqref{eq:conj_crystal} can be particularly challenging. The works in this direction are reviewed in detail in Ref.~\onlinecite{BlaLew-15}. In the two-dimensional case, it is known that the triangular lattice is the best among lattices for all $s\geq0$~\cite{Rankin-53,Cassels-59,Ennola-64,Diananda-64,Montgomery-88,NonVor-98,CheOsh-07,SanSer-12}. In $d=3$, it is only known that the FCC lattice is the best for large values of $s$\cite{Ennola-64b,Ryshkov-73}. Stability properties of 3D lattices under small deformations are studied in Ref.~\onlinecite{Fields-80,Gruber-12,Betermin-19}.
\end{remark}

Most authors only consider the minimal energy and do not attempt to prove that minimizers are really crystallized in the thermodynamic limit\cite{KatDun-84,BlaLew-15}.
We would like to suggest a more general problem and instead ask whether all the possible \emph{infinite Riesz equilibrium configurations} are crystallized. This is more general since, in principle, some of these optimal configurations could only occur in a thermodynamic limit with exotic boundary conditions. Nothing like this has been shown to our knowledge, even in dimension $d=1$ where we nevertheless expect the proof to be quite similar to the existing results.

Recall that we have defined Riesz equilibrium configurations at any chemical potential $\mu>0$ in Definition~\ref{def:equilibrium} when $s>d$. When $s<d$ we can introduce a similar definition using Theorem~\ref{thm:infinite_conf_Jellium} for $\max(d-2,0)\leq s<d$, with $\mu\in\R$ and $\rho_b>0$. We can choose any representation for the potential $\Phi(x)$, at the expense of shifting $\mu$. The corresponding conjecture is the following.

\begin{conjecture}[Crystallization -- configurations]\label{conj:crystal_state}
The Riesz equilibrium infinite configurations are exactly the $r(\cL+\tau)$ with $\tau\in\R^d$ and $\cL$ the minimizers for the right side of~\eqref{eq:conj_crystal}, where $r=\mu^{-1/s}$ for $s>d$ and $r=\rho_b^{-1/d}$ for $s<d$.
\end{conjecture}

\subsection{Limit with periodic boundary conditions}\label{sec:periodic_BC}

We have discussed infinite periodic systems and their optimality in some cases. We now turn to another important use of periodic potentials, even in situations where crystallization is not expected. The idea is to replace the thermodynamic limit studied in Sections~\ref{sec:thermo_limit_short_range}--\ref{sec:thermo_limit_long_range} by that of a periodic system with an increasing period $\ell$. This corresponds to working on the `torus' $\R^d/\ell\cL$ instead of the container $\Omega=\ell\omega$ with hard walls. Since the torus has no boundary, it is expected that the thermodynamic limit for this model should be somewhat better behaved in some delicate cases (e.g. very low $s$). This will be confirmed by the main result below. Even if the periodic model seems somewhat less physical than having a finite system in a container, many practical simulations on Jellium are made in the periodic setting~\cite{Mazars-11}. The expectation is that the final result must be the same as for the usual thermodynamic limit. One advantage of the periodic problem is that the system is translation-invariant (in the sense of torus translations, that is, modulo $\ell\cL$). Working with translation-invariant point processes has several well-known advantages in statistical mechanics although this restriction could in principle prevent from revealing the detailed properties of all the equilibrium states.

In order to describe the problem, we fix again a lattice $\cL$ of Wigner-Seitz unit cell $Q$ which, for simplicity, has volume $|Q|=1$. We consider the scaled lattice $\ell\cL$ by an amount $\ell>0$ and denote by
\begin{equation}
\cE_s^{\ell\cL}(x_1,...,x_N):=\sum_{1\leq j<k\leq N}V_s^{\ell\cL}(x_j-x_k)+\frac{NM_{\ell\cL}(s)}2
\label{eq:cE_s_per}
\end{equation}
the periodic energy which we have found in Lemma~\ref{lem:periodic_energy}. Up to a negligible logarithmic correction when $s=0$, this is the energy per unit cell of an infinite configuration with the $x_j$ repeated $(\ell\cL)$--periodically. It contains the somewhat unphysical interactions with the mirror images, which are however further and further away when $\ell$ increases. When $s<d$, the infinite system is immersed in a uniform background (properly chosen for $s\leq d-2$). The latter does not appear explicitly in the expression of the energy~\eqref{eq:cE_s_per} since $\int_{\ell Q}V^\cL_s=0$ by definition for $s<d$. In some sense, the background is only used to justify the use of the periodic potential $V^{\ell\cL}_s$. Neutrality is thus automatically imposed in the periodic model.

Note the scaling relation
\begin{equation}
 V_s^{\ell\cL}(\ell x)=\ell^{-s}V_s^{\cL}( x),\qquad \forall s\in(-2,\ii)\setminus\{d\},\quad \forall \ell>0.
 \label{eq:scaling_V_per}
\end{equation}
Similarly, the Madelung constant satisfies
\begin{equation}
M_{\ell\cL}(s)=\begin{cases}
\dps\frac{M_{\cL}(s)}{\ell^s}&\text{for $s\neq0$,}\\
M_{\cL}(0)+\log\ell&\text{for $s=0$.}
\end{cases}
\label{eq:scaling_Madelung}
\end{equation}
For $s>0$ we have therefore $M_{\ell\cL}(s)=o(1)$ and the last constant in~\eqref{eq:cE_s_per} can be safely removed. But for $s\leq0$ it is divergent. In fact, for $s<0$ we have $M_{\ell\cL}(s)=V^{\ell \cL}_s(0)$ and therefore the energy~\eqref{eq:cE_s_per} can be rewritten as
$$\cE_s^{\ell\cL}(x_1,...,x_N):=\frac12\sum_{1\leq j,k\leq N}V_s^{\ell\cL}(x_j-x_k)\qquad\text{for $-2<s<0$.}$$
Finally, we remark that
$$\lim_{\ell\to\ii}\left(V_s^{\ell\cL}(x)-M_{\ell\cL}(s)\right)=V_s(x)\qquad\text{for $s\in(-2,d)\setminus\{d\}$ and any fixed $x$.}$$
In the limit the constant $M_{\ell\cL}(s)$ can be removed for $s>0$. This shows that for two points $x_j,x_k\in\ell Q$ located at a finite distance to each other, the periodic interaction is essentially the same as that of the whole space (up to a large constant for $s\leq0$).

We can now look at the thermodynamic limit for this new problem as we discussed in Sections~\ref{sec:thermo_limit_short_range} and \ref{sec:thermo_limit_long_range}, with $V_s^{\ell\cL}$ in place of $V_s$, assuming the points are in the rescaled unit cell $\ell Q$. The canonical energy and free energy are defined by
$$\boxed{E_s^\cL(N,\ell Q):=\min_{x_1,...x_N\in\overline{\ell Q}}\cE_s^{\ell\cL}(x_1,...,x_N)}$$
and
$$\boxed{F_s^\cL(\beta,N,\ell Q)=-\frac1\beta\log\left(\frac1{N!}\int_{(\ell Q)^N}e^{-\beta\cE_s^{\ell \cL}(x_1,...,x_N)}\,\dx_1\cdots\dx_N\right).}$$
Since an $\ell$--periodic system is also $k\ell$--periodic, the variational principle provides
\begin{align}
  E^\cL_s(k^dN,k\ell Q)&\leq k^d\,E^\cL_s(N,\ell Q)+\frac{\log\ell}{2}(k^d-1)\delta_0(s),\nn\\
  F^\cL_s(\beta,k^dN,k\ell Q)&\leq k^d\,F^\cL_s(\beta,N,\ell Q)+\frac{\log\ell}{2}(k^d-1)\delta_0(s),
 \label{eq:cL_monotonicity}
\end{align}
for all $k\in\N$ and $\ell>0$. There is a similar estimate for any sublattice $\cL'\subset\ell\cL$, not just $\cL'=k\ell\cL$. 
The energy and free energy per unit volume are thus decreasing along sequences of the form $\ell_j=k^j\ell_0$ and $N_j=k^{jd}N_0$ (up to an unimportant logarithmic correction at $s=0$). The following is the periodic equivalent to Lemmas~\ref{lem:simple_estim_E_s} and~\ref{lem:lower_bound_Jellium}.

\begin{lemma}[Stability bounds with periodic boundary conditions]\label{lem:bounds_per}
Let $\cL$ be a lattice with normalized unit cell $Q$. Let $s\in(-2,\ii)\setminus\{d\}$. There exists a constant $c_\cL(s)$ (depending only on $\cL$ and $s$) such that
$$c_\cL(s) \left(\frac{N}{\ell^d}\right)^{\frac{s}{d}}N\leq E_s^\cL(N,\ell Q)\leq \frac{M_{\cL}(s)}2 \left(\frac{N}{\ell^d}\right)^{\frac{s}{d}}N,$$
$$c_\cL(s) \left(\frac{N}{\ell^d}\right)^{\frac{s}{d}}N+\beta^{-1}N\left(\log\frac{N}{\ell^d}-1\right)\leq F_s^\cL(\beta,N,\ell Q)\leq \frac{M_{\cL}(s)}2 \left(\frac{N}{\ell^d}\right)^{\frac{s}{d}}N+\beta^{-1}N\log\frac{N}{\ell^d}.$$
For $s<0$ and $s>d$ we can take $c_\cL=0$.
\end{lemma}

\begin{proof}
From the scaling relations~\eqref{eq:scaling_V_per} and~\eqref{eq:scaling_Madelung}, we can always assume that $N=\ell^d$. For the upper bound we use~\eqref{eq:cL_monotonicity} and obtain
$$F_s(\beta,\ell^d,\ell Q)\leq N\,F_s(\beta,1,Q)=\frac{M_{\cL}(s)}2N.$$
We have used here that the energy of one point is $\cE_s^{\cL}(x_1)=M_{\cL}(s)/2$ by~\eqref{eq:cE_s_per}.
When $s>d$ it is clear that $\cE_s^{\ell\cL}>0$. When $s<0$ we have $V_s^{\ell\cL}(0)=M_{\ell\cL}(s)$ so that
$$\cE_s^{\ell\cL}(x_1,...,x_N)=\frac12\sum_{1\leq j,k\leq N}V_s^{\ell\cL}(x_j-x_k)=\frac1{2(2\pi)^{d/2}}\sum_{k\in\cL^*/\ell}\widehat{V_s^{\ell\cL}}(k)\left|\sum_{j=1}^Ne^{ix_j\cdot k}\right|^2\geq0$$
as well. In these two cases we can thus take $c_\cL(s)=0$.
From Lemmas~\ref{lem:periodic_energy} and~\ref{lem:lower_bound_Jellium} we obtain
$E_s^\cL(\ell^d,\ell Q)\geq -c_1N$ for all $s> d-4$, with $c_1$ the constant in Lemma~\ref{lem:lower_bound_Jellium}. This implies that $F_s^\cL(\ell^d,\ell Q)\geq -(c_1+\beta^{-1})N$. In fact, from the existence of the thermodynamic limit in Theorems~\ref{thm:limit_C} and~\ref{thm:limit_C_Jellium} and Lemma~\ref{lem:periodic_energy}, we also have
$$\frac{E_s^\cL(\ell^d,\ell Q)}{\ell^d}\geq e(s),\qquad \frac{F_s^\cL(\beta,\ell^d,\ell Q)}{\ell^d}\geq f(s,\beta,1),\qquad \text{for all $s\geq\max(0, d-2)$.}$$
For lower values of $s$ one can get a universal lower bound by arguing as in Lemma~\ref{lem:lower_bound_Jellium} and modifying the background as outlined in Section~\ref{sec:periodic_potential}, but the result can also be read in Ref.~\onlinecite[Lem.~12]{HarSafSimSu-17} and Ref.~\onlinecite[Lem.~10.8.3]{BorHarSaf-19}.
\end{proof}

The following says that the thermodynamic limit exists for the canonical periodic problem for all  $s>-1$ and that it is independent of the lattice shape. Furthermore, it coincides with the limit with a sharp container for $s\geq d-2$.

\begin{theorem}[Thermodynamic limit with periodic boundary conditions]\label{thm:limit_C_per}
Let $\cL$ be a lattice with normalized unit cell $Q$. Let $\rho>0$, $d\geq1$ and $s\in(-1,\ii)\setminus\{d\}$. If $d=1$ we also allow $s=-1$. We have the limits
\begin{equation}
\lim_{\substack{N\to\ii\\ \frac{N}{\ell^d}\to\rho}}\frac{E_s^\cL(N,\ell Q)}{\ell^d}=e^{\rm per}(s)\rho^{1+\frac{s}d}-\delta_0(s)\frac{\rho\log\rho}{2d},\qquad
\lim_{\substack{N\to\ii\\ \frac{N}{\ell^d}\to\rho}}\frac{F^\cL_s(\beta,N,\ell Q)}{\ell^d}=f^{\rm per}(s,\beta,\rho),
\label{eq:thermo_limit_free_energy_per}
\end{equation}
where $e^{\rm per}(s)$ and $f^{\rm per}(s,\beta,\rho)$ are independent of the lattice $\cL$. In dimension $d=1$, the first limit is in fact valid for $s\in(-2,\ii)\setminus\{1\}$ and equals
\begin{equation}
 e^{\rm per}(s)=\begin{cases}
{\rm sgn}(s)\zeta(s)&\text{for $s\neq0$ in $d=1$,}\\
\zeta'(0)&\text{for $s=0$ in $d=1$.}
\end{cases}
 \label{eq:e_per_zeta_1D}
\end{equation}
In arbitrary dimension $d\geq1$, we have
\begin{equation}
 e^{\rm per}(s)=e(s),\qquad f^{\rm per}(s,\beta,\rho)=f(s,\beta,\rho),\qquad \text{for all $s\geq \max(0,d-2)$}
 \label{eq:equality_periodic}
\end{equation}
with the functions from Theorems~\ref{thm:limit_C} and~\ref{thm:limit_C_Jellium}.
\end{theorem}

In the short range case $s>d$, the result is due to~\textcite{FisLeb-70} (see also Ref. \onlinecite{AngNen-73} and \onlinecite{HarSafSim-14}). In fact, $V_s^\cL$ is positive for $s>d$ and all the local bounds of Section~\ref{sec:local_bd_short} hold in the periodic case. This can be used to prove that $E_s^\cL(N,\ell Q)=E_s(N,\ell Q)+o(\ell^d)$. The 1D Coulomb case $s=-1$ is treated by Kunz~\cite{Kunz-74}. The existence of the limit~\eqref{eq:thermo_limit_free_energy_per} for $-1<s<d$ and $\beta\leq\ii$ is an improvement on a result of Hardin \emph{et al} in Ref.~\onlinecite{HarSafSimSu-17}, who only dealt with $0\leq s<d$ and $\beta=+\ii$. The proof is provided below for the convenience of the reader. In dimension $d=1$, the universal optimality of $\cL=\Z$~\cite{CohKum-07,BraHarSaf-09,BorHarSaf-19,PetSer-20} implies the formula~\eqref{eq:e_per_zeta_1D} for $e^{\rm per}(s)$ in terms of the Riemann Zeta function. In fact, there is even equality for every finite $N$~\cite{CohKum-07,BorHarSaf-19}:
\begin{equation}
E^{\Z}_s(N,\ell Q)=\frac{N^{1+s}}{2\ell^s}M_\Z(s)-\frac{N}2\log\frac{N}{\ell}\delta_0(s),\qquad\forall N\in\N,\quad\forall \ell>0,\quad\forall s>-2.
\label{eq:optimality_per_Z}
\end{equation}
Recall that $M_\Z(s)=2{\rm sgn}(s)\zeta(s)$ for $s\neq0$ and $M_\Z(0)=2\zeta'(0)$.
For $s\geq \max(0,d-2)$, the limit~\eqref{eq:thermo_limit_free_energy_per} and the equality~\eqref{eq:equality_periodic} are in a series of works by Serfaty \emph{et al}  Ref.~\onlinecite{SanSer-15,PetSer-17,RouSer-16,LebSer-17,ArmSer-21} (see also Ref.~\onlinecite{CotPet-19b}). A simple proof in the Coulomb case was recently provided in Ref.~\onlinecite{LewLieSei-19b,Lauritsen-21} for $T=0$.

\begin{proof}[Proof of the limit~\eqref{eq:thermo_limit_free_energy_per}]
We closely follow Ref.~\onlinecite{HarSafSimSu-17} which only considered $0\leq s<d$ at $\beta^{-1}=0$. By scaling we can assume that $N=\ell^d$. We only consider $F^\cL_s(\beta,N,\ell Q)$. For $E_s^\cL(N,\ell Q)$ one should simply replace $\beta^{-1}$ by 0 in all our estimates.

Let $\ell_0\geq 3$ be any fixed integer. To any $\ell$ large enough we associate $k$ such that $(k-2)\ell_0\leq \ell <(k-1)\ell_0$, that is, $k=1+\lfloor\ell/\ell_0\rfloor$. We then let $\ell':=k\ell_0\geq \ell+\ell_0\geq \ell+3$ and $N'=(\ell')^d\geq N+3$. Let $\bP'$ be the Gibbs measure for $F^\cL_s(\beta,N',\ell' Q)$ and $\bQ':=\bP'/N'!$ be the associated probability over $(\ell' Q)^N$. In order to compare different values of the number of points, it is easier to rescale everything and go back to the macroscopic scale as follows
\begin{align*}
&F^\cL_s(\beta,N',\ell' Q)-\frac{N'M_{\ell'\cL}(s)}2\\
&\qquad=\frac{N'(N'-1)}{2}\int_{(\ell'Q)^{N'}} V_s^{\ell'\cL}(x_1-x_2)\,\rd\bQ' +\beta^{-1}\int_{(\ell'Q)^{N'}}\bQ'\log(N'!\bQ')\\
&\qquad=\frac{N'(N'-1)(N')^{-\frac{s}d}}{2}\int_{Q^{N'}} V_s^{\cL}(x_1-x_2)\,\rd\widetilde\bQ'+\beta^{-1}\int_{Q^{N'}}\widetilde\bQ'\log(\widetilde\bQ') +T\log\left(\frac{N'!}{(N')^{N'}}\right)
\end{align*}
where $\widetilde\bQ'$ is the rescaled probability on $Q^N$. Since $N'\geq N+3$, we can write
$$\int_{Q^{N'}} V_s^{\cL}(x_1-x_2)\,\rd\widetilde\bQ'=\int_{Q^N} V_s^{\cL}(x_1-x_2)\,\rd\widetilde\bQ$$
with the $N$th marginal $\widetilde\bQ:=\int_{Q^{N'-N}}\widetilde \bQ'$. Next we use the subadditivity of the entropy\cite{Wehrl-78}
$$\int_{Q^{N'}}\widetilde\bQ'\log(\widetilde\bQ')\geq \int_{Q^N}\widetilde\bQ\log(\widetilde\bQ)+\int_{Q^{N'-N}}\widetilde{\mathds R}\log(\widetilde{\mathds R})\geq \int_{Q^N}\widetilde\bQ\log(\widetilde\bQ),$$
with the complementary $(N'-N)$th marginal $\widetilde{\mathds R}:=\int_{Q^{N}}\widetilde \bQ'$. The terms are here all non-negative by Jensen, since $|Q|=1$. Scaling back everything to $\ell Q$, we arrive at
\begin{align}
&F^\cL_s(\beta,N',\ell' Q)\nn\\
&\geq\min_{\bQ}\Bigg\{\frac{N'(N'-1)(N')^{-\frac{s}d}}{N(N-1)N^{-\frac{s}d}}\int_{(\ell Q)^N}\cE_s^{\ell\cL}(x_1,...,x_N)\,\rd \bQ+\beta^{-1}\int_{(\ell Q)^N}\bQ\log(N!\bQ)\Bigg\}\nn\\
&\quad-(N')^{1-\frac{s}{d}}\frac{N'-N}{N-1}\frac{M_{\cL}(s)}2+T\log\left(\frac{N'!N^N}{(N')^{N'}N!}\right)+\delta_0(s)N'\left(\log\ell'-\frac{N'-1}{N-1}\log\ell\right).
\label{eq:lower_bd_periodic_F}
\end{align}
The minimum on the right side is very close to $F^\cL_s(\beta,N,\ell Q)$. In fact, we have $N'=N+K$ with $K\leq \ell_0^d(k^d-(k-2)^d)\leq CN^{1-\frac1d}$, hence
$\frac{N'(N'-1)(N')^{-s/d}}{N(N-1)N^{-s/d}}=1+O(N^{-\frac1d}).$
The first term on the second line is the average energy and it is of order $N$ by Lemma~\ref{lem:bounds_per}. We obtain
\begin{equation*}
F^\cL_s(\beta,N',\ell' Q)\geq F^\cL_s(\beta,N,\ell Q)+O\left(N^{1-\frac1d}+N^{1-\frac{s+1}d}+\delta_0(s)N^{1-\frac1d}\log N\right).
\end{equation*}
The right side is a $o(N)$ whenever $s>-1$. After passing to the limit $N\to\ii$, we get
\begin{align*}
\limsup_{\ell\to\ii}\frac{F^\cL_s(\beta,\ell^d,\ell Q)}{\ell^d}&\leq \limsup_{k\to\ii}\frac{F^\cL_s(\beta,(k\ell_0)^d,k\ell_0 Q)}{k^d\ell_0^d}+\frac{\log(\ell_0)}{2\ell_0^d}\delta_0(s)\\
&\leq \frac{F^\cL_s(\beta,\ell_0^d,\ell_0 Q)}{\ell_0^d}+\frac{\log(\ell_0)}{2\ell_0^d}\delta_0(s)
\end{align*}
where we used~\eqref{eq:cL_monotonicity} in the last inequality. Taking now $\ell_0\to\ii$ we conclude that the liminf equals the limsup, hence that the limit~\eqref{eq:thermo_limit_free_energy_per} exists for the lattice $\cL$.

It remains to show that the limit is independent of the lattice. The argument is exactly as in Ref.~\onlinecite{HarSafSimSu-17}. Let us write $\cL'=M\cL$ where $\det(M)=1$ and approximate $M$ by a sequence of matrices $M_k\in \Z^{d^2}/q_k$ with rational coefficients and $\det(M_k)=1$. Then $q_k k\cL'$ is a sub-lattice of $\cL_k:=M_k^{-1}M\cL$ so that, by~\eqref{eq:cL_monotonicity},
$$\frac{F^{\cL'}_s(\beta,(q_kk\ell_0)^d,q_kk\ell_0 Q')}{(q_kk\ell_0)^d}\leq \frac{F^{M_k^{-1}M\cL}_s(\beta,\ell_0^d,\ell_0 Q_k)}{\ell_0^d}+\frac{\log(\ell_0)}{2\ell_0^d}\delta_0(s).$$
Passing to the limit $k\to\ii$ at fixed $\ell_0$, we find
$$\lim_{\ell\to\ii}\frac{F^{\cL'}_s(\beta,\ell^d,\ell Q)}{\ell^d}\leq \frac{F^{\cL}_s(\beta,\ell_0^d,\ell_0 Q)}{\ell_0^d}.$$
On the right side we used the continuity of the free energy with respect to deformations of the lattice, at finite $\ell_0$. The result follows after passing to the limit $\ell_0\to\ii$ and then exchanging the role of $\cL$ and $\cL'$.
\end{proof}

We have mentioned the works of Serfaty \emph{et al} in Ref.~\onlinecite{SanSer-15,PetSer-17,RouSer-16,LebSer-17,ArmSer-21} where the equality~\eqref{eq:equality_periodic} with $e(s)$ and $f(s,\beta,\rho)$ is shown. In fact these authors can handle more general boundary conditions. Consider the potential
$$\Phi(x)=\sum_{j=1}^N V_s(x-x_j)-\rho_b\int_{\Omega}V_s(x-y)\,\dy,$$
generated by $N$ points $x_j$ in a background $\rho_b\1_\Omega$. In the Coulomb case we have Poisson's equation
$$-\Delta\Phi = |\bS^{d-1}|\bigg(\sum_{j=1}^N\delta_{x_j}-\rho_b\1_\Omega\bigg).$$
The Jellium energy $\cE_s(X,\Omega,\rho_b)$ can be expressed in terms of $\Phi$ only as follows
\begin{equation}
 \cE_s(X,\Omega,\rho_b)=\lim_{\eps\to0^+}\frac1{2|\bS^{d-1}|}\int_{\R^d\setminus\cup_jB_\eps(x_j)}|\nabla\Phi|^2-\frac{NV_s(\eps)}{2}
 \label{eq:energy_electric_field}
\end{equation}
The limit and the second term are necessary since point charges have an infinity self-energy. Instead of removing the balls one can also regularize the points~\cite{RouSer-16}. The conclusion is that it suffices to study the electric field $E=-\nabla \Phi$, seen as a function of the $x_j$. Apart from the technical difficulties arising from the divergence at the points, working with the electric field has many advantages (already noticed way back in Ref.~\onlinecite{Lenard-63,PenSmi-72,PenSmi-75,AizMar-80}).
In the periodic setting, we have similarly
$$-\Delta_{\ell\cL}\Phi_{\ell\cL}(x) = c_{d,s}\left(\sum_{j=1}^N\delta_{x_j}-\rho_b\1_{\ell Q}\right)$$
where $-\Delta_{\ell\cL}$ is the Laplacian on the domain $\ell Q\subset\R^d$ with periodic boundary conditions (essentially the Laplace-Beltrami operator on the torus) and $\Phi_{\ell\cL}(x)=\sum_{j=1}^NV_s^{\ell\cL}(x-x_j)$. We then have a formula similar to~\eqref{eq:energy_electric_field} for our periodic energy $\cE_s^{\ell\cL}(x_1,...,x_N)$ with an integral over $\ell Q$. These remarks suggest to consider other boundary conditions for the Laplacian on $\ell Q$. The Dirichlet and Neumann Laplacians are natural choices, well suited to comparison principles through the Dirichlet-Neumann bracketing method (Ref.~\onlinecite[Sec.~XIII.15]{ReeSim4}). In Ref.~\onlinecite[Prop.~5.6]{RouSer-16} and Ref.~\onlinecite{LebSer-17} it is shown that the Dirichlet and Neumann thermodynamic limit are the same, which then implies that all the considered limits coincide. One important fact is that configurations satisfying Neumann's boundary condition can easily be pasted together\cite{PenSmi-72}.

The proof in the Coulomb case is based on Poisson's equation. This approach was generalized to all $d-2<s<d$ in Ref.~\onlinecite{PetSer-17,LebSer-17} using instead the Caffarelli-Silvestre equation~\cite{CafSil-07}
$$-{\rm div}\left(|x_{d+1}|^{s+1-d}\nabla\tilde \Phi\right)=\tilde c_{d,s}\left(\sum_{j=1}^N\delta_{x_j}-\rho_b\1_{\ell Q}\right)\otimes\delta_0(x_{d+1})$$
for the extension $\tilde\Phi$ of $\Phi$ to $d+1$ dimensions, which we have already mentioned in the proof of Lemma~\ref{lem:separation}. This appears to be the main technical reason for the constraint $s\geq d-2$ in Ref.~\onlinecite{SanSer-15,PetSer-17,RouSer-16,LebSer-17}.

The periodic problem is not the only way to ensure nice properties such as translation-invariance. For the \emph{Uniform Riesz Gas}, studied in Ref.~\onlinecite{LewLieSei-18,LewLieSei-19,LewLieSei-19b,CotPet-19,CotPet-19b} and defined in Remark~\ref{rmk:UEG} below, one minimizes the Riesz energy or free energy (with background for $s<d$) under the additional constraint that the density is \emph{exactly equal to the background density $\rho_b\1_\Omega$}, hence the name ``uniform''. This does not make the point process translation-invariant but, at least, the one-point correlation function $\rho^{(1)}_\mathscr{P}$ is constant by definition. The interaction with the background then simplifies with the background self-interaction, and the energy becomes
$$\bE\left[\cE_s(\cdot,\Omega,\rho_b)\right]=\bE\left[\sum_{1\leq j<k\leq N}V_s(x_j-x_k)\right]-\frac{\rho_b^2}{2}\iint_{\Omega\times\Omega}V_s(x-y)\,\dx\,\dy.$$
In Ref.~\onlinecite{LewLieSei-18} it is shown that the thermodynamic limit exists for all $s>0$ and that the canonical and grand-canonical energies coincide. The equality with $e(s)$ and $f(s,\beta,\rho)$ goes through a comparison with the periodic problem~\cite{CotPet-19b,LewLieSei-19b} and it is so far only known for $\max(0,d-2)\leq s<d$. Let us notice that for the 1D uniform Riesz gas, crystallization was proved by Colombo, Di Marino and Stra~\cite{ColPasMar-15}. Due to the equality with $e(s)$, this provides a different proof of Theorem~\ref{thm:crystallization} in dimension $d=1$.

\begin{remark}[Free energy of the log gas]\label{rmk:log_gas}
As will be explained later, the log gas $s=0$ in dimension $d=1$ is an integrable system, like for $s=2$ (Remark~\ref{rmk:CSM}). The periodic free energy $F^\Z_0(\beta,N,[0,N])$ can be computed exactly, using Selberg's integral formula
\begin{equation}
\int_0^{2\pi}\cdots \int_0^{2\pi}\prod_{1\leq j<k\leq N}|e^{ i\theta_j}-e^{ i\theta_k}|^\beta\,\rd\theta_1\cdots\rd\theta_N=\frac{(2\pi)^N\Gamma(1+\frac12\beta N)}{\Gamma(1+\frac12\beta)^N}.
\label{eq:Selberg}
\end{equation}
This formula was conjectured by Dyson in~\onlinecite[Eq.~(133)]{Dyson-62a} and later proved in Ref.~\onlinecite{Wilson-62,Gunson-62,Good-70}. Using Theorem~\ref{thm:limit_C_per}, one can then obtain in the limit $N\to\ii$
\begin{equation}
f(0,\beta,\rho)=\rho\left(\beta^{-1}\log\Gamma\left(1+\frac\beta2\right)-\frac{\log(\pi\beta)}{2}\right)+\frac{2-\beta}{2\beta}\rho(\log\rho-1).
\label{eq:f_log_gas}
\end{equation}
This is a real-analytic function of $(\beta,\rho)$ on $(0,\ii)^2$. Note the particular value $\beta=2$ (the sine--2 point process at which there is a BKT transition, see Sections~\ref{sec:random_matrices} and~\ref{sec:phase_diag_long} below). The free energy is a convex function of $\rho$ for $\beta<2$ but it is concave for $\beta>2$ and linear at $\beta=2$. In the limit $\beta\to\ii$, we recover $e(0)=\zeta'(0)=-\log(2\pi)/2$.
\end{remark}

\section{Confined systems}\label{sec:confined}

In most practical applications, Riesz gases do not occur in a  `homogeneous' thermodynamic limit with a container $\Omega$ and hard walls (together with a uniform background for $s<d$), as we have studied in Sections~\ref{sec:thermo_limit_short_range} and~\ref{sec:thermo_limit_long_range}. In fact, in real experiments the system is usually confined by means of an external potential which typically varies over space~\cite{Derrick-69}. This is for instance how Coulomb crystals are produced in the laboratory since the 90s~\cite{ChuLin-94,HayTac-94,Thoetal-94,Thompson-15}.

The general idea is that if there are many points in regions where the external potential varies slowly, then we expect to obtain locally the infinite gas at equilibrium with the corresponding uniform chemical potential $\mu$. In this situation we think of the external potential as being macroscopic whereas the infinite gas occurs at the microscopic scale. One can even allow other thermodynamic parameters (e.g.~the temperature) to vary within the system.

The local convergence to the infinite gas in an external potential has been studied in the early days of statistical mechanics\cite{LebPer-61,LebPer-63,Millard-72,GarSim-72,SimGar-73,MarPre-72}. At the microscopic scale, a macroscopic external potential often ends up having the form $W(x/\ell)$ with $\ell\to\ii$. For a smooth function $W$, this is constant over boxes of side length $\ll \ell$. The proof that the system locally ressemble the infinite Riesz gas relies on proving that the interaction between these boxes is negligible. Local bounds are then very useful.

In order to state our main theorems, we now introduce the energy and free energy of a Riesz gas in an (unscaled) external potential $W$. For simplicity, we assume throughout the whole section that $W\in C^0(\overline{\omega})$ where $\omega$ is a bounded or unbounded domain in $\R^d$ with a smooth boundary. If $\omega$ is unbounded we further assume that
\begin{itemize}
\item $\lim_{|x|\to\ii}W(x)=+\ii$ if $\beta=+\ii$,
\item $\int_\omega e^{-\beta W(x)}\,\dx<\ii$ if $\beta<\ii$.
\end{itemize}
By convention we let $W\equiv+\ii$ on $\R^d\setminus\overline\omega$.
The canonical energy and free energy read
\begin{equation}
E_s(N,W):=\min_{x_1,...,x_N\in\R^d}\left\{\sum_{1\leq j<k}V_s(x_j-x_k)+\sum_{j=1}^NW(x_j)\right\},
\label{eq:E_s_N_V}
\end{equation}
\begin{equation}
F_s(N,\beta,W):=-\frac{1}{\beta}\log\left(\frac1{N!}\int_{\R^{dN}}e^{-\beta \left(\sum_{1\leq j<k}V_s(x_j-x_k)+\sum_{j=1}^NW(x_j)\right) }\dx_1\cdots \dx_N\right).
\label{eq:F_s_N_beta_W}
\end{equation}
The grand canonical energy and free energy are given by
\begin{align*}
E^{\rm GC}_s(W)&:=\min_{n\geq0}\min_{x_1,...,x_n\in\R^d}\left\{\sum_{1\leq j<k\leq n}V_s(x_j-x_k)+\sum_{j=1}^nW(x_j)\right\},\\
F^{\rm GC}_s(\beta,W)&:=-\beta^{-1}\log\left(\sum_{n\geq0}\frac{1}{n!}\int_{\R^{dn}} e^{-\beta\left(\sum_{1\leq j<k}V_s(x_j-x_k)+\sum_{j=1}^nW(x_j)\right)}\right).
\end{align*}
The chemical potential $\mu$ is typically included in the potential $W$. For $\Omega=N^{\frac1d}\omega$ with $|\omega|=1$ (resp.~$\Omega=\ell\omega$), the thermodynamic limit studied in Sections~\ref{sec:thermo_limit_short_range} and~\ref{sec:thermo_limit_long_range} corresponds to taking the potential
\begin{align}
W_N(x)&=N^{\frac{\max(0,d-s)}d}W(N^{-\frac1d}x)-\delta_0(s)\frac{N\log N}{d},\nn\\
W_\ell(x)&=\ell^{\max(0,d-s)} W(\ell^{-1}x)-\mu-\delta_0(s)\ell^d\log \ell,
\label{eq:rescaled}
\end{align}
for, respectively, the canonical and grand-canonical cases, where
\begin{equation}
W(x)=\begin{cases}
0 &\text{for $x\in\overline\omega$ and $s>d$,}\\
-\dps \int_\omega V_s(x-y)\,\dy &\text{for $x\in\overline\omega$ and $s<d$,}\\
+\ii&\text{for $x\notin\overline\omega$. }
\end{cases}
\label{eq:potential_omega}
\end{equation}
This potential is $C^\ii$ on $\omega$ and continuous on $\overline\omega$.

Our goal is to study the situation where $W$ is now a general function. We start by describing the results in the short range case before turning to the long range case.

\begin{remark}\label{rmk:confined_applies}
Many authors have studied Coulomb and Riesz gases in a general confining potential $W$. The thermodynamic limits we have studied in Sections~\ref{sec:thermo_limit_short_range} and~\ref{sec:thermo_limit_long_range} are covered in those works whenever the potential~\eqref{eq:potential_omega} is allowed. This is the case of most of the literature, possibly under regularity assumptions on $\omega$.

We emphasize that in many works the problem is expressed at the \emph{macroscopic scale} $\ell_{\rm macro}\sim N^{1/d}$ using the change of variable $x'=x/N^{1/d}$. Since $V_s$ is homogeneous of degree $-s$, this leads to rather simple expressions with a power of $N$ in front of the terms of the energy. This would not be the case for a more general interaction. Many papers concern only the macroscopic scale, or the intermediate \emph{mesoscopic scales}  $1\ll\ell_{\rm meso}\ll N^{1/d}$. The latter intermediate scales are not discussed at all in this review, which is solely about the microscopic system (the Coulomb or Riesz gas).
\end{remark}

\subsection{Short range case $s>d$}

In the short range case $s>d$, we have seen in Section~\ref{sec:thermo_limit_short_range} that the energy of points located at a finite distance to each other is typically of order $N$. So will be the potential energy if we choose as in~\eqref{eq:rescaled}
\begin{equation}
W_N(x)=W(N^{-\frac1d}x),\qquad W_\ell(x)= W(\ell^{-1}x),\qquad\text{for $s>d$}
\label{eq:rescaled_short}
\end{equation}
in, respectively, the canonical and grand-canonical cases. At the macroscopic scale this amounts to multiplying the interaction energy by $N^{-s/d}$ or $\ell^{-s}$. The slowly varying limit takes the following form for $s>d$.

\begin{theorem}[Slowly varying external fields, short range]\label{thm:confined_short_range}
Let $s>d$ and $W$ satisfying the previous conditions. Then the canonical energy and free energy have the limit
\begin{align}
\lim_{N\to\ii}\frac{E_s\big(N,W(N^{-\frac1d}\cdot)\big)}{N}&=\min_{\int_\omega\nu=1}\left\{\int_{\R^d} W(x)\nu(x)\,\dx+e(s)\int_{\R^d} \nu(x)^{1+\frac{s}d}\,\dx \right\}\nn\\
&=-\frac{s e(s)^{-\frac{d}{s}}}{d\left(1+\frac{s}{d}\right)^{1+\frac{d}{s}}}\int_\omega(\mu_W-W(x))_+^{1+\frac{d}{s}}\dx+\mu_W,\label{eq:limit_W_C_T0}\\
\lim_{N\to\ii}\frac{F_s\big(N,\beta,W(N^{-\frac1d}\cdot)\big)}{N}&=\min_{\int_\omega\nu=1}\left\{\int_{\R^d} W(x)\nu(x)\,\dx+\int_{\R^d}f\big(s,\beta,\nu(x)\big)\,\dx \right\}\nn\\
&=\int_\omega g\big(s,\beta,\mu_W-W(x)\big)\,\dx+\mu_W\label{eq:limit_W_C_T}
\end{align}
for a unique Lagrange multiplier $\mu_W$. For the first limit~\eqref{eq:limit_W_C_T0}, the latter is the unique solution to the equation
\begin{equation}
\left(\frac{d}{e(s)(d+s)}\right)^{\frac{d}s}\int_\omega(\mu_W-W)_+^{\frac{d}{s}}=1.
 \label{eq:mu_W}
\end{equation}
For~\eqref{eq:limit_W_C_T}, $\mu_W$ is the unique maximizer of the function $\mu\mapsto \int_\omega g(s,\beta,\mu-W(x))\,\dx+\mu$. 
In the grand-canonical case we have similarly
\begin{align*}
\lim_{\ell\to\ii}\frac{E^{\rm GC}_s\big(W(\ell^{-1}\cdot)\big)}{\ell^d}&=\min_\nu\left\{\int_{\R^d} W(x)\nu(x)\,\dx+e(s)\int_{\R^d} \nu(x)^{1+\frac{s}d}\,\dx \right\}\nn\\
&=-\frac{s e(s)^{-\frac{d}{s}}}{d\left(1+\frac{s}{d}\right)^{1+\frac{d}{s}}}\int_\omega W(x)_-^{1+\frac{d}{s}}\dx,\\
\lim_{\ell\to\ii}\frac{F^{\rm GC}_s\big(\beta,W(\ell^{-1}\cdot)\big)}{\ell^d}&=\min_{\nu}\left\{\int_{\R^d} W(x)\nu(x)\,\dx+\int_{\R^d} f\big(s,\beta,\nu(x)\big)\,\dx \right\}\nn\\
&=\int_\omega g(s,\beta,-W(x)\big)\,\dx.
\end{align*}
\end{theorem}

\begin{figure}[t]
\centering
\includegraphics[width=10cm]{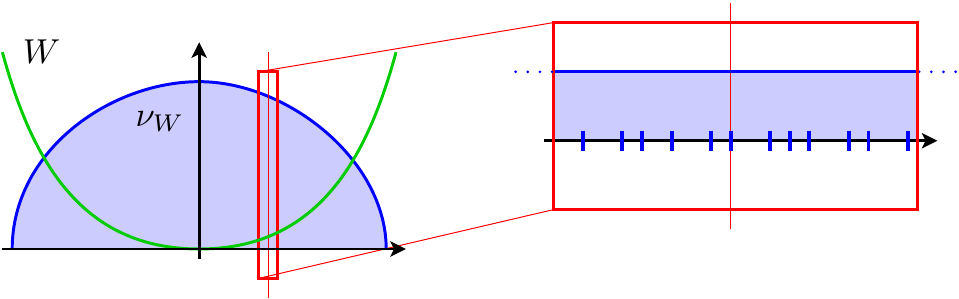}
\caption{At the macroscopic scale we have a fixed external potential $W$ and an associated equilibrium measure $\nu_W$ describing the shape of our system. When we zoom at the scale $\ell^{-1}\sim N^{-1/d}$ about a point $x$ in the support of $\nu_W$, we obtain our translation-invariant Coulomb or Riesz gas at the corresponding average density $\rho=\nu_W(x)$ or, equivalently, the chemical potential $\mu=\mu_W-W(x)$. Here $\mu_W$ is a macroscopic chemical potential used in the canonical case to guarantee that $\int\nu_W=1$. In the grand-canonical case we have simply $\mu_W=0$.
\label{fig:confined}}
\end{figure}

The statement involves the thermodynamic functions $f$ and $g$ from Theorems~\ref{thm:limit_C} and~\ref{thm:limit_GC}. The spirit of the result is that locally, in a neighborhood of any point $x$, we obtain the infinite gas with $W$ replaced by a constant. This brings $f$ and $g$ times the volume $\dx$ and thus an integral at the macroscopic scale. In the grand-canonical case, the local chemical potential is simply $\mu=-W(x)$. In the canonical case the constraint that the total number of points is fixed brings an additional macroscopic multiplier $\mu_W$ associated with a constraint $\int_\omega\nu=1$.

Theorem~\ref{thm:confined_short_range} can be found in the grand-canonical case in Ref.~\onlinecite{MarPre-72} (under weaker assumptions on $V$ and $\omega$) and in both the canonical and grand-canonical cases in Ref.~\onlinecite{GarSim-72} under the additional assumption that $\omega$ is bounded. As usual, the canonical and grand-canonical limits are dual to each other and one follows once the other one is proved. The zero-temperature case also follows from the case $\beta<\ii$ due to the continuity of the thermodynamic functions in the limit $\beta\to\ii$. The limit in the canonical case is also contained in the more recent works in Ref.~\onlinecite{HarSafVla-17,HarLebSafSer-18} which additionally prove a large deviation principle at $\beta<\ii$ and deal with $s=d$ at $\beta=+\ii$.

At $\beta=+\ii$, the minimizer of~\eqref{eq:limit_W_C_T0} is unique, due to the relation~\eqref{eq:mu_rho_T0}, given by
\begin{equation}
\boxed{ \nu_W(x)=\frac{\big(\mu_W-W(x)\big)_+^{\frac{d}s}}{\left(1+\frac{s}{d}\right)^{\frac{d}s}e(s)^{\frac{d}s}}.}
 \label{eq:nu_W_short}
\end{equation}
One can prove that the macroscopic empirical measure converges to $\nu_W$,\cite{HarSafVla-17}
\begin{equation}
\frac1N\sum_{j=1}^N\delta_{N^{-1/d}x_j}\wto \nu_W.
\label{eq:CV_nu_W}
\end{equation}
This means that $\nu_W$ is the shape of our gas at the macroscopic scale (Figure~\ref{fig:confined}). Note that for a harmonic potential $W(x)=|x|^2$, $\nu_W$ in~\eqref{eq:nu_W_short} is a kind of \emph{semi-circle law} \`a la Wigner\cite{Wigner-55}, which depends on the microscopic quantity $e(s)$. In fact, in $d=1$ we exactly obtain a half circle for the Calogero-Sutherland-Moser model~\cite{Calogero-71,Sutherland-71,Sutherland-71b,Moser-75,Sutherland-04} at $s=2$.
A similar result holds in the grand-canonical case with $\mu_W$ replaced by $0$ in~\eqref{eq:nu_W_short}.
For our container with hard walls~\eqref{eq:potential_omega}, we have $W=0$ (resp. $W=-\mu$ in the grand-canonical case) on $\omega$ and therefore the limiting macroscopic density $\nu_W$ is constant over $\omega$. For the thermodynamic limit studied in Section~\ref{sec:thermo_limit_short_range}, the system is perfectly uniform at the macroscopic scale.

At $\beta<\ii$, minimizers are not necessarily unique since there can be phase transitions. More precisely, the equation for a minimizer $\nu_W$ of~\eqref{eq:limit_W_C_T} reads
\begin{equation}
\boxed{ \frac{\partial f}{\partial\rho}\big(s,\beta,\nu_W(x)\big)=\mu_W-W(x).}
 \label{eq:nu_W_positive_temp}
\end{equation}
As we have discussed after Theorem~\ref{thm:limit_GC}, the derivative $\partial_\rho f(s,\beta,\rho)$ is a continuous non-decreasing function. It may be constant on some intervals, however, hence not invertible there. The value of $\nu_W(x)$ is not uniquely defined when $W(x)-\mu$ equals such values. However, there are at most countably many such bad points and we deduce that $\nu_W$ is well defined almost everywhere if the level sets $\{W=c\}$ have zero Lebesgue measure for all $c\in\R$. The inverse of $\partial_\rho f$ being exactly $-\partial_\mu g$, we can then rewrite~\eqref{eq:nu_W_positive_temp} as
\begin{equation}
\boxed{ \nu_W(x)=-\frac{\partial g}{\partial\mu}\big(s,\beta,\mu_W-W(x)\big)\qquad\text{for a.e. $x$}.}
 \label{eq:nu_W_positive_temp2}
\end{equation}
The Lagrange multiplier $\mu_W$ is then the unique solution to
\begin{equation}
\int_\omega\partial_\mu g(s,\beta,\mu_W-W(x)\big)\,\dx=-1. 
 \label{eq:mu_W2}
\end{equation}
There is a convergence similar to~\eqref{eq:CV_nu_W} for the rescaled density of the system~\cite{SimGar-73}.

It is possible to be more precise with regard to the behavior of the Gibbs measure at the microscopic scale. Well inside the system where $\nu_W$ is strictly positive (in the ``bulk''), it will typically converge to an equilibrium Riesz point process of the corresponding local uniform density (Figure~\ref{fig:confined}). This is proved using local bounds as in Section~\ref{sec:local_bd_short}, which stay valid in the presence of a bounded-below external potential (see, e.g., Ref.~\onlinecite[Lem.~2]{MarPre-72}). The situation is different close to the boundary $\partial\{\nu_W=0\}$ of the macroscopic system. At a point where the density $\nu_W$ jumps abruptly to 0, we expect to get an infinite gas living in a half space. The problem is more involved at a point where the density $\nu_W$ vanishes smoothly.

Instead of placing the points in an external potential, one can as well force them to live on a manifold. This formally corresponds to taking $\omega$ this manifold and $W$ infinite outside. The final result is the same, except that we of course get locally the thermodynamic functions and state of the space dimension of the manifold. This problem has been studied a lot, for instance on the sphere. We refer to Section~\ref{sec:random_matrices} below and to Ref.~\onlinecite{BorHarSaf-19} for an overview of the results in this direction at $T=0$.

\subsection{Long range case $s<d$}
As usual, the long range case is much more delicate. The leading energy of the system will typically not be of order $N$ due to the long range of the potential $V_s$, but rather $N^{2-s/d}$. The external potential must adapt to this energy scale and one should now take
\begin{equation}
W_N(x)=N^{1-\frac{s}{d}}W(N^{-\frac1d}x),\qquad W_\ell(x)= \ell^{d-s}W(\ell^{-1}x),\qquad\text{for $s<d$}
\label{eq:rescaled_long}
\end{equation}
in the canonical and grand-canonical cases, respectively, as was announced in~\eqref{eq:rescaled}.
To better understand the leading term, it is useful to rewrite the problem in macroscopic coordinates by letting $x'=x/N^{1/d}$. For instance, for the canonical energy~\eqref{eq:E_s_N_V} the scaling of $V_s$ gives
\begin{multline}
E_s\Big(N,N^{1-\frac{s}d}W(N^{-\frac1d}\cdot)\Big)+\frac{N(N-1)}{2d}\log N\delta_0(s)\\=N^{1-\frac{s}{d}}\min_{x'_1,...,x'_N\in\R^d}\left\{\frac1N\sum_{1\leq j<k}V_s(x'_j-x'_k)+\sum_{j=1}^NW(x'_j)\right\}.
\label{eq:E_s_N_V_rescaled}
\end{multline}
The right side is a mean-field limit~\cite{Choquet-58,Rougerie-LMU} and it is well known that
\begin{multline}
\lim_{N\to\ii}\frac{E_s\Big(N,N^{1-\frac{s}d}W(N^{-\frac1d}\cdot)\Big)+\frac{N^2}{2d}\log N\delta_0(s)}{N^{2-\frac{s}{d}}}\\=\min_{\int\nu=1}\left\{\frac12\iint_{\R^d\times\R^d}V_s(x-y)\rd\nu(x)\,\rd\nu(y)+\int_{\R^d}W(x)\,\rd\nu(x)\right\}=:I(W).
\label{eq:mean-field_limit}
\end{multline}
Since $V_s$ has a positive Fourier transform, the right side admits a unique minimizer $\nu_W$. It now solves the
implicit Euler-Lagrange equation
\begin{equation}
\boxed{\nu_W\ast V_s=\mu_W-W\qquad\nu_W\text{--a.e.,}}
 \label{eq:nu_W_MF}
\end{equation}
with the multiplier $\mu_W$ adjusted so that $\nu_W(\R^d)=1$. The rescaled empirical measure $N^{-1}\sum_{j=1}^N\delta_{N^{-1/d}x_j}$ converges weakly to $\nu_W$. This is therefore the shape of our gas at the macroscopic scale. Unlike the short range case $s>d$, $\nu_W$ only depends on $W$ and the interaction $V_s$. No microscopic information is seen at this scale. In the grand-canonical case ($s>0$) we have the similar limit
\begin{multline}
\lim_{\ell\to\ii}\frac{E_s^{\rm GC}\Big(N,\ell^{d-s}W(\ell^{-1}\cdot)\Big)}{\ell^{2d-s}}\\=\min_{\nu}\left\{\frac12\iint_{\R^d\times\R^d}V_s(x-y)\rd\nu(x)\,\rd\nu(y)+\int_{\R^d}W(x)\,\rd\nu(x)\right\}=:I^{\rm GC}(W)
\label{eq:mean-field_limit_GC}
\end{multline}
and the equation~\eqref{eq:nu_W_MF} now has $\mu_W=0$.

In some particular cases, it is possible to determine the unique solution $\nu_W$ of~\eqref{eq:nu_W_MF}, that is, the shape of the macroscopic system due to the potential $W$. In a harmonic well $W(x)=a|x|^2$ and for the Coulomb case $s=d-2$, $\nu_W$ is the characteristic function of a ball in all dimensions $d\geq1$. In $d=1$ at $s=0$ we obtain the famous \emph{Wigner semi-circle law}\cite{Wigner-55,Wigner-58}
\begin{equation}
\nu_W(x)=\frac{2}{\pi R^2}(R^2-x^2)^{\frac12}_+
\label{eq:semi-circle}
\end{equation}
for some $R>0$, which plays a central role in the theory of random matrices as we will see in Section~\ref{sec:random_matrices}. In the recent work Ref.~\onlinecite{AgaDhaKulKunMajMukSch-19}, the formula for $\nu_W$ was determined in a harmonic potential $W(x)=ax^2$ for all $-2<s<1$ in dimension $d=1$:
\begin{equation}
\nu_W(x)=\frac{\Gamma\left(\frac{4+s}{2}\right)}{R^{2+s}\sqrt\pi\,\Gamma\left(\frac{3+s}{2}\right)}(R^2-x^2)^{\frac{s+1}{2}}_+.
 \label{eq:formula_1D_nu_W_harmonic}
\end{equation}
The general shape is thus similar to the short range case~\eqref{eq:nu_W_short} but with a different power law. Also, for $s<-1$ the function is divergent at the boundary. This was later extended to the same system with a hard wall in Ref.~\onlinecite{KetKulKunMajMukSch-21}.

Another situation for which one can find $\nu_W$ explicitly, this time for all $-2<s<d$, is our flat container with hard walls studied in Section~\ref{sec:thermo_limit_long_range}. For the background potential~\eqref{eq:potential_omega} at $\rho_b=1$, we see that
$$\nu_W=\1_\omega,\qquad I(W)=-\frac12\iint_{\omega\times\omega}V_s(x-y)\,\dx\,\dy$$
so that the macroscopic density of the points is exactly equal to the uniform background and $I(W)$ corresponds to the background self-energy which we have added to our Jellium energy in~\eqref{eq:def_cE_s_Jellium}. At the macroscopic scale, the background is thus screened perfectly.

At positive temperature, the limit is exactly the same as in~\eqref{eq:mean-field_limit} for the free energy $F_s\big(N,N^{1-\frac{s}d}W(N^{-\frac1d}\cdot)\big)$. The reason is that, after scaling, the energy gets multiplied by $N^{1-s/d}$ which brings an effective $\beta_N=\beta N^{1-s/d}\to\ii$. It is not like in the short range potential case where the temperature can affect the form of the macroscopic system.

After having identified the macroscopic shape of our system, it is reasonable to expect that locally at the microscopic scale we should get again the Riesz gas, due to the slow variations of the external potential. Proving this result is much more delicate than in the short range case. The main difficulty is to show that different regions do not interact too much thanks to screening. This has been achieved for $d-2\leq s<d$ in the series of works by Serfaty \emph{et al} in Ref.~\onlinecite{SanSer-12,SanSer-14a,SanSer-15,BorSer-13,PetSer-17,RouSer-16,RotSer-15,Leble-15,LebSer-17,Serfaty-19}. The measure $\nu_W$ gives the average local density of the system and by writing the energy relative to $I(W)$ an effective uniform background at this density appears locally. For the result to hold, more stringent assumptions are needed than in the case $s>d$. The macroscopic density $\nu_W$ must be continuous and positive, whereas in Theorem~\ref{thm:confined_short_range} the regularity of the potential was sufficient.

\begin{theorem}[Slowly varying external fields, long range]\label{thm:confined_long_range}
Let $\max(0,d-2)\leq s<d$ and $W$ satisfying the previous conditions. If $\omega$ is unbounded and $s=0$, we assume in addition that $W(x)-\log|x|\to+\ii$ at infinity for $\beta=+\ii$ and $\int_\omega e^{-\beta (V(x)-\log|x|)}\dx<\ii$ for $\beta<\ii$. Concerning the macroscopic equilibrium measure $\nu_W$, we suppose that
\begin{itemize}
 \item $\nu_W$ has compact support $\Sigma$,
 \item $\partial\Sigma$ is $C^1$,
 \item $\nu_W$ is $C^1$ and strictly positive on $\Sigma$ and continuous on $\overline\Sigma$,
 \item either $\nu_W$ is strictly positive everywhere on the boundary, or vanishes everywhere like $\rd(x,\partial\Sigma)^\alpha$, for some $0<\alpha<1$.
\end{itemize}
Then the canonical energy and free energy admit the limit
\begin{multline*}
\lim_{N\to\ii}\frac{E_s\Big(N,N^{1-\frac{s}d}W(N^{-\frac1d}\cdot)\Big)+\frac{N^2}{2d}\log N\delta_0(s)-N^{2-\frac{s}{d}}I(W)}{N}\\
=e(s)\int_{\R^d} \nu_W(x)^{1+\frac{s}d}\,\dx-\frac{\delta_0(s)}{2d}\int_{\R^d}\nu_W(x)\log\nu_W(x)\,\dx, \label{eq:limit_W_C_T0_Jellium}
\end{multline*}
\begin{equation*}
\lim_{N\to\ii}\frac{F_s\big(N,\beta,N^{1-\frac{s}d}W(N^{-\frac1d}\cdot)\big)+\frac{N^2}{2d}\log N\delta_0(s)-N^{2-\frac{s}{d}}I(W)}{N}
=\int_{\R^d}f\big(s,\beta,\nu_W(x)\big)\,\dx.\label{eq:limit_W_C_T_Jellium}
\end{equation*}
\end{theorem}

As mentioned above, the theorem is proved in Ref.~\onlinecite{SanSer-12,SanSer-14a,SanSer-15,BorSer-13,PetSer-17,RouSer-16,RotSer-15,BetSan-18} at $\beta=+\ii$ and \onlinecite{Leble-15,LebSer-17,Serfaty-19} at $\beta<\ii$. A similar result is shown in Ref.~\onlinecite{BauBouNikYau-19} for $s=0$ in $d=2$ with a quantitative error term. A full  expansion in powers of $1/N$ was proved in Ref.~\onlinecite{AlbPasShc-01,BorGui-13,BorGuiKoz-15} at $s=0$ in $d=1$ using tools from random matrix theory (under more restrictive assumptions on the potential).

To our knowledge the grand-canonical problem has not been considered in the literature. We think that the approach of Ref.~\onlinecite{HaiLewSol_1-09,*HaiLewSol_2-09,LewLieSei-18,LewLieSei-19} based on the Graf-Schenker inequality~\cite{GraSch-95} for $s=d-2$ or that of Ref.~\onlinecite{CotPet-19,CotPet-19b} based on the Fefferman decomposition~\cite{Fefferman-85,Hughes-85,Gregg-89}  should provide similar limits for all $s>0$, with the functions $e^{\rm GC}(s)$ and $f^{\rm GC}(s,\beta,\rho)$ in Theorem~\ref{thm:limit_GC_Jellium}. From the equivalence of ensembles this would provide a different proof of Theorem~\ref{thm:confined_long_range} when $s\geq d-2$.

\begin{remark}[Fixing the density]\label{rmk:UEG}
The fixed external potential problem possesses a dual which plays an important role in classical and quantum density functional theory~\cite{LewLieSei-19_ppt}. This amounts to fixing the density $\rho(x)$ of the system. At zero temperature, the formulation in the canonical case is
$$C_s[\rho]=\min_{\rho_\bP=\rho}\int_{(\R^d)^N}\bigg(\sum_{1\leq j<k\leq N}V_s(x_j-x_k)\bigg)\rd\bP(x_1,\cdots,x_N)$$
where the minimum is over all symmetric probabilities $\bP$ on $(\R^d)^N$ of density $\rho$. This is a \emph{multi-marginal optimal transport} problem~\cite{Pass-15,LewLieSei-19_ppt}. We have the Legendre transforms
$$C_s[\rho]=\sup_{W}\left\{E_s(N,W)-\int\rho W\right\},\qquad E_s(N,W)=\inf_{\rho}\left\{C_s[\rho]+\int\rho W\right\},$$
where the last minimum is over densities of integral $N$. The limit similar to Theorems~\ref{thm:confined_short_range} and~\ref{thm:confined_long_range} corresponds to very spread out densities in the form $\nu(x/N^{1/d})$ with $\int\nu=1$. This limit was studied for $s=1$ in dimension $d=3$ in Ref.~\onlinecite{CotFriKlu-13,CotFriPas-15,LewLieSei-18,LewLieSei-19} and all $0<s<d$ in Ref.~\onlinecite{CotPet-19}. Quantitative bounds for a general slowly varying density are provided in Ref.~\onlinecite[Appendix]{LewLieSei-19}.
\end{remark}

\subsection{Examples: Feckete points, random matrices, Laughlin wavefunctions}\label{sec:random_matrices}

In this section we give some important practical examples of confined Riesz gases. We will not state any theorem and refer to the literature for more information.

\bigskip

\paragraph{Fekete points.}\label{sec:Feckete}
As we have mentioned above, we can confine our points to a manifold instead of using a smooth external potential. The zero-temperature minimizers are called \emph{Fekete points}~\cite{Fekete-23}. There are many works in this direction and we refer to Ref.~\onlinecite{SafKui-97,BraHarSaf-12,BorHarSaf-19} for an overview. The equivalent of the long range convergence in Theorem~\ref{thm:confined_long_range} is proved for smooth closed curves in dimension $d=1$ for all $s\geq -2$ in Ref.~\onlinecite{BorHarSaf-19} and seems to be otherwise only shown for $s=0$ and $d=2$ for the sphere in Ref.~\onlinecite{BetSan-18}. Fekete points on complex manifolds have also generated some interest in complex geometry~\cite{BerBouWitt-11,Berman-17,Berman-19,Berman-20}.

The case of the sphere has concentrated much of the attention, due to its link with the 7th open problem of Smale in 1998~\cite{Smale-98,Beltran-13}. Finding the optimal position of points on spheres has a lot of practical applications. In biology, this problem can explain the form of some viruses, and the arrangement  of pores on pollen grains; it is also studied in link with colloidosomes~\cite{Dinetal-02}. In physics, computations of Jellium on the sphere are used to avoid the use of the background and Ewald summation techniques~\cite{HanLevWei-79,CaiLevWeiHan-82,BerBoaPal-94}. This is often called the ``spherical boundary condition''. In quantum chemistry, it was proposed in Ref.~\onlinecite{LooGil-11,LooGil-13,AgbKnoGilLoo-15} that simulating electrons on a circle or a sphere could provide a better approximation to the energy of finite systems in density functional theory. Exact formulas for the 2D Coulomb gas settled on the sphere $\bS^2$ are derived for even $\Upsilon=\beta$ in Ref.~\onlinecite{Samaj-04,TelFor-12,SalTel-16}. The same system was also considered on a ``pseudo-sphere'' (a non-compact Riemann surface of constant negative curvature) in Ref.~\onlinecite{JanTel-98,FanJanTel-03}.

\bigskip

\paragraph{Random matrices.}
Many of the recent works on Riesz gases are motivated by the theory of random matrices~\cite{Mehta-10,AndGuiZei-10,Forrester-10,Forrester-22_ppt}. Introduced by Wigner~\cite{Wigner-55,Wigner-67} in 1955 and further developed by Dyson and Mehta~\cite{Dyson-62a,*Dyson-62b,*Dyson-62c,DysMeh-63a,DysMeh-63b} in 1962, this fascinating subject has always been tight to the theory of the \emph{log gas} at $s=0$. In short, if the entries of a matrix are independent Gaussian random variables, the statistical distribution of the eigenvalues
$\lambda_1,...,\lambda_N$ is (after a proper scaling) given by the log gas $s=0$ in the harmonic external potential $W(x)=|x|^2$. For hermitian matrices (Gaussian Unitary Ensemble \emph{GUE}) the problem is posed in $d=1$ since the eigenvalues are real, with the effective $\Upsilon=\beta=2$ (at density $\rho=1$). If one imposes that the matrices have real coefficients (Gaussian Orthogonal Ensemble \emph{GOE}), then $\beta=1$. When considering complex matrices without any
symmetry assumption (Ginibre ensemble\cite{Ginibre-65}), the same holds in dimension $d=2$ with $\beta=2$.  Dumitriu and Edelman have constructed a probability measure for tridiagonal matrices which provides all possible values of $\beta$ in a harmonic potential~\cite{DumEde-02}.

It is also possible to consider unitary or orthogonal matrices (Circular Unitary Ensemble \emph{CUE}, and
Circular Orthogonal Ensemble \emph{COE}, respectively), using the uniform law on this compact subset of matrices. Then, the eigenvalues are confined to the circle $\bS^1$, still with $\beta=2$ and $\beta=1$. Killip and Nenciu\cite{KilNen-04} have generalized to circular ensembles the result mentioned above by Dumitriu and Edelman\cite{DumEde-02}. They constructed classes of unitary matrices which provide the 2D Coulomb gas restricted to a circle (that is, the periodic 1D log gas) for all values of~$\beta$.

For the special values of $\beta$ mentioned above one can obtain explicit formulas. After passing to the limit $N\to\ii$, this provides a point process over the line (resp. the complex plane for the Ginibre ensemble) which is our Jellium infinite system. The two processes at $\beta=2$ (GUE for $d=1$ and Ginibre for $d=2$) are \emph{determinantal point processes}\cite{Soshnikov-00}, that is, the classical probability of a quantum quasi-free Fermi state~\cite{BacLieSol-94}. For $d=1$ we obtain nothing else but the \emph{non-interacting 1D Fermi gas}, which is no big surprise in light of a result of Calogero and Sutherland\cite{Calogero-71,Sutherland-71,Sutherland-71b,Moser-75,Sutherland-04} explained below. For $d=2$ and $\beta=2$ we rather obtain the \emph{noninteracting Fermi gas in a constant magnetic field perpendicular to the plane~\cite{Forrester-98} at filling factor 1}, that is, filling the entire lowest Landau level.\footnote{The particles have no spin and one should use the magnetic Laplacian. This is also equivalent to a Fermi system in a harmonic trap rotating at the largest possible speed~\cite{LacMajSch-19}.}

Using tools from the theory of random matrices, the limiting point process could be constructed for all values of $0<\beta<\ii$ in $d=1$. This was achieved for the circular ensemble (periodic case) in Ref.~\onlinecite{KilSto-09} and for the Gaussian case (harmonic external potential) in Ref.~\onlinecite{ValVir-09}. The two were later shown to coincide in Ref.~\onlinecite{Nakano-14,ValVir-17}. The corresponding point process (our Riesz gas at $s=0$, $d=1$ and finite $\beta$) is usually called the \emph{sine-$\beta$ point process}.

We have to mention a famous analogy to a quantum system. Calogero and Sutherland discovered in 1971~\cite{Calogero-71,Sutherland-71,Sutherland-71b,Moser-75,Sutherland-04,Forrester-84} that the square root of the $N$-point Gibbs measure of the $\beta$--log gas is the exact minimizer of a quantum problem, for all $N$ and all $\beta>0$. This is valid either on the circle or in a harmonic trap. In the Gaussian case the quantum problem consist of $N$ quantum particles interacting with the potential $1/x^2$ in the harmonic potential
$$\sum_{j=1}^N-\frac{\rd^2}{\rd x_j^2}+N^{-1}|x_j|^2+\frac{\beta(\beta-2)}{2}\sum_{1\leq j<k\leq N}\frac{1}{(x_j-x_k)^2}.$$
For circular ensembles, the same holds on $\bS^1$ with
$$\sum_{j=1}^N-\frac{\rd^2}{\rd x_j^2}+\frac{\pi^2\beta(\beta-2)}{2}\sum_{1\leq j<k\leq N}\frac1{\sin^2\pi(x_j-x_k)}.$$
After taking the limit $N\to\ii$, this says that the \emph{1D log gas at inverse temperature $\beta$ (sine-$\beta$)} coincides with a \emph{zero-temperature quantum gas with inverse-quadratic interaction}. The latter is the thermodynamic limit of the scaling-invariant quantum model with Hamiltonian
\begin{equation}
\sum_{j=1}^N-\frac{\rd^2}{\rd x_j^2}+\frac{\beta(\beta-2)}{2}\sum_{1\leq j<k\leq N}\frac1{(x_j-x_k)^2}.
\label{eq:Calogero-Sutherland}
\end{equation}
For the quantum system the interaction is repulsive for $\beta>2$, leading to a strong repulsion on the classical points, but attractive (though partly compensated by the kinetic energy) for $\beta<2$ and leading to a weaker repulsion between the points. The $N$-particle wavefunction is assumed to vanish on the diagonal due to the strong divergence of the interaction. At $\beta=2$ this condition must be retained and this leads to the non-interacting Fermi gas. The case $\beta=2$ is a transition point, as will be discussed in Section~\ref{sec:transitions}.

It is surprising that the Calogero-Sutherland-Moser (CSM) model occurs both at $s=0$ ($T=0$ quantum CSM) and $s=2$ ($T\geq0$ classical CSM). This link is further exploited in Ref.~\onlinecite{BogGirSch-09,AgaKulDha-19}.
As we have mentioned in Remark~\ref{rmk:CSM}, the quantum Calogero-Sutherland-Moser model can also be mapped to a non-interacting gas obeying the Haldane-Wu fractional statistics~\cite{Haldane-91,Wu-94,BerWu-94,Ha-94,Isakov-94,MurSha-94}. This is yet another interpretation of the one-dimensional log gas.

Many works in random matrix theory are devoted to proving some kind of ``universality''. This can be with respect to the shape of the external potential. Many results more precise than Theorem~\ref{thm:confined_long_range} are available in this setting, see, e.g., Ref.~\onlinecite{BouErdYau-12,BouErdYau-14} and the references therein. But it is also important to vary the law of the entries of the matrices. The universality in this context is called the \emph{Wigner-Dyson-Mehta conjecture} and was solved in Ref.~\onlinecite{ErdSchYau-11,TaoVu-11,ErdRamSchTaoVuYau-10,ErdYau-15}. Random matrices have been the object of many works which cannot be summarized in only a few lines. We refer to Ref.~\onlinecite{Mehta-10,ErdYau-12,Erdos-13,Deift-17,ErdYau-17,Chafai-21_ppt} for reviews and more recent results.

The original motivation of Wigner was to describe the energy levels of nuclei using a simplified model. Random matrices turned out to be extremely efficient in this respect~\cite{Mehta-10,BroFloFreMelPanWon-81,HaqPanBoh-82}. But the 1D log-gas is also believed to appear in many other settings~\cite{GuhMulWei-98,BouKea-13}. For instance, a famous conjecture states that it should also describe the local law of the eigenvalues of random Schr\"odinger operators in the delocalized phase (\onlinecite[Chap.~17]{AizWar-15}). A similar universality appears in quantum chaos where it is known as the \emph{Bohigas-Giannoni-Schmidt conjecture}~\cite{BohGiaSch-84,Berry-85,WinFrie-86,SimSzaAlt-93,SoAnlOttOer-95}. Zeros of random polynomials~\cite{ZeiZel-10} have also been considered.

In fact, the 1D log gas appears everywhere and is kind of universal in the class of repulsive point processes.
It has been found to properly describe the zeroes of the Riemann Zeta  function on the critical line~\cite{Montgomery-73,RudSar-96,BouKea-13,Odlyzko-87}, the waiting time between trains in the New York subway~\cite{JagTro-17} or buses in Mexico~\cite{KrbSeb-00,BaiBorDeiSui-06}, as well as the space between cars parked on the street or birds perched on a power line~\cite{Abul-Magd-06,Seba-07,Seba-09}. The rigorous understanding of this universality is an important question.

\bigskip

\paragraph{Laughlin wavefunction and Quantum Hall Effect.}
The \emph{fractional quantum Hall effect} is observed in two-dimensional electron systems subjected to low temperatures and strong magnetic fields. The conductance displays some plateaux at specific (quantized) fractional values of the filling factor~\cite{StoTsuGos-99,Callaway-91}. In 1983, Laughlin~\cite{Laughlin-83} proposed that for the special fractions $1/m$ this could be understood through the simple quantum wavefunction
$$\Psi_{\rm L}(z_1,...,z_N)=\prod_{1\leq j<k\leq N}(z_j-z_k)^me^{-\sum_{j=1}^N|z_j|^2},$$
with $m$ odd for electrons. Such functions also naturally arise for bosonic systems with $m$ even, when those are submitted to a very fast rotation~\cite{Cooper-08,LewSei-09,SeiYng-20}. Here we make the identification $\R^2\simeq \C$ and $z_j\in\C$. The modulus square of $\Psi_{\rm L}$ is the classical probability distribution of our Riesz Gibbs measure at $\beta=2m$ in a harmonic external potential. In particular, we obtain the Ginibre ensemble at $m=1$. This ``plasma analogy'' was used by Laughlin to argue that the excitations have the fractional charge $1/m$, which explained the fractional quantum Hall effect. Halperin\cite{Halperin-83} and Arovas-Schrieffer-Wilczek\cite{AroSchWil-84} proposed that fractionally charged quasiparticle excitations of the Laughlin states are \emph{anyons}~\cite{LeiMyr-77,Myrheim-99}, a kind of topological 2D quantum particle interpolating between bosons and fermions. This was confirmed later in experiments~\cite{PicRezHeiUmaBunMah-97,SamGlaJinEti-97,Martin_etal-04}.

The Laughlin wave function has been the source of many works relying on the plasma analogy, among which we only cite Ref.~\onlinecite{RouSerYng-13b,LieRouYng-18,LieRouYng-19,Rougerie-22_ppt}. A different line of research is to study the properties of Laughlin states on Riemann surfaces~\cite{Haldane-83,Klevtsov-16,KleXiaGeoWie-17,NemKle-21} or to interpret the different terms in the large-$N$ expansion of the Coulomb gas free energy in an external potential as geometric quantities~\cite{WieZab-00}. The 2D Coulomb gas appears in a variety of other situations with a geometric content, for instance in the study of the determinant of the Laplace-Beltrami operator on surfaces~\cite{OsgPhiSar-88}.

\section{Properties of Riesz gases, phase transitions}\label{sec:transitions}

The convex set $\cR_{s,\beta,\mu}$ of Riesz point processes was defined in Definitions~\ref{def:equilibrium} and~\ref{def:Gibbs_DLR} for $s>d$ and $\beta\leq\ii$. It is in Theorem~\ref{thm:infinite_conf_Jellium} for $\beta=+\ii$ and $d-2\leq s<d$. Translation-invariant Riesz point processes were defined in Ref.~\onlinecite{DerHarLebMai-21} for $s=0$ in dimension $d=1$ and in Ref.~\onlinecite{DerVas-21_ppt} for $d-1<s<d$ in all dimensions. We discuss in this section the known and expected properties of Riesz gases, with a focus on uniqueness.

When $\beta<\ii$, the simplest definition of a \emph{phase transition} is that $\cR_{s,\beta,\mu}$ is not reduced to one point, that is, Riesz point processes are not unique. As we will see, this concept is a little bit too restrictive, in particular in dimensions $d\in\{1,2\}$ where uniqueness is expected in most cases, although some kind of phase transition happens. At $\beta=+\ii$ we have explained that $\cR_{s,\beta,\mu}$ is never reduced to one point and discussed the crystallization conjecture in Section~\ref{sec:crystal_conjecture}.

To be more precise, a phase transition is often associated with the breaking of a symmetry of the system. In the thermodynamic limit the model is invariant under both translations and rotations. The conjecture is that when $\cR_{s,\beta,\mu}$ is not reduced to one point, it should be the convex hull of the $g\cdot\mathscr{P}_j$ for a few simple $\mathscr{P}_j$, where $g$ varies among all possible isometries of $\R^d$. In most cases we expect only one $\mathscr{P}_1$, hence uniqueness up to isometries. At first order phase transitions points we can have several $\mathscr{P}_j$. This is what corresponds to Conjecture~\ref{conj:crystal_state} at $\beta<+\ii$. Note that it is possible to break only part of the symmetries. Riesz point processes could be all translation-invariant but not rotation-invariant. Nevertheless, the simplest situation is that of a \emph{solid}, where we would have only one $\mathscr{P}_1$ which is periodic with respect to some lattice (as in Lemma~\ref{lem:periodic_energy_pt_process}) so that $\cR_{s,\beta,\mu}$ corresponds to the uniform probabilities over translations and rotations of this periodic system.

To our knowledge, the only available proof of a phase transition for Riesz gases is at $s=-1$ in $d=1$ (Section~\ref{sec:Kunz} below). There are more results for lattice gases\cite{FroSpe-82,FroSimSpe-76}.

An important tool to investigate symmetry breaking is the decay of the truncated correlation functions. We recall that the \emph{truncated $k$--point correlation functions} are defined recursively by
$$\rho^{(k)}_T(x_1,...,x_k):=\rho^{(k)}(x_1,...,x_k)-\sum_{\{1,...,k\}=I_1\cup\cdots\cup I_j}\rho^{(|I_1|)}_T(X_{I_1})\cdots\rho^{(|I_j|)}_T(X_{I_j}),$$
where the sum runs over all the (non trivial) partitions of $\{1,...,k\}$ and $X_{J}=(x_{i_1},...,x_{i_j})$ when $J=\{i_1,...,i_j\}$. For instance, the truncated two-point correlation function is
$$\rho^{(2)}_T(x,y):=\rho^{(2)}(x,y)-\rho^{(1)}(x)\rho^{(1)}(y).$$
Unlike the correlation functions which depend linearly on the point process, the truncated correlation functions are \emph{nonlinear} functions of the point process.

One important property of the convex set $\cR_{s,\beta,\mu}$, proved in Ref.~\onlinecite[App.~B]{LanRue-69} in the short range case $s>d$, is that its extreme points are characterized by the fact that their truncated correlations all go to 0 at infinity. In particular, if we can find one Riesz point process in $\cR_{s,\beta,\mu}$ which has some long range order (one of its truncated correlations does not converge to 0), we can conclude that the set is not reduced to one point, hence there is some phase transition. A typical situation encountered in practice is that of a translation-invariant process (e.g. obtained from the limit with periodic boundary conditions as in Section~\ref{sec:periodic_BC}) such that $\rho^{(2)}_T(x-y)\nrightarrow0$ when $|x-y|\to\ii$.

\subsection{Uniqueness and decay of correlation}

In this section we mention existing uniqueness results, either at small values of $\Upsilon=\beta\rho^{\frac{s}d}$, or in dimension $d\in\{1,2\}$ (Mermin-Wagner theorem).

\subsubsection{Short range case $s>d$}
When $\Upsilon=\beta\rho^{\frac{s}d}$ is small enough, it is known that $\cR_{s,\beta,\mu}$ is reduced to one point, that is, the Riesz point process is unique and no phase transition occurs. In particular the truncated correlation functions all tend to 0 at infinity. Their decay is universal, given by that of the potential.

\begin{theorem}[Uniqueness for $\Upsilon\ll1$~\cite{Ruelle-63,Penrose-63,Ruelle,BenGruMar-84}]\label{thm:unique}
Let $s>d$. For
\begin{equation}
\beta^{\frac{d}s} e^{\beta\mu}<\frac1{\dps e\int_{\R^d}(1-e^{-|x|^{-s}})\,\dx},
\label{eq:Mayer-radius}
\end{equation}
the set $\cR_{s,\beta,\mu}$ is reduced to one point $\mathscr{P}_{s,\beta,\mu}$, which is thus translation-- and rotation--invariant. The maps $(\beta,\mu)\mapsto \rho^{(n)}_{\mathscr{P}_{s,\beta,\mu}}$ and $(\beta,\mu)\mapsto f(s,\beta,\mu)$ are real-analytic in the region characterized by the condition~\eqref{eq:Mayer-radius} for all $n\geq1$. We have
\begin{equation}
 \frac{e^{\beta\mu}}{1+\beta^{\frac{d}s}e^{\beta\mu}\int_{\R^d}(1-e^{-|x|^{-s}})\,\dx}\leq\rho:=\rho^{(1)}_{\mathscr{P}_{s,\beta,\mu}}\leq e^{\beta\mu}.
 \label{eq:estim_rho_mu}
\end{equation}
The solution can be parametrized in terms of $\rho$ instead of $\mu$ and is again unique and real-analytic in the region
\begin{equation}
\Upsilon^{\frac{d}s}=\beta^{\frac{d}s}\rho<\frac1{\dps (1+e)\int_{\R^d}(1-e^{-|x|^{-s}})\,\dx}.
\label{eq:Mayer-radius2}
\end{equation}
The truncated two-point correlation function of $\mathscr{P}_{s,\beta,\mu}$ satisfies
\begin{equation}
\lim_{|x-y|\to\ii}|x-y|^s\left(\rho^{(2)}(|x-y|)-\rho^2\right)=-\beta\left(\rho+\int_{\R^d}\big(\rho^{(2)}(|z|)-\rho^2\big)\,\dz\right)^2.
\label{eq:decay_correlation_short}
\end{equation}
\end{theorem}

The proof of uniqueness and analyticity goes by applying the Banach fixed point to the Kirkwood-Salsburg equations~\eqref{eq:KS}~\cite{Ruelle-63,Penrose-63,Ruelle}. This is all explained in Ref.~\onlinecite[Chap.~4]{Ruelle}. The estimate~\eqref{eq:estim_rho_mu} confirms what we have mentioned before concerning the fact that $e^{\beta\mu}$ plays the role of a density. At low density we have $\rho(\mu)\sim e^{\beta\mu}$. That the function $x\mapsto \rho^{(2)}(|x|)-\rho^2$ is integrable is shown in Ref.~\onlinecite{Ruelle-64} and the precise decay~\eqref{eq:estim_rho_mu} is stated in Ref.~\onlinecite{Groeneveld-67} and proved in Ref.~\onlinecite{BenGruMar-84} (higher $\rho_T^{(k)}$ are handled there as well). Note that, by~\eqref{eq:nb_correlation}, the number inside the square on the right side of~\eqref{eq:decay_correlation_short} is nothing else but the variance of the number of points, per unit volume (see also Section~\ref{sec:prop} and Ref.~\onlinecite{GhoLeb-17}). The Ginibre inequality~\eqref{eq:variance_Ginibre} implies
$$\rho+\int_{\R^d}\big(\rho^{(2)}(|z|)-\rho^2\big)\,\dz\geq \frac{\rho}{1+\beta^{\frac{d}s}e^{\beta\mu}\int_{\R^d}(1-e^{-|x|^{-s}})\,\dx}>0.$$
In particular, the truncated two-point function decays exactly like $1/|x-y|^s$. In fact, it is argued in Ref.~\onlinecite{DunSou-76,GruLugMar-80,BenGruMar-84} that even for larger values of $\Upsilon$, it can never decay faster than the potential $1/|x-y|^s$. This is different in the long range case, as we will see below.

Of interest is also the behavior of $\rho^{(2)}(|x-y|)$ when $|x-y|\to0$, which describes the amount of repulsion between the points. From Lemma~\ref{lem:local_bound_GC} we already know that
$$\rho^{(2)}(|x-y|)\leq e^{2\beta\mu}e^{-\frac\beta{|x-y|^s}},$$
and the exact behavior is indeed expected to be $\rho^{(2)}(|x-y|)= e^{-\beta|x-y|^{-s}(1+o(1))}$. The fast divergence of $V_s$ at the origin results in an exponentially strong repulsion between the points.

As usual, dimensions $d=1$ and $d=2$ are special and continuous symmetries can not easily be broken, by the Van Hove~\cite{vanHove-50} and Mermin-Wagner~\cite{MerWag-66,Mermin-67,Mermin-68} theorems (except of course at $T=0$). The best result known so far seems to be the following.

\begin{theorem}[Mermin-Wagner]\label{thm:Mermin-Wagner}
For $d\in\{1,2\}$ and $s>d$, the equilibrium states in $\cR_{s,\beta,\mu}$ are all \emph{translation-invariant} for all $\mu\in\R$ and $\beta<\ii$.

For $d=1$ and $s>2$, the set $\cR_{s,\beta,\mu}$ is reduced to one point and the functions $(\beta,\mu)\mapsto\rho(\beta,\mu)$ and $(\beta,\mu)\mapsto g(s,\beta,\mu)$ are real-analytic on $\{(\mu,\beta)\in \R\times(0,\ii)\}$.
\end{theorem}

The first part of the theorem (absence of breaking of translations) is due to Fr\"ohlich and Pfister in Ref.~\onlinecite{FroPfi-81,FroPfi-86} (see also Ref.~\onlinecite{Georgii-99}). The second part of the theorem can be found in Ref.~\onlinecite{CamCapOl-83} (analyticity) and Ref.~\onlinecite{Papangelou-87} (uniqueness). It is not clear if the constraint $s>2$ is optimal, since discrete systems do exhibit phase transitions for $1<s \leq2$~\cite{Dyson-69,FroSpe-82}. However, it is known that the free energy is real-analytic at $s=2$ in $d=1$ by Remark~\ref{rmk:CSM}, and uniqueness is certainly expected as well.

Theorem~\ref{thm:Mermin-Wagner} does not handle rotational symmetry in dimension $d=2$. It is argued in Ref.~\onlinecite{FroPfi-81} that the result is also valid for rotations, but under the additional assumption that the truncated correlation functions decay like $1/|x-y|^{2+\eps}$ for some $\eps>0$. This was generalized in Ref.~\onlinecite{GruMar-81,GruMarOgu-82} in all dimensions $d\geq1$, using the BBGKY hierarchy~\eqref{eq:BBGKY_short}, leading to the conclusion that when the truncated correlations all decay faster than $1/|x-y|^{d+\eps}$, the state has to be invariant under both rotations and translations. As a consequence, rotation-invariance can be broken in dimension $d=2$, but then correlations have to decay slowly. The corresponding state is often called a solid.

Although translation-invariance is not broken, Berezinski, Kosterlitz and Thouless~\cite{KosTho-72,KosTho-73} have shown that topological defects can lead to a kind of ``quasi-breaking'' of translational symmetry in $d=2$, called the \emph{BKT phase transition}. We will not explain the theory here and will only retain the consequence that the truncated correlation functions move from the expected universal decay $1/|x-y|^{s}$ in~\eqref{eq:decay_correlation_short} to a much slower $\beta$--dependent polynomial decay with periodic modulations. Similar effects can happen in dimension $d=1$\cite{Schulz-81}, but only for small enough $s$. A typical example in dimension $d=1$  would be
\begin{equation}
 \rho^{(2)}_T(x-y)\underset{|x-y|\to\ii}\sim C\frac{\cos\big(2\pi\rho (x-y)\big)}{|x-y|^{p(\beta)}}
 \label{eq:typical_BKT}
\end{equation}
for a function $p(\beta)\to0$ when $\beta\to\ii$. In this situation one speaks of ``quasi long-range order''.

\subsubsection{Long range case $s<d$}\label{sec:phase_diag_long}

We review here what has been rigorously shown in the long range case $s<d$.

\bigskip

\paragraph{3D Coulomb ($s=1$).}
The Coulomb case $s=1$ in dimension $d=3$ has been the object of many works.
Mapping the model to a quantum field theory using the Sine-Gordon transformation~\cite{EdwLen-62,AlbHog-75,Frohlich-76,Park-77,FroPar-78,GruLugMar-80,FonMar-84} and based on previous results of Brydges and Federbush~\cite{Brydges-78,BryFed-80}, Imbrie~\cite{Imbrie-82} was able to construct Jellium equilibrium states for $\Upsilon=\beta\rho^{\frac13}\ll1$. He also proved that this is a gas phase with exponential decay of correlations. More precisely, at large distances the interaction potential is replaced by a Yukawa interaction $e^{-|x|/\ell_{\rm D}}/|x|$ where $\ell_{\rm D}$ is called the \emph{Debye length}. This is called \emph{Debye-Hückel screening} and is very specific to the Coulomb case. This is very different from what is happening in the short range case (Theorem~\ref{thm:unique}), where the decay is never better than that of the potential.

Following earlier works\cite{Brydges-78,BryFed-80} Imbrie had to perform the thermodynamic limit in a special way in order to better control the convergence of the cluster expansion. More precisely, his proof is in the grand-canonical case, adjusting $\mu$ to ensure exact neutrality. He works with two domains
 $\Omega\subset\Omega'$. On the bigger set $\Omega'$ the points interact with a well-adjusted short range potential.
The Coulomb interaction is turned on only in $\Omega$, with the Dirichlet boundary conditions mentioned in Section~\ref{sec:periodic_BC}. The limit $\Omega'\nearrow\R^3$ is taken first, which thus immerses the finite Coulomb system in an infinite fluid. Although the results of Imbrie in Ref.~\onlinecite{Imbrie-82} do not seem to directly apply to the Jellium problem in a container without the surrounding fluid, the independence with respect to boundary conditions proved in Ref.~\onlinecite{LebSer-17,ArmSer-21} suggests that the free energy should at least be the same. Imbrie's theorem must thus imply that $f(1,\Upsilon)$ is real analytic for $\Upsilon\ll1$ in $d=3$.

Many works have been devoted to Debye screening in the 80s, among which we cite Ref.~\onlinecite{Brydges-78,FroSpe-81b,FroSpe-81c,FedKen-85}. We refer to Ref.~\onlinecite{Rebenko-88,BryMar-99} for reviews.

\bigskip

\paragraph{The 2D Coulomb/Log-gas ($s=0$).}
Imbrie's work is stated in dimension $d=3$ for obvious physical reasons, but the proof probably works the exact same in all dimensions $d\geq3$ at $s=d-2$. The 2D case is always a bit different due to the divergence of the logarithm at infinity but it is nevertheless expected that the result should again be the same. Several works have dealt with 2D Coulomb systems, in particular the two-component plasma~\cite{FroPar-78,Yang-87,Falco-12} which has two kinds of particles with opposite charges and no background. The latter has more symmetries, which makes the proof easier from the point of view of quantum field theory.

For the 2D Coulomb gas, the only explicitly solvable point seems to be $\beta=2$ (Ginibre ensemble) at which one obtains for $\rho_b=1$\cite{Ginibre-65,Jancovici-81}
$$\rho^{(2)}_{T}(x-y)=-e^{-|x-y|^2}.$$
We find here the expected exponential decay for the Coulomb potential, due to Debye screening. Numerical simulations suggest that $\rho^{(2)}_{T}$ becomes non-monotonous for $\beta>2$\cite{Jancovici-81}, but the transition to a quasi-solid happens at a much bigger value of $\beta$, as we will describe in the next section. As already mentioned in~\eqref{eq:formula_f_2D_beta2}, the free energy $f(0,2,\rho)$ at $\beta=2$ was computed in Ref.~\onlinecite{DeuDewFur-79,AlaJan-81,Forrester-98}. The possibiblity of having an expansion in powers of $(\beta-2)$ near $\beta=2$ is discussed in Ref.~\onlinecite{Jancovici-81}. Some exact properties are discussed in Ref.~\onlinecite{TelFor-99,Samaj-04,TelFor-12,SalTel-16} when $\beta$ an even integer.

The occurrence of Berezinski-Kosterlitz-Thouless (BKT) phase transitions in 2D has been proved by Fr\"ohlich and Spencer in Ref.~\onlinecite{FroSpe-81,FroSpe-81b,FroSpe-81c}, but it is our understanding that the argument only applies to Coulomb lattice systems.

\bigskip

\paragraph{Exact properties for other values of $s$.}
For other values of $s$, the properties of long range Riesz gases were studied \emph{assuming} they exist and that their correlation functions solve a properly renormalized BBGKY equation. This program started with the important work of Gruber, Lugrin and Martin in Ref.~\onlinecite{GruLugMar-78,GruLugMar-80} and was then pursued in many subsequent articles\cite{MarYal-80,ChoFavGru-80,GruLebMar-81,GruMarOgu-82,BluGruLebMar-82,FonMar-84,AlaJan-84,LebMar-84,AlaMar-85,Martin-88,JanLebMag-93}. In the long range case, it is convenient to rewrite the BBGKY equation~\eqref{eq:BBGKY_short} in the form\cite{GruLugMar-78,GruLugMar-80}
\begin{align}
\nabla_{x_1}\rho^{(n)}(x_1,...,x_n)=&-\beta\bigg\{\sum_{j=2}^n\nabla V_s(x_1-x_j)\nn\\
&\qquad \qquad +\int_{\R^d}\nabla V_s(x_1-y)\big(\rho^{(1)}(y)-\rho_b\big)\,\dy\bigg\}\rho^{(n)}(x_1,...,x_n)\nn\\
& -\beta\int_{\R^d}\nabla V_s(x_1-y)\Big(\rho^{(n+1)}(x_1,...,x_n,y)-\rho^{(n)}(x_1,...,x_n)\rho^{(1)}(y)\Big)\dy
\label{eq:BBGKY}
\end{align}
for all $n\geq1$. Note that the term on the second line of~\eqref{eq:BBGKY}
$$\int_{\R^d}\nabla V_s(x-y)\big(\rho^{(1)}(y)-\rho_b\big)$$
is the average gradient of the potential $\Phi$ generated by the infinite system. It vanishes for a stationary point process. The last integral in~\eqref{eq:BBGKY} involves a kind of truncated function which we can hope to be integrable if correlations decay sufficiently fast.

Gruber, Lugrin and Martin~\cite{GruLugMar-78} proved that there are solutions to the BBGKY equations~\eqref{eq:BBGKY} for the 1D Coulomb gas $s=-1$ (discussed in detail in Section~\ref{sec:Kunz} below). Fontaine and Martin~\cite{FonMar-84} proved that the correlation functions constructed by Imbrie~\cite{Imbrie-82} for $s=1$ in $d=3$ also solve~\eqref{eq:BBGKY}. To our knowledge, it was not rigorously proved that the 1D log-gas (sine-$\beta$) is also a solution of the BBGKY hierarchy. For the two determinantal processes at $\beta=2$ in $d\in\{1,2\}$ (GUE and Ginibre) this must follow from the arguments in Ref.~\onlinecite{FonMar-84} or a calculation in the same spirit as in Ref.~\onlinecite{Tao-12}.

It is proved in Ref.~\onlinecite{GruMarOgu-82,AlaMar-85,Martin-88} that (whenever they exist) solutions
to~\eqref{eq:BBGKY} have truncated correlation functions $\rho^{(n)}_T$ which can essentially never decay faster than $|x-y|^{-(2d-s)}$ for $s\leq d-1$ and $|x-y|^{-(s+2)}$ for $d-1<s<d$. More precisely, for a fluid the Fourier transform of $\rho^{(2)}_T$ is expected to behaves like $|k|^{d-s}$ at the origin. This can be analytic only when $s\in d-2\N$ (like for Coulomb), in which case one expects Debye screening to happen and thus exponential decay of correlations. For $s\notin d-2\N$ the decay $|x-y|^{-(2d-s)}$ is conjectured to be optimal. This should be compared with the decay $|x-y|^{-s}$ found in the short range case~\eqref{eq:decay_correlation_short}.

On the other hand, at the origin $\rho^{(2)}(|x-y|)$ is believed to behave like $e^{-\beta/|x-y|^s}$ for $s>0$, and $|x-y|^\beta$ for $s=0$. The repulsion stays rather strong for $s>0$ but it is much weaker for $s=0$.

The question of whether the Mermin-Wagner argument applies in 1D and 2D for Jellium systems was studied in Ref.~\onlinecite{Baus-80,ChaDas-80,AlaJan-81b,MarMer-84,ReqWag-90}. The consensus seems to be that translation invariance can never be broken in 1D for all $0\leq s<1$, and in 2D for $1\leq s<2$. However, a rigorous proof is still lacking.

An important question is to determine the amount of screening in the system, due to the long range of the interaction. This can be measured in many different ways and we defer the discussion to Section~\ref{sec:prop} below.

\bigskip

\paragraph{The 1D log-gas / sine-$\beta$ ($s=0$).}\label{sec:sine-beta}
As usual, more is known about the log gas $s=0$ in $d=1$. The formula of the free energy $f(0,\beta,\rho)$ was given in~\eqref{eq:f_log_gas} in Remark~\ref{rmk:log_gas} and it is an analytic function of $\beta$ and $\rho$. It is proved in Ref.~\onlinecite{ErbHueLeb-21} that sine-$\beta$ is the unique minimizer (among stationary point processes) of the free energy per unit volume. Unfortunately, this does not show uniqueness for the DLR equation obtained in Ref.~\onlinecite{DerHarLebMai-21}. According to Ref.~\onlinecite{DerHarLebMai-21}, it however follows from Ref.~\onlinecite{KuiMin-19} that sine-2 is the unique solution to DLR.

It is instructive to take a look at the truncated two-point correlation function of the log gas at infinity, at the special values $\beta\in\{1,2,4\}$:
\begin{equation}
\rho^{(2)}_T(x-y)=
\begin{cases}
\dps -\frac{1}{\pi^2|x-y|^2}+\frac{3+\cos(2\pi|x-y|)}{2\pi^4|x-y|^4}+O\big(|x-y|^{-6}\big)&\text{for $\beta=1$,}\\[0.4cm]
\dps -\frac{\sin^2\pi|x-y|}{\pi^2|x-y|^2}&\text{for $\beta=2$,} \\[0.4cm]
\dps \frac{\cos(2\pi|x-y|)}{4|x-y|}-\frac{1+\frac\pi2\sin(2\pi|x-y|)}{4\pi^2|x-y|^2}+O\big(|x-y|^{-4}\big)&\text{for $\beta=4$.}
\end{cases}
\label{eq:correlation_RMT}
\end{equation}
We work here at $\rho_b=1$ for simplicity and refer to Ref.~\onlinecite[Chap.~7]{Forrester-10} for explicit expressions at $\beta\in\{1,4\}$. These formulas confirm that the decay is at most $1/|x-y|^2$, as discussed in the previous subsection. However, there are oscillatory terms and for $\beta=4$ oscillations occur to leading order.

Using the Calogero-Sutherland Hamiltonian~\eqref{eq:Calogero-Sutherland} and a theory of Haldane~\cite{Haldane-81} concerning the universal hydrodynamic theory for compressible quantum fluids, Forrester predicted in Ref.~\onlinecite{Forrester-84} (see also Ref.~\onlinecite{GanKam-01,AstGanLozSor-06}) that the two-point correlation function of the (unique) equilibrium state of the $\beta$--log gas should behave like
\begin{multline}
\rho^{(2)}_{T}(|x-y|)=-\frac{1}{\pi^2\beta|x-y|^2}+\sum_{m\geq1}\bigg(a_m\frac{\cos(2\pi m|x-y|)}{|x-y|^{4\beta^{-1}m^2}}+b_m\frac{\sin(2\pi m|x-y|)}{|x-y|^{4\beta^{-1}m^2+1}}\bigg)\\+O\left(\frac1{|x-y|^{4}}\right)
\label{eq:long-range-1D}
\end{multline}
for some coefficients $a_m,b_m$. In other words, there should always be oscillatory terms with a polynomial decay depending on $\beta$. The leading term in the expansion would thus be
\begin{equation}
\rho^{(2)}_{T}(|x-y|)\underset{|x-y|\to\ii}\sim
\begin{cases}
\dps-\frac{1}{\pi^2\beta|x-y|^2}&\text{for $\beta<2$,}\\[0.3cm]
\dps-\frac{1}{2\pi^2|x-y|^2}+\frac{\cos2\pi|x-y|}{2\pi^2|x-y|^2}&\text{for $\beta=2$,}\\[0.3cm]
\dps a_1\frac{\cos(2\pi |x-y|)}{|x-y|^{4\beta^{-1}}}&\text{for $\beta>2$.}
\end{cases}
\label{eq:decay_rho_2_conjectured}
\end{equation}
This is definitely true for $\beta\in\{1,2,4\}$ in~\eqref{eq:correlation_RMT}. We see that the decay at infinity displays a transition at $\beta=2$ from the universal monotonous decay $|x-y|^{-2}$ to an oscillating behavior with a non-universal decay depending of the temperature. This is the BKT transition that we have mentioned before. When $\beta\to\ii$, more and more oscillatory terms appear and  $\rho^{(2)}_{T}(|x-y|)$ converges to a periodic function. The corresponding point process is nothing but the uniform average of the crystal $\Z$, called in physics the \emph{floating crystal} and in random matrix theory the \emph{picked fence}. It is of course not clear when we can start calling the 1D log gas a ``solid'' or a ``quasi-solid''. Choosing the BKT transition $\beta=2$ seems the easiest choice.

Using the eigenfunctions and eigenvalues of the (integrable) Calogero-Sutherland Hamiltonian, Forrester could provide exact expressions for the correlation functions in terms of hypergeometric functions\cite{Forrester-10}. This enabled him to show the validity of~\eqref{eq:long-range-1D} for several possible values of $\beta$.
Even $\beta$'s are handled in Ref.~\onlinecite{Forrester-93} and rational $\beta$'s in Ref.~\onlinecite[Chap.~13]{Forrester-10}. It does not seem that~\eqref{eq:long-range-1D} has been shown in all cases. In Ref.~\onlinecite{GanKam-01,AstGanLozSor-06}, an approximate analysis based on the method of replicas provided the coefficients of the leading cosine terms.

\bigskip

\paragraph{The 1D Coulomb gas ($s=-1$).}\label{sec:Kunz}
The 1D Coulomb gas has been the object of many works~\cite{Prager-62,Baxter-63,Choquard-75,GruLugMar-78,AizMar-80,ChoKunMarNav-81,LugMar-82,JanLieSei-09,AizJanJun-10}. This is a beautiful model in statistical mechanics, where most of the interesting properties can be rigorously proved.

We have already mentioned the very simple formula~\eqref{eq:1D_Jellium_energy} for the energy, which immediately provides crystallization for every finite $N$. At positive temperature, the $N$-particle probability density in the interval $I_N=[-N/2,N/2]$ at unit density is
$$\bP_{N}(x_1,...,x_N)=Z_N^{-1}e^{-\beta\frac{N}{12}}\prod_{j=1}^Ne^{-\beta(x_j-j+1/2+N/2)^2}\1(x_j- x_{j-1}\geq0).$$
In his famous work Ref.~\onlinecite{Kunz-74}, Kunz rewrote this in terms of the fluctuations around the zero-temperature crystal. He introduced $y_j=x_j-j+1/2+N/2$ and obtained
$$\bP_{N}(x_1,...,x_N)=Z_N^{-1}e^{-\beta\frac{N}{12}}\prod_{j=1}^Ne^{-\beta y_j^2}\1(y_j- y_{j-1}\geq-1),$$
with the convention $y_0=0$.
This naturally brings the linear operator $K$ acting on $L^2(\R_+)$, with integral kernel $K(y,z)=e^{-\beta (y^2+z^2)/2}\1(y- z\geq-1)$. For instance we can express the partition function as
$$Z_N=e^{-\beta\frac{N}{12}}\pscal{e^{-\beta \frac{x^2}2}\1(x\leq N),K^Ne^{-\beta \frac{x^2}2}}_{L^2(\R_+)}.$$
Since $K$ has a positive kernel and is compact, the Perron-Frobenius theorem implies that it has a simple largest eigenvalue $\lambda_1(\beta)$, with a unique positive normalized eigenfunction $\psi_\beta\in L^2(\R_+)$. One concludes that the canonical free energy satisfies
$$F_{-1}(\beta,I_N)=-\frac{\log Z_N}{\beta}=N\left(\frac{1}{12}-\frac{\log\lambda_1(\beta)}{\beta}\right)-\frac2\beta \log\left(\int_0^\ii e^{-\beta \frac{y^2}{2}}\psi_\beta(y)\,\dy\right)+o(1).$$
In particular, $f(-1,\beta)=1/12-\beta^{-1}\log\lambda_1(\beta)$ is real-analytic, by analytic perturbation theory for the largest eigenvalue~\cite{Kato}. The correlation functions can also all be expressed in terms of $\lambda_1(\beta)$ and $\psi_\beta$ in a similar manner. Kunz finds for instance that in the thermodynamic limit the one-point correlation function converges to the $\Z$--periodic function
\begin{equation}
\rho^{(1)}(x)=\sum_{k\in\Z}\psi_\beta(-x-k-\tau)\psi_\beta(x+k+\tau),
 \label{eq:rho_1_Jellium_1D}
\end{equation}
for some $\tau\in[0,1)$. Since $\psi_\beta$ is exponentially decreasing, one can show that the truncated two-point correlation $\rho^{(2)}_T$ decreases exponentially, which is Debye-Hückel screening in 1D. In order to make sure that crystallization really happens, it is still needed to show that the density in~\eqref{eq:rho_1_Jellium_1D} is a genuine (non-constant) $\Z$--periodic function. Using analyticity, Kunz managed to prove this only for $\beta$ outside of a bounded countable set. A different argument was later given by Brascamp and Lieb\cite{BraLie-75} using inequalities for Gaussians. The non-trivial periodicity was finally shown by Aizenman and Martin\cite{AizMar-80} for all $\beta>0$. This was later generalized to the 1D quantum Jellium~\cite{BraLie-75,JanJun-14} and quasi-one dimensional systems~\cite{AizJanJun-10}.

Kunz also considered the periodic Jellium problem described in Section~\ref{sec:periodic_BC} and proved that the limit is the uniform average over $\tau$ of the previous periodic state. The breaking of symmetry can then not be detected in the density anymore and one should look into the truncated two-point correlation function, which does not decay at infinity (long range order). Finally, Kunz also considered an excess finite charge and proved that, although it affects the free energy to leading order in dimension $d=1$, it escapes to the boundary without affecting the limit of the correlation functions in the bulk.

A few years later, Aizenman and Martin\cite{AizMar-80} reformulated the problem in terms of the electric field, which is piecewise linear of slope 1 with a jump $-1$ at each $x_j$. Using the ergodic theorem they proved periodicity for all $\beta$ and showed that the associated point process solves a properly renormalized version of the canonical DLR, BBGKY and KMS equations. In order to make sense of the potential $\Phi(x)$ which we have largely studied in Section~\ref{sec:periodic}, they first proved that the point process is \emph{number rigid} (see Section~\ref{sec:prop}), that is, the number of points in an interval is deterministic once the outside configurations are fixed.

The conclusion is that for the 1D Coulomb case the set $\cR_{-1,\beta,\rho_b}$ is a non-trivial convex set of the form $\{\rho_b^{-1}(\tau+\mathscr{P}_\beta),\tau\in[0,1)\}$ where $\mathscr{P}_\beta$ is a $\Z$--periodic point process depending analytically of $\beta$, and $\rho_b^{-1}(\tau+\mathscr{P}_\beta)$ denotes the translation by $\tau$ and dilation by $\rho_b^{-1}$.

Note that the Perron-Frobenius (or Krein-Rutman) theorem is very often used for one-dimensional systems. This goes back at least to Van Hove~\cite{vanHove-49,vanHove-50}. Recently, Ducatez\cite{Ducatez-18_ppt} has suggested to use instead the Birkhoff-Hopf theorem and managed to construct the infinite Jellium Gibbs measure for any kind of background density (not necessarily a constant), bounded and strictly positive everywhere.

\subsection{Phase diagram: numerics}

Existing results on the shape of the phase diagram are essentially numerical. There is a huge number of works in the physical literature. In the long range case, most authors concentrate on $s=1$ in dimensions $d\in\{2,3\}$ (3D Coulomb) or $s=0$ in dimensions $d\in\{1,2\}$ (2D Coulomb). The case $s=3$ is also interesting when considering dipoles, but many more values of $s$ are typically considered in the short range case. Several works naturally focus on the quantum gas which we have not described in this article. The classical gas then occurs at low density for $s<2$ and high density for $s>2$ (depending on when the interaction $\sim\rho^{1+s/d}$ overcome the kinetic energy $\sim\rho^{1+2/d}$). Recall that $s=2$ is the Calogero-Sutherland model in $d=1$. 

For reviews about the physical properties and numerical simulations on Jellium we mention Ref.~\onlinecite{BauHan-80,Choquard-78,Mazars-11,LooGil-16}. The French-speaking reader can also read the excellent Ref.~\onlinecite{Alastuey-86}.

\begin{figure}[t]
\includegraphics[width=12.5cm]{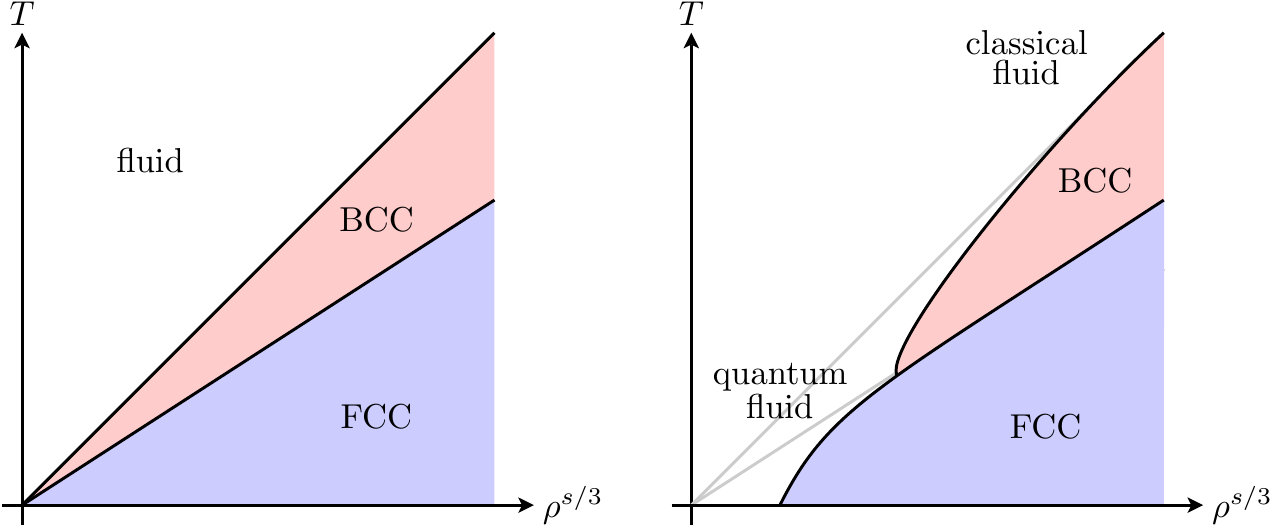}
\caption{Typical phase diagram in dimension $d=3$ in the short range case $s>3$, with $T=1/\beta$. The BCC phase has only been observed for $3<s\lesssim 7$. In the classical case (left) the system is parametrized by $\Upsilon =\rho^{s/3}/T$, whereas in the quantum case (right) the two variables are necessary. The classical case is recovered at high density where $\rho^{1+s/3}$ overcomes the kinetic energy which is of order $\rho^{5/3}$.\label{fig:phase_diagram_short}}

\bigskip

\includegraphics[width=12.5cm]{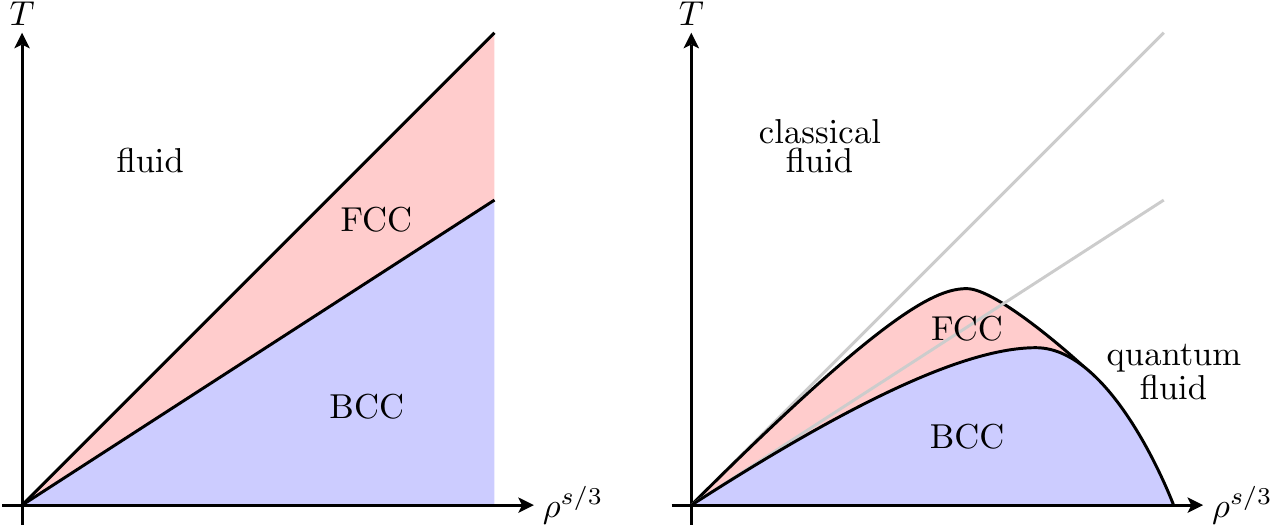}
\caption{Typical phase diagram in dimension $d=3$ in the long range case, here for $s<3/2$. The FCC phase is not present anymore in the Coulomb case $s=1$. The classical limit is now at low density where $\rho^{1+s/3}$ overcomes the kinetic energy.\label{fig:phase_diagram_long}}
\end{figure}

In dimension $d=3$, the phase diagram was found to have the form outlined in Figure~\ref{fig:phase_diagram_short}, for $s>3$. At small values of $\Upsilon=\beta\rho^{\frac{s}d}$ or $\Upsilon'=\beta e^{\beta\mu\frac{s}d}$, the system is a fluid by Theorem~\ref{thm:unique}. This is expected to happen until a critical value of $\Upsilon$. For larger values of $\Upsilon$, it is believed to form a solid. Numerical simulations~\cite{HooRosJohHenBarBro-70,HooGraJoh-71,LaiHay-92,AgrKof-95,AgrKof-95b,Traversset-14,PreSaiGia-05} suggest that there are two critical values $\Upsilon_{\rm FCC}(s)$ and $\Upsilon_{\rm BCC}(s)$ such that the system is a
$$\begin{cases}
\text{fluid} & \text{for $0< \Upsilon<\Upsilon_{\rm BCC}(s)$,}\\
\text{Body-Centered Cubic (BCC) crystal} & \text{for $\Upsilon_{\rm BCC}(s)<\Upsilon<\Upsilon_{\rm FCC}(s)$,}\\
\text{Face-Centered Cubic (FCC) crystal} & \text{for $\Upsilon>\Upsilon_{\rm FCC}(s)$,}
\end{cases}$$
where $\Upsilon_{\rm BCC}(s)<\Upsilon_{\rm FCC}(s)$ only for about $s\lesssim7$. Only the FCC crystal is present for $s\gtrsim 7$. A numerical approximation of the two functions $\Upsilon_{\rm FCC}(s)$ and $\Upsilon_{\rm BCC}(s)$ in terms of the parameter $s$ is computed in Ref.~\onlinecite[Fig.~8]{PreSaiGia-05}. The phase diagram in the quantum case resembles that of the classical case at large density, where the kinetic energy is negligible. But at low densities the kinetic energy becomes dominant, and the system is a quantum fluid (see Figure~\ref{fig:phase_diagram_short}).

In the long range case $s<d=3$, the phase diagram is believed to take a very similar form, except that the two lattices FCC and BCC must be interchanged for $s<3/2$\cite{Nijboer-75,BlaLew-15}, see Figure~\ref{fig:phase_diagram_long}. Most computations focus on the Coulomb case $s=1$, where only the BCC phase is expected to be present at all temperatures~\cite{BruSahTel-66,PolHan-73,SarMerCal-76,JonCep-96}. The phase transition happens at about $\Upsilon\simeq 175$. Note that when spin is taken into account, the system could undergo more phase transitions for the spin variable~\cite{AldCep-80,ZonLinCep-02,DruRadTraTowNee-04,CanBerCep-04,HolMor-20,AzaDru-22_ppt}.

In dimension $2$, translational symmetry will not be broken, by Theorem~\ref{thm:Mermin-Wagner}.
Kosterlitz, Thouless, Halperin, Nelson and Young~\cite{KosTho-72,KosTho-73,HalNel-78,NelHal-79,Young-79} have predicted that the melting of 2D crystals happens in a two-step procedure. Namely, the phase diagram resembles that of the left of Figure~\ref{fig:phase_diagram_short}, where `BCC' is replaced by an intermediate \emph{``hexatic'' phase} with fast decay of translational correlations and slow decay of (six-fold) orientational correlations, whereas `FCC' is replaced by a \emph{``solid'' or ``quasi-solid'' phase} which has slow decay of translational correlations and long range rotational order. The KTHNY scenario has been confirmed by many simulations and even experiments~\cite{MurWin-87,ZahLenMar-99,GasEisMarKei-10,KniWuHuaSerXiaPfeWes-18}. It is very hard to catch numerically. Numerical simulations and discussions may be found for several values of $s>d=2$ in Ref.~\onlinecite{BagAndSwo-96,LinZheTri-06,GasEisMarKei-10,KapKra-15,MazSal-19}

For many years, numerical simulations\cite{GanChaChe-79,GriAda-79,LeePerrSmi-82,CaiLevWeiHan-82} and even experiments~\cite{GriAda-79} indicated a first order phase transition in 2D Jellium systems, both for $s=0$ and $s=1$. See Ref.~\onlinecite{Strandburg-88} for a review. Recent numerical calculations however seem to confirm for both $s=0$ and $s=1$ the presence of the same two-step transition as in the short range case, with an intermediate hexatic phase\cite{CaiLev-86,Strandburg-88,MooPer-99,MutAok-99,HeCuiMaLiuZou-03,ClaCasCep-09,Mazars-15,Salazar-PhD,MazSal-19}. This happens around $\Upsilon\simeq 120$ and $140$ for $s=1$\cite{ClaCasCep-09} and  $\Upsilon\simeq134$ and $162$ for $s=0$~\cite{Stephenson-PhD}, which is way above the Ginibre point $\Upsilon=2$.

We end this section with comments on the one-dimensional Riesz gas, for which we found much less numerical results. Figure~\ref{fig:1D} summarizes the situation. The Riesz point process is periodic at $s=-1$ for all values of $\Upsilon$ and (known or believed to be) unique and translation-invariant for all $s\geq0$. At $s=2$, we recover the \emph{classical} Calogero-Sutherland-Moser model. At $s=0$ it is the sine-$\beta$ process of random matrices, which is equivalent to the \emph{quantum} Calogero-Sutherland model at zero temperature. There is a BKT transition at $\Upsilon=2$ (GUE). A natural question is to investigate the nature of the phase diagram in the plane $(\Upsilon,s)$. For instance, is the transition to a solid given by a smooth curve starting somewhere below $s=0$ and diverging with the temperature at $s=-1$? A similar question can be asked for the BKT transition. In a harmonic trap, one-dimensional Riesz gases were studied in Ref.~\onlinecite{AgaDhaKulKunMajMukSch-19,KetKulKunMajMukSch-21,SanKetAgaDhaKulKun-21,KetKulKunMajMukSch-22,FlaMajSch-22_ppt}.

\begin{figure}[t]
\includegraphics[width=9cm]{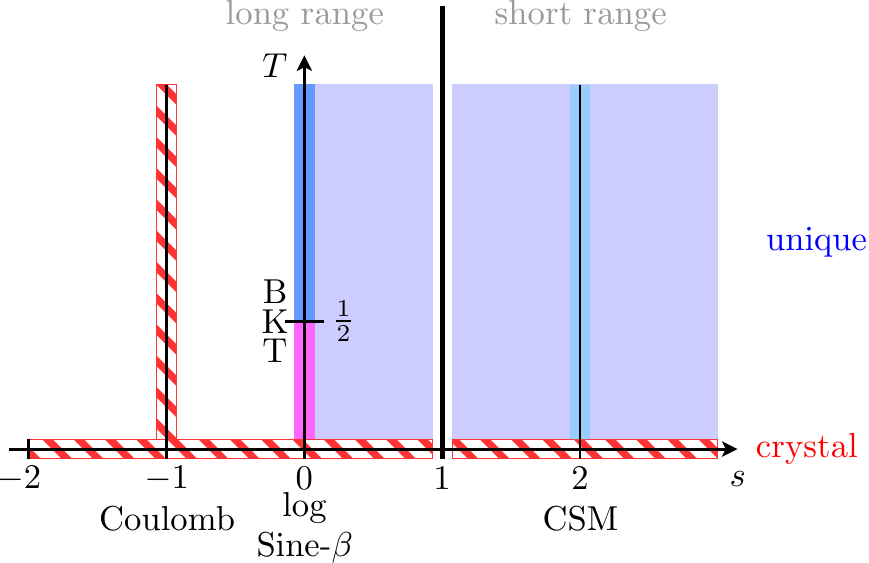}
\caption{Shape of the phase diagram of the 1D Riesz gas in the plane $(s,\Upsilon^{-1}= T)$ at unit density. The system is a solid at zero temperature for all $s$, which has only been proved so far for $s\geq0$ and $s=-1$ (Theorem~\ref{thm:crystallization}). The Coulomb gas ($s=-1$) is crystallized at all temperatures~\cite{Kunz-74}, whereas translation-invariance is known for $s>1$\cite{FroPfi-81} and expected for all $s\geq0$\cite{Baus-80,ReqWag-90}. Uniqueness is known for $s>2$\cite{Papangelou-87} and probably expected for all $s\geq0$. At $s=0$ we obtain the log-gas, also called sine-$\beta$ or Dyson gas, which appears in the theory of random matrices as well as many other situations. It is as well the zero-temperature quantum Calogero-Sutherland model describing particles interacting with a potential proportional to $1/(x_j-x_k)^2$. The point $T=1/2$ (GUE) is a Berezinski-Kosterlitz-Thouless transition where the leading term~\eqref{eq:decay_rho_2_conjectured} in the decay at infinity of the truncated two-point correlation function suddenly starts to oscillate. At $s=2$ we obtain the classical Calogero-Sutherland-Moser (CSM) model, which is completely integrable. An interesting question is to fill the gaps in the picture and understand, in particular, if there exists a smooth transition curve to a periodic crystal in the region $-1<s<0$. A similar question concerns the BKT transition.\label{fig:1D}}
\end{figure}

\subsection{Screening, hyperuniformity, rigidity}\label{sec:prop}
Many questions are open for Riesz gases. We focus in this section on some properties which are specific to long range systems and do not usually occur in the short range case, namely \emph{sum rules}, \emph{hyperuniformity} and \emph{rigidity}. These concepts have attracted a lot of attention and are all different (but related) ways of measuring the amount of \emph{screening} in the system. A recent summary of what is known in this direction can be read in Ref.~\onlinecite{GhoLeb-17}.

Usual point processes such as Poisson have a variance in the number of points in a domain which is extensive, that is, behaves like the volume of the considered set. In some cases, it can even grow faster. By~\eqref{eq:nb_correlation}, the variance in a domain $D$ is
$$\bE[n_D^2]-\bE[n_D]^2=\iint_{D\times D}\rho_T^{(2)}(x,y)\,\dx\,\dy+\int_ D\rho^{(1)}(x)\,\dx$$
with $\rho_T^{(2)}(x,y)=\rho^{(2)}(x,y)-\rho^{(1)}(x)\rho^{(1)}(y)$ the truncated two-point function. For a translation-invariant (stationary) point process with a sufficiently fast decay of correlation, this converges to
\begin{equation}
\lim_{D\nearrow\R^d}\frac{\bE[n_D^2]-\bE[n_D]^2}{|D|}=\int_{\R^d}\rho_T^{(2)}(z)\,\dz+\rho,
\label{eq:variance_per_unit_vol}
\end{equation}
where $\rho=\rho^{(1)}$ and, as usual, we identified $\rho^{(2)}_T(x,y)$ with $\rho^{(2)}_T(x-y)$. The right side of~\eqref{eq:variance_per_unit_vol} is called the \emph{bulk compressibility} and it appeared already in the decay of correlation~\eqref{eq:decay_correlation_short} in the short range case $s>d$ at low $\Upsilon$. The Ginibre inequality~\eqref{eq:variance_Ginibre} proves that it is always strictly positive for short range Riesz point processes. It could be infinite if the decay of correlation is not fast enough (integrable).

A point process is called \emph{hyperuniform} (or super homogeneous) when the variance is a $o(|D|)$~\cite{TorSti-03,Torquato-18}. Should the truncated two-point correlation be integrable, hyperuniformity is thus equivalent to the relation~\cite{GhoLeb-17}
\begin{equation}
\int_{\R^d}\rho_T^{(2)}(z)\,\dz+\rho=0.
\label{eq:sum_rule1}
\end{equation}
The short range Riesz point processes ($s>d$) are therefore never hyperuniform. In the long range case $s<d$ we have explained that, for $\Upsilon$ small enough, it is expected that there is a unique Gibbs point process, with correlations decaying like $|x-y|^{-(2d-s)}$ or even exponentially for $s\in d-2\N$. This corresponds to a Fourier transform behaving like $|k|^{d-s}$. Since $2d-s>d$, the integral in~\eqref{eq:sum_rule1} makes sense and they will thus be hyperuniform whenever~\eqref{eq:sum_rule1} holds.

The formula~\eqref{eq:sum_rule1} is called a \emph{canonical sum rule}. It is a non-trivial relation between the correlation functions, which is expected to be automatically satisfied due to the necessity of having a sufficiently strong screening in the system. Sum rules have been the object of many works in the 80s\cite{GruLugMar-80,MarYal-80,GruLebMar-81,BluGruLebMar-82,FonMar-84,LebMar-84,AlaJan-84,Martin-88} and their relation to hyperuniformity was also well known back then (see, e.g., Ref.~\onlinecite[Prop.~10]{GruLugMar-80}). One important result proved in Ref.~\onlinecite{GruLugMar-80,GruLebMar-81} is that the sum rule~\eqref{eq:sum_rule1} must always be satisfied for a solution of the BBGKY equations~\eqref{eq:BBGKY}, whenever correlations decay faster than $1/|x-y|^{d+\eps}$ for $s\leq d-1$ and faster than $1/|x-y|^{s+1+\eps}$ for $d-1<s<d$. From the expected decay mentioned earlier, we see that \emph{the Riesz point processes should be hyperuniform for all $s<d$ (at least when~$\Upsilon$ is small enough)}. Hyperuniformity is a typical feature of long range forces. It is easily verified to hold on all the exactly known solutions such as $\beta=2$ and $s=0$ in dimensions $d\in\{1,2\}$. This was in fact already mentioned by Dyson in his discussion of the 1D log gas in 1962\cite{Dyson-62c}.

For hyperuniform processes, the next natural question is to determine the precise behavior of the variance.
For Coulomb, the variance is expected to be exactly a surface term $|\partial D|=o(|D|)$. This is so far known for the 1D Coulomb gas, the 2D Coulomb gas at $\beta=2$ and the 3D Coulomb gas at small $\Upsilon$, see Ref.~\onlinecite{MarYal-80,Lebowitz-83,JanLebMag-93,JanLeb-01,AizGolLeb-01,GhoLeb-17}. The exact behavior of the variance for general point processes is studied in the recent work Ref.~\onlinecite{AdhGhoLeb-21}, in terms of the Fourier transform at the origin. For Riesz gases the conclusion is that the variance should go like $|D|^{s/d}$ for $d-1<s<d$, $|\partial D|\log|D|$ for $s=d-1$ and $|\partial D|$ for all $s<d-1$ (for small $\Upsilon$).

The sum rule~\eqref{eq:sum_rule1} might look mysterious at first sight. It is however nothing else but the usual relation between the one and two-point correlation functions in a \emph{canonical finite system}, for which
$$\iint_{\Omega\times\Omega}\rho^{(2)}=N(N-1)=\left(\int_\Omega\rho^{(1)}\right)^2-\int_\Omega\rho^{(1)}$$
by~\eqref{eq:nb_correlation}. In other words, the equation~\eqref{eq:sum_rule1} says that our infinite system is very close to being canonical locally. We have only discussed the first sum rule but there are similar relations between $\rho^{(n)}$ and $\rho^{(1)},...,\rho^{(n-1)}$ and they have been shown to all hold as well. When $s$ is further decreased, new sum rules appear, involving higher multipoles of $\rho^{(2)}_T$\cite{GruLebMar-81,BluGruLebMar-82,Martin-88}. This is a manifestation of the ability of the long range system to appropriately screen the background and thus effectively reduce the long range of the interaction.

Another important concept is that of \emph{number rigidity}\cite{GhoPer-17,GhoLeb-17}. This happens when the number of points in a domain $D$ is a deterministic function of the points outside of $D$. That is, the system is exactly canonical once the outside is fixed. This property was shown for the 1D Coulomb gas in Ref.~\onlinecite{AizMar-80} and crucially used there to define the potential in the DLR equation. For the 1D log gas, number rigidity was proved for all values of $\beta$ simultaneously in Ref.~\onlinecite{RedNaj-18,DerHarLebMai-21}. At $\beta=2$ this was shown before in Ref.~\onlinecite{Ghosh-15}. In Ref.~\onlinecite{DerVas-21_ppt} Dereudre and Vasseur show that the Riesz point process is \emph{not} number rigid for $d-1<s<d$ and conjecture it would start to be so at $s=d-1$. In Ref.~\onlinecite{GhoLeb-17b}, Ghosh and Lebowitz show that a hyperuniform point process for which
$$|\rho^{(2)}_T(x,y)|\leq \begin{cases}
\dps \frac{C}{1+|x-y|^2}&\text{for $d=1$,}\\
\dps \frac{C}{1+|x-y|^{4+\eps}}&\text{for $d=2$,}
\end{cases}
$$
is automatically number-rigid. In dimension $d=2$ where the maximal decay of correlation is $|x-y|^{-(2d-s)}=|x-y|^{-(4-s)}$, this suggests that number rigidity should rather start at $s=0$. In higher dimensions, things are less clear. Peres and Sly have found that even an exponential decay of correlations is not enough to guarantee that a hyperuniform point process is number rigid~\cite{PerSly-14}.

\section{Additional proofs}\label{sec:appendix}

\subsection{Equivalence of ensembles for $d-2<s<d$ and limit for net charges}\label{sec:net_charge}

In this section we study the Jellium energy for systems with a net charge and prove the limit~\eqref{eq:net_charge}. We will also explain how the existence of the thermodynamic limit in the canonical case (Theorem~\ref{thm:limit_C_Jellium}) follows from the grand-canonical limit in Theorem~\ref{thm:limit_GC_Jellium}, and why the two coincide: $f(s,\beta,\rho)=f^{\rm GC}(s,\beta,\rho)$.

To simplify our discussion, we assume everywhere that $\rho_b=1$ and will only state results for $F_s(\beta,N,\Omega):=F_s(\beta,N,\Omega,1)$, with the convention that this coincides with $E_s(N,\Omega,1)$ whenever $\beta=+\ii$.

\subsubsection{Proof of~\eqref{eq:net_charge} for non-neutral canonical systems}
First, we provide upper and lower bounds on the free energy for non-neutral systems, that is, $F_s(\beta,N\pm Q,N^{1/d}\omega)$ where $Q$ is called a `net charge'. This was already done in the Coulomb case in Ref.~\onlinecite{Kunz-74,LieLeb-72,LieNar-75,GraSch-95b}. We thus assume $s>\max(0,d-2)$ throughout. We work with $s>0$ to simplify the statement and since we are mainly interested in the link with the grand-canonical problem. The following arguments can be adapted to $s\in(-1,0]$ in dimension $d=1$. However, our proofs do not easily generalize to $s<d-2$ since we use some properties of the $s$-Riesz equilibrium measure of the domain $\omega$ which are only valid for $s\geq d-2$.

We recall that the $s$--Riesz capacity is
\begin{equation}
{\rm Cap}_s(\omega)^{-1}=\min_{\nu}\iint_{\overline\omega\times\overline\omega}\frac{\rd\nu(x)\,\rd\nu(y)}{|x-y|^s},
 \label{eq:capacity}
\end{equation}
as defined earlier in~\eqref{eq:mean-field}. For a bounded domain $\omega$ we have ${\rm Cap}_s(\omega)<\ii$ and there is a unique minimizer $\nu$. It concentrates on $\partial\omega$ for $s\leq d-2$ but lives over the whole of $\omega$ for $d-2<s<d$\cite{Landkof-72,GolNovRuf-15}. For $d-2\leq s<d$ the associated Riesz potential $\nu\ast|x|^{-s}$ is constant over $\omega$, equal to ${\rm Cap}_s(\omega)$. This plays a role in our proof.

The first lemma contains estimates that we can easily get after placing the excess or missing points in the equilibrium measure $\nu$.

\begin{lemma}[Net charge I]\label{lem:net_charge1}
Let $\max(0,d-2)<s<d$ and $\rho_b=1$. Let $\omega$ be a bounded open set with a $C^{1,\alpha}$ boundary, $\alpha>0$, and $|\omega|=1$. Define $\Omega_N:=N^{1/d}\omega$ with $N\in\N$. Then, for any $0\leq Q\leq N$ and $\beta\in(0,+\ii]$, we have for a constant $C$ depending only on $\omega$
\begin{align}
 F_s(\beta,N-Q,\Omega_N)&\geq F_s(\beta,N,\Omega_N)+\frac{Q^2}{2N^{\frac{s}d}{\rm Cap}_s(\omega)}-C Q\beta^{-1},\label{eq:negative_charge_lower}\\
  F_s(\beta,N+Q,\Omega_N)&\leq F_s(\beta,N,\Omega_N)+\frac{Q^2}{2N^{\frac{s}d}{\rm Cap}_s(\omega)}+C Q\beta^{-1}.
 \label{eq:positive_charge_upper}
\end{align}
\end{lemma}

\begin{proof}
Let $\nu$ be the $s$--Riesz equilibrium measure of $\omega$ and let $\nu_N=N^{-1}\nu(\cdot/N^{1/d})$ be that of $\Omega_N$. To obtain the first bound~\eqref{eq:negative_charge_lower}, we place $N-Q$ points in the Gibbs measure for $F_s(\beta,N,\Omega_N)$ (resp. minimizing positions for $\beta=+\ii$) together with $Q$ independent points distributed according to the measure $\nu_N$. We then we use this trial state for $F_s(\beta,N,\Omega_N)$. The interaction energy of the $Q$ points is
$$\frac{Q(Q-1)}2\iint_{\Omega\times\Omega}\frac{\rd\nu_N(x)\,\rd\nu_N(y)}{|x-y|^s}=\frac{Q(Q-1)}{2N^{\frac{s}d}{\rm Cap}_s(\omega)}.$$
Since $d-2<s<d$, the $Q$ points generate a constant potential over the whole domain $\Omega_N$, given by $V_Q(x):=Q\,\nu_N\ast|x|^{-s}=QN^{-\frac{s}d}{\rm Cap}_s(\omega)^{-1}$ . The interaction with the background and the other $N-Q$ points is thus equal to
$$(N-Q-N)V_Q=-\frac{Q^2}{N^{\frac{s}d}{\rm Cap}_s(\omega)}.$$
The entropy of the $Q$ points is
\begin{align*}
\frac1{Q!}\int_{\Omega^Q}Q!(\nu_N)^{\otimes Q}\log\big(Q!(\nu_N)^{\otimes Q}\big)&= \log(Q!)+Q\int_\Omega \nu_N\log\nu_N\\
&=\log(Q!)-Q\log N+Q\int_\omega \nu\log\nu\leq Q\int_\omega \nu\log\nu,
\end{align*}
since $Q!\leq Q^Q\leq N^Q$. Under our assumption on $\partial\omega$, it is known that $\nu$ is continuous in $\omega$ and diverges like $\rd(x,\partial\omega)^{(s-d)/2}$ close to the boundary (this follows from results in Ref.~\onlinecite{RosSer-16,RosSer-17}, as proved in Ref.~\onlinecite[Lem.~4.1]{GolNovRuf-22_ppt}). In particular
$\nu \in L^{p}(\omega)$ for $1\leq p<2/(d-s)$ and $\int_\omega \nu\log\nu$ is finite. Discarding the negative terms, we have thus proved that
$$F_s(\beta,N,\Omega_N)\leq F_s(\beta,N-Q,\Omega_N)-\frac{Q^2}{2N^{\frac{s}d}{\rm Cap}_s(\omega)}+Q\beta^{-1}\int_\omega \nu\log\nu.$$
This concludes the proof of~\eqref{eq:negative_charge_lower}.
The proof of~\eqref{eq:positive_charge_upper} goes along the same lines, using $N$ points in the Gibbs state for $F_s(\beta,N,\Omega_N)$ (resp. a minimizer for $\beta=+\ii$) and $Q$ points in the measure $\nu_N$.
\end{proof}

The reverse bounds are more complicated because we need to compensate the excess or missing charge using the background. We have much less freedom since the latter ought to be exactly constant in our definition of the Jellium energy. To simplify our discussion, we will only consider \emph{smooth convex domains} $\omega$ (star-shaped would suffice). Those satisfy the useful property that $\omega\subset \lambda\omega$ for all $\lambda\geq1$ (after a proper choice of origin). Note that in the first work Ref.~\onlinecite{LieLeb-72} on the thermodynamic limit for Coulomb systems, Lieb and Lebowitz only considered ellipsoids. General sets were handled later in Ref.~\onlinecite{GraSch-95b}. Our approach also works for general smooth domains, but then we have to allow small deformations of $\partial\omega$ and this goes beyond the scope of this paper.

The following is in the same spirit as Lemma~\ref{lem:net_charge1}, except that it involves the neutral problem $F_s(\beta,M,\Omega_M)$ for an $M=N+o(N)$ not necessarily equal to $N$.

\begin{lemma}[Net charge II]\label{lem:net_charge2}
Let $\max(0,d-2)<s<d$ and $\rho_b=1$. Let $\omega$ be a bounded convex open set with a $C^{1,\alpha}$ boundary, $\alpha>0$, and $|\omega|=1$. Define $\Omega_N:=N^{1/d}\omega$ with $N\in\N$. For any $Q\leq N$, there exists an integer $M\in\N$ with
\begin{equation}
 N+Q\leq M\leq N+CN^{1-\theta}Q^{\theta}+1,\qquad \theta:=\frac2{2+d-s},
 \label{eq:M_Q_excess}
\end{equation}
such that
\begin{multline}
 F_s(\beta,N+Q,\Omega_N)\geq F_s\big(\beta,M,\Omega_M\big)+\frac{Q^2}{2\big(1+(Q/N)^\theta\big)^sN^{\frac{s}d}{\rm Cap}_s(\omega)}\\-C\left(1+\beta^{-1}\right)(1+N^{1-\theta}Q^{\theta}).
  \label{eq:positive_charge_lower}
\end{multline}
Similarly, for any $Q\leq N/2$ there exists an $ N-Q\geq M\geq N-CN^{1-\theta}Q^{\theta}-1$ such that
\begin{multline}
 F_s(\beta,N-Q,\Omega)\leq F_s\big(\beta,M,\Omega_M\big)+\frac{Q^2}{2\big(1-(Q/N)^\theta\big)^sN^{\frac{s}d}{\rm Cap}_s(\omega)}\\+C\left(1+\beta^{-1}\right)(1+N^{1-\theta}Q^{\theta}).
  \label{eq:negative_charge_upper}
\end{multline}
The constant $C$ only depends on $\omega$.
\end{lemma}

\begin{proof}
Upon changing the origin of space, we can assume that $0\in\omega$. Our construction relies on an approximation of the $s$-Riesz equilibrium measure, which is bounded, supported in a neighborhood of $\omega$ and, more importantly, still produces an exactly constant potential in $\omega$. For instance, let us fix $0<\eta\leq1$ to be chosen later and consider the average over dilations
$$\nu'(x):=\eta^{-1}\int_1^{1+\eta}t^{-d}\nu(x/t)\,\dt$$
where $\nu$ is the $s$--Riesz equilibrium measure of $\omega$. The measure $\nu'$ is supported on $(1+\eta)\omega$ and satisfies $\nu'(x)\leq \kappa\eta^{\frac{s-d}{2}}$. The latter estimate follows from the fact that $\nu$ behaves like $\rd(x,\partial\omega)^{\frac{s-d}{2}}$ close to the boundary\cite{RosSer-16,RosSer-17,GolNovRuf-22_ppt}, and that $\partial\omega$ always intersects the line $x\R$ transversally (since $\omega$ is smooth and convex with $0\in\omega$). Instead of averaging over dilations we could similarly average over translations, using a convolution.

The rescaled measure $\nu_N':=N^{-1}\nu'(x/N^{1/d})$ is supported on $\Omega_N':=(1+\eta)\Omega_N$ and satisfies $\nu_N'\leq \kappa\eta^{\frac{s-d}{2}}N^{-1}$. By scaling it generates the constant Riesz potential
$$\nu_N'\ast|\cdot|^{-s}(x)=\frac1{\eta N^{\frac{s}d}{\rm Cap}_s(\omega)}\int_1^{1+\eta}\frac{1}{t^s}\,\dt,\qquad \forall x\in\Omega_N.$$
In $\Omega'_N\setminus \Omega_N$ the potential is not constant anymore. The idea is to replace the background $\1_{\Omega_N}$ by $\1_{\Omega_N}+Q\nu_N'$ in order to obtain a neutral system. Our problem is however settled with a uniform background and we have to compensate for the change due to $Q\nu_N'$. To this end, we use a constant background in $\Omega_N'$ slightly above $\1_{\Omega_N}+Q\nu_N'$, for instance
$$\rho_b':=1+\kappa QN^{-1}\eta^{\frac{s-d}{2}}+\frac{\tau}{(1+\eta)^d} N^{-1}$$
with  $\kappa$ the same constant as in the upper bound for $\nu'$ and $\tau\in[0,1)$ chosen to ensure that
$$M:=\rho'_b|\Omega_N'|=\left(1+\kappa QN^{-1}\eta^{\frac{s-d}{2}}\right)(1+\eta)^dN+\tau$$
is an integer. See Figure~\ref{fig:net_charge}. In the proof we will need that $\eta$ and $QN^{-1}\eta^{\frac{s-d}{2}}$ are both bounded. In applications they will in fact be small. Since
$M=N+O( Q\eta^{\frac{s-d}{2}}+\eta N+1),$
after optimizing over $\eta$, we are led to choosing
\begin{equation}
 \eta:=\left(\frac{Q}{N}\right)^{\frac{2}{2+d-s}}\leq1.
 \label{eq:eta_choice}
\end{equation}
With this choice we have $Q\eta^{\frac{s-d}{2}}= \eta N$ and thus $N+Q\leq M\leq N(1+C\eta)+\tau$.

\begin{figure}[t]
\includegraphics[width=7cm]{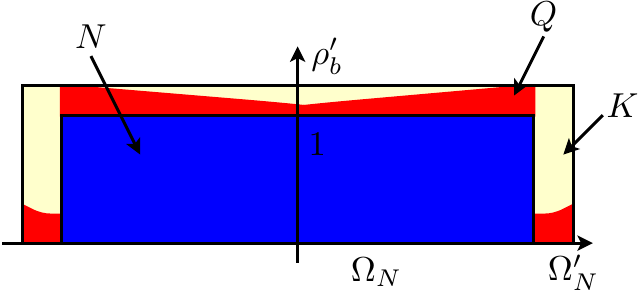}
\caption{Decomposition of the background $\rho'_b\1_{\Omega_N'}$ used in the proof of Lemma~\ref{lem:net_charge2}. In the smaller domain $\Omega_N=N^{1/d}\omega$ we have $N+Q$ points. To compensate the excess $Q$ we use a background $Q\nu_N'$ (in red in the picture), where $\nu_N'$ is a smooth approximation of the $s$-Riesz equilibrium measure of $\Omega_N$. It lives in the slightly larger set $\Omega'_N$ and is chosen so that the associated Riesz potential $Q\nu_N'\ast|x|^{-s}$ is still exactly constant in $\Omega_N$. In order to have a uniform  background on $\Omega_N'$ we also introduce the missing part $m_N$ (in yellow on the picture) with $\int m_N=K$ but completely screen it using $K$ points.\label{fig:net_charge}}
\end{figure}

Next we construct a trial state. Let
$$m_N:=\rho_b'\1_{\Omega_N'}-\1_{\Omega_N}-Q\nu_N'\geq0$$
denote the missing background charge and $K:=\int_{\Omega_N'}m_N=M-N-Q$. We have
$$m_N\leq (\kappa \eta +N^{-1})\1_{\Omega_N}+\left(1+\kappa\eta +N^{-1}\right)\1_{\Omega_N'\setminus\Omega_N}\leq (\kappa \eta +N^{-1})\1_{\Omega_N}+(2+\kappa)\1_{\Omega_N'\setminus\Omega_N}.$$
We pick as trial state the Gibbs state for $F_s(\beta,N+Q,\Omega_N,1)$, together with $K$ independent points chosen according to the measure $m_N/K$. These $K$ points are compensating $m_N$ exactly. No interaction is left with the rest of the system.  The free energy of the $K$ points together with the part $m_N$ from the background equals
\begin{multline*}
-\frac1{2K}\iint_{\Omega_N'\times\Omega_N'}\frac{m_N(x)\,m_N(y)\dx\,\dy}{|x-y|^s}+\beta^{-1}\log(K!)+\beta^{-1}\int_{\Omega_N'}m_N\log \frac{m_N}{K}\\
\leq \beta^{-1}\int_{\Omega_N'}m_N\log m_N\leq C\beta^{-1}(N\eta+1).
\end{multline*}
We are now left with our $N+Q$ points and the background $\1_{\Omega_N}+Q\nu_N'$. The background self-energy of $Q\nu_N'$ together with its interaction with the $N+Q$ points and the uniform background $\1_{\Omega_N}$ give
\begin{multline*}
\frac{Q^2}{2N^{\frac{s}{d}}}\iint\frac{\rd\nu'(x)\,\rd\nu'(y)}{|x-y|^s}-\frac{Q^2}{\eta N^{\frac{s}d}{\rm Cap}_s(\omega)}\int_1^{1+\eta}\frac{1}{t^s}\,\dt\\
\leq -\frac{Q^2}{2\eta N^{\frac{s}d}{\rm Cap}_s(\omega)}\int_1^{1+\eta}\frac{1}{t^s}\,\dt\leq -\frac{Q^2}{2 N^{\frac{s}d}(1+\eta)^s{\rm Cap}_s(\omega)},
\end{multline*}
where we used Jensen in the first bound. In total we arrive at the estimate
$$F_s(\beta,M,\Omega_N',\rho_b')\leq F_s(\beta,N+Q,\Omega_N,1)-\frac{Q^2}{2(1+\eta)^s N^{\frac{s}d}{\rm Cap}_s(\omega)}+C\beta^{-1}(N\eta+1).$$
From the scaling property~\eqref{eq:scaling_N_T0_Jellium} we have
$$F_s(\beta,M,\Omega_N',\rho_b')=(\rho'_b)^{\frac{s}{d}}F_s\big((\rho'_b)^{\frac{s}{d}}\beta,M,\Omega_M,1\big)\geq F_s(\beta,M,\Omega_M,1)-C(1+\beta^{-1})(N\eta+1).$$
The last bound is because the energy and entropy are both of order $N$, by Lemma~\ref{eq:lower_bound_Jellium}, hence changing $\beta$ to $(\rho'_b)^{\frac{s}{d}}\beta$ and multiplying by $\rho_b'$ generates an error of order $N\eta+1$. We obtain~\eqref{eq:positive_charge_lower}.

The proof of~\eqref{eq:negative_charge_upper} is the same, but reversed. We start with the domain $\Omega_N$ and consider the measure $\tilde\nu_N(x)=N^{-1}\tilde\nu(x/N^{1/d})$ with, this time,
$$\tilde\nu(x)=\frac1\eta\int_{1-\eta}^1t^{-d}\nu(x/t)\,\dt$$
where $\eta$ is the same as in~\eqref{eq:eta_choice}. The measure $\tilde\nu_N$ induces a constant potential over the smaller set $\tilde\Omega_N:=(1-\eta)\Omega_N$. We then subtract to the background $\1_{\Omega_N}$ the part $Q\tilde\nu_N$ and split the remainder into a constant term $\tilde \rho_b\1_{\tilde\Omega_N}$ with
$\tilde\rho_b=1-\kappa QN^{-1}\eta^{\frac{s-d}{2}}-\tau(1-\eta)^{-d} N^{-1}$
and a remainder $\tilde m_N=\1_\Omega-Q\tilde\nu_N-\tilde\rho_b\1_{\tilde\Omega_N}$. Then the integer appearing in the statement is $M=\tilde\rho_b(1-\eta)^dN$ and the proof of~\eqref{eq:negative_charge_upper} goes along the same lines as for~\eqref{eq:positive_charge_lower}.
\end{proof}

\begin{corollary}[Limit with net charge]\label{cor:net_charge}
Let $\max(0,d-2)<s<d$ and $\rho_b=1$. Let $\omega$ be a bounded convex open set with a $C^{1,\alpha}$ boundary, $\alpha>0$, and $|\omega|=1$. Then we have for every $q\in\R\cup\{\pm\ii\}$ and $\beta\in(0,+\ii]$
$$ \lim_{\substack{N\to\ii\\ N^{-\frac{d+s}{2d}}Q\to q}}\frac{F_s(\beta,N+Q,N^{1/d}\omega)}{N}=f(s,\beta)+\frac{q^2}{2\,{\rm Cap}_s(\omega)}
$$
where we recall our convention that $F_s(\beta,N,\Omega)=E_s(N,\Omega)$ for $\beta=+\ii$.
\end{corollary}

\begin{proof}
The result follows immediately from Lemma~\ref{lem:neutral_Jellium} when $q=\pm\ii$ and from Lemmas~\ref{lem:net_charge1} and~\ref{lem:net_charge2} for finite $q$. We are using here that the limit $F_s(\beta,M,\Omega_M)/M$ exists for neutral systems (Theorem~\ref{thm:limit_C_Jellium}) but in the next section we show how the existence of this limit follows from Lemmas~\ref{lem:net_charge1} and~\ref{lem:net_charge2}, together with the grand-canonical thermodynamic limit in Theorem~\ref{thm:limit_GC_Jellium}.
\end{proof}

\subsubsection{Existence of the canonical thermodynamic limit and coincidence with $f^{\rm GC}(s,\beta)$}
Here we explain how to show the thermodynamic limit in the canonical case using the grand-canonical limit in Theorem~\ref{thm:limit_GC_Jellium}, and the previous net charge estimates. We need two preliminary lemmas.
Recall from~\eqref{eq:def_E_s_GC_Jellium} that the grand-canonical energy at $T=0$ is by definition the minimum of the canonical energy over the number of points. The following states that the same holds at positive temperature, up to a logarithmic error in the volume.

\begin{lemma}[Minimizing over net charges]\label{lem:GC_min_charge}
Let $0<s<d$ and $\rho_b=1$. Let $\omega$ be any bounded open set with $|\omega|=1$. For every $\beta\in(0,+\ii)$ and $\mu\in\R$, we have
\begin{multline}
 F^{\rm GC}(\beta,\mu,N^{1/d}\omega)\leq \min_{n\geq0}\big\{F(\beta,n,N^{1/d}\omega)-\mu\, n\big\}\\
 \leq F^{\rm GC}(\beta,\mu,N^{1/d}\omega)+\beta^{-1}C\log(1+N)
 \label{eq:GC_min_charge}
\end{multline}
when $\N\ni N\geq N_0=C(1+\beta^{-1}+|\mu|)^{\frac{d}{d-s}}$. Here $C$ only depend on $\omega$.
\end{lemma}

We work here at integer volume $\ell^d=N$ for simplicity, since we are mainly interested in the comparison with the canonical problem.

\begin{proof}
Let $\Omega_N=N^{1/d}\omega$. From Lemma~\ref{lem:neutral_Jellium} and the same inequality as in~\eqref{eq:estim_lower_simple_temp} we have
\begin{equation}
  F(\beta,n,\Omega_N)\geq c'_1\frac{(n-N)^2}{N^{\frac{s}d}}-C(1+\beta^{-1})n, \qquad\forall n\in\N,
 \label{eq:lower_bd_F}
\end{equation}
for a universal constant $C$. We have estimated here the entropy by $\beta^{-1}n\log(n/N)\geq -\beta^{-1}e^{-1}N$.
In the neutral case $n=N$ we obtain after placing all the particles uniformly in $\Omega_N$
$$F(\beta,N,\Omega_N)\leq -\frac1{2N}\iint_{\Omega_N^2}\frac{\dx\,\dy}{|x-y|^s}+\beta^{-1}\log\frac{N!}{N^N}\leq0.$$
Let us denote by $G(\beta,\mu,\Omega_N)$ the minimum in~\eqref{eq:GC_min_charge}. Then the previous bounds show
$$-C(1+\beta^{-1}+|\mu|)N\leq G(\beta,\mu,\Omega_N)\leq F_s(\beta,N,\Omega_N)-\mu N\leq -\mu N$$
with the minimum attained for an $n$ satisfying
\begin{equation}
 |n-N|\leq C(1+\beta^{-1}+|\mu|)N^{\frac{s+d}{2d}}.
 \label{eq:estim_mu_beta_min_G}
\end{equation}
The new constant $C$ only depends on $\omega$ through $c_1'$ in~\eqref{eq:lower_bd_F}. For $N\geq N_0=C(1+\beta^{-1}+|\mu|)^{\frac{d}{d-s}}$, we have
\begin{equation}
F(\beta,n,\Omega_N)-\mu n\geq c\frac{(n-N)^2}{N^{\frac{s}d}}+G(\beta,\mu,\Omega_N),\qquad \forall n\geq 2N,
 \label{eq:estim_G_F_s}
\end{equation}
with $c=c'_1/2$. We are now ready to prove~\eqref{eq:GC_min_charge}. The first inequality is obtained by using the canonical Gibbs state for $F(\beta,n,\Omega_N)$ as a trial state in the grand-canonical problem. For the second inequality, we split the partition function in~\eqref{def:Z_GC} into two sums and use $F_s(\beta,n,\Omega_N)-\mu n\geq G(\beta,\mu,\Omega_N)$ for $n\leq 2N$ and~\eqref{eq:estim_G_F_s} for $n>2N$:
\begin{align*}
\sum_{n\geq0}e^{-\beta (F(\beta,n,\Omega_N)-\mu n)}&= \sum_{0\leq n\leq 2N}e^{-\beta (F(\beta,n,\Omega_N)-\mu n)}+\sum_{n>2N}e^{-\beta (F(\beta,n,\Omega_N)-\mu n)}\\
&\leq e^{-\beta G(\beta,\mu,\Omega_N)}\left(1+2N+\sum_{m\in\N}e^{-c\beta\frac{m^2}{N^{\frac{s}{d}}}}\right)\\
&\leq e^{-\beta G(\beta,\mu,\Omega_N)}\left(1+2N+C\beta^{-\frac12}N^{\frac{s}{2d}}\right)
\end{align*}
where we have recognized a Riemann sum. For $N\geq N_0$ we have $\beta^{-\frac12}N^{\frac{s}{2d}}\leq CN$ and therefore we obtain~\eqref{eq:GC_min_charge} after taking the logarithm and multiplying by $\beta^{-1}$.
\end{proof}

Lemma~\ref{lem:GC_min_charge} implies that the minimum $G(\beta,\mu,\Omega_N)$ has the same thermodynamic limit as the grand-canonical problem:
\begin{equation}
\lim_{N\to\ii}\frac{G(\beta,\mu,N^{1/d}\omega)}{N}=f^{\rm GC}(s,\beta,\mu)=f^{\rm GC}(s,\beta)-\mu.
\label{eq:GC_min_charge_limit}
\end{equation}

Our second lemma states that for convex sets, the neutral free energy is monotone.

\begin{lemma}[Monotonicity of the canonical (free) energy]\label{lem:monotonicity}
Let $\omega$ be a bounded convex open set with $|\omega|=1$. Let $\rho_b=1$. Define $\Omega_N:=N^{1/d}\omega$ with $N\in\N$.
Then we have
\begin{equation}
 F_s(\beta,N+1,\Omega_{N+1})\leq F_s(\beta,N,\Omega_N),
 \label{eq:F_s_monotone}
\end{equation}
for all $N\geq1$, all $\beta\in(0,+\ii]$ and all $0<s<d$.
\end{lemma}

\begin{proof}
Assuming $0\in\omega$, we have $\Omega_N\subset\Omega_{N+1}$. We construct a trial state for $\Omega_{N+1}$ by placing $N$ points in $\Omega_N$ with the Gibbs probability for $\Omega_N$ (resp. minimizing positions for $\beta=+\ii$), and one independent point in $\Omega_{N+1}\setminus\Omega_N$ in the uniform distribution $\1_{\Omega_{N+1}\setminus\Omega_N}$. This completely screens the background in $\Omega_{N+1}\setminus\Omega_N$ so that no interaction with the rest of the system is left. The free energy of the additional point is $-\frac{1}{2}\iint_{(\Omega_{N+1}\setminus\Omega_N)^2}V_s(x-y)\dx\,\dy\leq0$. It has no entropy since its probability density is a characteristic function.
\end{proof}

With this and the net charge estimates from the previous section, we can show that the canonical free energy converges to $f^{\rm GC}(s,\beta)$.

\begin{corollary}[Equivalence of ensemble]\label{cor:equivalence_ensembles}
Let $\max(0,d-2)<s<d$ and $\rho_b=1$. Let $\omega$ be a bounded convex open set with a $C^{1,\alpha}$ boundary, $\alpha>0$, and $|\omega|=1$. Then we have
$$\lim_{N\to\ii}\frac{F_s(\beta,N,N^{1/d}\omega)}{N} =f^{\rm GC}(s,\beta),$$
for all $\max(0,d-2)<s<d$ and all $\beta\in(0,+\ii]$.
\end{corollary}

We recall that the same result is shown for $s=d-2$ in $d\geq3$ in Ref.~\onlinecite{LieNar-75,SarMer-76}. Together with Theorem~\ref{thm:limit_GC_Jellium}, this provides a new proof of Theorem~\ref{thm:limit_C_Jellium} for smooth convex domains and $s>0$.

\begin{proof}
Denote again $\Omega_N=N^{1/d}\omega$. We have $F_s(\beta,N,\Omega_N)\geq F^{\rm GC}_s(\beta,0,\Omega_N)$ (we use $\mu=0$ throughout) and therefore
$$\liminf_{N\to\ii}\frac{F_s(\beta,N,\Omega_N)}{N} \geq f^{\rm GC}(s,\beta).$$
We only have to prove the reverse inequality with a limsup.
Let $N_k\to\ii$ a sequence realizing the limsup for $F_s$ and $N'_k:=N_k-N_k^p$ with $\frac{d+s}{2d}\theta+1-\theta<p<1$, where we recall that $\theta=\frac{2}{2+d-s}$. Let then $Q_k=O(N_k^{\frac{s+d}{2d}})$ (recall~\eqref{eq:estim_mu_beta_min_G}) be a positive or negative charge so that
$$G(\beta,0,N'_k)=\min_{n}F_s(\beta,n,\Omega_{N'_k})=F_s(\beta,N_k'+Q_k,\Omega_{N'_k}).$$
We apply~\eqref{eq:negative_charge_lower} and~\eqref{eq:positive_charge_lower} to $N=N_k'$ and $Q=Q_k$. This provides some $M_k$ with $N_k'+Q_k\leq M_k\leq N_k'+CN_k^{1-\theta}Q_k^\theta$ such that
$$F_s(\beta,N_k'+Q_k,\Omega_{N'_k})\geq F_s(\beta,M_k,\Omega_{M_k})-C(1+\beta^{-1})N_k^{1-\theta}Q_k^\theta.$$
However, due to our choice of $p$, we have $M_k\leq N_k-N_k^p+CN_k^{1-\theta}Q_k^\theta\leq N_k$ for $k$ large enough. By Lemma~\ref{lem:monotonicity}, we thus have $F_s(\beta,M_k,\Omega_{M_k})\geq F_s(\beta,N_k,\Omega_{N_k})$ and obtain due to Lemma~\ref{lem:GC_min_charge}
$$C\beta^{-1}\log(N_k)+F_s^{\rm GC}(\beta,0,\Omega_{N'_k})\geq F_s(\beta,N_k,\Omega_{N_k})-C(1+\beta^{-1})N_k^{1-\theta}Q_k^\theta.$$
From the grand-canonical thermodynamic limit in Theorem~\ref{thm:limit_GC_Jellium} we deduce that
$$\limsup_{N\to\ii}\frac{F_s(\beta,N,\Omega_N)}{N}\leq f^{\rm GC}(s,\beta).$$
Hence the canonical problem converges to the grand-canonical (free) energy.
\end{proof}

\subsection{Proof of Theorem~\ref{thm:infinite_conf_Jellium} on Jellium equilibrium configurations}\label{sec:proof_thm_Jellium}
In this section we provide the proof of Theorem~\ref{thm:infinite_conf_Jellium}. As usual, we can assume after scaling that $\rho_b=1$. We split our proof into several steps.

\bigskip

\paragraph*{Step 1: Finding an infinite cluster.}
We first need to find a point $\tau\in\Omega$ around which there will be infinitely many $x_{j,\ell}$ in the limit $\ell\to\ii$. Unfortunately, we have not proved that there cannot be large holes in our system (as we did in the short range case in Lemma~\ref{lem:local_bound_GC_T0}), which would imply that any $ \tau$ will do. In any case we expect much more. The points must be ``equidistributed'', that is, the number of points in sufficiently large balls $B_R(\tau)$ centered at any point $\tau\in\Omega$ should always be of order $\rho_b|B_R|$. In the Coulomb case equidistribution was proved in the canonical case in Ref.~\onlinecite{RotSer-15,PetRot-18,ArmSer-21}. Here we give a short proof of the weaker statement that \emph{almost all} $\tau$ have a local density close to $\rho_b$, in some average sense. Our argument is valid for all $0<s<d$ and applies similarly at positive temperature.

We first fix a small $\delta>0$ so that $|\{z\in\omega,\ \rd(z,\partial\omega)\geq\delta\}|\geq1/2$. We will not look at what is happening close to the boundary of $\Omega=\ell\omega$ at distance $\delta\ell$ and thus define
$$\Omega_-:=\left\{\tau\in\Omega\ :\ \rd(\tau,\partial\Omega)\geq \delta \ell\right\}.$$
Let then $\eps>0$ and consider the set of the bad $\tau$'s in $\Omega_-$ which have an average discrepancy per unit volume larger than $\eps$, for balls with a radius of order $2^k$:
$$A_{\ell,\eps,k}:=\left\{\tau\in\Omega_-\ :\ \left|\frac1{2^k}\int_{2^k}^{2^{k+1}}q_{r,\tau}(X_\ell)\,\rd r\right|\geq\eps\right\}.$$
We recall that
$$q_{r,\tau}(X)=\frac{\# X\cap B_r(\tau)-\rho_b|B_r(\tau)\cap \Omega|}{|B_r|}.$$
From Jensen's inequality and~\eqref{eq:estim_local_charge}, we have
\begin{equation*}
|A_{\ell,\eps,k}|\leq \eps^{-2}\int_{\R^d}\left|\frac1{2^k}\int_{2^k}^{2^{k+1}}q_{r,\tau}(X_\ell)\,\dr\right|^2\rd\tau\leq \frac1{\eps^2 2^k}\int_{2^k}^{2^{k+1}}\int_{\R^d}q_{r,\tau}(X_\ell)^2\rd\tau\,\dr\leq \frac{C|\Omega_-|}{\eps^22^{k(d-s)}}.
\end{equation*}
We are using here that $N_\ell\sim|\Omega|$ by Lemma~\ref{lem:neutral_Jellium} and that $|\Omega|\leq 2|\Omega_-|$.
After summing over $k$ we find
$$\left|\bigcup_{2^k\geq M}A_{\ell,\eps,k}\right|\leq \frac{C|\Omega_-|}{\eps^2M^{d-s}}.$$
Picking for instance $M=C\eps^{-\frac3{d-s}}$ we obtain
\begin{equation*}
\left|\left\{\tau\in\Omega_-\ :\ \left|\frac1{2^k}\int_{2^k}^{2^{k+1}}q_{r,\tau}(X_\ell)\,\rd r\right|\leq\eps,\ \forall 2^k\geq C\eps^{-\frac3{d-s}}\right\}\right|\geq (1-\eps)|\Omega_-|.
\end{equation*}
In other words, most of the points at distance $\delta\ell$ from the boundary have a small average discrepancy. Let then $\tau_\ell\in\Omega_-$ be any such point. It satisfies
$$1-\eps\leq \frac1{2^k}\int_{2^k}^{2^{k+1}}\frac{\#X_\ell\cap B_r(\tau_\ell)}{|B_r|}\dr\leq 1+ \eps,\qquad \text{for all}\quad C\eps^{-\frac3{d-s}}\leq 2^k\leq \frac{\delta \ell}2$$
where we have added the upper bound on $2^k$ to make sure that $B_{2^{k+1}}(\tau)\subset \Omega$. Recall that $\rho_b=1$. For the following it is sufficient to know that there are infinitely many points and thus we remark that the previous bound implies
$$\frac{1-\eps}{2^d}|B_{R}|\leq \#X_\ell\cap B_{R}(\tau_\ell)\leq 2^d(1+ \eps)|B_{R}|,\qquad \text{for all}\quad C\eps^{-\frac3{d-s}}\leq R\leq \frac{\delta \ell}2.$$
After translating $\omega$ we can assume $\tau_\ell\to0$. After passing to the limit $\ell\to\ii$, we obtain, as claimed, that there are infinitely many points located at a finite distance of the origin. Choosing the labels so that $j\mapsto |x_{j,\ell}|$ is non-decreasing, we find $x_{j,\ell}\to x_j$ for all $j$ after extraction of a further subsequence. From the previous construction, the limiting configuration satisfies
\begin{equation}
\frac{1-\eps}{2^d}|B_{R}|\leq \#X\cap B_{R}\leq 2^d(1+ \eps)|B_{R}|,\qquad \text{for all $R\geq C\eps^{-\frac3{d-s}}$.}
\label{eq:estim_nb_points}
\end{equation}
In fact, the lower and upper densities defined in Remark~\ref{rmk:energy_per_unit_vol} satisfy
$$\frac{\underline\rho(X)}{1+\eps}\leq 1\leq\frac{\overline\rho(X)}{1-\eps}.$$
This step is valid for all $0<s<d$ and a similar argument works at positive temperature.

\bigskip

\paragraph*{Step 2: Bounds on the potential.}
Due to the optimality of $N_\ell$, we have like in Lemma~\ref{lem:local_bound_GC_T0}
\begin{equation}
\Phi^{(j_0)}_\ell:=\sum_{j\neq j_0}\frac1{|x_{j_0,\ell}-x_{j,\ell}|^s}-\int_{\Omega}\frac{\dy}{|x_{j_0,\ell}-y|^s}\leq \mu,\qquad \forall 1\leq j_0\leq N_\ell,
\label{eq:Phi_j_0}
\end{equation}
and
\begin{equation}
 \Phi_\ell(x):=\sum_{j}\frac1{|x-x_{j,\ell}|^s}-\int_{\Omega}\frac{\dy}{|x-y|^s}\geq \mu,\qquad \forall x\in\overline \Omega\setminus X_\ell.
 \label{eq:Phi_x}
\end{equation}
In addition, moving one $x_{j_0,\ell}$ to a point $x$ should increase the energy. If $x_{j_0,\ell}$ is in the interior of $\Omega$, we thus obtain
$$\nabla_x\left(\sum_{j\neq j_0}\frac1{|x-x_{j,\ell}|^s}-\int_{\Omega}\frac{\dy}{|x-y|^s}\right)_{|x=x_{j_0,\ell}}=0.$$
Depending on the value of $s$, we can express this condition as
\begin{equation}
\sum_{j\neq j_0}\frac{x_{j_0,\ell}-x_{j,\ell}}{|x_{j_0,\ell}-x_{j,\ell}|^{s+2}}=\begin{cases}
\dps \int_{\Omega}\frac{x_{j_0,\ell}-y}{|x_{j_0,\ell}-y|^{s+2}}\dy &\text{for $0<s<d-1$,}\\[0.4cm]
\dps -\frac{\ell^{d-1-s}}{s}\int_{\partial\omega}\frac{n_\omega(y)}{|\ell^{-1}x_{j_0,\ell}-y|^{s}}\dy &\text{for $d-1\leq s<d$,}
\end{cases}
\label{eq:nabla_Phi}
\end{equation}
where $n_\omega$ is the outward unit at the boundary. At fixed $j_0$, the right side converges to 0 when $d-1<s<d$ and to
$-s^{-1}\int_{\partial\omega}|y|^{-s}\,n_\omega(y)\dy$ when $s=d-1$.
When $0<s<d-1$, the right side can diverge.

Next we prove a lower bound on $\Phi^{(j_0)}_\ell$. Integrating~\eqref{eq:Phi_x} against the function $\chi_\eta=|B_\eta|^{-1}\1_{B_\eta}$ centered at $x_{j_0,\ell}$, we obtain for $\ell$ large enough so that $B_\eta(x_{j_0,\ell})\subset\Omega$
\begin{multline}
\Phi^{(j_0)}_{\ell}\geq \mu-\frac1{|B_\eta|}\int_{B_\eta}\frac{\dy}{|y|^s}-\sum_{j\neq j_0}\left(\frac1{|x|^s}-\chi_\eta\ast\frac1{|x|^s}\right)(x_{j_0,\ell}-x_{j,\ell})\\
+\int_{\Omega}\left(\frac1{|x|^s}-\chi_\eta\ast\frac1{|x|^s}\right)(x_{j_0,\ell}-y)\,\dy.
\label{eq:lower_bd_Phi_j_0}
\end{multline}
The function $|x|^{-s}-\chi_\eta\ast|x|^{-s}$ decays like $|x|^{-s-2}$ at infinity and has compact support for $s=d-2$, by Newton's theorem. It is thus integrable over $\R^d$ and the last term of ~\eqref{eq:lower_bd_Phi_j_0} is finite. To see that the series on the right of~\eqref{eq:lower_bd_Phi_j_0} is also bounded, we treat separately the points $x_{j,\ell}$ located at a finite distance to the boundary. There are $o(\ell^d)$ such points by Lemma~\ref{lem:nb_points_boundary} and they are all located at a distance $\sim\ell$ of the origin. Thus the corresponding sum is a $o(\ell^{d-s-2})=o(1)$ for $d-2<s<d$. It simply vanishes for $s=d-2$. For the other points well inside $\Omega$, we use the separation given by Lemma~\ref{lem:separation} and conclude from Lemma~\ref{lem:simple_estim_sum} that $\Phi^{(j_0)}_{\ell}$ is bounded from below. Due to~\eqref{eq:Phi_j_0}, we have thus shown that the interaction of any point $x_{j_0}$ to the rest of the system is bounded for $d-2\leq s<d$. After extracting a further subsequence, we can assume that $\Phi^{(j_0)}_{\ell}$ converges to a constant $\Phi^{(j_0)}$ for every $j_0$.

\bigskip

\paragraph*{Step 3: Passing to the limit for $\Phi_\ell$ when $s>d-2$.}
We first assume $d-2<s<d-1$. We write, for any fixed $j_0$ and using the relation~\eqref{eq:nabla_Phi}
\begin{align*}
\Phi_\ell(x)=&\Phi^{(j_0)}_{\ell}+\frac1{|x-x_{j_0,\ell}|^s}+\sum_{j\neq j_0}\!\!\left(\frac1{|x-x_{j,\ell}|^s}-\frac1{|x_{j_0,\ell}-x_{j,\ell}|^s}+s\frac{(x-x_{j_0,\ell})\cdot (x_{j_0,\ell}-x_{j,\ell})}{|x_{j_0,\ell}-x_{j,\ell}|^{s+2}}\right)\\
&-\int_{\Omega}\left(\frac1{|x-y|^s}-\frac1{|x_{j_0,\ell}-y|^s}+s\frac{(x-x_{j_0,\ell})\cdot (x_{j_0,\ell}-y)}{|x_{j_0,\ell}-y|^{s+2}}\right)\!\dy.
\end{align*}
By the same argument as before (using the finite distance between the points and the number of points close to the boundary), the series is convergent as soon as $x$ stays away from the $x_{j,\ell}$. The last integral converges to
$$\int_{\R^d}\left(\frac1{|x-y|^s}-\frac1{|x_{j_0}-y|^s}+s\frac{(x-x_{j_0})\cdot (x_{j_0}-y)}{|x_{j_0}-y|^{s+2}}\right)\dy=0.$$
That this integral vanishes can be verified by going to Fourier coordinates. The function in the parenthesis is integrable for $d-2<s<d-1$ and its Fourier transform is proportional to
$$\frac{e^{ik\cdot x_{j_0}}\big(e^{ik\cdot (x-x_{j_0})}-1-ik\cdot (x-x_{j_0}) \big)}{|k|^{d-s}}=O(|k|^{2+s-d})\underset{k\to0}\longrightarrow0.$$
We have thus proved that $\Phi_\ell(x)\to \Phi(x)$, the function defined in~\eqref{eq:def_Phi}. The limit holds in $L^1_{\rm loc}(\R^d)$ as well as uniformly locally like in~\eqref{eq:limit_Phi_local}.

When $d-1<s<d$, the argument is exactly the same except that we do not need to subtract the gradient term in~\eqref{eq:nabla_Phi} and instead obtain~\eqref{def:Phi_simpler}. We need to use that
$$\int_{\R^d}\left(\frac1{|x-y|^s}-\frac1{|x_{j_0}-y|^s}\right)\dy=0\qquad\text{for $d-1<s<d,$}$$
which can also be proved by passing to Fourier coordinates. After passing to the limit in~\eqref{eq:nabla_Phi} we also obtain
$$\sum_{j\neq j_0}\frac{x_{j_0}-x_{j}}{|x_{j_0}-x_{j}|^{s+2}}=0\qquad\text{for $d-1<s<d$}$$
so that the formula~\eqref{eq:Phi_x} remains valid.

When $s=d-1$, we subtract the gradient term as we did for $s<d-1$ and use this time that
\begin{multline*}
\lim_{\ell\to\ii}\int_{\Omega}\left(\frac1{|x-y|^{d-1}}-\frac1{|x_{j_0,\ell}-y|^{d-1}}\right)\!\dy=\lim_{\ell\to\ii}\ell\int_{\omega}\left(\frac1{|\frac{x}\ell-y|^{d-1}}-\frac1{|\frac{x_{j_0,\ell}}\ell-y|^{d-1}}\right)\!\dy\\
=(x-x_{j_0})\cdot\nabla_x \left.\int_{\omega}\frac{\dy}{|x-y|^{d-1}}\right|_{x=0}=(x-x_{j_0})\cdot\int_{\partial\omega}n_\omega(y)\frac{\dy}{|y|^{d-1}}.
\end{multline*}

Once we know that $\Phi_\ell$ converges, we obtain the DLR equilibrium equation~\eqref{eq:DLR_T0_Jellium} by passing to the limit in the corresponding condition at finite $\ell$.

\bigskip

\paragraph*{Step 4: The Coulomb case.}
Next we deal with the special case $s=d-2$ in dimensions $d\geq3$. We know that $\Phi_\ell\geq\mu$ and arguing as for~\eqref{eq:lower_bd_Phi_j_0}, we also know that $\int_{B_R(x_{j_0,\ell})}\Phi_\ell$ is bounded for any fixed $R$. Thus $\Phi_\ell$ is bounded in $L^1_{\rm loc}$. In addition, we have
\begin{equation}
 -\Delta\Phi_\ell=(d-2)|\bS^{d-1}|\left(\sum_{j}\delta_{x_{j,\ell}}-\rho_b\1_{\Omega}\right).
 \label{eq:Poisson_Phi_Jellium_ell}
\end{equation}
This proves that $\Phi_\ell$ is subharmonic on $\R^d\setminus X_\ell$. By the mean-value inequality (Ref.~\onlinecite[Chap.~9]{LieLos-01}), we conclude that $\Phi_\ell$ is locally bounded in $\R^d\setminus X_\ell$. In fact, the function $f_\ell=\Phi_\ell-\sum_{x_{j,\ell}\in \overline{B_R}}|x-x_{j,\ell}|^{2-d}$ is $C^\ii$ and satisfies $\Delta f_\ell=(d-2)|\bS^{d-1}|\rho_b\1_{B_R}$ on any fixed ball $B_R$. This is sufficient to ensure that $f_\ell$ converges uniformly after extraction of a subsequence. After passing to the limit in~\eqref{eq:Poisson_Phi_Jellium_ell} we find that $\Phi$ solves~\eqref{eq:Poisson_Phi_Jellium}, in the sense of distributions.

We still have to show that the limiting function $\Phi$ is uniquely determined from the infinite configuration $X=\{x_j,\ j\geq1\}$, up to a constant. The precise statement is the following

\begin{lemma}[Uniqueness of $\Phi$]\label{lem:Liouville}
Let $\rho_b>0$ and $X=\{x_j\}_{j\in\N}$ be an infinite configuration of points in $\R^d$, with
$$|x_j-x_k|\geq\delta>0,\quad\text{for all $j\neq k$;}\qquad \overline{\rho}(X)=\limsup_{R\to\ii}\frac{\#X\cap B_{R}}{R^d}>0.$$
We consider the solutions $\Phi$ to Poisson's equation
$$-\Delta \Phi=|\bS^{d-1}|(d-2)\left(\sum_j \delta_{x_j}-\rho_b\right)\qquad\text{in $\cD'(\R^d)$,}$$
satisfying
\begin{equation}
 \left|\Phi(x)- \sum_j\frac{\1_{B_{\delta/2}(x_j)}}{|x-x_j|^{d-2}}\right|\leq \begin{cases}
C(1+|x|^2)&\text{on $\R^d$,}\\
C&\text{on $\bigcup_{j}B_{\delta/2}(x_j)$.}
\end{cases}
 \label{eq:assump_growth_Liouville}
\end{equation}
When they exist, such solutions are unique up to a constant.
\end{lemma}

\begin{proof}[Proof of Lemma~\ref{lem:Liouville}]
Consider two solutions $\Phi_1$ and $\Phi_2$. Then $F:=\Phi_1-\Phi_2$ is harmonic and satisfies $|F(x)|\leq 2C(1+|x|^2)$ on $\R^d$ due to our assumption~\eqref{eq:assump_growth_Liouville}. By Liouville's theorem, $F$ is a harmonic polynomial of degree 2 or less, that is, $F(x)=x^TMx+v\cdot x+\alpha$, where $v\in\R^d$ and $M$ is a $d\times d$ symmetric matrix with $M_{ii}=0$. The boundedness close to each $x_j$ implies that $F(x_j)$ is bounded. Expanding $F(x_j+y)$, we then find $|y^TMx_j|\leq C$ for all $y\in B_{\delta/2}$. This shows that $Mx_j$ is bounded and thus the distance of $x_j$ to the space $\ker(M)$ is bounded. If $\ker(M)$ had dimension $k<d$, this would imply that there are $O(R^k)=o(R^d)$ points in any large ball of radius $R$, which contradicts our assumption on the $x_j$. Thus $M\equiv0$. The boundedness of $F(x_j)=v\cdot x_j+\alpha$ implies by a similar argument that $v=0$. We have thus proved that $\Phi_1=\Phi_2+\text{cnst}$.
\end{proof}

Our function $\Phi(x)$ was obtained in the limit from $\Phi_\ell(x)$ and we have already explained that it is bounded below by $\mu$ and bounded uniformly around each $x_j$ when we subtract $|x-x_j|^{2-d}$. It only remains to prove the polynomial bound in~\eqref{eq:assump_growth_Liouville}. To this end, we discard the background and regularize the charges to obtain
$$\Phi_\ell(x)\leq \sum_j \frac1{|x|^{d-2}}\ast\chi(x-x_{j,\ell})+\sum_j\left(\frac{1}{|x|^{d-2}}-\frac1{|x|^{d-2}}\ast\chi\right)(x-x_{j_\ell})$$
where $\chi\in C^\ii_c(B_1,\R_+)$ is radial with $\int\chi=1$. In the second sum the terms are all compactly supported by Newton's theorem. The difference with the sum on the left of~\eqref{eq:assump_growth_Liouville} is bounded. Thus we have proved that
$$ \left|\Phi_\ell(x)- \sum_j\frac{\1_{B_{\delta/2}(x_{j,\ell})}}{|x-x_j|^{d-2}}\right|\leq\sum_j \frac1{|x|^{d-2}}\ast\chi(x-x_{j,\ell})+C.$$
Next we pick a large radius $r>0$ such that $\omega\subset B_{r-1}$ and use that $\chi$ is bounded to obtain for $\ell$ large enough
$$\sum_j \frac1{|x|^{d-2}}\ast\chi(x-x_{j,\ell})\leq C\1_{B_{\ell r}}\ast\frac1{|x|^{d-2}}=C|x|^2\quad\text{on $\Omega$}.$$
This concludes our proof that $\Phi$ is uniquely determined, hence that of Theorem~\ref{thm:infinite_conf_Jellium}.
\qed

\begin{acknowledgments}
I had the opportunity to interact with many people when writing this paper. I would like to thank Laurent B\'etermin, Djalil Chafa\"i, David Dereudre, Peter Forrester, Subhro Ghosh, Michael Goldman, Jonas Lampart, Thomas Lebl\'e, Joel Lebowitz, Elliott Lieb, Satya Majumdar, Simona Rota Nodari, Nicolas Rougerie, Gregory Schehr, Robert Seiringer, Sylvia Serfaty and Daniel Ueltschi for useful discussions.

This project has received funding from the European Research Council (ERC) under the European Union's Horizon 2020 research and innovation programme (grant agreement MDFT No 725528).
\end{acknowledgments}

\section*{References}

\begin{thebibliography}{518}%
\makeatletter
\providecommand \@ifxundefined [1]{%
 \@ifx{#1\undefined}
}%
\providecommand \@ifnum [1]{%
 \ifnum #1\expandafter \@firstoftwo
 \else \expandafter \@secondoftwo
 \fi
}%
\providecommand \@ifx [1]{%
 \ifx #1\expandafter \@firstoftwo
 \else \expandafter \@secondoftwo
 \fi
}%
\providecommand \natexlab [1]{#1}%
\providecommand \enquote  [1]{``#1''}%
\providecommand \bibnamefont  [1]{#1}%
\providecommand \bibfnamefont [1]{#1}%
\providecommand \citenamefont [1]{#1}%
\providecommand \href@noop [0]{\@secondoftwo}%
\providecommand \href [0]{\begingroup \@sanitize@url \@href}%
\providecommand \@href[1]{\@@startlink{#1}\@@href}%
\providecommand \@@href[1]{\endgroup#1\@@endlink}%
\providecommand \@sanitize@url [0]{\catcode `\\12\catcode `\$12\catcode
  `\&12\catcode `\#12\catcode `\^12\catcode `\_12\catcode `\%12\relax}%
\providecommand \@@startlink[1]{}%
\providecommand \@@endlink[0]{}%
\providecommand \url  [0]{\begingroup\@sanitize@url \@url }%
\providecommand \@url [1]{\endgroup\@href {#1}{\urlprefix }}%
\providecommand \urlprefix  [0]{URL }%
\providecommand \Eprint [0]{\href }%
\providecommand \doibase [0]{http://dx.doi.org/}%
\providecommand \selectlanguage [0]{\@gobble}%
\providecommand \bibinfo  [0]{\@secondoftwo}%
\providecommand \bibfield  [0]{\@secondoftwo}%
\providecommand \translation [1]{[#1]}%
\providecommand \BibitemOpen [0]{}%
\providecommand \bibitemStop [0]{}%
\providecommand \bibitemNoStop [0]{.\EOS\space}%
\providecommand \EOS [0]{\spacefactor3000\relax}%
\providecommand \BibitemShut  [1]{\csname bibitem#1\endcsname}%
\let\auto@bib@innerbib\@empty
\bibitem [{\citenamefont {Abo-Shaeer}\ \emph {et~al.}(2001)\citenamefont
  {Abo-Shaeer}, \citenamefont {Raman}, \citenamefont {Vogels},\ and\
  \citenamefont {Ketterle}}]{Ketterle-01}%
  \BibitemOpen
  \bibfield  {author} {\bibinfo {author} {\bibnamefont {Abo-Shaeer},
  \bibfnamefont {J.~R.}}, \bibinfo {author} {\bibnamefont {Raman},
  \bibfnamefont {C.}}, \bibinfo {author} {\bibnamefont {Vogels}, \bibfnamefont
  {J.~M.}}, \ and\ \bibinfo {author} {\bibnamefont {Ketterle}, \bibfnamefont
  {W.}},\ }\bibfield  {title} {\enquote {\bibinfo {title} {Observation of
  vortex lattices in {B}ose-{E}instein condensates},}\ }\href {\doibase
  10.1126/science.1060182} {\bibfield  {journal} {\bibinfo  {journal}
  {Science}\ }\textbf {\bibinfo {volume} {292}},\ \bibinfo {pages} {476--479}
  (\bibinfo {year} {2001})}\BibitemShut {NoStop}%
\bibitem [{\citenamefont {Abrikosov}(1957)}]{Abrikosov-57}%
  \BibitemOpen
  \bibfield  {author} {\bibinfo {author} {\bibnamefont {Abrikosov},
  \bibfnamefont {A.}},\ }\bibfield  {title} {\enquote {\bibinfo {title} {The
  magnetic properties of superconducting alloys},}\ }\href {\doibase
  10.1016/0022-3697(57)90083-5} {\bibfield  {journal} {\bibinfo  {journal} {J.
  Phys. Chem. Solids}\ }\textbf {\bibinfo {volume} {2}},\ \bibinfo {pages}
  {199--208} (\bibinfo {year} {1957})}\BibitemShut {NoStop}%
\bibitem [{\citenamefont {Abul-Magd}(2006)}]{Abul-Magd-06}%
  \BibitemOpen
  \bibfield  {author} {\bibinfo {author} {\bibnamefont {Abul-Magd},
  \bibfnamefont {A.}},\ }\bibfield  {title} {\enquote {\bibinfo {title}
  {Modelling gap-size distribution of parked cars using random-matrix
  theory},}\ }\href {\doibase https://doi.org/10.1016/j.physa.2005.10.059}
  {\bibfield  {journal} {\bibinfo  {journal} {Physica A}\ }\textbf {\bibinfo
  {volume} {368}},\ \bibinfo {pages} {536 -- 540} (\bibinfo {year}
  {2006})}\BibitemShut {NoStop}%
\bibitem [{\citenamefont {Adhikari}, \citenamefont {Ghosh},\ and\ \citenamefont
  {Lebowitz}(2021)}]{AdhGhoLeb-21}%
  \BibitemOpen
  \bibfield  {author} {\bibinfo {author} {\bibnamefont {Adhikari},
  \bibfnamefont {K.}}, \bibinfo {author} {\bibnamefont {Ghosh}, \bibfnamefont
  {S.}}, \ and\ \bibinfo {author} {\bibnamefont {Lebowitz}, \bibfnamefont
  {J.~L.}},\ }\bibfield  {title} {\enquote {\bibinfo {title} {Fluctuation and
  entropy in spectrally constrained random fields},}\ }\href {\doibase
  10.1007/s00220-021-04150-7} {\bibfield  {journal} {\bibinfo  {journal} {Comm.
  Math. Phys.}\ }\textbf {\bibinfo {volume} {386}},\ \bibinfo {pages}
  {749--780} (\bibinfo {year} {2021})}\BibitemShut {NoStop}%
\bibitem [{\citenamefont {Aftalion}(2007)}]{Aftalion-07}%
  \BibitemOpen
  \bibfield  {author} {\bibinfo {author} {\bibnamefont {Aftalion},
  \bibfnamefont {A.}},\ }\bibfield  {title} {\enquote {\bibinfo {title} {Vortex
  patterns in {B}ose {E}instein condensates},}\ }in\ \href {\doibase
  10.1090/conm/446/08623} {\emph {\bibinfo {booktitle} {Perspectives in
  nonlinear partial differential equations}}},\ \bibinfo {series} {Contemp.
  Math.}, Vol.\ \bibinfo {volume} {446}\ (\bibinfo  {publisher} {Amer. Math.
  Soc.},\ \bibinfo {address} {Providence, RI},\ \bibinfo {year} {2007})\ pp.\
  \bibinfo {pages} {1--18}\BibitemShut {NoStop}%
\bibitem [{\citenamefont {Aftalion}\ and\ \citenamefont
  {Blanc}(2006)}]{AftBla-06}%
  \BibitemOpen
  \bibfield  {author} {\bibinfo {author} {\bibnamefont {Aftalion},
  \bibfnamefont {A.}}\ and\ \bibinfo {author} {\bibnamefont {Blanc},
  \bibfnamefont {X.}},\ }\bibfield  {title} {\enquote {\bibinfo {title} {Vortex
  lattices in rotating {B}ose--{E}instein condensates},}\ }\href {\doibase
  10.1137/050632889} {\bibfield  {journal} {\bibinfo  {journal} {SIAM J. Math.
  Anal.}\ }\textbf {\bibinfo {volume} {38}},\ \bibinfo {pages} {874--893}
  (\bibinfo {year} {2006})}\BibitemShut {NoStop}%
\bibitem [{\citenamefont {Agarwal}\ \emph {et~al.}(2019)\citenamefont
  {Agarwal}, \citenamefont {Dhar}, \citenamefont {Kulkarni}, \citenamefont
  {Kundu}, \citenamefont {Majumdar}, \citenamefont {Mukamel},\ and\
  \citenamefont {Schehr}}]{AgaDhaKulKunMajMukSch-19}%
  \BibitemOpen
  \bibfield  {author} {\bibinfo {author} {\bibnamefont {Agarwal}, \bibfnamefont
  {S.}}, \bibinfo {author} {\bibnamefont {Dhar}, \bibfnamefont {A.}}, \bibinfo
  {author} {\bibnamefont {Kulkarni}, \bibfnamefont {M.}}, \bibinfo {author}
  {\bibnamefont {Kundu}, \bibfnamefont {A.}}, \bibinfo {author} {\bibnamefont
  {Majumdar}, \bibfnamefont {S.~N.}}, \bibinfo {author} {\bibnamefont
  {Mukamel}, \bibfnamefont {D.}}, \ and\ \bibinfo {author} {\bibnamefont
  {Schehr}, \bibfnamefont {G.}},\ }\bibfield  {title} {\enquote {\bibinfo
  {title} {Harmonically confined particles with long-range repulsive
  interactions},}\ }\href {\doibase 10.1103/PhysRevLett.123.100603} {\bibfield
  {journal} {\bibinfo  {journal} {Phys. Rev. Lett.}\ }\textbf {\bibinfo
  {volume} {123}},\ \bibinfo {pages} {100603} (\bibinfo {year}
  {2019})}\BibitemShut {NoStop}%
\bibitem [{\citenamefont {Agarwal}, \citenamefont {Kulkarni},\ and\
  \citenamefont {Dhar}(2019)}]{AgaKulDha-19}%
  \BibitemOpen
  \bibfield  {author} {\bibinfo {author} {\bibnamefont {Agarwal}, \bibfnamefont
  {S.}}, \bibinfo {author} {\bibnamefont {Kulkarni}, \bibfnamefont {M.}}, \
  and\ \bibinfo {author} {\bibnamefont {Dhar}, \bibfnamefont {A.}},\ }\bibfield
   {title} {\enquote {\bibinfo {title} {Some connections between the classical
  {C}alogero-{M}oser model and the log-gas},}\ }\href {\doibase
  10.1007/s10955-019-02349-6} {\bibfield  {journal} {\bibinfo  {journal} {J.
  Stat. Phys.}\ }\textbf {\bibinfo {volume} {176}},\ \bibinfo {pages}
  {1463--1479} (\bibinfo {year} {2019})}\BibitemShut {NoStop}%
\bibitem [{\citenamefont {Agboola}\ \emph {et~al.}(2015)\citenamefont
  {Agboola}, \citenamefont {Knol}, \citenamefont {Gill},\ and\ \citenamefont
  {Loos}}]{AgbKnoGilLoo-15}%
  \BibitemOpen
  \bibfield  {author} {\bibinfo {author} {\bibnamefont {Agboola}, \bibfnamefont
  {D.}}, \bibinfo {author} {\bibnamefont {Knol}, \bibfnamefont {A.~L.}},
  \bibinfo {author} {\bibnamefont {Gill}, \bibfnamefont {P.~M.~W.}}, \ and\
  \bibinfo {author} {\bibnamefont {Loos}, \bibfnamefont {P.-F.}},\ }\bibfield
  {title} {\enquote {\bibinfo {title} {{Uniform electron gases. III.
  Low-density gases on three-dimensional spheres}},}\ }\href {\doibase
  10.1063/1.4929353} {\bibfield  {journal} {\bibinfo  {journal} {J. Chem.
  Phys.}\ }\textbf {\bibinfo {volume} {143}},\ \bibinfo {pages} {084114}
  (\bibinfo {year} {2015})}\BibitemShut {NoStop}%
\bibitem [{\citenamefont {Agrawal}\ and\ \citenamefont
  {Kofke}(1995{\natexlab{a}})}]{AgrKof-95}%
  \BibitemOpen
  \bibfield  {author} {\bibinfo {author} {\bibnamefont {Agrawal}, \bibfnamefont
  {R.}}\ and\ \bibinfo {author} {\bibnamefont {Kofke}, \bibfnamefont {D.~A.}},\
  }\bibfield  {title} {\enquote {\bibinfo {title} {Solid-fluid coexistence for
  inverse-power potentials},}\ }\href {\doibase 10.1103/PhysRevLett.74.122}
  {\bibfield  {journal} {\bibinfo  {journal} {Phys. Rev. Lett.}\ }\textbf
  {\bibinfo {volume} {74}},\ \bibinfo {pages} {122--125} (\bibinfo {year}
  {1995}{\natexlab{a}})}\BibitemShut {NoStop}%
\bibitem [{\citenamefont {Agrawal}\ and\ \citenamefont
  {Kofke}(1995{\natexlab{b}})}]{AgrKof-95b}%
  \BibitemOpen
  \bibfield  {author} {\bibinfo {author} {\bibnamefont {Agrawal}, \bibfnamefont
  {R.}}\ and\ \bibinfo {author} {\bibnamefont {Kofke}, \bibfnamefont {D.~A.}},\
  }\bibfield  {title} {\enquote {\bibinfo {title} {Thermodynamic and structural
  properties of model systems at solid-fluid coexistence},}\ }\href {\doibase
  10.1080/00268979500100911} {\bibfield  {journal} {\bibinfo  {journal} {Mol.
  Phys.}\ }\textbf {\bibinfo {volume} {85}},\ \bibinfo {pages} {23--42}
  (\bibinfo {year} {1995}{\natexlab{b}})}\BibitemShut {NoStop}%
\bibitem [{\citenamefont {Aizenman}, \citenamefont {Goldstein},\ and\
  \citenamefont {Lebowitz}(2001)}]{AizGolLeb-01}%
  \BibitemOpen
  \bibfield  {author} {\bibinfo {author} {\bibnamefont {Aizenman},
  \bibfnamefont {M.}}, \bibinfo {author} {\bibnamefont {Goldstein},
  \bibfnamefont {S.}}, \ and\ \bibinfo {author} {\bibnamefont {Lebowitz},
  \bibfnamefont {J.~L.}},\ }\bibfield  {title} {\enquote {\bibinfo {title}
  {Bounded fluctuations and translation symmetry breaking in one-dimensional
  particle systems},}\ }\href {\doibase 10.1023/A:1010397401128} {\bibfield
  {journal} {\bibinfo  {journal} {J. Statist. Phys.}\ }\textbf {\bibinfo
  {volume} {103}},\ \bibinfo {pages} {601--618} (\bibinfo {year} {2001})},\
  \bibinfo {note} {{S}pecial issue dedicated to the memory of {J}oaquin {M}.
  {L}uttinger}\BibitemShut {NoStop}%
\bibitem [{\citenamefont {Aizenman}, \citenamefont {Jansen},\ and\
  \citenamefont {Jung}(2010)}]{AizJanJun-10}%
  \BibitemOpen
  \bibfield  {author} {\bibinfo {author} {\bibnamefont {Aizenman},
  \bibfnamefont {M.}}, \bibinfo {author} {\bibnamefont {Jansen}, \bibfnamefont
  {S.}}, \ and\ \bibinfo {author} {\bibnamefont {Jung}, \bibfnamefont {P.}},\
  }\bibfield  {title} {\enquote {\bibinfo {title} {{Symmetry Breaking in
  Quasi-1D Coulomb Systems}},}\ }\href {\doibase 10.1007/s00023-010-0067-y}
  {\bibfield  {journal} {\bibinfo  {journal} {Ann. Henri Poincar{\'e}}\
  }\textbf {\bibinfo {volume} {11}},\ \bibinfo {pages} {1453--1485} (\bibinfo
  {year} {2010})}\BibitemShut {NoStop}%
\bibitem [{\citenamefont {Aizenman}\ and\ \citenamefont
  {Martin}(1980)}]{AizMar-80}%
  \BibitemOpen
  \bibfield  {author} {\bibinfo {author} {\bibnamefont {Aizenman},
  \bibfnamefont {M.}}\ and\ \bibinfo {author} {\bibnamefont {Martin},
  \bibfnamefont {P.~A.}},\ }\bibfield  {title} {\enquote {\bibinfo {title}
  {Structure of {G}ibbs states of one dimensional {C}oulomb systems},}\ }\href
  {\doibase 10.1007/BF01941972} {\bibfield  {journal} {\bibinfo  {journal}
  {Comm. Math. Phys.}\ }\textbf {\bibinfo {volume} {78}},\ \bibinfo {pages}
  {99--116} (\bibinfo {year} {1980})}\BibitemShut {NoStop}%
\bibitem [{\citenamefont {Aizenman}\ and\ \citenamefont
  {Warzel}(2015)}]{AizWar-15}%
  \BibitemOpen
  \bibfield  {author} {\bibinfo {author} {\bibnamefont {Aizenman},
  \bibfnamefont {M.}}\ and\ \bibinfo {author} {\bibnamefont {Warzel},
  \bibfnamefont {S.}},\ }\href {\doibase 10.1090/gsm/168} {\emph {\bibinfo
  {title} {Random operators}}},\ \bibinfo {series} {Graduate Studies in
  Mathematics}, Vol.\ \bibinfo {volume} {168}\ (\bibinfo  {publisher} {American
  Mathematical Society, Providence, RI},\ \bibinfo {year} {2015})\ pp.\
  \bibinfo {pages} {xiv+326}\BibitemShut {NoStop}%
\bibitem [{\citenamefont {Alastuey}(1986)}]{Alastuey-86}%
  \BibitemOpen
  \bibfield  {author} {\bibinfo {author} {\bibnamefont {Alastuey},
  \bibfnamefont {A.}},\ }\bibfield  {title} {\enquote {\bibinfo {title}
  {Propri{\'e}t{\'e}s d'{\'e}quilibre du plasma classique {\`a} une composante
  en trois et deux dimensions},}\ }\href {\doibase
  10.1051/anphys:01986001106065300} {\bibfield  {journal} {\bibinfo  {journal}
  {Ann. Phys. Fr.}\ }\textbf {\bibinfo {volume} {11}},\ \bibinfo {pages}
  {653--738} (\bibinfo {year} {1986})}\BibitemShut {NoStop}%
\bibitem [{\citenamefont {Alastuey}\ and\ \citenamefont
  {Jancovici}(1981{\natexlab{a}})}]{AlaJan-81b}%
  \BibitemOpen
  \bibfield  {author} {\bibinfo {author} {\bibnamefont {Alastuey},
  \bibfnamefont {A.}}\ and\ \bibinfo {author} {\bibnamefont {Jancovici},
  \bibfnamefont {B.}},\ }\bibfield  {title} {\enquote {\bibinfo {title}
  {Absence of strict crystalline order in a two-dimensional electron system},}\
  }\href {\doibase 10.1007/BF01012815} {\bibfield  {journal} {\bibinfo
  {journal} {J. Stat. Phys.}\ }\textbf {\bibinfo {volume} {24}},\ \bibinfo
  {pages} {443--449} (\bibinfo {year} {1981}{\natexlab{a}})}\BibitemShut
  {NoStop}%
\bibitem [{\citenamefont {Alastuey}\ and\ \citenamefont
  {Jancovici}(1981{\natexlab{b}})}]{AlaJan-81}%
  \BibitemOpen
  \bibfield  {author} {\bibinfo {author} {\bibnamefont {Alastuey},
  \bibfnamefont {A.}}\ and\ \bibinfo {author} {\bibnamefont {Jancovici},
  \bibfnamefont {B.}},\ }\bibfield  {title} {\enquote {\bibinfo {title} {On the
  classical two-dimensional one-component {C}oulomb plasma},}\ }\href {\doibase
  10.1051/jphys:019810042010100} {\bibfield  {journal} {\bibinfo  {journal} {J.
  Phys. France}\ }\textbf {\bibinfo {volume} {42}},\ \bibinfo {pages} {1--12}
  (\bibinfo {year} {1981}{\natexlab{b}})}\BibitemShut {NoStop}%
\bibitem [{\citenamefont {Alastuey}\ and\ \citenamefont
  {Jancovici}(1984)}]{AlaJan-84}%
  \BibitemOpen
  \bibfield  {author} {\bibinfo {author} {\bibnamefont {Alastuey},
  \bibfnamefont {A.}}\ and\ \bibinfo {author} {\bibnamefont {Jancovici},
  \bibfnamefont {B.}},\ }\bibfield  {title} {\enquote {\bibinfo {title} {On
  potential and field fluctuations in two-dimensional classical charged
  systems},}\ }\href {https://doi.org/10.1007/BF01018558} {\bibfield  {journal}
  {\bibinfo  {journal} {J. Statist. Phys.}\ }\textbf {\bibinfo {volume} {34}},\
  \bibinfo {pages} {557--569} (\bibinfo {year} {1984})}\BibitemShut {NoStop}%
\bibitem [{\citenamefont {Alastuey}\ and\ \citenamefont
  {Martin}(1985)}]{AlaMar-85}%
  \BibitemOpen
  \bibfield  {author} {\bibinfo {author} {\bibnamefont {Alastuey},
  \bibfnamefont {A.}}\ and\ \bibinfo {author} {\bibnamefont {Martin},
  \bibfnamefont {P.~A.}},\ }\bibfield  {title} {\enquote {\bibinfo {title}
  {Decay of correlations in classical fluids with long-range forces},}\ }\href
  {https://doi.org/10.1007/BF01018670} {\bibfield  {journal} {\bibinfo
  {journal} {J. Statist. Phys.}\ }\textbf {\bibinfo {volume} {39}},\ \bibinfo
  {pages} {405--426} (\bibinfo {year} {1985})}\BibitemShut {NoStop}%
\bibitem [{\citenamefont {Albeverio}\ and\ \citenamefont
  {H{\o}egh-Krohn}(1975)}]{AlbHog-75}%
  \BibitemOpen
  \bibfield  {author} {\bibinfo {author} {\bibnamefont {Albeverio},
  \bibfnamefont {S.}}\ and\ \bibinfo {author} {\bibnamefont {H{\o}egh-Krohn},
  \bibfnamefont {R.}},\ }\bibfield  {title} {\enquote {\bibinfo {title}
  {Homogeneous random fields and statistical mechanics},}\ }\href
  {https://doi.org/10.1016/0022-1236(75)90058-0} {\bibfield  {journal}
  {\bibinfo  {journal} {J. Funct. Anal.}\ }\textbf {\bibinfo {volume} {19}},\
  \bibinfo {pages} {242--272} (\bibinfo {year} {1975})}\BibitemShut {NoStop}%
\bibitem [{\citenamefont {Albeverio}, \citenamefont {Pastur},\ and\
  \citenamefont {Shcherbina}(2001)}]{AlbPasShc-01}%
  \BibitemOpen
  \bibfield  {author} {\bibinfo {author} {\bibnamefont {Albeverio},
  \bibfnamefont {S.}}, \bibinfo {author} {\bibnamefont {Pastur}, \bibfnamefont
  {L.}}, \ and\ \bibinfo {author} {\bibnamefont {Shcherbina}, \bibfnamefont
  {M.}},\ }\bibfield  {title} {\enquote {\bibinfo {title} {On the {$1/n$}
  expansion for some unitary invariant ensembles of random matrices},}\ }\href
  {\doibase 10.1007/s002200100531} {\bibfield  {journal} {\bibinfo  {journal}
  {Comm. Math. Phys.}\ }\textbf {\bibinfo {volume} {224}},\ \bibinfo {pages}
  {271--305} (\bibinfo {year} {2001})},\ \bibinfo {note} {dedicated to Joel L.
  Lebowitz}\BibitemShut {NoStop}%
\bibitem [{\citenamefont {Ameur}(2018)}]{Ameur-18}%
  \BibitemOpen
  \bibfield  {author} {\bibinfo {author} {\bibnamefont {Ameur}, \bibfnamefont
  {Y.}},\ }\bibfield  {title} {\enquote {\bibinfo {title} {Repulsion in low
  temperature {$\beta$}-ensembles},}\ }\href {\doibase
  10.1007/s00220-017-3027-2} {\bibfield  {journal} {\bibinfo  {journal} {Comm.
  Math. Phys.}\ }\textbf {\bibinfo {volume} {359}},\ \bibinfo {pages}
  {1079--1089} (\bibinfo {year} {2018})}\BibitemShut {NoStop}%
\bibitem [{\citenamefont {Ameur}(2021)}]{Ameur-21}%
  \BibitemOpen
  \bibfield  {author} {\bibinfo {author} {\bibnamefont {Ameur}, \bibfnamefont
  {Y.}},\ }\bibfield  {title} {\enquote {\bibinfo {title} {A localization
  theorem for the planar {C}oulomb gas in an external field},}\ }\href
  {\doibase 10.1214/21-ejp613} {\bibfield  {journal} {\bibinfo  {journal}
  {Electron. J. Probab.}\ }\textbf {\bibinfo {volume} {26}},\ \bibinfo {pages}
  {Paper No. 46, 21} (\bibinfo {year} {2021})}\BibitemShut {NoStop}%
\bibitem [{\citenamefont {Ameur}\ and\ \citenamefont
  {Ortega-Cerd\`a}(2012)}]{AmeOrt-12}%
  \BibitemOpen
  \bibfield  {author} {\bibinfo {author} {\bibnamefont {Ameur}, \bibfnamefont
  {Y.}}\ and\ \bibinfo {author} {\bibnamefont {Ortega-Cerd\`a}, \bibfnamefont
  {J.}},\ }\bibfield  {title} {\enquote {\bibinfo {title} {Beurling-{L}andau
  densities of weighted {F}ekete sets and correlation kernel estimates},}\
  }\href {\doibase 10.1016/j.jfa.2012.06.011} {\bibfield  {journal} {\bibinfo
  {journal} {J. Funct. Anal.}\ }\textbf {\bibinfo {volume} {263}},\ \bibinfo
  {pages} {1825--1861} (\bibinfo {year} {2012})}\BibitemShut {NoStop}%
\bibitem [{\citenamefont {Ameur}\ and\ \citenamefont
  {Romero}(2022)}]{AmeRom-22_ppt}%
  \BibitemOpen
  \bibfield  {author} {\bibinfo {author} {\bibnamefont {Ameur}, \bibfnamefont
  {Y.}}\ and\ \bibinfo {author} {\bibnamefont {Romero}, \bibfnamefont
  {J.~L.}},\ }\href@noop {} {\enquote {\bibinfo {title} {The planar low
  temperature {C}oulomb gas: separation and equidistribution},}\ } (\bibinfo
  {year} {2022}),\ \Eprint {http://arxiv.org/abs/2010.10179} {arXiv:2010.10179
  [math.PR]} \BibitemShut {NoStop}%
\bibitem [{\citenamefont {Anderson}, \citenamefont {Guionnet},\ and\
  \citenamefont {Zeitouni}(2010)}]{AndGuiZei-10}%
  \BibitemOpen
  \bibfield  {author} {\bibinfo {author} {\bibnamefont {Anderson},
  \bibfnamefont {G.~W.}}, \bibinfo {author} {\bibnamefont {Guionnet},
  \bibfnamefont {A.}}, \ and\ \bibinfo {author} {\bibnamefont {Zeitouni},
  \bibfnamefont {O.}},\ }\href@noop {} {\emph {\bibinfo {title} {An
  introduction to random matrices}}},\ \bibinfo {series} {Cambridge Studies in
  Advanced Mathematics}, Vol.\ \bibinfo {volume} {118}\ (\bibinfo  {publisher}
  {Cambridge University Press, Cambridge},\ \bibinfo {year} {2010})\ pp.\
  \bibinfo {pages} {xiv+492}\BibitemShut {NoStop}%
\bibitem [{\citenamefont {Angelescu}\ and\ \citenamefont
  {Nenciu}(1973)}]{AngNen-73}%
  \BibitemOpen
  \bibfield  {author} {\bibinfo {author} {\bibnamefont {Angelescu},
  \bibfnamefont {N.}}\ and\ \bibinfo {author} {\bibnamefont {Nenciu},
  \bibfnamefont {G.}},\ }\bibfield  {title} {\enquote {\bibinfo {title} {On the
  independence of the thermodynamic limit on the boundary conditions in quantum
  statistical mechanics},}\ }\href
  {http://projecteuclid.org/euclid.cmp/1103858475} {\bibfield  {journal}
  {\bibinfo  {journal} {Comm. Math. Phys.}\ }\textbf {\bibinfo {volume} {29}},\
  \bibinfo {pages} {15--30} (\bibinfo {year} {1973})}\BibitemShut {NoStop}%
\bibitem [{\citenamefont {Armstrong}\ and\ \citenamefont
  {Serfaty}(2021)}]{ArmSer-21}%
  \BibitemOpen
  \bibfield  {author} {\bibinfo {author} {\bibnamefont {Armstrong},
  \bibfnamefont {S.}}\ and\ \bibinfo {author} {\bibnamefont {Serfaty},
  \bibfnamefont {S.}},\ }\bibfield  {title} {\enquote {\bibinfo {title} {Local
  laws and rigidity for {C}oulomb gases at any temperature},}\ }\href {\doibase
  10.1214/20-AOP1445} {\bibfield  {journal} {\bibinfo  {journal} {Ann.
  Probab.}\ }\textbf {\bibinfo {volume} {49}},\ \bibinfo {pages} {46--121}
  (\bibinfo {year} {2021})}\BibitemShut {NoStop}%
\bibitem [{\citenamefont {Arovas}, \citenamefont {Schrieffer},\ and\
  \citenamefont {Wilczek}(1984)}]{AroSchWil-84}%
  \BibitemOpen
  \bibfield  {author} {\bibinfo {author} {\bibnamefont {Arovas}, \bibfnamefont
  {D.}}, \bibinfo {author} {\bibnamefont {Schrieffer}, \bibfnamefont {J.~R.}},
  \ and\ \bibinfo {author} {\bibnamefont {Wilczek}, \bibfnamefont {F.}},\
  }\bibfield  {title} {\enquote {\bibinfo {title} {{Fractional Statistics and
  the Quantum Hall Effect}},}\ }\href {\doibase 10.1103/PhysRevLett.53.722}
  {\bibfield  {journal} {\bibinfo  {journal} {Phys. Rev. Lett.}\ }\textbf
  {\bibinfo {volume} {53}},\ \bibinfo {pages} {722--723} (\bibinfo {year}
  {1984})}\BibitemShut {NoStop}%
\bibitem [{\citenamefont {Astrakharchik}\ \emph {et~al.}(2006)\citenamefont
  {Astrakharchik}, \citenamefont {Gangardt}, \citenamefont {Lozovik},\ and\
  \citenamefont {Sorokin}}]{AstGanLozSor-06}%
  \BibitemOpen
  \bibfield  {author} {\bibinfo {author} {\bibnamefont {Astrakharchik},
  \bibfnamefont {G.~E.}}, \bibinfo {author} {\bibnamefont {Gangardt},
  \bibfnamefont {D.~M.}}, \bibinfo {author} {\bibnamefont {Lozovik},
  \bibfnamefont {Y.~E.}}, \ and\ \bibinfo {author} {\bibnamefont {Sorokin},
  \bibfnamefont {I.~A.}},\ }\bibfield  {title} {\enquote {\bibinfo {title}
  {Off-diagonal correlations of the {C}alogero-{S}utherland model},}\ }\href
  {\doibase 10.1103/PhysRevE.74.021105} {\bibfield  {journal} {\bibinfo
  {journal} {Phys. Rev. E}\ }\textbf {\bibinfo {volume} {74}},\ \bibinfo
  {pages} {021105} (\bibinfo {year} {2006})}\BibitemShut {NoStop}%
\bibitem [{\citenamefont {Azadi}\ and\ \citenamefont
  {Drummond}(2022)}]{AzaDru-22_ppt}%
  \BibitemOpen
  \bibfield  {author} {\bibinfo {author} {\bibnamefont {Azadi}, \bibfnamefont
  {S.}}\ and\ \bibinfo {author} {\bibnamefont {Drummond}, \bibfnamefont
  {N.~D.}},\ }\href@noop {} {\enquote {\bibinfo {title} {Low-density phase
  diagram of the three-dimensional electron gas},}\ } (\bibinfo {year}
  {2022}),\ \Eprint {http://arxiv.org/abs/2201.08743} {arXiv:2201.08743
  [cond-mat.str-el]} \BibitemShut {NoStop}%
\bibitem [{\citenamefont {Bach}, \citenamefont {Lieb},\ and\ \citenamefont
  {Solovej}(1994)}]{BacLieSol-94}%
  \BibitemOpen
  \bibfield  {author} {\bibinfo {author} {\bibnamefont {Bach}, \bibfnamefont
  {V.}}, \bibinfo {author} {\bibnamefont {Lieb}, \bibfnamefont {E.~H.}}, \ and\
  \bibinfo {author} {\bibnamefont {Solovej}, \bibfnamefont {J.~P.}},\
  }\bibfield  {title} {\enquote {\bibinfo {title} {Generalized {H}artree-{F}ock
  theory and the {H}ubbard model},}\ }\href {\doibase 10.1007/BF02188656}
  {\bibfield  {journal} {\bibinfo  {journal} {J. Statist. Phys.}\ }\textbf
  {\bibinfo {volume} {76}},\ \bibinfo {pages} {3--89} (\bibinfo {year}
  {1994})}\BibitemShut {NoStop}%
\bibitem [{\citenamefont {Bagchi}, \citenamefont {Andersen},\ and\
  \citenamefont {Swope}(1996)}]{BagAndSwo-96}%
  \BibitemOpen
  \bibfield  {author} {\bibinfo {author} {\bibnamefont {Bagchi}, \bibfnamefont
  {K.}}, \bibinfo {author} {\bibnamefont {Andersen}, \bibfnamefont {H.~C.}}, \
  and\ \bibinfo {author} {\bibnamefont {Swope}, \bibfnamefont {W.}},\
  }\bibfield  {title} {\enquote {\bibinfo {title} {Computer simulation study of
  the melting transition in two dimensions},}\ }\href {\doibase
  10.1103/PhysRevLett.76.255} {\bibfield  {journal} {\bibinfo  {journal} {Phys.
  Rev. Lett.}\ }\textbf {\bibinfo {volume} {76}},\ \bibinfo {pages} {255--258}
  (\bibinfo {year} {1996})}\BibitemShut {NoStop}%
\bibitem [{\citenamefont {Baik}\ \emph {et~al.}(2006)\citenamefont {Baik},
  \citenamefont {Borodin}, \citenamefont {Deift},\ and\ \citenamefont
  {Suidan}}]{BaiBorDeiSui-06}%
  \BibitemOpen
  \bibfield  {author} {\bibinfo {author} {\bibnamefont {Baik}, \bibfnamefont
  {J.}}, \bibinfo {author} {\bibnamefont {Borodin}, \bibfnamefont {A.}},
  \bibinfo {author} {\bibnamefont {Deift}, \bibfnamefont {P.}}, \ and\ \bibinfo
  {author} {\bibnamefont {Suidan}, \bibfnamefont {T.}},\ }\bibfield  {title}
  {\enquote {\bibinfo {title} {A model for the bus system in {C}uernavaca
  ({M}exico)},}\ }\href {\doibase 10.1088/0305-4470/39/28/S11} {\bibfield
  {journal} {\bibinfo  {journal} {J. Phys. A}\ }\textbf {\bibinfo {volume}
  {39}},\ \bibinfo {pages} {8965--8975} (\bibinfo {year} {2006})}\BibitemShut
  {NoStop}%
\bibitem [{\citenamefont {Barnes}\ and\ \citenamefont
  {Sloane}(1983)}]{BarSlo-83}%
  \BibitemOpen
  \bibfield  {author} {\bibinfo {author} {\bibnamefont {Barnes}, \bibfnamefont
  {E.~S.}}\ and\ \bibinfo {author} {\bibnamefont {Sloane}, \bibfnamefont
  {N.~J.~A.}},\ }\bibfield  {title} {\enquote {\bibinfo {title} {The optimal
  lattice quantizer in three dimensions},}\ }\href {\doibase 10.1137/0604005}
  {\bibfield  {journal} {\bibinfo  {journal} {SIAM J. Algebraic Discrete
  Methods}\ }\textbf {\bibinfo {volume} {4}},\ \bibinfo {pages} {30--41}
  (\bibinfo {year} {1983})}\BibitemShut {NoStop}%
\bibitem [{\citenamefont {Bauerschmidt}\ \emph {et~al.}(2017)\citenamefont
  {Bauerschmidt}, \citenamefont {Bourgade}, \citenamefont {Nikula},\ and\
  \citenamefont {Yau}}]{BauBouNikYau-17}%
  \BibitemOpen
  \bibfield  {author} {\bibinfo {author} {\bibnamefont {Bauerschmidt},
  \bibfnamefont {R.}}, \bibinfo {author} {\bibnamefont {Bourgade},
  \bibfnamefont {P.}}, \bibinfo {author} {\bibnamefont {Nikula}, \bibfnamefont
  {M.}}, \ and\ \bibinfo {author} {\bibnamefont {Yau}, \bibfnamefont {H.-T.}},\
  }\bibfield  {title} {\enquote {\bibinfo {title} {Local density for
  two-dimensional one-component plasma},}\ }\href {\doibase
  10.1007/s00220-017-2932-8} {\bibfield  {journal} {\bibinfo  {journal} {Comm.
  Math. Phys.}\ }\textbf {\bibinfo {volume} {356}},\ \bibinfo {pages}
  {189--230} (\bibinfo {year} {2017})}\BibitemShut {NoStop}%
\bibitem [{\citenamefont {Bauerschmidt}\ \emph {et~al.}(2019)\citenamefont
  {Bauerschmidt}, \citenamefont {Bourgade}, \citenamefont {Nikula},\ and\
  \citenamefont {Yau}}]{BauBouNikYau-19}%
  \BibitemOpen
  \bibfield  {author} {\bibinfo {author} {\bibnamefont {Bauerschmidt},
  \bibfnamefont {R.}}, \bibinfo {author} {\bibnamefont {Bourgade},
  \bibfnamefont {P.}}, \bibinfo {author} {\bibnamefont {Nikula}, \bibfnamefont
  {M.}}, \ and\ \bibinfo {author} {\bibnamefont {Yau}, \bibfnamefont {H.-T.}},\
  }\bibfield  {title} {\enquote {\bibinfo {title} {The two-dimensional
  {C}oulomb plasma: quasi-free approximation and central limit theorem},}\
  }\href {\doibase 10.4310/ATMP.2019.v23.n4.a1} {\bibfield  {journal} {\bibinfo
   {journal} {Adv. Theor. Math. Phys.}\ }\textbf {\bibinfo {volume} {23}},\
  \bibinfo {pages} {841--1002} (\bibinfo {year} {2019})}\BibitemShut {NoStop}%
\bibitem [{\citenamefont {Baus}(1980)}]{Baus-80}%
  \BibitemOpen
  \bibfield  {author} {\bibinfo {author} {\bibnamefont {Baus}, \bibfnamefont
  {M.}},\ }\bibfield  {title} {\enquote {\bibinfo {title} {Absence of
  long-range order with long-range potentials},}\ }\href {\doibase
  10.1007/BF01007993} {\bibfield  {journal} {\bibinfo  {journal} {J. Stat.
  Phys.}\ }\textbf {\bibinfo {volume} {22}},\ \bibinfo {pages} {111--119}
  (\bibinfo {year} {1980})}\BibitemShut {NoStop}%
\bibitem [{\citenamefont {Baus}\ and\ \citenamefont
  {Hansen}(1980)}]{BauHan-80}%
  \BibitemOpen
  \bibfield  {author} {\bibinfo {author} {\bibnamefont {Baus}, \bibfnamefont
  {M.}}\ and\ \bibinfo {author} {\bibnamefont {Hansen}, \bibfnamefont
  {J.-P.}},\ }\bibfield  {title} {\enquote {\bibinfo {title} {Statistical
  mechanics of simple {C}oulomb systems},}\ }\href {\doibase
  10.1016/0370-1573(80)90022-8} {\bibfield  {journal} {\bibinfo  {journal}
  {Phys. Rep.}\ }\textbf {\bibinfo {volume} {59}},\ \bibinfo {pages} {1--94}
  (\bibinfo {year} {1980})}\BibitemShut {NoStop}%
\bibitem [{\citenamefont {Baxter}(1963)}]{Baxter-63}%
  \BibitemOpen
  \bibfield  {author} {\bibinfo {author} {\bibnamefont {Baxter}, \bibfnamefont
  {R.~J.}},\ }\bibfield  {title} {\enquote {\bibinfo {title} {Statistical
  mechanics of a one-dimensional {C}oulomb system with a uniform charge
  background},}\ }\href {\doibase 10.1017/S0305004100003790} {\bibfield
  {journal} {\bibinfo  {journal} {Proc. Cambridge Philos. Soc.}\ }\textbf
  {\bibinfo {volume} {59}},\ \bibinfo {pages} {779--787} (\bibinfo {year}
  {1963})}\BibitemShut {NoStop}%
\bibitem [{\citenamefont {Becke}(1993)}]{Becke-93}%
  \BibitemOpen
  \bibfield  {author} {\bibinfo {author} {\bibnamefont {Becke}, \bibfnamefont
  {A.~D.}},\ }\bibfield  {title} {\enquote {\bibinfo {title}
  {Density-functional thermochemistry. {III}. {T}he role of exact exchange},}\
  }\href {\doibase 10.1063/1.464913} {\bibfield  {journal} {\bibinfo  {journal}
  {J. Chem. Phys.}\ }\textbf {\bibinfo {volume} {98}},\ \bibinfo {pages}
  {5648--5652} (\bibinfo {year} {1993})}\BibitemShut {NoStop}%
\bibitem [{\citenamefont {Beltr{\'a}n}(2013)}]{Beltran-13}%
  \BibitemOpen
  \bibfield  {author} {\bibinfo {author} {\bibnamefont {Beltr{\'a}n},
  \bibfnamefont {C.}},\ }\bibfield  {title} {\enquote {\bibinfo {title} {The
  state of the art in {S}male's 7th problem},}\ }in\ \href@noop {} {\emph
  {\bibinfo {booktitle} {Foundations of computational mathematics, {B}udapest
  2011}}},\ \bibinfo {series} {London Math. Soc. Lecture Note Ser.}, Vol.\
  \bibinfo {volume} {403}\ (\bibinfo  {publisher} {Cambridge Univ. Press,
  Cambridge},\ \bibinfo {year} {2013})\ pp.\ \bibinfo {pages}
  {1--15}\BibitemShut {NoStop}%
\bibitem [{\citenamefont {Benfatto}, \citenamefont {Gruber},\ and\
  \citenamefont {Martin}(1984)}]{BenGruMar-84}%
  \BibitemOpen
  \bibfield  {author} {\bibinfo {author} {\bibnamefont {Benfatto},
  \bibfnamefont {G.}}, \bibinfo {author} {\bibnamefont {Gruber}, \bibfnamefont
  {C.}}, \ and\ \bibinfo {author} {\bibnamefont {Martin}, \bibfnamefont
  {P.~A.}},\ }\bibfield  {title} {\enquote {\bibinfo {title} {Exact decay of
  correlations for infinite range continuous systems},}\ }\href@noop {}
  {\bibfield  {journal} {\bibinfo  {journal} {Helv. Phys. Acta}\ }\textbf
  {\bibinfo {volume} {57}},\ \bibinfo {pages} {63--85} (\bibinfo {year}
  {1984})}\BibitemShut {NoStop}%
\bibitem [{\citenamefont {Bergersen}, \citenamefont {Boal},\ and\ \citenamefont
  {Palffy-Muhoray}(1994)}]{BerBoaPal-94}%
  \BibitemOpen
  \bibfield  {author} {\bibinfo {author} {\bibnamefont {Bergersen},
  \bibfnamefont {B.}}, \bibinfo {author} {\bibnamefont {Boal}, \bibfnamefont
  {D.}}, \ and\ \bibinfo {author} {\bibnamefont {Palffy-Muhoray}, \bibfnamefont
  {P.}},\ }\bibfield  {title} {\enquote {\bibinfo {title} {Equilibrium
  configurations of particles on a sphere: the case of logarithmic
  interactions},}\ }\href {\doibase 10.1088/0305-4470/27/7/032} {\bibfield
  {journal} {\bibinfo  {journal} {J. Phys. A - Math. Gen}\ }\textbf {\bibinfo
  {volume} {27}},\ \bibinfo {pages} {2579--2586} (\bibinfo {year}
  {1994})}\BibitemShut {NoStop}%
\bibitem [{\citenamefont {Berman}, \citenamefont {Boucksom},\ and\
  \citenamefont {Nystr{\"o}m}(2011)}]{BerBouWitt-11}%
  \BibitemOpen
  \bibfield  {author} {\bibinfo {author} {\bibnamefont {Berman}, \bibfnamefont
  {R.}}, \bibinfo {author} {\bibnamefont {Boucksom}, \bibfnamefont {S.}}, \
  and\ \bibinfo {author} {\bibnamefont {Nystr{\"o}m}, \bibfnamefont {D.~W.}},\
  }\bibfield  {title} {\enquote {\bibinfo {title} {{Fekete points and
  convergence towards equilibrium measures on complex manifolds}},}\ }\href
  {\doibase 10.1007/s11511-011-0067-x} {\bibfield  {journal} {\bibinfo
  {journal} {Acta Math.}\ }\textbf {\bibinfo {volume} {207}},\ \bibinfo {pages}
  {1 -- 27} (\bibinfo {year} {2011})}\BibitemShut {NoStop}%
\bibitem [{\citenamefont {Berman}(2017)}]{Berman-17}%
  \BibitemOpen
  \bibfield  {author} {\bibinfo {author} {\bibnamefont {Berman}, \bibfnamefont
  {R.~J.}},\ }\bibfield  {title} {\enquote {\bibinfo {title} {Large deviations
  for {G}ibbs measures with singular {H}amiltonians and emergence of
  {K}\"{a}hler-{E}instein metrics},}\ }\href {\doibase
  10.1007/s00220-017-2926-6} {\bibfield  {journal} {\bibinfo  {journal} {Comm.
  Math. Phys.}\ }\textbf {\bibinfo {volume} {354}},\ \bibinfo {pages}
  {1133--1172} (\bibinfo {year} {2017})}\BibitemShut {NoStop}%
\bibitem [{\citenamefont {Berman}(2019)}]{Berman-19}%
  \BibitemOpen
  \bibfield  {author} {\bibinfo {author} {\bibnamefont {Berman}, \bibfnamefont
  {R.~J.}},\ }\bibfield  {title} {\enquote {\bibinfo {title} {Statistical
  mechanics of interpolation nodes, pluripotential theory and complex
  geometry},}\ }\href {\doibase 10.4064/ap180925-4-7} {\bibfield  {journal}
  {\bibinfo  {journal} {Ann. Polon. Math.}\ }\textbf {\bibinfo {volume}
  {123}},\ \bibinfo {pages} {71--153} (\bibinfo {year} {2019})}\BibitemShut
  {NoStop}%
\bibitem [{\citenamefont {{Berman}}(2020)}]{Berman-20}%
  \BibitemOpen
  \bibfield  {author} {\bibinfo {author} {\bibnamefont {{Berman}},
  \bibfnamefont {R.~J.}},\ }\bibfield  {title} {\enquote {\bibinfo {title} {{An
  invitation to K\"ahler-Einstein metrics and random point processes}},}\ }in\
  \href {\doibase 10.4310/SDG.2018.v23.n1.a2} {\emph {\bibinfo {booktitle}
  {Differential geometry, Calabi-Yau theory, and general relativity. Lectures
  given at conferences celebrating the 70th birthday of Shing-Tung Yau at
  Harvard University, Cambridge, MA, USA, May 2019}}}\ (\bibinfo  {publisher}
  {Somerville, MA: International Press},\ \bibinfo {year} {2020})\ pp.\
  \bibinfo {pages} {35--87}\BibitemShut {NoStop}%
\bibitem [{\citenamefont {Bernard}\ and\ \citenamefont {Wu}(1994)}]{BerWu-94}%
  \BibitemOpen
  \bibfield  {author} {\bibinfo {author} {\bibnamefont {Bernard}, \bibfnamefont
  {D.}}\ and\ \bibinfo {author} {\bibnamefont {Wu}, \bibfnamefont {Y.-S.}},\
  }\bibfield  {title} {\enquote {\bibinfo {title} {A note on statistical
  interactions and the thermodynamic {B}ethe ansatz},}\ }in\ \href@noop {}
  {\emph {\bibinfo {booktitle} {New developments in integrable systems and
  long-range interaction models}}},\ \bibinfo {series and number} {Nankai
  Lecture Notes on Mathematical Physics}\ (\bibinfo  {publisher} {World
  Scientific},\ \bibinfo {year} {1994})\BibitemShut {NoStop}%
\bibitem [{\citenamefont {Berry}(1985)}]{Berry-85}%
  \BibitemOpen
  \bibfield  {author} {\bibinfo {author} {\bibnamefont {Berry}, \bibfnamefont
  {M.~V.}},\ }\bibfield  {title} {\enquote {\bibinfo {title} {Semiclassical
  theory of spectral rigidity},}\ }\href {\doibase 10.1098/rspa.1985.0078}
  {\bibfield  {journal} {\bibinfo  {journal} {Proc. Roy. Soc. London Ser. A}\
  }\textbf {\bibinfo {volume} {400}},\ \bibinfo {pages} {229--251} (\bibinfo
  {year} {1985})}\BibitemShut {NoStop}%
\bibitem [{\citenamefont {B\'{e}termin}(2019)}]{Betermin-19}%
  \BibitemOpen
  \bibfield  {author} {\bibinfo {author} {\bibnamefont {B\'{e}termin},
  \bibfnamefont {L.}},\ }\bibfield  {title} {\enquote {\bibinfo {title} {Local
  optimality of cubic lattices for interaction energies},}\ }\href {\doibase
  10.1007/s13324-017-0205-5} {\bibfield  {journal} {\bibinfo  {journal} {Anal.
  Math. Phys.}\ }\textbf {\bibinfo {volume} {9}},\ \bibinfo {pages} {403--426}
  (\bibinfo {year} {2019})}\BibitemShut {NoStop}%
\bibitem [{\citenamefont {B\'{e}termin}\ and\ \citenamefont
  {Sandier}(2018)}]{BetSan-18}%
  \BibitemOpen
  \bibfield  {author} {\bibinfo {author} {\bibnamefont {B\'{e}termin},
  \bibfnamefont {L.}}\ and\ \bibinfo {author} {\bibnamefont {Sandier},
  \bibfnamefont {E.}},\ }\bibfield  {title} {\enquote {\bibinfo {title}
  {Renormalized energy and asymptotic expansion of optimal logarithmic energy
  on the sphere},}\ }\href {\doibase 10.1007/s00365-016-9357-z} {\bibfield
  {journal} {\bibinfo  {journal} {Constr. Approx.}\ }\textbf {\bibinfo {volume}
  {47}},\ \bibinfo {pages} {39--74} (\bibinfo {year} {2018})}\BibitemShut
  {NoStop}%
\bibitem [{\citenamefont {B{\'e}termin}, \citenamefont {\v{S}amaj},\ and\
  \citenamefont {Trav\v{e}nec}(2021)}]{BetSamTra-21_ppt}%
  \BibitemOpen
  \bibfield  {author} {\bibinfo {author} {\bibnamefont {B{\'e}termin},
  \bibfnamefont {L.}}, \bibinfo {author} {\bibnamefont {\v{S}amaj},
  \bibfnamefont {L.}}, \ and\ \bibinfo {author} {\bibnamefont {Trav\v{e}nec},
  \bibfnamefont {I.}},\ }\bibfield  {title} {\enquote {\bibinfo {title}
  {Three-dimensional lattice ground states for {R}iesz and {L}ennard--{J}ones
  type energies},}\ }\href {https://arxiv.org/abs/2107.14020} {\bibfield
  {journal} {\bibinfo  {journal} {ArXiV e-prints}\ } (\bibinfo {year}
  {2021})}\BibitemShut {NoStop}%
\bibitem [{\citenamefont {Bethuel}, \citenamefont {Brezis},\ and\ \citenamefont
  {H{\'e}lein}(1994)}]{BetBreHel-94}%
  \BibitemOpen
  \bibfield  {author} {\bibinfo {author} {\bibnamefont {Bethuel}, \bibfnamefont
  {F.}}, \bibinfo {author} {\bibnamefont {Brezis}, \bibfnamefont {H.}}, \ and\
  \bibinfo {author} {\bibnamefont {H{\'e}lein}, \bibfnamefont {F.}},\ }\href
  {\doibase 10.1007/978-1-4612-0287-5} {\emph {\bibinfo {title}
  {Ginzburg-{L}andau vortices}}},\ Progress in Nonlinear Differential Equations
  and their Applications, 13\ (\bibinfo  {publisher} {Birkh{\"a}user Boston,
  Inc., Boston, MA},\ \bibinfo {year} {1994})\ pp.\ \bibinfo {pages}
  {xxviii+159}\BibitemShut {NoStop}%
\bibitem [{\citenamefont {Bhaduri}, \citenamefont {Murthy},\ and\ \citenamefont
  {Sen}(2010)}]{BhaMurSen-10}%
  \BibitemOpen
  \bibfield  {author} {\bibinfo {author} {\bibnamefont {Bhaduri}, \bibfnamefont
  {R.~K.}}, \bibinfo {author} {\bibnamefont {Murthy}, \bibfnamefont
  {M.~V.~N.}}, \ and\ \bibinfo {author} {\bibnamefont {Sen}, \bibfnamefont
  {D.}},\ }\bibfield  {title} {\enquote {\bibinfo {title} {The virial expansion
  of a classical interacting system},}\ }\href {\doibase
  10.1088/1751-8113/43/4/045002} {\bibfield  {journal} {\bibinfo  {journal} {J.
  Phys. A}\ }\textbf {\bibinfo {volume} {43}},\ \bibinfo {pages} {045002, 8}
  (\bibinfo {year} {2010})}\BibitemShut {NoStop}%
\bibitem [{\citenamefont {Blanc}, \citenamefont {Bris},\ and\ \citenamefont
  {Lions}(2002)}]{BlaBriLio-02}%
  \BibitemOpen
  \bibfield  {author} {\bibinfo {author} {\bibnamefont {Blanc}, \bibfnamefont
  {X.}}, \bibinfo {author} {\bibnamefont {Bris}, \bibfnamefont {C.~L.}}, \ and\
  \bibinfo {author} {\bibnamefont {Lions}, \bibfnamefont {P.-L.}},\ }\bibfield
  {title} {\enquote {\bibinfo {title} {Caract{\'e}risation des fonctions de
  $\mathbb{R}^3$ {\`a} potentiel newtonien born{\'e}},}\ }\href {\doibase
  10.1016/S1631-073X(02)02203-3} {\bibfield  {journal} {\bibinfo  {journal} {C.
  R. Math. Acad. Sci. Paris}\ }\textbf {\bibinfo {volume} {334}},\ \bibinfo
  {pages} {15--21} (\bibinfo {year} {2002})}\BibitemShut {NoStop}%
\bibitem [{\citenamefont {Blanc}\ and\ \citenamefont
  {Lewin}(2015)}]{BlaLew-15}%
  \BibitemOpen
  \bibfield  {author} {\bibinfo {author} {\bibnamefont {Blanc}, \bibfnamefont
  {X.}}\ and\ \bibinfo {author} {\bibnamefont {Lewin}, \bibfnamefont {M.}},\
  }\bibfield  {title} {\enquote {\bibinfo {title} {The crystallization
  conjecture: A review},}\ }\href {\doibase 10.4171/EMSS/13} {\bibfield
  {journal} {\bibinfo  {journal} {EMS Surv. Math. Sci.}\ }\textbf {\bibinfo
  {volume} {2}},\ \bibinfo {pages} {255--306} (\bibinfo {year} {2015})},\
  \Eprint {http://arxiv.org/abs/1504.01153} {1504.01153} \BibitemShut {NoStop}%
\bibitem [{\citenamefont {Blum}\ \emph {et~al.}(1982)\citenamefont {Blum},
  \citenamefont {Gruber}, \citenamefont {Lebowitz},\ and\ \citenamefont
  {Martin}}]{BluGruLebMar-82}%
  \BibitemOpen
  \bibfield  {author} {\bibinfo {author} {\bibnamefont {Blum}, \bibfnamefont
  {L.}}, \bibinfo {author} {\bibnamefont {Gruber}, \bibfnamefont {C.}},
  \bibinfo {author} {\bibnamefont {Lebowitz}, \bibfnamefont {J.~L.}}, \ and\
  \bibinfo {author} {\bibnamefont {Martin}, \bibfnamefont {P.}},\ }\bibfield
  {title} {\enquote {\bibinfo {title} {Perfect screening for charged
  systems},}\ }\href {\doibase 10.1103/PhysRevLett.48.1769} {\bibfield
  {journal} {\bibinfo  {journal} {Phys. Rev. Lett.}\ }\textbf {\bibinfo
  {volume} {48}},\ \bibinfo {pages} {1769--1772} (\bibinfo {year}
  {1982})}\BibitemShut {NoStop}%
\bibitem [{\citenamefont {Bogomolny}, \citenamefont {Giraud},\ and\
  \citenamefont {Schmit}(2009)}]{BogGirSch-09}%
  \BibitemOpen
  \bibfield  {author} {\bibinfo {author} {\bibnamefont {Bogomolny},
  \bibfnamefont {E.}}, \bibinfo {author} {\bibnamefont {Giraud}, \bibfnamefont
  {O.}}, \ and\ \bibinfo {author} {\bibnamefont {Schmit}, \bibfnamefont {C.}},\
  }\bibfield  {title} {\enquote {\bibinfo {title} {{Random Matrix Ensembles
  Associated with Lax Matrices}},}\ }\href {\doibase
  10.1103/PhysRevLett.103.054103} {\bibfield  {journal} {\bibinfo  {journal}
  {Phys. Rev. Lett.}\ }\textbf {\bibinfo {volume} {103}},\ \bibinfo {pages}
  {054103} (\bibinfo {year} {2009})}\BibitemShut {NoStop}%
\bibitem [{\citenamefont {Bohigas}, \citenamefont {Giannoni},\ and\
  \citenamefont {Schmit}(1984)}]{BohGiaSch-84}%
  \BibitemOpen
  \bibfield  {author} {\bibinfo {author} {\bibnamefont {Bohigas}, \bibfnamefont
  {O.}}, \bibinfo {author} {\bibnamefont {Giannoni}, \bibfnamefont {M.~J.}}, \
  and\ \bibinfo {author} {\bibnamefont {Schmit}, \bibfnamefont {C.}},\
  }\bibfield  {title} {\enquote {\bibinfo {title} {Characterization of chaotic
  quantum spectra and universality of level fluctuation laws},}\ }\href
  {\doibase 10.1103/PhysRevLett.52.1} {\bibfield  {journal} {\bibinfo
  {journal} {Phys. Rev. Lett.}\ }\textbf {\bibinfo {volume} {52}},\ \bibinfo
  {pages} {1--4} (\bibinfo {year} {1984})}\BibitemShut {NoStop}%
\bibitem [{\citenamefont {Bonsall}\ and\ \citenamefont
  {Maradudin}(1977)}]{BonMar-77}%
  \BibitemOpen
  \bibfield  {author} {\bibinfo {author} {\bibnamefont {Bonsall}, \bibfnamefont
  {L.}}\ and\ \bibinfo {author} {\bibnamefont {Maradudin}, \bibfnamefont
  {A.~A.}},\ }\bibfield  {title} {\enquote {\bibinfo {title} {Some static and
  dynamical properties of a two-dimensional {W}igner crystal},}\ }\href
  {\doibase 10.1103/PhysRevB.15.1959} {\bibfield  {journal} {\bibinfo
  {journal} {Phys. Rev. B}\ }\textbf {\bibinfo {volume} {15}},\ \bibinfo
  {pages} {1959--1973} (\bibinfo {year} {1977})}\BibitemShut {NoStop}%
\bibitem [{\citenamefont {Borodachov}, \citenamefont {Hardin},\ and\
  \citenamefont {Saff}(2008)}]{BorHarSaf-08}%
  \BibitemOpen
  \bibfield  {author} {\bibinfo {author} {\bibnamefont {Borodachov},
  \bibfnamefont {S.~V.}}, \bibinfo {author} {\bibnamefont {Hardin},
  \bibfnamefont {D.~P.}}, \ and\ \bibinfo {author} {\bibnamefont {Saff},
  \bibfnamefont {E.~B.}},\ }\bibfield  {title} {\enquote {\bibinfo {title}
  {Asymptotics for discrete weighted minimal {R}iesz energy problems on
  rectifiable sets},}\ }\href {\doibase 10.1090/S0002-9947-07-04416-9}
  {\bibfield  {journal} {\bibinfo  {journal} {Trans. Amer. Math. Soc.}\
  }\textbf {\bibinfo {volume} {360}},\ \bibinfo {pages} {1559--1580} (\bibinfo
  {year} {2008})}\BibitemShut {NoStop}%
\bibitem [{\citenamefont {Borodachov}, \citenamefont {Hardin},\ and\
  \citenamefont {Saff}(2019)}]{BorHarSaf-19}%
  \BibitemOpen
  \bibfield  {author} {\bibinfo {author} {\bibnamefont {Borodachov},
  \bibfnamefont {S.~V.}}, \bibinfo {author} {\bibnamefont {Hardin},
  \bibfnamefont {D.~P.}}, \ and\ \bibinfo {author} {\bibnamefont {Saff},
  \bibfnamefont {E.~B.}},\ }\href {\doibase 10.1007/978-0-387-84808-2} {\emph
  {\bibinfo {title} {Discrete energy on rectifiable sets}}},\ Springer
  Monographs in Mathematics\ (\bibinfo  {publisher} {Springer, New York},\
  \bibinfo {year} {2019})\ pp.\ \bibinfo {pages} {xviii+666}\BibitemShut
  {NoStop}%
\bibitem [{\citenamefont {Borodin}\ and\ \citenamefont
  {Serfaty}(2013)}]{BorSer-13}%
  \BibitemOpen
  \bibfield  {author} {\bibinfo {author} {\bibnamefont {Borodin}, \bibfnamefont
  {A.}}\ and\ \bibinfo {author} {\bibnamefont {Serfaty}, \bibfnamefont {S.}},\
  }\bibfield  {title} {\enquote {\bibinfo {title} {Renormalized energy
  concentration in random matrices},}\ }\href
  {https://doi.org/10.1007/s00220-013-1716-z} {\bibfield  {journal} {\bibinfo
  {journal} {Comm. Math. Phys.}\ }\textbf {\bibinfo {volume} {320}},\ \bibinfo
  {pages} {199--244} (\bibinfo {year} {2013})}\BibitemShut {NoStop}%
\bibitem [{\citenamefont {Borot}\ and\ \citenamefont
  {Guionnet}(2013)}]{BorGui-13}%
  \BibitemOpen
  \bibfield  {author} {\bibinfo {author} {\bibnamefont {Borot}, \bibfnamefont
  {G.}}\ and\ \bibinfo {author} {\bibnamefont {Guionnet}, \bibfnamefont {A.}},\
  }\bibfield  {title} {\enquote {\bibinfo {title} {Asymptotic expansion of
  {$\beta$} matrix models in the one-cut regime},}\ }\href {\doibase
  10.1007/s00220-012-1619-4} {\bibfield  {journal} {\bibinfo  {journal} {Comm.
  Math. Phys.}\ }\textbf {\bibinfo {volume} {317}},\ \bibinfo {pages}
  {447--483} (\bibinfo {year} {2013})}\BibitemShut {NoStop}%
\bibitem [{\citenamefont {Borot}, \citenamefont {Guionnet},\ and\ \citenamefont
  {Kozlowski}(2015)}]{BorGuiKoz-15}%
  \BibitemOpen
  \bibfield  {author} {\bibinfo {author} {\bibnamefont {Borot}, \bibfnamefont
  {G.}}, \bibinfo {author} {\bibnamefont {Guionnet}, \bibfnamefont {A.}}, \
  and\ \bibinfo {author} {\bibnamefont {Kozlowski}, \bibfnamefont {K.~K.}},\
  }\bibfield  {title} {\enquote {\bibinfo {title} {Large-{$N$} asymptotic
  expansion for mean field models with {C}oulomb gas interaction},}\ }\href
  {\doibase 10.1093/imrn/rnu260} {\bibfield  {journal} {\bibinfo  {journal}
  {Int. Math. Res. Not. IMRN}\ ,\ \bibinfo {pages} {10451--10524}} (\bibinfo
  {year} {2015})}\BibitemShut {NoStop}%
\bibitem [{\citenamefont {Borwein}, \citenamefont {Borwein},\ and\
  \citenamefont {Shail}(1989)}]{BorBorSha-89}%
  \BibitemOpen
  \bibfield  {author} {\bibinfo {author} {\bibnamefont {Borwein}, \bibfnamefont
  {D.}}, \bibinfo {author} {\bibnamefont {Borwein}, \bibfnamefont {J.~M.}}, \
  and\ \bibinfo {author} {\bibnamefont {Shail}, \bibfnamefont {R.}},\
  }\bibfield  {title} {\enquote {\bibinfo {title} {Analysis of certain lattice
  sums},}\ }\href {\doibase 10.1016/0022-247X(89)90032-2} {\bibfield  {journal}
  {\bibinfo  {journal} {J. Math. Anal. Appl.}\ }\textbf {\bibinfo {volume}
  {143}},\ \bibinfo {pages} {126--137} (\bibinfo {year} {1989})}\BibitemShut
  {NoStop}%
\bibitem [{\citenamefont {Borwein}\ \emph {et~al.}(1988)\citenamefont
  {Borwein}, \citenamefont {Borwein}, \citenamefont {Shail},\ and\
  \citenamefont {Zucker}}]{BorBorShaZuc-88}%
  \BibitemOpen
  \bibfield  {author} {\bibinfo {author} {\bibnamefont {Borwein}, \bibfnamefont
  {D.}}, \bibinfo {author} {\bibnamefont {Borwein}, \bibfnamefont {J.~M.}},
  \bibinfo {author} {\bibnamefont {Shail}, \bibfnamefont {R.}}, \ and\ \bibinfo
  {author} {\bibnamefont {Zucker}, \bibfnamefont {I.~J.}},\ }\bibfield  {title}
  {\enquote {\bibinfo {title} {Energy of static electron lattices},}\ }\href
  {http://stacks.iop.org/0305-4470/21/1519} {\bibfield  {journal} {\bibinfo
  {journal} {J. Phys. A}\ }\textbf {\bibinfo {volume} {21}},\ \bibinfo {pages}
  {1519--1531} (\bibinfo {year} {1988})}\BibitemShut {NoStop}%
\bibitem [{\citenamefont {Borwein}, \citenamefont {Borwein},\ and\
  \citenamefont {Straub}(2014)}]{BorBorStr-14}%
  \BibitemOpen
  \bibfield  {author} {\bibinfo {author} {\bibnamefont {Borwein}, \bibfnamefont
  {D.}}, \bibinfo {author} {\bibnamefont {Borwein}, \bibfnamefont {J.~M.}}, \
  and\ \bibinfo {author} {\bibnamefont {Straub}, \bibfnamefont {A.}},\
  }\bibfield  {title} {\enquote {\bibinfo {title} {On lattice sums and {W}igner
  limits},}\ }\href {\doibase 10.1016/j.jmaa.2014.01.008} {\bibfield  {journal}
  {\bibinfo  {journal} {J. Math. Anal. Appl.}\ }\textbf {\bibinfo {volume}
  {414}},\ \bibinfo {pages} {489--513} (\bibinfo {year} {2014})}\BibitemShut
  {NoStop}%
\bibitem [{\citenamefont {Borwein}, \citenamefont {Borwein},\ and\
  \citenamefont {Taylor}(1985)}]{BorBorTay-85}%
  \BibitemOpen
  \bibfield  {author} {\bibinfo {author} {\bibnamefont {Borwein}, \bibfnamefont
  {D.}}, \bibinfo {author} {\bibnamefont {Borwein}, \bibfnamefont {J.~M.}}, \
  and\ \bibinfo {author} {\bibnamefont {Taylor}, \bibfnamefont {K.~F.}},\
  }\bibfield  {title} {\enquote {\bibinfo {title} {Convergence of lattice sums
  and {M}adelung's constant},}\ }\href {\doibase 10.1063/1.526675} {\bibfield
  {journal} {\bibinfo  {journal} {J. Math. Phys.}\ }\textbf {\bibinfo {volume}
  {26}},\ \bibinfo {pages} {2999--3009} (\bibinfo {year} {1985})}\BibitemShut
  {NoStop}%
\bibitem [{\citenamefont {Borwein}\ \emph {et~al.}(2013)\citenamefont
  {Borwein}, \citenamefont {Glasser}, \citenamefont {McPhedran}, \citenamefont
  {Wan},\ and\ \citenamefont {Zucker}}]{BorGlaMPh-13}%
  \BibitemOpen
  \bibfield  {author} {\bibinfo {author} {\bibnamefont {Borwein}, \bibfnamefont
  {J.~M.}}, \bibinfo {author} {\bibnamefont {Glasser}, \bibfnamefont {M.~L.}},
  \bibinfo {author} {\bibnamefont {McPhedran}, \bibfnamefont {R.~C.}}, \bibinfo
  {author} {\bibnamefont {Wan}, \bibfnamefont {J.~G.}}, \ and\ \bibinfo
  {author} {\bibnamefont {Zucker}, \bibfnamefont {I.~J.}},\ }\href {\doibase
  10.1017/CBO9781139626804} {\emph {\bibinfo {title} {Lattice sums then and
  now}}},\ \bibinfo {series} {Encyclopedia of Mathematics and its
  Applications}, Vol.\ \bibinfo {volume} {150}\ (\bibinfo  {publisher}
  {Cambridge University Press, Cambridge},\ \bibinfo {year} {2013})\ pp.\
  \bibinfo {pages} {xx+368},\ \bibinfo {note} {with a foreword by Helaman
  Ferguson and Claire Ferguson}\BibitemShut {NoStop}%
\bibitem [{\citenamefont {Bourgade}, \citenamefont {Erd\H{o}s},\ and\
  \citenamefont {Yau}(2012)}]{BouErdYau-12}%
  \BibitemOpen
  \bibfield  {author} {\bibinfo {author} {\bibnamefont {Bourgade},
  \bibfnamefont {P.}}, \bibinfo {author} {\bibnamefont {Erd\H{o}s},
  \bibfnamefont {L.}}, \ and\ \bibinfo {author} {\bibnamefont {Yau},
  \bibfnamefont {H.-T.}},\ }\bibfield  {title} {\enquote {\bibinfo {title}
  {Bulk universality of general {$\beta$}-ensembles with non-convex
  potential},}\ }\href {\doibase 10.1063/1.4751478} {\bibfield  {journal}
  {\bibinfo  {journal} {J. Math. Phys.}\ }\textbf {\bibinfo {volume} {53}},\
  \bibinfo {pages} {095221, 19} (\bibinfo {year} {2012})}\BibitemShut {NoStop}%
\bibitem [{\citenamefont {Bourgade}, \citenamefont {Erd\H{o}s},\ and\
  \citenamefont {Yau}(2014)}]{BouErdYau-14}%
  \BibitemOpen
  \bibfield  {author} {\bibinfo {author} {\bibnamefont {Bourgade},
  \bibfnamefont {P.}}, \bibinfo {author} {\bibnamefont {Erd\H{o}s},
  \bibfnamefont {L.}}, \ and\ \bibinfo {author} {\bibnamefont {Yau},
  \bibfnamefont {H.-T.}},\ }\bibfield  {title} {\enquote {\bibinfo {title}
  {Universality of general {$\beta$}-ensembles},}\ }\href {\doibase
  10.1215/00127094-2649752} {\bibfield  {journal} {\bibinfo  {journal} {Duke
  Math. J.}\ }\textbf {\bibinfo {volume} {163}},\ \bibinfo {pages} {1127--1190}
  (\bibinfo {year} {2014})}\BibitemShut {NoStop}%
\bibitem [{\citenamefont {Bourgade}\ and\ \citenamefont
  {Keating}(2013)}]{BouKea-13}%
  \BibitemOpen
  \bibfield  {author} {\bibinfo {author} {\bibnamefont {Bourgade},
  \bibfnamefont {P.}}\ and\ \bibinfo {author} {\bibnamefont {Keating},
  \bibfnamefont {J.~P.}},\ }\bibfield  {title} {\enquote {\bibinfo {title}
  {Quantum chaos, random matrix theory, and the {R}iemann
  {$\zeta$}-function},}\ }in\ \href {\doibase 10.1007/978-3-0348-0697-8\_4}
  {\emph {\bibinfo {booktitle} {Chaos}}},\ \bibinfo {series} {Prog. Math.
  Phys.}, Vol.~\bibinfo {volume} {66}\ (\bibinfo  {publisher}
  {Birkh\"{a}user/Springer, Basel},\ \bibinfo {year} {2013})\ pp.\ \bibinfo
  {pages} {125--168},\ \bibinfo {note} {proceedings of the 14th Poincar\'e
  Seminar held in Paris, June 5, 2010}\BibitemShut {NoStop}%
\bibitem [{\citenamefont {Boursier}(2021)}]{Boursier-21_ppt}%
  \BibitemOpen
  \bibfield  {author} {\bibinfo {author} {\bibnamefont {Boursier},
  \bibfnamefont {J.}},\ }\href@noop {} {\enquote {\bibinfo {title} {Optimal
  local laws and {CLT} for {1D} long-range {R}iesz gases},}\ } (\bibinfo {year}
  {2021}),\ \Eprint {http://arxiv.org/abs/2112.05881} {arXiv:2112.05881
  [math.PR]} \BibitemShut {NoStop}%
\bibitem [{\citenamefont {Brascamp}\ and\ \citenamefont
  {Lieb}(1975)}]{BraLie-75}%
  \BibitemOpen
  \bibfield  {author} {\bibinfo {author} {\bibnamefont {Brascamp},
  \bibfnamefont {H.~J.}}\ and\ \bibinfo {author} {\bibnamefont {Lieb},
  \bibfnamefont {E.~H.}},\ }\bibfield  {title} {\enquote {\bibinfo {title}
  {Some inequalities for {G}aussian measures and the long-range order of the
  one-dimensional plasma},}\ }in\ \href@noop {} {\emph {\bibinfo {booktitle}
  {Functional Integration and Its Applications}}},\ \bibinfo {editor} {edited
  by\ \bibinfo {editor} {\bibfnamefont {A.}~\bibnamefont {Arthurs}}}\ (\bibinfo
   {publisher} {Clarendon Press},\ \bibinfo {address} {Oxford},\ \bibinfo
  {year} {1975})\BibitemShut {NoStop}%
\bibitem [{\citenamefont {Brauchart}, \citenamefont {Dragnev},\ and\
  \citenamefont {Saff}(2014)}]{BraDraSaf-14}%
  \BibitemOpen
  \bibfield  {author} {\bibinfo {author} {\bibnamefont {Brauchart},
  \bibfnamefont {J.~S.}}, \bibinfo {author} {\bibnamefont {Dragnev},
  \bibfnamefont {P.~D.}}, \ and\ \bibinfo {author} {\bibnamefont {Saff},
  \bibfnamefont {E.~B.}},\ }\bibfield  {title} {\enquote {\bibinfo {title}
  {Riesz external field problems on the hypersphere and optimal point
  separation},}\ }\href {\doibase 10.1007/s11118-014-9387-8} {\bibfield
  {journal} {\bibinfo  {journal} {Potential Anal.}\ }\textbf {\bibinfo {volume}
  {41}},\ \bibinfo {pages} {647--678} (\bibinfo {year} {2014})}\BibitemShut
  {NoStop}%
\bibitem [{\citenamefont {Brauchart}, \citenamefont {Hardin},\ and\
  \citenamefont {Saff}(2009)}]{BraHarSaf-09}%
  \BibitemOpen
  \bibfield  {author} {\bibinfo {author} {\bibnamefont {Brauchart},
  \bibfnamefont {J.~S.}}, \bibinfo {author} {\bibnamefont {Hardin},
  \bibfnamefont {D.~P.}}, \ and\ \bibinfo {author} {\bibnamefont {Saff},
  \bibfnamefont {E.~B.}},\ }\bibfield  {title} {\enquote {\bibinfo {title} {The
  {R}iesz energy of the {$N$}th roots of unity: an asymptotic expansion for
  large {$N$}},}\ }\href {\doibase 10.1112/blms/bdp034} {\bibfield  {journal}
  {\bibinfo  {journal} {Bull. Lond. Math. Soc.}\ }\textbf {\bibinfo {volume}
  {41}},\ \bibinfo {pages} {621--633} (\bibinfo {year} {2009})}\BibitemShut
  {NoStop}%
\bibitem [{\citenamefont {Brauchart}, \citenamefont {Hardin},\ and\
  \citenamefont {Saff}(2012)}]{BraHarSaf-12}%
  \BibitemOpen
  \bibfield  {author} {\bibinfo {author} {\bibnamefont {Brauchart},
  \bibfnamefont {J.~S.}}, \bibinfo {author} {\bibnamefont {Hardin},
  \bibfnamefont {D.~P.}}, \ and\ \bibinfo {author} {\bibnamefont {Saff},
  \bibfnamefont {E.~B.}},\ }\bibfield  {title} {\enquote {\bibinfo {title} {The
  next-order term for optimal {R}iesz and logarithmic energy asymptotics on the
  sphere},}\ }in\ \href {\doibase 10.1090/conm/578/11483} {\emph {\bibinfo
  {booktitle} {Recent advances in orthogonal polynomials, special functions,
  and their applications}}},\ \bibinfo {series} {Contemp. Math.}, Vol.\
  \bibinfo {volume} {578}\ (\bibinfo  {publisher} {Amer. Math. Soc.,
  Providence, RI},\ \bibinfo {year} {2012})\ pp.\ \bibinfo {pages}
  {31--61}\BibitemShut {NoStop}%
\bibitem [{\citenamefont {Brody}\ \emph {et~al.}(1981)\citenamefont {Brody},
  \citenamefont {Flores}, \citenamefont {French}, \citenamefont {Mello},
  \citenamefont {Pandey},\ and\ \citenamefont {Wong}}]{BroFloFreMelPanWon-81}%
  \BibitemOpen
  \bibfield  {author} {\bibinfo {author} {\bibnamefont {Brody}, \bibfnamefont
  {T.~A.}}, \bibinfo {author} {\bibnamefont {Flores}, \bibfnamefont {J.}},
  \bibinfo {author} {\bibnamefont {French}, \bibfnamefont {J.~B.}}, \bibinfo
  {author} {\bibnamefont {Mello}, \bibfnamefont {P.~A.}}, \bibinfo {author}
  {\bibnamefont {Pandey}, \bibfnamefont {A.}}, \ and\ \bibinfo {author}
  {\bibnamefont {Wong}, \bibfnamefont {S.~S.~M.}},\ }\bibfield  {title}
  {\enquote {\bibinfo {title} {Random-matrix physics: spectrum and strength
  fluctuations},}\ }\href {\doibase 10.1103/RevModPhys.53.385} {\bibfield
  {journal} {\bibinfo  {journal} {Rev. Mod. Phys.}\ }\textbf {\bibinfo {volume}
  {53}},\ \bibinfo {pages} {385--479} (\bibinfo {year} {1981})}\BibitemShut
  {NoStop}%
\bibitem [{\citenamefont {Brush}, \citenamefont {Sahlin},\ and\ \citenamefont
  {Teller}(1966)}]{BruSahTel-66}%
  \BibitemOpen
  \bibfield  {author} {\bibinfo {author} {\bibnamefont {Brush}, \bibfnamefont
  {S.~G.}}, \bibinfo {author} {\bibnamefont {Sahlin}, \bibfnamefont {H.~L.}}, \
  and\ \bibinfo {author} {\bibnamefont {Teller}, \bibfnamefont {E.}},\
  }\bibfield  {title} {\enquote {\bibinfo {title} {Monte {C}arlo study of a
  one-component plasma. i},}\ }\href {\doibase 10.1063/1.1727895} {\bibfield
  {journal} {\bibinfo  {journal} {J. Chem. Phys.}\ }\textbf {\bibinfo {volume}
  {45}},\ \bibinfo {pages} {2102--2118} (\bibinfo {year} {1966})}\BibitemShut
  {NoStop}%
\bibitem [{\citenamefont {Brydges}(1978)}]{Brydges-78}%
  \BibitemOpen
  \bibfield  {author} {\bibinfo {author} {\bibnamefont {Brydges}, \bibfnamefont
  {D.~C.}},\ }\bibfield  {title} {\enquote {\bibinfo {title} {A rigorous
  approach to {D}ebye screening in dilute classical {C}oulomb systems},}\
  }\href {http://projecteuclid.org/euclid.cmp/1103901494} {\bibfield  {journal}
  {\bibinfo  {journal} {Comm. Math. Phys.}\ }\textbf {\bibinfo {volume} {58}},\
  \bibinfo {pages} {313--350} (\bibinfo {year} {1978})}\BibitemShut {NoStop}%
\bibitem [{\citenamefont {Brydges}\ and\ \citenamefont
  {Federbush}(1980)}]{BryFed-80}%
  \BibitemOpen
  \bibfield  {author} {\bibinfo {author} {\bibnamefont {Brydges}, \bibfnamefont
  {D.~C.}}\ and\ \bibinfo {author} {\bibnamefont {Federbush}, \bibfnamefont
  {P.}},\ }\bibfield  {title} {\enquote {\bibinfo {title} {Debye screening},}\
  }\href {http://projecteuclid.org/euclid.cmp/1103907873} {\bibfield  {journal}
  {\bibinfo  {journal} {Comm. Math. Phys.}\ }\textbf {\bibinfo {volume} {73}},\
  \bibinfo {pages} {197--246} (\bibinfo {year} {1980})}\BibitemShut {NoStop}%
\bibitem [{\citenamefont {Brydges}\ and\ \citenamefont
  {Martin}(1999)}]{BryMar-99}%
  \BibitemOpen
  \bibfield  {author} {\bibinfo {author} {\bibnamefont {Brydges}, \bibfnamefont
  {D.~C.}}\ and\ \bibinfo {author} {\bibnamefont {Martin}, \bibfnamefont
  {P.~A.}},\ }\bibfield  {title} {\enquote {\bibinfo {title} {Coulomb systems
  at low density: A review},}\ }\href {\doibase 10.1023/A:1004600603161}
  {\bibfield  {journal} {\bibinfo  {journal} {Journal of Statistical Physics}\
  }\textbf {\bibinfo {volume} {96}},\ \bibinfo {pages} {1163--1330} (\bibinfo
  {year} {1999})}\BibitemShut {NoStop}%
\bibitem [{\citenamefont {Caffarelli}\ and\ \citenamefont
  {Silvestre}(2007)}]{CafSil-07}%
  \BibitemOpen
  \bibfield  {author} {\bibinfo {author} {\bibnamefont {Caffarelli},
  \bibfnamefont {L.}}\ and\ \bibinfo {author} {\bibnamefont {Silvestre},
  \bibfnamefont {L.}},\ }\bibfield  {title} {\enquote {\bibinfo {title} {An
  extension problem related to the fractional {L}aplacian},}\ }\href {\doibase
  10.1080/03605300600987306} {\bibfield  {journal} {\bibinfo  {journal} {Comm.
  Partial Differential Equations}\ }\textbf {\bibinfo {volume} {32}},\ \bibinfo
  {pages} {1245--1260} (\bibinfo {year} {2007})}\BibitemShut {NoStop}%
\bibitem [{\citenamefont {Caglioti}\ \emph {et~al.}(1992)\citenamefont
  {Caglioti}, \citenamefont {Lions}, \citenamefont {Marchioro},\ and\
  \citenamefont {Pulvirenti}}]{CagLioMarPul-92}%
  \BibitemOpen
  \bibfield  {author} {\bibinfo {author} {\bibnamefont {Caglioti},
  \bibfnamefont {E.}}, \bibinfo {author} {\bibnamefont {Lions}, \bibfnamefont
  {P.-L.}}, \bibinfo {author} {\bibnamefont {Marchioro}, \bibfnamefont {C.}}, \
  and\ \bibinfo {author} {\bibnamefont {Pulvirenti}, \bibfnamefont {M.}},\
  }\bibfield  {title} {\enquote {\bibinfo {title} {A special class of
  stationary flows for two-dimensional {E}uler equations: a statistical
  mechanics description},}\ }\href
  {http://projecteuclid.org/getRecord?id=euclid.cmp/1104249078} {\bibfield
  {journal} {\bibinfo  {journal} {Comm. Math. Phys.}\ }\textbf {\bibinfo
  {volume} {143}},\ \bibinfo {pages} {501--525} (\bibinfo {year}
  {1992})}\BibitemShut {NoStop}%
\bibitem [{\citenamefont {Caglioti}\ \emph {et~al.}(1995)\citenamefont
  {Caglioti}, \citenamefont {Lions}, \citenamefont {Marchioro},\ and\
  \citenamefont {Pulvirenti}}]{CagLioMarPul-95}%
  \BibitemOpen
  \bibfield  {author} {\bibinfo {author} {\bibnamefont {Caglioti},
  \bibfnamefont {E.}}, \bibinfo {author} {\bibnamefont {Lions}, \bibfnamefont
  {P.-L.}}, \bibinfo {author} {\bibnamefont {Marchioro}, \bibfnamefont {C.}}, \
  and\ \bibinfo {author} {\bibnamefont {Pulvirenti}, \bibfnamefont {M.}},\
  }\bibfield  {title} {\enquote {\bibinfo {title} {A special class of
  stationary flows for two-dimensional {E}uler equations: a statistical
  mechanics description. {II}},}\ }\href
  {http://projecteuclid.org/getRecord?id=euclid.cmp/1104275293} {\bibfield
  {journal} {\bibinfo  {journal} {Comm. Math. Phys.}\ }\textbf {\bibinfo
  {volume} {174}},\ \bibinfo {pages} {229--260} (\bibinfo {year}
  {1995})}\BibitemShut {NoStop}%
\bibitem [{\citenamefont {Caillol}\ and\ \citenamefont
  {Levesque}(1986)}]{CaiLev-86}%
  \BibitemOpen
  \bibfield  {author} {\bibinfo {author} {\bibnamefont {Caillol}, \bibfnamefont
  {J.~M.}}\ and\ \bibinfo {author} {\bibnamefont {Levesque}, \bibfnamefont
  {D.}},\ }\bibfield  {title} {\enquote {\bibinfo {title} {Low-density phase
  diagram of the two-dimensional {C}oulomb gas},}\ }\href {\doibase
  10.1103/PhysRevB.33.499} {\bibfield  {journal} {\bibinfo  {journal} {Phys.
  Rev. B}\ }\textbf {\bibinfo {volume} {33}},\ \bibinfo {pages} {499--509}
  (\bibinfo {year} {1986})}\BibitemShut {NoStop}%
\bibitem [{\citenamefont {Caillol}\ \emph {et~al.}(1982)\citenamefont
  {Caillol}, \citenamefont {Levesque}, \citenamefont {Weis},\ and\
  \citenamefont {Hansen}}]{CaiLevWeiHan-82}%
  \BibitemOpen
  \bibfield  {author} {\bibinfo {author} {\bibnamefont {Caillol}, \bibfnamefont
  {J.~M.}}, \bibinfo {author} {\bibnamefont {Levesque}, \bibfnamefont {D.}},
  \bibinfo {author} {\bibnamefont {Weis}, \bibfnamefont {J.~J.}}, \ and\
  \bibinfo {author} {\bibnamefont {Hansen}, \bibfnamefont {J.~P.}},\ }\bibfield
   {title} {\enquote {\bibinfo {title} {{A Monte Carlo study of the classical
  two-dimensional one-component plasma}},}\ }\href {\doibase
  10.1007/BF01012609} {\bibfield  {journal} {\bibinfo  {journal} {J. Stat.
  Phys.}\ }\textbf {\bibinfo {volume} {28}},\ \bibinfo {pages} {325--349}
  (\bibinfo {year} {1982})}\BibitemShut {NoStop}%
\bibitem [{\citenamefont {Callaway}(1991)}]{Callaway-91}%
  \BibitemOpen
  \bibfield  {author} {\bibinfo {author} {\bibnamefont {Callaway},
  \bibfnamefont {D.~J.~E.}},\ }\bibfield  {title} {\enquote {\bibinfo {title}
  {Random matrices, fractional statistics, and the quantum {H}all effect},}\
  }\href {\doibase 10.1103/PhysRevB.43.8641} {\bibfield  {journal} {\bibinfo
  {journal} {Phys. Rev. B}\ }\textbf {\bibinfo {volume} {43}},\ \bibinfo
  {pages} {8641--8643} (\bibinfo {year} {1991})}\BibitemShut {NoStop}%
\bibitem [{\citenamefont {Calogero}(1971)}]{Calogero-71}%
  \BibitemOpen
  \bibfield  {author} {\bibinfo {author} {\bibnamefont {Calogero},
  \bibfnamefont {F.}},\ }\bibfield  {title} {\enquote {\bibinfo {title}
  {Solution of the one-dimensional {$N$}-body problems with quadratic and/or
  inversely quadratic pair potentials},}\ }\href {\doibase 10.1063/1.1665604}
  {\bibfield  {journal} {\bibinfo  {journal} {J. Mathematical Phys.}\ }\textbf
  {\bibinfo {volume} {12}},\ \bibinfo {pages} {419--436} (\bibinfo {year}
  {1971})}\BibitemShut {NoStop}%
\bibitem [{\citenamefont {Campa}\ \emph {et~al.}(2014)\citenamefont {Campa},
  \citenamefont {Dauxois}, \citenamefont {Fanelli},\ and\ \citenamefont
  {Ruffo}}]{CamDauFanRuf-14}%
  \BibitemOpen
  \bibfield  {author} {\bibinfo {author} {\bibnamefont {Campa}, \bibfnamefont
  {A.}}, \bibinfo {author} {\bibnamefont {Dauxois}, \bibfnamefont {T.}},
  \bibinfo {author} {\bibnamefont {Fanelli}, \bibfnamefont {D.}}, \ and\
  \bibinfo {author} {\bibnamefont {Ruffo}, \bibfnamefont {S.}},\ }\href
  {\doibase 10.1093/acprof:oso/9780199581931.001.0001} {\emph {\bibinfo {title}
  {{Physics of Long-Range Interacting Systems}}}}\ (\bibinfo  {publisher}
  {Oxford University Press},\ \bibinfo {year} {2014})\BibitemShut {NoStop}%
\bibitem [{\citenamefont {Campa}, \citenamefont {Dauxois},\ and\ \citenamefont
  {Ruffo}(2009)}]{CamDauRuf-09}%
  \BibitemOpen
  \bibfield  {author} {\bibinfo {author} {\bibnamefont {Campa}, \bibfnamefont
  {A.}}, \bibinfo {author} {\bibnamefont {Dauxois}, \bibfnamefont {T.}}, \ and\
  \bibinfo {author} {\bibnamefont {Ruffo}, \bibfnamefont {S.}},\ }\bibfield
  {title} {\enquote {\bibinfo {title} {Statistical mechanics and dynamics of
  solvable models with long-range interactions},}\ }\href {\doibase
  https://doi.org/10.1016/j.physrep.2009.07.001} {\bibfield  {journal}
  {\bibinfo  {journal} {Physics Reports}\ }\textbf {\bibinfo {volume} {480}},\
  \bibinfo {pages} {57--159} (\bibinfo {year} {2009})}\BibitemShut {NoStop}%
\bibitem [{\citenamefont {{Campanino}}, \citenamefont {{Capocaccia}},\ and\
  \citenamefont {{Olivieri}}(1983)}]{CamCapOl-83}%
  \BibitemOpen
  \bibfield  {author} {\bibinfo {author} {\bibnamefont {{Campanino}},
  \bibfnamefont {M.}}, \bibinfo {author} {\bibnamefont {{Capocaccia}},
  \bibfnamefont {D.}}, \ and\ \bibinfo {author} {\bibnamefont {{Olivieri}},
  \bibfnamefont {E.}},\ }\bibfield  {title} {\enquote {\bibinfo {title}
  {{Analyticity for one-dimensional systems with long-range superstable
  interactions}},}\ }\href {\doibase 10.1007/BF01009805} {\bibfield  {journal}
  {\bibinfo  {journal} {{J. Stat. Phys.}}\ }\textbf {\bibinfo {volume} {33}},\
  \bibinfo {pages} {437--476} (\bibinfo {year} {1983})}\BibitemShut {NoStop}%
\bibitem [{\citenamefont {Canc{\`e}s}, \citenamefont {Lahbabi},\ and\
  \citenamefont {Lewin}(2013)}]{CanLahLew-13}%
  \BibitemOpen
  \bibfield  {author} {\bibinfo {author} {\bibnamefont {Canc{\`e}s},
  \bibfnamefont {{\'E}.}}, \bibinfo {author} {\bibnamefont {Lahbabi},
  \bibfnamefont {S.}}, \ and\ \bibinfo {author} {\bibnamefont {Lewin},
  \bibfnamefont {M.}},\ }\bibfield  {title} {\enquote {\bibinfo {title}
  {Mean-field models for disordered crystals},}\ }\href {\doibase
  10.1016/j.matpur.2012.12.003} {\bibfield  {journal} {\bibinfo  {journal} {J.
  Math. Pures Appl.}\ }\textbf {\bibinfo {volume} {100}},\ \bibinfo {pages}
  {241--274} (\bibinfo {year} {2013})},\ \Eprint
  {http://arxiv.org/abs/1203.0402} {1203.0402} \BibitemShut {NoStop}%
\bibitem [{\citenamefont {C{\^a}ndido}, \citenamefont {Bernu},\ and\
  \citenamefont {Ceperley}(2004)}]{CanBerCep-04}%
  \BibitemOpen
  \bibfield  {author} {\bibinfo {author} {\bibnamefont {C{\^a}ndido},
  \bibfnamefont {L.}}, \bibinfo {author} {\bibnamefont {Bernu}, \bibfnamefont
  {B.}}, \ and\ \bibinfo {author} {\bibnamefont {Ceperley}, \bibfnamefont
  {D.~M.}},\ }\bibfield  {title} {\enquote {\bibinfo {title} {Magnetic ordering
  of the three-dimensional wigner crystal},}\ }\href {\doibase
  10.1103/PhysRevB.70.094413} {\bibfield  {journal} {\bibinfo  {journal} {Phys.
  Rev. B}\ }\textbf {\bibinfo {volume} {70}},\ \bibinfo {pages} {094413}
  (\bibinfo {year} {2004})}\BibitemShut {NoStop}%
\bibitem [{\citenamefont {Cassels}(1959)}]{Cassels-59}%
  \BibitemOpen
  \bibfield  {author} {\bibinfo {author} {\bibnamefont {Cassels}, \bibfnamefont
  {J.~W.~S.}},\ }\bibfield  {title} {\enquote {\bibinfo {title} {On a problem
  of {R}ankin about the {E}pstein zeta-function},}\ }\href {\doibase
  10.1017/S2040618500033906} {\bibfield  {journal} {\bibinfo  {journal} {Proc.
  Glasgow Math. Assoc.}\ }\textbf {\bibinfo {volume} {4}},\ \bibinfo {pages}
  {73--80} (\bibinfo {year} {1959})}\BibitemShut {NoStop}%
\bibitem [{\citenamefont {Catto}, \citenamefont {{Le Bris}},\ and\
  \citenamefont {Lions}(1998)}]{CatBriLio-98}%
  \BibitemOpen
  \bibfield  {author} {\bibinfo {author} {\bibnamefont {Catto}, \bibfnamefont
  {I.}}, \bibinfo {author} {\bibnamefont {{Le Bris}}, \bibfnamefont {C.}}, \
  and\ \bibinfo {author} {\bibnamefont {Lions}, \bibfnamefont {P.-L.}},\
  }\href@noop {} {\emph {\bibinfo {title} {The mathematical theory of
  thermodynamic limits: {T}homas-{F}ermi type models}}},\ Oxford Mathematical
  Monographs\ (\bibinfo  {publisher} {The Clarendon Press Oxford University
  Press},\ \bibinfo {address} {New York},\ \bibinfo {year} {1998})\ pp.\
  \bibinfo {pages} {xiv+277}\BibitemShut {NoStop}%
\bibitem [{\citenamefont {Catto}, \citenamefont {{Le Bris}},\ and\
  \citenamefont {Lions}(2001)}]{CatBriLio-01}%
  \BibitemOpen
  \bibfield  {author} {\bibinfo {author} {\bibnamefont {Catto}, \bibfnamefont
  {I.}}, \bibinfo {author} {\bibnamefont {{Le Bris}}, \bibfnamefont {C.}}, \
  and\ \bibinfo {author} {\bibnamefont {Lions}, \bibfnamefont {P.-L.}},\
  }\bibfield  {title} {\enquote {\bibinfo {title} {On the thermodynamic limit
  for {H}artree-{F}ock type models},}\ }\href {\doibase
  10.1016/S0294-1449(00)00059-7} {\bibfield  {journal} {\bibinfo  {journal}
  {Ann. Inst. H. Poincar{\'e} Anal. Non Lin{\'e}aire}\ }\textbf {\bibinfo
  {volume} {18}},\ \bibinfo {pages} {687--760} (\bibinfo {year}
  {2001})}\BibitemShut {NoStop}%
\bibitem [{\citenamefont {Ceperley}\ and\ \citenamefont
  {Alder}(1980)}]{AldCep-80}%
  \BibitemOpen
  \bibfield  {author} {\bibinfo {author} {\bibnamefont {Ceperley},
  \bibfnamefont {D.~M.}}\ and\ \bibinfo {author} {\bibnamefont {Alder},
  \bibfnamefont {B.~J.}},\ }\bibfield  {title} {\enquote {\bibinfo {title}
  {{Ground State of the Electron Gas by a Stochastic Method}},}\ }\href
  {\doibase 10.1103/PhysRevLett.45.566} {\bibfield  {journal} {\bibinfo
  {journal} {Phys. Rev. Lett.}\ }\textbf {\bibinfo {volume} {45}},\ \bibinfo
  {pages} {566--569} (\bibinfo {year} {1980})}\BibitemShut {NoStop}%
\bibitem [{\citenamefont {Chafa{\"i}}(2021)}]{Chafai-21_ppt}%
  \BibitemOpen
  \bibfield  {author} {\bibinfo {author} {\bibnamefont {Chafa{\"i}},
  \bibfnamefont {D.}},\ }\href@noop {} {\enquote {\bibinfo {title} {Aspects of
  coulomb gases},}\ } (\bibinfo {year} {2021}),\ \Eprint
  {http://arxiv.org/abs/2108.10653} {arXiv:2108.10653 [math.PR]} \BibitemShut
  {NoStop}%
\bibitem [{\citenamefont {Chakravarty}\ and\ \citenamefont
  {Dasgupta}(1980)}]{ChaDas-80}%
  \BibitemOpen
  \bibfield  {author} {\bibinfo {author} {\bibnamefont {Chakravarty},
  \bibfnamefont {S.}}\ and\ \bibinfo {author} {\bibnamefont {Dasgupta},
  \bibfnamefont {C.}},\ }\bibfield  {title} {\enquote {\bibinfo {title}
  {Absence of crystalline order in two dimensions},}\ }\href {\doibase
  10.1103/PhysRevB.22.369} {\bibfield  {journal} {\bibinfo  {journal} {Phys.
  Rev. B}\ }\textbf {\bibinfo {volume} {22}},\ \bibinfo {pages} {369--372}
  (\bibinfo {year} {1980})}\BibitemShut {NoStop}%
\bibitem [{\citenamefont {Chen}\ and\ \citenamefont
  {Oshita}(2007)}]{CheOsh-07}%
  \BibitemOpen
  \bibfield  {author} {\bibinfo {author} {\bibnamefont {Chen}, \bibfnamefont
  {X.}}\ and\ \bibinfo {author} {\bibnamefont {Oshita}, \bibfnamefont {Y.}},\
  }\bibfield  {title} {\enquote {\bibinfo {title} {An application of the
  modular function in nonlocal variational problems},}\ }\href {\doibase
  10.1007/s00205-007-0050-z} {\bibfield  {journal} {\bibinfo  {journal} {Arch.
  Ration. Mech. Anal.}\ }\textbf {\bibinfo {volume} {186}},\ \bibinfo {pages}
  {109--132} (\bibinfo {year} {2007})}\BibitemShut {NoStop}%
\bibitem [{\citenamefont {Choquard}(1975)}]{Choquard-75}%
  \BibitemOpen
  \bibfield  {author} {\bibinfo {author} {\bibnamefont {Choquard},
  \bibfnamefont {P.}},\ }\bibfield  {title} {\enquote {\bibinfo {title} {On the
  statistical mechanics of one-dimensional {C}oulomb systems},}\ }\href@noop {}
  {\bibfield  {journal} {\bibinfo  {journal} {Helv. Phys. Acta}\ }\textbf
  {\bibinfo {volume} {48}},\ \bibinfo {pages} {585--598} (\bibinfo {year}
  {1975})}\BibitemShut {NoStop}%
\bibitem [{\citenamefont {Choquard}(1978)}]{Choquard-78}%
  \BibitemOpen
  \bibfield  {author} {\bibinfo {author} {\bibnamefont {Choquard},
  \bibfnamefont {P.}},\ }\enquote {\bibinfo {title} {Selected topics on the
  equilibrium statistical mechanics of {C}oulomb systems},}\ in\ \href
  {\doibase 10.1007/978-1-4613-2868-1_11} {\emph {\bibinfo {booktitle}
  {Strongly Coupled Plasmas}}},\ \bibinfo {editor} {edited by\ \bibinfo
  {editor} {\bibfnamefont {G.}~\bibnamefont {Kalman}}\ and\ \bibinfo {editor}
  {\bibfnamefont {P.}~\bibnamefont {Carini}}}\ (\bibinfo  {publisher} {Springer
  US},\ \bibinfo {address} {Boston, MA},\ \bibinfo {year} {1978})\ pp.\
  \bibinfo {pages} {347--406}\BibitemShut {NoStop}%
\bibitem [{\citenamefont {Choquard}(2000)}]{Choquard-00}%
  \BibitemOpen
  \bibfield  {author} {\bibinfo {author} {\bibnamefont {Choquard},
  \bibfnamefont {P.}},\ }\bibfield  {title} {\enquote {\bibinfo {title}
  {Classical and quantum partition functions of the
  {C}alogero-{M}oser-{S}utherland model},}\ }in\ \href {\doibase
  10.1007/978-1-4612-1206-5\_8} {\emph {\bibinfo {booktitle}
  {Calogero-{M}oser-{S}utherland models ({M}ontr\'{e}al, {QC}, 1997)}}},\
  \bibinfo {series and number} {CRM Ser. Math. Phys.}\ (\bibinfo  {publisher}
  {Springer, New York},\ \bibinfo {year} {2000})\ pp.\ \bibinfo {pages}
  {117--125}\BibitemShut {NoStop}%
\bibitem [{\citenamefont {{Choquard}}, \citenamefont {{Favre}},\ and\
  \citenamefont {{Gruber}}(1980)}]{ChoFavGru-80}%
  \BibitemOpen
  \bibfield  {author} {\bibinfo {author} {\bibnamefont {{Choquard}},
  \bibfnamefont {P.}}, \bibinfo {author} {\bibnamefont {{Favre}}, \bibfnamefont
  {P.}}, \ and\ \bibinfo {author} {\bibnamefont {{Gruber}}, \bibfnamefont
  {C.}},\ }\bibfield  {title} {\enquote {\bibinfo {title} {On the equation of
  state of classical one-component systems with long-range forces},}\ }\href
  {\doibase 10.1007/BF01011574} {\bibfield  {journal} {\bibinfo  {journal} {J.
  Stat. Phys.}\ }\textbf {\bibinfo {volume} {23}},\ \bibinfo {pages} {405--442}
  (\bibinfo {year} {1980})}\BibitemShut {NoStop}%
\bibitem [{\citenamefont {Choquard}\ \emph {et~al.}(1981)\citenamefont
  {Choquard}, \citenamefont {Kunz}, \citenamefont {Martin},\ and\ \citenamefont
  {Navet}}]{ChoKunMarNav-81}%
  \BibitemOpen
  \bibfield  {author} {\bibinfo {author} {\bibnamefont {Choquard},
  \bibfnamefont {P.}}, \bibinfo {author} {\bibnamefont {Kunz}, \bibfnamefont
  {H.}}, \bibinfo {author} {\bibnamefont {Martin}, \bibfnamefont {P.~A.}}, \
  and\ \bibinfo {author} {\bibnamefont {Navet}, \bibfnamefont {M.}},\
  }\bibfield  {title} {\enquote {\bibinfo {title} {One-dimensional {C}oulomb
  systems},}\ }in\ \href@noop {} {\emph {\bibinfo {booktitle} {Physics in One
  Dimension}}},\ \bibinfo {editor} {edited by\ \bibinfo {editor} {\bibfnamefont
  {J.}~\bibnamefont {Bernasconi}}\ and\ \bibinfo {editor} {\bibfnamefont
  {T.}~\bibnamefont {Schneider}}}\ (\bibinfo  {publisher} {Springer Berlin
  Heidelberg},\ \bibinfo {address} {Berlin, Heidelberg},\ \bibinfo {year}
  {1981})\ pp.\ \bibinfo {pages} {335--350}\BibitemShut {NoStop}%
\bibitem [{\citenamefont {Choquet}(1959)}]{Choquet-58}%
  \BibitemOpen
  \bibfield  {author} {\bibinfo {author} {\bibnamefont {Choquet}, \bibfnamefont
  {G.}},\ }\bibfield  {title} {\enquote {\bibinfo {title} {Diam{\`e}tre
  transfini et comparaison de diverses capacit{\'e}s},}\ }in\ \href@noop {}
  {\emph {\bibinfo {booktitle} {S{\'e}minaire Brelot-Choquet-Deny. Th{\'e}orie
  du potentiel}}}\ (\bibinfo {year} {1958-1959})\ \bibinfo {note} {exp.
  4}\BibitemShut {NoStop}%
\bibitem [{\citenamefont {Chu}\ and\ \citenamefont {Lin}(1994)}]{ChuLin-94}%
  \BibitemOpen
  \bibfield  {author} {\bibinfo {author} {\bibnamefont {Chu}, \bibfnamefont
  {J.}}\ and\ \bibinfo {author} {\bibnamefont {Lin}, \bibfnamefont {I.}},\
  }\bibfield  {title} {\enquote {\bibinfo {title} {Coulomb lattice in a weakly
  ionized colloidal plasma},}\ }\href {\doibase 10.1016/0378-4371(94)90498-7}
  {\bibfield  {journal} {\bibinfo  {journal} {Physica A: Statistical Mechanics
  and its Applications}\ }\textbf {\bibinfo {volume} {205}},\ \bibinfo {pages}
  {183--190} (\bibinfo {year} {1994})}\BibitemShut {NoStop}%
\bibitem [{\citenamefont {Clark}, \citenamefont {Casula},\ and\ \citenamefont
  {Ceperley}(2009)}]{ClaCasCep-09}%
  \BibitemOpen
  \bibfield  {author} {\bibinfo {author} {\bibnamefont {Clark}, \bibfnamefont
  {B.~K.}}, \bibinfo {author} {\bibnamefont {Casula}, \bibfnamefont {M.}}, \
  and\ \bibinfo {author} {\bibnamefont {Ceperley}, \bibfnamefont {D.~M.}},\
  }\bibfield  {title} {\enquote {\bibinfo {title} {{Hexatic and Mesoscopic
  Phases in a 2D Quantum Coulomb System}},}\ }\href {\doibase
  10.1103/PhysRevLett.103.055701} {\bibfield  {journal} {\bibinfo  {journal}
  {Phys. Rev. Lett.}\ }\textbf {\bibinfo {volume} {103}},\ \bibinfo {pages}
  {055701} (\bibinfo {year} {2009})}\BibitemShut {NoStop}%
\bibitem [{\citenamefont {Cohn}\ and\ \citenamefont {Kumar}(2007)}]{CohKum-07}%
  \BibitemOpen
  \bibfield  {author} {\bibinfo {author} {\bibnamefont {Cohn}, \bibfnamefont
  {H.}}\ and\ \bibinfo {author} {\bibnamefont {Kumar}, \bibfnamefont {A.}},\
  }\bibfield  {title} {\enquote {\bibinfo {title} {Universally optimal
  distribution of points on spheres},}\ }\href {\doibase
  10.1090/S0894-0347-06-00546-7} {\bibfield  {journal} {\bibinfo  {journal} {J.
  Amer. Math. Soc.}\ }\textbf {\bibinfo {volume} {20}},\ \bibinfo {pages}
  {99--148} (\bibinfo {year} {2007})}\BibitemShut {NoStop}%
\bibitem [{\citenamefont {Cohn}\ \emph {et~al.}(2019)\citenamefont {Cohn},
  \citenamefont {Kumar}, \citenamefont {Miller}, \citenamefont {Radchenko},\
  and\ \citenamefont {Viazovska}}]{CohKumMilRadVia-19_ppt}%
  \BibitemOpen
  \bibfield  {author} {\bibinfo {author} {\bibnamefont {Cohn}, \bibfnamefont
  {H.}}, \bibinfo {author} {\bibnamefont {Kumar}, \bibfnamefont {A.}}, \bibinfo
  {author} {\bibnamefont {Miller}, \bibfnamefont {S.~D.}}, \bibinfo {author}
  {\bibnamefont {Radchenko}, \bibfnamefont {D.}}, \ and\ \bibinfo {author}
  {\bibnamefont {Viazovska}, \bibfnamefont {M.}},\ }\bibfield  {title}
  {\enquote {\bibinfo {title} {Universal optimality of the ${E}_8$ and {L}eech
  lattices and interpolation formulas},}\ }\href@noop {} {\bibfield  {journal}
  {\bibinfo  {journal} {Ann. of Math.}\ } (\bibinfo {year} {2019})},\ \bibinfo
  {note} {to appear},\ \Eprint {http://arxiv.org/abs/1902.05438}
  {arXiv:1902.05438 [math.MG]} \BibitemShut {NoStop}%
\bibitem [{\citenamefont {Coldwell-Horsfall}\ and\ \citenamefont
  {Maradudin}(1960)}]{ColMar-60}%
  \BibitemOpen
  \bibfield  {author} {\bibinfo {author} {\bibnamefont {Coldwell-Horsfall},
  \bibfnamefont {R.~A.}}\ and\ \bibinfo {author} {\bibnamefont {Maradudin},
  \bibfnamefont {A.~A.}},\ }\bibfield  {title} {\enquote {\bibinfo {title}
  {Zero-point energy of an electron lattice},}\ }\href {\doibase
  10.1063/1.1703670} {\bibfield  {journal} {\bibinfo  {journal} {J. Math.
  Phys.}\ }\textbf {\bibinfo {volume} {1}},\ \bibinfo {pages} {395--404}
  (\bibinfo {year} {1960})}\BibitemShut {NoStop}%
\bibitem [{\citenamefont {Colombo}, \citenamefont {{De Pascale}},\ and\
  \citenamefont {{Di Marino}}(2015)}]{ColPasMar-15}%
  \BibitemOpen
  \bibfield  {author} {\bibinfo {author} {\bibnamefont {Colombo}, \bibfnamefont
  {M.}}, \bibinfo {author} {\bibnamefont {{De Pascale}}, \bibfnamefont {L.}}, \
  and\ \bibinfo {author} {\bibnamefont {{Di Marino}}, \bibfnamefont {S.}},\
  }\bibfield  {title} {\enquote {\bibinfo {title} {Multimarginal optimal
  transport maps for one-dimensional repulsive costs},}\ }\href {\doibase
  10.4153/CJM-2014-011-x} {\bibfield  {journal} {\bibinfo  {journal} {Canad. J.
  Math.}\ }\textbf {\bibinfo {volume} {67}},\ \bibinfo {pages} {350--368}
  (\bibinfo {year} {2015})}\BibitemShut {NoStop}%
\bibitem [{\citenamefont {Conlon}, \citenamefont {Lieb},\ and\ \citenamefont
  {Yau}(1989)}]{ConLieYau-89}%
  \BibitemOpen
  \bibfield  {author} {\bibinfo {author} {\bibnamefont {Conlon}, \bibfnamefont
  {J.~G.}}, \bibinfo {author} {\bibnamefont {Lieb}, \bibfnamefont {E.~H.}}, \
  and\ \bibinfo {author} {\bibnamefont {Yau}, \bibfnamefont {H.-T.}},\
  }\bibfield  {title} {\enquote {\bibinfo {title} {The {C}oulomb gas at low
  temperature and low density},}\ }\href {\doibase 10.1007/BF01217775}
  {\bibfield  {journal} {\bibinfo  {journal} {Commun. Math. Phys.}\ }\textbf
  {\bibinfo {volume} {125}},\ \bibinfo {pages} {153--180} (\bibinfo {year}
  {1989})}\BibitemShut {NoStop}%
\bibitem [{\citenamefont {Conway}\ and\ \citenamefont
  {Sloane}(1999)}]{ConSlo-99}%
  \BibitemOpen
  \bibfield  {author} {\bibinfo {author} {\bibnamefont {Conway}, \bibfnamefont
  {J.~H.}}\ and\ \bibinfo {author} {\bibnamefont {Sloane}, \bibfnamefont
  {N.~J.~A.}},\ }\href {\doibase 10.1007/978-1-4757-6568-7} {\emph {\bibinfo
  {title} {Sphere packings, lattices and groups}}},\ \bibinfo {edition} {3rd}\
  ed.,\ \bibinfo {series} {Grundlehren der Mathematischen Wissenschaften
  [Fundamental Principles of Mathematical Sciences]}, Vol.\ \bibinfo {volume}
  {290}\ (\bibinfo  {publisher} {Springer-Verlag, New York},\ \bibinfo {year}
  {1999})\ pp.\ \bibinfo {pages} {lxxiv+703},\ \bibinfo {note} {with additional
  contributions by E. Bannai, R. E. Borcherds, J. Leech, S. P. Norton, A. M.
  Odlyzko, R. A. Parker, L. Queen and B. B. Venkov}\BibitemShut {NoStop}%
\bibitem [{\citenamefont {Cooper}(2008)}]{Cooper-08}%
  \BibitemOpen
  \bibfield  {author} {\bibinfo {author} {\bibnamefont {Cooper}, \bibfnamefont
  {N.}},\ }\bibfield  {title} {\enquote {\bibinfo {title} {Rapidly rotating
  atomic gases},}\ }\href {\doibase 10.1080/00018730802564122} {\bibfield
  {journal} {\bibinfo  {journal} {Advances in Physics}\ }\textbf {\bibinfo
  {volume} {57}},\ \bibinfo {pages} {539--616} (\bibinfo {year}
  {2008})}\BibitemShut {NoStop}%
\bibitem [{\citenamefont {Cotar}, \citenamefont {Friesecke},\ and\
  \citenamefont {Kl{\"u}ppelberg}(2013)}]{CotFriKlu-13}%
  \BibitemOpen
  \bibfield  {author} {\bibinfo {author} {\bibnamefont {Cotar}, \bibfnamefont
  {C.}}, \bibinfo {author} {\bibnamefont {Friesecke}, \bibfnamefont {G.}}, \
  and\ \bibinfo {author} {\bibnamefont {Kl{\"u}ppelberg}, \bibfnamefont {C.}},\
  }\bibfield  {title} {\enquote {\bibinfo {title} {Density functional theory
  and optimal transportation with {C}oulomb cost},}\ }\href {\doibase
  10.1002/cpa.21437} {\bibfield  {journal} {\bibinfo  {journal} {Comm. Pure
  Appl. Math.}\ }\textbf {\bibinfo {volume} {66}},\ \bibinfo {pages} {548--599}
  (\bibinfo {year} {2013})}\BibitemShut {NoStop}%
\bibitem [{\citenamefont {Cotar}, \citenamefont {Friesecke},\ and\
  \citenamefont {Pass}(2015)}]{CotFriPas-15}%
  \BibitemOpen
  \bibfield  {author} {\bibinfo {author} {\bibnamefont {Cotar}, \bibfnamefont
  {C.}}, \bibinfo {author} {\bibnamefont {Friesecke}, \bibfnamefont {G.}}, \
  and\ \bibinfo {author} {\bibnamefont {Pass}, \bibfnamefont {B.}},\ }\bibfield
   {title} {\enquote {\bibinfo {title} {Infinite-body optimal transport with
  {C}oulomb cost},}\ }\href {\doibase 10.1007/s00526-014-0803-0} {\bibfield
  {journal} {\bibinfo  {journal} {Calc. Var. Partial Differ. Equ.}\ }\textbf
  {\bibinfo {volume} {54}},\ \bibinfo {pages} {717--742} (\bibinfo {year}
  {2015})}\BibitemShut {NoStop}%
\bibitem [{\citenamefont {{Cotar}}\ and\ \citenamefont
  {{Petrache}}(2019)}]{CotPet-19b}%
  \BibitemOpen
  \bibfield  {author} {\bibinfo {author} {\bibnamefont {{Cotar}}, \bibfnamefont
  {C.}}\ and\ \bibinfo {author} {\bibnamefont {{Petrache}}, \bibfnamefont
  {M.}},\ }\bibfield  {title} {\enquote {\bibinfo {title} {Equality of the
  jellium and uniform electron gas next-order asymptotic terms for {C}oulomb
  and {R}iesz potentials},}\ }\href@noop {} {\bibfield  {journal} {\bibinfo
  {journal} {ArXiv e-prints 1707.07664 (version 5)}\ } (\bibinfo {year}
  {2019})},\ \Eprint {http://arxiv.org/abs/1707.07664v5} {arXiv:1707.07664v5
  [math-ph]} \BibitemShut {NoStop}%
\bibitem [{\citenamefont {Cotar}\ and\ \citenamefont
  {Petrache}(2019)}]{CotPet-19}%
  \BibitemOpen
  \bibfield  {author} {\bibinfo {author} {\bibnamefont {Cotar}, \bibfnamefont
  {C.}}\ and\ \bibinfo {author} {\bibnamefont {Petrache}, \bibfnamefont {M.}},\
  }\bibfield  {title} {\enquote {\bibinfo {title} {Next-order asymptotic
  expansion for {$N$}-marginal optimal transport with {C}oulomb and {R}iesz
  costs},}\ }\href {\doibase 10.1016/j.aim.2018.12.008} {\bibfield  {journal}
  {\bibinfo  {journal} {Adv. Math.}\ }\textbf {\bibinfo {volume} {344}},\
  \bibinfo {pages} {137--233} (\bibinfo {year} {2019})}\BibitemShut {NoStop}%
\bibitem [{\citenamefont {Cribier}\ \emph {et~al.}(1964)\citenamefont
  {Cribier}, \citenamefont {Jacrot}, \citenamefont {{Madhav Rao}},\ and\
  \citenamefont {Farnoux}}]{CriJacMadFar-64}%
  \BibitemOpen
  \bibfield  {author} {\bibinfo {author} {\bibnamefont {Cribier}, \bibfnamefont
  {D.}}, \bibinfo {author} {\bibnamefont {Jacrot}, \bibfnamefont {B.}},
  \bibinfo {author} {\bibnamefont {{Madhav Rao}}, \bibfnamefont {L.}}, \ and\
  \bibinfo {author} {\bibnamefont {Farnoux}, \bibfnamefont {B.}},\ }\bibfield
  {title} {\enquote {\bibinfo {title} {Mise en \'evidence par diffraction de
  neutrons d'une structure p\'eriodique du champ magn\'etique dans le {N}iobium
  supraconducteur},}\ }\href {\doibase
  https://doi.org/10.1016/0031-9163(64)90096-4} {\bibfield  {journal} {\bibinfo
   {journal} {Physics Letters}\ }\textbf {\bibinfo {volume} {9}},\ \bibinfo
  {pages} {106--107} (\bibinfo {year} {1964})}\BibitemShut {NoStop}%
\bibitem [{\citenamefont {Dahlberg}(1978)}]{Dahlberg-78}%
  \BibitemOpen
  \bibfield  {author} {\bibinfo {author} {\bibnamefont {Dahlberg},
  \bibfnamefont {B.~E.~J.}},\ }\bibfield  {title} {\enquote {\bibinfo {title}
  {On the distribution of {F}ekete points},}\ }\href
  {http://projecteuclid.org/euclid.dmj/1077312950} {\bibfield  {journal}
  {\bibinfo  {journal} {Duke Math. J.}\ }\textbf {\bibinfo {volume} {45}},\
  \bibinfo {pages} {537--542} (\bibinfo {year} {1978})}\BibitemShut {NoStop}%
\bibitem [{\citenamefont {Dauxois}\ \emph {et~al.}(2002)\citenamefont
  {Dauxois}, \citenamefont {Ruffo}, \citenamefont {Arimondo},\ and\
  \citenamefont {Wilkens}}]{DauRufAriWil-02}%
  \BibitemOpen
  \bibinfo {editor} {\bibnamefont {Dauxois}, \bibfnamefont {T.}}, \bibinfo
  {editor} {\bibnamefont {Ruffo}, \bibfnamefont {S.}}, \bibinfo {editor}
  {\bibnamefont {Arimondo}, \bibfnamefont {E.}}, \ and\ \bibinfo {editor}
  {\bibnamefont {Wilkens}, \bibfnamefont {M.}},\ eds.,\ \href {\doibase
  10.1007/3-540-45835-2} {\emph {\bibinfo {title} {{Dynamics and Thermodynamics
  of Systems with Long-Range Interactions}}}},\ \bibinfo {series} {Lecture
  Notes in Physics}, Vol.\ \bibinfo {volume} {602}\ (\bibinfo  {publisher}
  {Springer, Berlin, Heidelberg},\ \bibinfo {year} {2002})\BibitemShut
  {NoStop}%
\bibitem [{\citenamefont {{de-Picciotto}}\ \emph {et~al.}(1997)\citenamefont
  {{de-Picciotto}}, \citenamefont {{Reznikov}}, \citenamefont {{Heiblum}},
  \citenamefont {{Umansky}}, \citenamefont {{Bunin}},\ and\ \citenamefont
  {{Mahalu}}}]{PicRezHeiUmaBunMah-97}%
  \BibitemOpen
  \bibfield  {author} {\bibinfo {author} {\bibnamefont {{de-Picciotto}},
  \bibfnamefont {R.}}, \bibinfo {author} {\bibnamefont {{Reznikov}},
  \bibfnamefont {M.}}, \bibinfo {author} {\bibnamefont {{Heiblum}},
  \bibfnamefont {M.}}, \bibinfo {author} {\bibnamefont {{Umansky}},
  \bibfnamefont {V.}}, \bibinfo {author} {\bibnamefont {{Bunin}}, \bibfnamefont
  {G.}}, \ and\ \bibinfo {author} {\bibnamefont {{Mahalu}}, \bibfnamefont
  {D.}},\ }\bibfield  {title} {\enquote {\bibinfo {title} {{Direct observation
  of a fractional charge}},}\ }\href {\doibase 10.1038/38241} {\bibfield
  {journal} {\bibinfo  {journal} {Nature}\ }\textbf {\bibinfo {volume} {389}},\
  \bibinfo {pages} {162--164} (\bibinfo {year} {1997})}\BibitemShut {NoStop}%
\bibitem [{\citenamefont {Deift}(2017)}]{Deift-17}%
  \BibitemOpen
  \bibfield  {author} {\bibinfo {author} {\bibnamefont {Deift}, \bibfnamefont
  {P.}},\ }\bibfield  {title} {\enquote {\bibinfo {title} {Some open problems
  in random matrix theory and the theory of integrable systems. {II}},}\ }\href
  {\doibase 10.3842/SIGMA.2017.016} {\bibfield  {journal} {\bibinfo  {journal}
  {SIGMA Symmetry Integrability Geom. Methods Appl.}\ }\textbf {\bibinfo
  {volume} {13}},\ \bibinfo {pages} {Paper No. 016, 23} (\bibinfo {year}
  {2017})}\BibitemShut {NoStop}%
\bibitem [{\citenamefont {Deike}\ \emph {et~al.}(2001)\citenamefont {Deike},
  \citenamefont {Ballauff}, \citenamefont {Willenbacher},\ and\ \citenamefont
  {Weiss}}]{DeiBalWilWei-01}%
  \BibitemOpen
  \bibfield  {author} {\bibinfo {author} {\bibnamefont {Deike}, \bibfnamefont
  {I.}}, \bibinfo {author} {\bibnamefont {Ballauff}, \bibfnamefont {M.}},
  \bibinfo {author} {\bibnamefont {Willenbacher}, \bibfnamefont {N.}}, \ and\
  \bibinfo {author} {\bibnamefont {Weiss}, \bibfnamefont {A.}},\ }\bibfield
  {title} {\enquote {\bibinfo {title} {Rheology of thermosensitive latex
  particles including the high-frequency limit},}\ }\href {\doibase
  10.1122/1.1357820} {\bibfield  {journal} {\bibinfo  {journal} {J. Rheol.}\
  }\textbf {\bibinfo {volume} {45}},\ \bibinfo {pages} {709--720} (\bibinfo
  {year} {2001})}\BibitemShut {NoStop}%
\bibitem [{\citenamefont {Dereudre}\ \emph {et~al.}(2021)\citenamefont
  {Dereudre}, \citenamefont {Hardy}, \citenamefont {Lebl\'{e}},\ and\
  \citenamefont {Ma\"{\i}da}}]{DerHarLebMai-21}%
  \BibitemOpen
  \bibfield  {author} {\bibinfo {author} {\bibnamefont {Dereudre},
  \bibfnamefont {D.}}, \bibinfo {author} {\bibnamefont {Hardy}, \bibfnamefont
  {A.}}, \bibinfo {author} {\bibnamefont {Lebl\'{e}}, \bibfnamefont {T.}}, \
  and\ \bibinfo {author} {\bibnamefont {Ma\"{\i}da}, \bibfnamefont {M.}},\
  }\bibfield  {title} {\enquote {\bibinfo {title} {D{LR} equations and rigidity
  for the sine-beta process},}\ }\href {\doibase 10.1002/cpa.21963} {\bibfield
  {journal} {\bibinfo  {journal} {Comm. Pure Appl. Math.}\ }\textbf {\bibinfo
  {volume} {74}},\ \bibinfo {pages} {172--222} (\bibinfo {year}
  {2021})}\BibitemShut {NoStop}%
\bibitem [{\citenamefont {Dereudre}\ and\ \citenamefont
  {Vasseur}(2020)}]{DerVas-20}%
  \BibitemOpen
  \bibfield  {author} {\bibinfo {author} {\bibnamefont {Dereudre},
  \bibfnamefont {D.}}\ and\ \bibinfo {author} {\bibnamefont {Vasseur},
  \bibfnamefont {T.}},\ }\bibfield  {title} {\enquote {\bibinfo {title}
  {Existence of {G}ibbs point processes with stable infinite range
  interaction},}\ }\href {\doibase 10.1017/jpr.2020.39} {\bibfield  {journal}
  {\bibinfo  {journal} {J. Appl. Probab.}\ }\textbf {\bibinfo {volume} {57}},\
  \bibinfo {pages} {775--791} (\bibinfo {year} {2020})}\BibitemShut {NoStop}%
\bibitem [{\citenamefont {Dereudre}\ and\ \citenamefont
  {Vasseur}(2021)}]{DerVas-21_ppt}%
  \BibitemOpen
  \bibfield  {author} {\bibinfo {author} {\bibnamefont {Dereudre},
  \bibfnamefont {D.}}\ and\ \bibinfo {author} {\bibnamefont {Vasseur},
  \bibfnamefont {T.}},\ }\href@noop {} {\enquote {\bibinfo {title}
  {Number-rigidity and $\beta$-circular {R}iesz gas},}\ } (\bibinfo {year}
  {2021}),\ \Eprint {http://arxiv.org/abs/2104.09408} {arXiv:2104.09408
  [math.PR]} \BibitemShut {NoStop}%
\bibitem [{\citenamefont {Derrick}(1969)}]{Derrick-69}%
  \BibitemOpen
  \bibfield  {author} {\bibinfo {author} {\bibnamefont {Derrick}, \bibfnamefont
  {G.~H.}},\ }\bibfield  {title} {\enquote {\bibinfo {title} {Statistical
  mechanics in external force fields},}\ }\href {\doibase
  10.1103/PhysRev.181.457} {\bibfield  {journal} {\bibinfo  {journal} {Phys.
  Rev.}\ }\textbf {\bibinfo {volume} {181}},\ \bibinfo {pages} {457--462}
  (\bibinfo {year} {1969})}\BibitemShut {NoStop}%
\bibitem [{\citenamefont {Deutsch}, \citenamefont {Dewitt},\ and\ \citenamefont
  {Furutani}(1979)}]{DeuDewFur-79}%
  \BibitemOpen
  \bibfield  {author} {\bibinfo {author} {\bibnamefont {Deutsch}, \bibfnamefont
  {C.}}, \bibinfo {author} {\bibnamefont {Dewitt}, \bibfnamefont {H.~E.}}, \
  and\ \bibinfo {author} {\bibnamefont {Furutani}, \bibfnamefont {Y.}},\
  }\bibfield  {title} {\enquote {\bibinfo {title} {Debye thermodynamics for the
  two-dimensional one-component plasma},}\ }\href {\doibase
  10.1103/PhysRevA.20.2631} {\bibfield  {journal} {\bibinfo  {journal} {Phys.
  Rev. A}\ }\textbf {\bibinfo {volume} {20}},\ \bibinfo {pages} {2631--2633}
  (\bibinfo {year} {1979})}\BibitemShut {NoStop}%
\bibitem [{\citenamefont {{Di Marino}}, \citenamefont {Lewin},\ and\
  \citenamefont {Nenna}(2022)}]{MarLewNen-22_ppt}%
  \BibitemOpen
  \bibfield  {author} {\bibinfo {author} {\bibnamefont {{Di Marino}},
  \bibfnamefont {S.}}, \bibinfo {author} {\bibnamefont {Lewin}, \bibfnamefont
  {M.}}, \ and\ \bibinfo {author} {\bibnamefont {Nenna}, \bibfnamefont {L.}},\
  }\bibfield  {title} {\enquote {\bibinfo {title} {Grand-canonical optimal
  transport},}\ }\href@noop {} {\bibfield  {journal} {\bibinfo  {journal}
  {ArXiv e-prints}\ } (\bibinfo {year} {2022})},\ \Eprint
  {http://arxiv.org/abs/2201.06859} {arXiv:2201.06859 [math.OC]} \BibitemShut
  {NoStop}%
\bibitem [{\citenamefont {Diananda}(1964)}]{Diananda-64}%
  \BibitemOpen
  \bibfield  {author} {\bibinfo {author} {\bibnamefont {Diananda},
  \bibfnamefont {P.~H.}},\ }\bibfield  {title} {\enquote {\bibinfo {title}
  {Notes on two lemmas concerning the {E}pstein zeta-function},}\ }\href
  {\doibase 10.1017/S2040618500035036} {\bibfield  {journal} {\bibinfo
  {journal} {Proc. Glasgow Math. Assoc.}\ }\textbf {\bibinfo {volume} {6}},\
  \bibinfo {pages} {202--204} (\bibinfo {year} {1964})}\BibitemShut {NoStop}%
\bibitem [{\citenamefont {Dinsmore}\ \emph {et~al.}(2002)\citenamefont
  {Dinsmore}, \citenamefont {Hsu}, \citenamefont {Nikolaides}, \citenamefont
  {Marquez}, \citenamefont {Bausch},\ and\ \citenamefont {Weitz}}]{Dinetal-02}%
  \BibitemOpen
  \bibfield  {author} {\bibinfo {author} {\bibnamefont {Dinsmore},
  \bibfnamefont {A.~D.}}, \bibinfo {author} {\bibnamefont {Hsu}, \bibfnamefont
  {M.~F.}}, \bibinfo {author} {\bibnamefont {Nikolaides}, \bibfnamefont
  {M.~G.}}, \bibinfo {author} {\bibnamefont {Marquez}, \bibfnamefont {M.}},
  \bibinfo {author} {\bibnamefont {Bausch}, \bibfnamefont {A.~R.}}, \ and\
  \bibinfo {author} {\bibnamefont {Weitz}, \bibfnamefont {D.~A.}},\ }\bibfield
  {title} {\enquote {\bibinfo {title} {Colloidosomes: Selectively permeable
  capsules composed of colloidal particles},}\ }\href {\doibase
  10.1126/science.1074868} {\bibfield  {journal} {\bibinfo  {journal}
  {Science}\ }\textbf {\bibinfo {volume} {298}},\ \bibinfo {pages} {1006--1009}
  (\bibinfo {year} {2002})}\BibitemShut {NoStop}%
\bibitem [{\citenamefont {Dobru\v{s}in}(1968{\natexlab{a}})}]{Dobrushin-68b}%
  \BibitemOpen
  \bibfield  {author} {\bibinfo {author} {\bibnamefont {Dobru\v{s}in},
  \bibfnamefont {R.}},\ }\bibfield  {title} {\enquote {\bibinfo {title} {The
  problem of uniqueness of a gibbsian random field and the problem of phase
  transitions},}\ }\href {\doibase 10.1007/BF01075682} {\bibfield  {journal}
  {\bibinfo  {journal} {Functional Analysis and Its Applications}\ }\textbf
  {\bibinfo {volume} {2}},\ \bibinfo {pages} {302--312} (\bibinfo {year}
  {1968}{\natexlab{a}})}\BibitemShut {NoStop}%
\bibitem [{\citenamefont {Dobru\v{s}in}(1968{\natexlab{b}})}]{Dobrushin-68c}%
  \BibitemOpen
  \bibfield  {author} {\bibinfo {author} {\bibnamefont {Dobru\v{s}in},
  \bibfnamefont {R.~L.}},\ }\bibfield  {title} {\enquote {\bibinfo {title}
  {Description of a random field by means of conditional probabilities and
  conditions for its regularity},}\ }\href@noop {} {\bibfield  {journal}
  {\bibinfo  {journal} {Teor. Verojatnost. i Primenen}\ }\textbf {\bibinfo
  {volume} {13}},\ \bibinfo {pages} {201--229} (\bibinfo {year}
  {1968}{\natexlab{b}})}\BibitemShut {NoStop}%
\bibitem [{\citenamefont {Dobru\v{s}in}(1968{\natexlab{c}})}]{Dobrushin-68a}%
  \BibitemOpen
  \bibfield  {author} {\bibinfo {author} {\bibnamefont {Dobru\v{s}in},
  \bibfnamefont {R.~L.}},\ }\bibfield  {title} {\enquote {\bibinfo {title}
  {Gibbsian random fields for lattice systems with pairwise interactions},}\
  }\href@noop {} {\bibfield  {journal} {\bibinfo  {journal} {Funkcional. Anal.
  i Prilo\v{z}en.}\ }\textbf {\bibinfo {volume} {2}},\ \bibinfo {pages}
  {31--43} (\bibinfo {year} {1968}{\natexlab{c}})}\BibitemShut {NoStop}%
\bibitem [{\citenamefont {Dobru\v{s}in}(1969)}]{Dobrushin-69}%
  \BibitemOpen
  \bibfield  {author} {\bibinfo {author} {\bibnamefont {Dobru\v{s}in},
  \bibfnamefont {R.~L.}},\ }\bibfield  {title} {\enquote {\bibinfo {title}
  {Gibbsian random fields. {G}eneral case},}\ }\href@noop {} {\bibfield
  {journal} {\bibinfo  {journal} {Funkcional. Anal. i Prilo\v{z}en}\ }\textbf
  {\bibinfo {volume} {3}},\ \bibinfo {pages} {27--35} (\bibinfo {year}
  {1969})}\BibitemShut {NoStop}%
\bibitem [{\citenamefont {Dobru\v{s}in}\ and\ \citenamefont
  {Minlos}(1967)}]{DobMin-67}%
  \BibitemOpen
  \bibfield  {author} {\bibinfo {author} {\bibnamefont {Dobru\v{s}in},
  \bibfnamefont {R.~L.}}\ and\ \bibinfo {author} {\bibnamefont {Minlos},
  \bibfnamefont {R.~A.}},\ }\bibfield  {title} {\enquote {\bibinfo {title}
  {Existence and continuity of pressure in classical statistical physics},}\
  }\href@noop {} {\bibfield  {journal} {\bibinfo  {journal} {Teor. Verojatnost.
  i Primenen.}\ }\textbf {\bibinfo {volume} {12}},\ \bibinfo {pages} {595--618}
  (\bibinfo {year} {1967})}\BibitemShut {NoStop}%
\bibitem [{\citenamefont {Dragnev}(2002)}]{Dragnev-02}%
  \BibitemOpen
  \bibfield  {author} {\bibinfo {author} {\bibnamefont {Dragnev}, \bibfnamefont
  {P.~D.}},\ }\bibfield  {title} {\enquote {\bibinfo {title} {On the separation
  of logarithmic points on the sphere},}\ }in\ \href@noop {} {\emph {\bibinfo
  {booktitle} {Approximation theory, {X} ({S}t. {L}ouis, {MO}, 2001)}}},\
  \bibinfo {series and number} {Innov. Appl. Math.}\ (\bibinfo  {publisher}
  {Vanderbilt Univ. Press, Nashville, TN},\ \bibinfo {year} {2002})\ pp.\
  \bibinfo {pages} {137--144}\BibitemShut {NoStop}%
\bibitem [{\citenamefont {Dragnev}\ and\ \citenamefont
  {Saff}(2007)}]{DraSaf-07}%
  \BibitemOpen
  \bibfield  {author} {\bibinfo {author} {\bibnamefont {Dragnev}, \bibfnamefont
  {P.~D.}}\ and\ \bibinfo {author} {\bibnamefont {Saff}, \bibfnamefont
  {E.~B.}},\ }\bibfield  {title} {\enquote {\bibinfo {title} {Riesz spherical
  potentials with external fields and minimal energy points separation},}\
  }\href {\doibase 10.1007/s11118-006-9032-2} {\bibfield  {journal} {\bibinfo
  {journal} {Potential Anal.}\ }\textbf {\bibinfo {volume} {26}},\ \bibinfo
  {pages} {139--162} (\bibinfo {year} {2007})}\BibitemShut {NoStop}%
\bibitem [{\citenamefont {Drummond}\ \emph {et~al.}(2004)\citenamefont
  {Drummond}, \citenamefont {Radnai}, \citenamefont {Trail}, \citenamefont
  {Towler},\ and\ \citenamefont {Needs}}]{DruRadTraTowNee-04}%
  \BibitemOpen
  \bibfield  {author} {\bibinfo {author} {\bibnamefont {Drummond},
  \bibfnamefont {N.~D.}}, \bibinfo {author} {\bibnamefont {Radnai},
  \bibfnamefont {Z.}}, \bibinfo {author} {\bibnamefont {Trail}, \bibfnamefont
  {J.~R.}}, \bibinfo {author} {\bibnamefont {Towler}, \bibfnamefont {M.~D.}}, \
  and\ \bibinfo {author} {\bibnamefont {Needs}, \bibfnamefont {R.~J.}},\
  }\bibfield  {title} {\enquote {\bibinfo {title} {Diffusion quantum {M}onte
  {C}arlo study of three-dimensional {W}igner crystals},}\ }\href {\doibase
  10.1103/PhysRevB.69.085116} {\bibfield  {journal} {\bibinfo  {journal} {Phys.
  Rev. B}\ }\textbf {\bibinfo {volume} {69}},\ \bibinfo {pages} {085116}
  (\bibinfo {year} {2004})}\BibitemShut {NoStop}%
\bibitem [{\citenamefont {Ducatez}(2018)}]{Ducatez-18_ppt}%
  \BibitemOpen
  \bibfield  {author} {\bibinfo {author} {\bibnamefont {Ducatez}, \bibfnamefont
  {R.}},\ }\href@noop {} {\enquote {\bibinfo {title} {Analysis of the one
  dimensional inhomogeneous {J}ellium model with the {B}irkhoff-{H}opf
  theorem},}\ } (\bibinfo {year} {2018}),\ \Eprint
  {http://arxiv.org/abs/1806.07681} {arXiv:1806.07681 [math.SP]} \BibitemShut
  {NoStop}%
\bibitem [{\citenamefont {Dumitriu}\ and\ \citenamefont
  {Edelman}(2002)}]{DumEde-02}%
  \BibitemOpen
  \bibfield  {author} {\bibinfo {author} {\bibnamefont {Dumitriu},
  \bibfnamefont {I.}}\ and\ \bibinfo {author} {\bibnamefont {Edelman},
  \bibfnamefont {A.}},\ }\bibfield  {title} {\enquote {\bibinfo {title} {Matrix
  models for beta ensembles},}\ }\href {\doibase 10.1063/1.1507823} {\bibfield
  {journal} {\bibinfo  {journal} {J. Math. Phys.}\ }\textbf {\bibinfo {volume}
  {43}},\ \bibinfo {pages} {5830--5847} (\bibinfo {year} {2002})}\BibitemShut
  {NoStop}%
\bibitem [{\citenamefont {Duneau}\ and\ \citenamefont
  {Soillard}(1976)}]{DunSou-76}%
  \BibitemOpen
  \bibfield  {author} {\bibinfo {author} {\bibnamefont {Duneau}, \bibfnamefont
  {M.}}\ and\ \bibinfo {author} {\bibnamefont {Soillard}, \bibfnamefont {B.}},\
  }\bibfield  {title} {\enquote {\bibinfo {title} {Cluster properties of
  lattice and continuous systems},}\ }\href
  {http://projecteuclid.org/euclid.cmp/1103899725} {\bibfield  {journal}
  {\bibinfo  {journal} {Comm. Math. Phys.}\ }\textbf {\bibinfo {volume} {47}},\
  \bibinfo {pages} {155--166} (\bibinfo {year} {1976})}\BibitemShut {NoStop}%
\bibitem [{\citenamefont {Dyson}(1962{\natexlab{a}})}]{Dyson-62a}%
  \BibitemOpen
  \bibfield  {author} {\bibinfo {author} {\bibnamefont {Dyson}, \bibfnamefont
  {F.~J.}},\ }\bibfield  {title} {\enquote {\bibinfo {title} {Statistical
  theory of the energy levels of complex systems. {I}},}\ }\href {\doibase
  10.1063/1.1703773} {\bibfield  {journal} {\bibinfo  {journal} {J. Math.
  Phys.}\ }\textbf {\bibinfo {volume} {3}},\ \bibinfo {pages} {140--156}
  (\bibinfo {year} {1962}{\natexlab{a}})}\BibitemShut {NoStop}%
\bibitem [{\citenamefont {Dyson}(1962{\natexlab{b}})}]{Dyson-62b}%
  \BibitemOpen
  \bibfield  {author} {\bibinfo {author} {\bibnamefont {Dyson}, \bibfnamefont
  {F.~J.}},\ }\bibfield  {title} {\enquote {\bibinfo {title} {Statistical
  theory of the energy levels of complex systems. {II}},}\ }\href {\doibase
  10.1063/1.1703774} {\bibfield  {journal} {\bibinfo  {journal} {J. Math.
  Phys.}\ }\textbf {\bibinfo {volume} {3}},\ \bibinfo {pages} {157--165}
  (\bibinfo {year} {1962}{\natexlab{b}})}\BibitemShut {NoStop}%
\bibitem [{\citenamefont {Dyson}(1962{\natexlab{c}})}]{Dyson-62c}%
  \BibitemOpen
  \bibfield  {author} {\bibinfo {author} {\bibnamefont {Dyson}, \bibfnamefont
  {F.~J.}},\ }\bibfield  {title} {\enquote {\bibinfo {title} {Statistical
  theory of the energy levels of complex systems. {III}},}\ }\href {\doibase
  10.1063/1.1703775} {\bibfield  {journal} {\bibinfo  {journal} {J. Math.
  Phys.}\ }\textbf {\bibinfo {volume} {3}},\ \bibinfo {pages} {166--175}
  (\bibinfo {year} {1962}{\natexlab{c}})}\BibitemShut {NoStop}%
\bibitem [{\citenamefont {Dyson}(1967)}]{Dyson-67}%
  \BibitemOpen
  \bibfield  {author} {\bibinfo {author} {\bibnamefont {Dyson}, \bibfnamefont
  {F.~J.}},\ }\bibfield  {title} {\enquote {\bibinfo {title} {Ground-state
  energy of a finite system of charged particles},}\ }\href {\doibase
  10.1063/1.1705389} {\bibfield  {journal} {\bibinfo  {journal} {J. Math.
  Phys.}\ }\textbf {\bibinfo {volume} {8}},\ \bibinfo {pages} {1538--1545}
  (\bibinfo {year} {1967})}\BibitemShut {NoStop}%
\bibitem [{\citenamefont {Dyson}(1969)}]{Dyson-69}%
  \BibitemOpen
  \bibfield  {author} {\bibinfo {author} {\bibnamefont {Dyson}, \bibfnamefont
  {F.~J.}},\ }\bibfield  {title} {\enquote {\bibinfo {title} {Existence of a
  phase-transition in a one-dimensional {I}sing ferromagnet},}\ }\href
  {\doibase 10.1007/BF01645907} {\bibfield  {journal} {\bibinfo  {journal}
  {Comm. Math. Phys.}\ }\textbf {\bibinfo {volume} {12}},\ \bibinfo {pages}
  {91--107} (\bibinfo {year} {1969})}\BibitemShut {NoStop}%
\bibitem [{\citenamefont {Dyson}(1971)}]{Dyson-71}%
  \BibitemOpen
  \bibfield  {author} {\bibinfo {author} {\bibnamefont {Dyson}, \bibfnamefont
  {F.~J.}},\ }\bibfield  {title} {\enquote {\bibinfo {title} {Chemical binding
  in classical {C}oulomb lattices},}\ }\href {\doibase
  https://doi.org/10.1016/0003-4916(71)90294-6} {\bibfield  {journal} {\bibinfo
   {journal} {Annals of Physics}\ }\textbf {\bibinfo {volume} {63}},\ \bibinfo
  {pages} {1--11} (\bibinfo {year} {1971})}\BibitemShut {NoStop}%
\bibitem [{\citenamefont {Dyson}\ and\ \citenamefont
  {Lenard}(1967)}]{DysLen-67}%
  \BibitemOpen
  \bibfield  {author} {\bibinfo {author} {\bibnamefont {Dyson}, \bibfnamefont
  {F.~J.}}\ and\ \bibinfo {author} {\bibnamefont {Lenard}, \bibfnamefont
  {A.}},\ }\bibfield  {title} {\enquote {\bibinfo {title} {Stability of matter.
  {I}},}\ }\href {\doibase 10.1063/1.1705209} {\bibfield  {journal} {\bibinfo
  {journal} {J. Math. Phys.}\ }\textbf {\bibinfo {volume} {8}},\ \bibinfo
  {pages} {423--434} (\bibinfo {year} {1967})}\BibitemShut {NoStop}%
\bibitem [{\citenamefont {Dyson}\ and\ \citenamefont
  {Mehta}(1963)}]{DysMeh-63a}%
  \BibitemOpen
  \bibfield  {author} {\bibinfo {author} {\bibnamefont {Dyson}, \bibfnamefont
  {F.~J.}}\ and\ \bibinfo {author} {\bibnamefont {Mehta}, \bibfnamefont
  {M.~L.}},\ }\bibfield  {title} {\enquote {\bibinfo {title} {Statistical
  theory of the energy levels of complex systems. {IV}},}\ }\href {\doibase
  10.1063/1.1704008} {\bibfield  {journal} {\bibinfo  {journal} {J. Math.
  Phys.}\ }\textbf {\bibinfo {volume} {4}},\ \bibinfo {pages} {701--712}
  (\bibinfo {year} {1963})}\BibitemShut {NoStop}%
\bibitem [{\citenamefont {Edwards}\ and\ \citenamefont
  {Lenard}(1962)}]{EdwLen-62}%
  \BibitemOpen
  \bibfield  {author} {\bibinfo {author} {\bibnamefont {Edwards}, \bibfnamefont
  {S.~F.}}\ and\ \bibinfo {author} {\bibnamefont {Lenard}, \bibfnamefont
  {A.}},\ }\bibfield  {title} {\enquote {\bibinfo {title} {Exact statistical
  mechanics of a one-dimensional system with {C}oulomb forces. {II}. {T}he
  method of functional integration},}\ }\href
  {https://doi.org/10.1063/1.1724281} {\bibfield  {journal} {\bibinfo
  {journal} {J. Mathematical Phys.}\ }\textbf {\bibinfo {volume} {3}},\
  \bibinfo {pages} {778--792} (\bibinfo {year} {1962})}\BibitemShut {NoStop}%
\bibitem [{\citenamefont {Ellis}(1985)}]{Ellis-85}%
  \BibitemOpen
  \bibfield  {author} {\bibinfo {author} {\bibnamefont {Ellis}, \bibfnamefont
  {R.~S.}},\ }\href {\doibase 10.1007/978-1-4613-8533-2} {\emph {\bibinfo
  {title} {Entropy, large deviations, and statistical mechanics}}},\ \bibinfo
  {series} {Grundlehren der mathematischen Wissenschaften [Fundamental
  Principles of Mathematical Sciences]}, Vol.\ \bibinfo {volume} {271}\
  (\bibinfo  {publisher} {Springer-Verlag, New York},\ \bibinfo {year} {1985})\
  pp.\ \bibinfo {pages} {xiv+364}\BibitemShut {NoStop}%
\bibitem [{\citenamefont {Emersleben}(1923)}]{Emersleben-23}%
  \BibitemOpen
  \bibfield  {author} {\bibinfo {author} {\bibnamefont {Emersleben},},\
  }\bibfield  {title} {\enquote {\bibinfo {title} {{Zetafunktionen und
  elektrostatische Gitterpotentiale}},}\ }\href@noop {} {\bibfield  {journal}
  {\bibinfo  {journal} {Phys. Zeit.}\ }\textbf {\bibinfo {volume} {24}},\
  \bibinfo {pages} {73--97} (\bibinfo {year} {1923})}\BibitemShut {NoStop}%
\bibitem [{\citenamefont {Ennola}(1964{\natexlab{a}})}]{Ennola-64}%
  \BibitemOpen
  \bibfield  {author} {\bibinfo {author} {\bibnamefont {Ennola}, \bibfnamefont
  {V.}},\ }\bibfield  {title} {\enquote {\bibinfo {title} {A lemma about the
  {E}pstein zeta-function},}\ }\href {\doibase 10.1017/S2040618500035024}
  {\bibfield  {journal} {\bibinfo  {journal} {Proc. Glasgow Math. Assoc.}\
  }\textbf {\bibinfo {volume} {6}},\ \bibinfo {pages} {198--201} (\bibinfo
  {year} {1964}{\natexlab{a}})}\BibitemShut {NoStop}%
\bibitem [{\citenamefont {Ennola}(1964{\natexlab{b}})}]{Ennola-64b}%
  \BibitemOpen
  \bibfield  {author} {\bibinfo {author} {\bibnamefont {Ennola}, \bibfnamefont
  {V.}},\ }\bibfield  {title} {\enquote {\bibinfo {title} {On a problem about
  the {E}pstein zeta-function},}\ }\href@noop {} {\bibfield  {journal}
  {\bibinfo  {journal} {Proc. Cambridge Philos. Soc.}\ }\textbf {\bibinfo
  {volume} {60}},\ \bibinfo {pages} {855--875} (\bibinfo {year}
  {1964}{\natexlab{b}})}\BibitemShut {NoStop}%
\bibitem [{\citenamefont {Epstein}(1906)}]{Epstein-06}%
  \BibitemOpen
  \bibfield  {author} {\bibinfo {author} {\bibnamefont {Epstein}, \bibfnamefont
  {P.}},\ }\bibfield  {title} {\enquote {\bibinfo {title} {Zur {T}heorie
  allgemeiner {Z}etafunktionen. {II}},}\ }\href {\doibase 10.1007/BF01449900}
  {\bibfield  {journal} {\bibinfo  {journal} {Math. Ann.}\ }\textbf {\bibinfo
  {volume} {63}},\ \bibinfo {pages} {205--216} (\bibinfo {year}
  {1906})}\BibitemShut {NoStop}%
\bibitem [{\citenamefont {Erbar}, \citenamefont {Huesmann},\ and\ \citenamefont
  {Lebl\'{e}}(2021)}]{ErbHueLeb-21}%
  \BibitemOpen
  \bibfield  {author} {\bibinfo {author} {\bibnamefont {Erbar}, \bibfnamefont
  {M.}}, \bibinfo {author} {\bibnamefont {Huesmann}, \bibfnamefont {M.}}, \
  and\ \bibinfo {author} {\bibnamefont {Lebl\'{e}}, \bibfnamefont {T.}},\
  }\bibfield  {title} {\enquote {\bibinfo {title} {The one-dimensional log-gas
  free energy has a unique minimizer},}\ }\href {\doibase 10.1002/cpa.21977}
  {\bibfield  {journal} {\bibinfo  {journal} {Comm. Pure Appl. Math.}\ }\textbf
  {\bibinfo {volume} {74}},\ \bibinfo {pages} {615--675} (\bibinfo {year}
  {2021})}\BibitemShut {NoStop}%
\bibitem [{\citenamefont {Erd\H{o}s}(2013)}]{Erdos-13}%
  \BibitemOpen
  \bibfield  {author} {\bibinfo {author} {\bibnamefont {Erd\H{o}s},
  \bibfnamefont {L.}},\ }\bibfield  {title} {\enquote {\bibinfo {title}
  {Universality for random matrices and log-gases},}\ }in\ \href {\doibase
  10.4310/CDM.2012.v2012.n1.a2} {\emph {\bibinfo {booktitle} {Current
  developments in mathematics 2012}}}\ (\bibinfo  {publisher} {Int. Press,
  Somerville, MA},\ \bibinfo {year} {2013})\ pp.\ \bibinfo {pages}
  {59--132}\BibitemShut {NoStop}%
\bibitem [{\citenamefont {Erd\H{o}s}\ \emph {et~al.}(2010)\citenamefont
  {Erd\H{o}s}, \citenamefont {Ram\'{\i}rez}, \citenamefont {Schlein},
  \citenamefont {Tao}, \citenamefont {Vu},\ and\ \citenamefont
  {Yau}}]{ErdRamSchTaoVuYau-10}%
  \BibitemOpen
  \bibfield  {author} {\bibinfo {author} {\bibnamefont {Erd\H{o}s},
  \bibfnamefont {L.}}, \bibinfo {author} {\bibnamefont {Ram\'{\i}rez},
  \bibfnamefont {J.}}, \bibinfo {author} {\bibnamefont {Schlein}, \bibfnamefont
  {B.}}, \bibinfo {author} {\bibnamefont {Tao}, \bibfnamefont {T.}}, \bibinfo
  {author} {\bibnamefont {Vu}, \bibfnamefont {V.}}, \ and\ \bibinfo {author}
  {\bibnamefont {Yau}, \bibfnamefont {H.-T.}},\ }\bibfield  {title} {\enquote
  {\bibinfo {title} {Bulk universality for {W}igner {H}ermitian matrices with
  subexponential decay},}\ }\href {\doibase 10.4310/MRL.2010.v17.n4.a7}
  {\bibfield  {journal} {\bibinfo  {journal} {Math. Res. Lett.}\ }\textbf
  {\bibinfo {volume} {17}},\ \bibinfo {pages} {667--674} (\bibinfo {year}
  {2010})}\BibitemShut {NoStop}%
\bibitem [{\citenamefont {Erd\H{o}s}\ and\ \citenamefont
  {Yau}(2012)}]{ErdYau-12}%
  \BibitemOpen
  \bibfield  {author} {\bibinfo {author} {\bibnamefont {Erd\H{o}s},
  \bibfnamefont {L.}}\ and\ \bibinfo {author} {\bibnamefont {Yau},
  \bibfnamefont {H.-T.}},\ }\bibfield  {title} {\enquote {\bibinfo {title}
  {Universality of local spectral statistics of random matrices},}\ }\href
  {\doibase 10.1090/S0273-0979-2012-01372-1} {\bibfield  {journal} {\bibinfo
  {journal} {Bull. Amer. Math. Soc. (N.S.)}\ }\textbf {\bibinfo {volume}
  {49}},\ \bibinfo {pages} {377--414} (\bibinfo {year} {2012})}\BibitemShut
  {NoStop}%
\bibitem [{\citenamefont {Erd\H{o}s}\ and\ \citenamefont
  {Yau}(2015)}]{ErdYau-15}%
  \BibitemOpen
  \bibfield  {author} {\bibinfo {author} {\bibnamefont {Erd\H{o}s},
  \bibfnamefont {L.}}\ and\ \bibinfo {author} {\bibnamefont {Yau},
  \bibfnamefont {H.-T.}},\ }\bibfield  {title} {\enquote {\bibinfo {title} {Gap
  universality of generalized {W}igner and {$\beta$}-ensembles},}\ }\href
  {\doibase 10.4171/JEMS/548} {\bibfield  {journal} {\bibinfo  {journal} {J.
  Eur. Math. Soc. (JEMS)}\ }\textbf {\bibinfo {volume} {17}},\ \bibinfo {pages}
  {1927--2036} (\bibinfo {year} {2015})}\BibitemShut {NoStop}%
\bibitem [{\citenamefont {Erd\H{o}s}\ and\ \citenamefont
  {Yau}(2017)}]{ErdYau-17}%
  \BibitemOpen
  \bibfield  {author} {\bibinfo {author} {\bibnamefont {Erd\H{o}s},
  \bibfnamefont {L.}}\ and\ \bibinfo {author} {\bibnamefont {Yau},
  \bibfnamefont {H.-T.}},\ }\href@noop {} {\emph {\bibinfo {title} {A dynamical
  approach to random matrix theory}}},\ \bibinfo {series} {Courant Lecture
  Notes in Mathematics}, Vol.~\bibinfo {volume} {28}\ (\bibinfo  {publisher}
  {Courant Institute of Mathematical Sciences, New York; American Mathematical
  Society, Providence, RI},\ \bibinfo {year} {2017})\ pp.\ \bibinfo {pages}
  {ix+226}\BibitemShut {NoStop}%
\bibitem [{\citenamefont {Erd{\"o}s}, \citenamefont {Schlein},\ and\
  \citenamefont {Yau}(2011)}]{ErdSchYau-11}%
  \BibitemOpen
  \bibfield  {author} {\bibinfo {author} {\bibnamefont {Erd{\"o}s},
  \bibfnamefont {L.}}, \bibinfo {author} {\bibnamefont {Schlein}, \bibfnamefont
  {B.}}, \ and\ \bibinfo {author} {\bibnamefont {Yau}, \bibfnamefont {H.-T.}},\
  }\bibfield  {title} {\enquote {\bibinfo {title} {Universality of random
  matrices and local relaxation flow},}\ }\href {\doibase
  10.1007/s00222-010-0302-7} {\bibfield  {journal} {\bibinfo  {journal}
  {Invent. Math.}\ }\textbf {\bibinfo {volume} {185}},\ \bibinfo {pages}
  {75--119} (\bibinfo {year} {2011})}\BibitemShut {NoStop}%
\bibitem [{\citenamefont {Ewald}(1921)}]{Ewald-21}%
  \BibitemOpen
  \bibfield  {author} {\bibinfo {author} {\bibnamefont {Ewald}, \bibfnamefont
  {P.~P.}},\ }\bibfield  {title} {\enquote {\bibinfo {title} {Die berechnung
  optischer und elektrostatischer gitterpotentiale},}\ }\href {\doibase
  10.1002/andp.19213690304} {\bibfield  {journal} {\bibinfo  {journal} {Ann.
  Phys.}\ }\textbf {\bibinfo {volume} {369}},\ \bibinfo {pages} {253--287}
  (\bibinfo {year} {1921})}\BibitemShut {NoStop}%
\bibitem [{\citenamefont {Fabes}, \citenamefont {Kenig},\ and\ \citenamefont
  {Serapioni}(1982)}]{FabKenSer-82}%
  \BibitemOpen
  \bibfield  {author} {\bibinfo {author} {\bibnamefont {Fabes}, \bibfnamefont
  {E.~B.}}, \bibinfo {author} {\bibnamefont {Kenig}, \bibfnamefont {C.~E.}}, \
  and\ \bibinfo {author} {\bibnamefont {Serapioni}, \bibfnamefont {R.~P.}},\
  }\bibfield  {title} {\enquote {\bibinfo {title} {The local regularity of
  solutions of degenerate elliptic equations},}\ }\href {\doibase
  10.1080/03605308208820218} {\bibfield  {journal} {\bibinfo  {journal} {Comm.
  Partial Differential Equations}\ }\textbf {\bibinfo {volume} {7}},\ \bibinfo
  {pages} {77--116} (\bibinfo {year} {1982})}\BibitemShut {NoStop}%
\bibitem [{\citenamefont {Falco}(2012)}]{Falco-12}%
  \BibitemOpen
  \bibfield  {author} {\bibinfo {author} {\bibnamefont {Falco}, \bibfnamefont
  {P.}},\ }\bibfield  {title} {\enquote {\bibinfo {title}
  {Kosterlitz-{T}houless transition line for the two dimensional {C}oulomb
  gas},}\ }\href {\doibase 10.1007/s00220-012-1454-7} {\bibfield  {journal}
  {\bibinfo  {journal} {Comm. Math. Phys.}\ }\textbf {\bibinfo {volume}
  {312}},\ \bibinfo {pages} {559--609} (\bibinfo {year} {2012})}\BibitemShut
  {NoStop}%
\bibitem [{\citenamefont {Fantoni}, \citenamefont {Jancovici},\ and\
  \citenamefont {T\'{e}llez}(2003)}]{FanJanTel-03}%
  \BibitemOpen
  \bibfield  {author} {\bibinfo {author} {\bibnamefont {Fantoni}, \bibfnamefont
  {R.}}, \bibinfo {author} {\bibnamefont {Jancovici}, \bibfnamefont {B.}}, \
  and\ \bibinfo {author} {\bibnamefont {T\'{e}llez}, \bibfnamefont {G.}},\
  }\bibfield  {title} {\enquote {\bibinfo {title} {Pressures for a
  one-component plasma on a pseudosphere},}\ }\href {\doibase
  10.1023/A:1023671419021} {\bibfield  {journal} {\bibinfo  {journal} {J.
  Statist. Phys.}\ }\textbf {\bibinfo {volume} {112}},\ \bibinfo {pages}
  {27--57} (\bibinfo {year} {2003})}\BibitemShut {NoStop}%
\bibitem [{\citenamefont {Federbush}\ and\ \citenamefont
  {Kennedy}(1985)}]{FedKen-85}%
  \BibitemOpen
  \bibfield  {author} {\bibinfo {author} {\bibnamefont {Federbush},
  \bibfnamefont {P.}}\ and\ \bibinfo {author} {\bibnamefont {Kennedy},
  \bibfnamefont {T.}},\ }\bibfield  {title} {\enquote {\bibinfo {title}
  {Surface effects in {D}ebye screening},}\ }\href
  {http://projecteuclid.org/euclid.cmp/1104114462} {\bibfield  {journal}
  {\bibinfo  {journal} {Comm. Math. Phys.}\ }\textbf {\bibinfo {volume}
  {102}},\ \bibinfo {pages} {361--423} (\bibinfo {year} {1985})}\BibitemShut
  {NoStop}%
\bibitem [{\citenamefont {Fefferman}(1985)}]{Fefferman-85}%
  \BibitemOpen
  \bibfield  {author} {\bibinfo {author} {\bibnamefont {Fefferman},
  \bibfnamefont {C.}},\ }\bibfield  {title} {\enquote {\bibinfo {title} {The
  thermodynamic limit for a crystal},}\ }\href {\doibase 10.1007/BF01205785}
  {\bibfield  {journal} {\bibinfo  {journal} {Commun. Math. Phys.}\ }\textbf
  {\bibinfo {volume} {98}},\ \bibinfo {pages} {289--311} (\bibinfo {year}
  {1985})}\BibitemShut {NoStop}%
\bibitem [{\citenamefont {Fejes~T\'{o}th}(1959)}]{Toth-59}%
  \BibitemOpen
  \bibfield  {author} {\bibinfo {author} {\bibnamefont {Fejes~T\'{o}th},
  \bibfnamefont {L.}},\ }\bibfield  {title} {\enquote {\bibinfo {title}
  {\"{U}ber eine {P}unktverteilung auf der {K}ugel},}\ }\href {\doibase
  10.1007/BF02063286} {\bibfield  {journal} {\bibinfo  {journal} {Acta Math.
  Acad. Sci. Hungar.}\ }\textbf {\bibinfo {volume} {10}},\ \bibinfo {pages}
  {13--19 (unbound insert)} (\bibinfo {year} {1959})}\BibitemShut {NoStop}%
\bibitem [{\citenamefont {Fekete}(1923)}]{Fekete-23}%
  \BibitemOpen
  \bibfield  {author} {\bibinfo {author} {\bibnamefont {Fekete}, \bibfnamefont
  {M.}},\ }\bibfield  {title} {\enquote {\bibinfo {title} {{\"U}ber die
  {V}erteilung der {W}urzeln bei gewissen algebraischen {G}leichungen mit
  ganzzahligen {K}oeffizienten},}\ }\href {\doibase 10.1007/BF01504345}
  {\bibfield  {journal} {\bibinfo  {journal} {Math. Z.}\ }\textbf {\bibinfo
  {volume} {17}},\ \bibinfo {pages} {228--249} (\bibinfo {year}
  {1923})}\BibitemShut {NoStop}%
\bibitem [{\citenamefont {Fields}(1980)}]{Fields-80}%
  \BibitemOpen
  \bibfield  {author} {\bibinfo {author} {\bibnamefont {Fields}, \bibfnamefont
  {K.}},\ }\bibfield  {title} {\enquote {\bibinfo {title} {Locally minimal
  {E}pstein zeta functions},}\ }\href {\doibase 10.1112/S002557930000989X}
  {\bibfield  {journal} {\bibinfo  {journal} {Mathematika}\ }\textbf {\bibinfo
  {volume} {27}},\ \bibinfo {pages} {17--24} (\bibinfo {year}
  {1980})}\BibitemShut {NoStop}%
\bibitem [{\citenamefont {Fisher}(1964)}]{Fisher-64}%
  \BibitemOpen
  \bibfield  {author} {\bibinfo {author} {\bibnamefont {Fisher}, \bibfnamefont
  {M.~E.}},\ }\bibfield  {title} {\enquote {\bibinfo {title} {The free energy
  of a macroscopic system},}\ }\href {\doibase 10.1007/BF00250473} {\bibfield
  {journal} {\bibinfo  {journal} {Arch. Ration. Mech. Anal.}\ }\textbf
  {\bibinfo {volume} {17}},\ \bibinfo {pages} {377--410} (\bibinfo {year}
  {1964})}\BibitemShut {NoStop}%
\bibitem [{\citenamefont {Fisher}\ and\ \citenamefont
  {Lebowitz}(1970)}]{FisLeb-70}%
  \BibitemOpen
  \bibfield  {author} {\bibinfo {author} {\bibnamefont {Fisher}, \bibfnamefont
  {M.~E.}}\ and\ \bibinfo {author} {\bibnamefont {Lebowitz}, \bibfnamefont
  {J.~L.}},\ }\bibfield  {title} {\enquote {\bibinfo {title} {Asymptotic free
  energy of a system with periodic boundary conditions},}\ }\href
  {http://projecteuclid.org/euclid.cmp/1103842739} {\bibfield  {journal}
  {\bibinfo  {journal} {Comm. Math. Phys.}\ }\textbf {\bibinfo {volume} {19}},\
  \bibinfo {pages} {251--272} (\bibinfo {year} {1970})}\BibitemShut {NoStop}%
\bibitem [{\citenamefont {Fisher}\ and\ \citenamefont
  {Ruelle}(1966)}]{FisRue-66}%
  \BibitemOpen
  \bibfield  {author} {\bibinfo {author} {\bibnamefont {Fisher}, \bibfnamefont
  {M.~E.}}\ and\ \bibinfo {author} {\bibnamefont {Ruelle}, \bibfnamefont
  {D.}},\ }\bibfield  {title} {\enquote {\bibinfo {title} {The stability of
  many-particle systems},}\ }\href {\doibase 10.1063/1.1704928} {\bibfield
  {journal} {\bibinfo  {journal} {J. Math. Phys.}\ }\textbf {\bibinfo {volume}
  {7}},\ \bibinfo {pages} {260--270} (\bibinfo {year} {1966})}\BibitemShut
  {NoStop}%
\bibitem [{\citenamefont {Flack}, \citenamefont {Majumdar},\ and\ \citenamefont
  {Schehr}(2022)}]{FlaMajSch-22_ppt}%
  \BibitemOpen
  \bibfield  {author} {\bibinfo {author} {\bibnamefont {Flack}, \bibfnamefont
  {A.}}, \bibinfo {author} {\bibnamefont {Majumdar}, \bibfnamefont {S.~N.}}, \
  and\ \bibinfo {author} {\bibnamefont {Schehr}, \bibfnamefont {G.}},\
  }\bibfield  {title} {\enquote {\bibinfo {title} {Gap probability and full
  counting statistics in the one dimensional one-component plasma},}\ }\href
  {https://arxiv.org/abs/2202.12118} {\bibfield  {journal} {\bibinfo  {journal}
  {ArXiV e-prints}\ } (\bibinfo {year} {2022})}\BibitemShut {NoStop}%
\bibitem [{\citenamefont {Fontaine}\ and\ \citenamefont
  {Martin}(1984)}]{FonMar-84}%
  \BibitemOpen
  \bibfield  {author} {\bibinfo {author} {\bibnamefont {Fontaine},
  \bibfnamefont {J.-R.}}\ and\ \bibinfo {author} {\bibnamefont {Martin},
  \bibfnamefont {P.~A.}},\ }\bibfield  {title} {\enquote {\bibinfo {title}
  {Equilibrium equations and symmetries of classical {C}oulomb systems},}\
  }\href {https://doi.org/10.1007/BF01015731} {\bibfield  {journal} {\bibinfo
  {journal} {J. Statist. Phys.}\ }\textbf {\bibinfo {volume} {36}},\ \bibinfo
  {pages} {163--179} (\bibinfo {year} {1984})}\BibitemShut {NoStop}%
\bibitem [{\citenamefont {F{\"o}ppl}(1912)}]{Foppl-12}%
  \BibitemOpen
  \bibfield  {author} {\bibinfo {author} {\bibnamefont {F{\"o}ppl},
  \bibfnamefont {L.}},\ }\bibfield  {title} {\enquote {\bibinfo {title}
  {Stabile {A}nordnungen von {E}lektronen im {A}tom},}\ }\href {\doibase
  10.1515/crll.1912.141.251} {\bibfield  {journal} {\bibinfo  {journal} {J.
  Reine Angew. Math.}\ }\textbf {\bibinfo {volume} {141}},\ \bibinfo {pages}
  {251--302} (\bibinfo {year} {1912})}\BibitemShut {NoStop}%
\bibitem [{\citenamefont {Forrester}(1993)}]{Forrester-93}%
  \BibitemOpen
  \bibfield  {author} {\bibinfo {author} {\bibnamefont {Forrester},
  \bibfnamefont {P.}},\ }\bibfield  {title} {\enquote {\bibinfo {title} {Exact
  integral formulas and asymptotics for the correlations in the $1/r^2$ quantum
  many body system},}\ }\href {\doibase 10.1016/0375-9601(93)90661-I}
  {\bibfield  {journal} {\bibinfo  {journal} {Phys. Lett. A}\ }\textbf
  {\bibinfo {volume} {179}},\ \bibinfo {pages} {127--130} (\bibinfo {year}
  {1993})}\BibitemShut {NoStop}%
\bibitem [{\citenamefont {Forrester}(1984)}]{Forrester-84}%
  \BibitemOpen
  \bibfield  {author} {\bibinfo {author} {\bibnamefont {Forrester},
  \bibfnamefont {P.~J.}},\ }\bibfield  {title} {\enquote {\bibinfo {title}
  {Analogues between a quantum many body problem and the log-gas},}\ }\href
  {http://stacks.iop.org/0305-4470/17/i=10/a=018} {\bibfield  {journal}
  {\bibinfo  {journal} {J. Phys. A}\ }\textbf {\bibinfo {volume} {17}},\
  \bibinfo {pages} {2059} (\bibinfo {year} {1984})}\BibitemShut {NoStop}%
\bibitem [{\citenamefont {Forrester}(1998)}]{Forrester-98}%
  \BibitemOpen
  \bibfield  {author} {\bibinfo {author} {\bibnamefont {Forrester},
  \bibfnamefont {P.~J.}},\ }\bibfield  {title} {\enquote {\bibinfo {title}
  {Exact results for two-dimensional {C}oulomb systems},}\ \ }(\bibinfo {year}
  {1998})\ pp.\ \bibinfo {pages} {235--270},\ \bibinfo {note} {fundamental
  problems in statistical mechanics (Altenberg, 1997)}\BibitemShut {NoStop}%
\bibitem [{\citenamefont {Forrester}(2010)}]{Forrester-10}%
  \BibitemOpen
  \bibfield  {author} {\bibinfo {author} {\bibnamefont {Forrester},
  \bibfnamefont {P.~J.}},\ }\href@noop {} {\emph {\bibinfo {title} {Log-gases
  and random matrices}}},\ \bibinfo {series} {London Mathematical Society
  Monographs Series}, Vol.~\bibinfo {volume} {34}\ (\bibinfo  {publisher}
  {Princeton University Press, Princeton, NJ},\ \bibinfo {year} {2010})\ pp.\
  \bibinfo {pages} {xiv+791}\BibitemShut {NoStop}%
\bibitem [{\citenamefont {Forrester}(2022)}]{Forrester-22_ppt}%
  \BibitemOpen
  \bibfield  {author} {\bibinfo {author} {\bibnamefont {Forrester},
  \bibfnamefont {P.~J.}},\ }\bibfield  {title} {\enquote {\bibinfo {title} {A
  review of exact results for fluctuation formulas in random matrix theory},}\
  }\href {https://arxiv.org/abs/2204.03303} {\bibfield  {journal} {\bibinfo
  {journal} {ArXiV e-prints}\ } (\bibinfo {year} {2022})}\BibitemShut {NoStop}%
\bibitem [{\citenamefont {Fr{\"o}hlich}(1976)}]{Frohlich-76}%
  \BibitemOpen
  \bibfield  {author} {\bibinfo {author} {\bibnamefont {Fr{\"o}hlich},
  \bibfnamefont {J.}},\ }\bibfield  {title} {\enquote {\bibinfo {title}
  {Classical and quantum statistical mechanics in one and two dimensions:
  two-component {Y}ukawa- and {C}oulomb systems},}\ }\href
  {http://projecteuclid.org/euclid.cmp/1103899760} {\bibfield  {journal}
  {\bibinfo  {journal} {Comm. Math. Phys.}\ }\textbf {\bibinfo {volume} {47}},\
  \bibinfo {pages} {233--268} (\bibinfo {year} {1976})}\BibitemShut {NoStop}%
\bibitem [{\citenamefont {Fr{\"o}hlich}\ and\ \citenamefont
  {Park}(1978)}]{FroPar-78}%
  \BibitemOpen
  \bibfield  {author} {\bibinfo {author} {\bibnamefont {Fr{\"o}hlich},
  \bibfnamefont {J.}}\ and\ \bibinfo {author} {\bibnamefont {Park},
  \bibfnamefont {Y.~M.}},\ }\bibfield  {title} {\enquote {\bibinfo {title}
  {{Correlation inequalities and the thermodynamic limit for classical and
  quantum continuous systems}},}\ }\href
  {http://projecteuclid.org/euclid.cmp/1103901661} {\bibfield  {journal}
  {\bibinfo  {journal} {Comm. Math. Phys.}\ }\textbf {\bibinfo {volume} {59}},\
  \bibinfo {pages} {235--266} (\bibinfo {year} {1978})}\BibitemShut {NoStop}%
\bibitem [{\citenamefont {Fr{\"o}hlich}\ and\ \citenamefont
  {Pfister}(1981)}]{FroPfi-81}%
  \BibitemOpen
  \bibfield  {author} {\bibinfo {author} {\bibnamefont {Fr{\"o}hlich},
  \bibfnamefont {J.}}\ and\ \bibinfo {author} {\bibnamefont {Pfister},
  \bibfnamefont {C.}},\ }\bibfield  {title} {\enquote {\bibinfo {title} {On the
  absence of spontaneous symmetry breaking and of crystalline ordering in
  two-dimensional systems},}\ }\href
  {http://projecteuclid.org/euclid.cmp/1103920246} {\bibfield  {journal}
  {\bibinfo  {journal} {Comm. Math. Phys.}\ }\textbf {\bibinfo {volume} {81}},\
  \bibinfo {pages} {277--298} (\bibinfo {year} {1981})}\BibitemShut {NoStop}%
\bibitem [{\citenamefont {Fr{\"o}hlich}\ and\ \citenamefont
  {Pfister}(1986)}]{FroPfi-86}%
  \BibitemOpen
  \bibfield  {author} {\bibinfo {author} {\bibnamefont {Fr{\"o}hlich},
  \bibfnamefont {J.}}\ and\ \bibinfo {author} {\bibnamefont {Pfister},
  \bibfnamefont {C.-E.}},\ }\bibfield  {title} {\enquote {\bibinfo {title}
  {Absence of crystalline ordering in two dimensions},}\ }\href
  {http://projecteuclid.org/euclid.cmp/1104115175} {\bibfield  {journal}
  {\bibinfo  {journal} {Comm. Math. Phys.}\ }\textbf {\bibinfo {volume}
  {104}},\ \bibinfo {pages} {697--700} (\bibinfo {year} {1986})}\BibitemShut
  {NoStop}%
\bibitem [{\citenamefont {Fr{\"o}hlich}, \citenamefont {Simon},\ and\
  \citenamefont {Spencer}(1976)}]{FroSimSpe-76}%
  \BibitemOpen
  \bibfield  {author} {\bibinfo {author} {\bibnamefont {Fr{\"o}hlich},
  \bibfnamefont {J.}}, \bibinfo {author} {\bibnamefont {Simon}, \bibfnamefont
  {B.}}, \ and\ \bibinfo {author} {\bibnamefont {Spencer}, \bibfnamefont
  {T.}},\ }\bibfield  {title} {\enquote {\bibinfo {title} {Infrared bounds,
  phase transitions and continuous symmetry breaking},}\ }\href {\doibase
  10.1007/BF01608557} {\bibfield  {journal} {\bibinfo  {journal} {Comm. Math.
  Phys.}\ }\textbf {\bibinfo {volume} {50}},\ \bibinfo {pages} {79--95}
  (\bibinfo {year} {1976})}\BibitemShut {NoStop}%
\bibitem [{\citenamefont {Fr{\"o}hlich}\ and\ \citenamefont
  {Spencer}(1981{\natexlab{a}})}]{FroSpe-81}%
  \BibitemOpen
  \bibfield  {author} {\bibinfo {author} {\bibnamefont {Fr{\"o}hlich},
  \bibfnamefont {J.}}\ and\ \bibinfo {author} {\bibnamefont {Spencer},
  \bibfnamefont {T.}},\ }\bibfield  {title} {\enquote {\bibinfo {title}
  {Kosterlitz-{T}houless transition in the two-dimensional plane rotator and
  {C}oulomb gas},}\ }\href {https://doi.org/10.1103/PhysRevLett.46.1006}
  {\bibfield  {journal} {\bibinfo  {journal} {Phys. Rev. Lett.}\ }\textbf
  {\bibinfo {volume} {46}},\ \bibinfo {pages} {1006--1009} (\bibinfo {year}
  {1981}{\natexlab{a}})}\BibitemShut {NoStop}%
\bibitem [{\citenamefont {Fr{\"o}hlich}\ and\ \citenamefont
  {Spencer}(1981{\natexlab{b}})}]{FroSpe-81b}%
  \BibitemOpen
  \bibfield  {author} {\bibinfo {author} {\bibnamefont {Fr{\"o}hlich},
  \bibfnamefont {J.}}\ and\ \bibinfo {author} {\bibnamefont {Spencer},
  \bibfnamefont {T.}},\ }\bibfield  {title} {\enquote {\bibinfo {title} {The
  {K}osterlitz-{T}houless transition in two-dimensional abelian spin systems
  and the {C}oulomb gas},}\ }\href
  {http://projecteuclid.org/euclid.cmp/1103920388} {\bibfield  {journal}
  {\bibinfo  {journal} {Comm. Math. Phys.}\ }\textbf {\bibinfo {volume} {81}},\
  \bibinfo {pages} {527--602} (\bibinfo {year}
  {1981}{\natexlab{b}})}\BibitemShut {NoStop}%
\bibitem [{\citenamefont {Fr\"{o}hlich}\ and\ \citenamefont
  {Spencer}(1981)}]{FroSpe-81c}%
  \BibitemOpen
  \bibfield  {author} {\bibinfo {author} {\bibnamefont {Fr\"{o}hlich},
  \bibfnamefont {J.}}\ and\ \bibinfo {author} {\bibnamefont {Spencer},
  \bibfnamefont {T.}},\ }\bibfield  {title} {\enquote {\bibinfo {title} {On the
  statistical mechanics of classical {C}oulomb and dipole gases},}\ }\href
  {\doibase 10.1007/BF01011379} {\bibfield  {journal} {\bibinfo  {journal} {J.
  Statist. Phys.}\ }\textbf {\bibinfo {volume} {24}},\ \bibinfo {pages}
  {617--701} (\bibinfo {year} {1981})}\BibitemShut {NoStop}%
\bibitem [{\citenamefont {Fr{\"o}hlich}\ and\ \citenamefont
  {Spencer}(1982)}]{FroSpe-82}%
  \BibitemOpen
  \bibfield  {author} {\bibinfo {author} {\bibnamefont {Fr{\"o}hlich},
  \bibfnamefont {J.}}\ and\ \bibinfo {author} {\bibnamefont {Spencer},
  \bibfnamefont {T.}},\ }\bibfield  {title} {\enquote {\bibinfo {title} {The
  phase transition in the one-dimensional {I}sing model with {$1/r^{2}$}
  interaction energy},}\ }\href
  {http://projecteuclid.org/euclid.cmp/1103921047} {\bibfield  {journal}
  {\bibinfo  {journal} {Comm. Math. Phys.}\ }\textbf {\bibinfo {volume} {84}},\
  \bibinfo {pages} {87--101} (\bibinfo {year} {1982})}\BibitemShut {NoStop}%
\bibitem [{\citenamefont {Fuchs}(1935)}]{Fuchs-35}%
  \BibitemOpen
  \bibfield  {author} {\bibinfo {author} {\bibnamefont {Fuchs}, \bibfnamefont
  {K.}},\ }\bibfield  {title} {\enquote {\bibinfo {title} {{A Quantum
  Mechanical Investigation of the Cohesive Forces of Metallic Copper}},}\
  }\href {\doibase 10.1098/rspa.1935.0167} {\bibfield  {journal} {\bibinfo
  {journal} {Proc. Roy. Soc. London Ser. A.}\ }\textbf {\bibinfo {volume}
  {151}},\ \bibinfo {pages} {585--602} (\bibinfo {year} {1935})}\BibitemShut
  {NoStop}%
\bibitem [{\citenamefont {Gallavotti}(1999)}]{Gallavotti-99}%
  \BibitemOpen
  \bibfield  {author} {\bibinfo {author} {\bibnamefont {Gallavotti},
  \bibfnamefont {G.}},\ }\href {\doibase 10.1007/978-3-662-03952-6} {\emph
  {\bibinfo {title} {Statistical mechanics}}},\ Texts and Monographs in
  Physics\ (\bibinfo  {publisher} {Springer-Verlag, Berlin},\ \bibinfo {year}
  {1999})\ pp.\ \bibinfo {pages} {xiv+339},\ \bibinfo {note} {a short
  treatise}\BibitemShut {NoStop}%
\bibitem [{\citenamefont {Gallavotti}\ and\ \citenamefont
  {Marchioro}(1973)}]{GalMar-73}%
  \BibitemOpen
  \bibfield  {author} {\bibinfo {author} {\bibnamefont {Gallavotti},
  \bibfnamefont {G.}}\ and\ \bibinfo {author} {\bibnamefont {Marchioro},
  \bibfnamefont {C.}},\ }\bibfield  {title} {\enquote {\bibinfo {title} {On the
  calculation of an integral},}\ }\href {\doibase 10.1016/0022-247X(73)90008-5}
  {\bibfield  {journal} {\bibinfo  {journal} {J. Math. Anal. Appl.}\ }\textbf
  {\bibinfo {volume} {44}},\ \bibinfo {pages} {661--675} (\bibinfo {year}
  {1973})}\BibitemShut {NoStop}%
\bibitem [{\citenamefont {Gangardt}\ and\ \citenamefont
  {Kamenev}(2001)}]{GanKam-01}%
  \BibitemOpen
  \bibfield  {author} {\bibinfo {author} {\bibnamefont {Gangardt},
  \bibfnamefont {D.~M.}}\ and\ \bibinfo {author} {\bibnamefont {Kamenev},
  \bibfnamefont {A.}},\ }\bibfield  {title} {\enquote {\bibinfo {title}
  {{Replica treatment of the Calogero--Sutherland model}},}\ }\href {\doibase
  https://doi.org/10.1016/S0550-3213(01)00326-1} {\bibfield  {journal}
  {\bibinfo  {journal} {Nuclear Physics B}\ }\textbf {\bibinfo {volume}
  {610}},\ \bibinfo {pages} {578--594} (\bibinfo {year} {2001})}\BibitemShut
  {NoStop}%
\bibitem [{\citenamefont {Gann}, \citenamefont {Chakravarty},\ and\
  \citenamefont {Chester}(1979)}]{GanChaChe-79}%
  \BibitemOpen
  \bibfield  {author} {\bibinfo {author} {\bibnamefont {Gann}, \bibfnamefont
  {R.}}, \bibinfo {author} {\bibnamefont {Chakravarty}, \bibfnamefont {S.}}, \
  and\ \bibinfo {author} {\bibnamefont {Chester}, \bibfnamefont {G.}},\
  }\bibfield  {title} {\enquote {\bibinfo {title} {Monte {C}arlo simulation of
  the classical two-dimensional one-component plasma},}\ }\href {\doibase
  10.1103/PhysRevB.20.326} {\bibfield  {journal} {\bibinfo  {journal} {Phys.
  Rev. B}\ }\textbf {\bibinfo {volume} {20}},\ \bibinfo {pages} {326--344}
  (\bibinfo {year} {1979})}\BibitemShut {NoStop}%
\bibitem [{\citenamefont {Garrod}\ and\ \citenamefont
  {Simmons}(1972)}]{GarSim-72}%
  \BibitemOpen
  \bibfield  {author} {\bibinfo {author} {\bibnamefont {Garrod}, \bibfnamefont
  {C.}}\ and\ \bibinfo {author} {\bibnamefont {Simmons}, \bibfnamefont {C.}},\
  }\bibfield  {title} {\enquote {\bibinfo {title} {Rigorous statistical
  mechanics for nonuniform systems},}\ }\href {\doibase 10.1063/1.1666118}
  {\bibfield  {journal} {\bibinfo  {journal} {J. Mathematical Phys.}\ }\textbf
  {\bibinfo {volume} {13}},\ \bibinfo {pages} {1168--1176} (\bibinfo {year}
  {1972})}\BibitemShut {NoStop}%
\bibitem [{\citenamefont {Gasser}\ \emph {et~al.}(2010)\citenamefont {Gasser},
  \citenamefont {Eisenmann}, \citenamefont {Maret},\ and\ \citenamefont
  {Keim}}]{GasEisMarKei-10}%
  \BibitemOpen
  \bibfield  {author} {\bibinfo {author} {\bibnamefont {Gasser}, \bibfnamefont
  {U.}}, \bibinfo {author} {\bibnamefont {Eisenmann}, \bibfnamefont {C.}},
  \bibinfo {author} {\bibnamefont {Maret}, \bibfnamefont {G.}}, \ and\ \bibinfo
  {author} {\bibnamefont {Keim}, \bibfnamefont {P.}},\ }\bibfield  {title}
  {\enquote {\bibinfo {title} {Melting of crystals in two dimensions},}\ }\href
  {\doibase 10.1002/cphc.200900755} {\bibfield  {journal} {\bibinfo  {journal}
  {Chem. Phys. Chem.}\ }\textbf {\bibinfo {volume} {11}},\ \bibinfo {pages}
  {963--970} (\bibinfo {year} {2010})}\BibitemShut {NoStop}%
\bibitem [{\citenamefont {Ge}\ and\ \citenamefont {Sandier}(2021)}]{GeSan-21}%
  \BibitemOpen
  \bibfield  {author} {\bibinfo {author} {\bibnamefont {Ge}, \bibfnamefont
  {Y.}}\ and\ \bibinfo {author} {\bibnamefont {Sandier}, \bibfnamefont {E.}},\
  }\bibfield  {title} {\enquote {\bibinfo {title} {Lattices with finite
  renormalized coulombian interaction energy in the plane},}\ }\href {\doibase
  10.2140/tunis.2021.3.93} {\bibfield  {journal} {\bibinfo  {journal} {Tunis.
  J. Math.}\ }\textbf {\bibinfo {volume} {3}},\ \bibinfo {pages} {93--120}
  (\bibinfo {year} {2021})}\BibitemShut {NoStop}%
\bibitem [{\citenamefont {Genovese}\ and\ \citenamefont
  {Simonella}(2012)}]{GenSim-12}%
  \BibitemOpen
  \bibfield  {author} {\bibinfo {author} {\bibnamefont {Genovese},
  \bibfnamefont {G.}}\ and\ \bibinfo {author} {\bibnamefont {Simonella},
  \bibfnamefont {S.}},\ }\bibfield  {title} {\enquote {\bibinfo {title} {On the
  stationary {BBGKY} hierarchy for equilibrium states},}\ }\href {\doibase
  10.1007/s10955-012-0525-7} {\bibfield  {journal} {\bibinfo  {journal} {J.
  Stat. Phys.}\ }\textbf {\bibinfo {volume} {148}},\ \bibinfo {pages} {89--112}
  (\bibinfo {year} {2012})}\BibitemShut {NoStop}%
\bibitem [{\citenamefont {Georgii}(1976)}]{Georgii-76}%
  \BibitemOpen
  \bibfield  {author} {\bibinfo {author} {\bibnamefont {Georgii}, \bibfnamefont
  {H.-O.}},\ }\bibfield  {title} {\enquote {\bibinfo {title} {Canonical and
  grand canonical {G}ibbs states for continuum systems},}\ }\href
  {http://projecteuclid.org/euclid.cmp/1103899810} {\bibfield  {journal}
  {\bibinfo  {journal} {Comm. Math. Phys.}\ }\textbf {\bibinfo {volume} {48}},\
  \bibinfo {pages} {31--51} (\bibinfo {year} {1976})}\BibitemShut {NoStop}%
\bibitem [{\citenamefont {Georgii}(1994)}]{Georgii-94}%
  \BibitemOpen
  \bibfield  {author} {\bibinfo {author} {\bibnamefont {Georgii}, \bibfnamefont
  {H.-O.}},\ }\bibfield  {title} {\enquote {\bibinfo {title} {Large deviations
  and the equivalence of ensembles for gibbsian particle systems with
  superstable interaction},}\ }\href {\doibase 10.1007/BF01199021} {\bibfield
  {journal} {\bibinfo  {journal} {Probab. Theory Related Fields}\ }\textbf
  {\bibinfo {volume} {99}},\ \bibinfo {pages} {171--195} (\bibinfo {year}
  {1994})}\BibitemShut {NoStop}%
\bibitem [{\citenamefont {Georgii}(1995)}]{Georgii-95}%
  \BibitemOpen
  \bibfield  {author} {\bibinfo {author} {\bibnamefont {Georgii}, \bibfnamefont
  {H.-O.}},\ }\bibfield  {title} {\enquote {\bibinfo {title} {The equivalence
  of ensembles for classical systems of particles},}\ }\href {\doibase
  10.1007/BF02179874} {\bibfield  {journal} {\bibinfo  {journal} {J. Statist.
  Phys.}\ }\textbf {\bibinfo {volume} {80}},\ \bibinfo {pages} {1341--1378}
  (\bibinfo {year} {1995})}\BibitemShut {NoStop}%
\bibitem [{\citenamefont {Georgii}(1999)}]{Georgii-99}%
  \BibitemOpen
  \bibfield  {author} {\bibinfo {author} {\bibnamefont {Georgii}, \bibfnamefont
  {H.-O.}},\ }\bibfield  {title} {\enquote {\bibinfo {title} {Translation
  invariance and continuous symmetries in two-dimensional continuum systems},}\
  }in\ \href@noop {} {\emph {\bibinfo {booktitle} {Mathematical results in
  statistical mechanics ({M}arseilles, 1998)}}}\ (\bibinfo  {publisher} {World
  Sci. Publ., River Edge, NJ},\ \bibinfo {year} {1999})\ pp.\ \bibinfo {pages}
  {53--69}\BibitemShut {NoStop}%
\bibitem [{\citenamefont {Georgii}(2011)}]{Georgii-11}%
  \BibitemOpen
  \bibfield  {author} {\bibinfo {author} {\bibnamefont {Georgii}, \bibfnamefont
  {H.-O.}},\ }\href {\doibase 10.1515/9783110250329} {\emph {\bibinfo {title}
  {Gibbs measures and phase transitions}}},\ \bibinfo {edition} {2nd}\ ed.,\
  \bibinfo {series} {De Gruyter Studies in Mathematics}, Vol.~\bibinfo {volume}
  {9}\ (\bibinfo  {publisher} {Walter de Gruyter \& Co., Berlin},\ \bibinfo
  {year} {2011})\ pp.\ \bibinfo {pages} {xiv+545}\BibitemShut {NoStop}%
\bibitem [{\citenamefont {Georgii}\ and\ \citenamefont
  {Zessin}(1993)}]{GeoZes-93}%
  \BibitemOpen
  \bibfield  {author} {\bibinfo {author} {\bibnamefont {Georgii}, \bibfnamefont
  {H.-O.}}\ and\ \bibinfo {author} {\bibnamefont {Zessin}, \bibfnamefont
  {H.}},\ }\bibfield  {title} {\enquote {\bibinfo {title} {Large deviations and
  the maximum entropy principle for marked point random fields},}\ }\href
  {\doibase 10.1007/BF01192132} {\bibfield  {journal} {\bibinfo  {journal}
  {Probab. Theory Related Fields}\ }\textbf {\bibinfo {volume} {96}},\ \bibinfo
  {pages} {177--204} (\bibinfo {year} {1993})}\BibitemShut {NoStop}%
\bibitem [{\citenamefont {Ghosh}(2015)}]{Ghosh-15}%
  \BibitemOpen
  \bibfield  {author} {\bibinfo {author} {\bibnamefont {Ghosh}, \bibfnamefont
  {S.}},\ }\bibfield  {title} {\enquote {\bibinfo {title} {Determinantal
  processes and completeness of random exponentials: the critical case},}\
  }\href {\doibase 10.1007/s00440-014-0601-9} {\bibfield  {journal} {\bibinfo
  {journal} {Probab. Theory Related Fields}\ }\textbf {\bibinfo {volume}
  {163}},\ \bibinfo {pages} {643--665} (\bibinfo {year} {2015})}\BibitemShut
  {NoStop}%
\bibitem [{\citenamefont {Ghosh}\ and\ \citenamefont
  {Lebowitz}(2017{\natexlab{a}})}]{GhoLeb-17b}%
  \BibitemOpen
  \bibfield  {author} {\bibinfo {author} {\bibnamefont {Ghosh}, \bibfnamefont
  {S.}}\ and\ \bibinfo {author} {\bibnamefont {Lebowitz}, \bibfnamefont {J.}},\
  }\bibfield  {title} {\enquote {\bibinfo {title} {Number rigidity in
  superhomogeneous random point fields},}\ }\href {\doibase
  10.1007/s10955-016-1633-6} {\bibfield  {journal} {\bibinfo  {journal} {J.
  Stat. Phys.}\ }\textbf {\bibinfo {volume} {166}},\ \bibinfo {pages}
  {1016--1027} (\bibinfo {year} {2017}{\natexlab{a}})}\BibitemShut {NoStop}%
\bibitem [{\citenamefont {Ghosh}\ and\ \citenamefont
  {Lebowitz}(2017{\natexlab{b}})}]{GhoLeb-17}%
  \BibitemOpen
  \bibfield  {author} {\bibinfo {author} {\bibnamefont {Ghosh}, \bibfnamefont
  {S.}}\ and\ \bibinfo {author} {\bibnamefont {Lebowitz}, \bibfnamefont
  {J.~L.}},\ }\bibfield  {title} {\enquote {\bibinfo {title} {Fluctuations,
  large deviations and rigidity in hyperuniform systems: a brief survey},}\
  }\href {\doibase 10.1007/s13226-017-0248-1} {\bibfield  {journal} {\bibinfo
  {journal} {Indian J. Pure Appl. Math.}\ }\textbf {\bibinfo {volume} {48}},\
  \bibinfo {pages} {609--631} (\bibinfo {year}
  {2017}{\natexlab{b}})}\BibitemShut {NoStop}%
\bibitem [{\citenamefont {Ghosh}\ and\ \citenamefont
  {Peres}(2017)}]{GhoPer-17}%
  \BibitemOpen
  \bibfield  {author} {\bibinfo {author} {\bibnamefont {Ghosh}, \bibfnamefont
  {S.}}\ and\ \bibinfo {author} {\bibnamefont {Peres}, \bibfnamefont {Y.}},\
  }\bibfield  {title} {\enquote {\bibinfo {title} {Rigidity and tolerance in
  point processes: {G}aussian zeros and {G}inibre eigenvalues},}\ }\href
  {\doibase 10.1215/00127094-2017-0002} {\bibfield  {journal} {\bibinfo
  {journal} {Duke Math. J.}\ }\textbf {\bibinfo {volume} {166}},\ \bibinfo
  {pages} {1789--1858} (\bibinfo {year} {2017})}\BibitemShut {NoStop}%
\bibitem [{\citenamefont {Gigante}\ and\ \citenamefont
  {Leopardi}(2017)}]{GigLeo-17}%
  \BibitemOpen
  \bibfield  {author} {\bibinfo {author} {\bibnamefont {Gigante}, \bibfnamefont
  {G.}}\ and\ \bibinfo {author} {\bibnamefont {Leopardi}, \bibfnamefont {P.}},\
  }\bibfield  {title} {\enquote {\bibinfo {title} {Diameter bounded equal
  measure partitions of {A}hlfors regular metric measure spaces},}\ }\href
  {\doibase 10.1007/s00454-016-9834-y} {\bibfield  {journal} {\bibinfo
  {journal} {Discrete Comput. Geom.}\ }\textbf {\bibinfo {volume} {57}},\
  \bibinfo {pages} {419--430} (\bibinfo {year} {2017})}\BibitemShut {NoStop}%
\bibitem [{\citenamefont {Ginibre}(1965)}]{Ginibre-65}%
  \BibitemOpen
  \bibfield  {author} {\bibinfo {author} {\bibnamefont {Ginibre}, \bibfnamefont
  {J.}},\ }\bibfield  {title} {\enquote {\bibinfo {title} {Statistical
  ensembles of complex, quaternion, and real matrices},}\ }\href {\doibase
  10.1063/1.1704292} {\bibfield  {journal} {\bibinfo  {journal} {J. Math.
  Phys.}\ }\textbf {\bibinfo {volume} {6}},\ \bibinfo {pages} {440--449}
  (\bibinfo {year} {1965})}\BibitemShut {NoStop}%
\bibitem [{\citenamefont {Ginibre}(1967)}]{Ginibre-67}%
  \BibitemOpen
  \bibfield  {author} {\bibinfo {author} {\bibnamefont {Ginibre}, \bibfnamefont
  {J.}},\ }\bibfield  {title} {\enquote {\bibinfo {title} {Rigorous lower bound
  on the compressibility of a classical system},}\ }\href {\doibase
  https://doi.org/10.1016/0375-9601(67)90575-0} {\bibfield  {journal} {\bibinfo
   {journal} {Physics Letters A}\ }\textbf {\bibinfo {volume} {24}},\ \bibinfo
  {pages} {223--224} (\bibinfo {year} {1967})}\BibitemShut {NoStop}%
\bibitem [{\citenamefont {Giuliani}\ and\ \citenamefont
  {Vignale}(2005)}]{GiuVig-05}%
  \BibitemOpen
  \bibfield  {author} {\bibinfo {author} {\bibnamefont {Giuliani},
  \bibfnamefont {G.}}\ and\ \bibinfo {author} {\bibnamefont {Vignale},
  \bibfnamefont {G.}},\ }\href {http://www.google.dk/books?id=kFkIKRfgUpsC}
  {\emph {\bibinfo {title} {Quantum Theory of the Electron Liquid}}}\ (\bibinfo
   {publisher} {Cambridge University Press},\ \bibinfo {year}
  {2005})\BibitemShut {NoStop}%
\bibitem [{\citenamefont {Goldman}, \citenamefont {Novaga},\ and\ \citenamefont
  {Ruffini}(2015)}]{GolNovRuf-15}%
  \BibitemOpen
  \bibfield  {author} {\bibinfo {author} {\bibnamefont {Goldman}, \bibfnamefont
  {M.}}, \bibinfo {author} {\bibnamefont {Novaga}, \bibfnamefont {M.}}, \ and\
  \bibinfo {author} {\bibnamefont {Ruffini}, \bibfnamefont {B.}},\ }\bibfield
  {title} {\enquote {\bibinfo {title} {Existence and stability for a non-local
  isoperimetric model of charged liquid drops},}\ }\href {\doibase
  10.1007/s00205-014-0827-9} {\bibfield  {journal} {\bibinfo  {journal} {Arch.
  Ration. Mech. Anal.}\ }\textbf {\bibinfo {volume} {217}},\ \bibinfo {pages}
  {1--36} (\bibinfo {year} {2015})}\BibitemShut {NoStop}%
\bibitem [{\citenamefont {Goldman}, \citenamefont {Novaga},\ and\ \citenamefont
  {Ruffini}(2022)}]{GolNovRuf-22_ppt}%
  \BibitemOpen
  \bibfield  {author} {\bibinfo {author} {\bibnamefont {Goldman}, \bibfnamefont
  {M.}}, \bibinfo {author} {\bibnamefont {Novaga}, \bibfnamefont {M.}}, \ and\
  \bibinfo {author} {\bibnamefont {Ruffini}, \bibfnamefont {B.}},\ }\href@noop
  {} {\enquote {\bibinfo {title} {Rigidity of the ball for an isoperimetric
  problem with strong capacitary repulsion},}\ } (\bibinfo {year} {2022}),\
  \Eprint {http://arxiv.org/abs/2201.04376} {arXiv:2201.04376 [math.AP]}
  \BibitemShut {NoStop}%
\bibitem [{\citenamefont {Good}(1970)}]{Good-70}%
  \BibitemOpen
  \bibfield  {author} {\bibinfo {author} {\bibnamefont {Good}, \bibfnamefont
  {I.~J.}},\ }\bibfield  {title} {\enquote {\bibinfo {title} {Short proof of a
  conjecture by {D}yson},}\ }\href {\doibase 10.1063/1.1665339} {\bibfield
  {journal} {\bibinfo  {journal} {J. Mathematical Phys.}\ }\textbf {\bibinfo
  {volume} {11}},\ \bibinfo {pages} {1884} (\bibinfo {year}
  {1970})}\BibitemShut {NoStop}%
\bibitem [{\citenamefont {G{\"o}tze}(2004)}]{Gotze-04}%
  \BibitemOpen
  \bibfield  {author} {\bibinfo {author} {\bibnamefont {G{\"o}tze},
  \bibfnamefont {F.}},\ }\bibfield  {title} {\enquote {\bibinfo {title}
  {Lattice point problems and values of quadratic forms},}\ }\href {\doibase
  10.1007/s00222-004-0366-3} {\bibfield  {journal} {\bibinfo  {journal}
  {Invent. Math.}\ }\textbf {\bibinfo {volume} {157}},\ \bibinfo {pages}
  {195--226} (\bibinfo {year} {2004})}\BibitemShut {NoStop}%
\bibitem [{\citenamefont {Graf}\ and\ \citenamefont
  {Schenker}(1995{\natexlab{a}})}]{GraSch-95b}%
  \BibitemOpen
  \bibfield  {author} {\bibinfo {author} {\bibnamefont {Graf}, \bibfnamefont
  {G.~M.}}\ and\ \bibinfo {author} {\bibnamefont {Schenker}, \bibfnamefont
  {D.}},\ }\bibfield  {title} {\enquote {\bibinfo {title} {The free energy of
  systems with net charge},}\ }\href {\doibase 10.1007/BF00739156} {\bibfield
  {journal} {\bibinfo  {journal} {Lett. Math. Phys.}\ }\textbf {\bibinfo
  {volume} {35}},\ \bibinfo {pages} {75--83} (\bibinfo {year}
  {1995}{\natexlab{a}})}\BibitemShut {NoStop}%
\bibitem [{\citenamefont {Graf}\ and\ \citenamefont
  {Schenker}(1995{\natexlab{b}})}]{GraSch-95}%
  \BibitemOpen
  \bibfield  {author} {\bibinfo {author} {\bibnamefont {Graf}, \bibfnamefont
  {G.~M.}}\ and\ \bibinfo {author} {\bibnamefont {Schenker}, \bibfnamefont
  {D.}},\ }\bibfield  {title} {\enquote {\bibinfo {title} {On the molecular
  limit of {C}oulomb gases},}\ }\href {\doibase 10.1007/BF02099471} {\bibfield
  {journal} {\bibinfo  {journal} {Commun. Math. Phys.}\ }\textbf {\bibinfo
  {volume} {174}},\ \bibinfo {pages} {215--227} (\bibinfo {year}
  {1995}{\natexlab{b}})}\BibitemShut {NoStop}%
\bibitem [{\citenamefont {Gregg}(1989)}]{Gregg-89}%
  \BibitemOpen
  \bibfield  {author} {\bibinfo {author} {\bibnamefont {Gregg}, \bibfnamefont
  {J.~N.}},\ }\bibfield  {title} {\enquote {\bibinfo {title} {The existence of
  the thermodynamic limit in {C}oulomb-like systems},}\ }\href
  {https://projecteuclid.org:443/euclid.cmp/1104178763} {\bibfield  {journal}
  {\bibinfo  {journal} {Comm. Math. Phys.}\ }\textbf {\bibinfo {volume}
  {123}},\ \bibinfo {pages} {255--276} (\bibinfo {year} {1989})}\BibitemShut
  {NoStop}%
\bibitem [{\citenamefont {Grimes}\ and\ \citenamefont
  {Adams}(1979)}]{GriAda-79}%
  \BibitemOpen
  \bibfield  {author} {\bibinfo {author} {\bibnamefont {Grimes}, \bibfnamefont
  {C.~C.}}\ and\ \bibinfo {author} {\bibnamefont {Adams}, \bibfnamefont {G.}},\
  }\bibfield  {title} {\enquote {\bibinfo {title} {Evidence for a
  liquid-to-crystal phase transition in a classical, two-dimensional sheet of
  electrons},}\ }\href {\doibase 10.1103/PhysRevLett.42.795} {\bibfield
  {journal} {\bibinfo  {journal} {Phys. Rev. Lett.}\ }\textbf {\bibinfo
  {volume} {42}},\ \bibinfo {pages} {795--798} (\bibinfo {year}
  {1979})}\BibitemShut {NoStop}%
\bibitem [{\citenamefont {Groeneveld}(1967)}]{Groeneveld-67}%
  \BibitemOpen
  \bibfield  {author} {\bibinfo {author} {\bibnamefont {Groeneveld},},\ }in\
  \href {https://books.google.fr/books?id=rYkuAAAAIAAJ} {\emph {\bibinfo
  {booktitle} {{Statistical Mechanics: Foundations and Applications;
  Proceedings of the IUPAP meeting, Copenhagen 1966}}}},\ \bibinfo {editor}
  {edited by\ \bibinfo {editor} {\bibfnamefont {T.}~\bibnamefont {Bak}}},\
  \bibinfo {organization} {IUPAP (International Union of Pure and Applied
  Physics)}\ (\bibinfo  {publisher} {W. A. Benjamin},\ \bibinfo {year}
  {1967})\BibitemShut {NoStop}%
\bibitem [{\citenamefont {Gruber}, \citenamefont {Lebowitz},\ and\
  \citenamefont {Martin}(1981)}]{GruLebMar-81}%
  \BibitemOpen
  \bibfield  {author} {\bibinfo {author} {\bibnamefont {Gruber}, \bibfnamefont
  {C.}}, \bibinfo {author} {\bibnamefont {Lebowitz}, \bibfnamefont {J.}}, \
  and\ \bibinfo {author} {\bibnamefont {Martin}, \bibfnamefont {P.}},\
  }\bibfield  {title} {\enquote {\bibinfo {title} {Sum rules for inhomogeneous
  {C}oulomb systems},}\ }\href {\doibase 10.1063/1.442093} {\bibfield
  {journal} {\bibinfo  {journal} {J. Chem. Phys.}\ }\textbf {\bibinfo {volume}
  {75}},\ \bibinfo {pages} {944--954} (\bibinfo {year} {1981})}\BibitemShut
  {NoStop}%
\bibitem [{\citenamefont {Gruber}, \citenamefont {Lugrin},\ and\ \citenamefont
  {Martin}(1978)}]{GruLugMar-78}%
  \BibitemOpen
  \bibfield  {author} {\bibinfo {author} {\bibnamefont {Gruber}, \bibfnamefont
  {C.}}, \bibinfo {author} {\bibnamefont {Lugrin}, \bibfnamefont {C.}}, \ and\
  \bibinfo {author} {\bibnamefont {Martin}, \bibfnamefont {P.~A.}},\ }\bibfield
   {title} {\enquote {\bibinfo {title} {Equilibrium equations for classical
  systems with long range forces and application to the one dimensional
  {C}oulomb gas},}\ }\href@noop {} {\bibfield  {journal} {\bibinfo  {journal}
  {Helv Phys. Acta}\ }\textbf {\bibinfo {volume} {51}},\ \bibinfo {pages}
  {829--866} (\bibinfo {year} {1978})}\BibitemShut {NoStop}%
\bibitem [{\citenamefont {{Gruber}}, \citenamefont {{Lugrin}},\ and\
  \citenamefont {{Martin}}(1980)}]{GruLugMar-80}%
  \BibitemOpen
  \bibfield  {author} {\bibinfo {author} {\bibnamefont {{Gruber}},
  \bibfnamefont {C.}}, \bibinfo {author} {\bibnamefont {{Lugrin}},
  \bibfnamefont {C.}}, \ and\ \bibinfo {author} {\bibnamefont {{Martin}},
  \bibfnamefont {P.~A.}},\ }\bibfield  {title} {\enquote {\bibinfo {title}
  {Equilibrium properties of classical systems with long-range forces. {BBGKY}
  equation, neutrality, screening, and sum rules},}\ }\href {\doibase
  10.1007/BF01008049} {\bibfield  {journal} {\bibinfo  {journal} {J. Stat.
  Phys.}\ }\textbf {\bibinfo {volume} {22}},\ \bibinfo {pages} {193--236}
  (\bibinfo {year} {1980})}\BibitemShut {NoStop}%
\bibitem [{\citenamefont {Gruber}\ and\ \citenamefont
  {Martin}(1981)}]{GruMar-81}%
  \BibitemOpen
  \bibfield  {author} {\bibinfo {author} {\bibnamefont {Gruber}, \bibfnamefont
  {C.}}\ and\ \bibinfo {author} {\bibnamefont {Martin}, \bibfnamefont {P.}},\
  }\bibfield  {title} {\enquote {\bibinfo {title} {Translation invariance in
  statistical mechanics of classical continuous systems},}\ }\href {\doibase
  10.1016/0003-4916(81)90183-4} {\bibfield  {journal} {\bibinfo  {journal}
  {Ann. Physics}\ }\textbf {\bibinfo {volume} {131}},\ \bibinfo {pages}
  {56--72} (\bibinfo {year} {1981})}\BibitemShut {NoStop}%
\bibitem [{\citenamefont {Gruber}, \citenamefont {Martin},\ and\ \citenamefont
  {Oguey}(1982)}]{GruMarOgu-82}%
  \BibitemOpen
  \bibfield  {author} {\bibinfo {author} {\bibnamefont {Gruber}, \bibfnamefont
  {C.}}, \bibinfo {author} {\bibnamefont {Martin}, \bibfnamefont {P.~A.}}, \
  and\ \bibinfo {author} {\bibnamefont {Oguey}, \bibfnamefont {C.}},\
  }\bibfield  {title} {\enquote {\bibinfo {title} {Euclidean invariance in
  statistical mechanics of classical continuous system},}\ }\href
  {http://projecteuclid.org/euclid.cmp/1103921045} {\bibfield  {journal}
  {\bibinfo  {journal} {Comm. Math. Phys.}\ }\textbf {\bibinfo {volume} {84}},\
  \bibinfo {pages} {55--69} (\bibinfo {year} {1982})}\BibitemShut {NoStop}%
\bibitem [{\citenamefont {Gruber}(2012)}]{Gruber-12}%
  \BibitemOpen
  \bibfield  {author} {\bibinfo {author} {\bibnamefont {Gruber}, \bibfnamefont
  {P.~M.}},\ }\bibfield  {title} {\enquote {\bibinfo {title} {Application of an
  idea of {V}orono\"{\i} to lattice zeta functions},}\ }\href {\doibase
  10.1134/S0081543812010099} {\bibfield  {journal} {\bibinfo  {journal} {Proc.
  Steklov Inst. Math.}\ }\textbf {\bibinfo {volume} {276}},\ \bibinfo {pages}
  {103--124} (\bibinfo {year} {2012})}\BibitemShut {NoStop}%
\bibitem [{\citenamefont {Guhr}, \citenamefont {M{\"u}ller-Groeling},\ and\
  \citenamefont {Weidenm{\"u}ller}(1998)}]{GuhMulWei-98}%
  \BibitemOpen
  \bibfield  {author} {\bibinfo {author} {\bibnamefont {Guhr}, \bibfnamefont
  {T.}}, \bibinfo {author} {\bibnamefont {M{\"u}ller-Groeling}, \bibfnamefont
  {A.}}, \ and\ \bibinfo {author} {\bibnamefont {Weidenm{\"u}ller},
  \bibfnamefont {H.~A.}},\ }\bibfield  {title} {\enquote {\bibinfo {title}
  {Random-matrix theories in quantum physics: common concepts},}\ }\href
  {\doibase https://doi.org/10.1016/S0370-1573(97)00088-4} {\bibfield
  {journal} {\bibinfo  {journal} {Physics Reports}\ }\textbf {\bibinfo {volume}
  {299}},\ \bibinfo {pages} {189 -- 425} (\bibinfo {year} {1998})}\BibitemShut
  {NoStop}%
\bibitem [{\citenamefont {Gunson}(1962)}]{Gunson-62}%
  \BibitemOpen
  \bibfield  {author} {\bibinfo {author} {\bibnamefont {Gunson}, \bibfnamefont
  {J.}},\ }\bibfield  {title} {\enquote {\bibinfo {title} {Proof of a
  conjecture by {D}yson in the statistical theory of energy levels},}\ }\href
  {\doibase 10.1063/1.1724277} {\bibfield  {journal} {\bibinfo  {journal} {J.
  Mathematical Phys.}\ }\textbf {\bibinfo {volume} {3}},\ \bibinfo {pages}
  {752--753} (\bibinfo {year} {1962})}\BibitemShut {NoStop}%
\bibitem [{\citenamefont {Gustafsson}(2004)}]{Gustafsson-04}%
  \BibitemOpen
  \bibfield  {author} {\bibinfo {author} {\bibnamefont {Gustafsson},
  \bibfnamefont {B.}},\ }\bibfield  {title} {\enquote {\bibinfo {title}
  {Lectures on balayage},}\ }in\ \href@noop {} {\emph {\bibinfo {booktitle}
  {Clifford algebras and potential theory}}},\ \bibinfo {series} {Univ. Joensuu
  Dept. Math. Rep. Ser.}, Vol.~\bibinfo {volume} {7}\ (\bibinfo  {publisher}
  {Univ. Joensuu, Joensuu},\ \bibinfo {year} {2004})\ pp.\ \bibinfo {pages}
  {17--63}\BibitemShut {NoStop}%
\bibitem [{\citenamefont {Gustafsson}\ and\ \citenamefont
  {Putinar}(2007)}]{GusPut-07}%
  \BibitemOpen
  \bibfield  {author} {\bibinfo {author} {\bibnamefont {Gustafsson},
  \bibfnamefont {B.}}\ and\ \bibinfo {author} {\bibnamefont {Putinar},
  \bibfnamefont {M.}},\ }\bibfield  {title} {\enquote {\bibinfo {title}
  {Selected topics on quadrature domains},}\ }\href {\doibase
  10.1016/j.physd.2007.04.015} {\bibfield  {journal} {\bibinfo  {journal}
  {Phys. D}\ }\textbf {\bibinfo {volume} {235}},\ \bibinfo {pages} {90--100}
  (\bibinfo {year} {2007})}\BibitemShut {NoStop}%
\bibitem [{\citenamefont {Gustafsson}\ and\ \citenamefont
  {Shapiro}(2005)}]{GusSha-05}%
  \BibitemOpen
  \bibfield  {author} {\bibinfo {author} {\bibnamefont {Gustafsson},
  \bibfnamefont {B.}}\ and\ \bibinfo {author} {\bibnamefont {Shapiro},
  \bibfnamefont {H.~S.}},\ }\bibfield  {title} {\enquote {\bibinfo {title}
  {What is a quadrature domain?}}\ }in\ \href {\doibase
  10.1007/3-7643-7316-4\_1} {\emph {\bibinfo {booktitle} {Quadrature domains
  and their applications}}},\ \bibinfo {series} {Oper. Theory Adv. Appl.},
  Vol.\ \bibinfo {volume} {156}\ (\bibinfo  {publisher} {Birkh\"{a}user,
  Basel},\ \bibinfo {year} {2005})\ pp.\ \bibinfo {pages} {1--25}\BibitemShut
  {NoStop}%
\bibitem [{\citenamefont {Ha}(1994)}]{Ha-94}%
  \BibitemOpen
  \bibfield  {author} {\bibinfo {author} {\bibnamefont {Ha}, \bibfnamefont
  {Z.~N.~C.}},\ }\bibfield  {title} {\enquote {\bibinfo {title} {Exact
  dynamical correlation functions of {C}alogero-{S}utherland model and
  one-dimensional fractional statistics},}\ }\href {\doibase
  10.1103/PhysRevLett.73.1574} {\bibfield  {journal} {\bibinfo  {journal}
  {Phys. Rev. Lett.}\ }\textbf {\bibinfo {volume} {73}},\ \bibinfo {pages}
  {1574--1577} (\bibinfo {year} {1994})}\BibitemShut {NoStop}%
\bibitem [{\citenamefont {Hainzl}, \citenamefont {Lewin},\ and\ \citenamefont
  {Solovej}(2009{\natexlab{a}})}]{HaiLewSol_1-09}%
  \BibitemOpen
  \bibfield  {author} {\bibinfo {author} {\bibnamefont {Hainzl}, \bibfnamefont
  {C.}}, \bibinfo {author} {\bibnamefont {Lewin}, \bibfnamefont {M.}}, \ and\
  \bibinfo {author} {\bibnamefont {Solovej}, \bibfnamefont {J.~P.}},\
  }\bibfield  {title} {\enquote {\bibinfo {title} {The thermodynamic limit of
  quantum {C}oulomb systems. {P}art {I}. {G}eneral theory},}\ }\href {\doibase
  10.1016/j.aim.2008.12.010} {\bibfield  {journal} {\bibinfo  {journal}
  {Advances in Math.}\ }\textbf {\bibinfo {volume} {221}},\ \bibinfo {pages}
  {454--487} (\bibinfo {year} {2009}{\natexlab{a}})},\ \Eprint
  {http://arxiv.org/abs/0806.1708} {0806.1708} \BibitemShut {NoStop}%
\bibitem [{\citenamefont {Hainzl}, \citenamefont {Lewin},\ and\ \citenamefont
  {Solovej}(2009{\natexlab{b}})}]{HaiLewSol_2-09}%
  \BibitemOpen
  \bibfield  {author} {\bibinfo {author} {\bibnamefont {Hainzl}, \bibfnamefont
  {C.}}, \bibinfo {author} {\bibnamefont {Lewin}, \bibfnamefont {M.}}, \ and\
  \bibinfo {author} {\bibnamefont {Solovej}, \bibfnamefont {J.~P.}},\
  }\bibfield  {title} {\enquote {\bibinfo {title} {The thermodynamic limit of
  quantum {C}oulomb systems. {P}art {II}. {A}pplications},}\ }\href {\doibase
  10.1016/j.aim.2008.12.011} {\bibfield  {journal} {\bibinfo  {journal}
  {Advances in Math.}\ }\textbf {\bibinfo {volume} {221}},\ \bibinfo {pages}
  {488--546} (\bibinfo {year} {2009}{\natexlab{b}})},\ \Eprint
  {http://arxiv.org/abs/0806.1709} {0806.1709} \BibitemShut {NoStop}%
\bibitem [{\citenamefont {Haldane}(1981)}]{Haldane-81}%
  \BibitemOpen
  \bibfield  {author} {\bibinfo {author} {\bibnamefont {Haldane}, \bibfnamefont
  {F.~D.~M.}},\ }\bibfield  {title} {\enquote {\bibinfo {title} {Effective
  harmonic-fluid approach to low-energy properties of one-dimensional quantum
  fluids},}\ }\href {\doibase 10.1103/PhysRevLett.47.1840} {\bibfield
  {journal} {\bibinfo  {journal} {Phys. Rev. Lett.}\ }\textbf {\bibinfo
  {volume} {47}},\ \bibinfo {pages} {1840--1843} (\bibinfo {year}
  {1981})}\BibitemShut {NoStop}%
\bibitem [{\citenamefont {Haldane}(1983)}]{Haldane-83}%
  \BibitemOpen
  \bibfield  {author} {\bibinfo {author} {\bibnamefont {Haldane}, \bibfnamefont
  {F.~D.~M.}},\ }\bibfield  {title} {\enquote {\bibinfo {title} {{Fractional
  Quantization of the Hall Effect: A Hierarchy of Incompressible Quantum Fluid
  States}},}\ }\href {\doibase 10.1103/PhysRevLett.51.605} {\bibfield
  {journal} {\bibinfo  {journal} {Phys. Rev. Lett.}\ }\textbf {\bibinfo
  {volume} {51}},\ \bibinfo {pages} {605--608} (\bibinfo {year}
  {1983})}\BibitemShut {NoStop}%
\bibitem [{\citenamefont {Haldane}(1991)}]{Haldane-91}%
  \BibitemOpen
  \bibfield  {author} {\bibinfo {author} {\bibnamefont {Haldane}, \bibfnamefont
  {F.~D.~M.}},\ }\bibfield  {title} {\enquote {\bibinfo {title} {``{F}ractional
  statistics'' in arbitrary dimensions: A generalization of the {P}auli
  principle},}\ }\href {\doibase 10.1103/PhysRevLett.67.937} {\bibfield
  {journal} {\bibinfo  {journal} {Phys. Rev. Lett.}\ }\textbf {\bibinfo
  {volume} {67}},\ \bibinfo {pages} {937--940} (\bibinfo {year}
  {1991})}\BibitemShut {NoStop}%
\bibitem [{\citenamefont {Hall}(1979)}]{Hall-79}%
  \BibitemOpen
  \bibfield  {author} {\bibinfo {author} {\bibnamefont {Hall}, \bibfnamefont
  {G.~L.}},\ }\bibfield  {title} {\enquote {\bibinfo {title} {Correction to
  {F}uchs' calculation of the electrostatic energy of a {W}igner solid},}\
  }\href {\doibase 10.1103/PhysRevB.19.3921} {\bibfield  {journal} {\bibinfo
  {journal} {Phys. Rev. B}\ }\textbf {\bibinfo {volume} {19}},\ \bibinfo
  {pages} {3921--3932} (\bibinfo {year} {1979})}\BibitemShut {NoStop}%
\bibitem [{\citenamefont {Hall}(1981)}]{Hall-81}%
  \BibitemOpen
  \bibfield  {author} {\bibinfo {author} {\bibnamefont {Hall}, \bibfnamefont
  {G.~L.}},\ }\bibfield  {title} {\enquote {\bibinfo {title} {Response to
  ``{C}omment on the average potential of a {W}igner solid"},}\ }\href
  {\doibase 10.1103/PhysRevB.24.7415} {\bibfield  {journal} {\bibinfo
  {journal} {Phys. Rev. B}\ }\textbf {\bibinfo {volume} {24}},\ \bibinfo
  {pages} {7415--7418} (\bibinfo {year} {1981})}\BibitemShut {NoStop}%
\bibitem [{\citenamefont {Hall}\ and\ \citenamefont {Rice}(1980)}]{HalRic-80}%
  \BibitemOpen
  \bibfield  {author} {\bibinfo {author} {\bibnamefont {Hall}, \bibfnamefont
  {G.~L.}}\ and\ \bibinfo {author} {\bibnamefont {Rice}, \bibfnamefont
  {T.~R.}},\ }\bibfield  {title} {\enquote {\bibinfo {title} {Wigner solids,
  classical coulomb lattices, and invariant average potential},}\ }\href
  {\doibase 10.1103/PhysRevB.21.3757} {\bibfield  {journal} {\bibinfo
  {journal} {Phys. Rev. B}\ }\textbf {\bibinfo {volume} {21}},\ \bibinfo
  {pages} {3757--3759} (\bibinfo {year} {1980})}\BibitemShut {NoStop}%
\bibitem [{\citenamefont {Halperin}(1984)}]{Halperin-83}%
  \BibitemOpen
  \bibfield  {author} {\bibinfo {author} {\bibnamefont {Halperin},
  \bibfnamefont {B.~I.}},\ }\bibfield  {title} {\enquote {\bibinfo {title}
  {{Statistics of Quasiparticles and the Hierarchy of Fractional Quantized Hall
  States}},}\ }\href {\doibase 10.1103/PhysRevLett.52.1583} {\bibfield
  {journal} {\bibinfo  {journal} {Phys. Rev. Lett.}\ }\textbf {\bibinfo
  {volume} {52}},\ \bibinfo {pages} {1583--1586} (\bibinfo {year}
  {1984})}\BibitemShut {NoStop}%
\bibitem [{\citenamefont {Halperin}\ and\ \citenamefont
  {Nelson}(1978)}]{HalNel-78}%
  \BibitemOpen
  \bibfield  {author} {\bibinfo {author} {\bibnamefont {Halperin},
  \bibfnamefont {B.~I.}}\ and\ \bibinfo {author} {\bibnamefont {Nelson},
  \bibfnamefont {D.~R.}},\ }\bibfield  {title} {\enquote {\bibinfo {title}
  {Theory of two-dimensional melting},}\ }\href {\doibase
  10.1103/PhysRevLett.41.121} {\bibfield  {journal} {\bibinfo  {journal} {Phys.
  Rev. Lett.}\ }\textbf {\bibinfo {volume} {41}},\ \bibinfo {pages} {121--124}
  (\bibinfo {year} {1978})}\BibitemShut {NoStop}%
\bibitem [{\citenamefont {Hansen}, \citenamefont {Levesque},\ and\
  \citenamefont {Weis}(1979)}]{HanLevWei-79}%
  \BibitemOpen
  \bibfield  {author} {\bibinfo {author} {\bibnamefont {Hansen}, \bibfnamefont
  {J.~P.}}, \bibinfo {author} {\bibnamefont {Levesque}, \bibfnamefont {D.}}, \
  and\ \bibinfo {author} {\bibnamefont {Weis}, \bibfnamefont {J.~J.}},\
  }\bibfield  {title} {\enquote {\bibinfo {title} {Self-diffusion in the
  two-dimensional, classical electron gas},}\ }\href {\doibase
  10.1103/PhysRevLett.43.979} {\bibfield  {journal} {\bibinfo  {journal} {Phys.
  Rev. Lett.}\ }\textbf {\bibinfo {volume} {43}},\ \bibinfo {pages} {979--982}
  (\bibinfo {year} {1979})}\BibitemShut {NoStop}%
\bibitem [{\citenamefont {Haq}, \citenamefont {Pandey},\ and\ \citenamefont
  {Bohigas}(1982)}]{HaqPanBoh-82}%
  \BibitemOpen
  \bibfield  {author} {\bibinfo {author} {\bibnamefont {Haq}, \bibfnamefont
  {R.~U.}}, \bibinfo {author} {\bibnamefont {Pandey}, \bibfnamefont {A.}}, \
  and\ \bibinfo {author} {\bibnamefont {Bohigas}, \bibfnamefont {O.}},\
  }\bibfield  {title} {\enquote {\bibinfo {title} {Fluctuation properties of
  nuclear energy levels: Do theory and experiment agree?}}\ }\href {\doibase
  10.1103/PhysRevLett.48.1086} {\bibfield  {journal} {\bibinfo  {journal}
  {Phys. Rev. Lett.}\ }\textbf {\bibinfo {volume} {48}},\ \bibinfo {pages}
  {1086--1089} (\bibinfo {year} {1982})}\BibitemShut {NoStop}%
\bibitem [{\citenamefont {Hardin}\ \emph {et~al.}(2018)\citenamefont {Hardin},
  \citenamefont {Lebl\'{e}}, \citenamefont {Saff},\ and\ \citenamefont
  {Serfaty}}]{HarLebSafSer-18}%
  \BibitemOpen
  \bibfield  {author} {\bibinfo {author} {\bibnamefont {Hardin}, \bibfnamefont
  {D.~P.}}, \bibinfo {author} {\bibnamefont {Lebl\'{e}}, \bibfnamefont {T.}},
  \bibinfo {author} {\bibnamefont {Saff}, \bibfnamefont {E.~B.}}, \ and\
  \bibinfo {author} {\bibnamefont {Serfaty}, \bibfnamefont {S.}},\ }\bibfield
  {title} {\enquote {\bibinfo {title} {Large deviation principles for
  hypersingular {R}iesz gases},}\ }\href {\doibase 10.1007/s00365-018-9431-9}
  {\bibfield  {journal} {\bibinfo  {journal} {Constr. Approx.}\ }\textbf
  {\bibinfo {volume} {48}},\ \bibinfo {pages} {61--100} (\bibinfo {year}
  {2018})}\BibitemShut {NoStop}%
\bibitem [{\citenamefont {Hardin}, \citenamefont {Michaels},\ and\
  \citenamefont {Saff}(2019)}]{HarMicSaf-19}%
  \BibitemOpen
  \bibfield  {author} {\bibinfo {author} {\bibnamefont {Hardin}, \bibfnamefont
  {D.~P.}}, \bibinfo {author} {\bibnamefont {Michaels}, \bibfnamefont {T.~J.}},
  \ and\ \bibinfo {author} {\bibnamefont {Saff}, \bibfnamefont {E.~B.}},\
  }\bibfield  {title} {\enquote {\bibinfo {title} {Asymptotic linear
  programming lower bounds for the energy of minimizing {R}iesz and {G}auss
  configurations},}\ }\href {\doibase 10.1112/s0025579318000360} {\bibfield
  {journal} {\bibinfo  {journal} {Mathematika}\ }\textbf {\bibinfo {volume}
  {65}},\ \bibinfo {pages} {157--180} (\bibinfo {year} {2019})}\BibitemShut
  {NoStop}%
\bibitem [{\citenamefont {Hardin}\ \emph {et~al.}(2019)\citenamefont {Hardin},
  \citenamefont {Reznikov}, \citenamefont {Saff},\ and\ \citenamefont
  {Volberg}}]{HarRezSafVol-19}%
  \BibitemOpen
  \bibfield  {author} {\bibinfo {author} {\bibnamefont {Hardin}, \bibfnamefont
  {D.~P.}}, \bibinfo {author} {\bibnamefont {Reznikov}, \bibfnamefont {A.}},
  \bibinfo {author} {\bibnamefont {Saff}, \bibfnamefont {E.~B.}}, \ and\
  \bibinfo {author} {\bibnamefont {Volberg}, \bibfnamefont {A.}},\ }\bibfield
  {title} {\enquote {\bibinfo {title} {Local properties of {R}iesz minimal
  energy configurations and equilibrium measures},}\ }\href {\doibase
  10.1093/imrn/rnx262} {\bibfield  {journal} {\bibinfo  {journal} {Int. Math.
  Res. Not. IMRN}\ ,\ \bibinfo {pages} {5066--5086}} (\bibinfo {year}
  {2019})}\BibitemShut {NoStop}%
\bibitem [{\citenamefont {Hardin}\ and\ \citenamefont
  {Saff}(2004)}]{HarSaf-04}%
  \BibitemOpen
  \bibfield  {author} {\bibinfo {author} {\bibnamefont {Hardin}, \bibfnamefont
  {D.~P.}}\ and\ \bibinfo {author} {\bibnamefont {Saff}, \bibfnamefont
  {E.~B.}},\ }\bibfield  {title} {\enquote {\bibinfo {title} {Discretizing
  manifolds via minimum energy points},}\ }\href@noop {} {\bibfield  {journal}
  {\bibinfo  {journal} {Notices of the AMS}\ }\textbf {\bibinfo {volume}
  {51}},\ \bibinfo {pages} {1186--1194} (\bibinfo {year} {2004})}\BibitemShut
  {NoStop}%
\bibitem [{\citenamefont {Hardin}\ and\ \citenamefont
  {Saff}(2005)}]{HarSaf-05}%
  \BibitemOpen
  \bibfield  {author} {\bibinfo {author} {\bibnamefont {Hardin}, \bibfnamefont
  {D.~P.}}\ and\ \bibinfo {author} {\bibnamefont {Saff}, \bibfnamefont
  {E.~B.}},\ }\bibfield  {title} {\enquote {\bibinfo {title} {Minimal {R}iesz
  energy point configurations for rectifiable {$d$}-dimensional manifolds},}\
  }\href {\doibase 10.1016/j.aim.2004.05.006} {\bibfield  {journal} {\bibinfo
  {journal} {Adv. Math.}\ }\textbf {\bibinfo {volume} {193}},\ \bibinfo {pages}
  {174--204} (\bibinfo {year} {2005})}\BibitemShut {NoStop}%
\bibitem [{\citenamefont {Hardin}, \citenamefont {Saff},\ and\ \citenamefont
  {Simanek}(2014)}]{HarSafSim-14}%
  \BibitemOpen
  \bibfield  {author} {\bibinfo {author} {\bibnamefont {Hardin}, \bibfnamefont
  {D.~P.}}, \bibinfo {author} {\bibnamefont {Saff}, \bibfnamefont {E.~B.}}, \
  and\ \bibinfo {author} {\bibnamefont {Simanek}, \bibfnamefont {B.}},\
  }\bibfield  {title} {\enquote {\bibinfo {title} {Periodic discrete energy for
  long-range potentials},}\ }\href {\doibase 10.1063/1.4903975} {\bibfield
  {journal} {\bibinfo  {journal} {J. Math. Phys.}\ }\textbf {\bibinfo {volume}
  {55}},\ \bibinfo {pages} {123509, 27} (\bibinfo {year} {2014})}\BibitemShut
  {NoStop}%
\bibitem [{\citenamefont {Hardin}\ \emph {et~al.}(2017)\citenamefont {Hardin},
  \citenamefont {Saff}, \citenamefont {Simanek},\ and\ \citenamefont
  {Su}}]{HarSafSimSu-17}%
  \BibitemOpen
  \bibfield  {author} {\bibinfo {author} {\bibnamefont {Hardin}, \bibfnamefont
  {D.~P.}}, \bibinfo {author} {\bibnamefont {Saff}, \bibfnamefont {E.~B.}},
  \bibinfo {author} {\bibnamefont {Simanek}, \bibfnamefont {B.~Z.}}, \ and\
  \bibinfo {author} {\bibnamefont {Su}, \bibfnamefont {Y.}},\ }\bibfield
  {title} {\enquote {\bibinfo {title} {Next order energy asymptotics for
  {R}iesz potentials on flat tori},}\ }\href {\doibase 10.1093/imrn/rnw049}
  {\bibfield  {journal} {\bibinfo  {journal} {Int. Math. Res. Not. IMRN}\ ,\
  \bibinfo {pages} {3529--3556}} (\bibinfo {year} {2017})}\BibitemShut
  {NoStop}%
\bibitem [{\citenamefont {Hardin}, \citenamefont {Saff},\ and\ \citenamefont
  {Vlasiuk}(2017)}]{HarSafVla-17}%
  \BibitemOpen
  \bibfield  {author} {\bibinfo {author} {\bibnamefont {Hardin}, \bibfnamefont
  {D.~P.}}, \bibinfo {author} {\bibnamefont {Saff}, \bibfnamefont {E.~B.}}, \
  and\ \bibinfo {author} {\bibnamefont {Vlasiuk}, \bibfnamefont {O.~V.}},\
  }\bibfield  {title} {\enquote {\bibinfo {title} {Generating point
  configurations via hypersingular {R}iesz energy with an external field},}\
  }\href {\doibase 10.1137/16M107414X} {\bibfield  {journal} {\bibinfo
  {journal} {SIAM J. Math. Anal.}\ }\textbf {\bibinfo {volume} {49}},\ \bibinfo
  {pages} {646--673} (\bibinfo {year} {2017})}\BibitemShut {NoStop}%
\bibitem [{\citenamefont {Hardin}, \citenamefont {Saff},\ and\ \citenamefont
  {Whitehouse}(2012)}]{HarSafWhi-12}%
  \BibitemOpen
  \bibfield  {author} {\bibinfo {author} {\bibnamefont {Hardin}, \bibfnamefont
  {D.~P.}}, \bibinfo {author} {\bibnamefont {Saff}, \bibfnamefont {E.~B.}}, \
  and\ \bibinfo {author} {\bibnamefont {Whitehouse}, \bibfnamefont {J.~T.}},\
  }\bibfield  {title} {\enquote {\bibinfo {title} {Quasi-uniformity of minimal
  weighted energy points on compact metric spaces},}\ }\href {\doibase
  10.1016/j.jco.2011.10.009} {\bibfield  {journal} {\bibinfo  {journal} {J.
  Complexity}\ }\textbf {\bibinfo {volume} {28}},\ \bibinfo {pages} {177--191}
  (\bibinfo {year} {2012})}\BibitemShut {NoStop}%
\bibitem [{\citenamefont {Hayashi}\ and\ \citenamefont
  {Tachibana}(1994)}]{HayTac-94}%
  \BibitemOpen
  \bibfield  {author} {\bibinfo {author} {\bibnamefont {Hayashi}, \bibfnamefont
  {Y.}}\ and\ \bibinfo {author} {\bibnamefont {Tachibana}, \bibfnamefont
  {K.}},\ }\bibfield  {title} {\enquote {\bibinfo {title} {Observation of
  {C}oulomb-crystal formation from carbon particles grown in a methane
  plasma},}\ }\href {http://stacks.iop.org/1347-4065/33/i=6A/a=L804} {\bibfield
   {journal} {\bibinfo  {journal} {Jpn J. Appl. Phys.}\ }\textbf {\bibinfo
  {volume} {33}},\ \bibinfo {pages} {L804} (\bibinfo {year}
  {1994})}\BibitemShut {NoStop}%
\bibitem [{\citenamefont {He}\ \emph {et~al.}(2003)\citenamefont {He},
  \citenamefont {Cui}, \citenamefont {Ma}, \citenamefont {Liu},\ and\
  \citenamefont {Zou}}]{HeCuiMaLiuZou-03}%
  \BibitemOpen
  \bibfield  {author} {\bibinfo {author} {\bibnamefont {He}, \bibfnamefont
  {W.~J.}}, \bibinfo {author} {\bibnamefont {Cui}, \bibfnamefont {T.}},
  \bibinfo {author} {\bibnamefont {Ma}, \bibfnamefont {Y.~M.}}, \bibinfo
  {author} {\bibnamefont {Liu}, \bibfnamefont {Z.~M.}}, \ and\ \bibinfo
  {author} {\bibnamefont {Zou}, \bibfnamefont {G.~T.}},\ }\bibfield  {title}
  {\enquote {\bibinfo {title} {Phase transition in a classical two-dimensional
  electron system},}\ }\href {\doibase 10.1103/PhysRevB.68.195104} {\bibfield
  {journal} {\bibinfo  {journal} {Phys. Rev. B}\ }\textbf {\bibinfo {volume}
  {68}},\ \bibinfo {pages} {195104} (\bibinfo {year} {2003})}\BibitemShut
  {NoStop}%
\bibitem [{\citenamefont {Herring}(1952)}]{Herring-52}%
  \BibitemOpen
  \bibfield  {author} {\bibinfo {author} {\bibnamefont {Herring}, \bibfnamefont
  {C.}},\ }\bibfield  {title} {\enquote {\bibinfo {title} {Discussion note on
  {E}wald and {J}uretschke},}\ }in\ \href@noop {} {\emph {\bibinfo {booktitle}
  {Structure and Properties of Solid Surfaces}}},\ \bibinfo {editor} {edited
  by\ \bibinfo {editor} {\bibfnamefont {R.}~\bibnamefont {Gomer}}\ and\
  \bibinfo {editor} {\bibfnamefont {C.}~\bibnamefont {Stanley}}}\ (\bibinfo
  {publisher} {University of Chicago Press},\ \bibinfo {year} {1952})\ p.\
  \bibinfo {pages} {117},\ \bibinfo {note} {a Conference Arranged by The
  National Research Council, Lake Geneva, Wisconsin, Sept 1952}\BibitemShut
  {NoStop}%
\bibitem [{\citenamefont {Hess}\ \emph {et~al.}(1989)\citenamefont {Hess},
  \citenamefont {Robinson}, \citenamefont {Dynes}, \citenamefont {Valles},\
  and\ \citenamefont {Waszczak}}]{Hess_etal-89}%
  \BibitemOpen
  \bibfield  {author} {\bibinfo {author} {\bibnamefont {Hess}, \bibfnamefont
  {H.~F.}}, \bibinfo {author} {\bibnamefont {Robinson}, \bibfnamefont {R.~B.}},
  \bibinfo {author} {\bibnamefont {Dynes}, \bibfnamefont {R.~C.}}, \bibinfo
  {author} {\bibnamefont {Valles}, \bibfnamefont {J.~M.}}, \ and\ \bibinfo
  {author} {\bibnamefont {Waszczak}, \bibfnamefont {J.~V.}},\ }\bibfield
  {title} {\enquote {\bibinfo {title} {Scanning-tunneling-microscope
  observation of the abrikosov flux lattice and the density of states near and
  inside a fluxoid},}\ }\href {\doibase 10.1103/PhysRevLett.62.214} {\bibfield
  {journal} {\bibinfo  {journal} {Phys. Rev. Lett.}\ }\textbf {\bibinfo
  {volume} {62}},\ \bibinfo {pages} {214--216} (\bibinfo {year}
  {1989})}\BibitemShut {NoStop}%
\bibitem [{\citenamefont {Hiraoka}\ \emph {et~al.}(2020)\citenamefont
  {Hiraoka}, \citenamefont {Yang}, \citenamefont {Hagiya}, \citenamefont
  {Niozu}, \citenamefont {Matsuda}, \citenamefont {Huotari}, \citenamefont
  {Holzmann},\ and\ \citenamefont {Ceperley}}]{Hiraoka-etal-20}%
  \BibitemOpen
  \bibfield  {author} {\bibinfo {author} {\bibnamefont {Hiraoka}, \bibfnamefont
  {N.}}, \bibinfo {author} {\bibnamefont {Yang}, \bibfnamefont {Y.}}, \bibinfo
  {author} {\bibnamefont {Hagiya}, \bibfnamefont {T.}}, \bibinfo {author}
  {\bibnamefont {Niozu}, \bibfnamefont {A.}}, \bibinfo {author} {\bibnamefont
  {Matsuda}, \bibfnamefont {K.}}, \bibinfo {author} {\bibnamefont {Huotari},
  \bibfnamefont {S.}}, \bibinfo {author} {\bibnamefont {Holzmann},
  \bibfnamefont {M.}}, \ and\ \bibinfo {author} {\bibnamefont {Ceperley},
  \bibfnamefont {D.~M.}},\ }\bibfield  {title} {\enquote {\bibinfo {title}
  {Direct observation of the momentum distribution and renormalization factor
  in lithium},}\ }\href {\doibase 10.1103/PhysRevB.101.165124} {\bibfield
  {journal} {\bibinfo  {journal} {Phys. Rev. B}\ }\textbf {\bibinfo {volume}
  {101}},\ \bibinfo {pages} {165124} (\bibinfo {year} {2020})}\BibitemShut
  {NoStop}%
\bibitem [{\citenamefont {Hohenberg}\ and\ \citenamefont
  {Kohn}(1964)}]{HohKoh-64}%
  \BibitemOpen
  \bibfield  {author} {\bibinfo {author} {\bibnamefont {Hohenberg},
  \bibfnamefont {P.}}\ and\ \bibinfo {author} {\bibnamefont {Kohn},
  \bibfnamefont {W.}},\ }\bibfield  {title} {\enquote {\bibinfo {title}
  {Inhomogeneous electron gas},}\ }\href {\doibase 10.1103/PhysRev.136.B864}
  {\bibfield  {journal} {\bibinfo  {journal} {Phys. Rev.}\ }\textbf {\bibinfo
  {volume} {136}},\ \bibinfo {pages} {B864--B871} (\bibinfo {year}
  {1964})}\BibitemShut {NoStop}%
\bibitem [{\citenamefont {Holzmann}\ and\ \citenamefont
  {Moroni}(2020)}]{HolMor-20}%
  \BibitemOpen
  \bibfield  {author} {\bibinfo {author} {\bibnamefont {Holzmann},
  \bibfnamefont {M.}}\ and\ \bibinfo {author} {\bibnamefont {Moroni},
  \bibfnamefont {S.}},\ }\bibfield  {title} {\enquote {\bibinfo {title}
  {Itinerant-electron magnetism: The importance of many-body correlations},}\
  }\href {\doibase 10.1103/PhysRevLett.124.206404} {\bibfield  {journal}
  {\bibinfo  {journal} {Phys. Rev. Lett.}\ }\textbf {\bibinfo {volume} {124}},\
  \bibinfo {pages} {206404} (\bibinfo {year} {2020})}\BibitemShut {NoStop}%
\bibitem [{\citenamefont {Hoover}, \citenamefont {Gray},\ and\ \citenamefont
  {Johnson}(1971)}]{HooGraJoh-71}%
  \BibitemOpen
  \bibfield  {author} {\bibinfo {author} {\bibnamefont {Hoover}, \bibfnamefont
  {W.~G.}}, \bibinfo {author} {\bibnamefont {Gray}, \bibfnamefont {S.~G.}}, \
  and\ \bibinfo {author} {\bibnamefont {Johnson}, \bibfnamefont {K.~W.}},\
  }\bibfield  {title} {\enquote {\bibinfo {title} {Thermodynamic properties of
  the fluid and solid phases for inverse power potentials},}\ }\href {\doibase
  10.1063/1.1676196} {\bibfield  {journal} {\bibinfo  {journal} {J. Chem.
  Phys.}\ }\textbf {\bibinfo {volume} {55}},\ \bibinfo {pages} {1128--1136}
  (\bibinfo {year} {1971})}\BibitemShut {NoStop}%
\bibitem [{\citenamefont {Hoover}\ \emph {et~al.}(1970)\citenamefont {Hoover},
  \citenamefont {Ross}, \citenamefont {Johnson}, \citenamefont {Henderson},
  \citenamefont {Barker},\ and\ \citenamefont {Brown}}]{HooRosJohHenBarBro-70}%
  \BibitemOpen
  \bibfield  {author} {\bibinfo {author} {\bibnamefont {Hoover}, \bibfnamefont
  {W.~G.}}, \bibinfo {author} {\bibnamefont {Ross}, \bibfnamefont {M.}},
  \bibinfo {author} {\bibnamefont {Johnson}, \bibfnamefont {K.~W.}}, \bibinfo
  {author} {\bibnamefont {Henderson}, \bibfnamefont {D.}}, \bibinfo {author}
  {\bibnamefont {Barker}, \bibfnamefont {J.~A.}}, \ and\ \bibinfo {author}
  {\bibnamefont {Brown}, \bibfnamefont {B.~C.}},\ }\bibfield  {title} {\enquote
  {\bibinfo {title} {Soft-sphere equation of state},}\ }\href {\doibase
  10.1063/1.1672728} {\bibfield  {journal} {\bibinfo  {journal} {J. Chem.
  Phys.}\ }\textbf {\bibinfo {volume} {52}},\ \bibinfo {pages} {4931--4941}
  (\bibinfo {year} {1970})}\BibitemShut {NoStop}%
\bibitem [{\citenamefont {Hoover}, \citenamefont {Young},\ and\ \citenamefont
  {Grover}(1972)}]{HooYouGro-72}%
  \BibitemOpen
  \bibfield  {author} {\bibinfo {author} {\bibnamefont {Hoover}, \bibfnamefont
  {W.~G.}}, \bibinfo {author} {\bibnamefont {Young}, \bibfnamefont {D.~A.}}, \
  and\ \bibinfo {author} {\bibnamefont {Grover}, \bibfnamefont {R.}},\
  }\bibfield  {title} {\enquote {\bibinfo {title} {Statistical mechanics of
  phase diagrams. {I}. {I}nverse power potentials and the close-packed to
  body-centered cubic transition},}\ }\href {\doibase 10.1063/1.1677521}
  {\bibfield  {journal} {\bibinfo  {journal} {J. Chem. Phys.}\ }\textbf
  {\bibinfo {volume} {56}},\ \bibinfo {pages} {2207--2210} (\bibinfo {year}
  {1972})}\BibitemShut {NoStop}%
\bibitem [{\citenamefont {van Hove}(1949)}]{vanHove-49}%
  \BibitemOpen
  \bibfield  {author} {\bibinfo {author} {\bibnamefont {van Hove},
  \bibfnamefont {L.}},\ }\bibfield  {title} {\enquote {\bibinfo {title}
  {Quelques propri{\'e}t{\'e}s g{\'e}n{\'e}rales de l'int{\'e}grale de
  configuration d'un syst{\`e}me de particules avec interaction},}\ }\href
  {\doibase 10.1016/0031-8914(49)90059-2} {\bibfield  {journal} {\bibinfo
  {journal} {Physica}\ }\textbf {\bibinfo {volume} {15}},\ \bibinfo {pages}
  {951--961} (\bibinfo {year} {1949})}\BibitemShut {NoStop}%
\bibitem [{\citenamefont {van Hove}(1950)}]{vanHove-50}%
  \BibitemOpen
  \bibfield  {author} {\bibinfo {author} {\bibnamefont {van Hove},
  \bibfnamefont {L.}},\ }\bibfield  {title} {\enquote {\bibinfo {title} {Sur
  l'int{\'e}grale de configuration pour les syst{\`e}mes de particules {\`a}
  une dimension},}\ }\href@noop {} {\bibfield  {journal} {\bibinfo  {journal}
  {Physica}\ }\textbf {\bibinfo {volume} {16}},\ \bibinfo {pages} {137--143}
  (\bibinfo {year} {1950})}\BibitemShut {NoStop}%
\bibitem [{\citenamefont {Hughes}(2006)}]{Hugues-06}%
  \BibitemOpen
  \bibfield  {author} {\bibinfo {author} {\bibnamefont {Hughes}, \bibfnamefont
  {R.}},\ }\bibfield  {title} {\enquote {\bibinfo {title} {Theoretical
  practice: the {B}ohm-{P}ines quartet},}\ }\href {\doibase
  10.1162/posc.2006.14.4.457} {\bibfield  {journal} {\bibinfo  {journal}
  {Perspect. Sci.}\ }\textbf {\bibinfo {volume} {14}},\ \bibinfo {pages}
  {457--524} (\bibinfo {year} {2006})}\BibitemShut {NoStop}%
\bibitem [{\citenamefont {Hughes}(1985)}]{Hughes-85}%
  \BibitemOpen
  \bibfield  {author} {\bibinfo {author} {\bibnamefont {Hughes}, \bibfnamefont
  {W.}},\ }\bibfield  {title} {\enquote {\bibinfo {title} {Thermodynamics for
  {C}oulomb systems: a problem at vanishing particle densities},}\ }\href@noop
  {} {\bibfield  {journal} {\bibinfo  {journal} {J. Statist. Phys.}\ }\textbf
  {\bibinfo {volume} {41}},\ \bibinfo {pages} {975--1013} (\bibinfo {year}
  {1985})}\BibitemShut {NoStop}%
\bibitem [{\citenamefont {Huotari}\ \emph {et~al.}(2010)\citenamefont
  {Huotari}, \citenamefont {Soininen}, \citenamefont {Pylkk{\"a}nen},
  \citenamefont {H{\"a}m{\"a}l{\"a}inen}, \citenamefont {Issolah},
  \citenamefont {Titov}, \citenamefont {McMinis}, \citenamefont {Kim},
  \citenamefont {Esler}, \citenamefont {Ceperley}, \citenamefont {Holzmann},\
  and\ \citenamefont {Olevano}}]{Huotari-etal-10}%
  \BibitemOpen
  \bibfield  {author} {\bibinfo {author} {\bibnamefont {Huotari}, \bibfnamefont
  {S.}}, \bibinfo {author} {\bibnamefont {Soininen}, \bibfnamefont {J.~A.}},
  \bibinfo {author} {\bibnamefont {Pylkk{\"a}nen}, \bibfnamefont {T.}},
  \bibinfo {author} {\bibnamefont {H{\"a}m{\"a}l{\"a}inen}, \bibfnamefont
  {K.}}, \bibinfo {author} {\bibnamefont {Issolah}, \bibfnamefont {A.}},
  \bibinfo {author} {\bibnamefont {Titov}, \bibfnamefont {A.}}, \bibinfo
  {author} {\bibnamefont {McMinis}, \bibfnamefont {J.}}, \bibinfo {author}
  {\bibnamefont {Kim}, \bibfnamefont {J.}}, \bibinfo {author} {\bibnamefont
  {Esler}, \bibfnamefont {K.}}, \bibinfo {author} {\bibnamefont {Ceperley},
  \bibfnamefont {D.~M.}}, \bibinfo {author} {\bibnamefont {Holzmann},
  \bibfnamefont {M.}}, \ and\ \bibinfo {author} {\bibnamefont {Olevano},
  \bibfnamefont {V.}},\ }\bibfield  {title} {\enquote {\bibinfo {title}
  {Momentum distribution and renormalization factor in sodium and the electron
  gas},}\ }\href {\doibase 10.1103/PhysRevLett.105.086403} {\bibfield
  {journal} {\bibinfo  {journal} {Phys. Rev. Lett.}\ }\textbf {\bibinfo
  {volume} {105}},\ \bibinfo {pages} {086403} (\bibinfo {year}
  {2010})}\BibitemShut {NoStop}%
\bibitem [{\citenamefont {Ihm}\ and\ \citenamefont {Cohen}(1980)}]{IhmCoh-80}%
  \BibitemOpen
  \bibfield  {author} {\bibinfo {author} {\bibnamefont {Ihm}, \bibfnamefont
  {J.}}\ and\ \bibinfo {author} {\bibnamefont {Cohen}, \bibfnamefont {M.~L.}},\
  }\bibfield  {title} {\enquote {\bibinfo {title} {Comment on "{C}orrection to
  {F}uchs' calculation of the electrostatic energy of a {W}igner solid"},}\
  }\href {\doibase 10.1103/PhysRevB.21.3754} {\bibfield  {journal} {\bibinfo
  {journal} {Phys. Rev. B}\ }\textbf {\bibinfo {volume} {21}},\ \bibinfo
  {pages} {3754--3756} (\bibinfo {year} {1980})}\BibitemShut {NoStop}%
\bibitem [{\citenamefont {Imbrie}(1982)}]{Imbrie-82}%
  \BibitemOpen
  \bibfield  {author} {\bibinfo {author} {\bibnamefont {Imbrie}, \bibfnamefont
  {J.~Z.}},\ }\bibfield  {title} {\enquote {\bibinfo {title} {Debye screening
  for jellium and other {C}oulomb systems},}\ }\href
  {http://projecteuclid.org/euclid.cmp/1103922133} {\bibfield  {journal}
  {\bibinfo  {journal} {Comm. Math. Phys.}\ }\textbf {\bibinfo {volume} {87}},\
  \bibinfo {pages} {515--565} (\bibinfo {year} {1982})}\BibitemShut {NoStop}%
\bibitem [{\citenamefont {Isakov}(1994)}]{Isakov-94}%
  \BibitemOpen
  \bibfield  {author} {\bibinfo {author} {\bibnamefont {Isakov}, \bibfnamefont
  {S.~B.}},\ }\bibfield  {title} {\enquote {\bibinfo {title} {Statistical
  mechanics for a class of quantum statistics},}\ }\href {\doibase
  10.1103/PhysRevLett.73.2150} {\bibfield  {journal} {\bibinfo  {journal}
  {Phys. Rev. Lett.}\ }\textbf {\bibinfo {volume} {73}},\ \bibinfo {pages}
  {2150--2153} (\bibinfo {year} {1994})}\BibitemShut {NoStop}%
\bibitem [{\citenamefont {Ivi\'{c}}\ \emph {et~al.}(2006)\citenamefont
  {Ivi\'{c}}, \citenamefont {Kr\"{a}tzel}, \citenamefont {K\"{u}hleitner},\
  and\ \citenamefont {Nowak}}]{IviKraKuhNow-06}%
  \BibitemOpen
  \bibfield  {author} {\bibinfo {author} {\bibnamefont {Ivi\'{c}},
  \bibfnamefont {A.}}, \bibinfo {author} {\bibnamefont {Kr\"{a}tzel},
  \bibfnamefont {E.}}, \bibinfo {author} {\bibnamefont {K\"{u}hleitner},
  \bibfnamefont {M.}}, \ and\ \bibinfo {author} {\bibnamefont {Nowak},
  \bibfnamefont {W.~G.}},\ }\bibfield  {title} {\enquote {\bibinfo {title}
  {Lattice points in large regions and related arithmetic functions: recent
  developments in a very classic topic},}\ }in\ \href@noop {} {\emph {\bibinfo
  {booktitle} {Elementare und analytische {Z}ahlentheorie}}},\ \bibinfo
  {series} {Schr. Wiss. Ges. Johann Wolfgang Goethe Univ. Frankfurt am Main},
  Vol.~\bibinfo {volume} {20}\ (\bibinfo  {publisher} {Franz Steiner Verlag
  Stuttgart, Stuttgart},\ \bibinfo {year} {2006})\ pp.\ \bibinfo {pages}
  {89--128}\BibitemShut {NoStop}%
\bibitem [{\citenamefont {Jagannath}\ and\ \citenamefont
  {Trogdon}(2017)}]{JagTro-17}%
  \BibitemOpen
  \bibfield  {author} {\bibinfo {author} {\bibnamefont {Jagannath},
  \bibfnamefont {A.}}\ and\ \bibinfo {author} {\bibnamefont {Trogdon},
  \bibfnamefont {T.}},\ }\bibfield  {title} {\enquote {\bibinfo {title} {Random
  matrices and the {N}ew {Y}ork {C}ity subway system},}\ }\href {\doibase
  10.1103/PhysRevE.96.030101} {\bibfield  {journal} {\bibinfo  {journal} {Phys.
  Rev. E}\ }\textbf {\bibinfo {volume} {96}},\ \bibinfo {pages} {030101}
  (\bibinfo {year} {2017})}\BibitemShut {NoStop}%
\bibitem [{\citenamefont {Jancovici}(1981)}]{Jancovici-81}%
  \BibitemOpen
  \bibfield  {author} {\bibinfo {author} {\bibnamefont {Jancovici},
  \bibfnamefont {B.}},\ }\bibfield  {title} {\enquote {\bibinfo {title} {Exact
  results for the two-dimensional one-component plasma},}\ }\href {\doibase
  10.1103/PhysRevLett.46.386} {\bibfield  {journal} {\bibinfo  {journal} {Phys.
  Rev. Lett.}\ }\textbf {\bibinfo {volume} {46}},\ \bibinfo {pages} {386--388}
  (\bibinfo {year} {1981})}\BibitemShut {NoStop}%
\bibitem [{\citenamefont {Jancovici}\ and\ \citenamefont
  {Lebowitz}(2001)}]{JanLeb-01}%
  \BibitemOpen
  \bibfield  {author} {\bibinfo {author} {\bibnamefont {Jancovici},
  \bibfnamefont {B.}}\ and\ \bibinfo {author} {\bibnamefont {Lebowitz},
  \bibfnamefont {J.~L.}},\ }\bibfield  {title} {\enquote {\bibinfo {title}
  {Bounded fluctuations and translation symmetry breaking: a solvable model},}\
  }\href {\doibase 10.1023/A:1010349517967} {\bibfield  {journal} {\bibinfo
  {journal} {J. Statist. Phys.}\ }\textbf {\bibinfo {volume} {103}},\ \bibinfo
  {pages} {619--624} (\bibinfo {year} {2001})},\ \bibinfo {note} {special issue
  dedicated to the memory of Joaquin M. Luttinger}\BibitemShut {NoStop}%
\bibitem [{\citenamefont {Jancovici}, \citenamefont {Lebowitz},\ and\
  \citenamefont {Manificat}(1993)}]{JanLebMag-93}%
  \BibitemOpen
  \bibfield  {author} {\bibinfo {author} {\bibnamefont {Jancovici},
  \bibfnamefont {B.}}, \bibinfo {author} {\bibnamefont {Lebowitz},
  \bibfnamefont {J.~L.}}, \ and\ \bibinfo {author} {\bibnamefont {Manificat},
  \bibfnamefont {G.}},\ }\bibfield  {title} {\enquote {\bibinfo {title} {Large
  charge fluctuations in classical {C}oulomb systems},}\ }\href {\doibase
  10.1007/BF01048032} {\bibfield  {journal} {\bibinfo  {journal} {J. Statist.
  Phys.}\ }\textbf {\bibinfo {volume} {72}},\ \bibinfo {pages} {773--787}
  (\bibinfo {year} {1993})}\BibitemShut {NoStop}%
\bibitem [{\citenamefont {Jancovici}\ and\ \citenamefont
  {T\'{e}llez}(1998)}]{JanTel-98}%
  \BibitemOpen
  \bibfield  {author} {\bibinfo {author} {\bibnamefont {Jancovici},
  \bibfnamefont {B.}}\ and\ \bibinfo {author} {\bibnamefont {T\'{e}llez},
  \bibfnamefont {G.}},\ }\bibfield  {title} {\enquote {\bibinfo {title}
  {Two-dimensional {C}oulomb systems on a surface of constant negative
  curvature},}\ }\href {\doibase 10.1023/A:1023079916489} {\bibfield  {journal}
  {\bibinfo  {journal} {J. Statist. Phys.}\ }\textbf {\bibinfo {volume} {91}},\
  \bibinfo {pages} {953--977} (\bibinfo {year} {1998})}\BibitemShut {NoStop}%
\bibitem [{\citenamefont {Jansen}\ and\ \citenamefont
  {Jung}(2014)}]{JanJun-14}%
  \BibitemOpen
  \bibfield  {author} {\bibinfo {author} {\bibnamefont {Jansen}, \bibfnamefont
  {S.}}\ and\ \bibinfo {author} {\bibnamefont {Jung}, \bibfnamefont {P.}},\
  }\bibfield  {title} {\enquote {\bibinfo {title} {Wigner crystallization in
  the quantum 1d jellium at all densities},}\ }\href {\doibase
  10.1007/s00220-014-2032-y} {\bibfield  {journal} {\bibinfo  {journal} {Comm.
  Math. Phys.}\ ,\ \bibinfo {pages} {1--22}} (\bibinfo {year}
  {2014})}\BibitemShut {NoStop}%
\bibitem [{\citenamefont {Jansen}, \citenamefont {Lieb},\ and\ \citenamefont
  {Seiler}(2009)}]{JanLieSei-09}%
  \BibitemOpen
  \bibfield  {author} {\bibinfo {author} {\bibnamefont {Jansen}, \bibfnamefont
  {S.}}, \bibinfo {author} {\bibnamefont {Lieb}, \bibfnamefont {E.~H.}}, \ and\
  \bibinfo {author} {\bibnamefont {Seiler}, \bibfnamefont {R.}},\ }\bibfield
  {title} {\enquote {\bibinfo {title} {Symmetry breaking in {L}aughlin's state
  on a cylinder},}\ }\href {\doibase 10.1007/s00220-008-0576-4} {\bibfield
  {journal} {\bibinfo  {journal} {Comm. Math. Phys.}\ }\textbf {\bibinfo
  {volume} {285}},\ \bibinfo {pages} {503--535} (\bibinfo {year}
  {2009})}\BibitemShut {NoStop}%
\bibitem [{\citenamefont {Jones}\ and\ \citenamefont
  {Ceperley}(1996)}]{JonCep-96}%
  \BibitemOpen
  \bibfield  {author} {\bibinfo {author} {\bibnamefont {Jones}, \bibfnamefont
  {M.~D.}}\ and\ \bibinfo {author} {\bibnamefont {Ceperley}, \bibfnamefont
  {D.~M.}},\ }\bibfield  {title} {\enquote {\bibinfo {title} {{Crystallization
  of the One-Component Plasma at Finite Temperature}},}\ }\href {\doibase
  10.1103/PhysRevLett.76.4572} {\bibfield  {journal} {\bibinfo  {journal}
  {Phys. Rev. Lett.}\ }\textbf {\bibinfo {volume} {76}},\ \bibinfo {pages}
  {4572--4575} (\bibinfo {year} {1996})}\BibitemShut {NoStop}%
\bibitem [{\citenamefont {Kapfer}\ and\ \citenamefont
  {Krauth}(2015)}]{KapKra-15}%
  \BibitemOpen
  \bibfield  {author} {\bibinfo {author} {\bibnamefont {Kapfer}, \bibfnamefont
  {S.~C.}}\ and\ \bibinfo {author} {\bibnamefont {Krauth}, \bibfnamefont
  {W.}},\ }\bibfield  {title} {\enquote {\bibinfo {title} {Two-dimensional
  melting: From liquid-hexatic coexistence to continuous transitions},}\ }\href
  {\doibase 10.1103/PhysRevLett.114.035702} {\bibfield  {journal} {\bibinfo
  {journal} {Phys. Rev. Lett.}\ }\textbf {\bibinfo {volume} {114}},\ \bibinfo
  {pages} {035702} (\bibinfo {year} {2015})}\BibitemShut {NoStop}%
\bibitem [{\citenamefont {Kato}(1995)}]{Kato}%
  \BibitemOpen
  \bibfield  {author} {\bibinfo {author} {\bibnamefont {Kato}, \bibfnamefont
  {T.}},\ }\href@noop {} {\emph {\bibinfo {title} {Perturbation theory for
  linear operators}}},\ \bibinfo {edition} {2nd}\ ed.\ (\bibinfo  {publisher}
  {Springer},\ \bibinfo {year} {1995})\BibitemShut {NoStop}%
\bibitem [{\citenamefont {Katz}\ and\ \citenamefont
  {Duneau}(1984)}]{KatDun-84}%
  \BibitemOpen
  \bibfield  {author} {\bibinfo {author} {\bibnamefont {Katz}, \bibfnamefont
  {A.}}\ and\ \bibinfo {author} {\bibnamefont {Duneau}, \bibfnamefont {M.}},\
  }\bibfield  {title} {\enquote {\bibinfo {title} {The convergence of the
  one-dimensional ground states to an infinite lattice},}\ }\href {\doibase
  10.1007/BF01012914} {\bibfield  {journal} {\bibinfo  {journal} {J. Statist.
  Phys.}\ }\textbf {\bibinfo {volume} {37}},\ \bibinfo {pages} {257--268}
  (\bibinfo {year} {1984})}\BibitemShut {NoStop}%
\bibitem [{\citenamefont {Kethepalli}\ \emph {et~al.}(2021)\citenamefont
  {Kethepalli}, \citenamefont {Kulkarni}, \citenamefont {Kundu}, \citenamefont
  {Majumdar}, \citenamefont {Mukamel},\ and\ \citenamefont
  {Schehr}}]{KetKulKunMajMukSch-21}%
  \BibitemOpen
  \bibfield  {author} {\bibinfo {author} {\bibnamefont {Kethepalli},
  \bibfnamefont {J.}}, \bibinfo {author} {\bibnamefont {Kulkarni},
  \bibfnamefont {M.}}, \bibinfo {author} {\bibnamefont {Kundu}, \bibfnamefont
  {A.}}, \bibinfo {author} {\bibnamefont {Majumdar}, \bibfnamefont {S.~N.}},
  \bibinfo {author} {\bibnamefont {Mukamel}, \bibfnamefont {D.}}, \ and\
  \bibinfo {author} {\bibnamefont {Schehr}, \bibfnamefont {G.}},\ }\bibfield
  {title} {\enquote {\bibinfo {title} {Harmonically confined long-ranged
  interacting gas in the presence of a hard wall},}\ }\href {\doibase
  10.1088/1742-5468/ac2896} {\bibfield  {journal} {\bibinfo  {journal} {J.
  Stat. Mech. Theory Exp.}\ ,\ \bibinfo {pages} {Paper No. 103209, 36}}
  (\bibinfo {year} {2021})}\BibitemShut {NoStop}%
\bibitem [{\citenamefont {Kethepalli}\ \emph {et~al.}(2022)\citenamefont
  {Kethepalli}, \citenamefont {Kulkarni}, \citenamefont {Kundu}, \citenamefont
  {Majumdar}, \citenamefont {Mukamel},\ and\ \citenamefont
  {Schehr}}]{KetKulKunMajMukSch-22}%
  \BibitemOpen
  \bibfield  {author} {\bibinfo {author} {\bibnamefont {Kethepalli},
  \bibfnamefont {J.}}, \bibinfo {author} {\bibnamefont {Kulkarni},
  \bibfnamefont {M.}}, \bibinfo {author} {\bibnamefont {Kundu}, \bibfnamefont
  {A.}}, \bibinfo {author} {\bibnamefont {Majumdar}, \bibfnamefont {S.~N.}},
  \bibinfo {author} {\bibnamefont {Mukamel}, \bibfnamefont {D.}}, \ and\
  \bibinfo {author} {\bibnamefont {Schehr}, \bibfnamefont {G.}},\ }\bibfield
  {title} {\enquote {\bibinfo {title} {Edge fluctuations and third-order phase
  transition in harmonically confined long-range systems},}\ }\href {\doibase
  10.1088/1742-5468/ac52b2} {\bibfield  {journal} {\bibinfo  {journal} {J.
  Stat. Mech. Theory Exp.}\ }\textbf {\bibinfo {volume} {2022}},\ \bibinfo
  {pages} {033203} (\bibinfo {year} {2022})}\BibitemShut {NoStop}%
\bibitem [{\citenamefont {Kiessling}(1989)}]{Kiessling-89}%
  \BibitemOpen
  \bibfield  {author} {\bibinfo {author} {\bibnamefont {Kiessling},
  \bibfnamefont {M.~K.~H.}},\ }\bibfield  {title} {\enquote {\bibinfo {title}
  {On the equilibrium statistical mechanics of isothermal classical
  self-gravitating matter},}\ }\href {\doibase 10.1007/BF01042598} {\bibfield
  {journal} {\bibinfo  {journal} {J. Statist. Phys.}\ }\textbf {\bibinfo
  {volume} {55}},\ \bibinfo {pages} {203--257} (\bibinfo {year}
  {1989})}\BibitemShut {NoStop}%
\bibitem [{\citenamefont {Kiessling}(1993)}]{Kiessling-93}%
  \BibitemOpen
  \bibfield  {author} {\bibinfo {author} {\bibnamefont {Kiessling},
  \bibfnamefont {M.~K.-H.}},\ }\bibfield  {title} {\enquote {\bibinfo {title}
  {Statistical mechanics of classical particles with logarithmic
  interactions},}\ }\href {\doibase 10.1002/cpa.3160460103} {\bibfield
  {journal} {\bibinfo  {journal} {Comm. Pure. Appl. Math.}\ }\textbf {\bibinfo
  {volume} {46}},\ \bibinfo {pages} {27--56} (\bibinfo {year}
  {1993})}\BibitemShut {NoStop}%
\bibitem [{\citenamefont {Kiessling}\ and\ \citenamefont
  {Spohn}(1999)}]{KieSpo-99}%
  \BibitemOpen
  \bibfield  {author} {\bibinfo {author} {\bibnamefont {Kiessling},
  \bibfnamefont {M.~K.-H.}}\ and\ \bibinfo {author} {\bibnamefont {Spohn},
  \bibfnamefont {H.}},\ }\bibfield  {title} {\enquote {\bibinfo {title} {A note
  on the eigenvalue density of random matrices},}\ }\href {\doibase
  10.1007/s002200050516} {\bibfield  {journal} {\bibinfo  {journal} {Comm.
  Math. Phys.}\ }\textbf {\bibinfo {volume} {199}},\ \bibinfo {pages}
  {683--695} (\bibinfo {year} {1999})}\BibitemShut {NoStop}%
\bibitem [{\citenamefont {Killip}\ and\ \citenamefont
  {Nenciu}(2004)}]{KilNen-04}%
  \BibitemOpen
  \bibfield  {author} {\bibinfo {author} {\bibnamefont {Killip}, \bibfnamefont
  {R.}}\ and\ \bibinfo {author} {\bibnamefont {Nenciu}, \bibfnamefont {I.}},\
  }\bibfield  {title} {\enquote {\bibinfo {title} {Matrix models for circular
  ensembles},}\ }\href {\doibase 10.1155/S1073792804141597} {\bibfield
  {journal} {\bibinfo  {journal} {Int. Math. Res. Not.}\ ,\ \bibinfo {pages}
  {2665--2701}} (\bibinfo {year} {2004})}\BibitemShut {NoStop}%
\bibitem [{\citenamefont {Killip}\ and\ \citenamefont
  {Stoiciu}(2009)}]{KilSto-09}%
  \BibitemOpen
  \bibfield  {author} {\bibinfo {author} {\bibnamefont {Killip}, \bibfnamefont
  {R.}}\ and\ \bibinfo {author} {\bibnamefont {Stoiciu}, \bibfnamefont {M.}},\
  }\bibfield  {title} {\enquote {\bibinfo {title} {Eigenvalue statistics for
  {CMV} matrices: from {P}oisson to clock via random matrix ensembles},}\
  }\href {\doibase 10.1215/00127094-2009-001} {\bibfield  {journal} {\bibinfo
  {journal} {Duke Math. J.}\ }\textbf {\bibinfo {volume} {146}},\ \bibinfo
  {pages} {361--399} (\bibinfo {year} {2009})}\BibitemShut {NoStop}%
\bibitem [{\citenamefont {{Klevtsov}}(2016)}]{Klevtsov-16}%
  \BibitemOpen
  \bibfield  {author} {\bibinfo {author} {\bibnamefont {{Klevtsov}},
  \bibfnamefont {S.}},\ }\bibfield  {title} {\enquote {\bibinfo {title}
  {{Geometry and large $N$ limits in Laughlin states}},}\ }in\ \href@noop {}
  {\emph {\bibinfo {booktitle} {{Geometry and quantization. Lectures presented
  at the 6th school GEOQUANT, ICMAT, Madrid, Spain, September 7--18, 2015}}}}\
  (\bibinfo  {publisher} {Luxembourg: University of Luxembourg, Faculty of
  Science, Technology and Communication},\ \bibinfo {year} {2016})\ pp.\
  \bibinfo {pages} {63--127}\BibitemShut {NoStop}%
\bibitem [{\citenamefont {Klevtsov}\ \emph {et~al.}(2017)\citenamefont
  {Klevtsov}, \citenamefont {Ma}, \citenamefont {Marinescu},\ and\
  \citenamefont {Wiegmann}}]{KleXiaGeoWie-17}%
  \BibitemOpen
  \bibfield  {author} {\bibinfo {author} {\bibnamefont {Klevtsov},
  \bibfnamefont {S.}}, \bibinfo {author} {\bibnamefont {Ma}, \bibfnamefont
  {X.}}, \bibinfo {author} {\bibnamefont {Marinescu}, \bibfnamefont {G.}}, \
  and\ \bibinfo {author} {\bibnamefont {Wiegmann}, \bibfnamefont {P.}},\
  }\bibfield  {title} {\enquote {\bibinfo {title} {Quantum {H}all effect and
  {Q}uillen metric},}\ }\href {\doibase 10.1007/s00220-016-2789-2} {\bibfield
  {journal} {\bibinfo  {journal} {Comm. Math. Phys.}\ }\textbf {\bibinfo
  {volume} {349}},\ \bibinfo {pages} {819--855} (\bibinfo {year}
  {2017})}\BibitemShut {NoStop}%
\bibitem [{\citenamefont {Knighton}\ \emph {et~al.}(2018)\citenamefont
  {Knighton}, \citenamefont {Wu}, \citenamefont {Huang}, \citenamefont
  {Serafin}, \citenamefont {Xia}, \citenamefont {Pfeiffer},\ and\ \citenamefont
  {West}}]{KniWuHuaSerXiaPfeWes-18}%
  \BibitemOpen
  \bibfield  {author} {\bibinfo {author} {\bibnamefont {Knighton},
  \bibfnamefont {T.}}, \bibinfo {author} {\bibnamefont {Wu}, \bibfnamefont
  {Z.}}, \bibinfo {author} {\bibnamefont {Huang}, \bibfnamefont {J.}}, \bibinfo
  {author} {\bibnamefont {Serafin}, \bibfnamefont {A.}}, \bibinfo {author}
  {\bibnamefont {Xia}, \bibfnamefont {J.~S.}}, \bibinfo {author} {\bibnamefont
  {Pfeiffer}, \bibfnamefont {L.~N.}}, \ and\ \bibinfo {author} {\bibnamefont
  {West}, \bibfnamefont {K.~W.}},\ }\bibfield  {title} {\enquote {\bibinfo
  {title} {Evidence of two-stage melting of {W}igner solids},}\ }\href
  {\doibase 10.1103/PhysRevB.97.085135} {\bibfield  {journal} {\bibinfo
  {journal} {Phys. Rev. B}\ }\textbf {\bibinfo {volume} {97}},\ \bibinfo
  {pages} {085135} (\bibinfo {year} {2018})}\BibitemShut {NoStop}%
\bibitem [{\citenamefont {Kohn}\ and\ \citenamefont {Sham}(1965)}]{KohSha-65}%
  \BibitemOpen
  \bibfield  {author} {\bibinfo {author} {\bibnamefont {Kohn}, \bibfnamefont
  {W.}}\ and\ \bibinfo {author} {\bibnamefont {Sham}, \bibfnamefont {L.~J.}},\
  }\bibfield  {title} {\enquote {\bibinfo {title} {Self-consistent equations
  including exchange and correlation effects},}\ }\href {\doibase
  10.1103/PhysRev.140.A1133} {\bibfield  {journal} {\bibinfo  {journal} {Phys.
  Rev. (2)}\ }\textbf {\bibinfo {volume} {140}},\ \bibinfo {pages}
  {A1133--A1138} (\bibinfo {year} {1965})}\BibitemShut {NoStop}%
\bibitem [{\citenamefont {Kosterlitz}\ and\ \citenamefont
  {Thouless}(1972)}]{KosTho-72}%
  \BibitemOpen
  \bibfield  {author} {\bibinfo {author} {\bibnamefont {Kosterlitz},
  \bibfnamefont {J.~M.}}\ and\ \bibinfo {author} {\bibnamefont {Thouless},
  \bibfnamefont {D.~J.}},\ }\bibfield  {title} {\enquote {\bibinfo {title}
  {Long range order and metastability in two dimensional solids and
  superfluids. ({A}pplication of dislocation theory)},}\ }\href
  {http://stacks.iop.org/0022-3719/5/i=11/a=002} {\bibfield  {journal}
  {\bibinfo  {journal} {J. Phys. C: Solid State Phys}\ }\textbf {\bibinfo
  {volume} {5}},\ \bibinfo {pages} {L124} (\bibinfo {year} {1972})}\BibitemShut
  {NoStop}%
\bibitem [{\citenamefont {Kosterlitz}\ and\ \citenamefont
  {Thouless}(1973)}]{KosTho-73}%
  \BibitemOpen
  \bibfield  {author} {\bibinfo {author} {\bibnamefont {Kosterlitz},
  \bibfnamefont {J.~M.}}\ and\ \bibinfo {author} {\bibnamefont {Thouless},
  \bibfnamefont {D.~J.}},\ }\bibfield  {title} {\enquote {\bibinfo {title}
  {Ordering, metastability and phase transitions in two-dimensional systems},}\
  }\href {http://stacks.iop.org/0022-3719/6/i=7/a=010} {\bibfield  {journal}
  {\bibinfo  {journal} {J. Phys. C: Solid State Phys}\ }\textbf {\bibinfo
  {volume} {6}},\ \bibinfo {pages} {1181} (\bibinfo {year} {1973})}\BibitemShut
  {NoStop}%
\bibitem [{\citenamefont {Krb{\'{a}}lek}\ and\ \citenamefont
  {\v{S}eba}(2000)}]{KrbSeb-00}%
  \BibitemOpen
  \bibfield  {author} {\bibinfo {author} {\bibnamefont {Krb{\'{a}}lek},
  \bibfnamefont {M.}}\ and\ \bibinfo {author} {\bibnamefont {\v{S}eba},
  \bibfnamefont {P.}},\ }\bibfield  {title} {\enquote {\bibinfo {title} {The
  statistical properties of the city transport in {C}uernavaca ({M}exico) and
  random matrix ensembles},}\ }\href {\doibase 10.1088/0305-4470/33/26/102}
  {\bibfield  {journal} {\bibinfo  {journal} {J. Phys. A, Math. Gen.}\ }\textbf
  {\bibinfo {volume} {33}},\ \bibinfo {pages} {L229--L234} (\bibinfo {year}
  {2000})}\BibitemShut {NoStop}%
\bibitem [{\citenamefont {Kuijlaars}\ and\ \citenamefont {Mi\~{n}a
  D\'{\i}az}(2019)}]{KuiMin-19}%
  \BibitemOpen
  \bibfield  {author} {\bibinfo {author} {\bibnamefont {Kuijlaars},
  \bibfnamefont {A.~B.~J.}}\ and\ \bibinfo {author} {\bibnamefont {Mi\~{n}a
  D\'{\i}az}, \bibfnamefont {E.}},\ }\bibfield  {title} {\enquote {\bibinfo
  {title} {Universality for conditional measures of the sine point process},}\
  }\href {\doibase 10.1016/j.jat.2019.03.002} {\bibfield  {journal} {\bibinfo
  {journal} {J. Approx. Theory}\ }\textbf {\bibinfo {volume} {243}},\ \bibinfo
  {pages} {1--24} (\bibinfo {year} {2019})}\BibitemShut {NoStop}%
\bibitem [{\citenamefont {Kuijlaars}\ and\ \citenamefont
  {Saff}(1998)}]{KuiSaf-98}%
  \BibitemOpen
  \bibfield  {author} {\bibinfo {author} {\bibnamefont {Kuijlaars},
  \bibfnamefont {A.~B.~J.}}\ and\ \bibinfo {author} {\bibnamefont {Saff},
  \bibfnamefont {E.~B.}},\ }\bibfield  {title} {\enquote {\bibinfo {title}
  {Asymptotics for minimal discrete energy on the sphere},}\ }\href {\doibase
  10.1090/S0002-9947-98-02119-9} {\bibfield  {journal} {\bibinfo  {journal}
  {Trans. Amer. Math. Soc.}\ }\textbf {\bibinfo {volume} {350}},\ \bibinfo
  {pages} {523--538} (\bibinfo {year} {1998})}\BibitemShut {NoStop}%
\bibitem [{\citenamefont {Kunz}(1974)}]{Kunz-74}%
  \BibitemOpen
  \bibfield  {author} {\bibinfo {author} {\bibnamefont {Kunz}, \bibfnamefont
  {H.}},\ }\bibfield  {title} {\enquote {\bibinfo {title} {The one-dimensional
  classical electron gas},}\ }\href {\doibase 10.1016/0003-4916(74)90413-8}
  {\bibfield  {journal} {\bibinfo  {journal} {Ann. Phys. (NY)}\ }\textbf
  {\bibinfo {volume} {85}},\ \bibinfo {pages} {303--335} (\bibinfo {year}
  {1974})}\BibitemShut {NoStop}%
\bibitem [{\citenamefont {Lacroix-A-Chez-Toine}, \citenamefont {Majumdar},\
  and\ \citenamefont {Schehr}(2019)}]{LacMajSch-19}%
  \BibitemOpen
  \bibfield  {author} {\bibinfo {author} {\bibnamefont {Lacroix-A-Chez-Toine},
  \bibfnamefont {B.}}, \bibinfo {author} {\bibnamefont {Majumdar},
  \bibfnamefont {S.~N.}}, \ and\ \bibinfo {author} {\bibnamefont {Schehr},
  \bibfnamefont {G.}},\ }\bibfield  {title} {\enquote {\bibinfo {title}
  {Rotating trapped fermions in two dimensions and the complex {G}inibre
  ensemble: {E}xact results for the entanglement entropy and number
  variance},}\ }\href {\doibase 10.1103/PhysRevA.99.021602} {\bibfield
  {journal} {\bibinfo  {journal} {Phys. Rev. A}\ }\textbf {\bibinfo {volume}
  {99}},\ \bibinfo {pages} {021602} (\bibinfo {year} {2019})}\BibitemShut
  {NoStop}%
\bibitem [{\citenamefont {Lahbabi}(2013)}]{Lahbabi-PhD}%
  \BibitemOpen
  \bibfield  {author} {\bibinfo {author} {\bibnamefont {Lahbabi}, \bibfnamefont
  {S.}},\ }\emph {\bibinfo {title} {Etude math{\'e}matique de mod{\`e}les
  quantiques et classiques pour les mat{\'e}riaux al{\'e}atoires {\`a}
  l'{\'e}chelle atomique}},\ \href
  {http://tel.archives-ouvertes.fr/tel-00873213} {Ph.D. thesis},\ \bibinfo
  {school} {Universit{\'e} de Cergy-Pontoise} (\bibinfo {year}
  {2013})\BibitemShut {NoStop}%
\bibitem [{\citenamefont {Laird}\ and\ \citenamefont
  {Haymet}(1992)}]{LaiHay-92}%
  \BibitemOpen
  \bibfield  {author} {\bibinfo {author} {\bibnamefont {Laird}, \bibfnamefont
  {B.~B.}}\ and\ \bibinfo {author} {\bibnamefont {Haymet}, \bibfnamefont
  {A.}},\ }\bibfield  {title} {\enquote {\bibinfo {title} {Phase diagram for
  the inverse sixth power potential system from molecular dynamics computer
  simulation},}\ }\href {\doibase 10.1080/00268979200100071} {\bibfield
  {journal} {\bibinfo  {journal} {Mol. Phys.}\ }\textbf {\bibinfo {volume}
  {75}},\ \bibinfo {pages} {71--80} (\bibinfo {year} {1992})}\BibitemShut
  {NoStop}%
\bibitem [{\citenamefont {Landau}(1915)}]{Landau-15}%
  \BibitemOpen
  \bibfield  {author} {\bibinfo {author} {\bibnamefont {Landau}, \bibfnamefont
  {E.}},\ }\bibfield  {title} {\enquote {\bibinfo {title} {{Zur analytischen
  Zahlentheorie der definiten quadratischen Formen}},}\ }\href@noop {}
  {\bibfield  {journal} {\bibinfo  {journal} {Berliner Akademieber}\ ,\
  \bibinfo {pages} {458--476}} (\bibinfo {year} {1915})}\BibitemShut {NoStop}%
\bibitem [{\citenamefont {Landau}(1924)}]{Landau-24}%
  \BibitemOpen
  \bibfield  {author} {\bibinfo {author} {\bibnamefont {Landau}, \bibfnamefont
  {E.}},\ }\bibfield  {title} {\enquote {\bibinfo {title} {\"{U}ber
  {G}itterpunkte in mehrdimensionalen {E}llipsoiden},}\ }\href {\doibase
  10.1007/BF01187457} {\bibfield  {journal} {\bibinfo  {journal} {Math. Z.}\
  }\textbf {\bibinfo {volume} {21}},\ \bibinfo {pages} {126--132} (\bibinfo
  {year} {1924})}\BibitemShut {NoStop}%
\bibitem [{\citenamefont {Landkof}(1972)}]{Landkof-72}%
  \BibitemOpen
  \bibfield  {author} {\bibinfo {author} {\bibnamefont {Landkof}, \bibfnamefont
  {N.~S.}},\ }\href@noop {} {\emph {\bibinfo {title} {Foundations of modern
  potential theory}}}\ (\bibinfo  {publisher} {Springer-Verlag, New
  York-Heidelberg},\ \bibinfo {year} {1972})\ pp.\ \bibinfo {pages} {x+424},\
  \bibinfo {note} {translated from the Russian by A. P. Doohovskoy, Die
  Grundlehren der mathematischen Wissenschaften, Band 180}\BibitemShut
  {NoStop}%
\bibitem [{\citenamefont {Lanford}(1973)}]{Lanford-73}%
  \BibitemOpen
  \bibfield  {author} {\bibinfo {author} {\bibnamefont {Lanford}, \bibfnamefont
  {O.~E.}},\ }\bibfield  {title} {\enquote {\bibinfo {title} {Entropy and
  equilibrium states in classical statistical mechanics},}\ }in\ \href
  {\doibase 10.1007/BFb0112756} {\emph {\bibinfo {booktitle} {Statistical
  Mechanics and Mathematical Problems}}},\ \bibinfo {series} {Lecture Notes in
  Physics}, Vol.~\bibinfo {volume} {20},\ \bibinfo {editor} {edited by\
  \bibinfo {editor} {\bibfnamefont {A.}~\bibnamefont {Lenard}}}\ (\bibinfo
  {publisher} {Springer Berlin Heidelberg},\ \bibinfo {year} {1973})\ pp.\
  \bibinfo {pages} {1--113}\BibitemShut {NoStop}%
\bibitem [{\citenamefont {Lanford}\ and\ \citenamefont
  {Ruelle}(1969)}]{LanRue-69}%
  \BibitemOpen
  \bibfield  {author} {\bibinfo {author} {\bibnamefont {Lanford}, \bibfnamefont
  {III, O.~E.}}\ and\ \bibinfo {author} {\bibnamefont {Ruelle}, \bibfnamefont
  {D.}},\ }\bibfield  {title} {\enquote {\bibinfo {title} {Observables at
  infinity and states with short range correlations in statistical
  mechanics},}\ }\href {http://projecteuclid.org/euclid.cmp/1103841575}
  {\bibfield  {journal} {\bibinfo  {journal} {Comm. Math. Phys.}\ }\textbf
  {\bibinfo {volume} {13}},\ \bibinfo {pages} {194--215} (\bibinfo {year}
  {1969})}\BibitemShut {NoStop}%
\bibitem [{\citenamefont {Laughlin}(1983)}]{Laughlin-83}%
  \BibitemOpen
  \bibfield  {author} {\bibinfo {author} {\bibnamefont {Laughlin},
  \bibfnamefont {R.~B.}},\ }\bibfield  {title} {\enquote {\bibinfo {title}
  {{Anomalous Quantum Hall Effect: An Incompressible Quantum Fluid with
  Fractionally Charged Excitations}},}\ }\href {\doibase
  10.1103/PhysRevLett.50.1395} {\bibfield  {journal} {\bibinfo  {journal}
  {Phys. Rev. Lett.}\ }\textbf {\bibinfo {volume} {50}},\ \bibinfo {pages}
  {1395--1398} (\bibinfo {year} {1983})}\BibitemShut {NoStop}%
\bibitem [{\citenamefont {Lauritsen}(2021)}]{Lauritsen-21}%
  \BibitemOpen
  \bibfield  {author} {\bibinfo {author} {\bibnamefont {Lauritsen},
  \bibfnamefont {A.~r. B.~k.}},\ }\bibfield  {title} {\enquote {\bibinfo
  {title} {Floating {W}igner crystal and periodic jellium configurations},}\
  }\href {\doibase 10.1063/5.0053494} {\bibfield  {journal} {\bibinfo
  {journal} {J. Math. Phys.}\ }\textbf {\bibinfo {volume} {62}},\ \bibinfo
  {pages} {Paper No. 083305, 17} (\bibinfo {year} {2021})}\BibitemShut
  {NoStop}%
\bibitem [{\citenamefont {Lebl{\'e}}(2015)}]{Leble-15}%
  \BibitemOpen
  \bibfield  {author} {\bibinfo {author} {\bibnamefont {Lebl{\'e}},
  \bibfnamefont {T.}},\ }\bibfield  {title} {\enquote {\bibinfo {title} {A
  uniqueness result for minimizers of the 1{D} log-gas renormalized energy},}\
  }\href {https://doi.org/10.1016/j.jfa.2014.11.023} {\bibfield  {journal}
  {\bibinfo  {journal} {J. Funct. Anal.}\ }\textbf {\bibinfo {volume} {268}},\
  \bibinfo {pages} {1649--1677} (\bibinfo {year} {2015})}\BibitemShut {NoStop}%
\bibitem [{\citenamefont {Lebl{\'e}}(2016)}]{Leble-16}%
  \BibitemOpen
  \bibfield  {author} {\bibinfo {author} {\bibnamefont {Lebl{\'e}},
  \bibfnamefont {T.}},\ }\bibfield  {title} {\enquote {\bibinfo {title}
  {Logarithmic, {C}oulomb and {R}iesz energy of point processes},}\ }\href
  {https://doi.org/10.1007/s10955-015-1425-4} {\bibfield  {journal} {\bibinfo
  {journal} {J. Stat. Phys.}\ }\textbf {\bibinfo {volume} {162}},\ \bibinfo
  {pages} {887--923} (\bibinfo {year} {2016})}\BibitemShut {NoStop}%
\bibitem [{\citenamefont {Lebl\'{e}}(2017)}]{Leble-17}%
  \BibitemOpen
  \bibfield  {author} {\bibinfo {author} {\bibnamefont {Lebl\'{e}},
  \bibfnamefont {T.}},\ }\bibfield  {title} {\enquote {\bibinfo {title} {Local
  microscopic behavior for 2{D} {C}oulomb gases},}\ }\href {\doibase
  10.1007/s00440-016-0744-y} {\bibfield  {journal} {\bibinfo  {journal}
  {Probab. Theory Related Fields}\ }\textbf {\bibinfo {volume} {169}},\
  \bibinfo {pages} {931--976} (\bibinfo {year} {2017})}\BibitemShut {NoStop}%
\bibitem [{\citenamefont {{Lebl{\'e}}}\ and\ \citenamefont
  {{Serfaty}}(2017)}]{LebSer-17}%
  \BibitemOpen
  \bibfield  {author} {\bibinfo {author} {\bibnamefont {{Lebl{\'e}}},
  \bibfnamefont {T.}}\ and\ \bibinfo {author} {\bibnamefont {{Serfaty}},
  \bibfnamefont {S.}},\ }\bibfield  {title} {\enquote {\bibinfo {title} {{Large
  Deviation Principle for Empirical Fields of Log and Riesz Gases}},}\ }\href
  {\doibase 10.1007/s00222-017-0738-0} {\bibfield  {journal} {\bibinfo
  {journal} {Invent. Math.}\ }\textbf {\bibinfo {volume} {210}},\ \bibinfo
  {pages} {645--757} (\bibinfo {year} {2017})}\BibitemShut {NoStop}%
\bibitem [{\citenamefont {Lebowitz}(1983)}]{Lebowitz-83}%
  \BibitemOpen
  \bibfield  {author} {\bibinfo {author} {\bibnamefont {Lebowitz},
  \bibfnamefont {J.~L.}},\ }\bibfield  {title} {\enquote {\bibinfo {title}
  {Charge fluctuations in {C}oulomb systems},}\ }\href {\doibase
  10.1103/PhysRevA.27.1491} {\bibfield  {journal} {\bibinfo  {journal} {Phys.
  Rev. A}\ }\textbf {\bibinfo {volume} {27}},\ \bibinfo {pages} {1491--1494}
  (\bibinfo {year} {1983})}\BibitemShut {NoStop}%
\bibitem [{\citenamefont {Lebowitz}\ and\ \citenamefont
  {Martin}(1984)}]{LebMar-84}%
  \BibitemOpen
  \bibfield  {author} {\bibinfo {author} {\bibnamefont {Lebowitz},
  \bibfnamefont {J.~L.}}\ and\ \bibinfo {author} {\bibnamefont {Martin},
  \bibfnamefont {P.~A.}},\ }\bibfield  {title} {\enquote {\bibinfo {title} {On
  potential and field fluctuations in classical charged systems},}\ }\href
  {https://doi.org/10.1007/BF01770360} {\bibfield  {journal} {\bibinfo
  {journal} {J. Statist. Phys.}\ }\textbf {\bibinfo {volume} {34}},\ \bibinfo
  {pages} {287--311} (\bibinfo {year} {1984})}\BibitemShut {NoStop}%
\bibitem [{\citenamefont {Lebowitz}\ and\ \citenamefont
  {Percus}(1961)}]{LebPer-61}%
  \BibitemOpen
  \bibfield  {author} {\bibinfo {author} {\bibnamefont {Lebowitz},
  \bibfnamefont {J.~L.}}\ and\ \bibinfo {author} {\bibnamefont {Percus},
  \bibfnamefont {J.~K.}},\ }\bibfield  {title} {\enquote {\bibinfo {title}
  {Long-range correlations in a closed system with applications to nonuniform
  fluids},}\ }\href {\doibase 10.1103/PhysRev.122.1675} {\bibfield  {journal}
  {\bibinfo  {journal} {Phys. Rev.}\ }\textbf {\bibinfo {volume} {122}},\
  \bibinfo {pages} {1675--1691} (\bibinfo {year} {1961})}\BibitemShut {NoStop}%
\bibitem [{\citenamefont {Lebowitz}\ and\ \citenamefont
  {Percus}(1963)}]{LebPer-63}%
  \BibitemOpen
  \bibfield  {author} {\bibinfo {author} {\bibnamefont {Lebowitz},
  \bibfnamefont {J.~L.}}\ and\ \bibinfo {author} {\bibnamefont {Percus},
  \bibfnamefont {J.~K.}},\ }\bibfield  {title} {\enquote {\bibinfo {title}
  {Statistical thermodynamics of nonuniform fluids},}\ }\href {\doibase
  10.1063/1.1703877} {\bibfield  {journal} {\bibinfo  {journal} {J. Math.
  Phys.}\ }\textbf {\bibinfo {volume} {4}},\ \bibinfo {pages} {116--123}
  (\bibinfo {year} {1963})}\BibitemShut {NoStop}%
\bibitem [{\citenamefont {de~Leeuw}\ and\ \citenamefont
  {Perram}(1982)}]{LeePerrSmi-82}%
  \BibitemOpen
  \bibfield  {author} {\bibinfo {author} {\bibnamefont {de~Leeuw},
  \bibfnamefont {S.}}\ and\ \bibinfo {author} {\bibnamefont {Perram},
  \bibfnamefont {J.}},\ }\bibfield  {title} {\enquote {\bibinfo {title}
  {Statistical mechanics of two-dimensional {C}oulomb systems: {II}. {T}he
  two-dimensional one-component plasma},}\ }\href {\doibase
  10.1016/0378-4371(82)90156-X} {\bibfield  {journal} {\bibinfo  {journal}
  {Physica A}\ }\textbf {\bibinfo {volume} {113}},\ \bibinfo {pages} {546--558}
  (\bibinfo {year} {1982})}\BibitemShut {NoStop}%
\bibitem [{\citenamefont {{Leinaas}}\ and\ \citenamefont
  {{Myrheim}}(1977)}]{LeiMyr-77}%
  \BibitemOpen
  \bibfield  {author} {\bibinfo {author} {\bibnamefont {{Leinaas}},
  \bibfnamefont {J.~M.}}\ and\ \bibinfo {author} {\bibnamefont {{Myrheim}},
  \bibfnamefont {J.}},\ }\bibfield  {title} {\enquote {\bibinfo {title} {On the
  theory of identical particles},}\ }\href {\doibase 10.1007/BF02727953}
  {\bibfield  {journal} {\bibinfo  {journal} {Nuovo Cimento B Serie}\ }\textbf
  {\bibinfo {volume} {37}},\ \bibinfo {pages} {1--23} (\bibinfo {year}
  {1977})}\BibitemShut {NoStop}%
\bibitem [{\citenamefont {Lenard}(1963)}]{Lenard-63}%
  \BibitemOpen
  \bibfield  {author} {\bibinfo {author} {\bibnamefont {Lenard}, \bibfnamefont
  {A.}},\ }\bibfield  {title} {\enquote {\bibinfo {title} {Exact statistical
  mechanics of a one-dimensional system with {C}oulomb forces. {III}.
  {S}tatistics of the electric field},}\ }\href
  {https://doi.org/10.1063/1.1703988} {\bibfield  {journal} {\bibinfo
  {journal} {J. Mathematical Phys.}\ }\textbf {\bibinfo {volume} {4}},\
  \bibinfo {pages} {533--543} (\bibinfo {year} {1963})}\BibitemShut {NoStop}%
\bibitem [{\citenamefont {Lenard}\ and\ \citenamefont
  {Dyson}(1968)}]{DysLen-68}%
  \BibitemOpen
  \bibfield  {author} {\bibinfo {author} {\bibnamefont {Lenard}, \bibfnamefont
  {A.}}\ and\ \bibinfo {author} {\bibnamefont {Dyson}, \bibfnamefont {F.~J.}},\
  }\bibfield  {title} {\enquote {\bibinfo {title} {Stability of matter.
  {II}},}\ }\href {\doibase 10.1063/1.1664631} {\bibfield  {journal} {\bibinfo
  {journal} {J. Math. Phys.}\ }\textbf {\bibinfo {volume} {9}},\ \bibinfo
  {pages} {698--711} (\bibinfo {year} {1968})}\BibitemShut {NoStop}%
\bibitem [{\citenamefont {Lewin}\ and\ \citenamefont {Lieb}(2015)}]{LewLie-15}%
  \BibitemOpen
  \bibfield  {author} {\bibinfo {author} {\bibnamefont {Lewin}, \bibfnamefont
  {M.}}\ and\ \bibinfo {author} {\bibnamefont {Lieb}, \bibfnamefont {E.~H.}},\
  }\bibfield  {title} {\enquote {\bibinfo {title} {Improved {L}ieb-{O}xford
  exchange-correlation inequality with gradient correction},}\ }\href {\doibase
  10.1103/PhysRevA.91.022507} {\bibfield  {journal} {\bibinfo  {journal} {Phys.
  Rev. A}\ }\textbf {\bibinfo {volume} {91}},\ \bibinfo {pages} {022507}
  (\bibinfo {year} {2015})},\ \Eprint {http://arxiv.org/abs/1408.3358}
  {arXiv:1408.3358 [math-ph]} \BibitemShut {NoStop}%
\bibitem [{\citenamefont {Lewin}, \citenamefont {Lieb},\ and\ \citenamefont
  {Seiringer}(2018)}]{LewLieSei-18}%
  \BibitemOpen
  \bibfield  {author} {\bibinfo {author} {\bibnamefont {Lewin}, \bibfnamefont
  {M.}}, \bibinfo {author} {\bibnamefont {Lieb}, \bibfnamefont {E.~H.}}, \ and\
  \bibinfo {author} {\bibnamefont {Seiringer}, \bibfnamefont {R.}},\ }\bibfield
   {title} {\enquote {\bibinfo {title} {Statistical mechanics of the {U}niform
  {E}lectron {G}as},}\ }\href {\doibase 10.5802/jep.64} {\bibfield  {journal}
  {\bibinfo  {journal} {J. {\'E}c. polytech. Math.}\ }\textbf {\bibinfo
  {volume} {5}},\ \bibinfo {pages} {79--116} (\bibinfo {year} {2018})},\
  \Eprint {http://arxiv.org/abs/1705.10676} {arXiv:1705.10676 [math-ph]}
  \BibitemShut {NoStop}%
\bibitem [{\citenamefont {Lewin}, \citenamefont {Lieb},\ and\ \citenamefont
  {Seiringer}(2019{\natexlab{a}})}]{LewLieSei-19b}%
  \BibitemOpen
  \bibfield  {author} {\bibinfo {author} {\bibnamefont {Lewin}, \bibfnamefont
  {M.}}, \bibinfo {author} {\bibnamefont {Lieb}, \bibfnamefont {E.~H.}}, \ and\
  \bibinfo {author} {\bibnamefont {Seiringer}, \bibfnamefont {R.}},\ }\bibfield
   {title} {\enquote {\bibinfo {title} {Floating {W}igner crystal with no
  boundary charge fluctuations},}\ }\href {\doibase
  10.1103/PhysRevB.100.035127} {\bibfield  {journal} {\bibinfo  {journal}
  {Phys. Rev. B}\ }\textbf {\bibinfo {volume} {100}},\ \bibinfo {pages}
  {035127} (\bibinfo {year} {2019}{\natexlab{a}})},\ \Eprint
  {http://arxiv.org/abs/1905.09138} {arXiv:1905.09138} \BibitemShut {NoStop}%
\bibitem [{\citenamefont {Lewin}, \citenamefont {Lieb},\ and\ \citenamefont
  {Seiringer}(2019{\natexlab{b}})}]{LewLieSei-19}%
  \BibitemOpen
  \bibfield  {author} {\bibinfo {author} {\bibnamefont {Lewin}, \bibfnamefont
  {M.}}, \bibinfo {author} {\bibnamefont {Lieb}, \bibfnamefont {E.~H.}}, \ and\
  \bibinfo {author} {\bibnamefont {Seiringer}, \bibfnamefont {R.}},\ }\bibfield
   {title} {\enquote {\bibinfo {title} {The {L}ocal {D}ensity {A}pproximation
  in {D}ensity {F}unctional {T}heory},}\ }\href {\doibase
  10.2140/paa.2020.2.35} {\bibfield  {journal} {\bibinfo  {journal} {Pure Appl.
  Anal.}\ }\textbf {\bibinfo {volume} {2}},\ \bibinfo {pages} {35--73}
  (\bibinfo {year} {2019}{\natexlab{b}})},\ \Eprint
  {http://arxiv.org/abs/1903.04046} {arXiv:1903.04046} \BibitemShut {NoStop}%
\bibitem [{\citenamefont {Lewin}, \citenamefont {Lieb},\ and\ \citenamefont
  {Seiringer}(2020)}]{LewLieSei-19_ppt}%
  \BibitemOpen
  \bibfield  {author} {\bibinfo {author} {\bibnamefont {Lewin}, \bibfnamefont
  {M.}}, \bibinfo {author} {\bibnamefont {Lieb}, \bibfnamefont {E.~H.}}, \ and\
  \bibinfo {author} {\bibnamefont {Seiringer}, \bibfnamefont {R.}},\ }\bibfield
   {title} {\enquote {\bibinfo {title} {{Universal Functionals in Density
  Functional Theory}},}\ }\href@noop {} {\bibfield  {journal} {\bibinfo
  {journal} {ArXiv e-prints}\ } (\bibinfo {year} {2020})},\ \bibinfo {note}
  {chapter in a book ``Density Functional Theory" edited by \'Eric Canc\`es,
  Gero Friesecke \& Lin Lin},\ \Eprint {http://arxiv.org/abs/1912.10424}
  {arXiv:1912.10424} \BibitemShut {NoStop}%
\bibitem [{\citenamefont {Lewin}\ and\ \citenamefont
  {Seiringer}(2009)}]{LewSei-09}%
  \BibitemOpen
  \bibfield  {author} {\bibinfo {author} {\bibnamefont {Lewin}, \bibfnamefont
  {M.}}\ and\ \bibinfo {author} {\bibnamefont {Seiringer}, \bibfnamefont
  {R.}},\ }\bibfield  {title} {\enquote {\bibinfo {title} {Strongly correlated
  phases in rapidly rotating {B}ose gases},}\ }\href {\doibase
  10.1007/s10955-009-9833-y} {\bibfield  {journal} {\bibinfo  {journal} {J.
  Stat. Phys.}\ }\textbf {\bibinfo {volume} {137}},\ \bibinfo {pages}
  {1040--1062} (\bibinfo {year} {2009})},\ \Eprint
  {http://arxiv.org/abs/0906.0741} {arXiv:0906.0741} \BibitemShut {NoStop}%
\bibitem [{\citenamefont {Lewis}(1988{\natexlab{a}})}]{Lewis-88e}%
  \BibitemOpen
  \bibfield  {author} {\bibinfo {author} {\bibnamefont {Lewis}, \bibfnamefont
  {J.~T.}},\ }\bibfield  {title} {\enquote {\bibinfo {title} {The large
  deviation principle in statistical mechanics: an expository account},}\ }in\
  \href {\doibase 10.1007/BFb0077923} {\emph {\bibinfo {booktitle} {Stochastic
  mechanics and stochastic processes ({S}wansea, 1986)}}},\ \bibinfo {series}
  {Lecture Notes in Math.}, Vol.\ \bibinfo {volume} {1325}\ (\bibinfo
  {publisher} {Springer, Berlin},\ \bibinfo {year} {1988})\ pp.\ \bibinfo
  {pages} {141--155}\BibitemShut {NoStop}%
\bibitem [{\citenamefont {Lewis}(1988{\natexlab{b}})}]{Lewis-88a}%
  \BibitemOpen
  \bibfield  {author} {\bibinfo {author} {\bibnamefont {Lewis}, \bibfnamefont
  {J.~T.}},\ }\bibfield  {title} {\enquote {\bibinfo {title} {Limit theorems
  for stochastic processes associated with a boson gas},}\ }in\ \href@noop {}
  {\emph {\bibinfo {booktitle} {Mark {K}ac {S}eminar on {P}robability and
  {P}hysics. {S}yllabus 1985--1987 ({A}msterdam, 1985--1987)}}},\ \bibinfo
  {series} {CWI Syllabi}, Vol.~\bibinfo {volume} {17}\ (\bibinfo  {publisher}
  {Math. Centrum, Centrum Wisk. Inform., Amsterdam},\ \bibinfo {year} {1988})\
  pp.\ \bibinfo {pages} {137--146}\BibitemShut {NoStop}%
\bibitem [{\citenamefont {Lewis}, \citenamefont {Zagrebnov},\ and\
  \citenamefont {Pul\'{e}}(1988)}]{LewZagPul-88}%
  \BibitemOpen
  \bibfield  {author} {\bibinfo {author} {\bibnamefont {Lewis}, \bibfnamefont
  {J.~T.}}, \bibinfo {author} {\bibnamefont {Zagrebnov}, \bibfnamefont
  {V.~A.}}, \ and\ \bibinfo {author} {\bibnamefont {Pul\'{e}}, \bibfnamefont
  {J.~V.}},\ }\bibfield  {title} {\enquote {\bibinfo {title} {The large
  deviation principle for the {K}ac distribution},}\ }\href@noop {} {\bibfield
  {journal} {\bibinfo  {journal} {Helv. Phys. Acta}\ }\textbf {\bibinfo
  {volume} {61}},\ \bibinfo {pages} {1063--1078} (\bibinfo {year}
  {1988})}\BibitemShut {NoStop}%
\bibitem [{\citenamefont {{Li}}\ \emph {et~al.}(2021)\citenamefont {{Li}},
  \citenamefont {{Li}}, \citenamefont {{Regan}}, \citenamefont {{Wang}},
  \citenamefont {{Zhao}}, \citenamefont {{Kahn}}, \citenamefont {{Yumigeta}},
  \citenamefont {{Blei}}, \citenamefont {{Taniguchi}}, \citenamefont
  {{Watanabe}}, \citenamefont {{Tongay}}, \citenamefont {{Zettl}},
  \citenamefont {{Crommie}},\ and\ \citenamefont {{Wang}}}]{Li_etal-21}%
  \BibitemOpen
  \bibfield  {author} {\bibinfo {author} {\bibnamefont {{Li}}, \bibfnamefont
  {H.}}, \bibinfo {author} {\bibnamefont {{Li}}, \bibfnamefont {S.}}, \bibinfo
  {author} {\bibnamefont {{Regan}}, \bibfnamefont {E.~C.}}, \bibinfo {author}
  {\bibnamefont {{Wang}}, \bibfnamefont {D.}}, \bibinfo {author} {\bibnamefont
  {{Zhao}}, \bibfnamefont {W.}}, \bibinfo {author} {\bibnamefont {{Kahn}},
  \bibfnamefont {S.}}, \bibinfo {author} {\bibnamefont {{Yumigeta}},
  \bibfnamefont {K.}}, \bibinfo {author} {\bibnamefont {{Blei}}, \bibfnamefont
  {M.}}, \bibinfo {author} {\bibnamefont {{Taniguchi}}, \bibfnamefont {T.}},
  \bibinfo {author} {\bibnamefont {{Watanabe}}, \bibfnamefont {K.}}, \bibinfo
  {author} {\bibnamefont {{Tongay}}, \bibfnamefont {S.}}, \bibinfo {author}
  {\bibnamefont {{Zettl}}, \bibfnamefont {A.}}, \bibinfo {author} {\bibnamefont
  {{Crommie}}, \bibfnamefont {M.~F.}}, \ and\ \bibinfo {author} {\bibnamefont
  {{Wang}}, \bibfnamefont {F.}},\ }\bibfield  {title} {\enquote {\bibinfo
  {title} {{Imaging two-dimensional generalized Wigner crystals}},}\ }\href
  {\doibase 10.1038/s41586-021-03874-9} {\bibfield  {journal} {\bibinfo
  {journal} {Nature}\ }\textbf {\bibinfo {volume} {597}},\ \bibinfo {pages}
  {650--654} (\bibinfo {year} {2021})}\BibitemShut {NoStop}%
\bibitem [{\citenamefont {Lieb}\ and\ \citenamefont
  {Lebowitz}(1972)}]{LieLeb-72}%
  \BibitemOpen
  \bibfield  {author} {\bibinfo {author} {\bibnamefont {Lieb}, \bibfnamefont
  {E.~H.}}\ and\ \bibinfo {author} {\bibnamefont {Lebowitz}, \bibfnamefont
  {J.~L.}},\ }\bibfield  {title} {\enquote {\bibinfo {title} {The constitution
  of matter: {E}xistence of thermodynamics for systems composed of electrons
  and nuclei},}\ }\href {\doibase 10.1016/0001-8708(72)90023-0} {\bibfield
  {journal} {\bibinfo  {journal} {Advances in Math.}\ }\textbf {\bibinfo
  {volume} {9}},\ \bibinfo {pages} {316--398} (\bibinfo {year}
  {1972})}\BibitemShut {NoStop}%
\bibitem [{\citenamefont {Lieb}\ and\ \citenamefont {Loss}(2001)}]{LieLos-01}%
  \BibitemOpen
  \bibfield  {author} {\bibinfo {author} {\bibnamefont {Lieb}, \bibfnamefont
  {E.~H.}}\ and\ \bibinfo {author} {\bibnamefont {Loss}, \bibfnamefont {M.}},\
  }\href@noop {} {\emph {\bibinfo {title} {Analysis}}},\ \bibinfo {edition}
  {2nd}\ ed.,\ \bibinfo {series} {Graduate Studies in Mathematics},
  Vol.~\bibinfo {volume} {14}\ (\bibinfo  {publisher} {American Mathematical
  Society},\ \bibinfo {address} {Providence, RI},\ \bibinfo {year} {2001})\
  pp.\ \bibinfo {pages} {xxii+346}\BibitemShut {NoStop}%
\bibitem [{\citenamefont {Lieb}\ and\ \citenamefont
  {Narnhofer}(1975)}]{LieNar-75}%
  \BibitemOpen
  \bibfield  {author} {\bibinfo {author} {\bibnamefont {Lieb}, \bibfnamefont
  {E.~H.}}\ and\ \bibinfo {author} {\bibnamefont {Narnhofer}, \bibfnamefont
  {H.}},\ }\bibfield  {title} {\enquote {\bibinfo {title} {The thermodynamic
  limit for jellium},}\ }\href {\doibase 10.1007/BF01012066} {\bibfield
  {journal} {\bibinfo  {journal} {J. Stat. Phys.}\ }\textbf {\bibinfo {volume}
  {12}},\ \bibinfo {pages} {291--310} (\bibinfo {year} {1975})}\BibitemShut
  {NoStop}%
\bibitem [{\citenamefont {Lieb}, \citenamefont {Rougerie},\ and\ \citenamefont
  {Yngvason}(2018)}]{LieRouYng-18}%
  \BibitemOpen
  \bibfield  {author} {\bibinfo {author} {\bibnamefont {Lieb}, \bibfnamefont
  {E.~H.}}, \bibinfo {author} {\bibnamefont {Rougerie}, \bibfnamefont {N.}}, \
  and\ \bibinfo {author} {\bibnamefont {Yngvason}, \bibfnamefont {J.}},\
  }\bibfield  {title} {\enquote {\bibinfo {title} {Rigidity of the {L}aughlin
  liquid},}\ }\href {\doibase 10.1007/s10955-018-2082-1} {\bibfield  {journal}
  {\bibinfo  {journal} {J. Stat. Phys.}\ }\textbf {\bibinfo {volume} {172}},\
  \bibinfo {pages} {544--554} (\bibinfo {year} {2018})}\BibitemShut {NoStop}%
\bibitem [{\citenamefont {Lieb}, \citenamefont {Rougerie},\ and\ \citenamefont
  {Yngvason}(2019)}]{LieRouYng-19}%
  \BibitemOpen
  \bibfield  {author} {\bibinfo {author} {\bibnamefont {Lieb}, \bibfnamefont
  {E.~H.}}, \bibinfo {author} {\bibnamefont {Rougerie}, \bibfnamefont {N.}}, \
  and\ \bibinfo {author} {\bibnamefont {Yngvason}, \bibfnamefont {J.}},\
  }\bibfield  {title} {\enquote {\bibinfo {title} {Local incompressibility
  estimates for the {L}aughlin phase},}\ }\href {\doibase
  10.1007/s00220-018-3181-1} {\bibfield  {journal} {\bibinfo  {journal} {Comm.
  Math. Phys.}\ }\textbf {\bibinfo {volume} {365}},\ \bibinfo {pages}
  {431--470} (\bibinfo {year} {2019})}\BibitemShut {NoStop}%
\bibitem [{\citenamefont {Lieb}\ and\ \citenamefont
  {Seiringer}(2010)}]{LieSei-09}%
  \BibitemOpen
  \bibfield  {author} {\bibinfo {author} {\bibnamefont {Lieb}, \bibfnamefont
  {E.~H.}}\ and\ \bibinfo {author} {\bibnamefont {Seiringer}, \bibfnamefont
  {R.}},\ }\href@noop {} {\emph {\bibinfo {title} {The {S}tability of {M}atter
  in {Q}uantum {M}echanics}}}\ (\bibinfo  {publisher} {Cambridge Univ. Press},\
  \bibinfo {year} {2010})\BibitemShut {NoStop}%
\bibitem [{\citenamefont {Lieb}\ and\ \citenamefont
  {Simon}(1977)}]{LieSim-77b}%
  \BibitemOpen
  \bibfield  {author} {\bibinfo {author} {\bibnamefont {Lieb}, \bibfnamefont
  {E.~H.}}\ and\ \bibinfo {author} {\bibnamefont {Simon}, \bibfnamefont {B.}},\
  }\bibfield  {title} {\enquote {\bibinfo {title} {The {T}homas-{F}ermi theory
  of atoms, molecules and solids},}\ }\href {\doibase
  10.1016/0001-8708(77)90108-6} {\bibfield  {journal} {\bibinfo  {journal}
  {Adv. Math.}\ }\textbf {\bibinfo {volume} {23}},\ \bibinfo {pages} {22--116}
  (\bibinfo {year} {1977})}\BibitemShut {NoStop}%
\bibitem [{\citenamefont {Likos}(2001)}]{Likos-01}%
  \BibitemOpen
  \bibfield  {author} {\bibinfo {author} {\bibnamefont {Likos}, \bibfnamefont
  {C.~N.}},\ }\bibfield  {title} {\enquote {\bibinfo {title} {Effective
  interactions in soft condensed matter physics},}\ }\href {\doibase
  10.1016/S0370-1573(00)00141-1} {\bibfield  {journal} {\bibinfo  {journal}
  {Phys. Rep.}\ }\textbf {\bibinfo {volume} {348}},\ \bibinfo {pages}
  {267--439} (\bibinfo {year} {2001})}\BibitemShut {NoStop}%
\bibitem [{\citenamefont {Lin}, \citenamefont {Zheng},\ and\ \citenamefont
  {Trimper}(2006)}]{LinZheTri-06}%
  \BibitemOpen
  \bibfield  {author} {\bibinfo {author} {\bibnamefont {Lin}, \bibfnamefont
  {S.~Z.}}, \bibinfo {author} {\bibnamefont {Zheng}, \bibfnamefont {B.}}, \
  and\ \bibinfo {author} {\bibnamefont {Trimper}, \bibfnamefont {S.}},\
  }\bibfield  {title} {\enquote {\bibinfo {title} {Computer simulations of
  two-dimensional melting with dipole-dipole interactions},}\ }\href {\doibase
  10.1103/PhysRevE.73.066106} {\bibfield  {journal} {\bibinfo  {journal} {Phys.
  Rev. E}\ }\textbf {\bibinfo {volume} {73}},\ \bibinfo {pages} {066106}
  (\bibinfo {year} {2006})}\BibitemShut {NoStop}%
\bibitem [{\citenamefont {Loos}\ and\ \citenamefont {Gill}(2011)}]{LooGil-11}%
  \BibitemOpen
  \bibfield  {author} {\bibinfo {author} {\bibnamefont {Loos}, \bibfnamefont
  {P.-F.}}\ and\ \bibinfo {author} {\bibnamefont {Gill}, \bibfnamefont
  {P.~M.~W.}},\ }\bibfield  {title} {\enquote {\bibinfo {title} {Thinking
  outside the box: The uniform electron gas on a hypersphere},}\ }\href
  {\doibase 10.1063/1.3665393} {\bibfield  {journal} {\bibinfo  {journal} {J.
  Chem. Phys.}\ }\textbf {\bibinfo {volume} {135}},\ \bibinfo {pages} {214111}
  (\bibinfo {year} {2011})}\BibitemShut {NoStop}%
\bibitem [{\citenamefont {Loos}\ and\ \citenamefont {Gill}(2013)}]{LooGil-13}%
  \BibitemOpen
  \bibfield  {author} {\bibinfo {author} {\bibnamefont {Loos}, \bibfnamefont
  {P.-F.}}\ and\ \bibinfo {author} {\bibnamefont {Gill}, \bibfnamefont
  {P.~M.~W.}},\ }\bibfield  {title} {\enquote {\bibinfo {title} {{Uniform
  electron gases. I. Electrons on a ring}},}\ }\href {\doibase
  10.1063/1.4802589} {\bibfield  {journal} {\bibinfo  {journal} {J. Chem.
  Phys.}\ }\textbf {\bibinfo {volume} {138}},\ \bibinfo {pages} {164124}
  (\bibinfo {year} {2013})}\BibitemShut {NoStop}%
\bibitem [{\citenamefont {Loos}\ and\ \citenamefont {Gill}(2016)}]{LooGil-16}%
  \BibitemOpen
  \bibfield  {author} {\bibinfo {author} {\bibnamefont {Loos}, \bibfnamefont
  {P.-F.}}\ and\ \bibinfo {author} {\bibnamefont {Gill}, \bibfnamefont
  {P.~M.~W.}},\ }\bibfield  {title} {\enquote {\bibinfo {title} {The uniform
  electron gas},}\ }\href {\doibase https://doi.org/10.1002/wcms.1257}
  {\bibfield  {journal} {\bibinfo  {journal} {WIREs Computational Molecular
  Science}\ }\textbf {\bibinfo {volume} {6}},\ \bibinfo {pages} {410--429}
  (\bibinfo {year} {2016})}\BibitemShut {NoStop}%
\bibitem [{\citenamefont {Lugrin}\ and\ \citenamefont
  {Martin}(1982)}]{LugMar-82}%
  \BibitemOpen
  \bibfield  {author} {\bibinfo {author} {\bibnamefont {Lugrin}, \bibfnamefont
  {C.}}\ and\ \bibinfo {author} {\bibnamefont {Martin}, \bibfnamefont
  {P.~A.}},\ }\bibfield  {title} {\enquote {\bibinfo {title} {Functional
  integration treatment of one-dimensional ionic mixtures},}\ }\href {\doibase
  10.1063/1.525284} {\bibfield  {journal} {\bibinfo  {journal} {J. Math.
  Phys.}\ }\textbf {\bibinfo {volume} {23}},\ \bibinfo {pages} {2418--2429}
  (\bibinfo {year} {1982})}\BibitemShut {NoStop}%
\bibitem [{\citenamefont {Lundqvist}\ and\ \citenamefont
  {March}(1983)}]{LunMar-83}%
  \BibitemOpen
  \bibinfo {editor} {\bibnamefont {Lundqvist}, \bibfnamefont {S.}}\ and\
  \bibinfo {editor} {\bibnamefont {March}, \bibfnamefont {N.}},\ eds.,\
  \href@noop {} {\emph {\bibinfo {title} {Theory of the Inhomogeneous Electron
  Gas}}},\ Physics of Solids and Liquids\ (\bibinfo  {publisher} {Springer
  US},\ \bibinfo {year} {1983})\BibitemShut {NoStop}%
\bibitem [{\citenamefont {Madelung}(1918)}]{Madelung-18}%
  \BibitemOpen
  \bibfield  {author} {\bibinfo {author} {\bibnamefont {Madelung},
  \bibfnamefont {E.}},\ }\bibfield  {title} {\enquote {\bibinfo {title} {Das
  elektrische {F}eld in {S}ystemen von regelm{\"a}{\ss}ig angeordneten
  {P}unktladungen},}\ }\href@noop {} {\bibfield  {journal} {\bibinfo  {journal}
  {Physik. Z.}\ }\textbf {\bibinfo {volume} {19}},\ \bibinfo {pages} {524--533}
  (\bibinfo {year} {1918})}\BibitemShut {NoStop}%
\bibitem [{\citenamefont {Madison}\ \emph {et~al.}(2000)\citenamefont
  {Madison}, \citenamefont {Chevy}, \citenamefont {Wohlleben},\ and\
  \citenamefont {Dalibard}}]{Dalibard-00}%
  \BibitemOpen
  \bibfield  {author} {\bibinfo {author} {\bibnamefont {Madison}, \bibfnamefont
  {K.~W.}}, \bibinfo {author} {\bibnamefont {Chevy}, \bibfnamefont {F.}},
  \bibinfo {author} {\bibnamefont {Wohlleben}, \bibfnamefont {W.}}, \ and\
  \bibinfo {author} {\bibnamefont {Dalibard}, \bibfnamefont {J.}},\ }\bibfield
  {title} {\enquote {\bibinfo {title} {Vortex formation in a stirred
  {B}ose-{E}instein condensate},}\ }\href {\doibase 10.1103/PhysRevLett.84.806}
  {\bibfield  {journal} {\bibinfo  {journal} {Phys. Rev. Lett.}\ }\textbf
  {\bibinfo {volume} {84}},\ \bibinfo {pages} {806--809} (\bibinfo {year}
  {2000})}\BibitemShut {NoStop}%
\bibitem [{\citenamefont {Marchioro}\ and\ \citenamefont
  {Presutti}(1972)}]{MarPre-72}%
  \BibitemOpen
  \bibfield  {author} {\bibinfo {author} {\bibnamefont {Marchioro},
  \bibfnamefont {C.}}\ and\ \bibinfo {author} {\bibnamefont {Presutti},
  \bibfnamefont {E.}},\ }\bibfield  {title} {\enquote {\bibinfo {title}
  {Thermodynamics of particle systems in the presence of external macroscopic
  fields. {I}. {C}lassical case},}\ }\href
  {http://projecteuclid.org/euclid.cmp/1103858213} {\bibfield  {journal}
  {\bibinfo  {journal} {Comm. Math. Phys.}\ }\textbf {\bibinfo {volume} {27}},\
  \bibinfo {pages} {146--154} (\bibinfo {year} {1972})}\BibitemShut {NoStop}%
\bibitem [{\citenamefont {Martin}\ \emph {et~al.}(2004)\citenamefont {Martin},
  \citenamefont {Ilani}, \citenamefont {Verdene}, \citenamefont {Smet},
  \citenamefont {Umansky}, \citenamefont {Mahalu}, \citenamefont {Schuh},
  \citenamefont {Abstreiter},\ and\ \citenamefont {Yacoby}}]{Martin_etal-04}%
  \BibitemOpen
  \bibfield  {author} {\bibinfo {author} {\bibnamefont {Martin}, \bibfnamefont
  {J.}}, \bibinfo {author} {\bibnamefont {Ilani}, \bibfnamefont {S.}}, \bibinfo
  {author} {\bibnamefont {Verdene}, \bibfnamefont {B.}}, \bibinfo {author}
  {\bibnamefont {Smet}, \bibfnamefont {J.}}, \bibinfo {author} {\bibnamefont
  {Umansky}, \bibfnamefont {V.}}, \bibinfo {author} {\bibnamefont {Mahalu},
  \bibfnamefont {D.}}, \bibinfo {author} {\bibnamefont {Schuh}, \bibfnamefont
  {D.}}, \bibinfo {author} {\bibnamefont {Abstreiter}, \bibfnamefont {G.}}, \
  and\ \bibinfo {author} {\bibnamefont {Yacoby}, \bibfnamefont {A.}},\
  }\bibfield  {title} {\enquote {\bibinfo {title} {Localization of fractionally
  charged quasi-particles},}\ }\href {\doibase 10.1126/science.1099950}
  {\bibfield  {journal} {\bibinfo  {journal} {Science}\ }\textbf {\bibinfo
  {volume} {305}},\ \bibinfo {pages} {980--983} (\bibinfo {year}
  {2004})}\BibitemShut {NoStop}%
\bibitem [{\citenamefont {Martin}(1988)}]{Martin-88}%
  \BibitemOpen
  \bibfield  {author} {\bibinfo {author} {\bibnamefont {Martin}, \bibfnamefont
  {P.~A.}},\ }\bibfield  {title} {\enquote {\bibinfo {title} {Sum rules in
  charged fluids},}\ }\href {\doibase 10.1103/RevModPhys.60.1075} {\bibfield
  {journal} {\bibinfo  {journal} {Rev. Mod. Phys.}\ }\textbf {\bibinfo {volume}
  {60}},\ \bibinfo {pages} {1075--1127} (\bibinfo {year} {1988})}\BibitemShut
  {NoStop}%
\bibitem [{\citenamefont {{Martin}}\ and\ \citenamefont
  {{Yalcin}}(1980)}]{MarYal-80}%
  \BibitemOpen
  \bibfield  {author} {\bibinfo {author} {\bibnamefont {{Martin}},
  \bibfnamefont {P.~A.}}\ and\ \bibinfo {author} {\bibnamefont {{Yalcin}},
  \bibfnamefont {T.}},\ }\bibfield  {title} {\enquote {\bibinfo {title} {The
  charge fluctuations in classical {C}oulomb systems},}\ }\href {\doibase
  10.1007/BF01012866} {\bibfield  {journal} {\bibinfo  {journal} {J. Stat.
  Phys.}\ }\textbf {\bibinfo {volume} {22}},\ \bibinfo {pages} {435--463}
  (\bibinfo {year} {1980})}\BibitemShut {NoStop}%
\bibitem [{\citenamefont {Martinelli}\ and\ \citenamefont
  {Merlini}(1984)}]{MarMer-84}%
  \BibitemOpen
  \bibfield  {author} {\bibinfo {author} {\bibnamefont {Martinelli},
  \bibfnamefont {F.}}\ and\ \bibinfo {author} {\bibnamefont {Merlini},
  \bibfnamefont {D.}},\ }\bibfield  {title} {\enquote {\bibinfo {title} {A
  refined {M}ermin argument for the two-dimensional jellium},}\ }\href
  {https://doi.org/10.1007/BF01770361} {\bibfield  {journal} {\bibinfo
  {journal} {J. Statist. Phys.}\ }\textbf {\bibinfo {volume} {34}},\ \bibinfo
  {pages} {313--318} (\bibinfo {year} {1984})}\BibitemShut {NoStop}%
\bibitem [{\citenamefont {Mart\'{\i}nez-Finkelshtein}\ \emph
  {et~al.}(2004)\citenamefont {Mart\'{\i}nez-Finkelshtein}, \citenamefont
  {Maymeskul}, \citenamefont {Rakhmanov},\ and\ \citenamefont
  {Saff}}]{MarMayRakSaf-04}%
  \BibitemOpen
  \bibfield  {author} {\bibinfo {author} {\bibnamefont
  {Mart\'{\i}nez-Finkelshtein}, \bibfnamefont {A.}}, \bibinfo {author}
  {\bibnamefont {Maymeskul}, \bibfnamefont {V.}}, \bibinfo {author}
  {\bibnamefont {Rakhmanov}, \bibfnamefont {E.~A.}}, \ and\ \bibinfo {author}
  {\bibnamefont {Saff}, \bibfnamefont {E.~B.}},\ }\bibfield  {title} {\enquote
  {\bibinfo {title} {Asymptotics for minimal discrete {R}iesz energy on curves
  in {$\mathbb{R}^d$}},}\ }\href {\doibase 10.4153/CJM-2004-024-1} {\bibfield
  {journal} {\bibinfo  {journal} {Canad. J. Math.}\ }\textbf {\bibinfo {volume}
  {56}},\ \bibinfo {pages} {529--552} (\bibinfo {year} {2004})}\BibitemShut
  {NoStop}%
\bibitem [{\citenamefont {Mazars}(2011)}]{Mazars-11}%
  \BibitemOpen
  \bibfield  {author} {\bibinfo {author} {\bibnamefont {Mazars}, \bibfnamefont
  {M.}},\ }\bibfield  {title} {\enquote {\bibinfo {title} {Long ranged
  interactions in computer simulations and for quasi-2{D} systems},}\ }\href
  {\doibase https://doi.org/10.1016/j.physrep.2010.11.004} {\bibfield
  {journal} {\bibinfo  {journal} {Physics Reports}\ }\textbf {\bibinfo {volume}
  {500}},\ \bibinfo {pages} {43--116} (\bibinfo {year} {2011})}\BibitemShut
  {NoStop}%
\bibitem [{\citenamefont {Mazars}(2015)}]{Mazars-15}%
  \BibitemOpen
  \bibfield  {author} {\bibinfo {author} {\bibnamefont {Mazars}, \bibfnamefont
  {M.}},\ }\bibfield  {title} {\enquote {\bibinfo {title} {The melting of the
  classical two-dimensional {W}igner crystal},}\ }\href {\doibase
  10.1209/0295-5075/110/26003} {\bibfield  {journal} {\bibinfo  {journal}
  {{EPL} (Europhysics Letters)}\ }\textbf {\bibinfo {volume} {110}},\ \bibinfo
  {pages} {26003} (\bibinfo {year} {2015})}\BibitemShut {NoStop}%
\bibitem [{\citenamefont {Mazars}\ and\ \citenamefont
  {Salazar}(2019)}]{MazSal-19}%
  \BibitemOpen
  \bibfield  {author} {\bibinfo {author} {\bibnamefont {Mazars}, \bibfnamefont
  {M.}}\ and\ \bibinfo {author} {\bibnamefont {Salazar}, \bibfnamefont {R.}},\
  }\bibfield  {title} {\enquote {\bibinfo {title} {Topological defects in the
  two-dimensional melting},}\ }\href {\doibase 10.1209/0295-5075/126/56002}
  {\bibfield  {journal} {\bibinfo  {journal} {{EPL} (Europhysics Letters)}\
  }\textbf {\bibinfo {volume} {126}},\ \bibinfo {pages} {56002} (\bibinfo
  {year} {2019})}\BibitemShut {NoStop}%
\bibitem [{\citenamefont {Mehta}(2004)}]{Mehta-10}%
  \BibitemOpen
  \bibfield  {author} {\bibinfo {author} {\bibnamefont {Mehta}, \bibfnamefont
  {M.~L.}},\ }\href@noop {} {\emph {\bibinfo {title} {Random matrices}}},\
  \bibinfo {edition} {3rd}\ ed.,\ \bibinfo {series} {Pure and Applied
  Mathematics (Amsterdam)}, Vol.\ \bibinfo {volume} {142}\ (\bibinfo
  {publisher} {Elsevier/Academic Press, Amsterdam},\ \bibinfo {year} {2004})\
  pp.\ \bibinfo {pages} {xviii+688}\BibitemShut {NoStop}%
\bibitem [{\citenamefont {Mehta}\ and\ \citenamefont
  {Dyson}(1963)}]{DysMeh-63b}%
  \BibitemOpen
  \bibfield  {author} {\bibinfo {author} {\bibnamefont {Mehta}, \bibfnamefont
  {M.~L.}}\ and\ \bibinfo {author} {\bibnamefont {Dyson}, \bibfnamefont
  {F.~J.}},\ }\bibfield  {title} {\enquote {\bibinfo {title} {Statistical
  theory of the energy levels of complex systems. {V}},}\ }\href {\doibase
  10.1063/1.1704009} {\bibfield  {journal} {\bibinfo  {journal} {J.
  Mathematical Phys.}\ }\textbf {\bibinfo {volume} {4}},\ \bibinfo {pages}
  {713--719} (\bibinfo {year} {1963})}\BibitemShut {NoStop}%
\bibitem [{\citenamefont {Mermin}(1967)}]{Mermin-67}%
  \BibitemOpen
  \bibfield  {author} {\bibinfo {author} {\bibnamefont {Mermin}, \bibfnamefont
  {N.~D.}},\ }\bibfield  {title} {\enquote {\bibinfo {title} {Absence of
  ordering in certain classical systems},}\ }\href {\doibase 10.1063/1.1705316}
  {\bibfield  {journal} {\bibinfo  {journal} {J. Math. Phys.}\ }\textbf
  {\bibinfo {volume} {8}},\ \bibinfo {pages} {1061--1064} (\bibinfo {year}
  {1967})}\BibitemShut {NoStop}%
\bibitem [{\citenamefont {Mermin}(1968)}]{Mermin-68}%
  \BibitemOpen
  \bibfield  {author} {\bibinfo {author} {\bibnamefont {Mermin}, \bibfnamefont
  {N.~D.}},\ }\bibfield  {title} {\enquote {\bibinfo {title} {Crystalline order
  in two dimensions},}\ }\href {\doibase 10.1103/PhysRev.176.250} {\bibfield
  {journal} {\bibinfo  {journal} {Phys. Rev.}\ }\textbf {\bibinfo {volume}
  {176}},\ \bibinfo {pages} {250--254} (\bibinfo {year} {1968})}\BibitemShut
  {NoStop}%
\bibitem [{\citenamefont {Mermin}\ and\ \citenamefont
  {Wagner}(1966)}]{MerWag-66}%
  \BibitemOpen
  \bibfield  {author} {\bibinfo {author} {\bibnamefont {Mermin}, \bibfnamefont
  {N.~D.}}\ and\ \bibinfo {author} {\bibnamefont {Wagner}, \bibfnamefont
  {H.}},\ }\bibfield  {title} {\enquote {\bibinfo {title} {Absence of
  ferromagnetism or antiferromagnetism in one- or two-dimensional isotropic
  heisenberg models},}\ }\href {\doibase 10.1103/PhysRevLett.17.1133}
  {\bibfield  {journal} {\bibinfo  {journal} {Phys. Rev. Lett.}\ }\textbf
  {\bibinfo {volume} {17}},\ \bibinfo {pages} {1133--1136} (\bibinfo {year}
  {1966})}\BibitemShut {NoStop}%
\bibitem [{\citenamefont {Messer}\ and\ \citenamefont
  {Spohn}(1982)}]{MesSpo-82}%
  \BibitemOpen
  \bibfield  {author} {\bibinfo {author} {\bibnamefont {Messer}, \bibfnamefont
  {J.}}\ and\ \bibinfo {author} {\bibnamefont {Spohn}, \bibfnamefont {H.}},\
  }\bibfield  {title} {\enquote {\bibinfo {title} {Statistical mechanics of the
  isothermal {L}ane-{E}mden equation},}\ }\href {\doibase 10.1007/BF01342187}
  {\bibfield  {journal} {\bibinfo  {journal} {J. Statist. Phys.}\ }\textbf
  {\bibinfo {volume} {29}},\ \bibinfo {pages} {561--578} (\bibinfo {year}
  {1982})}\BibitemShut {NoStop}%
\bibitem [{\citenamefont {Millard}(1972)}]{Millard-72}%
  \BibitemOpen
  \bibfield  {author} {\bibinfo {author} {\bibnamefont {Millard}, \bibfnamefont
  {K.}},\ }\bibfield  {title} {\enquote {\bibinfo {title} {A statistical
  mechanical approach to the problem of a fluid in an external field},}\ }\href
  {\doibase 10.1063/1.1665958} {\bibfield  {journal} {\bibinfo  {journal} {J.
  Mathematical Phys.}\ }\textbf {\bibinfo {volume} {13}},\ \bibinfo {pages}
  {222--226} (\bibinfo {year} {1972})}\BibitemShut {NoStop}%
\bibitem [{\citenamefont {Montgomery}(1973)}]{Montgomery-73}%
  \BibitemOpen
  \bibfield  {author} {\bibinfo {author} {\bibnamefont {Montgomery},
  \bibfnamefont {H.~L.}},\ }\bibfield  {title} {\enquote {\bibinfo {title} {The
  pair correlation of zeros of the zeta function},}\ }in\ \href@noop {} {\emph
  {\bibinfo {booktitle} {Analytic number theory ({P}roc. {S}ympos. {P}ure
  {M}ath., {V}ol. {XXIV}, {S}t. {L}ouis {U}niv., {S}t. {L}ouis, {M}o.,
  1972)}}}\ (\bibinfo {year} {1973})\ pp.\ \bibinfo {pages}
  {181--193}\BibitemShut {NoStop}%
\bibitem [{\citenamefont {Montgomery}(1988)}]{Montgomery-88}%
  \BibitemOpen
  \bibfield  {author} {\bibinfo {author} {\bibnamefont {Montgomery},
  \bibfnamefont {H.~L.}},\ }\bibfield  {title} {\enquote {\bibinfo {title}
  {Minimal theta functions},}\ }\href {\doibase 10.1017/S0017089500007047}
  {\bibfield  {journal} {\bibinfo  {journal} {Glasgow Math. J.}\ }\textbf
  {\bibinfo {volume} {30}},\ \bibinfo {pages} {75--85} (\bibinfo {year}
  {1988})}\BibitemShut {NoStop}%
\bibitem [{\citenamefont {Moore}\ and\ \citenamefont
  {P{\'e}rez-Garrido}(1999)}]{MooPer-99}%
  \BibitemOpen
  \bibfield  {author} {\bibinfo {author} {\bibnamefont {Moore}, \bibfnamefont
  {M.~A.}}\ and\ \bibinfo {author} {\bibnamefont {P{\'e}rez-Garrido},
  \bibfnamefont {A.}},\ }\bibfield  {title} {\enquote {\bibinfo {title}
  {Absence of a finite-temperature melting transition in the classical
  two-dimensional one-component plasma},}\ }\href {\doibase
  10.1103/PhysRevLett.82.4078} {\bibfield  {journal} {\bibinfo  {journal}
  {Phys. Rev. Lett.}\ }\textbf {\bibinfo {volume} {82}},\ \bibinfo {pages}
  {4078--4081} (\bibinfo {year} {1999})}\BibitemShut {NoStop}%
\bibitem [{\citenamefont {Moser}(1975)}]{Moser-75}%
  \BibitemOpen
  \bibfield  {author} {\bibinfo {author} {\bibnamefont {Moser}, \bibfnamefont
  {J.}},\ }\bibfield  {title} {\enquote {\bibinfo {title} {Three integrable
  {H}amiltonian systems connected with isospectral deformations},}\ }\href
  {\doibase 10.1016/0001-8708(75)90151-6} {\bibfield  {journal} {\bibinfo
  {journal} {Advances in Math.}\ }\textbf {\bibinfo {volume} {16}},\ \bibinfo
  {pages} {197--220} (\bibinfo {year} {1975})}\BibitemShut {NoStop}%
\bibitem [{\citenamefont {Murray}\ and\ \citenamefont
  {Van~Winkle}(1987)}]{MurWin-87}%
  \BibitemOpen
  \bibfield  {author} {\bibinfo {author} {\bibnamefont {Murray}, \bibfnamefont
  {C.~A.}}\ and\ \bibinfo {author} {\bibnamefont {Van~Winkle}, \bibfnamefont
  {D.~H.}},\ }\bibfield  {title} {\enquote {\bibinfo {title} {Experimental
  observation of two-stage melting in a classical two-dimensional screened
  {C}oulomb system},}\ }\href {\doibase 10.1103/PhysRevLett.58.1200} {\bibfield
   {journal} {\bibinfo  {journal} {Phys. Rev. Lett.}\ }\textbf {\bibinfo
  {volume} {58}},\ \bibinfo {pages} {1200--1203} (\bibinfo {year}
  {1987})}\BibitemShut {NoStop}%
\bibitem [{\citenamefont {Murthy}\ and\ \citenamefont
  {Shankar}(1994)}]{MurSha-94}%
  \BibitemOpen
  \bibfield  {author} {\bibinfo {author} {\bibnamefont {Murthy}, \bibfnamefont
  {M.~V.~N.}}\ and\ \bibinfo {author} {\bibnamefont {Shankar}, \bibfnamefont
  {R.}},\ }\bibfield  {title} {\enquote {\bibinfo {title} {Thermodynamics of a
  one-dimensional ideal gas with fractional exclusion statistics},}\ }\href
  {\doibase 10.1103/PhysRevLett.73.3331} {\bibfield  {journal} {\bibinfo
  {journal} {Phys. Rev. Lett.}\ }\textbf {\bibinfo {volume} {73}},\ \bibinfo
  {pages} {3331--3334} (\bibinfo {year} {1994})}\BibitemShut {NoStop}%
\bibitem [{\citenamefont {Muto}\ and\ \citenamefont {Aoki}(1999)}]{MutAok-99}%
  \BibitemOpen
  \bibfield  {author} {\bibinfo {author} {\bibnamefont {Muto}, \bibfnamefont
  {S.}}\ and\ \bibinfo {author} {\bibnamefont {Aoki}, \bibfnamefont {H.}},\
  }\bibfield  {title} {\enquote {\bibinfo {title} {Crystallization of a
  classical two-dimensional electron system: Positional and orientational
  orders},}\ }\href {\doibase 10.1103/PhysRevB.59.14911} {\bibfield  {journal}
  {\bibinfo  {journal} {Phys. Rev. B}\ }\textbf {\bibinfo {volume} {59}},\
  \bibinfo {pages} {14911--14914} (\bibinfo {year} {1999})}\BibitemShut
  {NoStop}%
\bibitem [{\citenamefont {Myrheim}(1999)}]{Myrheim-99}%
  \BibitemOpen
  \bibfield  {author} {\bibinfo {author} {\bibnamefont {Myrheim}, \bibfnamefont
  {J.}},\ }\bibfield  {title} {\enquote {\bibinfo {title} {Anyons},}\ }in\
  \href@noop {} {\emph {\bibinfo {booktitle} {Topological aspects of low
  dimensional systems}}},\ \bibinfo {series} {Les Houches - \'Ecole d'\'Et\'e
  de Physique Th\'eorique}, Vol.~\bibinfo {volume} {69},\ \bibinfo {editor}
  {edited by\ \bibinfo {editor} {\bibfnamefont {A.}~\bibnamefont {Comtet}},
  \bibinfo {editor} {\bibfnamefont {T.}~\bibnamefont {Jolic{\oe}ur}}, \bibinfo
  {editor} {\bibfnamefont {S.}~\bibnamefont {Ouvry}}, \ and\ \bibinfo {editor}
  {\bibfnamefont {F.}~\bibnamefont {David}}}\ (\bibinfo {year} {1999})\ pp.\
  \bibinfo {pages} {265--413}\BibitemShut {NoStop}%
\bibitem [{\citenamefont {Nakano}(2014)}]{Nakano-14}%
  \BibitemOpen
  \bibfield  {author} {\bibinfo {author} {\bibnamefont {Nakano}, \bibfnamefont
  {F.}},\ }\bibfield  {title} {\enquote {\bibinfo {title} {Level statistics for
  one-dimensional {S}chr\"{o}dinger operators and {G}aussian beta ensemble},}\
  }\href {\doibase 10.1007/s10955-014-0987-x} {\bibfield  {journal} {\bibinfo
  {journal} {J. Stat. Phys.}\ }\textbf {\bibinfo {volume} {156}},\ \bibinfo
  {pages} {66--93} (\bibinfo {year} {2014})}\BibitemShut {NoStop}%
\bibitem [{\citenamefont {Nelson}\ and\ \citenamefont
  {Halperin}(1979)}]{NelHal-79}%
  \BibitemOpen
  \bibfield  {author} {\bibinfo {author} {\bibnamefont {Nelson}, \bibfnamefont
  {D.~R.}}\ and\ \bibinfo {author} {\bibnamefont {Halperin}, \bibfnamefont
  {B.~I.}},\ }\bibfield  {title} {\enquote {\bibinfo {title}
  {Dislocation-mediated melting in two dimensions},}\ }\href {\doibase
  10.1103/PhysRevB.19.2457} {\bibfield  {journal} {\bibinfo  {journal} {Phys.
  Rev. B}\ }\textbf {\bibinfo {volume} {19}},\ \bibinfo {pages} {2457--2484}
  (\bibinfo {year} {1979})}\BibitemShut {NoStop}%
\bibitem [{\citenamefont {Nemkov}\ and\ \citenamefont
  {Klevtsov}(2021)}]{NemKle-21}%
  \BibitemOpen
  \bibfield  {author} {\bibinfo {author} {\bibnamefont {Nemkov}, \bibfnamefont
  {N.}}\ and\ \bibinfo {author} {\bibnamefont {Klevtsov}, \bibfnamefont {S.}},\
  }\bibfield  {title} {\enquote {\bibinfo {title} {Liouville perturbation
  theory for {L}aughlin state and {C}oulomb gas},}\ }\href {\doibase
  10.1088/1751-8121/ac1483} {\bibfield  {journal} {\bibinfo  {journal} {J.
  Phys. A}\ }\textbf {\bibinfo {volume} {54}},\ \bibinfo {pages} {Paper No.
  335204, 22} (\bibinfo {year} {2021})}\BibitemShut {NoStop}%
\bibitem [{\citenamefont {Nijboer}(1975)}]{Nijboer-75}%
  \BibitemOpen
  \bibfield  {author} {\bibinfo {author} {\bibnamefont {Nijboer}, \bibfnamefont
  {B.~R.~A.}},\ }\bibfield  {title} {\enquote {\bibinfo {title} {On a relation
  between the scattering cross-section in dense media and the energy of a
  dilute electron gas},}\ }\href@noop {} {\bibfield  {journal} {\bibinfo
  {journal} {Philips Res. Rep.}\ }\textbf {\bibinfo {volume} {30}},\ \bibinfo
  {pages} {74} (\bibinfo {year} {1975})}\BibitemShut {NoStop}%
\bibitem [{\citenamefont {Nijboer}\ and\ \citenamefont
  {Ruijgrok}(1988)}]{NijRui-88}%
  \BibitemOpen
  \bibfield  {author} {\bibinfo {author} {\bibnamefont {Nijboer}, \bibfnamefont
  {B.~R.~A.}}\ and\ \bibinfo {author} {\bibnamefont {Ruijgrok}, \bibfnamefont
  {T.~W.}},\ }\bibfield  {title} {\enquote {\bibinfo {title} {On the energy per
  particle in three- and two-dimensional {W}igner lattices},}\ }\href {\doibase
  10.1007/BF01011562} {\bibfield  {journal} {\bibinfo  {journal} {J. Statist.
  Phys.}\ }\textbf {\bibinfo {volume} {53}},\ \bibinfo {pages} {361--382}
  (\bibinfo {year} {1988})}\BibitemShut {NoStop}%
\bibitem [{\citenamefont {Nonnenmacher}\ and\ \citenamefont
  {Voros}(1998)}]{NonVor-98}%
  \BibitemOpen
  \bibfield  {author} {\bibinfo {author} {\bibnamefont {Nonnenmacher},
  \bibfnamefont {S.}}\ and\ \bibinfo {author} {\bibnamefont {Voros},
  \bibfnamefont {A.}},\ }\bibfield  {title} {\enquote {\bibinfo {title}
  {Chaotic eigenfunctions in phase space},}\ }\href {\doibase
  10.1023/A:1023080303171} {\bibfield  {journal} {\bibinfo  {journal} {J. Stat.
  Phys.}\ }\textbf {\bibinfo {volume} {92}},\ \bibinfo {pages} {431--518}
  (\bibinfo {year} {1998})}\BibitemShut {NoStop}%
\bibitem [{Note1()}]{Note1}%
  \BibitemOpen
  \bibinfo {note} {For a smooth domain (for instance satisfying~\protect
  \textup {\hbox {\mathsurround \z@ \protect \normalfont (\ignorespaces \ref
  {eq:hyp_domain}\unskip \@@italiccorr )}} below), there are of the order of
  $\ell ^{d-1}R$ points located at a distance $R$ from the boundary. Those see
  a bounded potential. For the $N+o(N)$ particles inside, at a distance
  $\geqslant R$, the potential induced by the particles outside is of order
  $O(R^{d-s})$ by Lemma~\ref {lem:simple_estim_sum}. Thus the energy shift is
  of order $\ell ^{d-1}R+\ell ^dR^{d-s}$. After optimizing over $R$, this
  provides the claimed $O(\ell ^{d-\protect \frac {s-d}{s-d+1}})$.}\BibitemShut
  {Stop}%
\bibitem [{Note2()}]{Note2}%
  \BibitemOpen
  \bibinfo {note} {In the physics literature, the name `Jellium' is often
  employed for electrons (which are quantum with spin), whereas the
  `one-component plasma' is mainly used for classical particles as considered
  in the present article.}\BibitemShut {Stop}%
\bibitem [{Note3()}]{Note3}%
  \BibitemOpen
  \bibinfo {note} {The shift $\DOTSI \intop \ilimits@ _Q|y|^2\protect \tmspace
  +\thinmuskip {.1667em}{\protect \rm d}y$ is also called the \protect \emph
  {lattice quantizer}~\cite {ConSlo-99} and it is minimal for the BCC
  lattice~\cite {BarSlo-83} in 3D.}\BibitemShut {Stop}%
\bibitem [{Note4()}]{Note4}%
  \BibitemOpen
  \bibinfo {note} {The particles have no spin and one should use the magnetic
  Laplacian. This is also equivalent to a Fermi system in a harmonic trap
  rotating at the largest possible speed~\cite {LacMajSch-19}.}\BibitemShut
  {Stop}%
\bibitem [{\citenamefont {Odlyzko}(1987)}]{Odlyzko-87}%
  \BibitemOpen
  \bibfield  {author} {\bibinfo {author} {\bibnamefont {Odlyzko}, \bibfnamefont
  {A.~M.}},\ }\bibfield  {title} {\enquote {\bibinfo {title} {On the
  distribution of spacings between zeros of the zeta function},}\ }\href
  {\doibase 10.2307/2007890} {\bibfield  {journal} {\bibinfo  {journal} {Math.
  Comp.}\ }\textbf {\bibinfo {volume} {48}},\ \bibinfo {pages} {273--308}
  (\bibinfo {year} {1987})}\BibitemShut {NoStop}%
\bibitem [{\citenamefont {Onsager}(1939)}]{Onsager-39}%
  \BibitemOpen
  \bibfield  {author} {\bibinfo {author} {\bibnamefont {Onsager}, \bibfnamefont
  {L.}},\ }\bibfield  {title} {\enquote {\bibinfo {title} {Electrostatic
  interaction of molecules},}\ }\href {\doibase 10.1021/j150389a001} {\bibfield
   {journal} {\bibinfo  {journal} {J. Phys. Chem.}\ }\textbf {\bibinfo {volume}
  {43}},\ \bibinfo {pages} {189--196} (\bibinfo {year} {1939})}\BibitemShut
  {NoStop}%
\bibitem [{\citenamefont {Osgood}, \citenamefont {Phillips},\ and\
  \citenamefont {Sarnak}(1988)}]{OsgPhiSar-88}%
  \BibitemOpen
  \bibfield  {author} {\bibinfo {author} {\bibnamefont {Osgood}, \bibfnamefont
  {B.}}, \bibinfo {author} {\bibnamefont {Phillips}, \bibfnamefont {R.}}, \
  and\ \bibinfo {author} {\bibnamefont {Sarnak}, \bibfnamefont {P.}},\
  }\bibfield  {title} {\enquote {\bibinfo {title} {Extremals of determinants of
  {L}aplacians},}\ }\href {\doibase 10.1016/0022-1236(88)90070-5} {\bibfield
  {journal} {\bibinfo  {journal} {J. Funct. Anal.}\ }\textbf {\bibinfo {volume}
  {80}},\ \bibinfo {pages} {148--211} (\bibinfo {year} {1988})}\BibitemShut
  {NoStop}%
\bibitem [{\citenamefont {Papangelou}(1987)}]{Papangelou-87}%
  \BibitemOpen
  \bibfield  {author} {\bibinfo {author} {\bibnamefont {Papangelou},
  \bibfnamefont {F.}},\ }\bibfield  {title} {\enquote {\bibinfo {title} {On the
  absence of phase transition in continuous one-dimensional {G}ibbs systems
  with no hard core},}\ }\href {https://doi.org/10.1007/BF00363511} {\bibfield
  {journal} {\bibinfo  {journal} {Probab. Theory Related Fields}\ }\textbf
  {\bibinfo {volume} {74}},\ \bibinfo {pages} {485--496} (\bibinfo {year}
  {1987})}\BibitemShut {NoStop}%
\bibitem [{\citenamefont {Park}(1977)}]{Park-77}%
  \BibitemOpen
  \bibfield  {author} {\bibinfo {author} {\bibnamefont {Park}, \bibfnamefont
  {Y.~M.}},\ }\bibfield  {title} {\enquote {\bibinfo {title} {Massless quantum
  sine-{G}ordon equation in two space-time dimensions: correlation inequalities
  and infinite volume limit},}\ }\href {https://doi.org/10.1063/1.523230}
  {\bibfield  {journal} {\bibinfo  {journal} {J. Mathematical Phys.}\ }\textbf
  {\bibinfo {volume} {18}},\ \bibinfo {pages} {2423--2426} (\bibinfo {year}
  {1977})}\BibitemShut {NoStop}%
\bibitem [{\citenamefont {Parr}\ and\ \citenamefont {Yang}(1994)}]{ParYan-94}%
  \BibitemOpen
  \bibfield  {author} {\bibinfo {author} {\bibnamefont {Parr}, \bibfnamefont
  {R.}}\ and\ \bibinfo {author} {\bibnamefont {Yang}, \bibfnamefont {W.}},\
  }\href@noop {} {\emph {\bibinfo {title} {Density-Functional Theory of Atoms
  and Molecules}}},\ International Series of Monographs on Chemistry\ (\bibinfo
   {publisher} {Oxford University Press, USA},\ \bibinfo {year}
  {1994})\BibitemShut {NoStop}%
\bibitem [{\citenamefont {Pass}(2015)}]{Pass-15}%
  \BibitemOpen
  \bibfield  {author} {\bibinfo {author} {\bibnamefont {Pass}, \bibfnamefont
  {B.}},\ }\bibfield  {title} {\enquote {\bibinfo {title} {Multi-marginal
  optimal transport: theory and applications},}\ }\href {\doibase
  10.1051/m2an/2015020} {\bibfield  {journal} {\bibinfo  {journal} {ESAIM Math.
  Model. Numer. Anal.}\ }\textbf {\bibinfo {volume} {49}},\ \bibinfo {pages}
  {1771--1790} (\bibinfo {year} {2015})}\BibitemShut {NoStop}%
\bibitem [{\citenamefont {Paulin}, \citenamefont {Ackerson},\ and\
  \citenamefont {Wolfe}(1996)}]{PauAckWol-96}%
  \BibitemOpen
  \bibfield  {author} {\bibinfo {author} {\bibnamefont {Paulin}, \bibfnamefont
  {S.}}, \bibinfo {author} {\bibnamefont {Ackerson}, \bibfnamefont {B.~J.}}, \
  and\ \bibinfo {author} {\bibnamefont {Wolfe}, \bibfnamefont {M.}},\
  }\bibfield  {title} {\enquote {\bibinfo {title} {Equilibrium and shear
  induced nonequilibrium phase behavior of pmma microgel spheres},}\ }\href
  {\doibase 10.1006/jcis.1996.0113} {\bibfield  {journal} {\bibinfo  {journal}
  {J. Colloid Interf. Sci.}\ }\textbf {\bibinfo {volume} {178}},\ \bibinfo
  {pages} {251--262} (\bibinfo {year} {1996})}\BibitemShut {NoStop}%
\bibitem [{\citenamefont {Penrose}(1963)}]{Penrose-63}%
  \BibitemOpen
  \bibfield  {author} {\bibinfo {author} {\bibnamefont {Penrose}, \bibfnamefont
  {O.}},\ }\bibfield  {title} {\enquote {\bibinfo {title} {Convergence of
  fugacity expansions for fluids and lattice gases},}\ }\href {\doibase
  10.1063/1.1703906} {\bibfield  {journal} {\bibinfo  {journal} {J.
  Mathematical Phys.}\ }\textbf {\bibinfo {volume} {4}},\ \bibinfo {pages}
  {1312--1320} (\bibinfo {year} {1963})}\BibitemShut {NoStop}%
\bibitem [{\citenamefont {Penrose}\ and\ \citenamefont
  {Smith}(1972)}]{PenSmi-72}%
  \BibitemOpen
  \bibfield  {author} {\bibinfo {author} {\bibnamefont {Penrose}, \bibfnamefont
  {O.}}\ and\ \bibinfo {author} {\bibnamefont {Smith}, \bibfnamefont {E.~R.}},\
  }\bibfield  {title} {\enquote {\bibinfo {title} {Thermodynamic limit for
  classical systems with {C}oulomb interactions in a constant external
  field},}\ }\href {http://projecteuclid.org/euclid.cmp/1103858003} {\bibfield
  {journal} {\bibinfo  {journal} {Comm. Math. Phys.}\ }\textbf {\bibinfo
  {volume} {26}},\ \bibinfo {pages} {53--77} (\bibinfo {year}
  {1972})}\BibitemShut {NoStop}%
\bibitem [{\citenamefont {Perdew}(1991)}]{Perdew-91}%
  \BibitemOpen
  \bibfield  {author} {\bibinfo {author} {\bibnamefont {Perdew}, \bibfnamefont
  {J.~P.}},\ }\bibfield  {title} {\enquote {\bibinfo {title} {Unified {T}heory
  of {E}xchange and {C}orrelation {B}eyond the {L}ocal {D}ensity
  {A}pproximation},}\ }in\ \href@noop {} {\emph {\bibinfo {booktitle}
  {Electronic Structure of Solids '91}}},\ \bibinfo {editor} {edited by\
  \bibinfo {editor} {\bibfnamefont {P.}~\bibnamefont {Ziesche}}\ and\ \bibinfo
  {editor} {\bibfnamefont {H.}~\bibnamefont {Eschrig}}}\ (\bibinfo  {publisher}
  {Akademie Verlag, Berlin},\ \bibinfo {year} {1991})\ pp.\ \bibinfo {pages}
  {11--20}\BibitemShut {NoStop}%
\bibitem [{\citenamefont {Perdew}, \citenamefont {Burke},\ and\ \citenamefont
  {Ernzerhof}(1996)}]{PerBurErn-96}%
  \BibitemOpen
  \bibfield  {author} {\bibinfo {author} {\bibnamefont {Perdew}, \bibfnamefont
  {J.~P.}}, \bibinfo {author} {\bibnamefont {Burke}, \bibfnamefont {K.}}, \
  and\ \bibinfo {author} {\bibnamefont {Ernzerhof}, \bibfnamefont {M.}},\
  }\bibfield  {title} {\enquote {\bibinfo {title} {Generalized gradient
  approximation made simple},}\ }\href {\doibase 10.1103/PhysRevLett.77.3865}
  {\bibfield  {journal} {\bibinfo  {journal} {Phys. Rev. Lett.}\ }\textbf
  {\bibinfo {volume} {77}},\ \bibinfo {pages} {3865--3868} (\bibinfo {year}
  {1996})}\BibitemShut {NoStop}%
\bibitem [{\citenamefont {Perdew}\ and\ \citenamefont
  {Kurth}(2003)}]{PerKur-03}%
  \BibitemOpen
  \bibfield  {author} {\bibinfo {author} {\bibnamefont {Perdew}, \bibfnamefont
  {J.~P.}}\ and\ \bibinfo {author} {\bibnamefont {Kurth}, \bibfnamefont {S.}},\
  }\enquote {\bibinfo {title} {Density functionals for non-relativistic
  {C}oulomb systems in the new century},}\ in\ \href {\doibase
  10.1007/3-540-37072-2_1} {\emph {\bibinfo {booktitle} {A Primer in Density
  Functional Theory}}},\ \bibinfo {editor} {edited by\ \bibinfo {editor}
  {\bibfnamefont {C.}~\bibnamefont {Fiolhais}}, \bibinfo {editor}
  {\bibfnamefont {F.}~\bibnamefont {Nogueira}}, \ and\ \bibinfo {editor}
  {\bibfnamefont {M.~A.~L.}\ \bibnamefont {Marques}}}\ (\bibinfo  {publisher}
  {Springer Berlin Heidelberg},\ \bibinfo {address} {Berlin, Heidelberg},\
  \bibinfo {year} {2003})\ pp.\ \bibinfo {pages} {1--55}\BibitemShut {NoStop}%
\bibitem [{\citenamefont {Perdew}\ and\ \citenamefont
  {Wang}(1992)}]{PerWan-92}%
  \BibitemOpen
  \bibfield  {author} {\bibinfo {author} {\bibnamefont {Perdew}, \bibfnamefont
  {J.~P.}}\ and\ \bibinfo {author} {\bibnamefont {Wang}, \bibfnamefont {Y.}},\
  }\bibfield  {title} {\enquote {\bibinfo {title} {Accurate and simple analytic
  representation of the electron-gas correlation energy},}\ }\href {\doibase
  10.1103/PhysRevB.45.13244} {\bibfield  {journal} {\bibinfo  {journal} {Phys.
  Rev. B}\ }\textbf {\bibinfo {volume} {45}},\ \bibinfo {pages} {13244--13249}
  (\bibinfo {year} {1992})}\BibitemShut {NoStop}%
\bibitem [{\citenamefont {Peres}\ and\ \citenamefont {Sly}(2014)}]{PerSly-14}%
  \BibitemOpen
  \bibfield  {author} {\bibinfo {author} {\bibnamefont {Peres}, \bibfnamefont
  {Y.}}\ and\ \bibinfo {author} {\bibnamefont {Sly}, \bibfnamefont {A.}},\
  }\href@noop {} {\enquote {\bibinfo {title} {Rigidity and tolerance for
  perturbed lattices},}\ } (\bibinfo {year} {2014}),\ \Eprint
  {http://arxiv.org/abs/1409.4490} {arXiv:1409.4490 [math.PR]} \BibitemShut
  {NoStop}%
\bibitem [{\citenamefont {Petrache}\ and\ \citenamefont
  {Rota~Nodari}(2018)}]{PetRot-18}%
  \BibitemOpen
  \bibfield  {author} {\bibinfo {author} {\bibnamefont {Petrache},
  \bibfnamefont {M.}}\ and\ \bibinfo {author} {\bibnamefont {Rota~Nodari},
  \bibfnamefont {S.}},\ }\bibfield  {title} {\enquote {\bibinfo {title}
  {Equidistribution of jellium energy for {C}oulomb and {R}iesz
  interactions},}\ }\href {\doibase 10.1007/s00365-017-9395-1} {\bibfield
  {journal} {\bibinfo  {journal} {Constr. Approx.}\ }\textbf {\bibinfo {volume}
  {47}},\ \bibinfo {pages} {163--210} (\bibinfo {year} {2018})}\BibitemShut
  {NoStop}%
\bibitem [{\citenamefont {Petrache}\ and\ \citenamefont
  {Serfaty}(2017)}]{PetSer-17}%
  \BibitemOpen
  \bibfield  {author} {\bibinfo {author} {\bibnamefont {Petrache},
  \bibfnamefont {M.}}\ and\ \bibinfo {author} {\bibnamefont {Serfaty},
  \bibfnamefont {S.}},\ }\bibfield  {title} {\enquote {\bibinfo {title} {{Next
  Order Asymptotics and Renormalized Energy for Riesz Interactions}},}\ }\href
  {\doibase 10.1017/S1474748015000201} {\bibfield  {journal} {\bibinfo
  {journal} {J. Inst. Math. Jussieu}\ }\textbf {\bibinfo {volume} {16}},\
  \bibinfo {pages} {501--569} (\bibinfo {year} {2017})}\BibitemShut {NoStop}%
\bibitem [{\citenamefont {{Petrache}}\ and\ \citenamefont
  {{Serfaty}}(2020)}]{PetSer-20}%
  \BibitemOpen
  \bibfield  {author} {\bibinfo {author} {\bibnamefont {{Petrache}},
  \bibfnamefont {M.}}\ and\ \bibinfo {author} {\bibnamefont {{Serfaty}},
  \bibfnamefont {S.}},\ }\bibfield  {title} {\enquote {\bibinfo {title}
  {{Crystallization for Coulomb and Riesz interactions as a consequence of the
  Cohn-Kumar conjecture}},}\ }\href {\doibase 10.1090/proc/15003} {\bibfield
  {journal} {\bibinfo  {journal} {{Proc. Am. Math. Soc.}}\ }\textbf {\bibinfo
  {volume} {148}},\ \bibinfo {pages} {3047--3057} (\bibinfo {year}
  {2020})}\BibitemShut {NoStop}%
\bibitem [{\citenamefont {Placzek}, \citenamefont {Nijboer},\ and\
  \citenamefont {Hove}(1951)}]{PlaNijHov-51}%
  \BibitemOpen
  \bibfield  {author} {\bibinfo {author} {\bibnamefont {Placzek}, \bibfnamefont
  {G.}}, \bibinfo {author} {\bibnamefont {Nijboer}, \bibfnamefont {B.~R.~A.}},
  \ and\ \bibinfo {author} {\bibnamefont {Hove}, \bibfnamefont {L.~V.}},\
  }\bibfield  {title} {\enquote {\bibinfo {title} {Effect of short wavelength
  interference on neuteron scattering by dense systems of heavy nuclei},}\
  }\href {\doibase 10.1103/PhysRev.82.392} {\bibfield  {journal} {\bibinfo
  {journal} {Phys. Rev.}\ }\textbf {\bibinfo {volume} {82}},\ \bibinfo {pages}
  {392--403} (\bibinfo {year} {1951})}\BibitemShut {NoStop}%
\bibitem [{\citenamefont {Pollock}\ and\ \citenamefont
  {Hansen}(1973)}]{PolHan-73}%
  \BibitemOpen
  \bibfield  {author} {\bibinfo {author} {\bibnamefont {Pollock}, \bibfnamefont
  {E.}}\ and\ \bibinfo {author} {\bibnamefont {Hansen}, \bibfnamefont {J.}},\
  }\bibfield  {title} {\enquote {\bibinfo {title} {{Statistical Mechanics of
  Dense Ionized Matter. II. Equilibrium Properties and Melting Transition of
  the Crystallized One-Component Plasma}},}\ }\href {\doibase
  10.1103/PhysRevA.8.3110} {\bibfield  {journal} {\bibinfo  {journal} {Phys.
  Rev. A}\ }\textbf {\bibinfo {volume} {8}},\ \bibinfo {pages} {3110--3122}
  (\bibinfo {year} {1973})}\BibitemShut {NoStop}%
\bibitem [{\citenamefont {Prager}(1962)}]{Prager-62}%
  \BibitemOpen
  \bibfield  {author} {\bibinfo {author} {\bibnamefont {Prager}, \bibfnamefont
  {S.}},\ }\bibfield  {title} {\enquote {\bibinfo {title} {The one-dimensional
  plasma},}\ }in\ \href@noop {} {\emph {\bibinfo {booktitle} {Advances in
  {C}hemical {P}hysics, {V}ol. {IV}}}}\ (\bibinfo  {publisher} {Interscience,
  New York},\ \bibinfo {year} {1962})\ pp.\ \bibinfo {pages}
  {201--224}\BibitemShut {NoStop}%
\bibitem [{\citenamefont {Prestipino}, \citenamefont {Saija},\ and\
  \citenamefont {Giaquinta}(2005)}]{PreSaiGia-05}%
  \BibitemOpen
  \bibfield  {author} {\bibinfo {author} {\bibnamefont {Prestipino},
  \bibfnamefont {S.}}, \bibinfo {author} {\bibnamefont {Saija}, \bibfnamefont
  {F.}}, \ and\ \bibinfo {author} {\bibnamefont {Giaquinta}, \bibfnamefont
  {P.~V.}},\ }\bibfield  {title} {\enquote {\bibinfo {title} {Phase diagram of
  softly repulsive systems: The {G}aussian and inverse-power-law potentials},}\
  }\href {\doibase 10.1063/1.2064639} {\bibfield  {journal} {\bibinfo
  {journal} {J. Chem. Phys.}\ }\textbf {\bibinfo {volume} {123}},\ \bibinfo
  {pages} {144110} (\bibinfo {year} {2005})}\BibitemShut {NoStop}%
\bibitem [{\citenamefont {Preston}(1974)}]{Preston-74}%
  \BibitemOpen
  \bibfield  {author} {\bibinfo {author} {\bibnamefont {Preston}, \bibfnamefont
  {C.~J.}},\ }\href@noop {} {\emph {\bibinfo {title} {Gibbs states on countable
  sets}}},\ Cambridge Tracts in Mathematics, No. 68\ (\bibinfo  {publisher}
  {Cambridge University Press, London-New York},\ \bibinfo {year} {1974})\ pp.\
  \bibinfo {pages} {ix+128}\BibitemShut {NoStop}%
\bibitem [{\citenamefont {Radin}(1987)}]{Radin-87}%
  \BibitemOpen
  \bibfield  {author} {\bibinfo {author} {\bibnamefont {Radin}, \bibfnamefont
  {C.}},\ }\bibfield  {title} {\enquote {\bibinfo {title} {Low temperature and
  the origin of crystalline symmetry},}\ }\href {\doibase
  10.1142/S0217979287001675} {\bibfield  {journal} {\bibinfo  {journal}
  {Internat. J. Modern Phys. B}\ }\textbf {\bibinfo {volume} {1}},\ \bibinfo
  {pages} {1157--1191} (\bibinfo {year} {1987})}\BibitemShut {NoStop}%
\bibitem [{\citenamefont {Rankin}(1953)}]{Rankin-53}%
  \BibitemOpen
  \bibfield  {author} {\bibinfo {author} {\bibnamefont {Rankin}, \bibfnamefont
  {R.~A.}},\ }\bibfield  {title} {\enquote {\bibinfo {title} {A minimum problem
  for the {E}pstein zeta-function},}\ }\href {\doibase
  10.1017/S2040618500035668} {\bibfield  {journal} {\bibinfo  {journal} {Proc.
  Glasgow Math. Assoc.}\ }\textbf {\bibinfo {volume} {1}},\ \bibinfo {pages}
  {149--158} (\bibinfo {year} {1953})}\BibitemShut {NoStop}%
\bibitem [{\citenamefont {Rebenko}(1988)}]{Rebenko-88}%
  \BibitemOpen
  \bibfield  {author} {\bibinfo {author} {\bibnamefont {Rebenko}, \bibfnamefont
  {A.~L.}},\ }\bibfield  {title} {\enquote {\bibinfo {title} {Mathematical
  foundations of the equilibrium classical statistical mechanics of charged
  particles},}\ }\href {\doibase 10.1070/RM1988v043n03ABEH001744} {\bibfield
  {journal} {\bibinfo  {journal} {Uspekhi Mat. Nauk}\ }\textbf {\bibinfo
  {volume} {43}},\ \bibinfo {pages} {55--97, 271, 272} (\bibinfo {year}
  {1988})}\BibitemShut {NoStop}%
\bibitem [{\citenamefont {Reda}\ and\ \citenamefont
  {Najnudel}(2018)}]{RedNaj-18}%
  \BibitemOpen
  \bibfield  {author} {\bibinfo {author} {\bibnamefont {Reda}, \bibfnamefont
  {C.}}\ and\ \bibinfo {author} {\bibnamefont {Najnudel}, \bibfnamefont {J.}},\
  }\bibfield  {title} {\enquote {\bibinfo {title} {Rigidity of the {${\rm
  Sine}_{\beta}$} process},}\ }\href {\doibase 10.1214/18-ECP195} {\bibfield
  {journal} {\bibinfo  {journal} {Electron. Commun. Probab.}\ }\textbf
  {\bibinfo {volume} {23}},\ \bibinfo {pages} {Paper No. 94, 8} (\bibinfo
  {year} {2018})}\BibitemShut {NoStop}%
\bibitem [{\citenamefont {Reed}\ and\ \citenamefont {Simon}(1978)}]{ReeSim4}%
  \BibitemOpen
  \bibfield  {author} {\bibinfo {author} {\bibnamefont {Reed}, \bibfnamefont
  {M.}}\ and\ \bibinfo {author} {\bibnamefont {Simon}, \bibfnamefont {B.}},\
  }\href@noop {} {\emph {\bibinfo {title} {Methods of {M}odern {M}athematical
  {P}hysics. {IV}. {A}nalysis of operators}}}\ (\bibinfo  {publisher} {Academic
  Press},\ \bibinfo {address} {New York},\ \bibinfo {year} {1978})\ pp.\
  \bibinfo {pages} {xv+396}\BibitemShut {NoStop}%
\bibitem [{\citenamefont {Requardt}\ and\ \citenamefont
  {Wagner}(1990)}]{ReqWag-90}%
  \BibitemOpen
  \bibfield  {author} {\bibinfo {author} {\bibnamefont {Requardt},
  \bibfnamefont {M.}}\ and\ \bibinfo {author} {\bibnamefont {Wagner},
  \bibfnamefont {H.~J.}},\ }\bibfield  {title} {\enquote {\bibinfo {title}
  {Wigner crystallization and its relation to the poor decay of pair
  correlations in one-component plasmas of arbitrary dimension},}\ }\href
  {https://doi.org/10.1007/BF01026570} {\bibfield  {journal} {\bibinfo
  {journal} {J. Statist. Phys.}\ }\textbf {\bibinfo {volume} {58}},\ \bibinfo
  {pages} {1165--1180} (\bibinfo {year} {1990})}\BibitemShut {NoStop}%
\bibitem [{\citenamefont {Riesz}(8 40)}]{Riesz-38}%
  \BibitemOpen
  \bibfield  {author} {\bibinfo {author} {\bibnamefont {Riesz}, \bibfnamefont
  {M.}},\ }\bibfield  {title} {\enquote {\bibinfo {title} {Intégrales de
  {R}iemann--{L}iouville et potentiels},}\ }\href@noop {} {\bibfield  {journal}
  {\bibinfo  {journal} {Acta Sci. Math. (Szeged)}\ }\textbf {\bibinfo {volume}
  {9}},\ \bibinfo {pages} {1--42} (\bibinfo {year} {1938-40})}\BibitemShut
  {NoStop}%
\bibitem [{\citenamefont {Ros-Oton}\ and\ \citenamefont
  {Serra}(2016)}]{RosSer-16}%
  \BibitemOpen
  \bibfield  {author} {\bibinfo {author} {\bibnamefont {Ros-Oton},
  \bibfnamefont {X.}}\ and\ \bibinfo {author} {\bibnamefont {Serra},
  \bibfnamefont {J.}},\ }\bibfield  {title} {\enquote {\bibinfo {title}
  {Regularity theory for general stable operators},}\ }\href {\doibase
  10.1016/j.jde.2016.02.033} {\bibfield  {journal} {\bibinfo  {journal} {J.
  Differential Equations}\ }\textbf {\bibinfo {volume} {260}},\ \bibinfo
  {pages} {8675--8715} (\bibinfo {year} {2016})}\BibitemShut {NoStop}%
\bibitem [{\citenamefont {Ros-Oton}\ and\ \citenamefont
  {Serra}(2017)}]{RosSer-17}%
  \BibitemOpen
  \bibfield  {author} {\bibinfo {author} {\bibnamefont {Ros-Oton},
  \bibfnamefont {X.}}\ and\ \bibinfo {author} {\bibnamefont {Serra},
  \bibfnamefont {J.}},\ }\bibfield  {title} {\enquote {\bibinfo {title}
  {Boundary regularity estimates for nonlocal elliptic equations in {$C^1$} and
  {$C^{1,\alpha}$} domains},}\ }\href {\doibase 10.1007/s10231-016-0632-1}
  {\bibfield  {journal} {\bibinfo  {journal} {Ann. Mat. Pura Appl. (4)}\
  }\textbf {\bibinfo {volume} {196}},\ \bibinfo {pages} {1637--1668} (\bibinfo
  {year} {2017})}\BibitemShut {NoStop}%
\bibitem [{\citenamefont {{Rota Nodari}}\ and\ \citenamefont
  {Serfaty}(2015)}]{RotSer-15}%
  \BibitemOpen
  \bibfield  {author} {\bibinfo {author} {\bibnamefont {{Rota Nodari}},
  \bibfnamefont {S.}}\ and\ \bibinfo {author} {\bibnamefont {Serfaty},
  \bibfnamefont {S.}},\ }\bibfield  {title} {\enquote {\bibinfo {title}
  {Renormalized energy equidistribution and local charge balance in 2d
  {C}oulomb system},}\ }\href {\doibase 10.1093/imrn/rnu031} {\bibfield
  {journal} {\bibinfo  {journal} {Int. Math. Res. Not. (IMRN)}\ }\textbf
  {\bibinfo {volume} {11}},\ \bibinfo {pages} {3035--3093} (\bibinfo {year}
  {2015})}\BibitemShut {NoStop}%
\bibitem [{\citenamefont {{Rougerie}}(2015)}]{Rougerie-LMU}%
  \BibitemOpen
  \bibfield  {author} {\bibinfo {author} {\bibnamefont {{Rougerie}},
  \bibfnamefont {N.}},\ }\bibfield  {title} {\enquote {\bibinfo {title} {De
  {F}inetti theorems, mean-field limits and {B}ose-{E}instein condensation},}\
  }\href@noop {} {\bibfield  {journal} {\bibinfo  {journal} {ArXiv e-prints}\ }
  (\bibinfo {year} {2015})},\ \Eprint {http://arxiv.org/abs/1506.05263}
  {arXiv:1506.05263 [math-ph]} \BibitemShut {NoStop}%
\bibitem [{\citenamefont {Rougerie}(2022)}]{Rougerie-22_ppt}%
  \BibitemOpen
  \bibfield  {author} {\bibinfo {author} {\bibnamefont {Rougerie},
  \bibfnamefont {N.}},\ }\enquote {\bibinfo {title} {Book in honor of {E}lliott
  {L}ieb 90th birthday},}\ \ (\bibinfo  {publisher} {EMS Press},\ \bibinfo
  {year} {2022})\ Chap.\ \bibinfo {chapter} {The classical {J}ellium and the
  {L}aughlin phase}\BibitemShut {NoStop}%
\bibitem [{\citenamefont {{Rougerie}}\ and\ \citenamefont
  {{Serfaty}}(2016)}]{RouSer-16}%
  \BibitemOpen
  \bibfield  {author} {\bibinfo {author} {\bibnamefont {{Rougerie}},
  \bibfnamefont {N.}}\ and\ \bibinfo {author} {\bibnamefont {{Serfaty}},
  \bibfnamefont {S.}},\ }\bibfield  {title} {\enquote {\bibinfo {title}
  {{Higher Dimensional Coulomb Gases and Renormalized Energy Functionals}},}\
  }\href {\doibase 10.1002/cpa.21570} {\bibfield  {journal} {\bibinfo
  {journal} {Comm. Pure Appl. Math.}\ }\textbf {\bibinfo {volume} {69}},\
  \bibinfo {pages} {519--605} (\bibinfo {year} {2016})}\BibitemShut {NoStop}%
\bibitem [{\citenamefont {Rougerie}, \citenamefont {Serfaty},\ and\
  \citenamefont {Yngvason}(2013)}]{RouSerYng-13b}%
  \BibitemOpen
  \bibfield  {author} {\bibinfo {author} {\bibnamefont {Rougerie},
  \bibfnamefont {N.}}, \bibinfo {author} {\bibnamefont {Serfaty}, \bibfnamefont
  {S.}}, \ and\ \bibinfo {author} {\bibnamefont {Yngvason}, \bibfnamefont
  {J.}},\ }\bibfield  {title} {\enquote {\bibinfo {title} {{Quantum Hall Phases
  and Plasma Analogy in Rotating Trapped Bose Gases}},}\ }\href {\doibase
  10.1007/s10955-013-0766-0} {\bibfield  {journal} {\bibinfo  {journal} {J.
  Stat. Phys.}\ ,\ \bibinfo {pages} {1--49}} (\bibinfo {year} {2013})},\
  \Eprint {http://arxiv.org/abs/1301.1043} {1301.1043} \BibitemShut {NoStop}%
\bibitem [{\citenamefont {Rudnick}\ and\ \citenamefont
  {Sarnak}(1996)}]{RudSar-96}%
  \BibitemOpen
  \bibfield  {author} {\bibinfo {author} {\bibnamefont {Rudnick}, \bibfnamefont
  {Z.}}\ and\ \bibinfo {author} {\bibnamefont {Sarnak}, \bibfnamefont {P.}},\
  }\bibfield  {title} {\enquote {\bibinfo {title} {Zeros of principal
  {$L$}-functions and random matrix theory},}\ }\href {\doibase
  10.1215/S0012-7094-96-08115-6} {\bibfield  {journal} {\bibinfo  {journal}
  {Duke Math. J.}\ }\textbf {\bibinfo {volume} {81}},\ \bibinfo {pages}
  {269--322} (\bibinfo {year} {1996})},\ \bibinfo {note} {a celebration of John
  F. Nash, Jr.}\BibitemShut {Stop}%
\bibitem [{\citenamefont {Ruelle}(1963{\natexlab{a}})}]{Ruelle-63a}%
  \BibitemOpen
  \bibfield  {author} {\bibinfo {author} {\bibnamefont {Ruelle}, \bibfnamefont
  {D.}},\ }\bibfield  {title} {\enquote {\bibinfo {title} {Classical
  statistical mechanics of a system of particles},}\ }\href@noop {} {\bibfield
  {journal} {\bibinfo  {journal} {Helv. Phys. Acta}\ }\textbf {\bibinfo
  {volume} {36}},\ \bibinfo {pages} {183--197} (\bibinfo {year}
  {1963}{\natexlab{a}})}\BibitemShut {NoStop}%
\bibitem [{\citenamefont {Ruelle}(1963{\natexlab{b}})}]{Ruelle-63}%
  \BibitemOpen
  \bibfield  {author} {\bibinfo {author} {\bibnamefont {Ruelle}, \bibfnamefont
  {D.}},\ }\bibfield  {title} {\enquote {\bibinfo {title} {Correlation
  functions of classical gases},}\ }\href {\doibase
  10.1016/0003-4916(63)90336-1} {\bibfield  {journal} {\bibinfo  {journal}
  {Annals of Physics}\ }\textbf {\bibinfo {volume} {25}},\ \bibinfo {pages}
  {109--120} (\bibinfo {year} {1963}{\natexlab{b}})}\BibitemShut {NoStop}%
\bibitem [{\citenamefont {Ruelle}(1964)}]{Ruelle-64}%
  \BibitemOpen
  \bibfield  {author} {\bibinfo {author} {\bibnamefont {Ruelle}, \bibfnamefont
  {D.}},\ }\bibfield  {title} {\enquote {\bibinfo {title} {Cluster property of
  the correlation functions of classical gases},}\ }\href {\doibase
  10.1103/RevModPhys.36.580} {\bibfield  {journal} {\bibinfo  {journal} {Rev.
  Modern Phys.}\ }\textbf {\bibinfo {volume} {36}},\ \bibinfo {pages}
  {580--584} (\bibinfo {year} {1964})}\BibitemShut {NoStop}%
\bibitem [{\citenamefont {Ruelle}(1970)}]{Ruelle-70}%
  \BibitemOpen
  \bibfield  {author} {\bibinfo {author} {\bibnamefont {Ruelle}, \bibfnamefont
  {D.}},\ }\bibfield  {title} {\enquote {\bibinfo {title} {Superstable
  interactions in classical statistical mechanics},}\ }\href
  {http://projecteuclid.org/euclid.cmp/1103842505} {\bibfield  {journal}
  {\bibinfo  {journal} {Comm. Math. Phys.}\ }\textbf {\bibinfo {volume} {18}},\
  \bibinfo {pages} {127--159} (\bibinfo {year} {1970})}\BibitemShut {NoStop}%
\bibitem [{\citenamefont {Ruelle}(1999)}]{Ruelle}%
  \BibitemOpen
  \bibfield  {author} {\bibinfo {author} {\bibnamefont {Ruelle}, \bibfnamefont
  {D.}},\ }\href@noop {} {\emph {\bibinfo {title} {Statistical mechanics.
  Rigorous results}}}\ (\bibinfo  {publisher} {{Singapore: World Scientific.
  London: Imperial College Press}},\ \bibinfo {year} {1999})\BibitemShut
  {NoStop}%
\bibitem [{\citenamefont {Ruijsenaars}(1995)}]{Ruijsenaars-95}%
  \BibitemOpen
  \bibfield  {author} {\bibinfo {author} {\bibnamefont {Ruijsenaars},
  \bibfnamefont {S.}},\ }\bibfield  {title} {\enquote {\bibinfo {title}
  {Action-angle maps and scattering theory for some finite-dimensional
  integrable systems. {III}. {S}utherland type systems and their duals},}\
  }\href {\doibase 10.2977/prims/1195164440} {\bibfield  {journal} {\bibinfo
  {journal} {Publ. Res. Inst. Math. Sci.}\ }\textbf {\bibinfo {volume} {31}},\
  \bibinfo {pages} {247--353} (\bibinfo {year} {1995})}\BibitemShut {NoStop}%
\bibitem [{\citenamefont {Ry\v{s}kov}(1973)}]{Ryshkov-73}%
  \BibitemOpen
  \bibfield  {author} {\bibinfo {author} {\bibnamefont {Ry\v{s}kov},
  \bibfnamefont {S.~S.}},\ }\bibfield  {title} {\enquote {\bibinfo {title} {On
  the question of the final {$\zeta $}-optimality of lattices that yield the
  densest packing of {$n$}-dimensional balls},}\ }\href@noop {} {\bibfield
  {journal} {\bibinfo  {journal} {Siberian Math. J.}\ }\textbf {\bibinfo
  {volume} {14}},\ \bibinfo {pages} {743--750} (\bibinfo {year}
  {1973})}\BibitemShut {NoStop}%
\bibitem [{\citenamefont {Saff}\ and\ \citenamefont
  {Kuijlaars}(1997)}]{SafKui-97}%
  \BibitemOpen
  \bibfield  {author} {\bibinfo {author} {\bibnamefont {Saff}, \bibfnamefont
  {E.~B.}}\ and\ \bibinfo {author} {\bibnamefont {Kuijlaars}, \bibfnamefont
  {A.~B.~J.}},\ }\bibfield  {title} {\enquote {\bibinfo {title} {Distributing
  many points on a sphere},}\ }\href {\doibase 10.1007/BF03024331} {\bibfield
  {journal} {\bibinfo  {journal} {Math. Intelligencer}\ }\textbf {\bibinfo
  {volume} {19}},\ \bibinfo {pages} {5--11} (\bibinfo {year}
  {1997})}\BibitemShut {NoStop}%
\bibitem [{\citenamefont {Sakai}(1982)}]{Sakai-82}%
  \BibitemOpen
  \bibfield  {author} {\bibinfo {author} {\bibnamefont {Sakai}, \bibfnamefont
  {M.}},\ }\href@noop {} {\emph {\bibinfo {title} {Quadrature domains}}},\
  \bibinfo {series} {Lecture Notes in Mathematics}, Vol.\ \bibinfo {volume}
  {934}\ (\bibinfo  {publisher} {Springer-Verlag, Berlin-New York},\ \bibinfo
  {year} {1982})\ pp.\ \bibinfo {pages} {i+133}\BibitemShut {NoStop}%
\bibitem [{\citenamefont {Salazar}(2017)}]{Salazar-PhD}%
  \BibitemOpen
  \bibfield  {author} {\bibinfo {author} {\bibnamefont {Salazar}, \bibfnamefont
  {R.}},\ }\emph {\bibinfo {title} {{Exact results and melting theories in
  two-dimensional systems}}},\ \href
  {https://tel.archives-ouvertes.fr/tel-01689945} {Ph.D. thesis},\ \bibinfo
  {school} {{Universit{\'e} Paris Saclay ; Universidad de los Andes
  (Bogot{\'a})}} (\bibinfo {year} {2017})\BibitemShut {NoStop}%
\bibitem [{\citenamefont {Salazar}\ and\ \citenamefont
  {T\'{e}llez}(2016)}]{SalTel-16}%
  \BibitemOpen
  \bibfield  {author} {\bibinfo {author} {\bibnamefont {Salazar}, \bibfnamefont
  {R.}}\ and\ \bibinfo {author} {\bibnamefont {T\'{e}llez}, \bibfnamefont
  {G.}},\ }\bibfield  {title} {\enquote {\bibinfo {title} {Exact energy
  computation of the one component plasma on a sphere for even values of the
  coupling parameter},}\ }\href {\doibase 10.1007/s10955-016-1562-4} {\bibfield
   {journal} {\bibinfo  {journal} {J. Stat. Phys.}\ }\textbf {\bibinfo {volume}
  {164}},\ \bibinfo {pages} {969--999} (\bibinfo {year} {2016})}\BibitemShut
  {NoStop}%
\bibitem [{\citenamefont {{Salpeter}}(1961)}]{Salpeter-61}%
  \BibitemOpen
  \bibfield  {author} {\bibinfo {author} {\bibnamefont {{Salpeter}},
  \bibfnamefont {E.~E.}},\ }\bibfield  {title} {\enquote {\bibinfo {title}
  {Energy and pressure of a zero-temperature plasma},}\ }\href {\doibase
  10.1086/147194} {\bibfield  {journal} {\bibinfo  {journal} {Astrophys. J.}\
  }\textbf {\bibinfo {volume} {134}},\ \bibinfo {pages} {669} (\bibinfo {year}
  {1961})}\BibitemShut {NoStop}%
\bibitem [{\citenamefont {Salzberg}\ and\ \citenamefont
  {Prager}(1963)}]{SalPra-63}%
  \BibitemOpen
  \bibfield  {author} {\bibinfo {author} {\bibnamefont {Salzberg},
  \bibfnamefont {A.~M.}}\ and\ \bibinfo {author} {\bibnamefont {Prager},
  \bibfnamefont {S.}},\ }\bibfield  {title} {\enquote {\bibinfo {title}
  {Equation of state for a two‐dimensional electrolyte},}\ }\href {\doibase
  10.1063/1.1733553} {\bibfield  {journal} {\bibinfo  {journal} {J. Chem.
  Phys.}\ }\textbf {\bibinfo {volume} {38}},\ \bibinfo {pages} {2587--2587}
  (\bibinfo {year} {1963})},\ \Eprint
  {http://arxiv.org/abs/https://doi.org/10.1063/1.1733553}
  {https://doi.org/10.1063/1.1733553} \BibitemShut {NoStop}%
\bibitem [{\citenamefont {Saminadayar}\ \emph {et~al.}(1997)\citenamefont
  {Saminadayar}, \citenamefont {Glattli}, \citenamefont {Jin},\ and\
  \citenamefont {Etienne}}]{SamGlaJinEti-97}%
  \BibitemOpen
  \bibfield  {author} {\bibinfo {author} {\bibnamefont {Saminadayar},
  \bibfnamefont {L.}}, \bibinfo {author} {\bibnamefont {Glattli}, \bibfnamefont
  {D.~C.}}, \bibinfo {author} {\bibnamefont {Jin}, \bibfnamefont {Y.}}, \ and\
  \bibinfo {author} {\bibnamefont {Etienne}, \bibfnamefont {B.}},\ }\bibfield
  {title} {\enquote {\bibinfo {title} {{Observation of the $e/3$ Fractionally
  Charged Laughlin Quasiparticle}},}\ }\href {\doibase
  10.1103/PhysRevLett.79.2526} {\bibfield  {journal} {\bibinfo  {journal}
  {Phys. Rev. Lett.}\ }\textbf {\bibinfo {volume} {79}},\ \bibinfo {pages}
  {2526--2529} (\bibinfo {year} {1997})}\BibitemShut {NoStop}%
\bibitem [{\citenamefont {Sandier}\ and\ \citenamefont
  {Serfaty}(2012)}]{SanSer-12}%
  \BibitemOpen
  \bibfield  {author} {\bibinfo {author} {\bibnamefont {Sandier}, \bibfnamefont
  {E.}}\ and\ \bibinfo {author} {\bibnamefont {Serfaty}, \bibfnamefont {S.}},\
  }\bibfield  {title} {\enquote {\bibinfo {title} {From the {G}inzburg-{L}andau
  model to vortex lattice problems},}\ }\href {\doibase
  10.1007/s00220-012-1508-x} {\bibfield  {journal} {\bibinfo  {journal}
  {Commun. Math. Phys.}\ }\textbf {\bibinfo {volume} {313}},\ \bibinfo {pages}
  {635--743} (\bibinfo {year} {2012})}\BibitemShut {NoStop}%
\bibitem [{\citenamefont {Sandier}\ and\ \citenamefont
  {Serfaty}(2014)}]{SanSer-14a}%
  \BibitemOpen
  \bibfield  {author} {\bibinfo {author} {\bibnamefont {Sandier}, \bibfnamefont
  {E.}}\ and\ \bibinfo {author} {\bibnamefont {Serfaty}, \bibfnamefont {S.}},\
  }\bibfield  {title} {\enquote {\bibinfo {title} {1d log gases and the
  renormalized energy: crystallization at vanishing temperature},}\ }\href
  {\doibase 10.1007/s00440-014-0585-5} {\bibfield  {journal} {\bibinfo
  {journal} {Probab. Theory Related Fields}\ ,\ \bibinfo {pages} {1--52}}
  (\bibinfo {year} {2014})}\BibitemShut {NoStop}%
\bibitem [{\citenamefont {{Sandier}}\ and\ \citenamefont
  {{Serfaty}}(2015)}]{SanSer-15}%
  \BibitemOpen
  \bibfield  {author} {\bibinfo {author} {\bibnamefont {{Sandier}},
  \bibfnamefont {E.}}\ and\ \bibinfo {author} {\bibnamefont {{Serfaty}},
  \bibfnamefont {S.}},\ }\bibfield  {title} {\enquote {\bibinfo {title} {{2D
  Coulomb Gases and the Renormalized Energy}},}\ }\href {\doibase
  10.1214/14-AOP927} {\bibfield  {journal} {\bibinfo  {journal} {Annals of
  Proba.}\ }\textbf {\bibinfo {volume} {43}},\ \bibinfo {pages} {2026--2083}
  (\bibinfo {year} {2015})}\BibitemShut {NoStop}%
\bibitem [{\citenamefont {{Santra}}\ \emph {et~al.}(2021)\citenamefont
  {{Santra}}, \citenamefont {{Kethepalli}}, \citenamefont {{Agarwal}},
  \citenamefont {{Dhar}}, \citenamefont {{Kulkarni}},\ and\ \citenamefont
  {{Kundu}}}]{SanKetAgaDhaKulKun-21}%
  \BibitemOpen
  \bibfield  {author} {\bibinfo {author} {\bibnamefont {{Santra}},
  \bibfnamefont {S.}}, \bibinfo {author} {\bibnamefont {{Kethepalli}},
  \bibfnamefont {J.}}, \bibinfo {author} {\bibnamefont {{Agarwal}},
  \bibfnamefont {S.}}, \bibinfo {author} {\bibnamefont {{Dhar}}, \bibfnamefont
  {A.}}, \bibinfo {author} {\bibnamefont {{Kulkarni}}, \bibfnamefont {M.}}, \
  and\ \bibinfo {author} {\bibnamefont {{Kundu}}, \bibfnamefont {A.}},\
  }\bibfield  {title} {\enquote {\bibinfo {title} {{Gap statistics for confined
  particles with power-law interactions}},}\ }\href@noop {} {\bibfield
  {journal} {\bibinfo  {journal} {arXiv e-prints}\ ,\ \bibinfo {eid}
  {arXiv:2109.15026}} (\bibinfo {year} {2021})},\ \Eprint
  {http://arxiv.org/abs/2109.15026} {arXiv:2109.15026 [cond-mat.stat-mech]}
  \BibitemShut {NoStop}%
\bibitem [{\citenamefont {{Sari}}\ and\ \citenamefont
  {{Merlini}}(1976)}]{SarMer-76}%
  \BibitemOpen
  \bibfield  {author} {\bibinfo {author} {\bibnamefont {{Sari}}, \bibfnamefont
  {R.~R.}}\ and\ \bibinfo {author} {\bibnamefont {{Merlini}}, \bibfnamefont
  {D.}},\ }\bibfield  {title} {\enquote {\bibinfo {title} {On the
  {$\nu$}-dimensional one-component classical plasma: The thermodynamic limit
  problem revisited},}\ }\href {\doibase 10.1007/BF01011761} {\bibfield
  {journal} {\bibinfo  {journal} {J. Statist. Phys.}\ }\textbf {\bibinfo
  {volume} {14}},\ \bibinfo {pages} {91--100} (\bibinfo {year}
  {1976})}\BibitemShut {NoStop}%
\bibitem [{\citenamefont {Sari}, \citenamefont {Merlini},\ and\ \citenamefont
  {Calinon}(1976)}]{SarMerCal-76}%
  \BibitemOpen
  \bibfield  {author} {\bibinfo {author} {\bibnamefont {Sari}, \bibfnamefont
  {R.~R.}}, \bibinfo {author} {\bibnamefont {Merlini}, \bibfnamefont {D.}}, \
  and\ \bibinfo {author} {\bibnamefont {Calinon}, \bibfnamefont {R.}},\
  }\bibfield  {title} {\enquote {\bibinfo {title} {On the ground state of the
  one-component classical plasma},}\ }\href {\doibase
  10.1088/0305-4470/9/9/014} {\bibfield  {journal} {\bibinfo  {journal} {J.
  Phys. A}\ }\textbf {\bibinfo {volume} {9}},\ \bibinfo {pages} {1539--1551}
  (\bibinfo {year} {1976})}\BibitemShut {NoStop}%
\bibitem [{\citenamefont {Sarnak}\ and\ \citenamefont
  {Str{\"o}mbergsson}(2006)}]{SarStr-06}%
  \BibitemOpen
  \bibfield  {author} {\bibinfo {author} {\bibnamefont {Sarnak}, \bibfnamefont
  {P.}}\ and\ \bibinfo {author} {\bibnamefont {Str{\"o}mbergsson},
  \bibfnamefont {A.}},\ }\bibfield  {title} {\enquote {\bibinfo {title} {Minima
  of {E}pstein's zeta function and heights of flat tori},}\ }\href {\doibase
  10.1007/s00222-005-0488-2} {\bibfield  {journal} {\bibinfo  {journal}
  {Invent. Math.}\ }\textbf {\bibinfo {volume} {165}},\ \bibinfo {pages}
  {115--151} (\bibinfo {year} {2006})}\BibitemShut {NoStop}%
\bibitem [{\citenamefont {Schulz}(1981)}]{Schulz-81}%
  \BibitemOpen
  \bibfield  {author} {\bibinfo {author} {\bibnamefont {Schulz}, \bibfnamefont
  {H.}},\ }\bibfield  {title} {\enquote {\bibinfo {title} {Pairing transition
  of a one-dimensional classical plasma},}\ }\href
  {http://stacks.iop.org/0305-4470/14/3277} {\bibfield  {journal} {\bibinfo
  {journal} {J. Phys. A}\ }\textbf {\bibinfo {volume} {14}},\ \bibinfo {pages}
  {3277--3300} (\bibinfo {year} {1981})}\BibitemShut {NoStop}%
\bibitem [{\citenamefont {Schwartz}(1966)}]{Schwartz}%
  \BibitemOpen
  \bibfield  {author} {\bibinfo {author} {\bibnamefont {Schwartz},
  \bibfnamefont {L.}},\ }\href@noop {} {\emph {\bibinfo {title} {Th\'{e}orie
  des distributions}}},\ Publications de l'Institut de Math\'{e}matique de
  l'Universit\'{e} de Strasbourg, IX-X\ (\bibinfo  {publisher} {Hermann,
  Paris},\ \bibinfo {year} {1966})\ pp.\ \bibinfo {pages} {xiii+420},\ \bibinfo
  {note} {nouvelle \'{e}dition, enti\`{e}rement corrig\'{e}e, refondue et
  augment\'{e}e}\BibitemShut {NoStop}%
\bibitem [{\citenamefont {{\v{S}}eba}(2007)}]{Seba-07}%
  \BibitemOpen
  \bibfield  {author} {\bibinfo {author} {\bibnamefont {{\v{S}}eba},
  \bibfnamefont {P.}},\ }\bibfield  {title} {\enquote {\bibinfo {title}
  {Parking in the city},}\ }\href {\doibase 10.12693/APhysPolA.112.681}
  {\bibfield  {journal} {\bibinfo  {journal} {Acta Phys. Pol. A}\ }\textbf
  {\bibinfo {volume} {112}},\ \bibinfo {pages} {681--690} (\bibinfo {year}
  {2007})}\BibitemShut {NoStop}%
\bibitem [{\citenamefont {{\v{S}}eba}(2009)}]{Seba-09}%
  \BibitemOpen
  \bibfield  {author} {\bibinfo {author} {\bibnamefont {{\v{S}}eba},
  \bibfnamefont {P.}},\ }\bibfield  {title} {\enquote {\bibinfo {title}
  {Parking and the visual perception of space},}\ }\href {\doibase
  10.1088/1742-5468/2009/10/l10002} {\bibfield  {journal} {\bibinfo  {journal}
  {J. Stat. Mech.: Theory Exp.}\ }\textbf {\bibinfo {volume} {2009}},\ \bibinfo
  {pages} {L10002} (\bibinfo {year} {2009})}\BibitemShut {NoStop}%
\bibitem [{\citenamefont {Seidl}\ and\ \citenamefont
  {Perdew}(1994)}]{SeiPer-94}%
  \BibitemOpen
  \bibfield  {author} {\bibinfo {author} {\bibnamefont {Seidl}, \bibfnamefont
  {M.}}\ and\ \bibinfo {author} {\bibnamefont {Perdew}, \bibfnamefont
  {J.~P.}},\ }\bibfield  {title} {\enquote {\bibinfo {title} {Size-dependent
  ionization energy of a metallic cluster: {R}esolution of the classical
  image-potential paradox},}\ }\href {\doibase 10.1103/PhysRevB.50.5744}
  {\bibfield  {journal} {\bibinfo  {journal} {Phys. Rev. B}\ }\textbf {\bibinfo
  {volume} {50}},\ \bibinfo {pages} {5744--5747} (\bibinfo {year}
  {1994})}\BibitemShut {NoStop}%
\bibitem [{\citenamefont {Seiringer}\ and\ \citenamefont
  {Yngvason}(2020)}]{SeiYng-20}%
  \BibitemOpen
  \bibfield  {author} {\bibinfo {author} {\bibnamefont {Seiringer},
  \bibfnamefont {R.}}\ and\ \bibinfo {author} {\bibnamefont {Yngvason},
  \bibfnamefont {J.}},\ }\bibfield  {title} {\enquote {\bibinfo {title}
  {Emergence of {H}aldane pseudo-potentials in systems with short-range
  interactions},}\ }\href {\doibase 10.1007/s10955-020-02586-0} {\bibfield
  {journal} {\bibinfo  {journal} {J. Stat. Phys.}\ }\textbf {\bibinfo {volume}
  {181}},\ \bibinfo {pages} {448--464} (\bibinfo {year} {2020})}\BibitemShut
  {NoStop}%
\bibitem [{\citenamefont {Senff}\ and\ \citenamefont
  {Richtering}(1999)}]{SenRic-99}%
  \BibitemOpen
  \bibfield  {author} {\bibinfo {author} {\bibnamefont {Senff}, \bibfnamefont
  {H.}}\ and\ \bibinfo {author} {\bibnamefont {Richtering}, \bibfnamefont
  {W.}},\ }\bibfield  {title} {\enquote {\bibinfo {title} {Temperature
  sensitive microgel suspensions: Colloidal phase behavior and rheology of soft
  spheres},}\ }\href {\doibase 10.1063/1.479430} {\bibfield  {journal}
  {\bibinfo  {journal} {J. Chem. Phys.}\ }\textbf {\bibinfo {volume} {111}},\
  \bibinfo {pages} {1705--1711} (\bibinfo {year} {1999})}\BibitemShut {NoStop}%
\bibitem [{\citenamefont {Serfaty}(2014)}]{Serfaty-14b}%
  \BibitemOpen
  \bibfield  {author} {\bibinfo {author} {\bibnamefont {Serfaty}, \bibfnamefont
  {S.}},\ }\bibfield  {title} {\enquote {\bibinfo {title} {{Ginzburg--Landau
  vortices, Coulomb gases, and Abrikosov lattices}},}\ }\href {\doibase
  https://doi.org/10.1016/j.crhy.2014.06.001} {\bibfield  {journal} {\bibinfo
  {journal} {C. R. Phys.}\ }\textbf {\bibinfo {volume} {15}},\ \bibinfo {pages}
  {539--546} (\bibinfo {year} {2014})}\BibitemShut {NoStop}%
\bibitem [{\citenamefont {Serfaty}(2019)}]{Serfaty-19}%
  \BibitemOpen
  \bibfield  {author} {\bibinfo {author} {\bibnamefont {Serfaty}, \bibfnamefont
  {S.}},\ }\bibfield  {title} {\enquote {\bibinfo {title} {Microscopic
  description of {L}og and {C}oulomb gases},}\ }in\ \href@noop {} {\emph
  {\bibinfo {booktitle} {Random matrices}}},\ \bibinfo {series} {IAS/Park City
  Math. Ser.}, Vol.~\bibinfo {volume} {26}\ (\bibinfo  {publisher} {Amer. Math.
  Soc., Providence, RI},\ \bibinfo {year} {2019})\ pp.\ \bibinfo {pages}
  {341--387}\BibitemShut {NoStop}%
\bibitem [{\citenamefont {Shapir}\ \emph {et~al.}(2019)\citenamefont {Shapir},
  \citenamefont {Hamo}, \citenamefont {Pecker}, \citenamefont {Moca},
  \citenamefont {Legeza}, \citenamefont {Zarand},\ and\ \citenamefont
  {Ilani}}]{Shapir-etal-19}%
  \BibitemOpen
  \bibfield  {author} {\bibinfo {author} {\bibnamefont {Shapir}, \bibfnamefont
  {I.}}, \bibinfo {author} {\bibnamefont {Hamo}, \bibfnamefont {A.}}, \bibinfo
  {author} {\bibnamefont {Pecker}, \bibfnamefont {S.}}, \bibinfo {author}
  {\bibnamefont {Moca}, \bibfnamefont {C.~P.}}, \bibinfo {author} {\bibnamefont
  {Legeza}, \bibfnamefont {{\"O}.}}, \bibinfo {author} {\bibnamefont {Zarand},
  \bibfnamefont {G.}}, \ and\ \bibinfo {author} {\bibnamefont {Ilani},
  \bibfnamefont {S.}},\ }\bibfield  {title} {\enquote {\bibinfo {title}
  {{Imaging the electronic Wigner crystal in one dimension}},}\ }\href
  {\doibase 10.1126/science.aat0905} {\bibfield  {journal} {\bibinfo  {journal}
  {Science}\ }\textbf {\bibinfo {volume} {364}},\ \bibinfo {pages} {870--875}
  (\bibinfo {year} {2019})}\BibitemShut {NoStop}%
\bibitem [{\citenamefont {Sholl}(1967)}]{Sholl-67}%
  \BibitemOpen
  \bibfield  {author} {\bibinfo {author} {\bibnamefont {Sholl}, \bibfnamefont
  {C.~A.}},\ }\bibfield  {title} {\enquote {\bibinfo {title} {The calculation
  of electrostatic energies of metals by plane-wise summation},}\ }\href
  {\doibase 10.1088/0370-1328/92/2/321} {\bibfield  {journal} {\bibinfo
  {journal} {Proc. Physical Soc.}\ }\textbf {\bibinfo {volume} {92}} (\bibinfo
  {year} {1967}),\ 10.1088/0370-1328/92/2/321}\BibitemShut {NoStop}%
\bibitem [{\citenamefont {Simmons}\ and\ \citenamefont
  {Garrod}(1973)}]{SimGar-73}%
  \BibitemOpen
  \bibfield  {author} {\bibinfo {author} {\bibnamefont {Simmons}, \bibfnamefont
  {C.~S.}}\ and\ \bibinfo {author} {\bibnamefont {Garrod}, \bibfnamefont
  {C.}},\ }\bibfield  {title} {\enquote {\bibinfo {title} {The density of a
  nonuniform system in the thermodynamic limit},}\ }\href {\doibase
  10.1063/1.1666440} {\bibfield  {journal} {\bibinfo  {journal} {J.
  Mathematical Phys.}\ }\textbf {\bibinfo {volume} {14}},\ \bibinfo {pages}
  {1075--1087} (\bibinfo {year} {1973})}\BibitemShut {NoStop}%
\bibitem [{\citenamefont {{Simons}}, \citenamefont {{Szafer}},\ and\
  \citenamefont {{Altshuler}}(1993)}]{SimSzaAlt-93}%
  \BibitemOpen
  \bibfield  {author} {\bibinfo {author} {\bibnamefont {{Simons}},
  \bibfnamefont {B.~D.}}, \bibinfo {author} {\bibnamefont {{Szafer}},
  \bibfnamefont {A.}}, \ and\ \bibinfo {author} {\bibnamefont {{Altshuler}},
  \bibfnamefont {B.~L.}},\ }\bibfield  {title} {\enquote {\bibinfo {title}
  {{Universality in quantum chaotic spectra}},}\ }\href@noop {} {\bibfield
  {journal} {\bibinfo  {journal} {JETP Lett.}\ }\textbf {\bibinfo {volume}
  {57}},\ \bibinfo {pages} {276} (\bibinfo {year} {1993})}\BibitemShut
  {NoStop}%
\bibitem [{\citenamefont {Smale}(1998)}]{Smale-98}%
  \BibitemOpen
  \bibfield  {author} {\bibinfo {author} {\bibnamefont {Smale}, \bibfnamefont
  {S.}},\ }\bibfield  {title} {\enquote {\bibinfo {title} {Mathematical
  problems for the next century},}\ }\href {\doibase 10.1007/BF03025291}
  {\bibfield  {journal} {\bibinfo  {journal} {Math. Intell.}\ }\textbf
  {\bibinfo {volume} {20}},\ \bibinfo {pages} {7--15} (\bibinfo {year}
  {1998})}\BibitemShut {NoStop}%
\bibitem [{\citenamefont {Smith}\ and\ \citenamefont
  {Penrose}(1975)}]{PenSmi-75}%
  \BibitemOpen
  \bibfield  {author} {\bibinfo {author} {\bibnamefont {Smith}, \bibfnamefont
  {E.~R.}}\ and\ \bibinfo {author} {\bibnamefont {Penrose}, \bibfnamefont
  {O.}},\ }\bibfield  {title} {\enquote {\bibinfo {title} {An upper bound on
  the free energy for classical systems with {C}oulomb interactions in a
  varying external field},}\ }\href
  {http://projecteuclid.org/euclid.cmp/1103860525} {\bibfield  {journal}
  {\bibinfo  {journal} {Comm. Math. Phys.}\ }\textbf {\bibinfo {volume} {40}},\
  \bibinfo {pages} {197--213} (\bibinfo {year} {1975})}\BibitemShut {NoStop}%
\bibitem [{\citenamefont {{Smole{\'n}ski}}\ \emph {et~al.}(2021)\citenamefont
  {{Smole{\'n}ski}}, \citenamefont {{Dolgirev}}, \citenamefont {{Kuhlenkamp}},
  \citenamefont {{Popert}}, \citenamefont {{Shimazaki}}, \citenamefont
  {{Back}}, \citenamefont {{Lu}}, \citenamefont {{Kroner}}, \citenamefont
  {{Watanabe}}, \citenamefont {{Taniguchi}}, \citenamefont {{Esterlis}},
  \citenamefont {{Demler}},\ and\ \citenamefont
  {{Imamo{\v{g}}lu}}}]{Smolenski_etal-21}%
  \BibitemOpen
  \bibfield  {author} {\bibinfo {author} {\bibnamefont {{Smole{\'n}ski}},
  \bibfnamefont {T.}}, \bibinfo {author} {\bibnamefont {{Dolgirev}},
  \bibfnamefont {P.~E.}}, \bibinfo {author} {\bibnamefont {{Kuhlenkamp}},
  \bibfnamefont {C.}}, \bibinfo {author} {\bibnamefont {{Popert}},
  \bibfnamefont {A.}}, \bibinfo {author} {\bibnamefont {{Shimazaki}},
  \bibfnamefont {Y.}}, \bibinfo {author} {\bibnamefont {{Back}}, \bibfnamefont
  {P.}}, \bibinfo {author} {\bibnamefont {{Lu}}, \bibfnamefont {X.}}, \bibinfo
  {author} {\bibnamefont {{Kroner}}, \bibfnamefont {M.}}, \bibinfo {author}
  {\bibnamefont {{Watanabe}}, \bibfnamefont {K.}}, \bibinfo {author}
  {\bibnamefont {{Taniguchi}}, \bibfnamefont {T.}}, \bibinfo {author}
  {\bibnamefont {{Esterlis}}, \bibfnamefont {I.}}, \bibinfo {author}
  {\bibnamefont {{Demler}}, \bibfnamefont {E.}}, \ and\ \bibinfo {author}
  {\bibnamefont {{Imamo{\v{g}}lu}}, \bibfnamefont {A.}},\ }\bibfield  {title}
  {\enquote {\bibinfo {title} {{Signatures of Wigner crystal of electrons in a
  monolayer semiconductor}},}\ }\href {\doibase 10.1038/s41586-021-03590-4}
  {\bibfield  {journal} {\bibinfo  {journal} {Nature}\ }\textbf {\bibinfo
  {volume} {595}},\ \bibinfo {pages} {53--57} (\bibinfo {year}
  {2021})}\BibitemShut {NoStop}%
\bibitem [{\citenamefont {So}\ \emph {et~al.}(1995)\citenamefont {So},
  \citenamefont {Anlage}, \citenamefont {Ott},\ and\ \citenamefont
  {Oerter}}]{SoAnlOttOer-95}%
  \BibitemOpen
  \bibfield  {author} {\bibinfo {author} {\bibnamefont {So}, \bibfnamefont
  {P.}}, \bibinfo {author} {\bibnamefont {Anlage}, \bibfnamefont {S.~M.}},
  \bibinfo {author} {\bibnamefont {Ott}, \bibfnamefont {E.}}, \ and\ \bibinfo
  {author} {\bibnamefont {Oerter}, \bibfnamefont {R.~N.}},\ }\bibfield  {title}
  {\enquote {\bibinfo {title} {Wave chaos experiments with and without time
  reversal symmetry: Gue and goe statistics},}\ }\href {\doibase
  10.1103/PhysRevLett.74.2662} {\bibfield  {journal} {\bibinfo  {journal}
  {Phys. Rev. Lett.}\ }\textbf {\bibinfo {volume} {74}},\ \bibinfo {pages}
  {2662--2665} (\bibinfo {year} {1995})}\BibitemShut {NoStop}%
\bibitem [{\citenamefont {Soshnikov}(2000)}]{Soshnikov-00}%
  \BibitemOpen
  \bibfield  {author} {\bibinfo {author} {\bibnamefont {Soshnikov},
  \bibfnamefont {A.}},\ }\bibfield  {title} {\enquote {\bibinfo {title}
  {Determinantal random point fields},}\ }\href {\doibase
  10.1070/rm2000v055n05ABEH000321} {\bibfield  {journal} {\bibinfo  {journal}
  {Uspekhi Mat. Nauk}\ }\textbf {\bibinfo {volume} {55}},\ \bibinfo {pages}
  {107--160} (\bibinfo {year} {2000})}\BibitemShut {NoStop}%
\bibitem [{\citenamefont {Stephenson}(2003)}]{Stephenson-PhD}%
  \BibitemOpen
  \bibfield  {author} {\bibinfo {author} {\bibnamefont {Stephenson},
  \bibfnamefont {A.}},\ }\emph {\bibinfo {title} {Studies of the {O}ne
  {C}omponent {P}lasma}},\ \href@noop {} {Ph.D. thesis} (\bibinfo {year}
  {2003})\BibitemShut {NoStop}%
\bibitem [{\citenamefont {Stormer}, \citenamefont {Tsui},\ and\ \citenamefont
  {Gossard}(1999)}]{StoTsuGos-99}%
  \BibitemOpen
  \bibfield  {author} {\bibinfo {author} {\bibnamefont {Stormer}, \bibfnamefont
  {H.~L.}}, \bibinfo {author} {\bibnamefont {Tsui}, \bibfnamefont {D.~C.}}, \
  and\ \bibinfo {author} {\bibnamefont {Gossard}, \bibfnamefont {A.~C.}},\
  }\bibfield  {title} {\enquote {\bibinfo {title} {The fractional quantum
  {H}all effect},}\ }\href {\doibase 10.1103/RevModPhys.71.S298} {\bibfield
  {journal} {\bibinfo  {journal} {Rev. Mod. Phys.}\ }\textbf {\bibinfo {volume}
  {71}},\ \bibinfo {pages} {S298--S305} (\bibinfo {year} {1999})}\BibitemShut
  {NoStop}%
\bibitem [{\citenamefont {Strandburg}(1988)}]{Strandburg-88}%
  \BibitemOpen
  \bibfield  {author} {\bibinfo {author} {\bibnamefont {Strandburg},
  \bibfnamefont {K.~J.}},\ }\bibfield  {title} {\enquote {\bibinfo {title}
  {Two-dimensional melting},}\ }\href {\doibase 10.1103/RevModPhys.60.161}
  {\bibfield  {journal} {\bibinfo  {journal} {Rev. Mod. Phys.}\ }\textbf
  {\bibinfo {volume} {60}},\ \bibinfo {pages} {161--207} (\bibinfo {year}
  {1988})}\BibitemShut {NoStop}%
\bibitem [{\citenamefont {{Sun}}, \citenamefont {{Perdew}},\ and\ \citenamefont
  {{Ruzsinszky}}(2015)}]{SunPerRuz-15}%
  \BibitemOpen
  \bibfield  {author} {\bibinfo {author} {\bibnamefont {{Sun}}, \bibfnamefont
  {J.}}, \bibinfo {author} {\bibnamefont {{Perdew}}, \bibfnamefont {J.~P.}}, \
  and\ \bibinfo {author} {\bibnamefont {{Ruzsinszky}}, \bibfnamefont {A.}},\
  }\bibfield  {title} {\enquote {\bibinfo {title} {Semilocal density functional
  obeying a strongly tightened bound for exchange},}\ }\href {\doibase
  10.1073/pnas.1423145112} {\bibfield  {journal} {\bibinfo  {journal} {Proc.
  Nat. Acad. Sci. U.S.A.}\ }\textbf {\bibinfo {volume} {112}},\ \bibinfo
  {pages} {685--689} (\bibinfo {year} {2015})}\BibitemShut {NoStop}%
\bibitem [{\citenamefont {Sun}\ \emph {et~al.}(2016)\citenamefont {Sun},
  \citenamefont {Remsing}, \citenamefont {Zhang}, \citenamefont {Sun},
  \citenamefont {Ruzsinszky}, \citenamefont {Peng}, \citenamefont {Yang},
  \citenamefont {Paul}, \citenamefont {Waghmare}, \citenamefont {Wu},
  \citenamefont {Klein},\ and\ \citenamefont {Perdew}}]{Perdew_etal-16}%
  \BibitemOpen
  \bibfield  {author} {\bibinfo {author} {\bibnamefont {Sun}, \bibfnamefont
  {J.}}, \bibinfo {author} {\bibnamefont {Remsing}, \bibfnamefont {R.~C.}},
  \bibinfo {author} {\bibnamefont {Zhang}, \bibfnamefont {Y.}}, \bibinfo
  {author} {\bibnamefont {Sun}, \bibfnamefont {Z.}}, \bibinfo {author}
  {\bibnamefont {Ruzsinszky}, \bibfnamefont {A.}}, \bibinfo {author}
  {\bibnamefont {Peng}, \bibfnamefont {H.}}, \bibinfo {author} {\bibnamefont
  {Yang}, \bibfnamefont {Z.}}, \bibinfo {author} {\bibnamefont {Paul},
  \bibfnamefont {A.}}, \bibinfo {author} {\bibnamefont {Waghmare},
  \bibfnamefont {U.}}, \bibinfo {author} {\bibnamefont {Wu}, \bibfnamefont
  {X.}}, \bibinfo {author} {\bibnamefont {Klein}, \bibfnamefont {M.~L.}}, \
  and\ \bibinfo {author} {\bibnamefont {Perdew}, \bibfnamefont {J.~P.}},\
  }\bibfield  {title} {\enquote {\bibinfo {title} {Accurate first-principles
  structures and energies of diversely bonded systems from an efficient density
  functional},}\ }\href {\doibase 10.1038/nchem.2535} {\bibfield  {journal}
  {\bibinfo  {journal} {Nature Chemistry}\ }\textbf {\bibinfo {volume} {8}},\
  \bibinfo {pages} {831--836} (\bibinfo {year} {2016})}\BibitemShut {NoStop}%
\bibitem [{\citenamefont {{Sutherland}}(1971{\natexlab{a}})}]{Sutherland-71}%
  \BibitemOpen
  \bibfield  {author} {\bibinfo {author} {\bibnamefont {{Sutherland}},
  \bibfnamefont {B.}},\ }\bibfield  {title} {\enquote {\bibinfo {title}
  {Quantum many-body problem in one dimension: Ground state},}\ }\href
  {\doibase 10.1063/1.1665584} {\bibfield  {journal} {\bibinfo  {journal} {J.
  Mathematical Phys.}\ }\textbf {\bibinfo {volume} {12}},\ \bibinfo {pages}
  {246--250} (\bibinfo {year} {1971}{\natexlab{a}})}\BibitemShut {NoStop}%
\bibitem [{\citenamefont {{Sutherland}}(1971{\natexlab{b}})}]{Sutherland-71b}%
  \BibitemOpen
  \bibfield  {author} {\bibinfo {author} {\bibnamefont {{Sutherland}},
  \bibfnamefont {B.}},\ }\bibfield  {title} {\enquote {\bibinfo {title}
  {Quantum many-body problem in one dimension: Thermodynamics},}\ }\href
  {\doibase 10.1063/1.1665585} {\bibfield  {journal} {\bibinfo  {journal} {J.
  Mathematical Phys.}\ }\textbf {\bibinfo {volume} {12}},\ \bibinfo {pages}
  {251--256} (\bibinfo {year} {1971}{\natexlab{b}})}\BibitemShut {NoStop}%
\bibitem [{\citenamefont {Sutherland}(2004)}]{Sutherland-04}%
  \BibitemOpen
  \bibfield  {author} {\bibinfo {author} {\bibnamefont {Sutherland},
  \bibfnamefont {B.}},\ }\href {\doibase 10.1142/5552} {\emph {\bibinfo {title}
  {Beautiful Models: 70 Years of Exactly Solved Quantum Many-Body Problems}}}\
  (\bibinfo  {publisher} {World Scientific Press},\ \bibinfo {year}
  {2004})\BibitemShut {NoStop}%
\bibitem [{\citenamefont {Tao}(2012)}]{Tao-12}%
  \BibitemOpen
  \bibfield  {author} {\bibinfo {author} {\bibnamefont {Tao}, \bibfnamefont
  {T.}},\ }\href
  {https://terrytao.wordpress.com/2012/11/11/a-direct-proof-of-the-stationarity-of-the-dyson-sine-process-under-dyson-brownian-motion/}
  {\enquote {\bibinfo {title} {A direct proof of the stationarity of the
  {D}yson sine process under {D}yson {B}rownian motion},}\ }\bibinfo
  {howpublished} {blog article} (\bibinfo {year} {2012})\BibitemShut {NoStop}%
\bibitem [{\citenamefont {Tao}(2013)}]{Tao-13}%
  \BibitemOpen
  \bibfield  {author} {\bibinfo {author} {\bibnamefont {Tao}, \bibfnamefont
  {T.}},\ }\href {\doibase 10.1090/mbk/081} {\emph {\bibinfo {title}
  {Compactness and contradiction}}}\ (\bibinfo  {publisher} {American
  Mathematical Society, Providence, RI},\ \bibinfo {year} {2013})\ pp.\
  \bibinfo {pages} {xii+256}\BibitemShut {NoStop}%
\bibitem [{\citenamefont {Tao}\ and\ \citenamefont {Vu}(2011)}]{TaoVu-11}%
  \BibitemOpen
  \bibfield  {author} {\bibinfo {author} {\bibnamefont {Tao}, \bibfnamefont
  {T.}}\ and\ \bibinfo {author} {\bibnamefont {Vu}, \bibfnamefont {V.}},\
  }\bibfield  {title} {\enquote {\bibinfo {title} {Random matrices:
  universality of local eigenvalue statistics},}\ }\href {\doibase
  10.1007/s11511-011-0061-3} {\bibfield  {journal} {\bibinfo  {journal} {Acta
  Math.}\ }\textbf {\bibinfo {volume} {206}},\ \bibinfo {pages} {127--204}
  (\bibinfo {year} {2011})}\BibitemShut {NoStop}%
\bibitem [{\citenamefont {T\'{e}llez}\ and\ \citenamefont
  {Forrester}(1999)}]{TelFor-99}%
  \BibitemOpen
  \bibfield  {author} {\bibinfo {author} {\bibnamefont {T\'{e}llez},
  \bibfnamefont {G.}}\ and\ \bibinfo {author} {\bibnamefont {Forrester},
  \bibfnamefont {P.~J.}},\ }\bibfield  {title} {\enquote {\bibinfo {title}
  {Exact finite-size study of the 2{D} {OCP} at {$\Gamma=4$} and
  {$\Gamma=6$}},}\ }\href {\doibase 10.1023/A:1004654923170} {\bibfield
  {journal} {\bibinfo  {journal} {J. Statist. Phys.}\ }\textbf {\bibinfo
  {volume} {97}},\ \bibinfo {pages} {489--521} (\bibinfo {year}
  {1999})}\BibitemShut {NoStop}%
\bibitem [{\citenamefont {T\'{e}llez}\ and\ \citenamefont
  {Forrester}(2012)}]{TelFor-12}%
  \BibitemOpen
  \bibfield  {author} {\bibinfo {author} {\bibnamefont {T\'{e}llez},
  \bibfnamefont {G.}}\ and\ \bibinfo {author} {\bibnamefont {Forrester},
  \bibfnamefont {P.~J.}},\ }\bibfield  {title} {\enquote {\bibinfo {title}
  {Expanded {V}andermonde powers and sum rules for the two-dimensional
  one-component plasma},}\ }\href {\doibase 10.1007/s10955-012-0551-5}
  {\bibfield  {journal} {\bibinfo  {journal} {J. Stat. Phys.}\ }\textbf
  {\bibinfo {volume} {148}},\ \bibinfo {pages} {824--855} (\bibinfo {year}
  {2012})}\BibitemShut {NoStop}%
\bibitem [{\citenamefont {Thomas}\ \emph {et~al.}(1994)\citenamefont {Thomas},
  \citenamefont {Morfill}, \citenamefont {Demmel}, \citenamefont {Goree},
  \citenamefont {Feuerbacher},\ and\ \citenamefont
  {M{\"o}hlmann}}]{Thoetal-94}%
  \BibitemOpen
  \bibfield  {author} {\bibinfo {author} {\bibnamefont {Thomas}, \bibfnamefont
  {H.}}, \bibinfo {author} {\bibnamefont {Morfill}, \bibfnamefont {G.~E.}},
  \bibinfo {author} {\bibnamefont {Demmel}, \bibfnamefont {V.}}, \bibinfo
  {author} {\bibnamefont {Goree}, \bibfnamefont {J.}}, \bibinfo {author}
  {\bibnamefont {Feuerbacher}, \bibfnamefont {B.}}, \ and\ \bibinfo {author}
  {\bibnamefont {M{\"o}hlmann}, \bibfnamefont {D.}},\ }\bibfield  {title}
  {\enquote {\bibinfo {title} {Plasma crystal: Coulomb crystallization in a
  dusty plasma},}\ }\href {\doibase 10.1103/PhysRevLett.73.652} {\bibfield
  {journal} {\bibinfo  {journal} {Phys. Rev. Lett.}\ }\textbf {\bibinfo
  {volume} {73}},\ \bibinfo {pages} {652--655} (\bibinfo {year}
  {1994})}\BibitemShut {NoStop}%
\bibitem [{\citenamefont {Thompson}(2015)}]{Thompson-15}%
  \BibitemOpen
  \bibfield  {author} {\bibinfo {author} {\bibnamefont {Thompson},
  \bibfnamefont {R.~C.}},\ }\bibfield  {title} {\enquote {\bibinfo {title} {Ion
  coulomb crystals},}\ }\href {\doibase 10.1080/00107514.2014.989715}
  {\bibfield  {journal} {\bibinfo  {journal} {Contemporary Physics}\ }\textbf
  {\bibinfo {volume} {56}},\ \bibinfo {pages} {63--79} (\bibinfo {year}
  {2015})}\BibitemShut {NoStop}%
\bibitem [{\citenamefont {Thomson}(1904)}]{Thomson-04}%
  \BibitemOpen
  \bibfield  {author} {\bibinfo {author} {\bibnamefont {Thomson}, \bibfnamefont
  {J.}},\ }\bibfield  {title} {\enquote {\bibinfo {title} {Xxiv. {O}n the
  structure of the atom: an investigation of the stability and periods of
  oscillation of a number of corpuscles arranged at equal intervals around the
  circumference of a circle; with application of the results to the theory of
  atomic structure},}\ }\href {\doibase 10.1080/14786440409463107} {\bibfield
  {journal} {\bibinfo  {journal} {Philos. Mag.}\ }\textbf {\bibinfo {volume}
  {7}},\ \bibinfo {pages} {237--265} (\bibinfo {year} {1904})}\BibitemShut
  {NoStop}%
\bibitem [{\citenamefont {Torquato}(2018)}]{Torquato-18}%
  \BibitemOpen
  \bibfield  {author} {\bibinfo {author} {\bibnamefont {Torquato},
  \bibfnamefont {S.}},\ }\bibfield  {title} {\enquote {\bibinfo {title}
  {Hyperuniform states of matter},}\ }\href {\doibase
  10.1016/j.physrep.2018.03.001} {\bibfield  {journal} {\bibinfo  {journal}
  {Phys. Rep.}\ }\textbf {\bibinfo {volume} {745}},\ \bibinfo {pages} {1--95}
  (\bibinfo {year} {2018})}\BibitemShut {NoStop}%
\bibitem [{\citenamefont {Torquato}\ and\ \citenamefont
  {Stillinger}(2003)}]{TorSti-03}%
  \BibitemOpen
  \bibfield  {author} {\bibinfo {author} {\bibnamefont {Torquato},
  \bibfnamefont {S.}}\ and\ \bibinfo {author} {\bibnamefont {Stillinger},
  \bibfnamefont {F.~H.}},\ }\bibfield  {title} {\enquote {\bibinfo {title}
  {Local density fluctuations, hyperuniformity, and order metrics},}\ }\href
  {\doibase 10.1103/PhysRevE.68.041113} {\bibfield  {journal} {\bibinfo
  {journal} {Phys. Rev. E}\ }\textbf {\bibinfo {volume} {68}},\ \bibinfo
  {pages} {041113} (\bibinfo {year} {2003})}\BibitemShut {NoStop}%
\bibitem [{\citenamefont {Touchette}(2009)}]{Touchette-09}%
  \BibitemOpen
  \bibfield  {author} {\bibinfo {author} {\bibnamefont {Touchette},
  \bibfnamefont {H.}},\ }\bibfield  {title} {\enquote {\bibinfo {title} {The
  large deviation approach to statistical mechanics},}\ }\href {\doibase
  10.1016/j.physrep.2009.05.002} {\bibfield  {journal} {\bibinfo  {journal}
  {Phys. Rep.}\ }\textbf {\bibinfo {volume} {478}},\ \bibinfo {pages} {1--69}
  (\bibinfo {year} {2009})}\BibitemShut {NoStop}%
\bibitem [{\citenamefont {Travesset}(2014)}]{Traversset-14}%
  \BibitemOpen
  \bibfield  {author} {\bibinfo {author} {\bibnamefont {Travesset},
  \bibfnamefont {A.}},\ }\bibfield  {title} {\enquote {\bibinfo {title} {Phase
  diagram of power law and {L}ennard-{J}ones systems: {C}rystal phases},}\
  }\href {\doibase 10.1063/1.4898371} {\bibfield  {journal} {\bibinfo
  {journal} {J. Chem. Phys.}\ }\textbf {\bibinfo {volume} {141}},\ \bibinfo
  {pages} {164501} (\bibinfo {year} {2014})}\BibitemShut {NoStop}%
\bibitem [{\citenamefont {Uhlenbeck}(1968)}]{Uhlenbeck-68}%
  \BibitemOpen
  \bibfield  {author} {\bibinfo {author} {\bibnamefont {Uhlenbeck},
  \bibfnamefont {G.}},\ }\bibfield  {title} {\enquote {\bibinfo {title} {An
  outline of statistical mechanics},}\ }in\ \href@noop {} {\emph {\bibinfo
  {booktitle} {Fundamental problems in statistical mechanics {II}}}},\ \bibinfo
  {editor} {edited by\ \bibinfo {editor} {\bibfnamefont {E.}~\bibnamefont
  {Cohen}}}\ (\bibinfo  {publisher} {North-Holland Publishing Company},\
  \bibinfo {address} {Amsterdam},\ \bibinfo {year} {1968})\ pp.\ \bibinfo
  {pages} {1--29}\BibitemShut {NoStop}%
\bibitem [{\citenamefont {Valk\'{o}}\ and\ \citenamefont
  {Vir\'{a}g}(2009)}]{ValVir-09}%
  \BibitemOpen
  \bibfield  {author} {\bibinfo {author} {\bibnamefont {Valk\'{o}},
  \bibfnamefont {B.}}\ and\ \bibinfo {author} {\bibnamefont {Vir\'{a}g},
  \bibfnamefont {B.}},\ }\bibfield  {title} {\enquote {\bibinfo {title}
  {Continuum limits of random matrices and the {B}rownian carousel},}\ }\href
  {\doibase 10.1007/s00222-009-0180-z} {\bibfield  {journal} {\bibinfo
  {journal} {Invent. Math.}\ }\textbf {\bibinfo {volume} {177}},\ \bibinfo
  {pages} {463--508} (\bibinfo {year} {2009})}\BibitemShut {NoStop}%
\bibitem [{\citenamefont {Valk\'{o}}\ and\ \citenamefont
  {Vir\'{a}g}(2017)}]{ValVir-17}%
  \BibitemOpen
  \bibfield  {author} {\bibinfo {author} {\bibnamefont {Valk\'{o}},
  \bibfnamefont {B.}}\ and\ \bibinfo {author} {\bibnamefont {Vir\'{a}g},
  \bibfnamefont {B.}},\ }\bibfield  {title} {\enquote {\bibinfo {title} {The
  {$\rm Sine_\beta$} operator},}\ }\href {\doibase 10.1007/s00222-016-0709-x}
  {\bibfield  {journal} {\bibinfo  {journal} {Invent. Math.}\ }\textbf
  {\bibinfo {volume} {209}},\ \bibinfo {pages} {275--327} (\bibinfo {year}
  {2017})}\BibitemShut {NoStop}%
\bibitem [{\citenamefont {{Van Horn}}(2015)}]{VanHorn-15}%
  \BibitemOpen
  \bibfield  {author} {\bibinfo {author} {\bibnamefont {{Van Horn}},
  \bibfnamefont {H.~M.}},\ }\href {\doibase 10.1007/978-3-319-09369-7} {\emph
  {\bibinfo {title} {Unlocking the Secrets of White Dwarf Stars}}},\
  Astronomers' Universe\ (\bibinfo  {publisher} {Springer International
  Publishing Switzerland},\ \bibinfo {year} {2015})\BibitemShut {NoStop}%
\bibitem [{\citenamefont {Vaninsky}(1997)}]{Vaninsky-97}%
  \BibitemOpen
  \bibfield  {author} {\bibinfo {author} {\bibnamefont {Vaninsky},
  \bibfnamefont {K.~L.}},\ }\bibfield  {title} {\enquote {\bibinfo {title}
  {Gibbs' states for {M}oser-{C}alogero potentials},}\ }\href {\doibase
  10.1142/S0217979297000277} {\bibfield  {journal} {\bibinfo  {journal}
  {International Journal of Modern Physics B}\ }\textbf {\bibinfo {volume}
  {11}},\ \bibinfo {pages} {203--211} (\bibinfo {year} {1997})}\BibitemShut
  {NoStop}%
\bibitem [{\citenamefont {Vaninsky}(2000)}]{Vaninsky-00}%
  \BibitemOpen
  \bibfield  {author} {\bibinfo {author} {\bibnamefont {Vaninsky},
  \bibfnamefont {K.~L.}},\ }\bibfield  {title} {\enquote {\bibinfo {title}
  {Thermodynamics of {M}oser-{C}alogero potentials and {S}eiberg-{W}itten exact
  solution},}\ }in\ \href@noop {} {\emph {\bibinfo {booktitle}
  {Calogero-{M}oser-{S}utherland models ({M}ontr\'{e}al, {QC}, 1997)}}},\
  \bibinfo {series and number} {CRM Ser. Math. Phys.}\ (\bibinfo  {publisher}
  {Springer, New York},\ \bibinfo {year} {2000})\ pp.\ \bibinfo {pages}
  {497--506}\BibitemShut {NoStop}%
\bibitem [{\citenamefont {Ventevogel}(1978)}]{Ventevogel-78}%
  \BibitemOpen
  \bibfield  {author} {\bibinfo {author} {\bibnamefont {Ventevogel},
  \bibfnamefont {W.}},\ }\bibfield  {title} {\enquote {\bibinfo {title} {On the
  configuration of a one-dimensional system of interacting particles with
  minimum potential energy per particle},}\ }\href {\doibase
  10.1016/0378-4371(78)90136-X} {\bibfield  {journal} {\bibinfo  {journal}
  {Phys. A}\ }\textbf {\bibinfo {volume} {92}},\ \bibinfo {pages} {343--361}
  (\bibinfo {year} {1978})}\BibitemShut {NoStop}%
\bibitem [{\citenamefont {Ventevogel}\ and\ \citenamefont
  {Nijboer}(1979)}]{VenNij-79-2}%
  \BibitemOpen
  \bibfield  {author} {\bibinfo {author} {\bibnamefont {Ventevogel},
  \bibfnamefont {W.~J.}}\ and\ \bibinfo {author} {\bibnamefont {Nijboer},
  \bibfnamefont {B.~R.~A.}},\ }\bibfield  {title} {\enquote {\bibinfo {title}
  {On the configuration of systems of interacting particles with minimum
  potential energy per particle},}\ }\href {\doibase
  10.1016/0378-4371(79)90178-X} {\bibfield  {journal} {\bibinfo  {journal}
  {Phys. A}\ }\textbf {\bibinfo {volume} {98}},\ \bibinfo {pages} {274--288}
  (\bibinfo {year} {1979})}\BibitemShut {NoStop}%
\bibitem [{\citenamefont {Viazovska}(2021)}]{Viazovska-21}%
  \BibitemOpen
  \bibfield  {author} {\bibinfo {author} {\bibnamefont {Viazovska},
  \bibfnamefont {M.}},\ }\bibfield  {title} {\enquote {\bibinfo {title} {Almost
  impossible ${E}_{8}$ and {L}eech lattices},}\ }\href {\doibase
  10.4171/MAG-47} {\bibfield  {journal} {\bibinfo  {journal} {EMS Magazine}\
  }\textbf {\bibinfo {volume} {121}} (\bibinfo {year} {2021}),\
  10.4171/MAG-47}\BibitemShut {NoStop}%
\bibitem [{\citenamefont {\v{S}amaj}(2004)}]{Samaj-04}%
  \BibitemOpen
  \bibfield  {author} {\bibinfo {author} {\bibnamefont {\v{S}amaj},
  \bibfnamefont {L.}},\ }\bibfield  {title} {\enquote {\bibinfo {title} {Is the
  two-dimensional one-component plasma exactly solvable?}}\ }\href {\doibase
  10.1023/B:JOSS.0000044056.19438.2c} {\bibfield  {journal} {\bibinfo
  {journal} {J. Statist. Phys.}\ }\textbf {\bibinfo {volume} {117}},\ \bibinfo
  {pages} {131--158} (\bibinfo {year} {2004})}\BibitemShut {NoStop}%
\bibitem [{\citenamefont {Wehrl}(1978)}]{Wehrl-78}%
  \BibitemOpen
  \bibfield  {author} {\bibinfo {author} {\bibnamefont {Wehrl}, \bibfnamefont
  {A.}},\ }\bibfield  {title} {\enquote {\bibinfo {title} {General properties
  of entropy},}\ }\href {\doibase 10.1103/RevModPhys.50.221} {\bibfield
  {journal} {\bibinfo  {journal} {Rev. Modern Phys.}\ }\textbf {\bibinfo
  {volume} {50}},\ \bibinfo {pages} {221--260} (\bibinfo {year}
  {1978})}\BibitemShut {NoStop}%
\bibitem [{\citenamefont {de~Wette}(1980)}]{DeWette-80}%
  \BibitemOpen
  \bibfield  {author} {\bibinfo {author} {\bibnamefont {de~Wette},
  \bibfnamefont {F.~W.}},\ }\bibfield  {title} {\enquote {\bibinfo {title}
  {Comments on the electrostatic energy of a {W}igner solid},}\ }\href
  {\doibase 10.1103/PhysRevB.21.3751} {\bibfield  {journal} {\bibinfo
  {journal} {Phys. Rev. B}\ }\textbf {\bibinfo {volume} {21}},\ \bibinfo
  {pages} {3751--3753} (\bibinfo {year} {1980})}\BibitemShut {NoStop}%
\bibitem [{\citenamefont {Wigner}(1934)}]{Wigner-34}%
  \BibitemOpen
  \bibfield  {author} {\bibinfo {author} {\bibnamefont {Wigner}, \bibfnamefont
  {E.~P.}},\ }\bibfield  {title} {\enquote {\bibinfo {title} {On the
  interaction of electrons in metals},}\ }\href {\doibase
  10.1103/PhysRev.46.1002} {\bibfield  {journal} {\bibinfo  {journal} {Phys.
  Rev.}\ }\textbf {\bibinfo {volume} {46}},\ \bibinfo {pages} {1002--1011}
  (\bibinfo {year} {1934})}\BibitemShut {NoStop}%
\bibitem [{\citenamefont {Wigner}(1938)}]{Wigner-38}%
  \BibitemOpen
  \bibfield  {author} {\bibinfo {author} {\bibnamefont {Wigner}, \bibfnamefont
  {E.~P.}},\ }\bibfield  {title} {\enquote {\bibinfo {title} {Effects of the
  electron interaction on the energy levels of electrons in metals},}\ }\href
  {\doibase 10.1039/TF9383400678} {\bibfield  {journal} {\bibinfo  {journal}
  {Trans. Faraday Soc.}\ }\textbf {\bibinfo {volume} {34}},\ \bibinfo {pages}
  {678--685} (\bibinfo {year} {1938})}\BibitemShut {NoStop}%
\bibitem [{\citenamefont {Wigner}(1955)}]{Wigner-55}%
  \BibitemOpen
  \bibfield  {author} {\bibinfo {author} {\bibnamefont {Wigner}, \bibfnamefont
  {E.~P.}},\ }\bibfield  {title} {\enquote {\bibinfo {title} {Characteristic
  vectors of bordered matrices with infinite dimensions},}\ }\href
  {http://www.jstor.org/stable/1970079} {\bibfield  {journal} {\bibinfo
  {journal} {Ann. of Math.}\ }\bibinfo {series} {Second Series},\ \textbf
  {\bibinfo {volume} {62}},\ \bibinfo {pages} {pp. 548--564} (\bibinfo {year}
  {1955})}\BibitemShut {NoStop}%
\bibitem [{\citenamefont {Wigner}(1958)}]{Wigner-58}%
  \BibitemOpen
  \bibfield  {author} {\bibinfo {author} {\bibnamefont {Wigner}, \bibfnamefont
  {E.~P.}},\ }\bibfield  {title} {\enquote {\bibinfo {title} {On the
  distribution of the roots of certain symmetric matrices},}\ }\href {\doibase
  10.2307/1970008} {\bibfield  {journal} {\bibinfo  {journal} {Ann. of Math.
  (2)}\ }\textbf {\bibinfo {volume} {67}},\ \bibinfo {pages} {325--327}
  (\bibinfo {year} {1958})}\BibitemShut {NoStop}%
\bibitem [{\citenamefont {Wigner}(1967)}]{Wigner-67}%
  \BibitemOpen
  \bibfield  {author} {\bibinfo {author} {\bibnamefont {Wigner}, \bibfnamefont
  {E.~P.}},\ }\bibfield  {title} {\enquote {\bibinfo {title} {Random matrices
  in physics},}\ }\href {\doibase 10.1137/1009001} {\bibfield  {journal}
  {\bibinfo  {journal} {SIAM Review}\ }\textbf {\bibinfo {volume} {9}},\
  \bibinfo {pages} {1--23} (\bibinfo {year} {1967})}\BibitemShut {NoStop}%
\bibitem [{\citenamefont {Wilson}(1962)}]{Wilson-62}%
  \BibitemOpen
  \bibfield  {author} {\bibinfo {author} {\bibnamefont {Wilson}, \bibfnamefont
  {K.~G.}},\ }\bibfield  {title} {\enquote {\bibinfo {title} {Proof of a
  conjecture by {D}yson},}\ }\href {\doibase 10.1063/1.1724291} {\bibfield
  {journal} {\bibinfo  {journal} {J. Mathematical Phys.}\ }\textbf {\bibinfo
  {volume} {3}},\ \bibinfo {pages} {1040--1043} (\bibinfo {year}
  {1962})}\BibitemShut {NoStop}%
\bibitem [{\citenamefont {Wintgen}\ and\ \citenamefont
  {Friedrich}(1986)}]{WinFrie-86}%
  \BibitemOpen
  \bibfield  {author} {\bibinfo {author} {\bibnamefont {Wintgen}, \bibfnamefont
  {D.}}\ and\ \bibinfo {author} {\bibnamefont {Friedrich}, \bibfnamefont
  {H.}},\ }\bibfield  {title} {\enquote {\bibinfo {title} {Regularity and
  irregularity in spectra of the magnetized hydrogen atom},}\ }\href {\doibase
  10.1103/PhysRevLett.57.571} {\bibfield  {journal} {\bibinfo  {journal} {Phys.
  Rev. Lett.}\ }\textbf {\bibinfo {volume} {57}},\ \bibinfo {pages} {571--574}
  (\bibinfo {year} {1986})}\BibitemShut {NoStop}%
\bibitem [{\citenamefont {Witten}\ and\ \citenamefont
  {Pincus}(1986)}]{WitPin-86}%
  \BibitemOpen
  \bibfield  {author} {\bibinfo {author} {\bibnamefont {Witten}, \bibfnamefont
  {T.~A.}}\ and\ \bibinfo {author} {\bibnamefont {Pincus}, \bibfnamefont
  {P.~A.}},\ }\bibfield  {title} {\enquote {\bibinfo {title} {Colloid
  stabilization by long grafted polymers},}\ }\href {\doibase
  10.1021/ma00164a009} {\bibfield  {journal} {\bibinfo  {journal}
  {Macromolecules}\ }\textbf {\bibinfo {volume} {19}},\ \bibinfo {pages}
  {2509--2513} (\bibinfo {year} {1986})}\BibitemShut {NoStop}%
\bibitem [{\citenamefont {Wu}(1994)}]{Wu-94}%
  \BibitemOpen
  \bibfield  {author} {\bibinfo {author} {\bibnamefont {Wu}, \bibfnamefont
  {Y.-S.}},\ }\bibfield  {title} {\enquote {\bibinfo {title} {Statistical
  distribution for generalized ideal gas of fractional-statistics particles},}\
  }\href {\doibase 10.1103/PhysRevLett.73.922} {\bibfield  {journal} {\bibinfo
  {journal} {Phys. Rev. Lett.}\ }\textbf {\bibinfo {volume} {73}},\ \bibinfo
  {pages} {922--925} (\bibinfo {year} {1994})}\BibitemShut {NoStop}%
\bibitem [{\citenamefont {Yang}(1987)}]{Yang-87}%
  \BibitemOpen
  \bibfield  {author} {\bibinfo {author} {\bibnamefont {Yang}, \bibfnamefont
  {W.-S.}},\ }\bibfield  {title} {\enquote {\bibinfo {title} {Debye screening
  for two-dimensional {C}oulomb systems at high temperatures},}\ }\href
  {\doibase 10.1007/BF01009952} {\bibfield  {journal} {\bibinfo  {journal} {J.
  Statist. Phys.}\ }\textbf {\bibinfo {volume} {49}},\ \bibinfo {pages} {1--32}
  (\bibinfo {year} {1987})}\BibitemShut {NoStop}%
\bibitem [{\citenamefont {Young}(1979)}]{Young-79}%
  \BibitemOpen
  \bibfield  {author} {\bibinfo {author} {\bibnamefont {Young}, \bibfnamefont
  {A.~P.}},\ }\bibfield  {title} {\enquote {\bibinfo {title} {Melting and the
  vector {C}oulomb gas in two dimensions},}\ }\href {\doibase
  10.1103/PhysRevB.19.1855} {\bibfield  {journal} {\bibinfo  {journal} {Phys.
  Rev. B}\ }\textbf {\bibinfo {volume} {19}},\ \bibinfo {pages} {1855--1866}
  (\bibinfo {year} {1979})}\BibitemShut {NoStop}%
\bibitem [{\citenamefont {Zabrodin}\ and\ \citenamefont
  {Wiegmann}(2006)}]{WieZab-00}%
  \BibitemOpen
  \bibfield  {author} {\bibinfo {author} {\bibnamefont {Zabrodin},
  \bibfnamefont {A.}}\ and\ \bibinfo {author} {\bibnamefont {Wiegmann},
  \bibfnamefont {P.}},\ }\bibfield  {title} {\enquote {\bibinfo {title}
  {{Large-$N$ expansion for the 2D Dyson gas}},}\ }\href {\doibase
  10.1088/0305-4470/39/28/s10} {\ \textbf {\bibinfo {volume} {39}},\ \bibinfo
  {pages} {8933--8963} (\bibinfo {year} {2006})}\BibitemShut {NoStop}%
\bibitem [{\citenamefont {Zahn}, \citenamefont {Lenke},\ and\ \citenamefont
  {Maret}(1999)}]{ZahLenMar-99}%
  \BibitemOpen
  \bibfield  {author} {\bibinfo {author} {\bibnamefont {Zahn}, \bibfnamefont
  {K.}}, \bibinfo {author} {\bibnamefont {Lenke}, \bibfnamefont {R.}}, \ and\
  \bibinfo {author} {\bibnamefont {Maret}, \bibfnamefont {G.}},\ }\bibfield
  {title} {\enquote {\bibinfo {title} {Two-stage melting of paramagnetic
  colloidal crystals in two dimensions},}\ }\href {\doibase
  10.1103/PhysRevLett.82.2721} {\bibfield  {journal} {\bibinfo  {journal}
  {Phys. Rev. Lett.}\ }\textbf {\bibinfo {volume} {82}},\ \bibinfo {pages}
  {2721--2724} (\bibinfo {year} {1999})}\BibitemShut {NoStop}%
\bibitem [{\citenamefont {Zeitouni}\ and\ \citenamefont
  {Zelditch}(2010)}]{ZeiZel-10}%
  \BibitemOpen
  \bibfield  {author} {\bibinfo {author} {\bibnamefont {Zeitouni},
  \bibfnamefont {O.}}\ and\ \bibinfo {author} {\bibnamefont {Zelditch},
  \bibfnamefont {S.}},\ }\bibfield  {title} {\enquote {\bibinfo {title} {Large
  deviations of empirical measures of zeros of random polynomials},}\ }\href
  {\doibase 10.1093/imrn/rnp233} {\bibfield  {journal} {\bibinfo  {journal}
  {{Int. Math. Res. Not.}}\ }\textbf {\bibinfo {volume} {2010}},\ \bibinfo
  {pages} {3935--3992} (\bibinfo {year} {2010})}\BibitemShut {NoStop}%
\bibitem [{\citenamefont {{Zhang}}\ \emph {et~al.}(2017)\citenamefont
  {{Zhang}}, \citenamefont {{Pagano}}, \citenamefont {{Hess}}, \citenamefont
  {{Kyprianidis}}, \citenamefont {{Becker}}, \citenamefont {{Kaplan}},
  \citenamefont {{Gorshkov}}, \citenamefont {{Gong}},\ and\ \citenamefont
  {{Monroe}}}]{Zhang_etal-17}%
  \BibitemOpen
  \bibfield  {author} {\bibinfo {author} {\bibnamefont {{Zhang}}, \bibfnamefont
  {J.}}, \bibinfo {author} {\bibnamefont {{Pagano}}, \bibfnamefont {G.}},
  \bibinfo {author} {\bibnamefont {{Hess}}, \bibfnamefont {P.~W.}}, \bibinfo
  {author} {\bibnamefont {{Kyprianidis}}, \bibfnamefont {A.}}, \bibinfo
  {author} {\bibnamefont {{Becker}}, \bibfnamefont {P.}}, \bibinfo {author}
  {\bibnamefont {{Kaplan}}, \bibfnamefont {H.}}, \bibinfo {author}
  {\bibnamefont {{Gorshkov}}, \bibfnamefont {A.~V.}}, \bibinfo {author}
  {\bibnamefont {{Gong}}, \bibfnamefont {Z.~X.}}, \ and\ \bibinfo {author}
  {\bibnamefont {{Monroe}}, \bibfnamefont {C.}},\ }\bibfield  {title} {\enquote
  {\bibinfo {title} {{Observation of a many-body dynamical phase transition
  with a 53-qubit quantum simulator}},}\ }\href {\doibase 10.1038/nature24654}
  {\bibfield  {journal} {\bibinfo  {journal} {Nature}\ }\textbf {\bibinfo
  {volume} {551}},\ \bibinfo {pages} {601--604} (\bibinfo {year}
  {2017})}\BibitemShut {NoStop}%
\bibitem [{\citenamefont {{Zhou}}\ \emph {et~al.}(2021)\citenamefont {{Zhou}},
  \citenamefont {{Sung}}, \citenamefont {{Brutschea}}, \citenamefont
  {{Esterlis}}, \citenamefont {{Wang}}, \citenamefont {{Scuri}}, \citenamefont
  {{Gelly}}, \citenamefont {{Heo}}, \citenamefont {{Taniguchi}}, \citenamefont
  {{Watanabe}}, \citenamefont {{Zar{\'a}nd}}, \citenamefont {{Lukin}},
  \citenamefont {{Kim}}, \citenamefont {{Demler}},\ and\ \citenamefont
  {{Park}}}]{Zhou_etal-21}%
  \BibitemOpen
  \bibfield  {author} {\bibinfo {author} {\bibnamefont {{Zhou}}, \bibfnamefont
  {Y.}}, \bibinfo {author} {\bibnamefont {{Sung}}, \bibfnamefont {J.}},
  \bibinfo {author} {\bibnamefont {{Brutschea}}, \bibfnamefont {E.}}, \bibinfo
  {author} {\bibnamefont {{Esterlis}}, \bibfnamefont {I.}}, \bibinfo {author}
  {\bibnamefont {{Wang}}, \bibfnamefont {Y.}}, \bibinfo {author} {\bibnamefont
  {{Scuri}}, \bibfnamefont {G.}}, \bibinfo {author} {\bibnamefont {{Gelly}},
  \bibfnamefont {R.~J.}}, \bibinfo {author} {\bibnamefont {{Heo}},
  \bibfnamefont {H.}}, \bibinfo {author} {\bibnamefont {{Taniguchi}},
  \bibfnamefont {T.}}, \bibinfo {author} {\bibnamefont {{Watanabe}},
  \bibfnamefont {K.}}, \bibinfo {author} {\bibnamefont {{Zar{\'a}nd}},
  \bibfnamefont {G.}}, \bibinfo {author} {\bibnamefont {{Lukin}}, \bibfnamefont
  {M.~D.}}, \bibinfo {author} {\bibnamefont {{Kim}}, \bibfnamefont {P.}},
  \bibinfo {author} {\bibnamefont {{Demler}}, \bibfnamefont {E.}}, \ and\
  \bibinfo {author} {\bibnamefont {{Park}}, \bibfnamefont {H.}},\ }\bibfield
  {title} {\enquote {\bibinfo {title} {{Bilayer Wigner crystals in a transition
  metal dichalcogenide heterostructure}},}\ }\href {\doibase
  10.1038/s41586-021-03560-w} {\bibfield  {journal} {\bibinfo  {journal}
  {Nature}\ }\textbf {\bibinfo {volume} {595}},\ \bibinfo {pages} {48--52}
  (\bibinfo {year} {2021})}\BibitemShut {NoStop}%
\bibitem [{\citenamefont {Zong}, \citenamefont {Lin},\ and\ \citenamefont
  {Ceperley}(2002)}]{ZonLinCep-02}%
  \BibitemOpen
  \bibfield  {author} {\bibinfo {author} {\bibnamefont {Zong}, \bibfnamefont
  {F.~H.}}, \bibinfo {author} {\bibnamefont {Lin}, \bibfnamefont {C.}}, \ and\
  \bibinfo {author} {\bibnamefont {Ceperley}, \bibfnamefont {D.~M.}},\
  }\bibfield  {title} {\enquote {\bibinfo {title} {Spin polarization of the
  low-density three-dimensional electron gas},}\ }\href {\doibase
  10.1103/PhysRevE.66.036703} {\bibfield  {journal} {\bibinfo  {journal} {Phys.
  Rev. E}\ }\textbf {\bibinfo {volume} {66}},\ \bibinfo {pages} {036703}
  (\bibinfo {year} {2002})}\BibitemShut {NoStop}%
\end{thebibliography}

%

\end{document}